\documentclass[11pt,twoside,a4paper]{book}

\usepackage[T1]{fontenc}
\usepackage[utf8]{inputenc}
\usepackage[english]{babel}
\usepackage[pdftex]{graphicx}           
\usepackage{setspace}                   
\usepackage{indentfirst}                
\usepackage{makeidx}                    
\usepackage[nottoc]{tocbibind}          
\usepackage{courier}                    
\usepackage{type1cm}                    
\usepackage{listings}                   
\usepackage{titletoc}
\usepackage[fixlanguage]{babelbib}
\usepackage[font=small,format=plain,labelfont=bf,up,textfont=it,up]{caption}
\usepackage[usenames,svgnames,dvipsnames]{xcolor}
\usepackage[a4paper,top=2.54cm,bottom=2.0cm,left=2.0cm,right=2.54cm]{geometry} 
\usepackage[pdftex,plainpages=false,pdfpagelabels,pagebackref,colorlinks=true,citecolor=DarkGreen,linkcolor=NavyBlue,urlcolor=DarkRed,filecolor=green,bookmarksopen=true]{hyperref} 
\usepackage[all]{hypcap}                
\fontsize{60}{62}\usefont{OT1}{cmr}{m}{n}{\selectfont}

\usepackage{ragged2e}

\usepackage{epigraph}
\usepackage{framed}
\usepackage{tikz-cd,tkz-euclide}
\usepackage{tikz}
\usepackage{mdframed}
\usepackage{float}

\usepackage{mathrsfs,mathtools,amsmath,amsfonts,amssymb,amsthm,stmaryrd,bbm}
\newtheorem*{theorem*}{Theorem}
\newtheorem*{corollary*}{Corollary}
\newtheorem{theorem}{Theorem}[chapter]
\newtheorem{definition}[theorem]{Definition}
\newtheorem{proposition}[theorem]{Proposition}

\newtheorem{lemma}[theorem]{Lemma}
\newtheorem{corollary}[theorem]{Corollary}
\newtheorem{example}{Example}
\newtheorem{remark}[theorem]{Remark}

\DeclareMathOperator{\Dom}{Dom}
 
\DeclareMathOperator{\sgn}{sgn} 
 
\DeclareMathOperator{\spann}{span} 
\DeclareMathOperator{\ind}{ind}

\DeclareMathOperator{\Var}{Var}
\DeclareMathOperator{\Inv}{Inv}
\DeclareMathOperator{\Id}{Id}
\DeclareMathOperator{\supp}{supp}
\DeclareMathOperator{\ran}{Ran}

\DeclareMathOperator{\Iso}{Iso}
\newcommand{\partiald}[2]{\frac{\partial {#1}}{\partial {#2}}}

\newcommand*\circled[1]{\tikz[baseline=(char.base)]{
            \node[shape=circle,draw,inner sep=2pt] (char) {#1};}}
\newcommand{\tnorm}[1]{{\left\vert\kern-0.25ex\left\vert\kern-0.25ex\left\vert #1 
    \right\vert\kern-0.25ex\right\vert\kern-0.25ex\right\vert}}

\usepackage{fancyhdr}
\renewcommand{\chaptermark}[1]{\markboth{\MakeUppercase{#1}}{}}

\graphicspath{{./figuras/}}             
\makeindex                              
\fontsize{60}{62}\usefont{OT1}{cmr}{m}{n}{\selectfont}
\cleardoublepage
\normalsize

\begin{document}
\frontmatter 
\fancyhead[RO]{{\footnotesize\rightmark}\hspace{2em}\thepage}
\setcounter{tocdepth}{2}
\fancyhead[LE]{\thepage\hspace{2em}\footnotesize{\leftmark}}
\fancyhead[RE,LO]{}
\fancyhead[RO]{{\footnotesize\rightmark}\hspace{2em}\thepage}

\onehalfspacing  

\thispagestyle{empty}
\begin{center}
    \vspace*{2.3cm}
    \textbf{\Large{Thermodynamic Formalism for Generalized Countable Markov Shifts}}\\
    
    \vspace*{1.2cm}
    \Large{Thiago Costa Raszeja}
    
    \vskip 2cm
    \textsc{
    Thesis presented\\[-0.25cm] 
    to the \\[-0.25cm]
    Institute of Mathematics and Statistics\\[-0.25cm]
    of the\\[-0.25cm]
    University of S\~ao Paulo\\[-0.25cm]
    in partial fulfillment of the requirements\\[-0.25cm]
    for the degree\\[-0.25cm]
    of\\[-0.25cm]
    Doctor of Science}
    
    \vskip 1.5cm
    Program: Applied Mathematics\\
    Advisor: Prof. Dr. Rodrigo Bissacot\\ 

   	\vskip 1cm
    \normalsize{During the development of this work the author was supported by CAPES and CNPq.}
    
    \vskip 0.5cm
    \normalsize{S\~ao Paulo, December 2020}
\end{center}

%
%
%




%
%
%
%
\newpage
\thispagestyle{empty}
    \begin{center}
        \vspace*{2.3 cm}
        \textbf{\Large{Thermodynamic Formalism for Generalized Countable Markov Shifts}}\\
        \vspace*{2 cm}
    \end{center}

    \vskip 2cm

    \begin{FlushRight}
     This is the final version of the thesis and it contains corrections \\
     and changes suggested by the examiner commitee during the\\
     defense of the original work, realized in December 17th, 2020.\\
     A copy of the original version of this text is avaliable at the \\
     Institute of Mathematics and Statistics of the University of S\~ao Paulo.
    
    \vskip 2cm

   \end{FlushRight}
    \vskip 4.2cm

    \begin{quote}
    \noindent Examiner Committee:
    
    \begin{itemize}
		\item Prof. Dr. Severino Toscano do R\^ego Melo - IME-USP (President)
		\item Prof. Dr. Godofredo Iommi - Pontificia Universidad Cat\'olica de Chile
		\item Prof. Dr. Bartosz Kosma Kwa\'sniewski  - Uniwersytet w Bia\l{}ystok
		\item Prof. Dr. Jean Renault - Universit\'e d'Orl\'eans
        \item Prof. Dr. Aidan Sims - University of Wollongong
    \end{itemize}
      
    \end{quote}
\pagebreak

\pagenumbering{roman}     

\chapter*{Aknowledgements}

This is a (perhaps too simple) part of this thesis where I can express my gratitude for all those who supported me in many different ways, from nice moments laughing together to intense mathematical discussions. I feel very lucky to meet so nice people and it was amazing to interact with all of you during all these years, and I wish the best for you. I hope to see you all again soon!

First I thank my family, to which I dedicate this thesis, for all the love and care. To my parents, Airton Raszeja and Solange Maria P. C. Raszeja, and my brother, Fabr\'icio C. Raszeja. They always supported me in my career and all the changes in it, thank you very much!

I thank my advisor (and friend), prof. Rodrigo Bissacot, for amazing discussions, for tons of motivation, great courses, and for all the hard work which allowed the occurrence of my professional phase transition: from a theoretical physicist into a mathematician physicist. Also, I thank his spouse, prof. Patricia L. Cunha and their son, Fabricio Bissacot, for many great conversations and laughs, and for being so supportive to me during this journey.

Also, I thank everybody from prof. Rodrigo's research group for being so supportive and enthusiastic on Mathematics, which implied a lot of interesting discussions, and for participating in my seminars: Lucas Affonso, Elmer R. B\'eltran, Henrique Corsini, Eric O. Endo, Rodrigo Frausino, Bruno Kimura, Rafael P. Lima, Jo\~ao V. Maia, Renan L. Molinari, D\'ebora Oliveira, Samir Salmen, Leonardo Teramatsu, Kelvyn Welsch. Special thanks to Lucas, Henrique, Rodrigo F., Rafael, and Jo\~ao Maia for careful readings and for pointing some text mistakes. Besides, the author thanks Rafael and Rodrigo F. for discussions on groupoid C$^*$-algebras theory, and to Elmer for teaching me a lot about standard Markov shift spaces and for suggesting us to try the potential of Theorem \ref{thm:complete_characterization_eigenmeasures_renewal} to study the characterization of eigenmeasures on the generalized renewal shift. A special thanks to Rodrigo F. as co-author of the paper related to this thesis \cite{BEFR2018} and for many technical discussions. The author also thanks Christian T\'afula for great mathematical discussions.

I thank prof. Ruy Exel, also co-author of the paper related to this thesis, for coming to S\~ao Paulo for an intense month of discussions and work and for later discussions at ICM2018. 

Also, I thank prof. Marcelo Laca for advising me during one semester at UVic, for great discussions and all the support. In addition, I thank everybody I had the opportunity to interact with at Victoria: prof. Ian Putnam, prof. Chris Bose, Chris Bruce, Mitch Haslehurst, Joseph Horan, Anna Duwenig, Jeremy Hume, Lorena Aguirre Salazar, Alexander Chernyavsky, Barbara MacDonald, Chris MacDonald, Scot MacDonald, Eugene Oppong, Derrick Oppong, Karina Perry and Yvonne Trott.  

I thank professor Manfred Denker for showing me his and Yuri's paper \cite{DenYu2015} during a conference in Rio de Janeiro, which allowed us to explore even further the characterization of the eigenmeasures of the generalized Markov shifts.

I also thank prof. Walter de Siqueira Pedra for the great course on C$^*$-algebras and foundations of Quantum Mechanics and Statistical Mechanics, which made me change my research topic from my master's to my Ph.D., by showing how interesting the Operator Algebras can be.

I thank prof. Oswaldo Rio Branco de Oliveira for the great course on Measure Theory, which was crucial for my research.

Furthermore, I thank all friends that motivated and inspired me: Gabriel de Lara Melo, Luiz Roberto de Lara Melo, Elza Lima de Lara, Paulo Henrique Macedo, Carolina Alexiou, Thomas Gauran, Pietro Raele, Wilson Maruyama, Ana Carolina Souza, Caroline Andreassa, Vin\'icius Fran\c{c}\~ao, Nicholas Busic, Pryscila B. Grossch\"adl, Benedito Faustinoni, Juliana P. de Souza, Deborah Yohana Bertoldo, Vitor Yukio Fugivala, N\'ickolas Alves, Pedro Tredezini, prof. Gabrielle Weber, prof. Rebeca Bacani, Gabriela Lima Lichtenstein, Rita Stasevskas Kujawski, Ricardo Correa da Silva, Andr\'e Z. Vitorelli,  Arthur Loureiro, Gustavo Soares, Guilherme Ruiz, Brendon Faccion, J\'ulio Vin\'ius Rodrigues Miguel, Andr\'e Rodrigo da Silva, Danilo F. Putinato, Rodrigo Fuscaldo, prof. Francisco Nascimento, Talita Lopes da Cruz, Gabriel P. e Dechiche, Rafael R. de Campos and Thiago G. Tartaro.

I thank the examiner committee for participating in my Ph.D. defense: prof. Severino T. do R\^ego Melo, prof. Godofredo Iommi, prof. Bartosz K. Kwa\'sniewski, prof. Jean Renault and prof. Aidan Sims. A special thanks to Italo Cipriano for the big e-mail with many good suggestions and showing me many typos I committed.

Finally, I thank CNPq and CAPES for the financial support during this Ph.D., and the Institute of Mathematics and Statistics of the University of S\~ao Paulo.

\chapter*{Resumo}

\noindent RASZEJA, T. \textbf{Formalismo Termodin\^amico para Shifts de Markov Cont\'aveis Generalizados}. 
2020. 218 f.
Tese (Doutorado) - Instituto de Matem\'atica e Estat\'istica,
Universidade de S\~ao Paulo, S\~ao Paulo, 2020.
\\

Shifts de Markov com alfabeto enumer\'avel, os quais denotamos por $\Sigma_A$ para uma matriz 0-1 infinita $A$, s\~ao objetos centrais em Din\^amica Simb\'olica e Teoria Erg\'odica. R. Exel e M. Laca introduziram suas correspondentes \'algebras de operadores como uma generaliza\c{c}\~ao das \'algebras de Cuntz-Krieger para um alfabeto infinito e cont\'avel. Eles introduziram o conjunto $X_A=\Sigma_A \cup Y_A$, que \'e um tipo de {\it shift de Markov cont\'avel generalizado}, uma vez que coincide com o espa\c{c}o $\Sigma_A$ no caso localmente compacto. O espa\c{c}o $X_A$ cont\'em como subconjuntos densos o shift de Markov usual e um subconjunto de palavras finitas permitidas $Y_A$, este \'ultimo \'e denso quando for n\~ao vazio. Desenvolvemos o formalismo termodin\^amico para os shifts de Markov generalizados, introduzindo a no\c{c}\~ao de medida conforme em $X_A$ e explorando suas conex\~oes com o formalismo termodin\^amico usual em $\Sigma_A$. Novos fen\^omenos surgem, como diferentes tipos de transi\c{c}\~ao de fase e novas medidas conformes que n\~ao s\~ao detectadas pelo formalismo termodin\^amico cl\'assico quando a matriz n\~ao \'e row-finite. Dado um potencial $F$ e inverso da temperatura $\beta$, estudamos o problema de exist\^encia e aus\^encia de medidas conformes $\mu_{\beta}$ associadas a $\beta F$. Apresentamos exemplos onde existe um valor cr\'itico $\beta_c$, em que temos exist\^encia de probabilidades conformes satisfazendo $\mu_{\beta}(\Sigma_A)=0$ para todo $\beta > \beta_c$ e, na topologia fraca$^*$, quando tomamos o limite $\beta$ indo para $\beta_c$, o conjunto de probabilidades conformes para inverso de temperatura $\beta > \beta_c$ colapsa para a probabilidade conforme usual $\mu_{\beta_c}$ tal que $\mu_{\beta_c}(\Sigma_A) = 1$. Estudamos em detalhe o shift renewal generalizado e modifica\c{c}\~oes deste. Destacamos a bije\c{c}\~ao entre os elementos do alfabeto que s\~ao emissores infinitos e medidas de probabilidade conformes para essa classe de shifts do tipo renewal. Provamos a exist\^encia de automedidas de probabilidade da transforma\c{c}\~ao de Ruelle para temperaturas baixas o suficiente para um potencial particular no shift de renewal generalizado; estas medidas n\~ao s\~ao detectadas no renewal shift usual, dado que, para temperaturas baixas, o potencial $\beta F$ \'e transiente.
\\

\noindent \textbf{Palavras-chave:} transi\c{c}\~ao de fase, shift de Markov com alfabeto enumer\'avel, formalismo termodin\^amico, Exel-Laca \'algebra, estados KMS, din\^amica simb\'olica, sistemas din\^amicos, medidas conformes.

\chapter*{Abstract}
\noindent RASZEJA, T. \textbf{Thermodynamic Formalism for Generalized Countable Markov Shifts}. 
2020. 218 pages.
PhD Thesis - Institute of Mathematics and Statistics,
University of S\~ao Paulo, S\~ao Paulo, 2020.
\\

Countable Markov shifts, which we denote by $\Sigma_A$ for a 0-1 infinite matrix $A$, are central objects in Symbolic Dynamics and Ergodic Theory. R. Exel and M. Laca have introduced the corresponding operator algebras as a generalization of the Cuntz-Krieger algebras for an infinite and countable alphabet. They introduced the set $X_A=\Sigma_A \cup Y_A$, a kind of {\it generalized countable Markov shift} that coincides with the space $\Sigma_A$ in the locally compact case. The space $X_A$ contains as dense subsets the standard countable Markov $\Sigma_A$ and a subset of finite allowed words $Y_A$. The last one is dense when it is non-empty. We develop the thermodynamic formalism for the generalized countable Markov shifts $X_A$, introducing the notion of conformal measure in $X_A$ and exploring its connections with the usual thermodynamic formalism for $\Sigma_A$. New phenomena appear, as different types of phase transitions and new conformal measures that are not detected in the classical thermodynamic formalism when the matrix $A$ is not row-finite. Given a potential $F$ and inverse of temperature $\beta$, we study the problem of the existence and absence of conformal measures $\mu_{\beta}$ associated with $\beta F$. We present examples where there exists a critical $\beta_c$ such that we have the existence of conformal probabilities satisfying $\mu_{\beta}(\Sigma_A)=0$ for every $\beta > \beta_c$ and, on the weak$^*$ topology, when we take the limit on $\beta$ going to $\beta_c$, the set of conformal probabilities for the inverse of temperature $\beta > \beta_c$ collapses to the standard conformal probability $\mu_{\beta_c}$ such that $\mu_{\beta_c}(\Sigma_A)=1$. We study in detail the generalized renewal shift and modifications of it. We highlight the bijection between the elements of the alphabet, which are infinite emitters, and extremal conformal probability measures for this class of renewal type shifts. We prove the existence and uniqueness of the eigenmeasure (probability) of the Ruelle's transformation at low enough temperature for a particular potential on the generalized renewal shift; these measures are not detected on the standard renewal shift since for low temperatures, the potential $\beta F$ is transient. 
\\

\noindent \textbf{Keywords:} phase transition, countable Markov shift, thermodynamic formalism, Exel-Laca algebra, KMS states, symbolic dynamics, dynamical systems, conformal measures.

\tableofcontents    



\listoffigures            
\listoftables            

\mainmatter

\fancyhead[RE,LO]{\thesection}

\singlespacing              

\chapter*{Introduction}
\addcontentsline{toc}{chapter}{Introduction}
\markboth{INTRODUCTION}{}
The Thermodynamic Formalism consists in a rigorous and abstract approach from the point of view of Mathematical Physics into the foundations of Equilibrium Statistical Mechanics. Its Dynamical System and Ergodic Theory flavors were introduced by D. Ruelle \cite{Ruelle1967,Ruelle1978}, Y. Sinai \cite{Sinai1972} and R. Bowen \cite{Bowen1975}. One of the main topics of the area is the study of phase transition phenomena \cite{AizBarsFern1987,BelBisEndo2020,Dobrushin1965,Dobrushin1968,EnterFernSokal1993,FriedliVelenik2018,Georgii2011,Iommi2007,KuchQuasWolf2020,Peierls1936,Pesin2014,PirogovSinai1974,Sarig2001,Sarig2000}. A phase transition consists into an abrupt change in the system and it can have many different meanings, some of them are equivalent and some not, as we explain below:
\begin{itemize}
    \item[$(a)$] the Ising model on $\mathbb{Z}^d$, consisting in the study of ferromagnetism phenomena on a lattice for a two state model (see chapter 3 of \cite{FriedliVelenik2018}), where the configurations are elements of $\Omega = \{-1,+1\}^{\mathbb{Z}^d}$. The set of Gibbs measures is the closed convex hull of thermodynamic limits of the Gibbs distributions (see Theorem C.4.1 of \cite{Ellis2006} or in \cite{Israel1979}) on increasing finite sets $\{\Lambda_n\}_{n \in \mathbb{N}}$ which invade $\mathbb{Z}^d$, as we explain next. Given a boundary condition $\eta \in \Omega$, consider a configuration on a finite box $\Lambda \subset \mathbb{Z}^d$ as an element of the finite set
    \begin{equation*}
        \Omega_\Lambda^\eta := \{\omega \in \Omega: \omega_i = \eta_i \text{ for every } i \in \Lambda^c\}
    \end{equation*}
    The energy of a configuration $\omega \in \Omega_\Lambda^\eta$ of the Ising model in $\mathbb{Z}^d$ with inverse of temperature $\beta>0$ is given by
    \begin{equation*}
        \beta \mathcal{H}_{\Lambda,h}^\eta(\omega) = -\beta \sum_{\{i,j\} \in \mathscr{E}_\Lambda^b} \sigma_i (\omega) \sigma_j (\omega) - \beta h \sum_{i \in \Lambda} \sigma_i (\omega),
    \end{equation*}
    where $h \in \mathbb{R}$ is the external field $\sigma_i(\omega)$ is the projection at the coordinate $i$ of $\omega$. Also,
    \begin{equation*}
        \mathscr{E}_\Lambda^b := \{\{i,j\}\subset \mathbb{Z}^d:\{i,j\}\cap\Lambda \neq \emptyset \text{ and } i \sim j\},
    \end{equation*}
    where $i \sim j$ means that $i$ and $j$ are nearest neighbors, that is, $\|i-j\|_1 = \sum_{k=1}^d |i_k-j_k| = 1$. The Gibbs distributions at the set $\Lambda$ with boundary condition $\eta$ at parameters $\beta$ and $h$ is defined by
    \begin{equation*}
        \mu_{\Lambda,\beta,h}^\eta(\omega) = \begin{cases}
                                                \frac{e^{-\beta\mathcal{H}_{\Lambda,h}^\eta(\omega)}}{Z_{\Lambda,\beta,h}^\eta}, \quad \text{if } \omega \in \Omega_\Lambda^\eta,\\
                                                0, \quad \text{otherwise};
                                             \end{cases}
    \end{equation*}
    where $Z_{\Lambda,\beta,h}^\eta$ is the partition function, given by
    \begin{equation*}
        Z_{\Lambda,\beta,h}^\eta := \sum_{\omega \in \Omega_\Lambda^\eta} e^{-\beta\mathcal{H}_{\Lambda,h}^\eta(\omega)}.
    \end{equation*}
    The pressure is thermodynamic limit with certain control on the boundaries\footnote{The ratio of the size of the boundary of $\Lambda_n$ over its volume goes to zero as it $n$ goes to infinity.} of the family $\{\Lambda_n\}_\mathbb{N}$,
    \begin{equation*}
        P(\beta,h) = \lim_n P_{\Lambda_n}^\eta(\beta,h),
    \end{equation*}
    where the pressure at finite volume $\Lambda$ is defined by
    \begin{equation*}
        P_{\Lambda}^\eta(\beta,h) := \frac{1}{|\Lambda|}\log Z_{\Lambda,\beta,h}^\eta.
    \end{equation*}
    The function $P(\beta,h)$ is a well-defined convex function in which the limit does not depend on the choice of the sequence $\{\Lambda_n\}$ and of the boundary condition $\eta$. For fixed $h$ and $\beta$, some of the most important Gibbs measures here are the weak$^*$ limit of sequences in the form $\{\mu_{\Lambda_n,\beta,h}^{\eta_n}\}_{n \in \mathbb{N}}$. In Probability Theory, also for the Statistical Mechanics community, a phase transition occurs when there exists more than one Gibbs measure, for low enough temperature. On the other hand, for the Dynamical Systems and Ergodic Theory communities and, in part of the literature of Mathematical Physics, a phase transition occurs when $P(\beta,h)$ is not differentiable or analytic with respect to $\beta$. Sometimes, the point where the system passes from one to several measures coincides with the non-analyticity point $\beta_c$ of the pressure function. Theorem 3.25 of \cite{FriedliVelenik2018} grants that there is no phase transition for $d=1$ or $h \neq 0$. On the other hand, for $d = 2$ and $h = 0$, Onsager \cite{Onsager1944} determined the pressure explicitly, and the non-analyticity point is precisely the critical point $\beta_c$ which separates the regions of uniqueness and where we have the existence of more than one Gibbs measure: there exists a critical value $\beta_c \in (0,\infty)$ s.t. for $\beta \leq \beta_c$ we have the uniqueness of the Gibbs measure, and for $\beta > \beta_c$ we have non-uniqueness. For $d > 2$ and $h = 0$, the critical point of phase transition in the sense of uniqueness of Gibbs measures \cite{AizBarsFern1987,Ellis2006} is a point of non-analyticity of the pressure. The study of $P(\beta,0)$ analyticity with respect to $\beta$ was recently closed by S. Ott \cite{Ott2019,Ott2020}. Here, the Gibbs measures are in the sense of Dobrushin-Lanford-Ruelle \cite{Dobrushin1968,Dobrushin1968b,Dobrushin1968c,LanRu1969}, also called DLR measures.
    
    \item[$(b)$] Although the two notions for phase transition are equivalent in the model above, the equivalence does not hold in general. In \cite{BisCasCioPres2015}, the authors considered a slight modification in the model, by taking
    \begin{equation*}
        \mathcal{H}_{\Lambda}^\eta(\omega) = -J \sum_{\{i,j\} \in \mathscr{E}_\Lambda^b} \sigma_i (\omega) \sigma_j (\omega) - \sum_{i \in \Lambda} h_i \sigma_i (\omega),
    \end{equation*}
    that is, the external field is not constant anymore, and in this case, it is taken to be
    \begin{equation*}
        h(x) =  \begin{cases}
                   \frac{h^*}{\|x\|^\alpha}, \quad x \neq 0;\\
                   h^*, \quad x = 0,
                \end{cases}
    \end{equation*}
    where $\alpha > 0$, $J > 0$ and $h^* > 0$. By taking van Hove sequences, it can be shown that the pressure for this model is the same as in item $(a)$ for $h = 0$, that is, when there is no external field, and therefore there exists a critical value $\beta_c$ s.t. the pressure is not analytic. In \cite{BisCasCioPres2015}, if $0 < \alpha < 1$, they proved that there exists $\beta_\alpha < \infty$ s.t. for $\beta > \beta_\alpha$ we have the uniqueness of the DLR measure. On the other hand, by the Dobrushin Uniqueness Theorem \cite{Dobrushin1968}, we also have a unique DLR measure for $\beta < 1/(2dJ)$. L. Ciolleti and R. Vila \cite{CiolettiVila2015} completed the description, proving that such uniqueness holds for every $\beta >0$. Therefore we have a model on $\mathbb{Z}^d$ such that the notions of phase transition on lack of analyticity of the pressure and the non-uniqueness of the DLR measures are not equivalent.
\end{itemize}

In the context of symbolic dynamics, when we compare the standard theory of Markov shifts with the generalized one presented in this thesis, we have a similar change of behavior in the sense of dissociation of the notions of phase transitions. In fact, we have by Theorem 5 of \cite{Sarig2001}, which says that for the renewal shift space and a bounded above potential with summable variations s.t. its induced potential is locally H\"older continuous, there exists a critical value $\beta_c$ s.t. the Gurevich pressure $P_G(\beta F)$ is not analytic for $\beta = \beta_c$. Moreover, due to the Generalized RPF Theorem \cite{Sarig2015}, we have the existence of a unidimensional space of conservative eigenmeasures for $\beta < \beta_c$ and the absence of conservative eigenmeasures for $\beta > \beta_c$. Eigenmeasures are examples of DLR measures \cite{Shwartz2019,BelBisEndo2020}, which sometimes coincide, sometimes not \cite{BelBisEndo2020,CioLopesStad2020}. On the generalized renewal shift, we have an example of similar behavior in $(b)$, since we have a unique probability eigenmeasure for every $\beta >0$, even with the pressure having a point which is non-analytic.  



The first connection between C$^*$-algebras and symbolic spaces was born in 1977 for the full shift case by J. Cuntz \cite{Cuntz1977}. It was generalized for Markov shifts in the case of finite symbols J. Cuntz and W. Krieger \cite{CK1980} three years later, where it was introduced the so-called Cuntz-Krieger algebras $\mathcal{O}_A$, with $A$ being the transition matrix. These algebras are universal C$^*$-algebras generated by a family of partial isometries indexed by the alphabet under some conditions. They codify the Markov shift spaces in the sense that a product of these generators is non-zero if and only if the word correspondent to the product is admissible. Moreover, $C(\Sigma_A)$ is a commutative C$^*$-subalgebra of $\mathcal{O}_A$. Remarkable results for the case row-finite matrix $A$ were achieved by A. Kumjian, D. Pask, I. Raeburn and J. Renault in \cite{KumjPaskRaebRen1997,KumjPaskRaeb1998} by using the Groupoid C$^*$-algebra theory (see \cite{Renault1980}). In 1999, R. Exel and M. Laca \cite{EL1999} generalized the Cuntz-Krieger theory for the case in which the matrix $A$ is not row-finite. For this general case, $\Sigma_A$ may not be locally compact. The generalized Markov shift $X_A$ arises as a natural object to replace $\Sigma_A$, and the C$^*$-subalgebra $C_0(X_A)$ takes the place of $C(\Sigma_A)$. The generalized setting includes $\Sigma_A$, and the dynamics can be extended almost to the whole space, excepting by the which we call empty-stem configurations.

Since then, even with the progress of the Thermodynamic Formalism on countable Markov shifts $\Sigma_A$ in the last 20 years, no study on Thermodynamic Formalism was made for $X_A$. This thesis aims to start the thermodynamic study of $X_A$ and relate it to what was developed so far for $\Sigma_A$. It holds that $\Sigma_A$ is dense in $X_A$, and for locally compact standard Markov shift, these sets coincide. The compatibility of Borel $\sigma$-algebras for both settings, in the sense that the Borel sets of $\Sigma_A$ are included in the Borel $\sigma$-algebra on $X_A$, plays a crucial role in order to allow that conformal measures from the standard theory can be seen as conformal ones in the generalized setting. Also, the topologies on each setting are partially similar. For instance, unlike in $\Sigma_A$ case, the complement of cylinders is not simply a union of cylinders, but they have an extra set that lives in $Y_A = \Sigma_A^c$. Also, $X_A$ is always locally compact, and in many cases, it is compact, even for $\Sigma_A$ not locally compact. On the other hand, the Gurevich pressure, which its exponential is the Ruelle's operator eigenvalue for suitable potentials (recurrent), only considers the periodic points, which cannot live in $Y_A$. However, by Denker-Yuri \cite{DenYu2015} results for Iterated Function Systems (IFS), it is possible to define the pressure at a point $P(\beta F,x)$, $x \in X_A$, and for a class of potentials presented in this thesis (section \ref{sec:pressures}), this pressure coincides with the Gurevich pressure on every point. Moreover, by Denker-Yuri's eigenmeasure existence theorem, we grant for every $\beta > 0$ that we have an eigenmeasure for the associated eigenvalue $e^{P(\beta F,x)}$. We studied phase transition phenomena for conformal measures and eigenmeasures, and by the set $X_A$, we can detect new of the aforementioned measures than the standard setting. Furthermore, we highlight a length-type phase transition phenomena: the existence of a critical value $\beta_c$ s.t. in the region $\beta \leq \beta_c$, we have the existence of a probability eigenmeasure that gives mass zero for $Y_A$, and for $\beta > \beta_c$, the eigenmeasure also exists, but vanishes in $\Sigma_A$. The term \textit{length-type} refers to that we can see the elements of $Y_A$ as finite words; sometimes, they have some multiplicity. For the case of conformal measures, we have similar results but with the standard measure (living on $\Sigma_A$) existing only for $\beta = \beta_c$ instead of $\beta \leq \beta_c$, and the absence of these measures for $\beta < \beta_c$.

The main contributions of this thesis are:
\begin{itemize}
    \item we prove the equivalence among several notions of conformal measure, including the standard and generalized settings, beyond the particular case of shift spaces (generalized or not). The result unifies Sarig's definition of conformal measure, quasi-invariant measures on groupoids, fixed point measures for Ruelle's operator, and others. It holds for locally compact Hausdorff second countable spaces;
    \item it is a general fact that for regular potentials $F$ and standard Markov shift spaces with finite alphabet, we have the existence of KMS$_\beta$ states only for a unique $\beta_c$, which corresponds to the value for which we have pressure zero and the potential $\beta F$ normalizes the Ruelle's operator ($L_{\beta F}(1) = 1$) \cite{BratJorgOstro2004,Exel2004,KerrPinz2002,KesStadStrat2007}. We give some examples which indicate that, at least for the class of renewal type shifts on the generalized setting, the behavior is different: we have an interval of values of $\beta$ for which we have the existence of KMS$_\beta$ states. In other words: for low temperatures, we have the existence of KMS$_\beta$-states (fixed points of the Ruelle's transformation). For the generalized renewal shift, this is true for any potential with a positive lower bound;
    \item for the renewal-type class of shift spaces and low temperatures, we have a bijection between the infinite emitters and the extremal KMS$_\beta$ states. In particular, for the pair renewal and prime renewal shifts, we have a phase transition on the number of KMS$_\beta$ states: for high temperatures, we have the absence of these states, and for temperature low enough, we have infintely many KMS$_\beta$ states. For the pair renewal shift, we calculated the critical value explicitly for the phase transition, namely $\beta_c = \log(1+\sqrt{2})$, and for this value, we have a unique KMS$_\beta$ state;  
    \item we study the weak convergence of measures on generalized countable Markov shifts. The study completes recent results about weak convergence of measures on the countable setting \cite{IommiVelozo2019}. The control of weak$^*$ convergence of the probability measures allows us to prove that the set of $e^{\beta F}$-conformal measures giving mass zero to $\Sigma_A$ (living on $Y_A$, for $\beta > \beta_c$) converges to a probability measure $\mu_{\beta_c}$ such that $\mu_{\beta_c}(\Sigma_A) = 1$;
    \item we introduced the concept of pressure on generalized renewal shifts by adapting a notion of Denker-Yuri for Iterated Function Systems. With this notion of pressure, we can produce examples where the pressure coincides with the standard Gurevich pressure on countable Markov shifts, and we found a new type of phase transition: the length-type;
    \item  we present an example of potential for the generalized renewal shift such that we have the existence of a unique eigenmeasure for each $\beta > 0$ but, after a critical value $\beta_c$, the measure gives mass $1$ to the set of finite words $Y_A$. For $\beta \geq \beta_c$, the Gurevich pressure is zero and therefore, in this case, the eigenmeasure is a fixed point of the Ruelle's transformation, and then it is associated to the a unique KMS$_\beta$ state. The the set of KMS$_\beta$ states is described for every $\beta >0$ since for $\beta < \beta_c$ we have the absence of KMS$_\beta$ states.
\end{itemize}


The first three chapters contain the general mathematical background and, they were written to make this thesis friendly to read for both Dynamical Systems and C$^*$-algebras communities. We briefly describe their contents as follows.

In chapter \ref{ch:Markov_shift_space}, we present the Markov shift spaces, which we denote by $\Sigma_A$ for a transition matrix $A$, and some of its properties, the different regularities for potentials in this theory and the Gurevich pressure, the notions of conformal measures in both senses of Denker-Urba\'nski and Sarig, the Ruelle's operator, and the eigenmeasures associated with it. Moreover, we also present essential notions and results from the Thermodynamic Formalism on Markov shifts, such as the equivalence between the aforementioned measures (Corollary \ref{cor:equivalences_conformality_classical}), the notion of recurrence for potentials, the generalized Ruelle-Perron-Frobenius Theorem, which is a result of the existence of eigenmeasures based on the recurrence of the potential; the Discriminant Theorem, which is a powerful result to decide the recurrence of the potential, and a result on the linearity of the pressure on the inverse of the temperature, after a critical value.

Chapter \ref{ch:C_star_algebras} contains the basics of C$^*$-algebras and a significant result for this thesis: the Gelfand's theorem for C$^*$-algebras. Also, in this chapter, we have an entire section dedicated to the construction of universal C$^*$-algebras and its properties, which are essential to define both Cuntz-Krieger and Exel-Laca algebras properly. Furthermore, we briefly introduce the notion of KMS-states, which are connected to conformal measures.

Chapter \ref{ch:Groupoid_algebras} presents topological groupoids, and its main objective is to construct the definition of a C$^*$-algebra of a groupoid and present some of its properties. One of the main topological features for groupoids we studied is the etalicity of a topological groupoid. We briefly present the notion of quasi-invariant measure, which connects the conformal measures to the KMS states via 1-cocycles associated to continuous potentials.

Now, we describe the chapters \ref{ch:CK_EL} and \ref{ch:TF_on_Generalized_Countable_Markov_shifts}, which are the main subject of this thesis, which comprehends the paper \cite{BEFR2018}.

In chapter \ref{ch:CK_EL}, we introduce the Cuntz-Krieger algebras and their generalizations, the Exel-Laca algebras, both denoted by $\mathcal{O}_A$. We show some of their properties. We compare each other, for instance, that the Exel-Laca algebras are, in fact, the universal algebras which generalize the Cuntz-Krieger algebras. Moreover, in the crucial part of this chapter, we present the generalized Markov shift spaces, denoted by $X_A$, which is the spectrum of a particular commutative C$^*$-subalgebra of $\mathcal{O}_A$. Therefore it is a locally compact space. This topological aspect is very useful for Topology and Measure Theory results, such as the Urysohn's lemma and the Riesz-Markov-Kakutani representation theorem. We present some important facts about this set, such as the density of the standard Markov shift in this space and the density of its complement $Y_A$, when this one is non-empty, and the compatibility of Borel sets between the standard and generalized shifts. It is also presented a characterization of a particular partition of $Y_A$ made of countable sets, which we call by $Y_A$-families. For the latter, we also explore the cylinder topology in this new setting. This chapter counts with three examples: the renewal shift, the pair renewal shift, and the prime renewal shift. These examples have the property that the standard shift space is not locally compact; however, the generalized one is compact. Moreover, the respective $Y_A$ sets in these cases are countable, in particular, with a finite number of infinite emitters for the renewal and pair renewal shifts, an infinite number of it for the prime renewal shift. We also present an example of non locally compact $\Sigma_A$ such that $X_A$ is not compact and such that we have uncountable elements in $Y_A$. At the end of the chapter, we have a section dedicated to present the Renault-Deaconu groupoid, a groupoid constructed from a Singly Generated Dynamical System (SGDS), and its essential connection between the Exel-Laca algebras and the Renault-Deaconu groupoid C$^*$-algebra of the SGDS $(X_A,\sigma)$, where $\sigma$ is the shift map: these C$^*$-algebras are isomorphic.

Chapter \ref{ch:TF_on_Generalized_Countable_Markov_shifts} comprehends the Thermodynamic Formalism on $X_A$. First, we define and discuss the weak$^*$ convergence of measures and state theorem \ref{thm:convergence_measures_Bogachev}, a well-known result of Measure Theory on weak$^*$ convergence with hypotheses on the basic sets which fits with our topology. Then, we introduce the notions of conformal measure in both senses of Denker-Urba\'nski and Sarig and the generalized version of the Ruelle operator, which we call Ruelle transformation, and jointly with the notion of quasi-invariant measure, we obtained an equivalence result among these notions (theorem \ref{thm:equivalences_conformal_measures_generalized_Markov_shift}), similar to Corollary \ref{cor:equivalences_conformality_classical}. Moreover, we proved that this theory is compatible with the standard one, in the sense that, in Sarig's sense, every conformal measure, when restricted to the Borel $\sigma$-algebra on $\Sigma_A$, is a conformal measure on $\Sigma_A$, and every conformal measure in the standard setting, when seen as a measure which vanishes out of $\Sigma_A$, is a conformal measure in the generalized setting. Furthermore, we proved that the extremal conformal (Denker-Urba\'nski sense) probability measures on $X_A$ that vanishes in $\Sigma_A$ are precisely those that live on a unique $Y_A$-family associated to an infinite emitter. Each one of these families may have at most one conformal probability living on them. We studied the existence of these new measures on  renewal, pair renewal, and prime renewal shift spaces and, in all these cases, we discovered extremal conformal measures beyond of where the standard formalism reaches. For instance, on the renewal shift case, let $U = \Dom \sigma \subsetneq X_A$ and a given potential $F$, we have the following results (Theorem \ref{theorem.potentialwithonecoordinate} and Corollary \ref{cor:phase_transition_conformal_potential_1}, respectively).
\begin{theorem*} For the generalized renewal shift, consider a potential $F:X_A\setminus\{\xi^0\} \to \mathbb{R}$ and $\beta>0$, we have the following:
\begin{itemize}
    \item[$(i)$] If $\inf F>0$, for $\beta>\frac{\log 2}{\inf F}$, there exists a unique $e^{\beta F}$-conformal probability measure $\mu_\beta$ that vanishes in $\Sigma_A$.
    \item[$(ii)$] If $0\leq \sup F < + \infty$ and $\beta\leq \frac{\log 2}{\sup F}$, there are no $e^{\beta F}$-conformal probability measures that vanish in $\Sigma_A$.
\end{itemize}
\end{theorem*}
\begin{corollary*}
For the generalized renewal shift, let $F\equiv1$. Then, for the constant $\beta_c=\log 2$, the result follows:
\begin{itemize}
    \item[$(i)$] For $\beta>\beta_c$ we have a unique $e^\beta$-conformal probability measure that vanishes on $\Sigma_A$.
    \item[$(ii)$] For $\beta\leq\beta_c$ there is no $e^\beta$- conformal probability measure that vanishes on $\Sigma_A$.
\end{itemize}
\end{corollary*}
In the corollary above, for $\beta_c = \log 2$, there is a unique conformal measure, which lives on $\Sigma_A$, and no other conformal measure living on $\Sigma_A$ for any other $\beta$. A $e^{\beta F}$-measures $\mu$ define a KMS$_\beta$-state $\varphi_\mu$ of the Renault-Deaconu full groupoid C$^*$-algebra $C^*(\mathcal{G}(X_A,\sigma)$ for a $1$-cocycle dynamics on the potential by setting
\begin{equation*}
    \varphi_\mu(f) = \int_X f(x,0,x) d\mu(x), \quad f \in C_c(\mathcal{G}(X_A,\sigma)).
\end{equation*}
For the pair renewal and prime renewal shift cases, we obtained similar results. After the aforementioned studies, we searched for phase transition phenomena for the eigenmeasure probabilities on the renewal shift. We discovered what we call length-type phase transition, which we explain here by one of the examples of this thesis: it consists in the change of space that the eigenmeasure lives after a critical value $\beta_c$; by decreasing the temperature, the measure once were living on $\Sigma_A$ passes to live on $Y_A$. We characterize all the probability eigenmeasures associated with the eigenvalue being the exponential of the Gurevich pressure for each $\beta$. In order to do this, we see $(X_A,\sigma)$ as an Interated Function System (IFS), and by a result of Denker and Yuri \cite{DenYu2015}, we grant the existence of eigenmeasures for the eigenvalue being the exponential of the pressure at a point, a different concept of pressure. For the class of potentials on the renewal shift that are bounded above and such that they can be written in the form
\begin{equation*}
    F(x) = g(x_0) - g(x_0 +1),
\end{equation*}
where $g$ is a continuous function and $x_0$ the first coordinate of the stem of the configuration, for every $\beta >0$ and every point $\xi \in X_A$, both Gurevich pressure and pressure at the point on $\beta F$ coincide. Then the eigenmeasures of the Denker-Yuri theorem are eigenmeasures associated to the aforementioned eingevalue on the Gurevich pressure. We obtain the following theorem:
\begin{theorem*} Let $A$ be the renewal shift transition matrix and $X_A$ its generalized Markov shift space. Consider the potential $F:U \to \mathbb{R}$ given by $$F(x)=\log(x_0)-\log(x_0+1).$$. Then, for every $\beta > 0$, there exists a unique eigenmeasure associated to the eigenvalue $\lambda_\beta = e^{P_G(\beta F)}$. Moreover, there is critical value $\beta_c$, which is the (real) solution for $\zeta(\beta_c) = 2$ such that
\begin{itemize}
    \item[$(i)$] if $\beta >\beta_c$, then the eigenmeasure lives on $Y_A$;
    \item[$(ii)$] if $\beta \leq \beta_c$, then the eigenmeasure lives on $\Sigma_A$.
\end{itemize}
\end{theorem*}
The theorem above shows the length-type phase transition and all the eigenmeasure are completely described.

\chapter{Thermodynamic Formalism on Countable Markov Shifts}
\label{ch:Markov_shift_space}

In the present chapter we introduce the classical notion of Markov shift space and some of its topological properties. In addition we present the concepts of eigenmeasures and conformal measures. This chapter is based mostly in the lecture notes \cite{Sarig2009}, the survey \cite{Sarig2015} and the PhD thesis \cite{Beltran2019}. Also, we are based on the papers \cite{Daon2013,DenUr1991,Iommi2007,Sarig1999,Sarig2001,Ruelle1976} and the book \cite{Walters2000}.

\section{Markov Shift Spaces}

We consider a countable alphabet $S = \{1,...,n\}$, $n \in \mathbb{N}$, or $S = \mathbb{N}$. The elements of $S$ are called \emph{letters} or \emph{symbols} and any finite sequence in $S$ is a \emph{word}. For any word $w$, we define its \emph{length} $|w|$ as being the number of its elements and we write $w = w_0 \cdots w_{|w|-1}$ or $w=w_0, ..., w_{|w|-1}$. The \emph{empty word} is the unique word on $S$ with length zero and it will be denoted by $\emptyset$. A \emph{transition matrix} $A$ over $S$ is a $\{0,1\}$ matrix with entries indexed by $S \times S$. A word $w$ is said to be \emph{admissible} if either $A(w_0,w_1) = \cdots = A(w_{|w|-2},w_{|w|-1}) = 1$ and $|w| \geq 2$ or $|w| \in {0,1}$.

\begin{definition}[Markov shift space]\label{def:Markov_shift} Given a countable alphabet $S$ and a transition matrix $A$, the \emph{Markov shift space} is the pair $(\Sigma_A,\sigma)$ where
\begin{equation*}
    \Sigma_A := \left\{x \in S^{\mathbb{N}_0}: A(x_j,x_{j+1}) = 1, \text{ for all } j \in \mathbb{N}_0\right\},
\end{equation*}
where $\mathbb{N}_0 = \mathbb{N}\cup\{0\}$, endowed with the topology generated by the metric $d:\Sigma_A \times \Sigma_A \to [0,1]$ defined by
\begin{equation*}
    d(x,y) = 2^{-\inf\{p:x_p \neq y_p\}};
\end{equation*}
and $\sigma: \Sigma_A \to \Sigma_A$ is the \emph{shift map}, given by
\begin{equation*}
    x = x_0 x_1 x_2 \cdots \mapsto \sigma(x) = x_1 x_2 \cdots
\end{equation*}
\end{definition}

\begin{remark} Many authors also call a Markov shift space as defined above as an \emph{one-sided Markov shift space}, since there is a similar construction for the \emph{two-sided Markov shift}, in which case we write 
\begin{equation*}
    \Sigma_A := \left\{x \in S^{\mathbb{Z}}: A(x_j,x_{j+1}) = 1, \text{ for all } j \in \mathbb{Z}\right\}.
\end{equation*}
In the present work we shall restrict ourselves to one-sided Markov shift spaces.
\end{remark}

\begin{remark} For every transition matrix $A$, $\Sigma_A$ is a closed subset of $S^{\mathbb{N}_0}$ and the map $\sigma$ is continuous. 
\end{remark}

There is a natural way to represent a Markov shift by directed graphs: take $S$ as the set of vertices and the set of edges as $\{(i,j) \in S\times S: A(i,j)=1\}$. In other words, $A(i,j) = 1$ if and only if there is an edge from $i$ to $j$. The Markov shift space can be seen as the set of infinite paths in this graph that start in some symbol. We call such graphs \emph{symbolic graphs}. The example below illustrates this construction.

\begin{example} Consider the Markov shift for $S = \{1,2,3,4\}$ and the transition matrix
\begin{equation}
 A = \begin{pmatrix}
        1 & 1 & 1 & 1\\
        1 & 0 & 0 & 1\\
        0 & 1 & 0 & 0\\
        0 & 1 & 1 & 0
     \end{pmatrix}.
\end{equation}
The graph associated to $\Sigma_A$ (or equivalently, to $A$) is the figure below.
\[
\begin{tikzcd}
\circled{1}\arrow[loop left] \arrow[r,bend left] \arrow[d] \arrow[rd]& \circled{2} \arrow[l] \arrow[d]\\
\circled{3}\arrow[ru] & \circled{4} \arrow[u,bend right] \arrow[l]
\end{tikzcd}
\]
\end{example}

In this thesis we will refer to a Markov shift space simply by $\Sigma_A$, and the reader will be warned about any features of the alphabet and the matrix when necessary. 

\begin{definition}[Cylinder sets] Consider a Markov shift space $\Sigma_A$. For any non-empty word $w=w_0 \cdots w_{|w|-1}$. The set $[w]:=\{(x_i)_{i\in \mathbb{N}_0}\in \Sigma_A: x_k=w_k; k = 0,\dots, |w|-1 \}$ is said to be a \emph{cylinder set} of $\Sigma_A$.
\end{definition}

\begin{remark} Every cylinder set is clopen and the family of all cylinder sets forms a topological basis. Indeed, we have
\begin{equation*}
    \Sigma_A = \bigcup_{i \in S}[i],
\end{equation*}
and for every two admissible words $w$ and $v$ it follows that
\begin{equation*}
    [w] \cap [v] = \begin{cases}
                        [w], \quad \text{if $w$ is a sub-word of $v$};\\
                        [v], \quad \text{if $v$ is a sub-word of $w$};\\
                        \emptyset, \quad \text{otherwise}.
                   \end{cases}
\end{equation*}
\end{remark}

Note that if $w$ is not admissible, then $[w] = \emptyset$. In addition, for $i \in S$, if $A$ has all the entries of the $i$-th row being zero, we have $[i] = \emptyset$. On the other hand, if $A$ has zero columns, say the $i$-th column, we have that $\sigma^{-n}([i]) = \emptyset$ and then $\Sigma_A$ has dynamical drawbacks. In order to avoid these problems we assume for the rest of this chapter the following.

\begin{mdframed} \textbf{Standing hypothesis:} for any Markov shift space the transition matrix has no zero rows or columns.
\end{mdframed}

We are basically interested in two special classes of Markov shifts which we define next.

\begin{definition}[Transitivity and Mixing properties] We say that a Markov shift $\Sigma_A$ is
\begin{itemize}
 \item[$(i)$] \emph{transitive} (or \emph{topologically transitive}) if for every pair of symbols $i,j \in S$ there exists an admissible finite word $w$ such that $iwj$ is admissible.
 \item[$(ii)$] \emph{topologically mixing} if for every pair of symbols $i,j \in S$ there exists $N=N_{i,j} \in \mathbb{N}$ such that, for every $n > N$, there is an admissible word $w$, $|w| = n$, such that $iwj$ is admissible.
\end{itemize}
In addition, we also say that $A$ is transitive (or topologically mixing) when $\Sigma_A$ is transitive (or topologically mixing). 
\end{definition}



Most of the studied Markov shifts are transitive and we depict some of them next.

\begin{example}[Full shift]\label{exa:full_shift} The full shift consists of $\Sigma_A$ without transition restrictions: $A(i,j) = 1$ for every $i,j \in S$, that is $\Sigma_A = S^{\mathbb{N}_0}$. Its symbolic graph is the complete directed graph on $S$. This shift space has the property that every word is admissible, and hence any word can be inserted between any two symbols. Therefore it is clear that the full shift is topologically mixing. 
\end{example}

\begin{example}[Renewal shift]\label{exa:renewal_shift} The renewal shift is a Markov shift space for the alphabet $\mathbb{N}$ and the following transition matrix:
\begin{equation*}
 A(i,j) = \begin{cases}
                1, \quad \text{if } j=i+1 \text{ or } i = 1;\\
                0, \quad \text{otherwise}.
          \end{cases}
\end{equation*}
Its symbolic graph is the following.
\[
\begin{tikzcd}
\circled{1}\arrow[loop left]\arrow[r,bend left]\arrow[rr,bend left]\arrow[rrr,bend left]\arrow[rrrr, bend left]&\circled{2}\arrow[l]&\circled{3}\arrow[l]&\circled{4}\arrow[l]&\arrow[l]\cdots
\end{tikzcd}
\]
In addition, the renewal shift is topologically mixing. Indeed, given two symbols $i,j \in \mathbb{N}$, it holds that
\begin{itemize}
 \item if $i =1$, then $i1^kj$ is admissible for every $k \in \mathbb{N}_0$, where
 \begin{equation*}
  1^k = \underbrace{111\cdots111}_{k \text{times}};
 \end{equation*}
 \item if $i \neq 1$, then $i,i-1,...,1^kj$ is admissible for every $k \in \mathbb{N}_0$.
\end{itemize}
Consequently, the topological mixing property holds. 
\end{example}

\begin{example} Consider $\Sigma_A$ for $S = \mathbb{N}$ and
\begin{equation*}
 A(i,j) = \begin{cases}
                1, \quad \text{if } $j=i+1$ \text{ or if } i = 1 \text{ and } j \text{ is even};\\
                0, \quad \text{otherwise},
          \end{cases}
\end{equation*}
that is,
\begin{equation*}
    A=\begin{pmatrix}
       0 & 1 & 0 & 1 & 0 & \cdots\\
       1 & 0 & 0 & 0 & 0 & \cdots\\
       0 & 1 & 0 & 0 & 0 & \cdots\\
       0 & 0 & 1 & 0 & 0 & \cdots\\
       \vdots & \vdots & \vdots & \vdots & \vdots &\ddots
      \end{pmatrix}.
\end{equation*}
The symbolic graph of $A$ is the following.
\[
\begin{tikzcd}
\circled{1}\arrow[r,bend left]\arrow[rrr,bend left]\arrow[rrrrr, bend left]&\circled{2}\arrow[l]&\circled{3}\arrow[l]&\circled{4}\arrow[l]&\circled{5}\arrow[l]&\arrow[l]\cdots
\end{tikzcd}
\]
We claim $A$ is transitive. Indeed, for any $i,j \in \mathbb{N}$, we may take the path $i,i-1,...,1,2j,2j-1,...,j$, and therefore the claim is proved. However, $\Sigma_A$ is not topologically mixing. If the word $1w2$ is admissible, then $w$ has even length. 
\end{example}

Now, we present some important topological properties of the Markov shift space.

\begin{proposition} The metric space $(\Sigma_A,d)$ is complete and separable.
\end{proposition}

\begin{proof} \textbf{Completeness:} Let $(x^n)_\mathbb{N}$ be a Cauchy sequence in $\Sigma_A$. Then, for every $k \in \mathbb{N}$ there exists $N = N_k \in \mathbb{N}$ such that for any $n,m> N$ we have
\begin{equation*}
    d(x^n,x^m) = 2^{-\inf\{p\in \mathbb{N}_0: x^n_p\neq x^m_p \}} < 2^{-k},
\end{equation*}
and hence for $x^n_p = x^m_p$ for $n,m>N$ and $p<k$. Then, for every $p_0 \in \mathbb{N}_0$, the sequences $(x_p^n)_{n \in \mathbb{N}}$ are eventually constant for each $p \leq p_0$, that is, the limit $\lim x_p^n$ exists for each $p \leq p_0$. Take $x\in S^{\mathbb{N}_0}$ such that $x_p = \lim x^n_p$. It is straightforward that $x \in \Sigma_A$. We claim that $\lim x^n = x$. Indeed, given $k \in \mathbb{N}$ there exists $N = N_k \in \mathbb{N}$ such that $x_p^n = x_p$ for every $p < k$ and $n > N$, hence 
\begin{equation*}
    d(x^n,x) = 2^{-\inf\{p\in \mathbb{N}_0: x^n_p\neq x_p \}} < 2^{-k},
\end{equation*}
and therefore the metric space $(\Sigma_A,d)$ is complete.

\textbf{Separability:} the family of all cylinder sets forms a countable basis for the topology, therefore $\Sigma_A$ is separable.
\end{proof}

\begin{remark} The proposition above holds independently of any hypotheses on the transition matrix. 
\end{remark}

\begin{proposition}\label{prop:compactness_Markov} Consider a countable alphabet $S$ and any transition matrix $A$ over $S$. Then $\Sigma_A$ corresponding to the alphabet $S$ and the matrix $A$ is compact if and only if $S$ is finite.
\end{proposition}
\begin{proof} The discrete topology on $S$ is compact if and only if $S$ is finite. Also, the product topology on $S^{\mathbb{N}_0}$ is compact if and only if the topology on each coordinate space is compact. Since $\Sigma_A$ is a closed subset of $S^{\mathbb{N}_0}$, we conclude that if $S$ is finite then $\Sigma_A$ is compact. For the inverse implication we must observe that, since $A$ has not zero rows, the coordinate projection $\pi_0:\Sigma_A \to S$, given by
\begin{equation*}
\pi_0(x) = x_0,
\end{equation*}
satisfies
\begin{equation*}
\pi_0^{-1}(\{i\}) \neq \emptyset, \quad \text{for all } i \in S,
\end{equation*}
and therefore
\begin{equation*}
\pi_0(\Sigma_A) = S.
\end{equation*}
If $S$ is countably infinite, then $S$ is not compact. Suppose that $\Sigma_A$ is compact. by the continuity of the projection maps, we would have that $\pi_0(\Sigma_A) = S$ is compact, a contradiction.
\end{proof}

The last proposition describes precisely what are the compact Markov shifts, namely, those related to a finite alphabet. Although these spaces are not compact for the infinitely countable alphabet case, they might be locally compact. Proposition \ref{prop:local_compactness_Markov} will show that the Markov shift spaces which are locally compact are precisely those whose transition matrix is \emph{row-finite}, that is, for every $i \in S$ we have that the set $\{j \in S: A(i,j) =1\}$ is finite. In order to prove the aformentioned proposition, we need the following definition. Given a transition matrix $A: S \times S \to \{0,1\}$, $k \in \mathbb{N}_0$ and $i \in S$, let
\begin{align*}
    A^k(i)&:= \begin{cases}
                \{i\}, \quad \text{if } k=0;\\
                \{j \in \mathbf{N}:A(i,j) =1\} , \quad \text{if } k=1.
              \end{cases}
\end{align*}
If $k \geq 2$ we shall define $A^k(i)$ as follows:
\begin{align*}
    p \in A^k(i) &\iff \exists p_1,...,p_{k-1} \in \mathbf{N}: A(i,p_1) \left(\prod_{j = 1}^{k-2}A(p_j,p_{j+1})\right)A(p_{k-1},p) = 1, \quad\text{if} \quad k>2,\\
    p \in A^2(i) &\iff \exists p_1 \in \mathbf{N}: A(i,p_1)A(p_1,p)=1.
\end{align*}

\begin{proposition}\label{prop:local_compactness_Markov} Consider a Markov shift space $\Sigma_A$ with $A$ transitive. The following are equivalent:
\begin{itemize}
    \item[$(i)$] $A$ is row-finite, that is, for each line of the transition matrix there is a finite number of $1$'s;
    \item[$(ii)$] the cylinder sets of $\Sigma_A$ are compact;
    \item[$(iii)$] $\Sigma_A$ is locally compact.
\end{itemize}
\end{proposition}
\begin{proof} We prove the following: $(i)$ implies $(ii)$, $(ii)$ implies $(iii)$, and $(iii)$ implies $(i)$.

\textbf{$(i)$ implies $(ii)$:} it is straightforward to notice that
\begin{equation}\label{eq:cylinder_contained_prod_A_k}
    [i] \subseteq \{i\} \times A(i) \times A^2(i) \times \cdots = \prod_{j=0}^\infty A^{j}(i)
\end{equation}
for every $i \in \mathbb{N}$. Since $A$ is row-finite, the sets $A^k(i)$ are finite for every $k \in \mathbb{N}$ and therefore they are compact sets in the discrete topology of $\mathbb{N}$. By Tychonov's Theorem, the RHS of \eqref{eq:cylinder_contained_prod_A_k} is compact, and since $[i]$ is a closed subset of it, we conclude that $[i]$ is compact. Moreover, for any cylinder set $[w]$, $w=w_o \cdots w_{n-1}$, $n \in \mathbb{N}$, we have that $[w]$ is a closed subset of the compact set $[w_0]$, and therefore $[w]$ is compact due to the Hausdorff property of $\Sigma_A$.


\textbf{$(ii)$ implies $(iii)$:} for any $x \in \Sigma_A$ we have that $x \in [x_0]$, then $[x_0]$ is a compact neighborhood of $x$ and therefore $\Sigma_A$ is locally compact.

\textbf{$(iii)$ implies $(i)$:} suppose $A$ is not row-finite. Then there exists $i \in S$ such that $|A(i)| = \infty$. By local compactness and Hausdorff property, for every $x \in \Sigma_A$ there exists an open neighborhood of $x$ such that $\overline{U}$ is compact. We may write $U$ as a union of the basic sets
\begin{equation*}
    U = \bigcup_{w \in L} [\underline{w}],
\end{equation*}
where $L$ is a family of admissible words. Hence, there exists an admissible word $v\in L$ such that $x \in [v] \subseteq U \subseteq \overline{U}$, where we conclude that $[v]$ is compact. The transitivity implies that there exists a finite admissible word $u$ such that the cylinder $[vui] \subseteq [v]$ is non-empty, and hence $[vui]$ is a closed non-compact subset of $[v]$ because its open cover $\{[vuij]:j \in A(i)\}$ does not admit a finite subcover. By the Hausdorff property we conlude $[v]$ is not compact because it has a non-compact closed subset, a contradiction.
\end{proof}

From now on, given a topological space $X$, $\mathcal{B}_X$ will always denote the Borel $\sigma$-algebra on $X$.

\section{Potentials and Gurevich pressure}

We call a measurable function $f:\Sigma_A\to\mathbb{R}$ with respect to the Borel $\sigma$-algebra $\mathcal{B}_{\Sigma_A}$ a potential. It defines a model on a classical Markov shift space. In order to maintain the relation between dynamics and C$^*$-algebras well-defined, potential functions will be, henceforth, supposed not only measurable but also continuous.

Potentials can be classified by their regularity and here we present some of the main classes. And for our purposes, we study them under following assumption.  

\begin{mdframed} \textbf{Standing hypothesis:} until the end of this chapter, every Markov shift space considered is transitive.
\end{mdframed}

We introduce the notions of Birkhoff sum and variations of a potential. These concepts will be used to define some of the potential classes. In particular, the Birkhoff sum will be used to contruct some of the conformal measures for generalized Markov shift spaces which are not detected in the standard setting presented in this chapter. In this section, the alphabet considered is $\mathbb{N}$.

\begin{definition}[Birkhoff sum] Given a continuous potential $f:\Sigma_A \to \mathbb{R}$ and $n \in \mathbb{N}$, the \emph{$n$-th Birkhoff sum} of $f$ is the function
\begin{equation*}
    f_n(x):= \sum_{i=0}^{n-1} f\circ \sigma^i(x).
\end{equation*}
\end{definition}

\begin{definition}[Variation] Given a continuous potential $f:\Sigma_A \to \mathbb{R}$ and $n \in \mathbb{N}$, the \emph{$n$-th variation} of $f$ is the quantity
\begin{equation*}
    \Var_n f := \sup\{|f(x)-f(y)|: x_i = y_i, i \in \{0,...,n-1\}\}.
\end{equation*}
\end{definition}

\begin{definition}\label{def:regularity_potentials} Given a potential $f: \Sigma_A \to \mathbb{R}$, we say that $f$
\begin{itemize}
    \item is \emph{summable (or exp-summable)} if 
    \begin{equation}\label{eq:summable_potential}
        \sum_{i \in \mathbb{N}} e^{\sup f\vert_{[i]}}<\infty;
    \end{equation}
    \item is \emph{locally H\"older}\footnote{In the literature this property is also called \emph{weakly H\"{o}lder}.} if there exists a constant $H_f>0$ and $r \in (0,1)$ such that for all $k \geq 2$ natural number, we have
    \begin{equation*}
        \Var_k f \leq H_f r^k;
    \end{equation*}
    \item \emph{has summable variations} if 
    \begin{equation*}
        \sum_{k \in \mathbb{N}} \Var_k f<\infty;
    \end{equation*}
    \item \emph{satisfies Walters' condition \cite{Walters1978}} if
    \begin{equation*}
        \sup_{n \in \mathbb{N}} \Var_{n+k}f_n < \infty \quad \text{for each }k \in \mathbb{N} \quad \text{and} \quad \lim_k \sup_{n \in \mathbb{N}} \Var_{n+k}f_n = 0.
    \end{equation*}
\end{itemize}
\end{definition}
 
Note that summable potentials cannot be bounded below, for the series \eqref{eq:summable_potential} would not converge otherwise. Furthermore, any locally H\"older potential has summable variations and every potential having summable variations satisfies the Walters' condition. 

The next lemma presents sufficient conditions to bound Birkhoff sums of potentials satisfying Walters' condition over cylinders.

\begin{lemma}\label{lemma:Walters_var_1_finite_implies_bounded_diference_on_all_cilynders} Let $f:\Sigma_A \to \mathbb{R}$ be a potential satisfying Walters' condition with $\Var_1 f < \infty$. There exists $M>0$ s.t. for every cylinder set $[w]$ with $|w|>0$, we have
\begin{equation*}
    |f_{|w|}(x) - f_{|w|}(y)| \leq M,
\end{equation*}
for every $x,y \in [w]$.
\end{lemma}

\begin{proof} Let $C:= \sup_{n\geq 1} \Var_{n+1} f_n + \Var_1 f$. Note that $C$ is finite by hypotheses. Consider a cylinder $[w]$ s.t. $|w| > 0$. For any $x,y \in [w]$, we have two cases:
\begin{itemize}
    \item if $|w| = 1$, then
    \begin{equation*}
        |f(x) - f(y)| \leq \sup_{x,y \in [w]}\left|f(y) - f(x)\right| \leq \Var_{1}f \leq C;
    \end{equation*}
    \item if $|w| >1$, then
    \begin{align*}
        |f_{|w|}(y) - f_{|w|}(x)| &\leq |f_{|w|-1}(y) - f_{|w|-1}(x)| + |f(\sigma^{|w|-1}(x)) - f(\sigma^{|w|-1}(y))| \\ 
        &\leq \Var_{|w|} f_{|w|-1} + \Var_{1}f  \leq C.
    \end{align*}
\end{itemize}
\end{proof}

\begin{lemma} \label{lemma:Walters_bounded_var_n_plus_1_f_n} Let $f:\Sigma_A \to \mathbb{R}$ be a potential satisfying Walters' condition. Then for every $n \in \mathbb{N}$ we have $\sup_{[w]} |f_n| < \infty$, where $w$ is an admissible word of length $n+1$. 
\end{lemma}

\begin{proof} Take $y = wy' \in [w]$. Then
\begin{align*}
    \sup_{[w]} f_n &= \sup_{x \in [w]} |f_n(x) - f_n(y) + f_n(y)| \leq \sup_{x \in [w]} |f_n(x) - f_n(y)| + |f_n(y)| \\
    &\leq \Var_{n+1} f_n + |f_n(y)| < \infty.
\end{align*}
Since $|f_n(y)|$ is constant, we conclude that $\sup_{[w]} f_n$ is finite. 
\end{proof}


In order to define the Gurevich pressure, we define partition functions on $\Sigma_A$ as follows. Let $\Sigma_A$ be topologically mixing and let $f:\Sigma_A \to \mathbb{R}$ be a potential. For each $n \in \mathbb{N}$ and $a \in S$, we define 
\begin{equation}
	Z_n(f,[a]) := \sum_{\sigma^n x = x} e^{f_n (x)} \mathbbm{1}_{[a]}(x).
\end{equation}

\begin{lemma}\label{lemma:Z_n_b_bounded_Z_n_a} Suppose a potential $f$ satisfies Walters' condition, then for each $a,b \in S$, there exists constants $C_1,C_2 > 0$ and $k_1,k_2 \in \mathbb{N}_0$ such that
\begin{equation}\label{eq:partition_function_sandwiched}
 C_1 Z_{n-k_1}(f,[a]) \leq Z_n(f,[b]) \leq C_2 Z_{n+k_2}(f,[a])
\end{equation}
for all $n > k_1$.
\end{lemma}

\begin{proof} Fix $a,b \in S$. By transitivity, there exist two admissible words $w$ and $v$ with $|w|>0$ and $|v|>0$, and such that $aub$ and $bva$ admissible. Let $k = |u|+|v|+2$, define the bijective map
\begin{align*}
    \psi:\{x \in [b]: \sigma^n x = x\} &\to \{x \in [a]: \sigma^{n+k} x = x\},\\
    x &\mapsto (aux_0\cdots x_nv)^\infty,
\end{align*}
where $(aux_0\cdots x_nv)^\infty$ is the periodic word consisting of the repetition of the word $aux_0\cdots x_nv$. By the Walters property we get for every $x \in [b]$, $\sigma^n x = x$, that
\begin{align*}
    \left|f_{n+k}(\psi(x)) - f_n(x)\right| &= \left|\sum_{i=0}^{n+k-1}f\circ \sigma^i((aux_0\cdots x_nv)^\infty) - f_n(x)\right| \\
    &= \left|\sum_{i=0}^{|u|} f\circ \sigma^i((aux_0\cdots x_nv)^\infty) + \sum_{i=|u|+1}^{n+k-1}f\circ \sigma^i((aux_0\cdots x_nv)^\infty) - f_n(x)\right|\\
    &= \left| f_{|u|+1}((aux_0\cdots x_nv)^\infty) + \sum_{i=0}^{n+|v|}f\circ \sigma^i \circ \sigma^{|u|+1}((aux_0\cdots x_nv)^\infty) - f_n(x)\right|\\
    &\leq \sup_{[aub]}\left|f_{|u|+1}\right| + \left| \sum_{i=0}^{n+|v|}f\circ \sigma^i ((x_0\cdots x_nvau)^\infty) - f_n(x)\right|\\
    &\stackrel{|v|>0}{\leq} \sup_{[aub]}\left|f_{|u|+1}\right| + \left| \sum_{i=n}^{n+|v|}f\circ \sigma^i ((x_0\cdots x_nvau)^\infty)\right|\\
    &+\left|\sum_{i=0}^{n-1}f\circ \sigma^i ((x_0\cdots x_nvau)^\infty)- f_n(x)\right|\\
\end{align*}
\begin{align*}
    &\leq \sup_{[aub]}\left|f_{|u|+1}\right| + \left| \sum_{i=0}^{|v|}f\circ \sigma^i \circ \sigma^{n} ((x_0\cdots x_nvau)^\infty)\right|\\
    &+\left|\sum_{i=0}^{n-1}f\circ \sigma^i ((x_0\cdots x_nvau)^\infty)- f_n(x)\right|\\
    &\leq \sup_{[aub]}\left|f_{|u|+1}\right| + \left| \sum_{i=0}^{|v|}f\circ \sigma^i ((bvaux_0\cdots x_{n-1})^\infty)\right| + \sup_n[\Var_{n+1}f_n]\\
    &\leq \sup_{[aub]}\left|f_{|u|+1}\right| + \sup_{[bva]}\left|f_{|v|+1}\right| + \sup_n[\Var_{n+1}f_n].\\
\end{align*}
Define
\begin{equation*}
    \log M:= \sup_{[aub]}\left|f_{|u|+1}\right| + \sup_{[bva]}\left|f_{|v|+1}\right| + \sup_n[\Var_{n+1}f_n],
\end{equation*}
and notice that $\log M < \infty$. Indeed, $\sup_n[\Var_{n+1}f_n] < \infty$ due to the Walters' condition. Also, $\sup_{[aub]}\left|f_{|u|+1}\right|< \infty$ and $\sup_{[bva]}\left|f_{|v|+1}\right| < \infty$ hold because of Lemma \ref{lemma:Walters_bounded_var_n_plus_1_f_n}. It is straightforward to observe that $M = M(a,b,v,u)$ does not depend on $n$. In particular we have
\begin{equation*}
    f_{n+k}(\psi(x)) \geq f_n(x) - \log M,
\end{equation*}
and since $\psi$ is one-to-one we get
\begin{align*}
    Z_n(f,[b]) &= M \sum_{\sigma^n x = x} e^{f_n(x)-\log M} \mathbbm{1}_{[b]}(x) \leq M \sum_{\sigma^n x = x} e^{f_{n+k}(\psi(x))} \mathbbm{1}_{[a]}(\psi(x)) \\
    &= M \sum_{\sigma^{n+k} x = x} e^{f_{n+k}(x)} \mathbbm{1}_{[a]}(x) = M Z_{n+k}(f,[a]),
\end{align*}
that is
\begin{align*}
    Z_n(f,[b]) \leq  M Z_{n+k}(f,[a]),
\end{align*}
for every $n \in \mathbb{N}$. Conversely there exist $M'>0$ and $k' \in \mathbb{N}_0$ such that
\begin{align*}
    Z_n(f,[a]) \leq  M' Z_{n+k'}(f,[b]),
\end{align*}
that is,
\begin{align*}
    (M')^{-1}Z_{m-k'}(f,[a]) \leq Z_{m}(f,[b]),
\end{align*}
for all $m > k'$. By taking $C_1 = (M')^{-1}$, $C_2 = M$, $k_1 = k'$ and $k_2 = k$ we get
\begin{equation*}
  C_1 Z_{n-k_1}(f,[a]) \leq Z_n(f,[b]) \leq C_2 Z_{n+k_2}(f,[a])
\end{equation*}
for every $n > k_1$.
\end{proof}

\begin{proposition} Let $\Sigma_A$ be topologically mixing and $f:\Sigma_A \to \mathbb{R}$ be a potential satisfying the Walters' condition. For every $a \in S$, the limit
\begin{equation}
    \lim_n \frac{1}{n}\log Z_n(f,[a])
\end{equation}
exists in the extended real numbers set and it does not depend on $a$. Moreover, such limit is never $-\infty$. 
\end{proposition}

\begin{proof} Let $a \in S$. Since $\Sigma_A$ is topologically mixing, there exists $n_0 \in \mathbb{N}$ such that for every $n \geq n_0$ there exists an admissible word $w$ such that $|w|= n+1$ and $w_0 = w_{|w|-1} = a$, and then $Z_n(f,[a]) > 0$ for every $n\geq n_0$. Hence, the sequence $(\zeta_n)_{n\geq n_0}$, where $\zeta_n := \log Z_n(f,[a])$ is well-defined. We prove that such sequence is almost super-additive, that is, there exists some $c \in \mathbb{R}$ satisfying
\begin{equation}\label{eq:sequence_almost_super_additivity}
    \zeta_{n+m} \geq \zeta_n + \zeta_m - c,
\end{equation}
for every $n,m \geq n_0$. In fact,
\begin{align*}
    Z_n(f,[a])Z_m(f,[a]) &= \sum_{\substack{x \in [a]: \\ \sigma^n x = x}} \sum_{\substack{y \in [a]: \\ \sigma^m y = y}} e^{f_n (x)+f_m (y)}. 
\end{align*}
Now, given $x \in [a]$ $n$-periodic word, that is, $x$ satisfies $\sigma^n x = x$, and $y \in [a]$ $m$-periodic word, we may write $x = uuu\cdots := u^\infty$ and $y = vvv\cdots := v^\infty$, where $u = u(x)$ and $v = v(y)$ are finite admissible words satisfying $|u| = n$, $|v| = m$ and $u_0 = v_0 = a$. So there exists the $n+m$-periodic point $z^{x,y} \in [a]$ given by $z = uvuv \cdots := (uv)^\infty$. Then,
\begin{align*}
    Z_n(f,[a])Z_m(f,[a]) &= \sum_{\substack{x \in [a]: \\ \sigma^n x = x}} \sum_{\substack{y \in [a]: \\ \sigma^m y = y}} e^{f_n (x)+ f_m (y)-f_{n+m} (z^{x,y})+f_{n+m}(z^{x,y})}\\
    &\leq \sum_{\substack{x \in [a]: \\ \sigma^n x = x}} \sum_{\substack{y \in [a]: \\ \sigma^m y = y}} e^{|f_n (x)+ f_m (y)-f_{n+m} (z^{x,y})|}e^{f_{n+m}(z^{x,y})}. 
\end{align*}
We observe that
\begin{align*}
    |f_n (x)+ f_m (y)-f_{n+m} (z^{x,y})| &= \left|f_n (u^\infty)+ f_m (v^\infty)-  f_{n+m} ((uv)^\infty)\right| \\
                                         &= \left|f_n (u^\infty)-f_n ((uv)^\infty) + f_m (v^\infty)-  f_m ((vu)^\infty)\right|\\
                                         &\leq  \left|f_n (u^\infty)-f_n ((uv)^\infty) \right| + \left|f_m (v^\infty)-  f_m ((vu)^\infty)\right|\\
                                         &\leq \Var_{n+1} f_n + \Var_{m+1} f_m \leq c,
\end{align*}
where $c = 2 \sup_{n \in \mathbb{N}}\Var_{n+1} f_n < \infty$, due to the Walters' condition. So,
\begin{align*}
    Z_n(f,[a])Z_m(f,[a]) &\leq e^c \sum_{\substack{x \in [a]: \\ \sigma^n x = x}} \sum_{\substack{y \in [a]: \\ \sigma^m y = y}} e^{f_{n+m}(z^{x,y})} \leq e^c \sum_{\substack{z \in [a]: z = uv, \\ \sigma^n u^\infty = u^\infty\\ \sigma^m v^\infty = v^\infty}} e^{f_{n+m}(z)}\\
    &\leq e^c \sum_{\substack{z \in [a]: \\ \sigma^{n+m} z = z}} e^{f_{n+m}(z)} = e^c Z_{n+m}(f,[a]), 
\end{align*}
that is
\begin{align*}
    Z_n(f,[a])Z_m(f,[a]) &\leq e^c Z_{n+m}(f,[a]). 
\end{align*}
By applying $\log$ in the inequality above we obtain \eqref{eq:sequence_almost_super_additivity}. There are two possibilities: either $\zeta_n = \infty$ for some $m \geq n_0$ or $- \infty < \zeta_n < \infty$ for every $n \geq n_0$. In the first case, the almost supper-additivity implies that $\lim \frac{\zeta_n}{n} = \infty$. Now, suppose the second case and fix $m > n_0$. For every $n > m$ we may write $n = q_n m + r_n$, where $q_n \in \mathbb{N}$ and $r_n \in \{0,...,m-1\}$. By applying the almost supper-additivity $q_n - 1$ times, we have
\begin{align*}
    \frac{\zeta_n}{n} &= \frac{\zeta_{q_n m + r_n}}{n} \geq \frac{q_n \zeta_m + \zeta_{r_n} - q_nc}{q_n m + r_n}\geq \frac{q_n \zeta_m + \zeta_{r_n} - q_nc}{p_n m}
\end{align*}
where 
\begin{equation*}
    p_n = \begin{cases}
            q_n + 1, \quad \text{if } q_n \zeta_m + \zeta_{r_n} - q_nc \geq 0;\\
            q_n - 1, \quad \text{otherwise}.
          \end{cases}
\end{equation*}
Then,
\begin{align*}
    \frac{\zeta_n}{n} &\geq \frac{q_n \zeta_m + \zeta_{r_n} - q_nc}{p_n m} = \left(\frac{q_n}{p_n}\right) \frac{\zeta_m}{m} + \frac{\zeta_{r_n}}{p_n m} - \left(\frac{q_n}{p_n}\right)\frac{c}{m}. 
\end{align*}
By taking the limit inferior on $n$ in the inequality above, we have
\begin{equation*}
    \liminf_n \frac{\zeta_n}{n} \geq \frac{\zeta_m}{m} - \frac{c}{m},
\end{equation*}
and then
\begin{equation*}
    \liminf_n \frac{\zeta_n}{n} \geq \sup_{m\geq n_0}\left(\frac{\zeta_m}{m} - \frac{c}{m}\right) \geq \limsup_m\left(\frac{\zeta_m}{m} - \frac{c}{m}\right) = \limsup_m\frac{\zeta_m}{m}.
\end{equation*}
Then, $\frac{\zeta_n}{n}$ converges. Now we prove that its limit does not depends on the choice of $a$. Indeed, let $b \in S$. By Lemma \ref{lemma:Z_n_b_bounded_Z_n_a}, there exists constants $C_1,C_2 > 0$ and $k_1,k_2 \in \mathbb{N}_0$ such that
\begin{equation*}
    C_1 Z_{n-k_1}(f,[a]) \leq Z_n(f,[b]) \leq C_2 Z_{n+k_2}(f,[a])
\end{equation*}
for all $n > k_1$. By applying $\log$ in the inequalities above we obtain.
\begin{equation*}
    \log C_1 + \log Z_{n-k_1}(f,[a]) \leq \log Z_n(f,[b]) \leq \log C_2 + \log Z_{n+k_2}(f,[a]).
\end{equation*}
The left inequality above gives
\begin{equation*}
    \frac{\log C_1}{n-k_1} + \frac{\log Z_{n-k_1}(f,[a])}{n-k_1} \leq \frac{\log Z_n(f,[b])}{n-k_1},
\end{equation*}
and by the limit on $n$ we obtain
\begin{equation*}
    \lim_n \frac{\log Z_n(f,[a])}{n} \leq \lim_n \frac{\log Z_n(f,[b])}{n},
\end{equation*}
and similarly on the right inequality we obtain
\begin{equation*}
    \lim_n \frac{\log Z_n(f,[b])}{n} \leq \lim_n \frac{\log Z_n(f,[a])}{n},
\end{equation*}
that is, $\lim_n \frac{\log Z_n(f,[b])}{n} = \lim_n \frac{\log Z_n(f,[a])}{n}$. Finally, such limit cannot be $-\infty$ since for every $x \in [a]$ such that $\sigma^{n_0}x = x$ we have that $\sigma^{kn_0}x = x$ and then $Z_{kn_0}(f,[a]) \geq e^{f_{kn_0}(x)} = e^{kf_{n_0}}(x)$,  so $\lim_k \frac{\log Z_{kn_0}(f,[a])}{kn_0} \geq \frac{f_{n_0}(x)}{n_0} > - \infty$.
\end{proof}

One of central notions in Thermodynamic Formalism is the concept of pressure.

\begin{definition}[Gurevich pressure] Let $\Sigma_A$ be topologically mixing and $f:\Sigma_A \to \mathbb{R}$ be a potential satisfying the Walters' condition. The Gurevich pressure of the potential $f$ is the quantity
\begin{equation*}
 P_G(f) := \lim_n \frac{1}{n}\log Z_n(f,[a]),
\end{equation*}
where $a \in S$.
\end{definition}

\begin{remark} The Gurevich pressure was constructed originally in \cite{Gurevich1969} by B. M. Gurevich, for the zero potential, corresponding to the Gurevich entropy. Later, this construction was generalized in \cite{Gurevich1984} for Markovian potentials, that is, potentials depending on the two first coordinates. In the original definition, the construction done above was made for the finite symbolic subgraphs and the Gurevich pressure was defined as the supremum on the finite subgraphs for these pressures. The construction presented is based on Sarig papers \cite{Sarig1999,Sarig2015} and notes \cite{Sarig2009} and it is equivalent, since in this case, the supremum construction done by Gurevich is a theorem. Sarig \cite{Sarig1999,Sarig2009} and Daon \cite{Daon2013} extended the Gurevich pressure for a larger class of potentials.
\end{remark}

\begin{remark}\label{remar:finite_pressure} A natural question which arises when we study the notion of Gurevich pressure is to provide conditions that ensures its finiteness. For the renewal shift endowed with a potential $F$, it is sufficient that $\sup F < \infty$. Indeed, the quantity of $n$-periodic elements in $[1]$ is $2^{n-1}$ and then
\begin{align*}
    Z_n(\beta F,[1]) = \sum_{\substack{x \in [1], \\ \sigma^x = x}} e^{\beta F_n(x)} \leq 2^{n-1} e^{\beta n \sup F}, \quad n \in \mathbb{N}, 
\end{align*}
and hence
\begin{align*}
    P_G(\beta F) &= \lim_n \frac{1}{n}\log Z_n(\beta,[a]) \leq \lim_n \frac{1}{n}\log \left[2^{n-1} e^{\beta n \sup F}\right]\\
    &= \lim_n \frac{1}{n}\left[(n-1) \log 2 + \beta n \sup F\right]\\
    &= \log 2 + \beta \sup F < \infty,
\end{align*}
and therefore $P(\beta F) < \infty$ for every $\beta > 0$. More general results can be found in \cite{Sarig2009,Sarig2015,Beltran2019}.
\end{remark}

\section{Conformal measures and Eigenmeasures}

The notion of conformal measure was introduced by Patterson \cite{Patterson1976}, and it was constructed originally to calculate the Hausdorff dimension of the limit set of a finitely generated Fuchsian group of the second kind. This idea was later adapted to the context of Markov shift spaces by Denker and Urba\'nski \cite{DenUr1991} and they are, under some conditions, eigenmeasures of the Ruelle's operator for the eigenvalue $1$. On the other hand, Sarig also defined a notion of conformal measure \cite{Sarig2009} which is not equal to the one stated by Denker and Urba\'nski; however, it is strictly related to it, and it is related to eigenmeasures as well. In this section we present and relate these notions. In addition, we present the recurrence modes and some existence and uniqueness results about conformal measures.

We start presenting the notion of conformal measure by Denker and Urba\'nski \cite{DenUr1991}.

\begin{definition}[Conformal measure - Denker-Urba\'nski] Let $(X,\mathcal{F})$ be a measurable space, $\sigma: X \to X$ a measurable endomorphism and $D:X \to [0,\infty)$ also measurable. A set $B\subseteq X$ is called special if $B\in \mathcal{F}$ and $\sigma_B:=\sigma\vert_B: B\to \sigma(B)$ is injective. A measure $\mu$ in $X$ is said to be $D$-conformal in the sense of Denker-Urba\'nski if
\begin{equation}\label{eq:conformal_Ur_sets_potential_classical}
    \mu(\sigma(B)) = \int_B D d\mu,
\end{equation}
for every special set $B$.
\end{definition}

\begin{remark} In the definition above, for $X = \Sigma_A$, we have that every Borel set contained in a cylinder is a special set. 
\end{remark}

Before we present the conformal measures constructed by Sarig we introduce some necessary concepts.

\begin{definition}[Non-singularity] Let $(X,\mathcal{F},\mu)$ be a measure space and $T:X\to X$ be a measurable map. We say $T$ is non-singular or $\mu$ is non-singular if $\mu \circ T^{-1} \sim \mu$, i.e.,
\begin{equation}\label{eq:non-singular}
    \mu\left(T^{-1}E\right) = 0 \iff \mu\left(E\right) = 0, \quad E \in \mathcal{F}.
\end{equation}
\end{definition}

For the next definition we observe first that, considering $X = \Sigma_A$ and $T = \sigma$, we have that $\sigma(E\cap [a])$ is a borel set of $\Sigma_A$. Indeed, we claim the restriction
\begin{align*}
    \sigma\vert_{[a]}:[a] &\to \sigma([a]) = \{y \in \Sigma_A: ay \in \Sigma_A\}\\
    x &\mapsto \sigma(x)
\end{align*}
is a homeomorphism. It is straightforward to notice that $\sigma\vert_{[a]}$ is a bijection, and the inverse map is given by
\begin{equation*}
    \sigma\vert_{[a]}^{-1}(y) = ay.
\end{equation*}
Also, since $\sigma$ is continuous, we have that $\sigma\vert_{[a]}$ is continuous as well (subspace topology). On the other hand, for any sequence $(y^n)_\mathbb{N}$ in $\sigma([a])$ such that $y^n \to y \in \Sigma_A$. Note that $y \in \sigma([a])$, because the convergence $y^n \to y$ is equivalent to the coordinate-wise convergence, and in particular there exists $N_0 \in \mathbb{N}$ such that $n \geq N_0$ implies $y_0^n = y_0$ and therefore $A(a,y_0) = A(a,y^n) = 1$, that is, $y \in \sigma([a])$. Now,
\begin{equation*}
    \sigma\vert_{[a]}^{-1}(y^n) = ay^n \to ay = \sigma\vert_{[a]}^{-1}(y),
\end{equation*}
hence $\sigma\vert_{[a]}^{-1}$ is continuous, and therefore $\sigma\vert_{[a]}$ is a homeomorphism and this argument can be extended to any cylinder, which proves that $\sigma$ is a local homeomorphism. Observe that
\begin{equation}\label{eq:sigma_cylinder_is_union_of_cylinders}
    \sigma([a]) = \bigsqcup_{\substack{b \in S,\\ A(a,b) = 1}}[b]
\end{equation}
for every $a \in S$, and hence $\sigma([a])$ is an open subset of $\Sigma_A$. So given $B \in \mathcal{B}_{\Sigma_A}$, we have that $B \cap [a] \in \mathcal{B}_{\Sigma_A} \cap [a] := \{C \cap [a]: C \in \mathcal{B}_{\Sigma_A}\} \subseteq \mathcal{B}_{\Sigma_A}$. Since homeomorphisms perserves Borel sets, we have that $\sigma\vert_{[a]}(B \cap [a]) = \sigma(B \cap [a]) \in \mathcal{B}_{\Sigma_A} \cap \sigma([a]) \subseteq \mathcal{B}_{\Sigma_A}$. This discussion summarizes that the measure defined next is in fact well-defined.

\begin{definition}\label{def:mu_odot_sigma} Suppose $\mu$ is non-singular measure on $\Sigma_A$ with alphabet $S$ and set the measure
\begin{equation}\label{eq:mu_odot_sigma}
    \mu \odot \sigma(E) := \sum_{a \in S}\mu\left(\sigma(E\cap[a])\right), \quad E \in \mathcal{B}_{\Sigma_A}.
\end{equation} 
\end{definition}

\begin{remark} Note that in general $\mu \odot \sigma(E) \neq \mu (\sigma(E))$. For instance, if $E = [ab] \sqcup [cb]$, where $a,b,c \in S$, $A(a,b) = A(c,b) = 1$, $a\neq c$, and $\mu([b])>0$. We have
\begin{equation*}
    \mu \odot \sigma(E) = 2 \mu ( \sigma(E)).
\end{equation*}
This also shows that the $\mu \odot \sigma$ gives the measure of $T(E)$ by considering its multiplicity, in the sense that it considers how much $E$ is spreaded by the action of $\sigma$ over the basic cylinders $[a]$, $a \in S$.
\end{remark}

It follows directly from Definition \ref{def:mu_odot_sigma} that $\mu \odot \sigma \ll \mu$. Indeed, given a Borel set $E$, $\mu(E) = 0$ if and only if $\mu(E\cap[a]) = 0$ for every $a \in S$. It is straightforward the equation
\begin{equation*}
    \mu\odot \sigma(E\cap[a]) = \mu([a]),
\end{equation*}
hence
\begin{equation*}
    \mu \odot \sigma(E) = \sum_{a \in S} \mu \odot \sigma(E\cap [a]) = 0. 
\end{equation*}

Also, it is clear that the non-singularity condition in Definition \ref{def:mu_odot_sigma} does not require the non-singularity of $\mu$. However, this extra condition is added in order to construct a Radon-Nikodym derivative which defines the conformal measure in the sense of Sarig. The next lemma shows how to evaluate integrals for $\mu \odot \sigma$ and also that the non-singularity makes the measures $\mu$ and $\mu \odot \sigma$ equivalent.

\begin{lemma}\label{lemma:properties_mu_odot_sigma} For a Markov shift space $\Sigma_A$ the following holds.
\begin{item}
    \item[$1.$] For all non-negative Borel functions $f: \Sigma_A \to \mathbb{R}$,
    \begin{equation*}
        \int_{\Sigma_A} f d\mu \odot \sigma = \sum_{a \in S} \int_{\sigma([a])} f(ax)d\mu(x).
    \end{equation*}
    \item[$2.$] $\mu \sim \mu \odot \sigma$.
\end{item}
\end{lemma}

\begin{proof} $1.$ For $f = \mathbbm{1}_B$, where $B$ is a Borel set, we have that
\begin{align*}
    \int_{\Sigma_A} \mathbbm{1}_B d\mu \odot \sigma &= \mu \odot \sigma(B) = \sum_{a \in S}\mu(\sigma(B \cap [a])) = \sum_{a \in S}\int_{\sigma(B \cap [a])}d\mu = \sum_{a \in S}\int_{\Sigma_A}\mathbbm{1}_{\sigma(B \cap [a])}d\mu\\
    &= \sum_{a \in S}\int_{\Sigma_A}\mathbbm{1}_{B \cap [a]}\circ \sigma\vert_{[a]}^{-1}d\mu = \sum_{a \in S}\int_{\Sigma_A}\left(\mathbbm{1}_{B}\circ \sigma\vert_{[a]}^{-1}\right)\left(\mathbbm{1}_{[a]}\circ \sigma\vert_{[a]}^{-1}\right)d\mu\\
    &= \sum_{a \in S}\int_{\sigma([a])}\left(\mathbbm{1}_{B}\circ \sigma\vert_{[a]}^{-1}\right)d\mu = \sum_{a \in S}\int_{\sigma([a])}\mathbbm{1}_{B} (ax)d\mu(x).
\end{align*}
By linearity of the integral and the positive sums above it follows that the result also holds for simple functions. Now let $f$ be a positive Borel function. Then, there exists a pointwise increasing sequence of simple functions $(\varphi_n)_{n \in \mathbb{N}}$ such that $\varphi_n \to f$ pointwise. We have that
\begin{align*}
    \int_{\Sigma_A} f d\mu \odot \sigma &= \int_{\Sigma_A} \lim_n \varphi_n d\mu \odot \sigma \stackrel{(\bullet)}{=} \lim_n \int_{\Sigma_A}  \varphi_n d\mu \odot \sigma = \lim_n \sum_{a \in S}\int_{\sigma([a])}\varphi_n (ax)d\mu(x)\\
    &\stackrel{(\bullet)}{=}  \sum_{a \in S}\lim_n\int_{\sigma([a])}\varphi_n (ax)d\mu(x)\stackrel{(\bullet)}{=}  \sum_{a \in S}\lim_n\int_{\sigma([a])}\varphi_n (ax)d\mu(x),
\end{align*}
where in $(\bullet)$ we used the Monotone Convergence Theorem.

$2.$ As we discussed previously, we only need to prove that $\mu \ll \mu \odot \sigma$. Suppose that for $B$ Borel set of $\Sigma_A$ we have $\mu \odot \sigma(B) = 0$. Then $\mu(\sigma(B\cap [a])) = \mu(\sigma\vert_{[a]}(B\cap [a])) = 0$ for every $a \in S$. By the non-singularity property it follows that $\mu(\sigma^{-1}\sigma\vert_{[a]}(B\cap [a])) = 0$. Since $\sigma^{-1}\vert_{[a]}\sigma\vert_{[a]}(B\cap [a]) \subseteq \sigma^{-1}\sigma\vert_{[a]}(B\cap [a])$ and $\sigma$ is a local homeomorphism as we discussed previously, one gets, for all $a \in S$,
\begin{equation*}
    0 \leq \mu(B\cap [a]) = \mu(\sigma^{-1}\vert_{[a]}\sigma\vert_{[a]}(B\cap [a])) \leq \mu(\sigma^{-1}\sigma\vert_{[a]}(B\cap [a])) = 0 \implies \mu(B\cap [a]) = 0, 
\end{equation*}
and therefore $\mu(B) = 0$.
\end{proof}

\begin{definition}[Conformal measure - Sarig] Given a Borel non-singular $\sigma$-finite measure $\mu$ on $\Sigma_A$ and a potential $f:\Sigma_A \to \mathbb{R}$, we say $\mu$ is $(\beta F, \lambda)$-conformal in the sense of Sarig if there exists $\lambda >0$ such that
\begin{equation*}
\dfrac{d\mu\odot\sigma}{d\mu}(x)= \lambda e^{-\beta F(x)}\quad \mu\; a.e \; x\in \Sigma_A.
\end{equation*}
\end{definition}

In order to present the Ruelle's operator, we present the notion of transfer operator, which requires the following lemma.

\begin{lemma}\label{lemma:transfer_operator_well_defined} Let $T$ be a non-singular map on a $\sigma$-finite measure space $(X, \mathcal{F}, \mu)$. Then, the Radon-Nikodym derivative
\begin{equation}\label{eq:radon_nikodym_transfer_operator}
    \frac{d \mu_f \circ T^{-1}}{d \mu},
\end{equation}
where, $f \in L^1(\mu)$ and $d\mu_f := f d\mu$ exists and it belongs to $L^1(\mu)$. 
\end{lemma}

\begin{proof} Let $B \in \mathcal{F}$ such that $\mu(B) = 0$. By the non-singularity of $T$ we have that $\mu(T^{-1}B) = 0$ and then
\begin{equation*}
    0\leq \mu_f\circ T^{-1}(E) = \mu_f(T^{-1}E) = \int_{T^{-1}E} f d\mu \leq \|f\|_{L^1} \mu(T^{-1}E) = 0,
\end{equation*}
and hence $\mu_f\circ T^{-1}(E) = 0$. So $\mu_f \circ T^{-1} \ll \mu$, and since $(X,\mathcal{F},\mu)$ is $\sigma$-finite, the Radon-Nikodym Theorem grants the existence of the derivative \eqref{eq:radon_nikodym_transfer_operator}. On the other hand, denoting the derivative in \eqref{eq:radon_nikodym_transfer_operator} by $h$, we have
\begin{align*}
    \left\|h\right\|_{L^1} &= \int \sgn(h) h d\mu = \int \sgn(h) \frac{d \mu_f \circ T^{-1}}{d \mu} d\mu = \int \sgn(h) d \mu_f \circ T^{-1} \\
    &= \int \sgn(h)(T(x)) d \mu_f (x) = \int \sgn(h)(T(x)) f(x) d \mu (x) \leq \|f\|_{L^1} < \infty, 
\end{align*}
and therefore $\frac{d \mu_f \circ T^{-1}}{d \mu} \in L^1(\mu)$.
\end{proof}

The previous lemma ensures that the transfer operator is well-defined.

\begin{definition}[Transfer operator] Let $T$ be a non-singular map on a $\sigma$-finite measure space $(X, \mathcal{F}, \mu)$. The transfer transfer operator is the map $\widehat{T}:L^1(\mu) \to L^1(\mu)$ given by
\begin{equation*}
    \widehat{T}f:=\frac{d \mu_f \circ T^{-1}}{d \mu},
\end{equation*} 
\end{definition}

The next proposition presents some properties of the transfer operator, and its proof can be found in \cite{Sarig2009}.

\begin{proposition}\label{prop:transfer_operator_properties} Let $T$ be a non-singular map on a $\sigma$-finite measure space $(X, \mathcal{F}, \mu)$. The following are true.
\begin{itemize}
 \item[$(i)$] If $f \in L^1(\mu)$, then $\widehat{T}f$ is the unique function in $L^1(\mu)$ such that
 \begin{equation*}
    \int \psi \widehat{T}f d\mu = \int (\psi \circ T) f d\mu,
 \end{equation*}
for every $\psi \in L^\infty(\mu)$.
 \item[$(ii)$] $\widehat{T}$ is positive, in the sense that for every $f \in L^1(\mu)$ such that $f \geq 0$ $a.e.$, we have $\widehat{T}f \geq 0$ $a.e.$.
 \item[$(iii)$] $\widehat{T}$ is a bounded linear operator on $L^1(\mu)$, satisfying $\|\widehat{T}\| = 1$.
 \item[$(iv)$] $\widehat{T}^* \mu = \mu$, in the sense that
 \begin{equation*}
    \int \widehat{T}f d\mu = \int f d\mu,
 \end{equation*}
 for every $f \in L^1(\mu)$.
\end{itemize}
\end{proposition}

\begin{proposition}\label{prop:transfer_operator_Markov_shifts} Consider the Markov shift $\Sigma_A$ and a non-singular measure $\mu$ on $\mathcal{B}_{\Sigma_A}$. The transfer operator $\widehat{\sigma}$ is given by
\begin{equation*}
    (\widehat{\sigma}f)(x) = \sum_{y \in \sigma^{-1}x} \dfrac{d\mu}{d\mu\odot\sigma}(y) f(y) \quad \mu-a.e.
\end{equation*}
for every $f \in L^1(\mu)$. 
\end{proposition}

\begin{proof} First, note that
\begin{equation}\label{eq:equivalent_sums_Ruelle_operator}
    \sum_{y \in \sigma^{-1}x} \dfrac{d\mu}{d\mu\odot\sigma}(y) f(y) = \sum_{a \in S} \mathbbm{1}_{\sigma([a])}(x)\dfrac{d\mu}{d\mu\odot\sigma}(ax) f(ax)
\end{equation}
for every $x \in \Sigma_A$. Indeed, it is a straightforward consequence from the following chain of equivalences:
\begin{equation*}
    y \in \sigma^{-1} x \iff \sigma y = x \iff y = ax \text{ for some } a \in S \iff x \in \sigma([a])\text{ for some } a \in S.
\end{equation*}
Now, given $\psi \in L^\infty(\mu)$, we have that
\begin{align*}
    \int \psi(x) \sum_{y \in \sigma^{-1}x} \dfrac{d\mu}{d\mu\odot\sigma}(y) f(y) d\mu(x) &= \int \psi(x) \sum_{a \in S} \mathbbm{1}_{\sigma([a])}(x) \dfrac{d\mu}{d\mu\odot\sigma}(ax) f(ax) d\mu(x)\\
    &= \sum_{a \in S}  \int \psi(x) \mathbbm{1}_{\sigma([a])}(x)\dfrac{d\mu}{d\mu\odot\sigma}(ax) f(ax) d\mu(x)\\
    &= \sum_{a \in S}  \int_{\sigma([a])} \psi \circ \sigma(ax) \dfrac{d\mu}{d\mu\odot\sigma}(ax) f(ax) d\mu(x).
\end{align*}
Observe that $H:= (\psi \circ \sigma) \left(\dfrac{d\mu}{d\mu\odot\sigma}\right) f \in L^1(\mu)$. By dividing $H$ into its positive and negative parts, we have
\begin{align*}
    \int \psi(x) \sum_{y \in \sigma^{-1}x} \dfrac{d\mu}{d\mu\odot\sigma}(y) f(y) d\mu(x) &= \sum_{a \in S}  \int_{\sigma([a])} H(ax) d\mu(x)\\
    &= \left(\sum_{a \in S} \int_{\sigma([a])} H_+(ax) d\mu(x)\right) -\left(\sum_{a \in S} \int_{\sigma([a])} H_-(ax) d\mu(x) \right)\\
    &\stackrel{(\bullet)}{=} \left(\int H_+ d\mu\odot \sigma \right) -\left(\int H_- d\mu\odot \sigma \right)\\
    &= \int H d\mu\odot \sigma,
\end{align*}
where in $(\bullet)$ we used Lemma \ref{lemma:properties_mu_odot_sigma}. Then,
\begin{align*}
    \int \psi(x) \sum_{y \in \sigma^{-1}x} \dfrac{d\mu}{d\mu\odot\sigma}(y) f(y) d\mu(x) &= \int (\psi \circ \sigma) f \dfrac{d\mu}{d\mu\odot\sigma}  d\mu\odot \sigma = \int (\psi \circ \sigma) f  d\mu.
\end{align*}
By Proposition \ref{prop:transfer_operator_properties} $(i)$ we conclude that
\begin{equation*}
    (\widehat{\sigma}f)(x) = \sum_{y \in \sigma^{-1}x} \dfrac{d\mu}{d\mu\odot\sigma}(y) f(y) \quad \mu-a.e.
\end{equation*}
for every $f \in L^1(\mu)$. 
\end{proof}

\begin{remark} The Radon-Nikodym derivative $\dfrac{d\mu}{d\mu\odot\sigma}$ is also refered as the Jacobian of $\mu$. 
\end{remark}

It is straightforward to observe that, if $\mu$ is a $(\beta F,\lambda)$-conformal measure, then
\begin{equation*}
    (\widehat{\sigma}f)(x) = \lambda \sum_{y \in \sigma^{-1}x} e^{\beta F(y)} f(y) \quad \mu-a.e.,
\end{equation*}
and since the conformal measures in this sense are non-singular we have that the equality above is well-defined. By removing the term $\lambda^{-1}$, we have the Ruelle's operator, as defined below.

\begin{definition}[Ruelle's operator] Let $f: \Sigma_A \to \mathbb{R}$ be a potential on the Markov shift space $\Sigma_A$ and consider $\nu$ a $(\lambda,\beta F)$-conformal measure. The Ruelle's operator is the function $L_{\beta F}:L^1(\nu) \to L^1(\nu)$, defined by 
\begin{equation}\label{eq:Ruelle_operator}
    \left(L_{\beta F} f\right)(x) := \sum_{y \in \sigma^{-1}x} e^{\beta F(y)} f(y).
\end{equation}
\end{definition}

\begin{remark} \label{remark:transfer_operator_Ruelle_log_jacobian} By Proposition \ref{prop:transfer_operator_Markov_shifts} we have that the transfer operator $\hat{\sigma}$ of a non-singular measure on $\Sigma_A$ is the Ruelle's operator of its $\log$ Jacobian, i.e.,
\begin{equation*}
    \hat{\sigma}(f) = L_{\log \left(\frac{d\mu}{d\mu\odot\sigma}\right)}(f).
\end{equation*}
\end{remark}

\begin{remark} The Ruelle's operator defined as above was constructed for an arbitrary alphabet $S$, in order to make the sum in the RHS of \eqref{eq:Ruelle_operator} be convergent $\nu-a.e$. However, for locally compact Markov shifts, the same sum converges for any $f \in C_c(\Sigma_A)$. In particular, if $\Sigma_A$ is compact, we may define the Ruelle's operator on $C(\Sigma_A)$. See, for instance, the references \cite{Bowen2008,Ruelle1976} for the compact case. For the locally compact case, see \cite{Shwartz2019}. In addition, there is an different approach of the Ruelle's operator without the local compactness on \cite{MaulUr2001}, by defining the Ruelle's operator on the bounded continuous functions and restricting the potentials to the exp-summable case. 
\end{remark}

\begin{definition}[Eigenmeasures of the Ruelle's operator] Given a Borel $\sigma$-finite measure $\mu$ on $\Sigma_A$, a measurable potential $F: \Sigma_A \to \mathbb{R}$ and $\lambda>0$, we say that $\mu$ is an eigenmeasure of $L_F$ with eigenvalue $\lambda$ when
\begin{equation*}
    \int L_F f d\mu = \lambda \int f d\mu
\end{equation*}
for every $f \in L^1(\mu)$.
\end{definition}

The different notions of conformality and eigenmeasures of the Ruelle's operator are connected, as we show next.

\begin{theorem}\label{thm:equivalence_conformality_Sarig_eigen_measure_Ruelle_operator} Let $\Sigma_A$ be a Markov shift, $F: \Sigma_A \to \mathbb{R}$ a measurable potential and $\mu$ a $\sigma$-finite measure. Then, $\mu$ is $(F,\lambda)$-conformal if and only if it is an eigenmeasure of the Ruelle's operator $L_F$ for the eigenvalue $\lambda$.
\end{theorem}

\begin{proof} Let $B \in \mathcal{B}_{\Sigma_A}$. By item $1.$ from lemma \ref{lemma:properties_mu_odot_sigma}, we have that
\begin{align*}
    \int \mathbbm{1}_B e^{F} d \mu \odot \sigma &= \sum_{a \in S} \int_{\sigma([a])} \mathbbm{1}_B(ax) e^{F(ax)} d\mu(x) \stackrel{(\bullet)}{=} \int \sum_{a \in S} \mathbbm{1}_{\sigma([a])}(x) \mathbbm{1}_B(ax) e^{F(ax)} d\mu(x)\\
    &\stackrel{(\bullet \bullet)}{=} \int \sum_{y \in \sigma^{-1}x} \mathbbm{1}_B(y) e^{F(y)} d\mu(x) = \int L_F \mathbbm{1}_B d\mu,
\end{align*}
where the equality $(\bullet)$ is a consequence of the Monotone Convergence Theorem, while in $(\bullet \bullet)$ is proven by similar arguments used to prove \eqref{eq:equivalent_sums_Ruelle_operator}. Hence we have
\begin{align}\label{eq:integral_identity_Ruelle_operator_mu_odot_sigma}
    \int \mathbbm{1}_B e^{F} d \mu \odot \sigma = \int L_F \mathbbm{1}_B d\mu.
\end{align}
If $\mu$ is a $(F,\lambda)$-conformal measure, then by \eqref{eq:integral_identity_Ruelle_operator_mu_odot_sigma} we get
\begin{align*}
    \mu(B) &= \int \mathbbm{1}_B d\mu = \int \mathbbm{1}_B \dfrac{d\mu}{d \mu \odot \sigma} d \mu \odot \sigma = \lambda^{-1}\int \mathbbm{1}_B e^F d \mu \odot \sigma = \int L_F \mathbbm{1}_B d\mu,
\end{align*}
that is,
\begin{align}\label{eq:conformal_Sarig_is_eigenmeasure_characteristic_function}
    \int L_F \mathbbm{1}_B d\mu = \lambda \int \mathbbm{1}_B d\mu. 
\end{align}
By linearity of the integral, the identity \eqref{eq:conformal_Sarig_is_eigenmeasure_characteristic_function} is also valid for simple functions. Now, given $f \in L^1(\mu)$, $f \geq 0$, there exists a crescent sequence $\{g_n\}_{n \in \mathbb{N}}$ of simple functions that converges to pointwise to $f$, that is
\begin{equation*}
    0 \leq g_n(x) \leq g_{n+1}(x) \leq f(x) g_n(x) \to f(x)
\end{equation*}
for every $n \in \mathbb{N}$ and $x \in \Sigma_A$, with $g_n(x) \to f(x)$. By the Monotone Convergence Theorem, we have
\begin{align*}
    \lambda \int f d\mu &= \lambda \int \lim_n g_n d\mu = \lim_n \lambda \int g_n d\mu = \lim_n \int L_F g_n d\mu \\
    &= \int \lim_n L_F g_n d\mu = \int L_F \lim_n g_n d\mu = \int L_F f d\mu,
\end{align*}
where in the last equality above we used that $L_F$ is continuous on $L^1(\mu)$ because $\widehat{\sigma}$ is bounded due to the item $(iii)$ of Proposition \ref{prop:transfer_operator_properties}. For the general case $f \in L^1(\mu)$ is proved simply by applying the previous case for $f_+$ and $f_-$, and by using the linearity of the integral and $L_F$. Conversely, let $\mu$ be an eigenmeasure of $L_F$ for the eigenvalue $\lambda$. Since $\mu$ is a $\sigma$-finite measure, there exists a pairwise disjoint measurable partition $\{E_k\}_{k \in \mathbb{N}}$ of $\Sigma_A$ satisfying $\mu(E_k) < \infty$, for every $k \in \mathbb{N}$. Then,
\begin{align}\label{eq:mu_partition_to_mu_odot_sigma_partition}
    \mu(B\cap E_k) = \int \mathbbm{1}_{B\cap E_k} d\mu &\stackrel{(\dagger)}{=} \lambda^{-1} \int L_F \mathbbm{1}_{B\cap E_k} d\mu \stackrel{(\ddagger)}{=} \lambda^{-1} \int \mathbbm{1}_{B\cap E_k} e^F d \mu \odot \sigma,
\end{align}
where in $(\dagger)$ we used the definition of eigenmeasure\footnote{Observe that the eigenmeasure equation has its validity granted for $L^1(\mu)$ and since it $\mu$ is not necessarily a finite measure, we need to use a partition that witnesses the $\sigma$-finiteness of $\mu$.} and in $(\ddagger)$ we used the equation \eqref{eq:integral_identity_Ruelle_operator_mu_odot_sigma}. Then,
\begin{equation}\label{eq:measure_on_B_equals_to_integral_on_B_derivative_Radon_Nikodym}
    \int_B d\mu = \sum_{k \in \mathbb{N}} \mu(B \cap E_k) \stackrel{(\spadesuit)}{=} \sum_{k \in \mathbb{N}} \int \mathbbm{1}_{B\cap E_k} \lambda^{-1} e^F d \mu \odot \sigma = \int_B \lambda^{-1} e^F d \mu \odot \sigma,
\end{equation}
where in $(\spadesuit)$ we used \eqref{eq:mu_partition_to_mu_odot_sigma_partition}. By similar arguments used to prove that a $(F,\lambda)$-conformal measure is an eigenmeasure, we have that the result above is also valid for simple functions, then for positive functions in $L^1(\mu)$ and finally for any function in $L^1(\mu)$. We conclude that
\begin{equation*}
    \frac{d \mu \odot \sigma}{d\mu}(x) = \lambda e^{-F(x)}, \quad \mu-a.e. x \in \Sigma_A,
\end{equation*}
that is, $\mu$ is a $(F,\lambda)$-conformal measure. Observe that in this case the non-singularity of the eigenmeasure is straightforward because of the validity of the equation \eqref{eq:measure_on_B_equals_to_integral_on_B_derivative_Radon_Nikodym}. Indeed, given $B \in \mathcal{B}_{\Sigma_A}$ such that $\mu(B) = 0$ and $\mu(\sigma^{-1} B)>0$, then $\mu \odot \sigma (\sigma^{-1} B) >0$ and hence \eqref{eq:measure_on_B_equals_to_integral_on_B_derivative_Radon_Nikodym} does not hold. Covnersely, if $\mu(B) > 0$ and $\mu(\sigma^{-1} B)=0$, we necessarily have $\mu \odot \sigma (\sigma^{-1} B) =0$ and again we have a contradiction on the validity of \eqref{eq:measure_on_B_equals_to_integral_on_B_derivative_Radon_Nikodym}.
\end{proof}

\begin{theorem}\label{thm:equivalence_conformality_Sarig_and_DU} Let $\Sigma_A$ be a Markov shift and $F: \Sigma_A \to \mathbb{R}$ a measurable potential. A Borel measure $\mu$ is $e^F$-conformal (Denker-Urba\'nski) if and only if it is non-singular and $(-F,1)$-conformal (Sarig).
\end{theorem}

\begin{proof} Let $\mu$ be a $e^F$-conformal measure. Observe that for every $B \in \mathcal{B}_{\Sigma_A}$ and $a \in S$ we have that $B \cap [a]$ is a special set. Then,
\begin{equation}\label{eq:mu_sigma_B_cap_cylinder_a}
    \mu(\sigma(B \cap [a])) = \int_{B \cap [a]} e^F d \mu.
\end{equation}
By \eqref{eq:mu_sigma_B_cap_cylinder_a} we have
\begin{equation*}
    \mu \odot \sigma(B) = \sum_{a \in S} \mu(\sigma(B \cap [a])) = \sum_{a \in S} \int_{B \cap [a]} e^F d \mu = \int_B e^F d \mu,
\end{equation*}
and then $\mu$ is non-singular and
\begin{equation*}
    \frac{d \mu \odot \sigma}{d \mu} = e^F,
\end{equation*}
hence $\mu$ is also $(-F,1)$-conformal. Conversely, suppose that $\mu$ is non-singular and it is $(-F,1)$-conformal measure and let $B \in \mathcal{B}_{\Sigma_A}$ be a special set. Then, for every $a,b \in \mathbb{N}$, $a \neq b$, we have
\begin{equation*}
    \sigma(E \cap [a]) \cap \sigma(B \cap [b]) = \emptyset.
\end{equation*}
Indeed, w.l.o.g. suppose that $B \cap [a] \neq \emptyset$ and $B \cap [b] \neq \emptyset$. If $y \in \sigma(B \cap [a]) \cap \sigma(B \cap [b])$, then there exist $x^a \in B \cap [a]$ and $x^b \in B \cap [b]$ such that $y = \sigma\vert_B(x^a) = \sigma\vert_B(x^b)$. Since $B$ is special, we have that $\sigma\vert_B$ is injective and therefore $x^a = x^b$, which is a contradiction since $[a] \cap [b] = \emptyset$. Moreover, we have that
\begin{equation}\label{eq:union_sigma_B_cap_cylinders_equals_to_sigma_B}
    \bigsqcup_{a \in S} \sigma(B\cap [a]) = \sigma B.
\end{equation}
Indeed, it is straightforward that $\sigma(B \cap [a]) \subseteq \sigma B$. Now, let $y \in \sigma B$, and hence there exists $x \in B$ such that $\sigma x = y$. Consequently, there exists $b \in S$ such that $x \in B \cap [b]$, and therefore $y \in \sigma(B \cap [b])$, which proves \eqref{eq:union_sigma_B_cap_cylinders_equals_to_sigma_B}. Then
\begin{align*}
    \int_B e^F d\mu &= \sum_{a \in S} \int_{B\cap [a]} e^F d\mu = \sum_{a \in S} \int_{B\cap [a]} e^F \frac{d\mu}{d\mu \odot \sigma} d\mu \odot \sigma = \sum_{a \in S} \int_{B\cap [a]} 1 d\mu \odot \sigma \\
    &= \sum_{a \in S} \mu \odot \sigma (B\cap [a]) = \sum_{a \in S} \mu ( \sigma (B\cap [a])) \stackrel{(\dagger)}{=} \mu ( \sigma B),
\end{align*}
where in $(\dagger)$ we used \eqref{eq:union_sigma_B_cap_cylinders_equals_to_sigma_B}. We conclude that $\mu$ is $e^F$-conformal.
\end{proof}

The next result summarizes the last two theorems.

\begin{corollary}\label{cor:equivalences_conformality_classical} On a Markov shift $\Sigma_A$, consider a potential $F:\Sigma_A \to \mathbb{R}$ and $\mu$ a Borel measure on $\Sigma_A$. The following are equivalent:
\begin{itemize}
    \item[$(i)$] $\mu$ is $e^F$-conformal (in the sense of Denker-Urba\'nski);
    \item[$(ii)$] $\mu$ is non-singular and it is $(-F,1)$-conformal (in the sense of Sarig);
    \item[$(iii)$] $\mu$ is an eigenmeasure of $L_{-F}$ with eigenvalue $1$.
\end{itemize}
\end{corollary}

\begin{proof} It is a straightforward result from Theorem \ref{thm:equivalence_conformality_Sarig_eigen_measure_Ruelle_operator}, for $\lambda = 1$, and Theorem \ref{thm:equivalence_conformality_Sarig_and_DU}.
\end{proof}

\begin{remark}\label{remark:equivalences_conformality_classical_general_eigenvalue} The result above is also valid for general $\lambda$, by exchanging the potential $F$ by $F - \log \lambda$. Also, it is straightforward the validity of the equivalences here studied when we adapt the potential in order to include the inverse of the temperature $\beta > 0$. 
\end{remark}

\section{Recurrence, Conservativity and Existence of Conformal Measures}

In the last section we presented the different notions of conformal measures and the notion of eigen-measures of the Ruelle's operator, and how these definitions are related, by presenting and proving results about equivalences between these concepts. Now we present conditions that grants the existence of the conformal measures by two different ways of approach. The first one is a Measure Theory approach, and it is based on the notion of recurrence of the potential, which is inspired on the theory of Markov chains (see \cite{Sarig2009}). The second one is an approch by the point of the view of Analysis, constructed by M. Denker and M. Yuri \cite{DenYu2015}, and it will be presented for the context of Generalized Markov shift spaces in chapter \ref{ch:term_form_X_A}.

In order to present the definitions of recurrence we define the following. Let $a, n \in \mathbb{N}$, $x \in \Sigma_A$ and $F:\Sigma_A \to \mathbb{R}$ be a potential, we define
\begin{itemize}
    \item $\varphi_a(x):= \mathbbm{1}_{[a]}(x) \inf\{n\in \mathbb{N}:\sigma^n(x) \in [a]\}$ (first return time);
    \item $Z_n^*(F,[a]) := \sum_{\sigma^n x = x}e^{F_n(x)}\mathbbm{1}_{[\varphi_a = n]}(x)$;
    \item $\lambda:= e^{P_G(F)} = \lim_n [Z_n(F,[a])]^{1/n}$.
\end{itemize}

\begin{definition}[Modes of Recurrence] Let $\Sigma_A$ be topologically mixing and $F: \Sigma_A \to \mathbb{R}$ be a potential such that $P_G(F) < \infty$. Fix $a \in \mathbb{N}$. We say that 
\begin{itemize}
    \item $F$ is recurrent if $\sum_{n \in \mathbb{N}} \lambda^{-n} Z_n(F,[a]) = \infty$;
    \item $F$ is positive recurrent if it is recurrent and $\sum_{n \in \mathbb{N}}n \lambda^{-n} Z_n^*(F,[a]) < \infty$;
    \item $F$ is null recurrent if it is recurrent and $\sum_{n \in \mathbb{N}}n \lambda^{-n} Z_n^*(F,[a]) = \infty$;
    \item $F$ is transient if $\sum_{n \in \mathbb{N}} \lambda^{-n} Z_n(F,[a]) < \infty$.
\end{itemize}
\end{definition}

Now, we present the notion of conservativity and we study its relation with the recurrence modes.

\begin{definition}[wandering sets and conservativity] Let $T$ be non-singular map on a sigma finite measure space $(X,\mathcal{F},\nu)$. A set $W\in \mathcal{F}$ is said to be a wandering set if $\{T^{-n}W\}_{n\in \mathbb{N}_0}$ is a pairwise disjoint family. We say that $T$ (or $\mu$) is conservative if every wandering set $W \in \mathcal{F}$ satisfies $W=\emptyset$ or $W = \emptyset \mod\nu$. 
\end{definition}

\begin{remark} Equivalently, $W$ is a wandering set when $T^{-k}(W)\cap W=\emptyset$ for every $k\in \mathbb{N}$. 
\end{remark}

The next result is an equivalence for conservativity.

\begin{theorem}[Halmos]\label{thm:Halmos} Suppose $T$ is a non-singular map on a $\sigma$-finite measure space $(X,\mathcal{F},\nu)$. $T$ is conservative iff the following holds for every measurable set $E$ of strictly positive measure: for a.e $x\in E$, we have $T^n(x)\in E$ for infinitely many positive $n$'s.
\end{theorem}

\begin{proof} First suppose that for every measurable set $E$ of strictly positive measure we have that for a.e $x\in E$, we have $T^n(x)\in E$ for infinitely many positive $n$'s. For every wandering set $W$ we have necessarily that, for a.e. $x \in W$, it holds that $T^k(x) \notin W$ for every $k \in \mathbb{N}$, and by hypothesis we have $\nu(W) = 0$, and therefore $\nu$ is conservative. Conversely, suppose $T$ conservative, and assume that there is a measurable set $E$ s.t. $\nu\{x\in E: \# \{n\geq 0: T^n(x)\in E\}<\infty\}\neq 0$. Set
\begin{equation*}
    E_N:=\{x\in E: |\{n\geq 0:T^n(x)\in E\}|=N\}.
\end{equation*}
There exists $N\in\mathbb{N}$ satisfying $\nu(E_N)\neq 0$. For every $k\in \mathbb{N}$. We claim that $T^{-k}E_N\cap E_N=\emptyset$. Indeed, if there is a point $x\in T^{-k}E_N\cap E_N$, then this point would visit $E$ at least $N+1$ times (once at time zero, then $N$ times at times $k$ or larger). But this contradicts the definition of $E_N$, and then the claim is proved. Consequently, $E_N$ is a wandering set of positive measure, which is a contradiction since $T$ is conservative.
\end{proof}

\begin{proposition}\label{prop:T_conservative_sum_T_n} Let $T$ be a non-singular map of a $\sigma$-finite measure space.
 \begin{itemize}
     \item[1.] If there is a non-negative $f\in L^1$ s.t $\sum_{n \in \mathbb{N}}\hat{T}^nf=\infty$ a.e, then $T$ is conservative.
     \item[2.]If there is a strictly positive $f\in L^1$ s.t $\sum_{n \in \mathbb{N}}\hat{T}^nf<\infty$ on a set of positive measure, then $T$ is not conservative.
 \end{itemize}
\end{proposition}

\begin{proof} Denote the underlying measure space by $(X,\mathcal{F},\nu)$. Suppose $f$ is non-negative integrable s.t. $\sum_{n\in \mathbb{N}}\hat{T}^nf=\infty$ a.e. We will show that every wondering set $W$ has measure zero. Since $W$ is wandering, $\sum_{n \in \mathbb{N}} 1_W\circ T^n\leq 1$ everywhere. Indeed, every $x$ must be on only one of the $T^{-n}(W)$ or in no one. In the later the function is zero, otherwise it is one. Now, by Monotone Convergence Theorem and Theorem \ref{prop:transfer_operator_properties} $(i)$, we have
\begin{equation*}
   \|f\|_1\geq \int f\left(\sum_{n \in \mathbb{N}} 1_W\circ T^n\right)d\nu=\sum_{n \in \mathbb{N}}\int f 1_W\circ T^n d\nu=\sum_{n \in \mathbb{N}} \int_W\hat{T}^n fd\nu=\int_W\sum_{n \in \mathbb{N}}\hat{T}^n f d\nu. 
\end{equation*}
However, $\sum_{n \in \mathbb{N}}\hat{T}^n f=\infty$ a.e., so $W$ must have measure zero. This proves item 1.
 
For item 2, suppose that $f$ is a stricly positive integrable function such that $U:= [\sum \hat{T}^n f<\infty]$ has positive measure and we shall show that $T$ is not conservative. Since $\sum_n \hat{T}^n <\infty$ on $U$, $\exists B\subseteq U$ of positive measure such that $\int_B \sum \hat{T}^n f d\nu <\infty$, this is due to the measure being $\sigma$-finite. Now, due to the Monotone Convergence Theorem,
\begin{equation*}
    \int f \sum_{n \in \mathbb{N}} 1_B\circ T^n d\nu =\sum_{n \in \mathbb{N}} \int f  1_B\circ T^n d\nu =\sum_{n \in \mathbb{N}} \int_B \hat{T}^n f d\nu=\int_B \sum_{n \in \mathbb{N}}\hat{T}^n f d\nu<\infty,
\end{equation*}
whence since $f$ is strictly positive $\sum_{n \in \mathbb{N}} 1_B\circ T^n<\infty$ a.e in $X$, whence in $B$. It follows that for a.e $x\in B, T^n(x)\in B$ only finite many times. By Theorem \ref{thm:Halmos}, $T$ is not conservative, proving item 2.
\end{proof}

The next theorem relates recurrence with conservativity.

\begin{theorem}\label{thm:equivalence_conservativity_recurrence} Let $\Sigma_A$ be a transitive Markov shift space, $F$ be a potential, and suppose $\nu$ is a non-singular measure which is finite on cylinders and s.t. $\dfrac{d\nu}{d\nu\odot \sigma}=\lambda^{-1}e^{F}$. If $f$ satisfies the Walters property, then $\nu$ is conservative iff for some $a \in S$ (whence all $a \in S$),
\begin{equation}\label{eq:series_Z_n}
    \sum_{n=1}^\infty \lambda^{-1}Z_n(F,[a])= \infty.
\end{equation}
\end{theorem}

\begin{proof} Since $F$ satisfies the Walters property, by Lemma \ref{lemma:Z_n_b_bounded_Z_n_a} we observe that if the series $\sum_{n=1}^\infty \lambda^{-1}Z_n(F,a)$ converges for some $a \in S$, then it converges for all $a \in S$. Also, the transfer operator is $\hat{\sigma}g(x)=\lambda^{-1}\sum_{\sigma y=x}e^{F(y)}g(y)$. By applying $\hat{\sigma}$ on $g$ a finite number of times, we have by induction that
\begin{equation*}
    \hat{\sigma}^ng(x)=\lambda^{-n}\sum_{\sigma^n y=x}e^{F(y)}g(y).
\end{equation*}
For every state $b$, there are constants $C_1,C_2>0$ and $k_1,k_2\in \mathbb{N}$ satisfying
\begin{equation}\label{eq:transfer_operator_characteristic_bounded_by_Z_partition}
 C_1 \lambda^{-(n-k_1)}Z_{n-k_1}(F,b) \leq \hat{\sigma}^n 1_{[a]}(x) \leq  C_2 \lambda ^{-(n+k_2)}Z_{n+k_2}(F,b),
\end{equation}
for every $x \in [b]$ and every $n> k_1$. We omit the proof, since it is similar to the proof of Lemma \ref{lemma:Z_n_b_bounded_Z_n_a}. Suppose $\sum_n \lambda^{-n}Z_n(F,[a])=\infty$, then \eqref{eq:transfer_operator_characteristic_bounded_by_Z_partition} implies that $\sum_n \hat{T}^n1_{[a]}=\infty$ everywhere on $\Sigma_A$. By Proposition \ref{prop:T_conservative_sum_T_n}, item 1, $T$ is conservative. Conversely, suppose that $\sum_n \lambda^{-n}Z_n(F,[a])<\infty$ for some state $a$. Then $\sum_n \lambda^{-n}Z_n(F,[b])<\infty$ for every $b \in S$. By \eqref{eq:transfer_operator_characteristic_bounded_by_Z_partition}, for each $b \in S$ there exists a constant $C_{ab}>0$ s.t $\sum_n \hat{T}^n 1_{[a]}\leq C_{ab}$ on $[b]$. Choose $b$ s.t. $\nu [b]\neq 0$, and consider $\{\epsilon_a\}_{a\in S}$ be positive numbers s.t $\sum_a \epsilon_a C_{ab}<\infty$ and $\sum_a \epsilon_a \nu [a]<\infty$. The last inequality is possible because the measure is finite on cylinders. The function
\begin{equation*}
  h:= \sum_{a\in S} \epsilon_a 1_{[a]}
\end{equation*}
is positive and it belongs to $L^1(\nu)$. Moreover, for a.e. $x\in [b]$, it follows that
\begin{equation*}
    \sum_{n=1}^\infty \hat{\sigma}^n h=\sum_{n=1}^\infty\sum_{a\in S}\epsilon_a\hat{\sigma}^n 1_{[a]} \leq \sum_{a \in S} \epsilon_a C_{ab}<\infty.
\end{equation*}
Since $\nu[b]\neq 0$, we obtain by Proposition \ref{prop:T_conservative_sum_T_n} that $\nu$ is not conservative.  
\end{proof}

It is immediate from theorem above that the conservativity of a conformal measure depends exclusively on the recurrence, a property of the potential. Indeed, for a topologically mixing $\Sigma_A$ the equation \eqref{eq:series_Z_n} is precisely the definition of a recurrent potential when $c = \lambda$.

Next, we present the generalized Ruelle-Perron-Frobenius Theorem, a result on the existence of conformal measures on Markov shifts in the sense of Sarig, based on the recurrence of the potential.

\begin{theorem}[Generalized RPF Theorem]\label{thm:RPF_Generalized} Consider $\Sigma_A$ topologically mixing and $F$ a potential on $\Sigma_A$ satisfying the Walters condition such that $P_G(F)< \infty$. The following is true:
\begin{itemize}
    \item[$(a)$] $F$ is recurrent if and only if there exists $\lambda>0$, $h:\Sigma_A \to \mathbb{R}$ non-negative and continuous function, and $\nu$ conservative measure, strictly positive and finite on cylinders, such that $L_f h = \lambda h$ and $L_F^*\nu = \lambda \nu$. In this case, $\lambda = e^{P_G(F)}$;
    \item[$(b)$] $F$ is positive recurrent if and only if there exists $\lambda>0$, $h:\Sigma_A \to \mathbb{R}$ non-negative and continuous function and $\nu$ conservative measure, strictly positive and finite on cylinders, such that $L_f h = \lambda h$, $L_F^*\nu = \lambda \nu$ and $\int h d\nu < \infty$. In this case, $\lambda = e^{P_G(F)}$. Moreover, for each cylinder $[a]$ we have
    \begin{equation}\label{eq:convergence_PR_power_L}
        \lambda^{-n}L_F^n \mathbbm{1}_{[a]} \to \frac{h \nu[a]}{\int h d\nu}
    \end{equation}
    point-wise in $\Sigma_A$ and also uniformly on compact sets;
    \item[$(c)$] $F$ is null recurrent if and only if there exists $\lambda>0$, $h:\Sigma_A \to \mathbb{R}$ non-negative and continuous function and $\nu$ conservative measure,finite on cylinders, such that $L_f h = \lambda h$, $L_F^*\nu = \lambda \nu$ and $\int h d\nu = \infty$. In this case, $\lambda = e^{P_G(F)}$. Moreover, for each cylinder $[a]$ we have
    \begin{equation}\label{eq:convergence_NR_power_L}
        \lambda^{-n}L_F^n \mathbbm{1}_{[a]} \to 0
    \end{equation}
    uniformly on compact sets;
    \item[$(d)$] $F$ is transient if and only if there is not a conservative measure which is finite on cylinders $\nu$ such that $L_F^*\nu = \lambda \nu$ for some $\lambda > 0$.
\end{itemize}
\end{theorem}

Some important observations must to be made about the statements of theorem above.

\textbf{(1) Conservativity:} the conservativity does not interfere in theorem above because of Theorem \ref{cor:equivalences_conformality_classical} and Theorem \ref{thm:equivalence_conservativity_recurrence}, in the sense that every eigenmeasure which is finite on cylinders necessarily is conservative when $F$ is recurrent.

\textbf{(2) Uniqueness (eigenvalue):} Proposition 3.3 in Sarig's notes is part of the proof of the generalized RPF Theorem and it is stated here as Proposition \ref{prop:recurrence_equivalence_existence_eigenmeasure} below. Remark \ref{remark:any_conservative_eigenmeasure_has_the_same_eigenvalue_exponential_of_Gurevich_pressure}, about the proof of this result, implies that the eigenvalues associated to the eigenmeasures are the same, namely $\lambda = e^{P_G(F)}$.

Next propositions are the Propositions 3.3 and 3.4 of Sarig's notes.

\begin{proposition}[proposition 3.3 of \cite{Sarig2009}] \label{prop:recurrence_equivalence_existence_eigenmeasure} A potential $F$ satisfying Walters's condition is recurrent if and only if there exists a conservative measure $\mu$ which is finite on cylinders such that for some $\lambda > 0$, $L_F^* \mu = \lambda \mu$. In this case $\lambda = e^{P_G(F)}$ and $\mu$ gives any cylinder strictly positive measure.
\end{proposition}

\begin{remark} \label{remark:any_conservative_eigenmeasure_has_the_same_eigenvalue_exponential_of_Gurevich_pressure} In the proof of proposition above, for \textbf{any} conservative eigenmeasure which is finite on cylinders, we necessarily have that its associated eigenvalue is $\lambda = e^{P_G(F)}$. And then this implies that $F$ is recurrent. So observe that for recurrent potentials, not only there exists a conservative eigenmeasure, but also every eigenmeasure has the exponential of the Gurevich's pressure as its eigenvalue.
\end{remark}

\textbf{(3) Positive mass:} the measures of the RPF generalized Theorem are \emph{strictly positive} on cylinders.

In particular, we present the following result about the existence of the eigenfucntion which is part of the proof of the generalized RPF Theorem \ref{thm:RPF_Generalized}. 

\begin{proposition}[Proposition 3.4 of \cite{Sarig2009}]\label{prop:recurrence_implies_eigenfunction_associated_to_conservative_eigenmeasure} If $F$ is a recurrent potential satisfying Walters' condition and $\mu$ is a conservative eigenmeasure which is finite on cylinders. Then there exists a positive continous function $h$ s.t.
\begin{itemize}
 \item[1.] $L_F h = e^{P_G(F)}h$;
 \item[2.] $\Var_k[\log h] \leq \sup_n[\Var_{n+k} F_n] \leq \sum_{\ell \geq k+1}^\infty \Var_\ell F$;
 \item[3.] $\log h$ is uniformly continuous and $\Var_1[\log h] < \infty$.
\end{itemize} 
\end{proposition}

\begin{remark} In the original text, proposition above is for a measure in the conditions of Proposition \ref{prop:recurrence_equivalence_existence_eigenmeasure}. We simply explicited it. 
\end{remark}

The next proposition is part of Proposition 3.5 of \cite{Sarig2009}.

\begin{proposition}\label{prop:PR_implies_ACIM_finite} Suppose F satisifies Walters' condition and it is positive recurrent, and let $h$ and $\mu$ as in Proposition \ref{prop:recurrence_equivalence_existence_eigenmeasure} and Proposition \ref{prop:recurrence_implies_eigenfunction_associated_to_conservative_eigenmeasure}. Then $\int h d\mu < \infty$.
\end{proposition}



The next result is the unidimensionality of the eigenmeasures for positive recurrent potentials.

\begin{proposition}\label{prop:eigenmeasures_for_positive_recurrent_potentials_dimension_1} Consider $\Sigma_A$ topologically mixing and $F$ a potential on $\Sigma_A$ satisfying the Walters\footnote{Same valid for summable variations.} condition such that $P_G(F)< \infty$. If $F$ is positive recurrent, then the family of eigen-functions and eigen-measures associated to $\lambda = e^{P_G(F)}$ have both dimension 1.
\end{proposition}

\begin{proof} Let $\mu$ and $\nu$ eigen-measures associated to $f$ and $h$ respectively. If $F$ is positive recurrent, then by \eqref{eq:convergence_PR_power_L} we have that $\frac{\nu[a]}{\mu[a]}$ is constant for every $a \in \mathbb{N}$. Indeed, for each $x \in \Sigma_A$ and each $a \in \mathbin{N}$ we have
\begin{align}
    \lambda^{-n}L_F^n \mathbbm{1}_{[a]}(x) &\to \frac{f(x) \mu[a]}{\int f d\mu} \label{eq:convergence_PR_dim_one_mu},\\
    \lambda^{-n}L_F^n \mathbbm{1}_{[a]}(x) &\to \frac{h(x) \nu[a]}{\int h d\nu} \label{eq:convergence_PR_dim_one_nu}.
\end{align}
By dividing the second equation above by the first one we obtain
\begin{equation*}
    1 = \frac{h(x)\nu[a]}{\int h d\nu} \frac{\int f d\mu}{f(x)\mu[a]} \implies \frac{\nu[a]}{\mu[a]} = \alpha \frac{f(x)}{h(x)},
\end{equation*}
where $\alpha = \frac{\int h d\nu}{\int f d\mu}>0$. Repeating the calculation for fixed $x$ and $b \in \mathbb{N}$ one gets
\begin{equation*}
    \frac{\nu[a]}{\mu[a]} = \frac{\nu[b]}{\mu[b]} = \alpha \frac{f(x)}{h(x)},
\end{equation*}
therefore $f = \gamma h$ for some $\gamma>0$. In a similar way, we have $\nu[a] = \gamma' \mu[a] \frac{\int h d\nu}{\int f d \mu}$. This can be extended to the algebra of the cylinders. By Caratheodory Extension Theorem it is extendend to the Borel $\sigma$-algebra as well.
\end{proof}

In order to obtain a characterization of the recurrence of the potentials, Sarig proved an important result called Discriminant Theorem \cite{Sarig2001}, stated in this thesis. In order to provide the statement of this result we introduce the notion of induced systems \cite{Sarig2001,Sarig2001_Null}.

\begin{definition}[Induced Markov shift] Given a Markov shift space $\Sigma_A$ with alphabet $S$. Fixed $a \in S$, set $S_{\ind} := \{[w]: w \in \mathfrak{W}^*, w_i = a \iff i = 0, [wa] \neq \emptyset \}$, where $\mathfrak{W}^*$ is the set of finite admissible non-empty words. The induced Markov shift space on $a$ is the set $\Sigma_A^{\ind}(a) := S_{\ind}^{\mathbb{N}_0}$, where the induced shift map $\sigma_{\ind}: \Sigma_A^{\ind}(a) \to \Sigma_A^{\ind}(a)$ is given by
\begin{equation*}
    \sigma_{\ind}(([w^0],[w^1],[w^2],...)) := ([w^1],[w^2],[w^3],...),
\end{equation*}
for every $([w^0],[w^1],[w^2],...) \in \Sigma_A^{\ind}(a)$. Let $\pi_{\ind} :\Sigma_A^{\ind}(a) \to [a]$ the map given by
\begin{equation*}
    \pi_{\ind}(([w^0],[w^1],[w^2],...)) := (w^0w^1w^2...).
\end{equation*}
Given a potential $F:\Sigma_A \to \mathbb{R}$, we define the induced potential on $[a]$, $F^{\ind}: \Sigma_A^{\ind}(a) \to \mathbb{R}$, by
\begin{equation*}
    F^{\ind} := \left(\sum_{k=0}^{\varphi_a - 1} F\circ \sigma^k \right)\circ \pi_{\ind}.
\end{equation*}
The pair $(\Sigma_A^{\ind}(a),F^{\ind})$ is said to be the induced system on $[a]$.
\end{definition}

In other words, the induced system $\Sigma_A^{\ind}(a)$ is a full shift where symbols are a particular cylinders in $[a]$. Besides that, the induced potential value taken on a sequence in the induced space is taken on the Birkhoff sum on the first return of the a priori potential on the word on the original Markov shift by connecting the symbols of the sequence of the induced shift space.

In terms of regularity, a locally H\"{o}lder potential induces a locally H\"{o}lder induced potential. However, a potential with summable variations does not necessarily imply that its induced potential does also have summable variations. Anyway, for the topologically mixing case, the Gurevich pressure is well-defined for the induced system as stated below. 

\begin{lemma}[Lemma 2 of \cite{Sarig2001}] Suppose that $\Sigma_A$ is topologically mixing, $S$ be its alphabet, and let $F:\Sigma_A \to \mathbb{R}$ be a potential with summable variations. Fixed $a \in S$, let $(\Sigma_A^{\ind}(a),F^{\ind})$ be the induced system on $a$. Then the following limit exists for all $[w] \in S_{\ind}$ (although it may be infinite) and is independent of the choice of $[w]$:
\begin{equation*}
    P_G(F^{\ind}) = \lim_n \frac{1}{n} \log Z_n(F^{\ind},[w]).
\end{equation*}
\end{lemma} 

\begin{definition} Let $\Sigma_A$ be topologically mixing and let $F : \Sigma_A \to \mathbb{R}$ have summable variations and finite Gurevich pressure. Fix $a \in S$ and let $(\Sigma_A^{\ind}(a),F^{\ind})$ be the induced system on $[a]$. Set
\begin{equation*}
    p^*_a[F] := \sup\{p : P_G((F + p)^{\ind}) < \infty\}
\end{equation*}
The $a$-discriminant of $F$, denoted by $\Delta_a[F]$ is defined as follows:
\begin{equation*}
    \Delta_a[F] := \sup\{P_G((F + p)^{\ind}) : p<p^*_a[\phi]\} \leq \infty.
\end{equation*}
\end{definition}

Still assuming the topological mixing property, there are some important properties envolving the $a$-discriminant, the Gurevich pressure and the partitions $Z_n^*(F,[a])$, which is essentialy Proposition 3 of \cite{Sarig2001}, presented next.

\begin{proposition}\label{prop:properties_of_discriminant} Let $\Sigma_A$ be topologically mixing and let $F$ be a potential with summable variations and finite Gurevich pressure. Also consider the induced system $(\Sigma_A^{\ind}(a),F^{\ind})$ on $[a]$, where $a \in S$. Then $P_G((F+p)^{\ind})$, seen as a function on $p$, is convex, strictly increasing and continuous in $(-\infty,p_a^*[F])$. Moreover, the following identities hold:
\begin{align}
    \Delta_a[F] &= P_G((F+p_a^*[F])^{\ind}),\label{eq:discriminant_is_Gurevich_pressure_induced_Markov_shift}
\end{align}
\begin{align}
    \left|\Delta_a[F] - \log\left(\sum_{k=1}^\infty R^k Z_k^*(F,[a])\right) \right| &\leq \sum_{k \geq 2} \Var_k(F), \label{eq:absol_discriminant_power_series},
\end{align}
\begin{align}
    p_a^*[F] &= - \lim_n \sup \frac{1}{n} \log Z_n^*(F,[a]),
    \label{eq:p_a_star_is_limsup_partition_Z_star}
\end{align}
\begin{align}
    \left|P_G((F+p_a^*[F])^{\ind}) - \log\left(\sum_{k=1}^\infty e^{kp} Z_k^*(F,[a])\right) \right| &\leq \sum_{k \geq 2} \Var_k(F), \label{eq:absol_Gurevich_pressure_power_series},
\end{align}
where $R$ in the identity \eqref{eq:absol_discriminant_power_series} is the convergence radius of the power series $\sum_{k=1}^\infty \xi^k Z_k^*(\phi,[a])$.
\end{proposition}

\begin{remark}\label{remark:properties_of_discriminant} It is important to observe that the relations \eqref{eq:absol_discriminant_power_series} and \eqref{eq:p_a_star_is_limsup_partition_Z_star} are consequences from the inequality \eqref{eq:absol_Gurevich_pressure_power_series}. Also, for potentials which depend on the first coordinate only, that is, $F(x) = F(x_0)$ for every $x \in \Sigma_A$, we have
\begin{equation}\label{eq:Bernoulli_potentials_variation_zero}
    \sum_{k \geq 2} \Var_k(F) = 0,
\end{equation}
since the $n$-variation is zero for each $n \geq 2$ and it is straightforward that they have summable variations.
\end{remark}

Now we present the Sarig's Discriminant Theorem.

\begin{theorem}[Discriminant Theorem] \label{thm:discriminant_theorem}
Let $\Sigma_A$ be a topologically mixing and suppose $F : \Sigma_A \to \mathbb{R}$ be some potential with summable variations satisfying $P_G(F) < \infty$. For a fixed state $a \in S$,
\begin{enumerate}
   \item The equation $P_G((F + p)^{\ind}) = 0$ has a unique solution $p(F)$ if $\Delta_a[F] \geq 0$ and no solution if $\Delta_a[F] < 0$. The Gurevich pressure of $F$ is given by 
   $$P_G(F) =\begin{cases} -p(F)\quad \text{if}\; \Delta_a[F]\geq 0\\
   -p_a^*(F)\quad \text{if}\; \Delta_a[F]<0
   \end{cases}
    $$

\item  $\phi$ is positive recurrent if $\Delta_a[F] > 0$ and transient if $\Delta_a[F] < 0$. In the case
$\Delta_a[F] = 0$, $F$ is either positive recurrent or null recurrent.
\end{enumerate}
\end{theorem}

\begin{remark} For $\Sigma_A$ topologically mixing and a weakly H\"older continuous potential $F$, we have the following:
\begin{itemize}
    \item if $\Delta_a[F] > 0$ for some $a \in \mathbb{N}$, then $\Delta_b[F] > 0$ for every $b \in \mathbb{N}$;
    \item if $\Delta_a[F] = 0$ for some $a \in \mathbb{N}$, then $\Delta_b[F] = 0$ for every $b \in \mathbb{N}$;
    \item if $\Delta_a[F] < 0$ for some $a \in \mathbb{N}$, then $\Delta_b[F] < 0$ for every $b \in \mathbb{N}$.
\end{itemize}
The last statement is straightforward, since $\Delta_a[F] < 0$ if and only if $F$ is transient. The remaining cases are consequences of the potential $F$ having or not the spectral gap property (see \cite{Sarig2015} for further details): $\Delta_a[F] > 0$ if and only if the potential is positive recurrent and it has the spectral gap property; and $\Delta_a[F] = 0$ is equivalent to state that $F$ is recurrent and it has not the spectral gap property.
\end{remark}

In particular, for the renewal shift we have the following result.

\begin{theorem}[Theorem 5 of \cite{Sarig2001}]\label{thm:phase_transition_renewal_shift_Sarig_Gurevich_pressure} Let $\Sigma_A$ be the renewal shift and consider a potential $F : \Sigma_A \to \mathbb{R}$ with summable variations satisfying $\sup F < \infty$ and such that $F^{\ind}$ is locally H\"{o}lder continuous. Then
there exists $0 < \beta_c \leq \infty$ such that:
\begin{enumerate}
    \item $\beta F$ is strongly positive recurrent\footnote{For a potential $F$, we say $F$ is strongly positive recurrent when $\Delta_a[F] > 0$ for some $a \in S$.} for $0 < \beta < \beta_c$ and transient for $\beta > \beta_c$;
    \item $P_G(\beta F)$ is real analytic in $(0, \beta_c)$ and linear in $(\beta_c, \infty)$. It is continuous but not
analytic at $\beta_c$ (in case $\beta_c < \infty)$.
    \item Set $A_n := e^{\sup\{F_n(x):x\in [0,n-1,... ,0]\}}$ and let $R(\beta)$ be the radius of convergence of
    \begin{equation*}
        F_\beta(t) := \sum_{n \in \mathbb{N}} A_n^\beta t^n.
    \end{equation*}
    If $F_\beta(R(\beta))$ is infinite for every $\beta$, then $\beta_c = \infty$. If there exists $\beta > 0$ such that $F_\beta(R(\beta)) < 1$, then $\beta_c < \infty$.
\end{enumerate}
\end{theorem}

Next we present a potential on the renewal Markov shift space, which is one of the essential examples of this thesis. This potential was suggested by Elmer R. Beltr\'an.

\begin{example}\label{exa:renewal_potential_log}

Let $\Sigma_A$ be the renewal shift and consider the potential
\begin{equation}\label{eq:conservativity_phase_transition_potential}
    F(x) = \log (x_0) -\log(x_0+1).
\end{equation} 
In order to study the thermodynamic properties in this model, we actually consider the potential $\beta F$, where $\beta >0$ is the inverse of the temperature. Such potential depends on the first coordinate only, and by the remark \ref{remark:properties_of_discriminant}, the sum of its variations is zero. Moreover, as it is proved in Example \ref{exa:renewal_shift}, $\Sigma_A$ is topologically mixing. Then, by the inequality \eqref{eq:absol_discriminant_power_series} in Proposition \ref{prop:properties_of_discriminant} and the identity \eqref{eq:Bernoulli_potentials_variation_zero} for $a = 1$, we get
\begin{equation}\label{eq:discriminant_1_equals_series}
    \Delta_1[\beta F] = \log\left(\sum_{k=1}^\infty R^k Z_k^*(\beta F,[1])\right),
\end{equation}
and we can calculate the $1$-discriminant by calculating the series above. For every $n$-periodic sequence $x \in \Sigma_A$ we have
\begin{align*}
    \mathbbm{1}_{[\varphi_1 = n]}(x) = 1 &\iff \varphi_1(x) = n \iff x \in [1] \text{ and } \inf\{m\geq 1:\sigma^mx \in [1]\} = n\\
    &\iff x = \overline{1,n,n-1,...,2}.
\end{align*}
Note that $\sigma^n(\overline{1,n,n-1,...,2}) = \overline{1,n,n-1,...,2}$. Then,
\begin{align}\label{eq:Z_star_renewal}
    Z_n^*(\beta F,[1]) = e^{\beta F_n(\overline{1,n,n-1,...,2})}.
\end{align}
The equality above can be used to calculate the convergence radius $R$ of the series in the RHS of \eqref{eq:discriminant_1_equals_series}. Indeed, we get
\begin{align*}
    R = \lim_{n \to \infty} \frac{Z_n^*(\beta F,[1])}{Z_{n+1}^*(\beta F,[1])} = \lim_{n \to \infty} \frac{e^{\beta F_n(\overline{1,n,n-1,...,2})}}{e^{\beta F_{n+1}(\overline{1,n+1,n,...,2})}},
\end{align*}
and since 
\begin{align*}
        F_n(\overline{1,n,n-1,...,2}) = \sum_{k=0}^{n-1} F(\sigma^k(\overline{1,n,...,2})) = \sum_{k = 1}^n\left[\log(k) - \log(k+1)\right] = - \log(n+1), \quad n \in \mathbb{N},
\end{align*}
it follows that 
\begin{align*}
    Z_n^*(\beta F,[1]) = \frac{1}{(n+1)^\beta} \quad \text{and} \quad R = \lim_{n \to \infty} \frac{e^{ -\beta \log(n+1)}}{e^{-\beta  \log(n+2)}} = \lim_{n \to \infty} e^{-\beta \log \left(\frac{n+1}{n+2}\right)}  = 1.
\end{align*}
Now, we calculate $\Delta_1[\beta F]$:
\begin{equation}\label{eq:discriminant_identity}
    \Delta_1[\beta F] = \log\left(\sum_{k=1}^\infty R^k Z_k^*(\beta F,[1])\right) =  \log\left(\sum_{k=1}^\infty \frac{1}{(k+1)^\beta}\right).
\end{equation}
So we obtain
\begin{align}\label{eq:discriminant_final}
    \Delta_1[\beta F] &= \log\left(\sum_{k=1}^\infty \frac{1}{k^\beta}-1\right).
\end{align}
Let $\beta_c >0$ be the unique solution\footnote{$\beta_c \approx 1.72865$.} of $\zeta(\beta_c) = 2$, where $\zeta$ is the Riemann Zeta function. For $0 < \beta \leq 1$, the series in \eqref{eq:discriminant_final} diverges and therefore $\Delta_1[F] = \infty > 0$. Now, for $\beta > 1$, we may write
\begin{equation*}
    \Delta_1[\beta F] = \log\left(\zeta(\beta)-1\right).
\end{equation*}
Since the $\zeta$ is strictly decreasing for $\beta >1$, we have that $\zeta(\beta) - 1 > 1$ for $1 < \beta < \beta_c$ and in this case we obtain $0 < \Delta_1 [\beta F] < \infty$. Also,
\begin{equation*}
    \Delta_1 [\beta_c F] =  0.
\end{equation*}
Furthermore, observe that $\zeta(\beta) > 1$ for every $\beta > 1$ and it has $1$ as its horizontal asymptote. Then, for $\beta > \beta_c$ we have that $0 < \zeta(\beta) - 1 < 1$, and therefore $\Delta_1[\beta F] < 0$ for $\beta > \beta_c$. Therefore, By the Discriminant Theorem, $\beta F$ is positive recurrent for $\beta < \beta_c$ and it is transient for $\beta > \beta_c$. Consequently, by the Generalized RPF Theorem, there exists a conservative eigenmeasure of the Ruelle's operator $L_{\beta F}$ which is finite on cylinders for $\beta \in (0,\beta_c)$ and we have the absence of such measures for $\beta \in (\beta_c,\infty)$.
Now note that since $F$ depends only on the first coordinate, we may consider its extension on the real strictly positive numbers on the first symbol, that is, 
\begin{equation*}
    F(y) = \log (y) - \log(y+1), \quad y \in \mathbb{R}_+^*.
\end{equation*}
Moreover,
\begin{equation*}
    \dfrac{dF}{dy}(y) = \frac{1}{y} - \frac{1}{y+1}>0, \quad \text{for all } y \in \mathbb{R}_+^*.
\end{equation*}
Then $F$ is a strictly increasing function on $y$, and then
\begin{equation}\label{eq:sup_potential_phase_transition_bounded}
    \sup_{x \in \Sigma_A} F(x) = \sup_{y \in \mathbb{R}_+^*} F(y) = \lim_{y \to \infty} F(y) = 0 < \infty,
\end{equation}
and hence $\sup F < \infty$. On the other hand we claim that $F^{\ind}$ is locally H\"{o}lder. In fact, for every $x \in \Sigma_A^{\ind}(1)$ we have that $x = [w^0],[w^1][w^2]...$, where, for each $n \in \mathbb{N}_0$, $w^n = w^n(x)$ is a non-empty admissible word starting with $1$ such that does not have the symbol $1$ in any other position which is not the first letter, satisfying $A(w^n_{|w^n|-1},1) = 1$. Observe that $\varphi_1\left(\pi^{ind}(x)\right) = |w^0(x)|$. Then,
\begin{equation}\label{eq:F_ind_renewal_potential_phase_transition}
    F^{\ind}(x) = \sum_{k = 0}^{|w^0(x)|-1}F \circ \sigma^k(w^0(x)w^1(x)w^2(x)\cdots).
\end{equation}
Given $x,y \in \Sigma_A^{\ind}(1)$, for each $m \in \mathbb{N}_0$, the identity $x_m = y_m$ in terms of induced system means that $w^m(x) = w^m(y)$, and hence
\begin{equation*}
    \Var_k F^{\ind} = \sup\left\{\left|F^{\ind}(x)-F^{\ind}(y)\right|: w^m(x) = w^m(y), m \in \{0,...,k-1\}\right\}.
\end{equation*}
For the potential \eqref{eq:conservativity_phase_transition_potential} and every $x, y \in \Sigma_A^{\ind}(1)$ such that $w^0(x) = w^0(y) = w^0$, it follows that $\varphi_1\left(\pi^{ind}(x)\right) = \varphi_1\left(\pi^{ind}(y)\right) = |w^0|$ and then 
\begin{align*}
    \left|F^{\ind}(x)-F^{\ind}(y)\right| &= \left|\sum_{k = 0}^{|w^0|-1}\left[F \circ \sigma^k(w^0w^1(x)w^2(x)\cdots)-F^{\ind}\circ\sigma^k(w^0w^1(y)w^2(y)\cdots)\right]\right|\\
    &= \left|\sum_{k = 0}^{|w^0|-1}\left[\log\left(\frac{w^0_k}{w^0_k+1}\right) - \log\left(\frac{w^1_0}{w^0_k+1}\right) \right]\right| = 0,
\end{align*}
and therefore
\begin{equation*}
    \Var_k F^{\ind} = 0,
\end{equation*}
for every $k \geq 2$, and so it is straightforward that $F^{\ind}$ is locally H\"{o}lder. This regularity and the inequality \eqref{eq:sup_potential_phase_transition_bounded} shows that the potential $F$ satisfies the hypotheses of Theorem \ref{thm:phase_transition_renewal_shift_Sarig_Gurevich_pressure}. Then $P_G(\beta F)$ is linear on the variable $\beta$ for $\beta \geq \beta_c$. In other words, $P_G(\beta F) = \beta p_1^*[F]$. By item 2 of Discriminant Theorem and equation \eqref{eq:p_a_star_is_limsup_partition_Z_star} from Proposition \ref{prop:properties_of_discriminant}, we have that
\begin{equation} \label{eq:coefficient_pressure_zero_for_beta_greater_than_critical}
    p_1^*[F] = \limsup_n \frac{1}{n} \log Z_n^*(F,1) = \limsup_n \frac{1}{n} \log \left(\frac{1}{n+1}\right) = 0.
\end{equation}
\end{example}

\chapter{C\texorpdfstring{$^*$}{TEXT}-algebras}
\label{ch:C_star_algebras}

In this chapter we introduce some important concepts about $C^*$-algebras that are crucial in this work. We introduce the general algebraic setting such as definition of C$^*$-algebras, morphisms, universal algebras. In addition, we also present the basics of the dynamical setting on these algebras, such as the notions of C$^*$-dynamical systems and the construction of the Kubo-Martin-Schwinger (KMS) states and some of its properties.

\section{Algebras, Banach Algebras and C$^*$-Algebras}
\label{sec:c-star-basics}

From now, we assume that the reader has some familiarity with some topics on functional analysis, specially on theorems about normed and Banach vector spaces. The basic notions of the C$^*$-algebras are based on \cite{Murphy1990,Davidson1996}, while the construction of universal algebras are mostly based on \cite{Blackadar1985,Blackadar2006,Tasca2015}. For the KMS theory we used the references \cite{Bratteli1987vol1,Bratteli1996vol2}. 

For any vector over a field space we will denote the usual conventions for addition $+$ and scalar multiplication $\cdot$.

\begin{definition}[Algebra]\label{def:algebra} The $4$-tuple $(A,+,\cdot, \circ)$ is said to be an algebra over a field $\mathbb{K}$ ($\mathbb{K} =\mathbb{C}$ or $\mathbb{R}$) if $(A,+,\cdot)$ is a vector space over $\mathbb{K}$ and $\circ$ is an operation, namely
\begin{align*}
   \circ: A \times A \to A, \\
    (a,b) \mapsto a \circ b =: ab,
\end{align*}
called the (algebra) product and it satisfies the following properties.

\begin{itemize}
    \item \emph{Associativity}\footnote{Some texts define algebras without associativity. In the context of operator algebras, the algebras are associative by default.}: for every $a,b,c \in A$,
    \begin{equation*}
        a(bc) = (ab)c;
    \end{equation*}
    \item \emph{distribution over vector addition}: for any $a,b,c \in A$,
    \begin{align*}
        a(b+c) = ab + ac, \\
        (a+b)c = ac + bc;
    \end{align*}
    \item \emph{commutativity with relation to the scalar multiplication}: for all $a,b \in A$ and $\lambda \in \mathbb{K}$,
    \begin{equation*}
        \lambda(ab) = (\lambda a)b = a(\lambda b).
    \end{equation*}
\end{itemize}

An algebra is said to be \emph{commutative} if for every $a,b \in A$,
\begin{equation*}
    ab = ba.
\end{equation*}

An algebra is said to be \emph{unital} if there exists $1 \in A$ s.t. for all $a \in A$ 
\begin{equation*}
    1a = a1 = a.
\end{equation*}
In this case $1$ is called \emph{unity} of $A$.
\end{definition}

From now we will refer the algebra $(A,+,\cdot,\circ)$ simply by $A$. Also, the field will be always $\mathbb{C}$ except if we are specifying the field we are working with it. Most of the facts we prove or mention here are easily particularized to $\mathbb{R}$.

\begin{definition} A subspace $B$ of an algebra $A$ is a \emph{subalgebra} if it is algebraically closed with relation to the algebra product, i.e., $ab \in B$ for every $a,b \in B$.
\end{definition}

\begin{example}\label{exa:algebras} We present some examples of algebras:
\begin{itemize}
    \item[$(a)$] consider a topological space $X$. The set $C(X)$ of all complex continuous functions on $X$ is a commutative unital algebra over $\mathbb{C}$ with the algebra product being the pointwise product
    \begin{equation*}
        (fg)(x):= f(x)g(x), \quad f,g \in C(X), \text{ } x \in X;
    \end{equation*}
    \item[$(b)$] let $\mathcal{H}$ be a Hilbert space. The set $\mathfrak{B}(\mathcal{H})$ of the bounded linear operators on $\mathcal{H}$ is an unital algebra with the matrix product being the composition of operators. In particular, given $n \in \mathbb{N}$, the vector space $M_n(\mathbb{C})$ of the $n \times n$ matrices with complex entries endowed with the matrix product as the algebra product is a non-commutative unital algebra over $\mathbb{C}$; 
\end{itemize}

\end{example}

\begin{definition}[Involution] Let $A$ be an algebra. An involution on $A$ is a unary operation $*:A \to A$ satisfying the following: given $a,b,c \in A$ and $\lambda \in \mathbb{C}$, then
\begin{itemize}
    \item[$(i)$] $(\lambda a+b)^* = \overline{\lambda} a^* + b^*$;
    \item[$(ii)$] $(ab)^* = b^*a^*$;
    \item[$(iii)$] $(a^*)^* = a$.
\end{itemize} 
\end{definition}

\begin{remark}\label{remark:involution_0_1} It is straightforward that $a^* = 0$ if and only if $a = 0$. Also, on a unital $*$-algebra we necessarily have $1^* = 1$. 
\end{remark}

\begin{definition} A $*$-algebra is an algebra endowed with an involution. Given a $*$-algebra $A$, a $*$-subalgebra $B$ of $A$ is a subalgebra of $A$ which is a $*$-algebra with respect to the involution on $A$ restricted to $B$. 
\end{definition}

\begin{example}\label{exa:star_algebras} We introduce involutions on the algebras of the Example \ref{exa:algebras}, turning them into $*$-algebras:

\begin{itemize}
    \item[$(a)$] for a topological space $X$, the algebra $C(X)$ endowed with the involution assinged by
    \begin{equation*}
        (f^*)(x):= \overline{f(x)}, \quad f\in C(X), \text{ } x \in X,
    \end{equation*}
    is a $*$-algebra;
    \item[$(b)$] for a given Hilbert space $\mathcal{H}$, the algebra $\mathfrak{B}(\mathcal{H})$ can be endowed with the involution that maps each operator to its adjoint operator, that is, for $T \in \mathfrak{B}(\mathcal{H})$, $T^* \in \mathfrak{B}(\mathcal{H})$ is the unique operator such that
    \begin{equation*}
        (Ax,y) = (x,A^*y)
    \end{equation*}
    for every $x,y \in \mathcal{H}$, where $(\cdot, \cdot)$ is the inner product. In particular, the complex algebra $M_n(\mathbb{C})$, $n \in \mathbb{N}$, admits an involution as follows. For each $a = (a_{ij}) \in M_n(\mathbb{C})$, let $a^* \in M_n(\mathbb{C})$ defined by $a^*_{ij}:= (\overline{a}_{ji})$. The assignment $a \mapsto a^*$ defines an involution on $M_n(\mathbb{C})$, turning this algebra into a $*$-algebra;
\end{itemize}
\end{example}

Not every subalgebra is a $*$-subalgebra, as we show in the next example.

\begin{example} Consider the $*$-algebra $M_2(\mathbb{C})$ and take the subalgebra
\begin{equation*}
    B=\left\{\begin{pmatrix}
                a & b\\
                0 & 0
           \end{pmatrix} : a,b \in \mathbb{C}
    \right\}.
\end{equation*}
Suppose that there exists an involution $\varphi:M_2(\mathbb{C}) \to M_2(\mathbb{C})$ which its restriction to $B$ is also an involution. Then,
\begin{align*}
    \varphi \begin{pmatrix}
                0 & 1\\
                0 & 0
           \end{pmatrix}
    = \begin{pmatrix}
          c & d\\
          0 & 0
       \end{pmatrix},
\end{align*}
for some $c,d \in \mathbb{C}$. So,
\begin{align*}
    \left[\varphi \begin{pmatrix}
                0 & 1\\
                0 & 0
           \end{pmatrix}\right]^2
    = \begin{pmatrix}
          c^2 & cd\\
          0 & 0
       \end{pmatrix}.
\end{align*}
However, since $\varphi$ is an involution, we necessarily have
\begin{align}\label{eq:inv_squared_zero}
    \left[\varphi \begin{pmatrix}
                0 & 1\\
                0 & 0
           \end{pmatrix}\right]^2
    = \left[\varphi \begin{pmatrix}
                0 & 1\\
                0 & 0
           \end{pmatrix}\right]
    \left[\varphi \begin{pmatrix}
                0 & 1\\
                0 & 0
           \end{pmatrix}\right]
    = \varphi \left[ \begin{pmatrix}
                        0 & 1\\
                        0 & 0
                     \end{pmatrix}
              \begin{pmatrix}
                0 & 1\\
                0 & 0
              \end{pmatrix}\right]
    = \varphi \begin{pmatrix}
                  0 & 0\\
                  0 & 0
              \end{pmatrix} = 0,
\end{align}
and then $c = 0$ and hence,
\begin{align*}
    \varphi \begin{pmatrix}
                0 & 1\\
                0 & 0
           \end{pmatrix}
    = \begin{pmatrix}
          0 & d\\
          0 & 0
       \end{pmatrix}
    = d\begin{pmatrix}
          0 & 1\\
          0 & 0
       \end{pmatrix}.
\end{align*}
By \eqref{eq:inv_squared_zero} we have
\begin{align*}
    0 = \left[\varphi \begin{pmatrix}
                0 & 1\\
                0 & 0
           \end{pmatrix}\right]^2
    = \overline{d}\varphi \begin{pmatrix}
                0 & 1\\
                0 & 0
           \end{pmatrix}
    = |d|^2 \begin{pmatrix}
                0 & 1\\
                0 & 0
           \end{pmatrix},
\end{align*}
and then $d = 0$ and therefore
\begin{align*}
    \varphi \begin{pmatrix}
                0 & 1\\
                0 & 0
           \end{pmatrix} = 0,
\end{align*}
which is a contradiction due to Remark \ref{remark:involution_0_1}. Then $B$ cannot admit an involution and therefore it is not a $*$-subalgebra of $M_2(\mathbb{C})$.
\end{example}

\begin{definition} Given a $*$-algebra $A$, an element $a \in A$ is said to be
\begin{itemize}
    \item[$(i)$] self-adjoint when $a^* = a$;
    \item[$(ii)$] idempotent when $a^2 = a$;
    \item[$(iii)$] a projection when it is self-adjoint and idempotent;
    \item[$(iv)$] an isometry when $A$ is unital and $a^*a = 1$;
\end{itemize}
In addition, given a set $Y \subseteq A$, we define $Y^* = \{y^*: y \in Y\}$, and we say that $Y$ is self-adjoint when $Y = Y^*$. 
\end{definition}

\begin{proposition} For every element $a$ in a $*$-algebra $A$, there exist unique self-adjoint elements $b,c \in A$ such that $a = b +ic$.
\end{proposition}

\begin{proof} The existence is straightforward by taking
\begin{equation*}
    b = \frac{a + a^*}{2} \quad \text{and} c = \frac{a - a^*}{2}.
\end{equation*}
Now, for the uniqueness, suppose that there exist $b',c' \in A$ self-adjoint elements satisfying $a = b' + ic'$. Then
\begin{equation}\label{eq:decomposition_self_adjoint}
  (b - b') + i(c - c') = 0.
\end{equation}
By applying the involution in the equation above and the self-adjoint property of $b,b',c$ and $c'$, we obtain
\begin{equation}\label{eq:decomposition_self_adjoint_inv}
  (b - b') - i(c - c') = 0.
\end{equation}
By summing the equation \eqref{eq:decomposition_self_adjoint_inv} in \eqref{eq:decomposition_self_adjoint} we obtain $b = b'$, and by subtracting \eqref{eq:decomposition_self_adjoint_inv} from \eqref{eq:decomposition_self_adjoint} we get $c = c'$.
\end{proof}

\begin{definition}[Normed algebra] Given an algebra $A$ such that the vector space $(A,+,\cdot)$ has a norm $\|.\|$, we say that $(A,+,\cdot, \circ, \|\cdot\|)$ is a normed algebra if its submultiplicave with relation to the algebra product: for every $a,b \in A$,
\begin{equation*}
    \|ab\| \leq \|a\| \|b\|, \quad a,b \in A.
\end{equation*}
We also will denote the normed algebra $(A,+,\cdot, \circ, \|\cdot\|)$ simply by $A$. A subalgebra of a normed algebra $A$ that is closed in norm is said to be a normed subalgebra of $A$.
\end{definition}

\begin{example}\label{exa:normed_algebras} Using the same list of examples as in Example \ref{exa:algebras} and its respective enumeration, we present norms that turn those algebras into normed algebras.

\begin{itemize}
    \item[$(a)$] Both vector space norms
    \begin{equation*}
        \|a\|:= \sup\{\|av\|: v \in \mathbb{C}^n,\text{ } \|v\|\leq 1\}, \quad a \in M_n(\mathbb{C}),
    \end{equation*}
    and 
    \begin{equation*}
        \|a\|_1:= \sum_{i,j = 1}^n|a_{ij}|, \quad a \in M_n(\mathbb{C}),
    \end{equation*}
    make $M_n(\mathbb{C})$ a normed algebra\footnote{Note that both norms generate the same topology, since $M_n(\mathbb{C})$ is a finite dimension vector space.};
    \item[$(b)$] here we consider $C_b(X)$ instead of $C(X)$, the set of all bounded complex continuous functions on $X$. The vector space norm
    \begin{equation*}
        \|f\| := \sup_{x \in X} |f(x)|, \quad f \in C(X),
    \end{equation*}
    makes $C_b(X)$ a normed algebra;
    \item[$(c)$] the vector space norm
    \begin{equation*}
        \|T\| := \sup_{x \in \mathcal{H}} \frac{\|Tx\|}{\|x\|}, \quad T \in \mathfrak{B}(\mathcal{H}),
    \end{equation*}
    makes $\mathfrak{B}(\mathcal{H})$ a normed algebra.
\end{itemize}

\end{example}

Given a set $X$ and a metric $d: X \times X \to [0,\infty)$, we recall that the metric space $(X,d)$ is said to be complete if every Cauchy sequence in $X$ converges with respect to $d$. Also, we recall that any norm $\|\cdot\|$ on a vector space $X$ induces a metric, given by
\begin{equation*}
    d_{\|\cdot\|}(x,y):= \|x-y\|, \quad x,y \in X.
\end{equation*}
The completeness structure on the Banach vector spaces is directly transfered to normed algebras as in definition below.

\begin{definition}[Banach algebra] A normed algebra $A$ is said to be a Banach algebra if it is complete with respect to its norm, i.e. $(A,+,\cdot, ||\cdot||)$ is a Banach space. If $A$ is a Banach algebra endowed with an involution $*$ such that $\|a^*\| = \|a\|$ for every $a \in A$, then we call it a Banach $*$-algebra.
\end{definition}

\begin{example}\label{exa:Banach_algebras} We recall the normed algebras of Example \ref{exa:normed_algebras} and once again we keep the listing order.

\begin{itemize}
    \item[$(a)$] Both norms for the normed algebra $M_n(\mathbb{C})$ in the item $(a)$ of Example \ref{exa:normed_algebras} are Banach algebras;
    \item[$(b)$] here we consider $X$ a locally compact Hausdorff topological space and $C_0(X)$, the set of all complex continuous functions on $X$ which vanishes at the infinity, i.e. 
    \begin{equation*}
        C_0(X):= \left\{f \in C(X): \mathfrak{C}(f,\epsilon) \text{ is compact for all }\epsilon>0  \right\},
    \end{equation*}
    where 
    \begin{equation*}
        \mathfrak{C}(f,\epsilon):=\{x \in X: |f(x)|\geq \epsilon\}.
    \end{equation*}
    $C_0(X)$ is a Banach algebra;
    \item[$(c)$] $\mathfrak{B}(\mathcal{H})$ is a Banach algebra.
\end{itemize}
Moreover, if we endow these Banach algebras with the respective involutions of Example \ref{exa:star_algebras}, they become Banach $*$-algebras.
\end{example}

\begin{definition}[C$^*$-algebra] A C$^*$-algebra $A$ is a $*$-Banach algebra that satisfies the C$^*$-property, namely
\begin{equation*}
    \|a^*a\| = \|a\|^2,
\end{equation*}
for every $a \in A$. 
\end{definition}

\begin{definition}[Quotient vector space] Let $A$ be a vector space and $B \subseteq A$ a vector subspace. We define the equivalence relation $\sim_B$ on $A$ as follows: given $x,y \in A$, we say that $x$ is equivalent to $y$, denoted by $x \sim_B y$ when $x-y \in B$. It is straightforward that the equivalence classes of $\sim_B$ are
\begin{equation*}
    [x] := x + B = \{x + b:b \in B \}.
\end{equation*}
The quotient space $A/B$ is the set of all equivalence classes as above and it has the vector space structure as next. Given $x,y \in A$ and $\lambda \in \mathbb{C}$ we define the addition on $A/B$ as
\begin{equation*}
    (x+B) + (y+B) := (x+y)+B,
\end{equation*}
and its product by scalar as
\begin{equation*}
    \lambda (x+B):= (\lambda x) + B. 
\end{equation*}
\end{definition}

\begin{definition}[Ideals] Given an algebra $A$, let $I \subseteq A$ be a subspace. $I$ is said to be a left ideal of $A$ if 
\begin{equation*}
    a \in A\text{, } b \in I \implies ab \in I. 
\end{equation*}
Analogously, $I$ is said to be a right ideal of $A$ if 
\begin{equation*}
    a \in A\text{, } b \in I \implies ba \in I. 
\end{equation*}
Also, we say that $I$ $I$ is a two-sided ideal of $A$ if $I$ is a left and a right ideal of $A$. The two-sided ideals $\{0\}$ and $A$ are said to be trivial ideals. A two-sided ideal is said to be maximal if it is a proper ideal of $A$ that is not contained in any other proper two-sided ideal of $A$.
\end{definition}

\begin{example}\label{exa:ideal} Let $X$ be a compact topological space and consider the Banach algebra $C(X)$ with the usual operations and the supremum norm and let $y \in X$. The set
\begin{equation*}
    M_y := \{f \in C(X): f(y) = 0\}
\end{equation*}
is a two-sided closed\footnote{in the norm topology.} ideal of $C(X)$ of codimension\footnote{Given a vector space $V$ and $S$ a subspace of $V$, the codimension of $S$ is the dimension of the quotient space $V/S$.} $1$ and therefore it is a maximal ideal of $C(X)$.
\end{example}

\begin{theorem} \label{thm:ideal_quotient_star_algebra} Let $A$ be a $*$-algebra and $I \subseteq A$ be a two-sided self-adjoint ideal. Define for every $(a+I),(b+I) \in A/I$ the assignments
\begin{equation}\label{eq:prod_quotient}
    (a+I)(b+I):= ab +I
\end{equation}
and
\begin{equation}\label{eq:inv_quotient}
    (a+I)^*:= a^* +I.
\end{equation}
Then, \eqref{eq:prod_quotient} and \eqref{eq:inv_quotient} define operations of product and involution on $A/I$ respectively, and therefore it turns $A/I$ into a $*$-algebra. 
\end{theorem}

\begin{proof} Since $A/I$ is a vector space, it remains to prove that the operations \eqref{eq:prod_quotient} and \eqref{eq:inv_quotient} are well-defined and satisfy the axioms of product and involution, respectively. Since the axioms are straightforward, we only prove that such operations are well defined. For the product let $x,y \in A$ such that $x+I = a+I$ and $y+I = b+I$. These equalities are true if and only if $x = a + h_x$ and $y = b + h_y$, $h_x, h_y \in I$. For every $c \in A$ we have the following equivalences:
\begin{align*}
    c \in xy + I \iff c \in (a + h_x)(b + h_y) +I \iff c \in ab + \underbrace{ah_y + bh_x}_{\in I} +I \iff c \in ab+I,
\end{align*}
and therefore $xy + I = ab+I$, that is, the product is well defined. Now, for the involution, for any $x \in A$ with $x+I = a+I$, we have the following equivalences for any $c \in A$.
\begin{align*}
    c \in x^* + I \iff c \in (a + h_x)^* +I \iff c \in a^* + \underbrace{x^*}_{\in I} +I \iff c \in a^*+I,
\end{align*}
and therefore $x^* + I = a^* + I$ and then the involution is also well defined. Consequently, $A/I$ is a $*$-algebra.
\end{proof}

\begin{definition}[Partial isometry] Let $A$ be a C$^*$-algebra and $a \in A$. We say $a$ is a partial isometry when $a^*a$ is a projection. 
\end{definition}

\begin{proposition} For a given C$^*$-algebra $A$ and $a \in A$, the following are equivalent:
\begin{itemize}
    \item[$(i)$] $a$ is a partial isometry;
    \item[$(ii)$] $a = aa^*a$;
    \item[$(iii)$] $aa^*$ is a projection.
\end{itemize}
\end{proposition}

\begin{proof}  We prove the chain $(i) \implies (ii) \implies (iii) \implies (i)$.

\textbf{Proof of $\mathbf{(i) \implies (ii)}$:} let $p = a^*a$ and $z = aa^*a - a$. Then,
\begin{equation*}
    z^*z = (a^*aa^* - a^*)(aa^*a - a) = p^3 - p^2 - p^2 + p = 0,
\end{equation*}
where we used in the last equality above that $p$ is a projection. By the C$^*$-property, we have that $\|z\|^2 = \|z^*z\| = 0$ and therefore $z = 0$, that is, $a = aa^*a$.

\textbf{Proof of $\mathbf{(ii) \implies (iii)}$:} it is straightforward by multiplying the identity in $(ii)$ by $a^*$ from the right.

\textbf{Proof of $\mathbf{(iii) \implies (i)}$:} by similar proof done for $(i) \implies (ii)$ we get $a^* = a^*aa^*$ and by multiplying this result by $a$ from the right we get that $a^*a$ is a projection, that is, $a$ is a partial isometry.
\end{proof}

\begin{definition}[C$^*$-seminorm and norm] A C$^*$-seminorm on a $*$-algebra $A$ is a function $\|\cdot\|:A \to \mathbb{R}_+$ such that, for every $a,b \in A$ and $\lambda \in \mathbb{C}$, we have
\begin{itemize}
    \item[$(i)$] $\|\cdot\|$ is a seminorm, that is, it satisfies
    \begin{itemize}
        \item[$\bullet$] $\|\lambda a\| = |\lambda| \|a\|$,
        \item[$\bullet$] $\|a+b\| \leq \|a\| + \|b\|$;
    \end{itemize}
    \item[$(ii)$] $\|\cdot\|$ is submultiplicative, that is, $\|ab\|\leq \|a\|\|b\|$;
    \item[$(iii)$] $\|a^*\| = \|a\|$;
    \item[$(iv)$] $\|a^*a\| = \|a\|^2$.
\end{itemize}
If $\|\cdot\|$ is a norm instead of a seminorm, we say that such function is a C$^*$-norm.
\end{definition}

\begin{proposition}\label{prop:C_star_seminorm_two_sided_ideal} Let $A$ be a $*$-algebra and $\|\cdot\| : A \to \mathbb{R}_+$ be a C$^*$-seminorm on $A$. Let $N:= \{a \in A: \|a\| = 0\}$. Then, $N$ is a two-sided self-adjoint ideal. 
\end{proposition}

\begin{proof} Let $x \in A$ and $a,b \in N$, then
\begin{equation*}
    0 \leq \|xa\| \leq \|x\|\|a\| = \|x\|\cdot 0 = 0,
\end{equation*}
hence $\|xa\| = 0$ and then $xa \in N$. By similar calculations we also obtain $ax \in N$. On the other hand, for every $\lambda \in \mathbb{C}$ we also have
\begin{equation*}
    0 \leq \|a + \lambda b\| \leq \|a\| + |\lambda| \|b\| = 0 + |\lambda| 0 = 0,
\end{equation*}
and then $a + \lambda b \in N$, that is, $N$ is a vector subspace of $A$. So far we have that $N$ is a two-sided ideal. Now, it also holds that
\begin{equation*}
    0 \leq \|a^*\| = \|a^*\| = 0,
\end{equation*}
and therefore $N$ is self-adjoint.
\end{proof}

\begin{theorem} \label{thm:quotient_star_normed_algebra} Let $A$ be a $*$-algebra and $\|\cdot\| : A \to \mathbb{R}_+$ be a C$^*$-seminorm on $A$. Let $N:= \{a \in A: \|a\| = 0\}$. Then the quotient $A/N$ is a normed $*$-algebra for the C$^*$-norm $\tnorm{ \cdot} : A/N \to \mathbb{R}_+$ defined by
\begin{equation*}
    \tnorm{b + N} := \|b\|, \quad b \in A.
\end{equation*}
\end{theorem}

\begin{proof} By Proposition \ref{prop:C_star_seminorm_two_sided_ideal} we have that $N$ is a two-sided self-adjoint ideal, and Theorem \ref{thm:ideal_quotient_star_algebra} grants that $A/N$ is in fact a $*$-algebra. We claim that $\tnorm{\cdot}$ is a well-defined function. In fact, given $a,b \in A$ such that $a + N = b + N$ we have that $a-b \in N$ and hence $\|a-b\| = 0$. Then,
\begin{equation*}
    0 \leq |\|a\|-\|b\|| \leq \|a-b\| = 0,
\end{equation*}
and therefore $\|a\| = \|b\|$. In particular, $\tnorm{a+N} = 0$ if and only if $\|a\| = 0$, that is, $a \in N$. This proves that $\tnorm{\cdot}$ is in fact a norm, since the remaining requirements for such function be a norm are satisfied as follows:
\begin{align*}
    \tnorm{\lambda a + N} &= \|\lambda a\| = |\lambda| \|a\| = |\lambda| \tnorm{a + N}, \quad \text{for every } \lambda \in \mathbb{C};\\
    \tnorm{\lambda (a+b) + N} &= \|a + b\| \leq \|a\| + \|b\| = \tnorm{a+N} + \tnorm{b+N}.
\end{align*}
The submultiplicativity is also straightforward:
\begin{equation*}
    \tnorm{ab +N} = \|ab\| \leq \|a\|\|b\| = \tnorm{a+N} \tnorm{b+N}.
\end{equation*}
Moreover,
\begin{equation*}
    \tnorm{a^* +N} = \|a^*\| = \|a\| = \tnorm{a+N}
\end{equation*}
and
\begin{equation*}
    \tnorm{a^*a +N} = \|a^*a\| = \|a\|^2 = \tnorm{a+N}^2.
\end{equation*}
Therefore $\tnorm{\cdot}$ is a C$^*$-norm.
\end{proof}

\begin{proposition}\label{prop:completion_star_algebra_C_star_norm} Let $A$ be a $*$-algebra and $\|\cdot\| : A \to \mathbb{R}_+$ be a C$^*$-norm on $A$. The completion $A^{\|\cdot\|}$ of $A$ under $\|\cdot\|$ is a C$^*$-algebra. 
\end{proposition}

\begin{proof} Let $a$ be an element of the completion of $A$ under $\|\cdot\|$. In particular, for every normed algebra, its completion is a Banach algebra. We claim that the involution map is uniformly continuous on $A$. Indeed, let $a,b \in A$ and $\epsilon >0$. If $\|a-b\|< \epsilon$, then
\begin{equation*}
    \|a^*-b^*\| = \|(a-b)^*\| = \|a-b\| <\epsilon,
\end{equation*}
which proves the claim. Then, since $A$ is dense on its completion, there exists a unique extension of $*$ on $A^{\|\cdot\|}$, which is an involution as well. Indeed, let $a,b \in A^{\|\cdot\|}$ and $\lambda \in \mathbb{C}$. Then there are sequences $(a_n)_{\mathbb{N}}$ and $(b_n)_{\mathbb{N}}$ on $A$ such that $a = \lim_n a_n$ and $b = \lim_n b_n$, then
\begin{align*}
    (a^*)^* &= ((\lim_n a_n)^*)^* = \lim_n ((a_n)^*)^* = \lim_n a_n) = a,\\
    (a+\lambda b)^* &= (\lim_n a_n+\lambda \lim_n b_n)^* = \lim_n (a_n+\lambda b_n)^* = \lim_n (a_n^*+\overline{\lambda} b_n^*) = a^*+\overline{\lambda} b^*,\\
    (ab)^* &= ((\lim_n a_n) (\lim_n b_n))^* = \lim_n (a_n b_n)^* = \lim_n (b_n^*a_n^*) = b^*a^*.
\end{align*}
The properties between the norm and the involution can be proved similarly. 
\end{proof}

\begin{definition}[Morphisms]\label{def:morphisms} Given two algebras $A$ and $B$ and a linear operator $\varphi:A \to B$, we say that
\begin{itemize}
    \item $\varphi$ is a homomorphism if it is a multiplicative map, i.e. for all $a,b \in A$ we have
    \begin{align*}
        \varphi(ab) = \varphi(a)\varphi(b);
    \end{align*}
    \item $\varphi$ is an isomorphism if it is a bijective homomorphism;
    \item $\varphi$ is an endomorphism if it is homomorphism and $A = B$;
    \item $\varphi$ is an automorphism if it is an endomorphism and an isomorphism (i.e. if it is a bijective endomorphism).
\end{itemize}

We say that a homomorphism $\varphi:A \to B$ is unital if both $A$ and $B$ are unital and $\varphi(1) = 1$.
\end{definition}

In this part we will follow Murphy's book \cite{Murphy1990}.

We will denote the algebra of all polynomials on the variable $z$ and complex coefficients by $\mathbb{C}[z]$ and observe that such algebra is normed with the norm
\begin{equation*}
    \|p\|:= \sup_{|\lambda| \leq 1}|p(\lambda)|.
\end{equation*}
It is important to notice that such normed algebra is not complete.

Let $a$ be an element of a unital algebra $A$ and $p \in \mathbb{C}[z]$ given by

\begin{equation*}
    p(z):= \sum_{k=0}^n \lambda_k z^k, \quad \lambda_k \in \mathbb{C}
\end{equation*}
for $k = 0,...,n$. We define 
\begin{equation*}
    p(a):= \sum_{k=0}^n \lambda_k a^k, \quad \lambda_k \in \mathbb{C}.
\end{equation*}
It is straightforward to see that the map
\begin{align*}
    \mathbb{C}[z] &\to A,\\
    p &\mapsto p(a),
\end{align*}
is a unital homomorphism.

\begin{definition}[Invertible elements] Given a unital algebra $A$, we say that $a \in A$ is invertible if there exista $b\in A$ s.t.
\begin{equation*}
    ab=ba=1.
\end{equation*}
$b$ is called an inverse of $a$.
\end{definition}

We define the set
\begin{equation*}
    \Inv A := \{a \in A: a \text{ is invertible}\} 
\end{equation*}
of the invertible elements of the unital algebra $A$.

Given a unital algebra $A$, prove that $\Inv A$ is a group under the algebra multiplication.

Now we are ready to introduce the spectral theory and we start with the fundamental definition of spectrum of an element of an algebra.

\begin{definition}[Spectrum of an element] Let $A$ be a unital algebra. The spectrum of an element $a \in A$ is defined as the set
\begin{equation*}
    \sigma(a) \equiv \sigma_A(a):= \{\lambda \in \mathbb{C}: \lambda 1 - a \notin \Inv A\}.
\end{equation*}
\end{definition}

\begin{example} We recall the recurrent main examples of these notes.
\begin{itemize}
    \item[$(a)$] For the algebra $M_n(\mathbb{C})$ and $a \in M_n(\mathbb{C})$, $\sigma(a)$ is the set of eigenvalues of $a$;
    \item[$(b)$] consider $X$ a compact Hausdorff topological space and the algebra $C(X)$. For given $f \in C(X)$, we have that $\sigma(f) = f(X)$.
\end{itemize}
\end{example}

\begin{remark} For any $a$ and $b$ elements of a unital algebra $A$, then $1-ab$ is invertible if and only if $1-ba$ is invertible. Indeed, if $c = (1-ab)^{-1}$ then $(1-ba)^{-1} = 1+bca$, obtained by the following calculations:
\begin{align*}
    (1+bca)(1-ba) &= 1 - ba + bca - bcaba = 1 - ba + bc(1-ab)a = 1 - ba + ba = 1,\\
    (1-ba)(1+bca) &= 1 + bca - ba - babca = 1 - ba + b(1-ab)ca = 1 - ba + ba = 1.
\end{align*}
The converse is analogous and we omit the proof. Consequently, we have $\sigma(ab)\setminus\{0\} = \sigma(ba)\setminus\{0\}$. Indeed,
\begin{align*}
    \lambda \in \sigma(ab)\setminus\{0\} &\iff \lambda 1 - ab \notin \Inv A, \quad \lambda \neq 0 \iff \lambda\left(1-\frac{ab}{\lambda}\right) \notin \Inv A, \quad \lambda \neq 0 \\
    &\iff \lambda\left(1-\frac{ba}{\lambda}\right) \notin \Inv A , \quad \lambda \neq 0  \iff \lambda 1 - ba \notin \Inv A, \quad \lambda \neq 0 \\
    &\iff \lambda \in \sigma(ba)\setminus\{0\}.
\end{align*}
\end{remark}

From now on, we will omit `$1$' in $\lambda 1$.

\begin{theorem}\label{teo:polynomials_spectrum} Let $a$ be an element of a unital algebra $A$. If $\sigma(a) \neq \emptyset$ and $p \in \mathbb{C}[z]$, then
\begin{equation*}
    \sigma(p(a)) = p(\sigma(a)).
\end{equation*}
\end{theorem}

\begin{proof} We may suppose that $p$ is not constant. If $\mu \in \mathbb{C}$, there is $n \in \mathbb{N}$ and elements $\lambda_0,..., \lambda_n \in \mathbb{C}$, $\lambda_0 \neq 0$, s.t.

\begin{equation*}
    p(z) - \mu = \lambda_0(z-\lambda_1)\cdots(z-\lambda_n),
\end{equation*}
and hence 

\begin{equation*}
    p(a) - \mu = \lambda_0(a-\lambda_1)\cdots(a-\lambda_n).
\end{equation*}
It is straightforward to notice that $p(a)-\mu \in \Inv A$ if and only if $(a-\lambda_k)$ is invertible for each $k=1,...,n$. It follows that $\mu \in \sigma(p(a))$ if and only if $\mu \in p(\lambda)$ for some $\lambda \in \sigma(a)$. The proof is complete.
\end{proof}

The theorem above is also valid for analytic functions instead of only polynomials in the case of analytical functions. Moreover, it is also valid for continuous functions on the case of $C^*$-algebras.

\begin{theorem}\label{teo:Neumann_series_Banach_algebra} Let $A$ be a unital Banach algebra and $a \in A$ with $\|a\|<1$. Then, $1-a \in \Inv A$ and 
\begin{equation*}
    (1-a)^{-1} = \sum_{n=0}^\infty a^n.
\end{equation*}
\end{theorem}

\begin{proof} Since $\|a\|<1$, one can use the submultiplicativity of the norm and obtain that
\begin{equation*}
    \sum_{n=0}^\infty \|a^n\| \leq \sum_{n=0}^\infty \|a\|^n = (1-\|a\|)^{-1}< \infty,
\end{equation*}
hence $\sum_{n=0}^\infty \|a^n\|$ is convergent and therefore $\sum_{n=0}^\infty a^n$ converges, we say to $b \in A$. Since $(1-a)(1+\cdots+a^n) = 1-a^{n+1}$ converges to $(1-a)b = b(1-a)$ and to $1$ we conclude that $b = (1-a)^{-1}$.
\end{proof}

The series above is called Neumann series for $(1-a)^{-1}$. We will only state the next theorem, but its proof is found in Murphy's book \cite{Murphy1990}.

\begin{theorem}\label{teo:Inv_A_open_inversion_map_differentiable} If $A$ is a unital Banach algebra, then $\Inv A$ is an open set and the inversion map
\begin{align*}
    \Inv A &\to A,\\
    a &\mapsto a^{-1},
\end{align*}
is differentiable.
\end{theorem}

\begin{lemma}\label{lema:spectrum_closed_differentiability} Let $A$ be a unital Banach algebra and $a \in A$. The spectrum $\sigma(a)$ of $a$ is a closed subset of 
\begin{equation*}
    D_{\|a\|}(0) :=\{z \in \mathbb{C}:|z|<\|a\|\}.
\end{equation*}
Furthermore, the map
\begin{align*}
    \mathbb{C}\setminus \sigma(a) &\to A,\\
    \lambda &\mapsto (a-\lambda)^{-1},
\end{align*}
is differentiable.
\end{lemma}

\begin{proof} If $|\lambda|>\|a\|$, then $\|\lambda^{-1}a\|<1$, and hence $1-\lambda^{-1}a \in \Inv A$, so is $\lambda-a$ and therefore $\lambda \notin \sigma(a)$. Thus $\lambda \in \sigma(a)$ necessarily implies that $|\lambda| \leq \|a\|$ and therefore $\sigma(a) \subseteq D_{\|a\|}(0)$. Now, let $F:\mathbb{C} \to A$ given by
\begin{equation*}
    F(\lambda) = \lambda - a,
\end{equation*}
which is a continuous function. Hence $\varphi^{-1}(\Inv A) = \mathbb{C}\setminus \sigma(a)$ is an open set\footnote{It is an equivalent definition of continuity: $f:X \to Y$ is continuous if and only if $f^{-1}(A)$ is an open set of $X$ for every $A \subseteq Y$ open.}, and therefore $\sigma(a)$ is closed. The differentiability of the map $\lambda \to (a-\lambda)^{-1}$ comes from the previous theorem.
\end{proof}

Note that the previous lemma actually says that $\sigma(a)$ is a compact set, since $\sigma(a)$ is closed and it is contained in the closed disc $D_{\|a\|}(0)$ which is a bounded set. For the next important result we recall the Liouville's Theorem from the complex analysis below.

\begin{theorem}[Liouville] Every entire bounded function must be constant.
\end{theorem}

We recall the reader that one of the equivalent definitions of an entire function is a function that it is holomorphic everywhere, that is, it is a complex valued function $f$ of one or more complex variables that is complex differentiable in a neighborhood for each point of its domain. By `bounded' we mean simply that there exists $M \geq 0$ s.t. $|f(z)|\leq M$ for every $z \in \mathbb{C}^n$.

\begin{theorem}[Gelfand] If $a$ is an element of a unital Banach algebra $A$, then the spectrum of $a$ is non-empty.
\end{theorem}

\begin{proof} Suppose that $\sigma(a) \neq \emptyset$. If $|\lambda|> 2 \|a\|$, then $\|\lambda^{-1}a\|<1/2$ and therefore $1 - \|\lambda^{-1}a\|>1/2$. Hence by Theorem \ref{teo:Neumann_series_Banach_algebra} we have that
\begin{align*}
    \left\|(1-\lambda^{-1}a)-1\right\| &= \left\|\sum_{n=1}^\infty(\lambda^{-1}a)^n\right\| \leq \sum_{n=1}^\infty\left\|\lambda^{-1}a\right\|^n = \frac{\left\|\lambda^{-1}a\right\|}{1-\left\|\lambda^{-1}a\right\|} \leq 2 \left\|\lambda^{-1}a\right\| <1.
\end{align*}
Consequently, $\left\|(1-\lambda^{-1}a)-1\right\|<2$, and hence
\begin{equation*}
    \left\|(a - \lambda)^{-1}\right\| = \left\|\lambda^{-1}(1-\lambda^{-1}a)\right\| < \frac{2}{\lambda} < \|a\|^{-1}.
\end{equation*}
Note that $\sigma(a) = \emptyset$ implies necessarily that $a \neq 0$ and then the inequality above makes sense. By Lemma \ref{lema:spectrum_closed_differentiability} we know that in particular the map
\begin{equation*}
    \lambda \mapsto (a-\lambda)^{-1}
\end{equation*}
is continuous and therefore bounded on the compact disc $2\|a\| D_1(0)$. Therefore the same map is bounded in all $\mathbb{C}$, and hence there is positive number $M$ s.t. $\left|(a-\lambda)^{-1}\right\|\leq M$.

Let $\tau \in A^*$. Hence the function $\lambda \mapsto \tau((a-\lambda)^{-1})$ is entire and bounded by $M\|\tau\|$, and by the Liouville's Theorem we conclude that such function is constant. In particular $\tau(a^{-1}) = \tau((a-1)^{-1})$ and hence $a^{-1}=(a-1)^{-1}$, since the elements of $A^*$ separates points on $A$. We conclude that $a = a-1$, a contradiction. The proof is complete.
\end{proof}



\begin{definition}[Spectral radius] For $a$ an element of a unital Banach algebra $A$, the spectral radius of $a$ is the number
\begin{equation*}
    r(a) = \sup_{\lambda \in \sigma(a)} |\lambda|.
\end{equation*}
\end{definition}
Observe that by compactness of $\sigma(a)$ makes the $r(a)$ well defined as we call it a number. Moreover, since $\sigma(ab)\setminus\{0\} = \sigma(ba)\setminus\{0\}$ for every $a,b \in A$, we have that $r(ab) = r(ba)$.

\begin{example} For a compact Hausdorff space $X$ and $A = C(X)$, we have that $r(f) = \|f\|_\infty$ for all $f \in A$.
\end{example}

\begin{example} If $A = M_2(\mathbb{C})$ and
\begin{equation*}
    a = \begin{pmatrix}
    0&1 \\ 0&0
    \end{pmatrix},
\end{equation*}
we have that $r(a) = 0$ and that $\|a\|=1$ on any of the two possible norms presented in Example \ref{exa:normed_algebras} $(a)$.
\end{example}

The next theorem we state relates the spectral radius with the norm of a unital Banach algebra.

\begin{theorem}[Beurling]\label{thm:Beurling} For a unital Banach algebra we have that
\begin{equation*}
    r(a) = \inf_{n \geq 1} \|a^n\|^{1/n} = \lim_{n \to \infty} \|a^n\|^{1/n}, \quad \text{for all } a \in A.
\end{equation*}
\end{theorem}

\begin{example} $A = C^1([0,1])$ is a Banach algebra under the usual pointwise operations and with a norm given by

\begin{equation*}
    \|f\|:= \|f\|_\infty+\|f'\|_\infty, \quad f \in C^1([0,1]).
\end{equation*}
On this context, consider the inclusion $x:[0,1] \to \mathbb{C}$. Then $x \in A$. We have that $\|x^n\| = 1 + n$ for every $n \in \mathbb{N}$ and therefore $r(x)=\lim_{n \to \infty}(1+n)^{1/n} = 1 < 2 = \|x\|$.
\end{example}

In general we do need to impose that our Banach algebra a priori be unital. There exists a canonical process of including an unity element on a non-unital algebra called \textbf{unitization} which we will explain here. Let $A$ be an algebra and set $\widetilde{A}:= A \oplus \mathbb{C}$ as a vector space. Define the algebra multiplication on $\widetilde{A}$ as it follows.

\begin{equation*}
    (a,\lambda)(b,\mu) = (ab+\lambda b + \mu a, \lambda \mu).
\end{equation*}

The unit on $\widetilde{A}$ is $(0,1)$ and $\widetilde{A}$ is called the \textbf{unitization} of $A$. The map

\begin{align*}
    A \to \widetilde{A},\\
    a \mapsto (a,0),
\end{align*}
is an injective homomorphism, which we use to identify $A$ as an ideal of $\widetilde{A}$.

We now use the notation $ a + \lambda := (a,\lambda)$. The map
\begin{align*}
    \widetilde{A} \to \mathbb{C},\\
    a+\lambda \mapsto \lambda,
\end{align*}
is a unital homomorphism with kernel $A$, called canonical homomorphism.

If $A$ is a normed algebra, we define a norm on $\widetilde{A}$ by

\begin{equation*}
    \|a + \lambda\|:= \|a\| + |\lambda|.
\end{equation*}
if $A$ is a Banach algebra, so is $\widetilde{A}$.

Now it is easy to set the spectral theory for a non-unital Banach algebra. We simply set that 

\begin{equation*}
    \sigma_A(a):= \sigma_{\widetilde{A}}(a), \quad a \in A
\end{equation*}
where $\sigma_A(a)$ is the spectrum of $a$ as an element of $A$ and $\sigma_{\widetilde{A}}(a)$ is the spectrum of $a$ as an element of $\widetilde{A}$. Now we also can set the spectral radius of $a \in A$ in analogous way of the unital case, namely

\begin{equation*}
    r(a) = \sup_{\lambda \in \sigma_A(a)}|\lambda|.
\end{equation*}
Note that in this case $0 \in \sigma(a)$ for every element in $A$. Now all the results we studied have applications for non-unital Banach algebras.

Now we will briefly study the relationship between the spectrum of elements of an unital Banach algebra and the spectrum of the algebra itself. It is important the reader has in mind that the word `spectrum' is overloaded of different meanings here.

Let $A$ and $B$ be two algebras where $A$ is not unital and $B$ is unital. Given $\varphi:A \to B$, then there exists a unique unital homomorphism $\widetilde{\varphi}:\widetilde{A} \to B$ which extends $\varphi$.

\begin{proposition}\label{prop:homomorphism_spectrum} If $\varphi:A \to B$ is a unital homomorphism between the unital algebras $A$ and $B$, then $\sigma(\varphi(a)) \subseteq \sigma(a)$ for every $a \in A$. Consequently, $\varphi(\Inv A) \subseteq \Inv B$.
\end{proposition}

\begin{proof} Suppose that there exists $\lambda \in \sigma(\varphi(a))$ s.t. $\lambda \notin \sigma(a)$. We have that

\begin{equation*}
    a-\lambda \in \Inv A \quad \text{and} \quad \varphi(a)-\lambda \notin \Inv B.
\end{equation*}
Let $b = (a-\lambda)^{-1}$, we have that

\begin{equation*}
    1 = \varphi(1) = \varphi((a-\lambda)b) = (\varphi(a)-\lambda)\varphi(b)
\end{equation*}
and similarly we also obtain $\varphi(b)(\varphi(a)-\lambda)=1$. Therefore we conclude that $\varphi(b)=(\varphi(a)-\lambda)^{-1}$, leading to a contradiction. The last part of the proposition is straightforward.
\end{proof}

\begin{theorem}\label{thm:self_adjoint_norm_is_spectral_radius} For every self-adjoint element $a$ in a C$^*$-algebra we have $r(a) = \|a\|$. 
\end{theorem}

\begin{proof} In this case we have $\|a^2\| = \|a^*a\| = \|a\|^2$, by induction one can obtain
\begin{equation*}
    \|a^{2^n}\| = \|a\|^{2^n}.
\end{equation*}
By Theorem \ref{thm:Beurling} we have
\begin{equation*}
    r(a) = \lim_n \|a^n\|^{1/n} = \lim_n \|a^{2^n}\|^{2^{-n}} = \lim_n \|a\| = \|a\|.  \tag*{\qedhere}
\end{equation*}
\end{proof}

\begin{theorem}\label{thm:star_homomorphism_is_norm_decreasing} Let $\varphi:A\to B$ be a $*$-homomorphism from a $*$-Banach algebra to a C$^*$-algebra. Then, $\varphi$ is norm-decreasing. 
\end{theorem}

\begin{proof} W.l.o.g. we may suppose $A$ and $B$ unital. By Proposition \ref{prop:homomorphism_spectrum} we have $\sigma(\varphi(a)) \subseteq \sigma(a)$ for every $a \in A$, and then
\begin{align*}
    \|\varphi(a)\|^2 &= \|\varphi(a)^*\varphi(a)\| = \|\varphi(a^*a)\| \stackrel{(\dagger)}{=} r(\varphi(a^*a)) \leq r(a^*a) \leq \|a^*a\| \leq \|a\|^2,
\end{align*}
where in $(\dagger)$ we used Theorem \ref{thm:self_adjoint_norm_is_spectral_radius}.
\end{proof}

\begin{definition}[Characters] A character on an abelian algebra $A$ is a non-zero homomorphism $\tau:A \to \mathbb{C}$. The set of all characters on $A$ is said to be the spectrum of $A$ and it will be denoted by $\widehat{A}$.
\end{definition}

The next theorem links the characters of an abelian algebra to the maximal ideals of the same algebra.

\begin{theorem}[Theorem 1.3.3. of \cite{Murphy1990}] Let $A$ be a unital abelian Banach algebra.
\begin{itemize}
    \item if $\tau \in \widehat{A}$, then $\|\tau\|=1$;
    \item let $\mathfrak{I}(A)$ be the set of the maximal ideals on $A$. The map
    \begin{align*}
        \widehat{A} &\to \mathfrak{I}(A),\\
        \tau &\mapsto \ker \tau,
    \end{align*}
    is bijective.
\end{itemize}
\end{theorem}
The theorem above says that every maximal ideal in $A$ is a kernel of some character on $A$. However, the maximal ideals are not the unique relation on which the characters are participating, as we can see in the theorem below.

\begin{theorem}[Theorem 1.3.4. of \cite{Murphy1990}] Let $A$ be an abelian Banach algebra.
\begin{itemize}
    \item if $A$ is unital, then
    \begin{equation*}
        \sigma(a)=\{\tau(a):\tau \in \widehat{A}\}, \quad a \in A;
    \end{equation*}
    \item if $A$ is not unital, then
    \begin{equation*}
        \sigma(a)=\{\tau(a):\tau \in \widehat{A}\}\cup\{0\}, \quad a \in A.
    \end{equation*}
\end{itemize}
\end{theorem}

Naturally we may endow $\widehat{A}$ with the weak$^*$ topology.

Now we present the most important result of this section for this thesis. For a given abelian Banach algebra $A$ with $\widehat{A}\neq \emptyset$ and $a \in A$, we define the Gelfand transform of $a$, denoted by $\widehat{a}$, the function

\begin{align*}
    \widehat{a}: \widehat{A} &\to \mathbb{C},\\
    \tau &\mapsto \tau(a).
\end{align*}
The weak$^*$ topology on $\widehat{A}$ is the smallest topology that makes $\widehat{a}$ continuous for every $a \in A$. We present the Gelfand's representation Theorem.

\begin{theorem}[Gelfand representation] Suppose that $A$ is an abelian Banach algebra s.t. $\widehat{A} \neq \emptyset$. The map
\begin{align*}
    A &\to C_0(\widehat{A}),\\
    a &\mapsto \widehat{a},
\end{align*}
is a norm-decreasing homomorphism, and
\begin{equation*}
    r(a) = \|a\|_\infty, \quad a \in A.
\end{equation*}
If $A$ is unital, then $\sigma(a) = \widehat{a}(\widehat{A})$. If $A$ is non-unital, then $\sigma(a) = \widehat{a}(\widehat{A})\cup\{0\}$.
\end{theorem}

\begin{corollary}[Gelfand representation theorem for C$^*$-algebras]\label{corollary:Gelfand_C_star} Suppose that $A$ is a non-zero abelian C$^*$-algebra, then the Gelfand representation
\begin{align*}
    A &\to C_0(\widehat{A}),\\
    a &\mapsto \widehat{a},
\end{align*}
is an isometric $*$-isomorphism.
\end{corollary}

\begin{lemma}[Lemma 2.2.2 of \cite{Murphy1990}]\label{lemma:positivity_equivalence} Let $A$ be a unital C$^*$-algebra, $a \in A$ a self-ajoint element and $t \in \mathbb{R}$. Then $a \geq 0$ if $\|a-t1\|\leq t$. In the inverse direction, if $\|a\| \leq t$ and $a \geq 0$, then $\|a-t1\| \leq t$. 
\end{lemma}

\begin{corollary}\label{cor:positivity_t_equals_to_one} Let $A$ be a unital C$^*$-algebra and $a \in A$ a self-ajoint element satisfying $\|a\|\leq 1$. The element $a$ is positive if and only if $\|a-1\|\leq 1$. 
\end{corollary}

Now we present the Uryshon's Lemma in the version for locally compact Hausdorff spaces. This lemma and the Gelfand's Theorem will be used to prove Theorem \ref{thm:dense_character_general}, which states that sets of separating point characters are dense in the spectrum of a commutative C$^*$-algebra.

\begin{lemma}[Urysohn's Lemma for LCH spaces]\label{lemma:Urysohn_LCH} Let $X$ be a locally compact Hausdorff space. If $K,F \subseteq X$ are disjoint sets s.t. $K$ is compact and $F$ is closed, then there exists a continuous function $f:X \to [0,1]$ s.t. $f\vert_K \equiv 1$ and $f\vert_F \equiv 0$.
\end{lemma}

\begin{theorem}\label{thm:dense_character_general} Given a commutative $C^*$ algebra $B$, let $Y \subseteq \widehat{B}$ such that for any $a \in B$ it follows that
\begin{equation}\label{eq:Y_separates_spectrum_version}
    \varphi(a) = 0 \quad \forall \varphi \in Y \implies a=0,
\end{equation}
i.e., $Y$ separates points in $B$. Then $Y$ is dense in $\widehat{B}$ (weak$^*$ topology).
\end{theorem}

\begin{proof} By the Gelfand representation theorem for commutative $C^*$-algebras, it follows that the map 
\begin{align*}
    \Upsilon : B &\to C_0(\widehat{B}) \\
     b &\mapsto \widehat{b},
\end{align*}
where $\widehat{b}$ is the evaluation map on the spectrum, is an isometric $*$-isomorphism. Seeing $B$ as $C_0(X)$, $X = \widehat{B}$, which is locally compact\footnote{We recall that the spectrum of a commutative Banach algebra in general is locally compact, and it is compact if such algebra is unital.}, and $Y \subseteq X$, we may write the property the condition \eqref{eq:Y_separates_spectrum_version} as

\begin{equation}\label{eq:Y_separates_function_version}
    f(y) = 0 \quad \forall y \in Y \implies f\equiv 0.
\end{equation}
Observe that every point in $X$ has a relatively compact neighborhood\footnote{Fixed a topological space, s set is said to be relatively compact if its closure is compact.} since $X$ is locally compact and Hausdorff.

Now, suppose that $Y$ is not dense in $\widehat{B}$. Then there exists a non-empty open set $U \subseteq X$ such that $U \cap Y = \emptyset$. By the last paragraph, w.l.o.g. we may assume that $U$ is relatively compact. Take $x_0 \in U$. Since $\{x_0\}$ is compact, by the Urysohn's Lemma for LCH spaces (Lemma \ref{lemma:Urysohn_LCH}), there exists a continuous function $F:X \to [0,1]$ such that $F(x_0) = 1$ and $F\vert_{U^c} = 0$. Note that supp$F = \overline{U}$, which is compact since $U$ is relatively compact. We conclude that $F \in C_c(X)\subseteq C_0(X)$. Since $Y \subseteq U^c$, we have that $F(y) = 0$ for all $y \in Y$ and $F \not\equiv 0$, a contradiction due to the hypothesis of the validity of \eqref{eq:Y_separates_function_version}.
\end{proof}

\begin{definition} Let $A$ be a C$^*$-algebra $A$. An element $a \in A$ is said to be positive if it is self-adjoint and
\begin{equation*}
    \sigma(a) \subseteq \mathbb{R}_+.
\end{equation*}
The set of positive elements of $A$ will be denoted by $A^+$.
\end{definition}

\begin{remark}\label{positive_C_star_poset} For every C$^*$-algebra $A$, there is a natural partial order on $A^+$: we say that $a \geq b$, when $a-b \geq 0$. 
\end{remark}




\begin{definition} A linear map $\varphi:A \to B$ between two C$^*$-algebras $A$ and $B$ is said to be positive when $\varphi(A^+) \subseteq \varphi(B^+)$. In particular, when $B = \mathbb{C}$, $\varphi$ is said to be a positive linear functional.
\end{definition}

\begin{remark} Equivalently, $\varphi$, as in definition above, is a positive linear functional if and only if $\varphi(a^*a) \geq 0$. This is straightforward consequence of Theorem 2.2.5 (1) of \cite{Murphy1990}, which gives the following characterization of positive elements of a C$^*$-algebra $A$:
\begin{equation*}
    A^+ = \{a^*a: a \in A\}.
\end{equation*}
\end{remark}

\begin{remark} In particular, every $*$-homomorphism is positive. Observe that there are positive functionals that are not $*$-homomorphisms. For instance, take $A = C(\mathbb{T})$ and consider $m$ as being the normalized arc length measure on $\mathbb{T}$. The linear functional $\varphi : C(\mathbb{T}) \to \mathbb{C}$ given by
\begin{equation*}
    \varphi(g) := \int_{\mathbb{T}} g dm
\end{equation*}
is positive and it is not a $*$-homomorphism.
\end{remark}

The next result is the Theorem 3.3.1. of \cite{Murphy1990}.

\begin{theorem} 
\label{thm:positivefunctionalbounded}
If $\varphi$ is a positive linear functional on a C$^*$-algebra $A$, then it is bounded.
\end{theorem}


\section{Universal algebras generated by relations}

One of the main structures of this work is the Exel-Laca algebra, which is a universal algebra generated by partial isometries under certain relations. In order to present these algebras properly, we introduce the notion of universal algebras. We will see further that the relations of the Exel-Laca algebras are polynomials and then it is sufficient to study these universal structures under the condition where the relations are polynomials, as it is in \cite{Blackadar1985,Blackadar2006,Tasca2015}. However we inform the reader that these universal structures can be realized for a more general nature of the relations, we mention for example the work \cite{Phillips1988}, where the relations can be any statement on the generators that makes sense in a C$^*$-algebra.

First, we construct the free associative complex $*$-algebra generated by a set.

\begin{definition}[Free semigroup] Let $\mathcal{D}$ be a non-empty set. The free semigroup $\mathcal{F}_\mathcal{D}$ generated by $\mathcal{D}$ is the semigroup of the finite non-empty sequences of elements of $\mathcal{D}$, endowed by the juxtaposition product of these sequences: given $x_1\cdots x_n, y_1\cdots y_m \in \mathcal{F}_\mathcal{D}$, $n,m\in \mathbb{N}$, then
\begin{equation*}
    (x_1\cdots x_n) \cdot (y_1\cdots y_m) := x_1\cdots x_ny_1\cdots y_m.
\end{equation*}
The set $\mathcal{D}$ is said to be the generator set of $\mathcal{F}_\mathcal{D}$.
\end{definition}

It is straightforward that $\mathcal{F}_\mathcal{D}$ is in fact a semigroup, since the juxtaposition product is associative. 

\begin{definition}[Free algebra] Let $\mathcal{D}$ be a non-empty set. The free algebra $\mathcal{A}_\mathcal{D}$ generated by $\mathcal{D}$ is the algebra over $\mathbb{C}$ spanned by the elements of the free semigroup $\mathcal{F}_\mathcal{D}$, that is, the set of the complex polinomials on the non-commutative variables of $\mathcal{D}$:
\begin{equation*}
    \mathcal{A}_\mathcal{D} := \spann\left\{\mathcal{F}_\mathcal{D}\right\} = \left\{\sum_{i = 1}^n\lambda_i r_i: n \in \mathbb{N}; \lambda_i \in \mathbb{C}, r_i \in \mathcal{F}_\mathcal{D}, 1 \leq i \leq n\right\}.
\end{equation*}
where the operations of addition, product by scalar are the usual ones for polynomials. The algebra product is the one inherited one from the semigroup $\mathcal{F}_\mathcal{D}$. We say that $\mathcal{D}$ is the generator set of $\mathcal{A}_D$.
\end{definition}

Now we introduce an involution on the free algebra.

\begin{definition}[Free $*$-algebra] Let $D$ be a non-empty set. Define a copy of $D$ given by
\begin{equation*}
  D^*=\{d^*: d \in D\}.
\end{equation*}
Consider the set $\mathcal{D}= D\sqcup D^*$. The free associative complex $*$-algebra generated by the set $D$ is the free algebra $\mathcal{A}_D$ endowed with the following involution $*$. Given $x\in \mathcal{D}$, define
\begin{equation*}
    x^*:= \begin{cases}
            h, \quad \text{if } x = h^* \in D^*,\\
            h^*,\quad \text{if } x = h \in D.
          \end{cases}
\end{equation*}
Now, for every $r = x_1 \cdots x_n \in \mathcal{F}_\mathcal{D}$, $x_i \in \mathcal{D}$ for every $i$, set
\begin{equation*}
    r_i^* = x_n^* \cdots x_1^* \in \mathcal{F}_\mathcal{D}.
\end{equation*}
Finally, extend the operation above by taking
\begin{equation*}
    \left(\sum_{i = 1}^n\lambda_i r_i\right)^* = \sum_{i = 1}^n\lambda_i r_i^*,
\end{equation*}
for every $\sum_{i = 1}^n\lambda_i r_i \in \mathcal{A}_D$, where $r_i \in \mathcal{F}_\mathcal{D}$ and $\lambda_i \in \mathbb{C}$ for each $ 1 \leq i \leq n$. The set $D$ is said to be the generator set of the $*$-algebra $\mathcal{A}_D$.
\end{definition}

\begin{remark} Observe that the free $*$-algebras are not unital, however, as we see further in this chapter, this does not interfer on the construction of unital universal C$^*$-algebras. 
\end{remark}

From now, we fix $D$ a generator set of a free $*$-algebra $\mathcal{A}_D$, and we refer to $D$ simply by `generator set'.

\begin{proposition}\label{prop:homomorphism_extension_free_algebra} Let $D$ be the generator and a C$^*$-algebra $B$. Also and consider a function $\Theta_0:D \to  B$. There exists a unique $*$-homomorphism $\Theta:\mathcal{A}_D \to B$ which is the extension of $\Theta_0$.
\end{proposition}

\begin{proof} For every $d \in D$, $g_1,...,g_n \in \mathcal{F}_{\mathcal{D}}$, and $a = \sum_{i=1}^m \lambda_i x_i \in \mathcal{A}_D$, $\lambda_i \in \mathbb{C}$ and $x_i \in \mathcal{F}_{\mathcal{D}}$ for all $i$, define $\Theta$ as follows:
\begin{align}
    \Theta(d) &:= \Theta_0(d);\label{eq:theta_extension}\\
    \Theta(d^*) &:= \Theta(d)^* = \Theta_0(d)^*;\label{eq:theta_involution}\\
    \Theta(g_1\cdots g_n) &:= \Theta(g_1)\cdots \Theta(g_n)\label{eq:theta_prod};\\
    \Theta\left(\sum_{i=1}^m \lambda_i x_i\right) &:= \sum_{i=1}^m \lambda_i \Theta(x_i). \label{eq:theta_linearity}
\end{align}
By \eqref{eq:theta_extension}, it is clear that $\Theta$ is an extension of $\Theta_0$. The identity \eqref{eq:theta_linearity} imposes the linearity of $\Theta$. Consider now $b = \sum_{j=1}^{\ell} \lambda'_j y_j \in \mathcal{A}_D$, $\lambda'_j \in \mathbb{C}$ and $y_j \in \mathcal{F}_{\mathcal{D}}$ for all $j$. By \eqref{eq:theta_linearity} and \eqref{eq:theta_prod} we have
\begin{align*}
    \Theta(ab) &= \Theta\left(\left[\sum_{i=1}^m \lambda_i x_i\right] \left[ \sum_{j=1}^{\ell} \lambda'_j y_j\right]\right) = \Theta\left(\sum_{i=1}^m \sum_{j=1}^{\ell} (\lambda_i \lambda'_j) (x_i y_j) \right) = \sum_{i=1}^m \sum_{j=1}^{\ell} (\lambda_i \lambda'_j) \Theta(x_i y_j)\\
    &= \sum_{i=1}^m \sum_{j=1}^{\ell} (\lambda_i \lambda'_j) \Theta(x_i)\Theta (y_j) = \left[\sum_{i=1}^m \lambda_i \Theta(x_i)\right] \left[\sum_{j=1}^{\ell} \lambda'_j \Theta(y_j)\right]\\
    &= \Theta\left(\sum_{i=1}^m \lambda_i x_i\right) \Theta\left(\sum_{j=1}^{\ell} \lambda'_j y_j\right) = \Theta(a)\Theta(b),
\end{align*}
and then $\Theta$ is a homomorphism. Now, write $x_i = a_{1,i} \cdots a_{\phi(i),i}$  for each $i = 1,...,m$, where $\phi:\{1,...,m\}\to \mathbb{N}$ is a function and $a_{p,q}\in \mathcal{D}$ for every $1\leq p \leq \phi(q)$ and each $1 \leq q \leq m$. Then, by \eqref{eq:theta_involution} and the fact that $\Theta$ is a homomorphism, we have
\begin{align*}
    \Theta(a^*) &= \Theta\left(\left[\sum_{i=1}^m \lambda_i x_i\right]^*\right) = \Theta\left(\sum_{i=1}^m \overline{\lambda_i} a_{\phi(i),i}^* \cdots a_{1,i}^*\right) = \sum_{i=1}^m \overline{\lambda_i} \Theta(a_{\phi(i),i}^* \cdots a_{1,i}^*)\\
    &= \sum_{i=1}^m \overline{\lambda_i} \Theta(a_{\phi(i),i}^*) \cdots \Theta(a_{1,i}^*) = \sum_{i=1}^m \overline{\lambda_i} \Theta(a_{\phi(i),i})^* \cdots \Theta(a_{1,i})^*\\
    &= \sum_{i=1}^m \overline{\lambda_i} (\Theta(a_{1,i}) \cdots \Theta(a_{\phi(i),i}))^* = \sum_{i=1}^m \overline{\lambda_i} (\Theta(a_{1,i} \cdots a_{\phi(i),i}))^* = \Theta\left(\sum_{i=1}^m \lambda_i x_i\right)^* = \Theta(a)^*.
\end{align*}
Therefore $\Theta$ is a $*$-homomorphism. Given $\tilde{\Theta}:\mathcal{A}_D \to B$ a $*$-homomorphism which extends $\Theta_0$, we have $\tilde{\Theta}(d) = \Theta_0(d) = \Theta(d)$ for every $d \in D$, and since $\tilde{\Theta}$ is a homomorphism, we also have $\tilde{\Theta}(d^*) = \tilde{\Theta}(d)^* = \Theta(d)^*$. Then, for every $a = \sum_{i=1}^m \lambda_i a_{1,i} \cdots a_{\phi(i),i} \in \mathcal{A}_D$, we obtain
\begin{align*}
    \tilde{\Theta}(a) &= \tilde{\Theta}\left(\sum_{i=1}^m \lambda_i a_{1,i} \cdots a_{\phi(i),i}\right) = \sum_{i=1}^m \lambda_i\tilde{\Theta}( a_{1,i}) \cdots \tilde{\Theta}(a_{\phi(i),i}) = \sum_{i=1}^m \lambda_i\Theta( a_{1,i}) \cdots \Theta(a_{\phi(i),i})\\
    &= \sum_{i=1}^m \lambda_i \Theta(a_{1,i} \cdots a_{\phi(i),i}) = \Theta\left(\sum_{i=1}^m \lambda_i a_{1,i} \cdots a_{\phi(i),i}\right) = \Theta(a).
\end{align*}
Therefore, $\tilde{\Theta} = \Theta$, that is, such homomorphism is the unique one such that extends $\Theta_0$.
\end{proof}

\begin{remark} Observe that the function $\Theta_0$ can also be extended uniquely to the unitization $\tilde{\mathcal{A}}_D$ by taking $\Theta(1) = 1$. 
\end{remark}

Now that we constructed the free $*$-algebra, the next important step on the construction of universal C$^*$-algebras it is to define relations on the generator elements, and then we may construct representations for the free $*$-algebra, which allows us to construct a C$^*$-seminorm that can be turned into a C$^*$-norm. During this construction, it arises naturally the fact that the relations between the generator elements must satisfy an admissibility condition for the universal C$^*$-algebra does exist. 

\begin{definition}[Relations] A relation on a non-unital $*$-algebra $B$ is a pair $(r,t) \in \tilde{B} \times \mathbb{R}_+$, where $\tilde{B}$ is the unitization of $B$.
\end{definition}

\begin{definition}[Representations of generators] Let $D$ be the set of generators of the $*$-algebra $\mathcal{A}_D$ and $\mathscr{R}\subset \tilde{\mathcal{A}}_D \times \mathbb{R}_+$ a family of relations on $\mathcal{A}_D$. Given $B$ a C$^*$-algebra and a function $\Theta_0: D \to B$, we say that $\Theta_0$ is a representation for the pair $(D,\mathscr{R})$ if
\begin{equation*}
    \|\Theta(r)\| \leq t,
\end{equation*}
for every $(r,t) \in \mathscr{R}$, where $\Theta:\mathfrak{A}_D \to \mathfrak{B}$ is the unique $*$-homomorphism extending $\Theta_0$ such that 
\begin{equation*}
    \mathfrak{A}_D = \begin{cases}
                        \mathcal{A}_D, \quad \text{if }\mathscr{R} \subset \mathcal{A}_D \times \mathbb{R}_+,\\
                        \tilde{\mathcal{A}}_D, \quad \text{otherwise};
                     \end{cases}
    \quad \text{and} \quad
    \mathfrak{B} = \begin{cases}
                        B, \quad \text{if }\mathscr{R} \subset \mathcal{A}_D \times \mathbb{R}_+,\\
                        B, \quad \text{if }\mathscr{R} \not\subset \mathcal{A}_D \times \mathbb{R}_+ \text{ and $B$ is unital},\\
                        \tilde{B}, \quad \text{otherwise}.
                   \end{cases}
\end{equation*}
\end{definition}

\begin{definition}[Admissibility] Let $D$ be the set of generators of the $*$-algebra $\mathcal{A}_D$ and $\mathscr{R}\subset \tilde{\mathcal{A}}_D \times \mathbb{R}_+$ a family of relations on $\mathcal{A}_D$. We say the pair $(D,\mathscr{R})$ is admissible if for every $d \in D$ there exists $K_d \in \mathbb{R}_+$ such that $\|\Theta(d)\| \leq K_d$ for every $\Theta$ representation for $(D,\mathscr{R})$.
\end{definition}

\begin{example}\label{exa:pairs} Here are some examples of admissible and non-admissible pairs:
\begin{itemize}
    \item[$(i)$] by taking $D = \{x\}$ and $\mathscr{R} = \emptyset$, we have that the functions in the form
    \begin{equation*}
        \Theta_{0,\lambda}(x) = \begin{pmatrix}
                                 \lambda & 0\\
                                 0 & \lambda
                                \end{pmatrix} \in M_2(\mathbb{C}),
        \quad \lambda \in \mathbb{C},
    \end{equation*}
    are representations of $(D,\mathscr{R})$, by vacuosity. However, this pair is not admissible, since for each $K \in \mathbb{R}_+$, we can take $\lambda \in \mathbb{C}$ such that $|\lambda| > K$ and then
    \begin{equation*}
        \|\Theta_\lambda(x)\| = \left\|\begin{pmatrix}
                                 \lambda & 0\\
                                 0 & \lambda
                                \end{pmatrix}\right\| = |\lambda| > K;
    \end{equation*}
    \item[$(ii)$] by taking $D = \{x\}$ and $\mathscr{R} = \{(x-x^*,0)\}$, we have that the functions in the form $\Theta_{0,\lambda}(x) = \lambda$, $\lambda \in \mathbb{R}$, are representations of $(D,\mathscr{R})$. Indeed, their respective extensions $\Theta_\lambda$ give
    \begin{equation*}
        \|\Theta(x-x^*)\| = \|\Theta(x) - \Theta(x)^*\| = \|\lambda - \lambda\| = 0.
    \end{equation*}
    However, such pair is not admissible, since for every $K \in \mathbb{R}_+$, we may take $\lambda > K$ and then
    \begin{equation*}
        \|\Theta(x)\| = \lambda > K;
    \end{equation*}
    \item[$(iii)$] by taking $D = \{x\}$ and $\mathscr{R} = \{(x-x^*,0),(x,1),(1-x^2,1)\}$, we have that the functions in the form $\Theta_{0,\lambda}(x) = \lambda$, $\lambda \in [0,1]$, are representations of $(D,\mathscr{R})$. Indeed, their respective extensions, $\Theta_\lambda$ give
    \begin{align*}
        \|\Theta(x-x^*)\| &= \|\Theta(x) - \Theta(x)^*\| = \|\lambda - \lambda\| = 0,\\
        \|\Theta(x)\| &= \lambda \leq 1,\\
        \|1-\Theta(x)^2\| &= \|1-\lambda^2\| \leq 1,
    \end{align*}
    Moreover, such pair is admissible, since for every $\Theta$ representation for $(D,\mathscr{R})$ we have
    \begin{equation*}
        \|\Theta(x)\| \leq 1;
    \end{equation*}
    \item[$(iv)$] by taking $D = \{x\}$ and $\mathscr{R} = \{(1-x^*x,0),(1-xx^*,0)\}$, we have that the functions in the form $\Theta_{0,\lambda}(x) = f_\lambda$, $\lambda \in \mathbb{R}$, where $f_\lambda:[0,2\pi] \to \mathbb{C}$, is the function $f(y) = e^{-i\lambda y}$, are representations of $(D,\mathscr{R})$. Indeed, their respective extensions, $\Theta_\lambda$ give
    \begin{align*}
        \|\Theta(1-x^*x)\| &= \|\Theta(1)-\Theta(x^*)\Theta(x)\| = \|1 - e^{i\lambda y}e^{-i\lambda y}\| = \|1-1\| = 0,\\
        \|\Theta(1-xx^*)\| &= \|\Theta(1)-\Theta(x)\Theta(x^*)\| = \|1 - e^{-i\lambda y}e^{i\lambda y}\| = \|1-1\| = 0,
    \end{align*}
    Also, such pair is admissible, since for every $\Theta$ representation for $(D,\mathscr{R})$ we have
    \begin{equation*}
        \|\Theta(x)\|^2 = \|\Theta(x)^*\Theta(x)\| = \|1\| = 1,
    \end{equation*}
    that is $\|x\| = 1$;
    \item[$(v)$] take $D = \{x_{ij}:1 \leq i,j \leq n\}$ and $\mathscr{R} = \{(x_{ij}-x_{ji}^*,0),(x_{ij}x_{k\ell}-\delta_{jk}x_{i\ell},0): 1\leq i,j,k,\ell\leq n\}$, where $\delta_{ij}$ is the Kronecker delta. Consider the function $\Theta_0:D \to M_n(\mathbb{C})$, given by
    \begin{equation*}
        \Theta_0(x_{ij}) = U_{ij},\quad 1 \leq i,j \leq n,
    \end{equation*}
    where $U_{ij}$ is the matrix unit of $M_n(\mathbb{C})$ given by
    \begin{equation}
        (U_{ij})_{pq} = \delta_{ip}\delta_{jq},
    \end{equation}
    that is, it assigns $1$ for the entry in the $i$-th row and $j$-th column of $U_{ij}$ zero for its remaining entries. $\Theta_0$ is a representation for $(D,\mathscr{R})$ since the following is true:
    \begin{align*}
        \|\Theta(x_{ij}-x_{ji}^*)\| &= \|U_{ij}-U_{ji}^*\| = \|U_{ij}-U_{ij}\| = 0,\\
        \|\Theta(x_{ij}x_{k\ell}-\delta_{jk}x_{i\ell})\| &= \|U_{ij}U_{k\ell}-\delta_{jk}U_{i\ell}\| = 0.
    \end{align*}
    Furthermore, the pair $(D,\mathscr{R})$ is admissible. In fact, for every $\Theta_0$ representation of $(D,\mathscr{R})$ we have for every $1 \leq i \leq n$ that
    \begin{align*}
        \|\Theta(x_{ii})\|^2 &= \|\Theta(x_{ii})^*\Theta(x_{ii})\| = \|\Theta(x_{ii})\Theta(x_{ii})\| = \|\Theta(x_{ii})\|, 
    \end{align*}
    that is, $\|\Theta(x_{ii})\| \in \{0,1\}$ and then $\|\Theta(x_{ii})\| \leq 1$. On the other hand, for every $1 \leq i,j \leq n$
    \begin{align*}
        \|\Theta(x_{ij})\|^2 &= \|\Theta(x_{ij})^*\Theta(x_{ij})\| = \|\Theta(x_{ij})\Theta(x_{ij})\| = \|\Theta(x_{jj})\|^2 \leq 1. 
    \end{align*}
\end{itemize}
\end{example}

\begin{proposition}\label{prop:representation_C_star_seminorm} Let $D$ a set of generators and $\mathscr{R} \subset \tilde{\mathcal{A}}_D \times \mathbb{R}_+$ a family of relations such that the pair $(D,\mathscr{R})$ is admissible. Define $\Gamma$ to be the set of all representations of the pair $(D,\mathscr{R})$ for all possible C$^*$-algebras. Let $\tnorm{\cdot}:\mathcal{A}_D \to \mathbb{R}_+$ the function defined by
\begin{equation}\label{eq:representation_C_star_seminorm}
    \tnorm{a} := \sup_{\Theta \in \Gamma}\|\Theta(a)\|, \quad a \in \mathcal{A}_D,
\end{equation}
is a C$^*$-seminorm on $\mathcal{A}_D$.
\end{proposition}

\begin{proof} Firstly we show that $\tnorm{\cdot}$ is well-defined. In fact, let $a = \sum_{i=1}^m \lambda_m x_m \in \mathcal{A}_D$, where $\lambda_i \in \mathbb{C}$ and $x_i \in \mathcal{F}_{\mathcal{D}}$ for all $i$. We can write $x_i = a_{1,i} \cdots a_{\phi(i),i}$  for each $i = 1,...,m$, where $\phi:\{1,...,m\}\to \mathbb{N}$ is a function and $a_{p,q}\in \mathcal{D}$ for every $1\leq p \leq \phi(q)$ and each $1 \leq q \leq m$. Since $(D,\mathscr{R})$ is admissible, we have that for each $d \in D$ that there exists $K_d \in \mathbb{R}_+$ such that $\|\Theta(d)\| \leq K_d$ for every representation $\Theta$ and then $\|\Theta(d^*)\| \leq K_d$, and consequently for every $i$ and $\Theta$ we have
\begin{equation*}
    \|\Theta(x_i)\| \leq \prod_{j=1}^{\phi(i)}\|\Theta(a_{j,i})\| \leq K_{x_i} < \infty
\end{equation*}
where $K_{x_i}:=\prod_{j=1}^{\phi(i)}K_{a_{j,i}}$. Then,
\begin{equation*}
    \|\Theta(a)\| \leq \sum_{i=1}^m K_{x_i} < \infty,
\end{equation*}
and since it does not depend on $\Theta$, we get
\begin{equation*}
    \tnorm{a} \leq \infty,
\end{equation*}
and hence $\tnorm{\cdot}$ is well-defined. Now, we prove that this function is in fact a C$^*$-seminorm: for every $a,b \in \mathcal{A}_D$ and $\lambda \in \mathbb{C}$ we have
\begin{align*}
    \tnorm{\lambda a} &= \sup_{\Theta \in \Gamma}\|\Theta(\lambda a)\| = \sup_{\Theta \in \Gamma}\|\lambda\Theta(a)\| = \sup_{\Theta \in \Gamma}\{|\lambda|\|\Theta(a)\|\} = |\lambda|\sup_{\Theta \in \Gamma}\|\Theta(a)\| = |\lambda|\tnorm{a};\\
    \tnorm{a+b} &= \sup_{\Theta \in \Gamma}\|\Theta(a+b)\| = \sup_{\Theta \in \Gamma}\|\Theta(a)+\Theta(b)\| \leq \sup_{\Theta \in \Gamma}\{\|\Theta(a)\|+\|\Theta(b)\|\} \leq \sup_{\Theta \in \Gamma}\|\Theta(a)\| + \sup_{\Theta \in \Gamma}\|\Theta(b)\|.
\end{align*}
Then $\tnorm{\cdot}$ is a seminorm. Moreover,
\begin{align*}
    \tnorm{ab} &= \sup_{\Theta \in \Gamma}\|\Theta(ab)\| = \sup_{\Theta \in \Gamma}\|\Theta(a)\Theta(b)\| \leq \sup_{\Theta \in \Gamma}\{\|\Theta(a)\|\|\Theta(b)\|\} \leq \left(\sup_{\Theta \in \Gamma}\|\Theta(a)\|\right) \left(\sup_{\Theta \in \Gamma}\|\Theta(b)\|\right);\\
    \tnorm{a^*} &= \sup_{\Theta \in \Gamma}\|\Theta(a^*)\| = \sup_{\Theta \in \Gamma}\|\Theta(a)^*\| = \sup_{\Theta \in \Gamma}\|\Theta(a)\| = \tnorm{a};\\
    \tnorm{a^*a} &= \sup_{\Theta \in \Gamma}\|\Theta(a^*a)\| = \sup_{\Theta \in \Gamma}\|\Theta(a)^*\Theta(a)\| = \sup_{\Theta \in \Gamma}\|\Theta(a)\|^2 = \left(\sup_{\Theta \in \Gamma}\|\Theta(a)\|\right)^2 = \tnorm{a}^2.
\end{align*}
Therefore, $\tnorm{\cdot}$ is a C$^*$-seminorm. 
\end{proof}

By Proposition \ref{prop:representation_C_star_seminorm}, the map $\tnorm{\cdot}$, as in \eqref{eq:representation_C_star_seminorm}, is a C$^*$-seminorm for and $\mathcal{A}_D$ and then we may conclude by the proposition \ref{prop:C_star_seminorm_two_sided_ideal} that $N:=\{a\in \mathcal{A}_D: \tnorm{a} = 0\}$ is two-sided self-adjoint ideal. Furthermore, Theorem \ref{thm:quotient_star_normed_algebra} ensures that $\mathcal{A}_D/N$ is a normed $*$-algebra for the C$^*$-norm
\begin{equation}\label{eq:universal_norm}
    \|a + N\| = \tnorm{a}, \quad a \in \mathcal{A}_D.
\end{equation}
And by Proposition \ref{prop:completion_star_algebra_C_star_norm}, the norm completion $(\mathcal{A}_D/N)^{\|\cdot\|}$ is a C$^*$-algebra and it is our central definition of this chapter. 

\begin{definition}[Universal C$^*$-algebra] Let $D$ be a generator set and $\mathscr{R}$ be a family of relations on $\mathcal{A}_D$, such that the pair $(D,\mathscr{R})$ is admissible. The universal C$^*$-algebra generated by the set $D$ satifying the relations $\mathscr{R}$ is the C$^*$-algebra $(\mathcal{A}_D/N)^{\|\cdot\|}$, where $\|\cdot\|$ is the norm in \eqref{eq:universal_norm}, and we denote it by $C^*(D,\mathscr{R})$. 
\end{definition}

\begin{remark} Observe that $\mathcal{A}_D/N$ is dense on $C^*(D,\mathscr{R})$ by construction. This fact will be used further when we present examples of universal C$^*$-algebras and specially when we study the Cuntz-Krieger algebras and their generalizations, the Exel-Laca algebras. 
\end{remark}

Universal C$^*$-algebras generated by sets and relations have two main aspects. The first one is that their generators in fact satisfy the relations that their free algebra is imposed to obey. The second one is the universal property, which consists in the existence of a unique $*$-homomorphism between the universal algebra and any other C$^*$-algebra that admits a representation for the generating pair. The next proposition proves the first aspect and, after we present some important examples, we prove the universal property. From now on, given a universal C$^*$-algebra $C^*(D,\mathscr{R})$, we set $\Psi:\mathcal{A}_D \to C^*(D,\mathscr{R})$ as being the canonical projection map.

\begin{proposition}\label{prop:universal_C_star_algebra_satisfies_its_relations} Let $(D,\mathscr{R})$ be an admissible pair. Then the elements $\Psi(a)$, $a \in \mathcal{A}_D$, including the elements of $\tilde{\mathcal{A}}_D$ when $\mathscr{R} \not\subset \mathcal{A}_D \times \mathbb{R}_+$, satisfy the respective relations in $\mathscr{R}$ which they are imposed to obey before applying $\Psi$. That is, if $(a,t) \in \mathscr{R}$, then
\begin{equation*}
    \|\Psi(a)\| \leq t.
\end{equation*}
\end{proposition}

\begin{proof} Given $(a,t) \in \mathscr{R}$, then
\begin{equation*}
    \|\Psi(a)\| = \tnorm{a} = \sup_{\Theta \in \Gamma}\|\Theta(a)\| \leq \sup_{\Theta \in \Gamma} t = t.
\end{equation*} 
\end{proof}

\begin{example} We present next some well-known examples of universal C$^*$-algebras.
\begin{itemize}
    \item[$(a)$] Let $\mathfrak{A}$ a C$^*$-algebra. Take $D = A$ and $\mathscr{R}$ as the set of the identities from the axioms of C$*$-algebras, for instance, $((ab)^*-b^*a^*,0)$ for every $a.b \in A$. Then $C^*(A,\mathscr{R}) \simeq A$.
    \item[$(b)$] Set $D = \{x\}$ and $\mathscr{R} = \{(x-x^*,0),(x,1),(1-x^2,1)\}$. Then $C^*(D,\mathscr{R}) \simeq C_0((0,1])$. In fact, $\mathcal{A}_D$ is generated by only one element and then it is abelian. Consequently $\mathcal{A}_D/N$ is also abelian, and so it is $C^*(D,\mathscr{R})$. We indentify its spectrum $\Omega[C^*(D,\mathscr{R})]$. For every $\varphi \in \Omega[C^*(D,\mathscr{R})]$, since this C$^*$-algebra is generated by only one element, $\varphi$ is completely determined simply by evaluating $\varphi(\Psi(x))$, because its extension to the whole C$^*$-algebra is unique. In other words, the elements of $\Omega[C^*(D,\mathscr{R})]$ are determined by the equation
    \begin{equation*}
        \varphi_\lambda(\Psi(x)) = \lambda, \quad \varphi_\lambda \in \Omega[C^*(D,\mathscr{R})], \quad \lambda \in \mathbb{C}.
    \end{equation*}
    Now we determine the possible values for $\lambda$ by using Proposition \ref{prop:universal_C_star_algebra_satisfies_its_relations}. The relation $(x-x^*,0)$ gives
    \begin{equation*}
        \lambda = \varphi_\lambda(\Psi(x)) = \varphi_\lambda(\Psi(x^*)) = \varphi_\lambda(\Psi(x)^*) = \overline{\lambda},
    \end{equation*}
    and then $\lambda \in \mathbb{R}$. Now, since every $*$-homomorphism from a Banach $*$-algebra to a C$^*$-algebra is norm decreasing, the relation $(x,1)$ gives
    \begin{equation*}
        |\lambda| = |\varphi_\lambda(\Psi(x))| \leq \|\Psi(x)\| \leq 1.
    \end{equation*}
    On the other hand, by Corollary \ref{cor:positivity_t_equals_to_one} and $(1-x^2,1)$ we have that $\Psi(x) \geq 0$ and then $\lambda \geq 0$, since $\varphi_\lambda \in \sigma(\Psi(x))\cup \{0\}$ (see Theorem 1.3.4 of \cite{Murphy1990}). Then, $\lambda \in [0,1]$. However, for $\lambda = 0$ we necessarily have $\varphi_0(C^*(D,\mathscr{R})) = \{0\}$, and therefore $\varphi_0$ is not a character. We obtained a bijection between $\Omega[C^*(D,\mathscr{R})]$ and $(0,1]$, defined by the map
    \begin{equation}\label{eq:bijection_characters_semi_interval}
        \varphi_\lambda \mapsto \lambda.
    \end{equation}
    Let $(\varphi_{\lambda_n})_\mathbb{N}$ be a sequence on $\Omega[C^*(D,\mathscr{R})]$ and an element $\varphi_\lambda \in \Omega[C^*(D,\mathscr{R})]$. We have that $\varphi_{\lambda_n}\rightharpoonup \varphi_\lambda$ (weak$^*$ topology) if and only if $\varphi_{\lambda_n}(\Psi(x)) \to \varphi_\lambda(\Psi(x))$, which is equivalent to the convergence $\lambda_n \to \lambda$ and therefore, the map \eqref{eq:bijection_characters_semi_interval} is a homeomorphism. By the Gelfand representation Theorem (Corollary \ref{corollary:Gelfand_C_star}) we have that $C^*(D,\mathscr{R}) \simeq C_0((0,1])$.
    \item[$(c)$] Let $D = \{x,1\}$ and $\mathscr{R} = \{(x-x^*,0),(x,1),(1-x^2,1),(1-1^*,0),(1-1^2,0),(1x-x,0),(x1-x,0)\}$. Then $C^*(D,\mathscr{R}) \simeq C([0,1])$ by similar proof as done in the previous item.
    \item[$(d)$] Consider $D = \{x_{ij}:1 \leq i,j \leq n\}$ and $\mathscr{R} = \{(x_{ij}-x_{ji}^*,0),(x_{ij}x_{k\ell}-\delta_{jk}x_{i\ell},0): 1\leq i,j,k,\ell\leq n\}$, where $\delta_{ij}$ is the Kronecker delta. Then, $C^*(D,\mathscr{R}) \simeq M_n(\mathbb{C})$. Indeed, the $*$-homomorphism $\Theta: C^*(D,\mathscr{R}) \to M_n(\mathbb{C})$ maps each generator $x_{ij}$ to its respective matrix unit $U_{ij}$, which is the matrix that assigns $1$ to the entry on the $i$-th row at the $j$-column and zero at the remaining ones, is surjective since $\mathcal{U}:=\{U_{ij}: 1 \leq i,j \leq n\}$ is a basis for $M_n(\mathbb{C})$. We prove that $\Theta$ is a injective as well. We observe that the generator set is algebraically closed under involutions and products, and then any element $x \in C^*(D,\mathscr{R})$ is in the form
    \begin{equation*}
        x = \sum_{k = 1}^n \sum_{\ell = 1}^n \lambda_{k \ell} \Psi(x_{k \ell}),
    \end{equation*}
    where $\lambda_{k\ell} \in \mathbb{C}$ for every $k$ and $\ell$. Now, suppose that $\Theta(x) = 0$. Then,
    \begin{equation*}
        0=\sum_{k = 1}^n \sum_{\ell = 1}^n \lambda_{k \ell} U_{k \ell},
    \end{equation*}
    and since $\mathcal{U}$ is a basis, we have that $\lambda_{k \ell} = 0$ for all $k$ and $\ell$ and then $x=0$. Therefore, $\Theta$ is injective.
\end{itemize}
\end{example}

The next result presents and proves the universal property, one of the main properties that characterizes universal C$^*$-algebras.

\begin{theorem}[Universal property]\label{thm:universal_property} Let $(D,\mathscr{R})$ be an admissible pair, $B$ be a C$^*$-algebra and $\Theta_0:D \to B$ be a representation of $(D,\mathscr{R})$ on $B$. Then, there exists a unique $*$-homomorphism $\Upsilon: C^*(D,\mathscr{R}) \to B$ such that $\Upsilon(\Psi(d)) = \Theta_0(d)$ for every $d \in D$, that is, the diagram
\[
\begin{tikzcd}[column sep = huge, row sep = huge]
  D \arrow[d,"\Psi"] \arrow[r,"\Theta_0"] & B\\
  C^*(D,\mathscr{R}) \arrow[ur, "\Upsilon"] &
\end{tikzcd}
\]
commutes.
\end{theorem}

\begin{proof} Let $\Upsilon_0 : \mathcal{A}_D/N \to B$ be the mapping given by
\begin{equation*}
    \Upsilon_0(a+N) := \Theta(a),
\end{equation*}
where $\Theta$ is the unique $*$-homomorphism that extends $\Theta_0$ to $\mathcal{A}_D$, as granted by Proposition \ref{prop:homomorphism_extension_free_algebra}. We claim that $\Upsilon_0$ is well-defined, since if for given $a, b \in \mathcal{A}_D$ satisfying $a+N=b+N$, we have $a-b+N=0+N$ and then
\begin{equation*}
    0 = \tnorm{a-b+N} = \sup_{\Xi \in \Gamma}\|\Xi(a-b)\| \geq \|\Theta(a-b)\| \geq 0,
\end{equation*}
and hence $\Theta(a) = \Theta(b)$, that is, $\Upsilon_0(a+N) = \Upsilon_0(b+N)$, and the claim is proved. It is straightforward that $\Theta$ is a $*$-homomorphism between $\mathcal{A}_D$ and $B$, and consequently $\Upsilon_0$ is a $*$-homomorphism between $\mathcal{A}_D/N$ and $B$. Besides that, $\Upsilon_0$ is norm-decreasing. Indeed, for every $a \in \mathcal{A}_D/N$ we have
\begin{equation*}
    \|\Upsilon_0(a+N)\| = \|\Theta(a)\| \leq \sup_{\Theta_0 \in \Gamma}\|\Theta(a)\| = \|a+N\| = \tnorm{a+N}.
\end{equation*}
By density, $\Upsilon_0$ is extended to $C^*(D,\mathscr{R})$ to a norm-decreasing $*$-homomorphism $\Upsilon : C^*(D,\mathscr{R}) \to B$. Now we prove that the diagram of the statement of this theorem in fact commutes. Let $d \in D$. We have
\begin{equation*}
    \Upsilon(\Psi(d)) = \Upsilon(d+N) = \Theta(d) = \Theta_0(d),
\end{equation*}
and we conclude that $\Upsilon \circ \Psi = \Theta_0$, and therefore the diagram of the statement commutes. Now, for the uniqueness, suppose that there exists a $*$-homomorphism $\Phi:C^*(D,\mathscr{R}) \to B$ such that $\Phi(\Psi(d)) = \Theta_0(d)$ for every $d \in D$. Then, $\Phi(\Psi(d)) = \Theta_0(d) = \Upsilon(\Psi(d))$ for every $d \in D$, and since $D$ generates $C^*(D,\mathscr{R})$ we conclude that $\Phi = \Upsilon$.
\end{proof}

The next result shows the uniqueness of the universal C$^*$-algebra.

\begin{theorem}[Uniquenes of the Universal C$^*$-algebra]\label{thm:uniqueness_universal_algebra} Let $B$ be a C$^*$-algebra and $\pi: D \to B$ be a representation for the admissible pair $(D,\mathscr{R})$. If for every representation, $\Phi:D \to C$, of $(D,\mathscr{R})$ on any C$^*$-algebra $C$ there exists a unique $*$-homomorphism $\rho: B \to C$ satisfying $\Phi = \rho \circ \pi$, that is, the diagram
\[
\begin{tikzcd}[column sep = huge, row sep = huge]
  D \arrow[d,"\pi"] \arrow[r,"\Phi"] & C\\
  B \arrow[ur, "\rho"] &
\end{tikzcd}
\]
commutes, then $B \simeq C^*(D,\mathscr{R})$.
\end{theorem}

\begin{proof} Given a C$^*$-algebra $C$, the Theorem \ref{thm:universal_property} and the hypothesis gives that the following diagram commutes
\[
\begin{tikzcd}[column sep = huge, row sep = huge]
    & B \arrow[dl,"\Upsilon_2"] &\\
    B \arrow[dr,"\Upsilon_4"] & D \arrow[d,"\Psi_1"] \arrow[r,"\Psi_1"] \arrow[u,"\Psi_2"] \arrow[l,"\Psi_2"] & C^*(D,\mathscr{R}) \arrow[ul,"\Upsilon_3"] \\
    & C^*(D,\mathscr{R}) \arrow[ur,"\Upsilon_1"] &
\end{tikzcd}
\]
where $\Upsilon_1 : C^*(D,\mathscr{R}) \to C^*(D,\mathscr{R})$, $\Upsilon_2 : B \to B$, $\Upsilon_3 : C^*(D,\mathscr{R}) \to B$ and $\Upsilon_4 : B \to C^*(D,\mathscr{R})$ are the unique $*$-homomorphisms satisfying respectively $\Upsilon_1 \circ \Psi_1 = \Psi_1$, $\Upsilon_2 \circ \Psi_2 = \Psi_2$, $\Upsilon_3 \circ \Psi_1 = \Psi_2$ and $\Upsilon_4 \circ \Psi_2 = \Psi_1$. Immediately we obtain $\Upsilon_1 = \Id_{C^*(D,\mathscr{R})}$ and $\Upsilon_2 = \Id_B$. Moreover, we have that $\Upsilon_3 \circ \Upsilon_4:B \to B$ and $\Upsilon_4 \circ \Upsilon_3:C^*(D,\mathscr{R}) \to C^*(D,\mathscr{R})$ are $*$-homomorphisms satisfying
\begin{equation*}
    (\Upsilon_3 \circ \Upsilon_4)\circ \Psi_2 = \left(\Upsilon_3 \circ \Id_{C^*(D,\mathscr{R})} \circ \Upsilon_4\right)\circ \Psi_2 = \Psi_2 \quad \text{and} \quad (\Upsilon_4 \circ \Upsilon_3)\circ \Psi_1 = \left(\Upsilon_4 \circ \Id_B \circ \Upsilon_3\right)\circ \Psi_1 = \Psi_1,
\end{equation*}
By the uniqueness of $\Upsilon_1$ and $\Upsilon_2$, we conclude that
\begin{equation*}
    \Upsilon_3 \circ \Upsilon_4 = \Id_B \quad \text{and} \quad \Upsilon_4 \circ \Upsilon_3 = \Id_{C^*(D,\mathscr{R})}
\end{equation*}
and therefore $B \simeq C^*(D,\mathscr{R})$.
\end{proof}

\section{KMS states}
\label{section:kms}

As the last topic of this chapter, we introduce the notion of Kubo-Martin-Schwinger (KMS) state \cite{Kubo1957,MartinSchwinger1959}. This section is based mostly on O. Bratelli's books \cite{Bratteli1987vol1,Bratteli1996vol2}. Before we start with the Mathematical Background in its generality, we briefly present a heuristic motivation from Physics on this subject.

Consider a physical system with $n$ possible configurations and let $\{E_1,...,E_n\}$ be the possible distinct energies for each configuration. Also set $p=\{p_1,...,p_n\}$ as the respective probabilities associated to each state. The expected energy is given by
\begin{equation*}
    \overline{E}(p) = \sum_{j=1}^n p_j E_j.
\end{equation*}
If we simply consider as a physical principle only the minimal expected energy, the weight choice which minimizes the energy would be $1$ for the smallest $E_j$ and zero for the rest. However, only with that principle, the system gets too much rigid, being a deterministic system. In order to deal non-deterministic systems we must consider the `lack of information' of the system and such value is determined by the Shannon's information entropy, given by
\begin{equation*}
    S(p):= - \sum_{j=1}^np_j \ln p_j.
\end{equation*}
The principle of maximum entropy states that the probability distribution which best represents the physical system is the one which maximizes $S$. In order to conceal at same time the minimum energy and the maximum entropy one may consider the minimization of the quantity
\begin{equation*}
    F(p) = \overline{E}(p) - TS(p),
\end{equation*}
where $T$ is usually seen as the temperature times the Boltzmann constant. $F$ is said to be the free energy. In order minimizate it, we recall that $g(p):= \sum_{j=1}^np_j$ is the constant function equals to $1$, since it is a probability. We shall minimize $F$ by using the Lagrange multipliers' method. The minimum is obtained where $\nabla F$ is parallel to $\nabla g$, with the constraint $g(p)-1 = 0$. We have that
\begin{equation*}
    \partiald{F}{p_j} = \lambda \partiald{g}{p_j}, \quad j=1,...,n.
\end{equation*}
Notice that
\begin{equation*}
    \partiald{F}{p_j} = \lambda \partiald{}{p_j}\left(\sum_{k=1}^np_kE_k + T\sum_{k=1}^np_k \ln p_k\right) = E_j + T\left(\ln p_j + 1\right).
\end{equation*}
Considering the contraint $\partiald{g}{p_j} = 1$, we obtain
\begin{equation*}
    E_j + T \ln p_j + T = \lambda.
\end{equation*}
By taking $\beta = \frac{1}{T}$ we hae
\begin{equation*}
    \ln p_j = \frac{\lambda - T - E_j}{T} = c - \beta E_j,
\end{equation*}
where $c = \frac{\lambda - T}{T}$, and so
\begin{equation*}
    p_j = e^c e^{-\beta E_j}.
\end{equation*}
The condition $g \equiv 1$ gives 
\begin{equation*}
    e^c \sum_{k=1}^n e^{-\beta E_k} = 1 \implies e^c \frac{1}{\sum_{k=1}^n e^{-\beta E_k}},
\end{equation*}
and the final distribution obtained is given by
\begin{equation*}
    p_j = \frac{e^{-\beta E_j}}{\sum_{k=1}^n e^{-\beta E_k}}.
\end{equation*}
The probability distribution above is the so-called Gibbs state. Roughly speaking, the idea of quantization consists in to see the observables of the physical system as a non-commutative algebra. In our example of $n$ states we may take $M_n(\mathbb{C})$ and set
\begin{equation*}
    E = \begin{pmatrix}
            E_1 &0  &\cdots& 0 \\
            0 & E_2 &\cdots & 0 \\
            \vdots & \vdots & \ddots &\vdots \\
            0 &0 & \cdots & E_n
        \end{pmatrix}, \quad
    P = \begin{pmatrix}
            p_1 &0  &\cdots& 0 \\
            0 & p_2 &\cdots & 0 \\
            \vdots & \vdots & \ddots &\vdots \\
            0 &0 & \cdots & p_n
        \end{pmatrix}, \quad
    a = \begin{pmatrix}
            a_1 &0  &\cdots& 0 \\
            0 & a_2 &\cdots & 0 \\
            \vdots & \vdots & \ddots &\vdots \\
            0 &0 & \cdots & a_n
        \end{pmatrix},
\end{equation*}
where $a$ is a physical observable. Note that 
\begin{equation*}
    \overline{E}(p) = \text{tr}(EP),
\end{equation*}
which motivates to consider a linear functional $\omega$ s.t.
\begin{equation*}
    \omega(a) = \text{tr}(aP)
\end{equation*}
for physical observables as $a$. Observe that $\omega(1)=1$ and 
\begin{align*}
    \omega(a^*a) &= \text{tr}(a^*aP) = \text{tr}(a^*aP^{1/2}P^{1/2}) = \text{tr}(P^{1/2}a^*aP^{1/2}) \\
    &= \text{tr}((P^{1/2})^*a^*aP^{1/2}) =  \text{tr}((aP^{1/2})^*aP^{1/2}) \geq 0.
\end{align*}
The condition above suggests some positivity condition. Also,
\begin{equation*}
    P = \frac{e^{-\beta E}}{\text{tr}(e^{-\beta E})}
\end{equation*}
and $Z(\beta):= \text{tr}(e^{-\beta E})$ is the partition function. In this algebra, one can introduce some dynamics, since the Quantum Mechanics operators have a time evolution. In this case, it is the Heisenberg dynamics (see chapter 2 of \cite{Sakurai2017} for further details), which is a $1$-parameter group of automorphisms $\{\theta_t\}_{t \in \mathbb{R}}$ where the parameter is the time and it is described as follows: given $a \in M_n(\mathbb{C)}$ we set 
\begin{equation}\label{eq:Heisenberg_dynamics}
    \theta_t(a) = e^{itE}ae^{-itE}.
\end{equation}
Through some analyticity conditions we may extend the parameter for the entire plane $\mathbb{C}$. Also, the Gibbs state $\omega$ satisfies the condition
\begin{equation}\label{eq:KMS_intro}
    \omega(ba) = \omega(a \theta_{i\beta}(b)), \quad a,b \in M_n(\mathbb{C}),
\end{equation}
the so-called the Kubo-Martin-Schwinger condition. In 1967, R. Haag, N. M. Hugenholtz and M. Winnink proposed the KMS condition as a criterion for equilibrium states for the quantum setting in Statistical Mechanics in the paper \cite{HaagHugWin1967}.

In this section, given a normed vector space $X$, we denote the dual space of $X$ by $X'$. Also, $\sigma(X, X')$ denotes the weak topology on $X$. A function $f: \mathbb{R} \rightarrow X$ is said to be $\sigma(X,X')$-continuous when 
\begin{equation*}
    \eta \circ f: \mathbb{R} \rightarrow \mathbb{C}
\end{equation*}
is continuous for every $\eta \in X'$.

\begin{definition} A state on a C$^*$-algebra $A$ is a positive linear functional $\omega:A \to \mathbb{C}$ s.t. $\|\omega\| = 1$. 
\end{definition}

Our objective now is to define dynamical systems on a C$^*$-algebra, and the notion of one-parameter group of automorphisms is crucial for it.

\begin{definition} 
\label{def:groupofautomorphisms}
Given a C$^*$-algebra $A$, an {one-parameter group of $\ast$-automorphisms} is a family $\tau = \{ \tau_t \}_{t \in \mathbb{R}}$ where
\begin{itemize}
\item[$(i)$] $\tau_t$ is a $*$-automorphism on $A$ for every $t \in \mathbb{R}$; 
\item[$(ii)$] $\tau_{t+s} = \tau_t \circ \tau_s$ for every $t, s \in \mathbb{R}$;
\item[$(iii)$] $\tau_0 = id$.
\end{itemize}

A {C$^*$-dynamical system} is a pair $(A, \tau)$ consisting in a C$^*$-algebra $A$ and an one-parameter group of $\ast$-automorphisms $\tau = \{ \tau_t \}_{t \in \mathbb{R}}$ which is {strongly continuous}, that is, the mapping
\begin{align}
    \mathbb{R} &\to A, \nonumber\\
    t &\mapsto \tau_t(a),\label{eq:strong_continuity}
\end{align}
is norm continuous for every $a \in A$. In this case, $\tau$ is said to be the {dynamics} on $A$.
\end{definition}

\begin{remark}\label{remark:dynamicsextension} It is sufficient to define the dynamics $\tau$ on a dense $*$-subalgebra, since in this case it admits a unique extension to a dynamics on the whole C$^*$-algebra. For a proof for this fact, see Lemma 5.1.12 of \cite{Lima2019}. 
\end{remark}

\begin{example} The Heisenberg dynamics $\theta = \{\theta_t\}_{t \in \mathbb{R}}$, defined in \eqref{eq:Heisenberg_dynamics}, is a strongly continuous one-parameter group of $\ast$-automorphisms, and then $(M_n(\mathbb{C}),\theta)$ is a C$^*$-dynamical system.
\end{example}




In particular every $*$-automorphism on a C$^*$-algebra is an isometry and then $\| \tau_t \| = 1$ for every $t \in \mathbb{R}$. Moreover, by Theorem 3.10 of \cite{Brezis2011}, the norm continuity \eqref{eq:strong_continuity} in Definition \ref{def:groupofautomorphisms}, also called strong continuity, is equivalent to the $\sigma(X,X')$-continuity for the same map. Therefore, for every $(A,\tau)$ C$^*$-dynamical system, $\tau$ is an one-parameter $\sigma(A,A')$-continuous group of isometries, as in Definition 2.5.17 of \cite{Bratteli1987vol1}.

In other words, for every C$^*$-dynamical system $(A,\tau)$, the mapping
\begin{align*}
   \mathbb{R} &\to \mathbb{C},\\
   t &\mapsto \eta \circ \tau_t(a),
\end{align*}
is continuous for every $a \in A$ and every $\eta \in A'$. We shall see further in this section that there are special elements in $A$ such that the map above can be extended analytically for some strip of $\mathbb{C}$, in particular, some of them can be extended in the whole $\mathbb{C}$. Furthermore they are dense in $A$. These elements, called analytic elements, shall be used to define KMS states and they will be defined next.

For the definition below, given $\lambda \in (0,\infty]$, we define the set
\begin{equation*}
    I_\lambda = \{ z \in \mathbb{C} : \vert Im(z) \vert < \lambda \}.
\end{equation*}

\begin{definition} 
\label{def:analytic}
Let $X$ be a Banach space and $\tau$ be an one-parameter $\sigma(X, X')$-continuous group of isometries. We say an element $a \in X$ is {analytic} for $\tau$ when there exists $\lambda > 0$ and a function $f: I_\lambda \to X$, where the following holds
\begin{itemize}
    \item[$(a)$] $f(t) = \tau_t(a)$ for every $t \in \mathbb{R}$;
    \item[$(b)$] the function $\eta \circ f$ is analytic on $I_\lambda$ for every $\eta \in X'$.
\end{itemize}
In this setting, we write $$\tau_z(a) = f(z), \quad \text{$z \in I_\lambda$}.$$

When $\lambda = \infty$, the element $a$ is said to be {entire analytic} for $\tau$. Under this setting, we denote by $X_\tau$ the set of entire analytic elements.
\end{definition}

\begin{remark}
\label{remark:analyticvectorspace}
Observe that $X_\tau$ is a vector space. In fact, given $a_1$ and $a_2$ entire analytic elements and $\alpha \in \mathbb{C}$, let $f_1$ and $f_2$ be the corresponding functions as in Definition \ref{def:analytic} which witness the entire analyticity of $a_1$ and $a_2$ respectively. For every $t \in \mathbb{R}$ we have that
\begin{align}
\label{eqn:entirevectorspace}
f_1(t) + \alpha f_2(t) = \tau_t(a_1) + \alpha \tau_t(a_2) = \tau_t(a_1 + \alpha a_2).
\end{align}
Moreover, $\eta \circ (f_1 + \alpha f_2)$ is analytic for every $\eta \in X'$. Therefore $a_1 + \alpha a_2$ is entire analytic. 
\end{remark}

Further we shall see, when $X=A$ and $\tau$ is a dynamics in the definition above, that $A_\tau$ is a dense $*$-subalgebra of $A$.


The next result is Proposition 2.5.18 of \cite{Bratteli1987vol1}.

\begin{proposition} \label{prop:operator_on_integral} Let $X$ be a Banach space and $\{\tau_t\}_{t \in \mathbb{R}}$ be an one-parameter $\sigma(X,X')$-continuous group of isometries, and let $\mu$ be a Borel measure of bounded variation on $\mathbb{R}$. It follows that for each $a \in X$ there exists $b \in X$ such that
\begin{equation*}
    \eta(b) = \int \eta(\tau_t(a)) d\mu(t)
\end{equation*}
for any $\eta \in X'$. 
\end{proposition}

Next we present a result on existence and density of the entire analytic elements in a Banach space $X$: Proposition 2.5.22 of \cite{Bratteli1987vol1}. For a more detailed proof, see \cite{Lima2019}.

\begin{proposition} 
\label{prop:analyticdense}
If $\tau$ is a one-parameter $\sigma(X,X')$-continuous group of isometries, and $a \in X$, define for $n \in \mathbb{N}$,
\begin{align*}
a_n := \sqrt{\frac{n}{\pi}} \int_{-\infty}^{\infty} \tau_t(a)e^{-nt^2}dt.
\end{align*}
For each $n$, $a_n$ is an entire analytic element for $\tau$ and $\| a_n \| \leq \| a \|$ and $\lim_n a_n \to a$ on the $\sigma(X,X')$ topology. In particular, $X_\tau$, is a $\sigma(X,X')$-dense subspace of $X$.
\end{proposition}

\begin{remark} Proposition \ref{prop:operator_on_integral} is crucial for the proposition above, since it grants that
\begin{equation*}
    f_n(z) := \sqrt{\frac{n}{\pi}} \int_{-\infty}^{\infty} \tau_t(a)e^{-n(t-z)^2}dt
\end{equation*}
is well-defined for every $z \in \mathbb{C}$. Observe that, for $s \in \mathbb{R}$ we have $f_n(s) = \tau_s(a_n)$, for every $n \in \mathbb{N}$.
\end{remark}

As a consequence of Proposition \ref{prop:analyticdense} we have Corollary 2.5.22 of \cite{Bratteli1987vol1}, and this one is presented next.

\begin{corollary} Let $X$ be a Banach space and consider $\tau$ an one-parameter $\sigma(X,X')$-continuous group of isometries. Then, the map
\begin{align*}
    t \mapsto \tau_t
\end{align*}
is strongly continuous, i.e., the mapping
\begin{align*}
    \mathbb{R} &\to X\\
    t &\mapsto \tau_t(a),
\end{align*}
is norm continuous for every $a \in X$, and $X$ contains a norm-dense set of entire analytic elements. 
\end{corollary}

\begin{remark}\label{remark:equivalence_between_analytic_and_strongly_analytic} Given $X$ a Banach space and $\tau$ a one-parameter $\sigma(X,X')$-continuous group of isometries, for $\lambda > 0$, the element $a \in X$ is said to be {strongly analytic} on $I_\lambda$ if there exists $f: I_\lambda \to X$ satisfying (a) in Definition \ref{def:analytic} and such that the limit
\begin{align*}
\lim_{h \rightarrow 0} \frac{f(z + h) - f(z)}{h}
\end{align*}
exists for every $z \in I_\lambda$. It is true that $a \in X$ is analytic if and only if it is strongly analytic.
\end{remark}



\begin{lemma} 
\label{lemma:tauzplusw}
Given $X$ a Banach space and $\tau$ a one-parameter $\sigma(X,X')$-continuous group of isometries, let $a \in X_\tau$. For every $w \in \mathbb{C}$, we have $\tau_w(a) \in X_\tau$. Also, for every $z \in \mathbb{C}$ it follows that
\begin{align*}
\tau_{z+w}(a) = \tau_z \circ \tau_w(a).
\end{align*}
\end{lemma}

\begin{proof}
Since $a \in X_\tau$, there exists a function $f: \mathbb{C} \rightarrow X$ satisfying
\begin{itemize}
 \item $f(t) = \tau_t(a)$ for every $t \in \mathbb{R}$;
 \item $\eta \circ f$ is entire analytic for every $\eta \in X'$.
\end{itemize}
Now, observe that for $s \in \mathbb{R}$ and $\eta \in X'$, the function $\eta \circ f_s$, where $f_s: \mathbb{C} \rightarrow X$ is defined by
\begin{equation*}
    f_s(w) := f(s+w),
\end{equation*}
is entire analytic.  Let $g: \mathbb{C} \rightarrow X$ be a complex function given by $g(w) = \tau_s \circ \tau_w(a)$. It is straightforward to notice that $\eta \circ g$ is entire analytic for every $\eta \in X'$ because
\begin{align*}
\eta \circ g(w)
= \eta \circ \tau_s \circ \tau_w(a)
= (\eta \circ \tau_s) \circ f(z),
\end{align*}
for every $w \in \mathbb{C}$, $\eta \circ \tau_s \in X'$ and $a \in X_\tau$. Moreover, for each $t \in \mathbb{R}$, we have
\begin{align*}
\eta \circ g(t) = \eta \circ \tau_s \circ \tau_t(a) = \eta \circ \tau_{s+t}(a) = \eta \circ f(t+s) = \eta \circ f_s(t).
\end{align*}
Then by the Uniqueness Theorem for Analytic Functions (Theorem 6.9 of \cite{BakNewman2010}) we necessarily have that $\eta \circ g = \eta \circ f_s$. Since $\eta$ is arbitrary and $X'$ separates poins in $X$, we have that $g$ and $f_s$ coincide. Consequently,
\begin{align}
\label{eqn:taustauw}
\tau_s \circ \tau_w(a) = \tau_{s+w}(a), \quad w \in \mathbb{C}.
\end{align}

Now we fix $w \in \mathbb{C}$ and assume $s$ is a real variable. Define $f_w: \mathbb{C} \rightarrow X$ by $f_w(z) = f(w+z)$. Since $a$ is entire analytic, the function $\eta \circ f_w$ is entire analytic for every $\eta \in X*$. By equation \eqref{eqn:taustauw}, we have for all $t \in \mathbb{R}$ the equality
\begin{align*}
f_w(t) = f(w+t) = \tau_{t+w}(a) = \tau_t \circ \tau_w(a).
\end{align*}
We conclude that $\tau_w(a)$ is entire analytic and
\begin{align*}
\tau_z\circ \tau_w(a) = f_w(z) = f(w+z) = \tau_{z+w}(a).
\end{align*}
\end{proof}

\begin{definition} Let $(A,\tau)$ be a C$^*$-dynamical system. A set $B \subseteq A$ is said to be $\tau$-invariant if $\tau_t(B) \subseteq B$ for every $t \in \mathbb{R}$. 
\end{definition}

\begin{lemma}
\label{lemma:atausubalgebra}
Given a C$^*$-dynamical system $(A, \tau)$, the set $A_\tau$ is a $\ast$-subalgebra of $A$ and it is $\tau$-invariant.
\end{lemma}

\begin{proof} Remark \ref{remark:analyticvectorspace} gives that $A_\tau$ is a vector space, and Lemma \ref{lemma:tauzplusw} grants that $A_\tau$ is $\tau$-invariant, so it remains to prove that the product and the involution are algebraically closed in $A_\tau$.

Let $a, b \in A_\tau$, and $f_1, f_2: \mathbb{C} \rightarrow A$ complex functions which witness that $a$ and $b$ are entire analytic respectively, that is,
\begin{equation*}
    f_1(t) = \tau_t(a) \quad \text{and} \quad f_2(t) = \tau_t(b), \quad t \in \mathbb{R},
\end{equation*}
and for every $\eta \in A'$, $\eta \circ f_1$ and $\eta \circ f_1$ are analytic on $\mathbb{C}$. It is straightforward that for the function $f = f_1 f_2$ holds that $\eta \circ f$ is analytic for every $\eta \in X'$ and that
\begin{equation*}
    f(t) = f_1(t)f_2(t) = \tau_t(a)\tau_t(b) = \tau_t(ab), \quad t \in \mathbb{R},
\end{equation*}
and then $ab$ is entire analytic. Now, let $f_1^*$ given by $f_1^*(z) = f_1(\overline{z})^*$ for every $z \in \mathbb{C}$. It follows that 
\begin{align*}
f_1^*(t) = f_1(t)^* = \tau_t(a)^* = \tau_t(a^*), \quad t \in \mathbb{R}.
\end{align*}
For $\eta \in A'$, consider $\widetilde{\eta} \in A'$ given by $\widetilde{\eta}(c) = \overline{\eta(c^*)}$ for every $c \in A$. It is straightforward to observe that $$\eta \circ f_1^*(z) = \overline{\widetilde{\eta} \circ f_1(\overline{z})}.$$

We claim that the map
$$z \mapsto \overline{\widetilde{\eta} \circ f_1(\overline{z})}$$
is analytic. Indeed, we have that
\begin{align*}
\lim_{h \rightarrow 0} \frac{\overline{\widetilde{\eta} \circ f_1(\overline{z} + \overline{h})} - \overline{\widetilde{\eta} \circ f_1(\overline{z})}}{h}
&= \overline{\left( \lim_{h \rightarrow 0} \frac{\widetilde{\eta} \circ f_1(\overline{z} + \overline{h}) - \widetilde{\eta} \circ f_1(\overline{z})}{\overline{h}} \right)}.
\end{align*}
By the RHS of the equation above, the limit exists beacuse $\widetilde{\eta} \circ f_1$ is analytic. By Remark \ref{remark:equivalence_between_analytic_and_strongly_analytic}, we conclude that $\eta \circ f_1^*$ is analytic, then $a^* \in A_\tau$.
\end{proof}

\begin{definition}[KMS states] 
For a given $(A,\tau)$ C$^*$-dynamical system and $\beta \in \mathbb{R}$, a state $\omega$ on $A$ is said to be a $\tau$-$\text{KMS}_\beta$-state when
\begin{align*}
\omega(a\tau_{i\beta}(b)) = \omega(ba)
\end{align*}
for every $a, b$ in a norm-dense $\ast$-subalgebra $A_0$ of $\tau$-invariant entire analytic elements.
\end{definition}

We simply say $\omega$ is a $\text{KMS}_\beta$-state, that is, we omit $\tau$, when this one is understood. In this case, for $\beta$ fixed, we simply say $\omega$ is a KMS state.

\begin{remark}
In particular, every $\omega$ KMS$_0$-state, that is, for $\beta = 0$, we have
\begin{align}\label{eq:KMS_tracial}
\omega(ab) = \omega(ba),
\end{align}
for every $a, b \in A_0$. The state $\omega$ is continuous due to Theorem \ref{thm:positivefunctionalbounded}, and then \eqref{eq:KMS_tracial} holds for every $a, b \in A$. In this case, $\omega$ is said to be a {tracial} state.
\end{remark}

\chapter{Groupoids and Groupoid C\texorpdfstring{$^*$}{TEXT}-algebras}
\label{ch:Groupoid_algebras}
The main objective of this chaper is to provide an introduction to the theory of groupoid C$^*$-algebras. Our main obejective is to present the definitions of full and reduced groupoid C$^*$-algebras. Further, when we present the Cuntz-Krieger and Exel-Laca algebras, we will use the groupoid approach on these algebras, which in these particular cases the groupoid used is the Renault-Deaconu groupoid, and the Markov shifts spaces play a crucial role when we define such structures. This chapter is based on the J. Renault's book \cite{Renault1980}, the lecture notes of I. Putnam and A. Sims \cite{Putnam2016,Sims2017}, and the master theses of R. Frausino and R. Lima \cite{Frausino2018,Lima2019}.  
.

\section{Groupoids, Topological Groupoids and Etalicity}

We can understand a groupoid as a generalization of the concept of group, in the sense that there there is defined a product and there is a identity and inverses, however the product is not defined for every pair and in general the identity is not unique, since its product is not defined for every element. Moreover, the product of an element by its inverse can result in different identities, depending on which order the product is realized.

\begin{definition}[Groupoids]\label{def:groupoid}

A groupoid consists in a $4$-tuple $(G,G^{(2)},\cdot,^{-1})$, where $G$ is a set, $G^{(2)} \subseteq G \times G$ is named set of the composable parts, $\cdot:G^{(2)} \to G$ is the product (or composition) operation, and $^{-1}:G \to G$, satisfying the following conditions:
\begin{itemize}
    \item[$(G1)$] $(g^{-1})^{-1} = g$, for every $g \in G$;
    \item[$(G2)$] if $(g_1,g_2),(g_2,g_3) \in G^{(2)}$, then $(g_1g_2,g_3),(g_1,g_2g_3) \in G^{(2)}$, and in this case $(g_1g_2)g_3 = g_1(g_2g_3):= g_1g_2g_3$;
    \item[$(G3)$] $(g,g^{-1}) \in G^{(2)}$ for all $g \in G$, and given $(g_1,g_2) \in G^{(2)}$ it is true that
    \begin{equation*}
        (g_1^{-1}g_1)g_2 = g_1^{-1}(g_1g_2) = g_2 \quad \text{and} \quad g_1(g_2g_2^{-1}) = (g_1g_2)g_2^{-1} = g_1.
    \end{equation*}
\end{itemize}
\end{definition}
Unless we say the opposite, we always refer to a groupoid $(G,G^{(2)},\cdot,^{-1})$ by its set $G$. 

As we commented before, the groupoid can be seen as a generalization of the concept of group, and the definition above presents some properties which are similar to a group, namely the existence of inverse elements and an associative product. However, it remains to present the units of the groupoid. We define the unit space of a groupoid $G$ as being the set
\begin{equation*}
    G^{(0)}:= \{g^{-1}g:g \in G\} \subseteq G.
\end{equation*}
Note that $G^{(0)}= \{gg^{-1}:g \in G\}$ because $g^{-1}g = hh^{-1}$ by taking $h = g^{-1}$. Also we define two surjective maps
\begin{align*}
    r:G \to G^{(0)} \quad \text{and} \quad s:G \to G^{(0)},
\end{align*}
defined by
\begin{equation*}
    r(g):= gg^{-1} \quad \text{and} \quad s(g):= g^{-1}g,
\end{equation*}
and named ranged and source maps, respectively. The next lemma shows that the elements of the unit space can be interpreted as a generalization of the group notion of identity.

\begin{lemma}\label{lemma:properties_r_s} Given a groupoid $G$ and $g \in G$ we have that
\begin{itemize}
    \item[$(i)$] $(r(g),g),(g,s(g)) \in G^{(2)}$;
    \item[$(ii)$] $r(g)g = gs(g) = g$;
    \item[$(iii)$] $r(g^{-1}) = s(g)$;
    \item[$(iv)$] $g^{-1}$ is the unique element satisfying simultaneously $(g,g^{-1}) \in G^{(2)}$ and $gg^{-1} = r(g)$;
    \item[$(v)$] $g^{-1}$ is the unique element satisfying simultaneously $(g^{-1},g) \in G^{(2)}$ and $g^{-1}g = s(g)$.
\end{itemize}
\end{lemma}

\begin{proof} \textbf{($\mathbf{i}$):} by $(G3)$ we have that $(g,g^{-1}), (g^{-1},g)\in G$, where we used $(G1)$ for the pair $(g^{-1},g)$. By $(G2)$ we have $(gg^{-1},g),(g,g^{-1}g) \in G^{(2)}$, that is, $(r(g),g),(g,s(g)) \in G^{(2)}$.  

\textbf{($\mathbf{ii}$):} by $(i)$ and $(G3)$ for $g_2 = g_1 = g$ we have $r(g)g = gs(g) = gg^{-1}g = g$.

\textbf{($\mathbf{iii}$):} by $(G1)$ we have that $r(g^{-1}) = g^{-1}(g^{-1})^{-1} = g^{-1}g = s(g)$.

\textbf{($\mathbf{iv}$):} let $(g,h) \in G^{(2)}$ such that $gh = r(g) = gg^{-1}$. By $(G3)$ we have that $(g^{-1},g) \in G^{(2)}$ and by $(G2)$ it follows that $(g^{-1}g,h) \in G^{(2)}$, and by $(G3)$, $(ii)$ and $(iii)$ we get $h = g^{-1}gh = g^{-1}r(g) = g^{-1}s(g^{-1}) = g^{-1}$.

\textbf{($\mathbf{v}$):} similar steps as in $(iv)$.
\end{proof}

\begin{proposition}[Cancellation laws]\label{prop:cancellation_groupoids} Given a groupoid $G$, the following properties hold:
\begin{itemize}
    \item if $(g_1,h),(g_2,h) \in G^{(2)}$ such that $g_1 h = g_2 h$, then $g_1 = g_2$;
    \item if $(h,g_1),(h,g_2) \in G^{(2)}$ such that $h g_1 = h g_2$, then $g_1 = g_2$.
\end{itemize}
\end{proposition}

\begin{proof} We prove the first one, since the second one is proved by similar steps. By $(G3)$ we have that $(h,h^{-1}) \in G^{(2)}$, $g_1 = g_1 h h^{-1}$ and $g_2 = g_2 h h^{-1}$, and by $(G2)$ we get these products are in fact well defined and they are associative, then
\begin{equation*}
    g_1 = (g_1 h) h^{-1} = (g_2 h) h^{-1} = g_2. 
\end{equation*}
\end{proof}

The next result, among other consequences, it gives us another characterization of the unit space, in terms of terms of the range and source maps.

\begin{proposition}\label{prop:arrow_approach_groupoids} Let $G$ be a groupoid. The following holds:
\begin{itemize}
    \item[$(a)$] $(g,h) \in G^{(2)} \iff s(g) = r(h)$;
    \item[$(b)$] $r(gh) = r(g)$ and $s(gh) = s(h)$ for every $(g,h) \in G^{(2)}$;
    \item[$(c)$] $(gh)^{-1} = h^{-1} g^{-1}$ for every $(g,h) \in G^{(2)}$;
    \item[$(d)$] $r(x) = x = s(x)$ for every $x \in G^{(0)}$.
\end{itemize}
\end{proposition}

\begin{proof} \textbf{($\mathbf{a}$):} if $s(g) = r(h)$, then $g^{-1}g=hh^{-1}$. By Lemma \ref{lemma:properties_r_s} $(i)$ we have that
\begin{equation*}
    (g,s(g)) = (g,g^{-1}g) = (g,hh^{-1}) \in G^{(2)}
\end{equation*}
and
\begin{equation*}
    (r(h),h) = (hh^{-1},h) \in G^{(2)}.
\end{equation*}
Hence, by $(G2)$ and $(G3)$, it holds that $(g,hh^{-1}h) = (g,h) \in G^{(2)}$. Conversely, let $(g,h) \in G^{(2)}$. By $(G1)$ and $(G3)$ we have that $(g^{-1},g)\in G^{(2)}$ and then $(G2)$ implies that $(g^{-1},gh) \in G^{(2)}$. Again by $(G3)$, we obtain $s(g)h = g^{-1}gh = h = r(h)h$. By Proposition \ref{prop:cancellation_groupoids}, we conclude that $s(g) = g^{-1}g = r(h)$.

\textbf{($\mathbf{b}$):} by $(G1)$ and $(G3)$ we have that $(g,g^{-1}),(g^{-1},g) \in G^{(2)}$. Combining this result with $(g,h) \in G^{(2)}$ we obtain by $(G2)$ the following:
\begin{align*}
    (g^{-1},g),(g,h) \in G^{(2)} &\implies (g^{-1},gh) \in G^{(2)},\\
    (g,g^{-1}),(g^{-1},gh) \in G^{(2)} &\implies (gg^{-1},gh) = (r(g),gh) \in G^{(2)};
\end{align*}
and Lemma \ref{lemma:properties_r_s} $(ii)$ gives $r(g)gh = (r(g)g)h = gh = r(gh)gh$, and Proposition \ref{prop:cancellation_groupoids} let us obtain $r(g) = r(gh)$. The remaining identity is proved with similar steps.

\textbf{($\mathbf{c}$):} let $(g,h) \in G^{(2)}$. We have that
\begin{equation}\label{eq:h_inv_g_inv_G_2}
    s(h^{-1}) \stackrel{(\bullet)}{=} r(h) = s(g) \stackrel{(\bullet)}{=} r(g^{-1}) \implies (h^{-1},g^{-1}) \in G^{(2)},
\end{equation}
where in $(\bullet)$ we used the lemma \ref{lemma:properties_r_s} $(iii)$. Also,
\begin{equation}\label{eq:gh_comma_h_inv_G_2}
    s(gh) \stackrel{(\dagger)}{=} s(h) \stackrel{(\bullet)}{=} r(h^{-1}) \implies (gh,h^{-1}) \in G^{(2)},
\end{equation}
where in $(\dagger)$ we used the last item. The axiom $(G2)$ and the relations \eqref{eq:h_inv_g_inv_G_2} and \eqref{eq:gh_comma_h_inv_G_2} together imply $(gh,h^{-1}g^{-1}),(ghh^{-1},g^{-1}) \in G^{(2)}$ and then
\begin{equation*}
    (gh)(h^{-1}g^{-1}) = (ghh^{-1})g^{-1} \stackrel{(G3)}{=} gg^{-1} = r(g) = r(gh),
\end{equation*}
that is, $(gh)(h^{-1}g^{-1}) = r(gh)$. This fact, together with $(gh,h^{-1}g^{-1}) \in G^{(2)}$, implies by the uniqueness in Lemma \ref{lemma:properties_r_s} $(iv)$ that $(gh)^{-1} = h^{-1}g^{-1}$.

\textbf{($\mathbf{d}$):} given $x \in G^{(0)}$, we have $x = g^{-1}g$ for some $g \in G$. By $(b)$, we get
\begin{equation*}
    r(x) = r(g^{-1}g) = r(g^{-1}) = s(g) = g^{-1}g = x,
\end{equation*}
and
\begin{equation*}
    s(x) = s(g^{-1}g) = s(g) = g^{-1}g = x.
\end{equation*}
\end{proof}

The proposition above leads to a more graphical idea about what is a groupoid: a groupoid $G$ is a set of arrows from $G^{(0)}$ to itself, such that each arrow comes from its source to its range. The set $G^{(2)}$ represents the pair of arrows (i.e. orientend paths) that can be connected, that is, for $(g,h) \in G^{(2)}$, since the source of $g$ coincide with the range of $h$ we are able to define a resultant arrow connecting $s(h)$ to $r(g)$ as an definitive path between these points of $G^{(0)}$. Even the elements of $G^{(0)}$ are also arrows from elements of $G^{(0)}$ to themselves, and in this particular case these arrows are cyclic, since $r(x) = x = s(x)$ as presented in Proposition \ref{prop:arrow_approach_groupoids}. The figure \ref{fig:groupoid_arrow_approach} illustrates such picturesque descricption presented in this paragraph.

\begin{figure}[h]
\caption{A groupoid seen as a set of arrows from $G^{(0)}$ to itself. In $I$ we se two non-composable elements $g$ and $h$ of the groupoid. In $II$ we see an element of $G^{(0)}$ represented as an arrow, however it is important to understand the case illustrated here is a particular case, since in general there are elements on the groupoid that the range coincides with the source that do not belong to $G^{(0)}$ (see Example \ref{exa:groups}). The picture $III$ represents the groupoid product of two elements $g_2$ and $g_1$, and in $IV$ there is a representation of the inverse operation of an element $h_1$.}
\centering

\tikzset{every picture/.style={line width=0.75pt}} 
\resizebox{1\textwidth}{!}{

\begin{tikzpicture}[x=0.75pt,y=0.75pt,yscale=-1,xscale=1]

\draw [line width=0.75]    (676.52,247) .. controls (621.5,248) and (602.5,282) .. (683.5,287) .. controls (778.54,287.99) and (758.91,245.86) .. (692.22,247.6) ;
\draw [shift={(690.19,247.67)}, rotate = 357.76] [fill={rgb, 255:red, 0; green, 0; blue, 0 }  ][line width=0.08]  [draw opacity=0] (10.72,-5.15) -- (0,0) -- (10.72,5.15) -- (7.12,0) -- cycle    ;

\draw  [line width=0.75]  (470.5,75) .. controls (604.5,70) and (849.5,89) .. (904.5,157) .. controls (959.5,225) and (789.5,273) .. (800.5,335) .. controls (811.5,397) and (995.5,524) .. (901.5,601) .. controls (807.5,678) and (558.5,669) .. (475.5,667) .. controls (392.5,665) and (170.5,656) .. (86.5,580) .. controls (2.5,504) and (186.5,465) .. (157.5,398) .. controls (128.5,331) and (-24.5,294) .. (30.5,216) .. controls (85.5,138) and (336.5,80) .. (470.5,75) -- cycle ;
\draw [line width=0.75]    (329,288.75) .. controls (348.3,235.67) and (323.83,224.97) .. (301.43,215.3) ;
\draw [shift={(299,214.25)}, rotate = 383.5] [fill={rgb, 255:red, 0; green, 0; blue, 0 }  ][line width=0.08]  [draw opacity=0] (10.72,-5.15) -- (0,0) -- (10.72,5.15) -- (7.12,0) -- cycle    ;

\draw [line width=0.75]    (409,200.25) .. controls (398.22,230.14) and (402.81,246.1) .. (431.07,280.22) ;
\draw [shift={(432.83,282.33)}, rotate = 230.02] [fill={rgb, 255:red, 0; green, 0; blue, 0 }  ][line width=0.08]  [draw opacity=0] (10.72,-5.15) -- (0,0) -- (10.72,5.15) -- (7.12,0) -- cycle    ;

\draw [line width=0.75]    (262.3,522.67) .. controls (279.83,463.89) and (326.89,457.68) .. (356.97,504.65) ;
\draw [shift={(358.33,506.83)}, rotate = 238.7] [fill={rgb, 255:red, 0; green, 0; blue, 0 }  ][line width=0.08]  [draw opacity=0] (10.72,-5.15) -- (0,0) -- (10.72,5.15) -- (7.12,0) -- cycle    ;

\draw [line width=0.75]    (367,507.5) .. controls (395.91,466.32) and (442.69,456.85) .. (459.79,525.88) ;
\draw [shift={(460.3,528)}, rotate = 256.85] [fill={rgb, 255:red, 0; green, 0; blue, 0 }  ][line width=0.08]  [draw opacity=0] (10.72,-5.15) -- (0,0) -- (10.72,5.15) -- (7.12,0) -- cycle    ;

\draw [line width=0.75]    (261.63,537.33) .. controls (288.83,616.93) and (406.58,633.7) .. (458.19,542.06) ;
\draw [shift={(458.97,540.67)}, rotate = 478.78] [fill={rgb, 255:red, 0; green, 0; blue, 0 }  ][line width=0.08]  [draw opacity=0] (10.72,-5.15) -- (0,0) -- (10.72,5.15) -- (7.12,0) -- cycle    ;

\draw [line width=0.75]    (609.5,532.5) .. controls (636.7,612.1) and (752.6,628.7) .. (804.19,537.06) ;
\draw [shift={(804.97,535.67)}, rotate = 478.78] [fill={rgb, 255:red, 0; green, 0; blue, 0 }  ][line width=0.08]  [draw opacity=0] (10.72,-5.15) -- (0,0) -- (10.72,5.15) -- (7.12,0) -- cycle    ;

\draw [line width=0.75]    (804.5,523) .. controls (774.28,444.5) and (658.27,424.94) .. (610.22,518.42) ;
\draw [shift={(609.5,519.83)}, rotate = 296.59000000000003] [fill={rgb, 255:red, 0; green, 0; blue, 0 }  ][line width=0.08]  [draw opacity=0] (10.72,-5.15) -- (0,0) -- (10.72,5.15) -- (7.12,0) -- cycle    ;

\draw  [fill={rgb, 255:red, 0; green, 0; blue, 0 }  ,fill opacity=1 ] (803,529.5) .. controls (803,527.57) and (804.57,526) .. (806.5,526) .. controls (808.43,526) and (810,527.57) .. (810,529.5) .. controls (810,531.43) and (808.43,533) .. (806.5,533) .. controls (804.57,533) and (803,531.43) .. (803,529.5) -- cycle ;
\draw  [fill={rgb, 255:red, 0; green, 0; blue, 0 }  ,fill opacity=1 ] (290.67,212.5) .. controls (290.67,210.57) and (292.23,209) .. (294.17,209) .. controls (296.1,209) and (297.67,210.57) .. (297.67,212.5) .. controls (297.67,214.43) and (296.1,216) .. (294.17,216) .. controls (292.23,216) and (290.67,214.43) .. (290.67,212.5) -- cycle ;
\draw  [fill={rgb, 255:red, 0; green, 0; blue, 0 }  ,fill opacity=1 ] (323,294.83) .. controls (323,292.9) and (324.57,291.33) .. (326.5,291.33) .. controls (328.43,291.33) and (330,292.9) .. (330,294.83) .. controls (330,296.77) and (328.43,298.33) .. (326.5,298.33) .. controls (324.57,298.33) and (323,296.77) .. (323,294.83) -- cycle ;
\draw  [fill={rgb, 255:red, 0; green, 0; blue, 0 }  ,fill opacity=1 ] (433.67,286.5) .. controls (433.67,284.57) and (435.23,283) .. (437.17,283) .. controls (439.1,283) and (440.67,284.57) .. (440.67,286.5) .. controls (440.67,288.43) and (439.1,290) .. (437.17,290) .. controls (435.23,290) and (433.67,288.43) .. (433.67,286.5) -- cycle ;
\draw  [fill={rgb, 255:red, 0; green, 0; blue, 0 }  ,fill opacity=1 ] (408.67,193.83) .. controls (408.67,191.9) and (410.23,190.33) .. (412.17,190.33) .. controls (414.1,190.33) and (415.67,191.9) .. (415.67,193.83) .. controls (415.67,195.77) and (414.1,197.33) .. (412.17,197.33) .. controls (410.23,197.33) and (408.67,195.77) .. (408.67,193.83) -- cycle ;
\draw  [fill={rgb, 255:red, 0; green, 0; blue, 0 }  ,fill opacity=1 ] (680.67,247.5) .. controls (680.67,245.57) and (682.23,244) .. (684.17,244) .. controls (686.1,244) and (687.67,245.57) .. (687.67,247.5) .. controls (687.67,249.43) and (686.1,251) .. (684.17,251) .. controls (682.23,251) and (680.67,249.43) .. (680.67,247.5) -- cycle ;
\draw  [fill={rgb, 255:red, 0; green, 0; blue, 0 }  ,fill opacity=1 ] (258,529.5) .. controls (258,527.57) and (259.57,526) .. (261.5,526) .. controls (263.43,526) and (265,527.57) .. (265,529.5) .. controls (265,531.43) and (263.43,533) .. (261.5,533) .. controls (259.57,533) and (258,531.43) .. (258,529.5) -- cycle ;
\draw  [fill={rgb, 255:red, 0; green, 0; blue, 0 }  ,fill opacity=1 ] (358,512.5) .. controls (358,510.57) and (359.57,509) .. (361.5,509) .. controls (363.43,509) and (365,510.57) .. (365,512.5) .. controls (365,514.43) and (363.43,516) .. (361.5,516) .. controls (359.57,516) and (358,514.43) .. (358,512.5) -- cycle ;
\draw  [fill={rgb, 255:red, 0; green, 0; blue, 0 }  ,fill opacity=1 ] (457.33,534.83) .. controls (457.33,532.9) and (458.9,531.33) .. (460.83,531.33) .. controls (462.77,531.33) and (464.33,532.9) .. (464.33,534.83) .. controls (464.33,536.77) and (462.77,538.33) .. (460.83,538.33) .. controls (458.9,538.33) and (457.33,536.77) .. (457.33,534.83) -- cycle ;
\draw  [fill={rgb, 255:red, 0; green, 0; blue, 0 }  ,fill opacity=1 ] (605.67,525.5) .. controls (605.67,523.57) and (607.23,522) .. (609.17,522) .. controls (611.1,522) and (612.67,523.57) .. (612.67,525.5) .. controls (612.67,527.43) and (611.1,529) .. (609.17,529) .. controls (607.23,529) and (605.67,527.43) .. (605.67,525.5) -- cycle ;

\draw (685,225) node  [font=\Large]  {$s( x) =x=r( x)$};
\draw (391.83,237) node  [font=\Large]  {$g$};
\draw (415,172) node  [font=\Large]  {$s( g)$};
\draw (435,301) node  [font=\Large]  {$r( g)$};
\draw (350,241) node  [font=\Large]  {$h$};
\draw (324,312) node  [font=\Large]  {$s( h)$};
\draw (294,191) node  [font=\Large]  {$r( h)$};
\draw (113,232) node  [font=\Huge]  {$G^{( 0)}$};
\draw (683,302) node  [font=\Large]  {$x$};
\draw (422,457) node  [font=\Large]  {$g_{2}$};
\draw (361,623) node  [font=\Large]  {$g_{2} g_{1}$};
\draw (362,526) node  [font=\large]  {$r( g_{1}) =s( g_{2})$};
\draw (312,456) node  [font=\Large]  {$g_{1}$};
\draw (232,528) node  [font=\Large]  {$s( g_{1})$};
\draw (491,534) node  [font=\Large]  {$r( g_{2})$};
\draw (247,188) node  [font=\LARGE] [align=left] {I.};
\draw (600,186) node  [font=\LARGE] [align=left] {II.};
\draw (245,441) node  [font=\LARGE] [align=left] {III.};
\draw (600,443) node  [font=\LARGE] [align=left] {IV.};
\draw (838,527) node  [font=\Large]  {$r( h_{1})$};
\draw (578,524) node  [font=\Large]  {$s( h_{1})$};
\draw (699,618) node  [font=\Large]  {$h_{1}$};
\draw (702,437) node  [font=\Large]  {$h^{-1}_{1}$};

\end{tikzpicture}

}

\label{fig:groupoid_arrow_approach}
\end{figure}
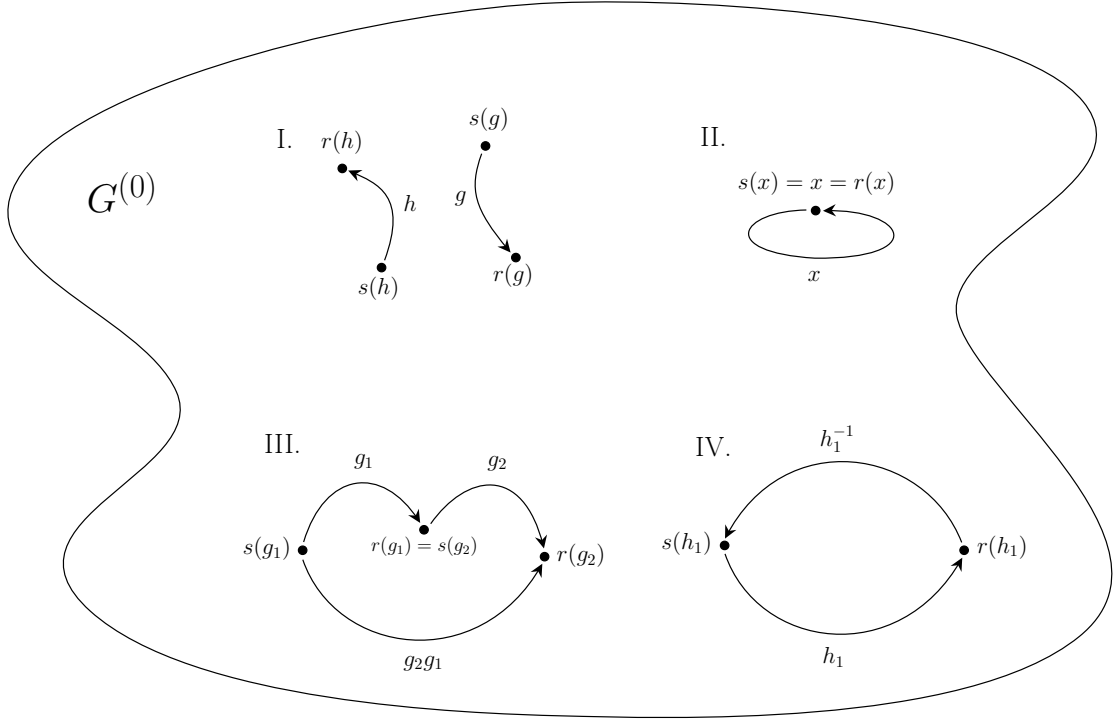

Now we present some examples of groupoids. As we mentioned previously, a groupoid is a generalization of the notion of group, and therefore groups are groupoids.

\begin{example}[Groups]\label{exa:groups} Consider a group $G$ and its unit $e$. By setting $G^{(0)} := \{ e \}$ and $G^{(2)} = G \times G$, there is a unique way to define the range and source maps, namely $r(g) := e =: s(g)$. The groupoid axioms in this case follows directly from the group axioms. The figure \ref{fig:group_groupoid_arrows} illustrates the group by the arrow approach of the groupoid structure.

\begin{figure}[h]
\caption{A group seen as a groupoid and represented by arrows. The unit $e$ is the unique element of $G^{(0)}$ and the elements $g_1$, $g_2$ and $g_3$ are elements of the group.}
\centering

\tikzset{every picture/.style={line width=0.75pt}} 

\begin{tikzpicture}[x=0.75pt,y=0.75pt,yscale=-1,xscale=1]

\draw [line width=0.75]    (122.95,118.75) .. controls (120.43,124.66) and (68.2,219) .. (128.2,224.5) .. controls (184.96,219.51) and (143.71,142.17) .. (135.03,121.94) ;
\draw [shift={(133.95,119.25)}, rotate = 430.56] [fill={rgb, 255:red, 0; green, 0; blue, 0 }  ][line width=0.08]  [draw opacity=0] (10.72,-5.15) -- (0,0) -- (10.72,5.15) -- (7.12,0) -- cycle    ;

\draw  [color={rgb, 255:red, 208; green, 2; blue, 27 }  ,draw opacity=1 ][fill={rgb, 255:red, 208; green, 2; blue, 27 }  ,fill opacity=1 ] (125.67,109.5) .. controls (125.67,107.57) and (127.23,106) .. (129.17,106) .. controls (131.1,106) and (132.67,107.57) .. (132.67,109.5) .. controls (132.67,111.43) and (131.1,113) .. (129.17,113) .. controls (127.23,113) and (125.67,111.43) .. (125.67,109.5) -- cycle ;
\draw [line width=0.75]    (125.06,99.47) .. controls (122.17,93.73) and (83.84,-7.06) .. (40.91,35.21) .. controls (8.07,81.78) and (93.77,100.21) .. (114.81,106.64) ;
\draw [shift={(117.57,107.55)}, rotate = 200.77] [fill={rgb, 255:red, 0; green, 0; blue, 0 }  ][line width=0.08]  [draw opacity=0] (10.72,-5.15) -- (0,0) -- (10.72,5.15) -- (7.12,0) -- cycle    ;

\draw [line width=0.75]    (138.87,107.71) .. controls (145.01,105.81) and (250.73,84.58) .. (216.1,35.28) .. controls (175.57,-4.77) and (143.29,76.73) .. (133.49,96.43) ;
\draw [shift={(132.13,99)}, rotate = 300.24] [fill={rgb, 255:red, 0; green, 0; blue, 0 }  ][line width=0.08]  [draw opacity=0] (10.72,-5.15) -- (0,0) -- (10.72,5.15) -- (7.12,0) -- cycle    ;

\draw (147.6,115.4) node  [font=\Large,color={rgb, 255:red, 208; green, 2; blue, 27 }  ,opacity=1 ]  {$e$};
\draw (25.2,22) node  [font=\Large,color={rgb, 255:red, 0; green, 0; blue, 0 }  ,opacity=1 ]  {$g_{1}$};
\draw (235.2,18.4) node  [font=\Large,color={rgb, 255:red, 0; green, 0; blue, 0 }  ,opacity=1 ]  {$g_{2}$};
\draw (129.6,238.4) node  [font=\Large,color={rgb, 255:red, 0; green, 0; blue, 0 }  ,opacity=1 ]  {$g_{3}$};

\end{tikzpicture}

\label{fig:group_groupoid_arrows}
\end{figure}
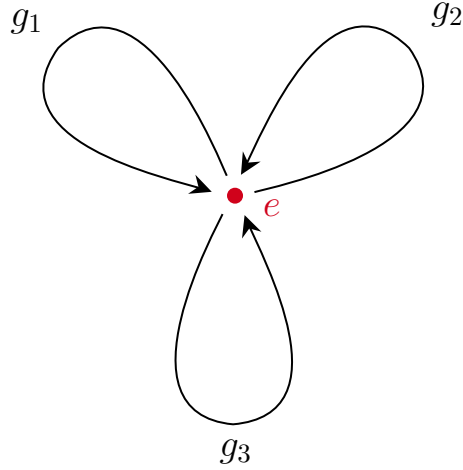
\end{example}

A groupoid that in general is not a group is the one that arises from equivalence relations on a set. In particular, the Renault-Deaconu groupoid, which we will present later, is the most important groupoid of this thesis and it comes from an equivalence relation. For now, we present the general setting for an arbitrary equivalence relation.

\begin{example}[Groupoids from equivalence relations] Let $X$ be a set and consider an equivalence relation $\sim$ on it. Define
\begin{equation*}
    G := \{ (x, y) \in X \times X: x \sim y \} \quad \text{ and } \quad G^{(2)} := \{((x,y),(y,z)): x \sim y, y \sim z \},
\end{equation*}
with the product and inversion given respectively by
\begin{equation*}
    (x,y) \cdot (y,z) := (x,z) \quad \text{and} \quad (x,y)^{-1} := (y,x).
\end{equation*}
The reflexivity and transitivity axioms of equivalence relation we have that the product is well-defined. On the other hand, the asymmetry axiom makes the inverse map well-defined as well. Therefore $G$ is a groupoid. Consequently, we have
\begin{align*}
    G^{(0)} &= \{ (x, x): x \in X \},
\end{align*}
with the range and source maps being chosen\footnote{Alternatively, we may exchange the definitions of $r$ and $s$.} as $r(x,y) = (x,x)$ and $s(x,y) = (y,y)$.
\end{example}

A more concrete example is shown next.

\begin{example}[groupoid from the general linear groups] Take $G = \bigsqcup_{n \in \mathbb{N}} GL_n(\mathbb{C})$ where $GL_n(\mathbb{C})$ is the general linear group of degree $n$, that is, the set of all invertible complex $n \times n$ matrices. In this case, $G^{(2)} = \bigsqcup_{n \in \mathbb{N}}GL_n(\mathbb{C}) \times GL_n(\mathbb{C})$, that is, $G^{(2)}$ is the set of pairs of matrices of same dimension. The product an inverse maps are the usual ones for matrices. Also, $G^{(0)}$ is the set of all usual indentity matrices.
\end{example}

The example above can be extended for $G$ being a disjoint union of groups, where the composable parts are the pairs of elements in a same group. The groupoid product is the usual group operations on each group, and the set of units consists in all group indentities.

From now on, for any groupoid $G$ we use the notations $r^{-1}(x):= r^{-1}(\{x\})$ and $s^{-1}(x):= s^{-1}(\{x\})$ for every $x\in G^{(0)}$. We introduce next some special subsets on a groupoid, such as fibers related to the range and source maps and isotropy groups.

\begin{definition}[$r/s$-fibers] Let $G$ be a groupoid and $x,y \in G^{(0)}$. Define the sets
\begin{equation*}
    G^x := r^{-1}(x) \quad \text{and} \quad G_x:= s^{-1}(x),
\end{equation*}
respectively named $r$-fiber of $G$ over $x$ and $s$-fiber of $G$ over $x$. Also define their intersection $G^x_y:= G^x\cap G_y = r^{-1}(x) \cap s^{-1}(y)$.
\end{definition}

\begin{proposition} Given a groupoid $G$ and $x \in G^{(0)}$. The set $G^x_x \subseteq G$, endowed with the product and inverse operations inherited from the groupoid, is a group. 
\end{proposition}

\begin{proof} First we prove that the product and inverse maps restricted to $G_x^x$ are well-defined. Indeed, for every $g,h \in G_x^x$ the pair $(g,h)$ is composable by Proposition \ref{prop:arrow_approach_groupoids} since $s(g) = x = r(h)$. Moreover, by the same proposition, item $(b)$ it follows that
\begin{equation*}
    r(gh) = r(g) = x = s(h) = s(gh),
\end{equation*}
and then $gh \in G_x^x$. On the other hand, by Lemma \ref{lemma:properties_r_s}, we have that
\begin{equation*}
    r(g^{-1}) = s(g) = x = r(g) = s(g^{-1}),
\end{equation*}
and hence $g^{-1} \in G^x_x$. Now, we prove that $G_x^x$ satisfies the axioms of a group. The associativity follows straightforward from $(G2)$. We claim that $x$ is the group identity. Indeed, we have that $x \in G^x_x$ by Proposition \ref{prop:arrow_approach_groupoids} $(d)$. In addition, by definition of range and source maps, for every $g \in G^x_x$ one gets
\begin{equation}\label{eq:isotropy_inverse}
    g^{-1}g = s(g) = x = r(g) = gg^{-1}.
\end{equation}
The group inverse element of $g$ is its groupoid inverse element $g^{-1}$, as it is also proved by the identities of \eqref{eq:isotropy_inverse}. 
\end{proof}

The last proposition allows us to define the isotropy group.

\begin{definition}[Isotropy group] Let $G$ be a groupoid and $x \in G^{(0)}$. The group $G^x_x$ is called the isotropy group at $x$. 
\end{definition}

From this point we introduce the following notation: given $x \in G^{(0)}$ we use subscript (respect. superscript) $x$ to denote elements of $G_x$ (respect. $G^x$), for example $h_x \in G_x$ (respect. $h^x \in G^x$. We use both notations for elements in the isotropy group at $x$, for instance, $h_x^x \in G_x^x$. 

In groups, vector spaces and algebras, we have the notion of topological groups, topological vector spaces and topological algebras, consisting into endowing these structures which make their operations continuous. The same can be done for groupoids and, under some assumptions on the chosen topology, we may define a $*$-algebra of compactly supported continuous complex functions on the groupoid, which essentially allows us to construct the groupoid C$^*$-algebras.

\begin{definition}[Topological groupoid] Given a groupoid $G$, we say that $G$ is a topological groupoid when it is endowed with a topology such that the inverse and product operations are continuous, where for the last one we endowed $G^{(2)}$ with the subspace topology of the product topology on $G \times G$.
\end{definition}

\begin{remark}\label{remark:r_and_s_continuous} It is an immediate consequence of the definitions of topological groupoid, range and source maps that $r$ and $s$ are continuous.
\end{remark}

We define next the notion of etalicity, the main assumption to construct the groupoid C$^*$-algebras.

\begin{definition}[Etalicity] 
A topological groupoid is said to be \'etale when the range and source maps are local homeomorphisms.
\end{definition}

\begin{example}
\label{ex:discreteetale}
Every discrete groupoid $G$ is \'etale, because the subsets $\{ g \}, \{ r(g) \}$ and $\{ s(g) \}$ are open in $G$ and then the maps $r\vert_{\lbrace g \rbrace}: \lbrace g \rbrace \rightarrow \lbrace r(g) \rbrace$, $s\vert_{\lbrace g \rbrace}: \lbrace g \rbrace \rightarrow \lbrace s(g) \rbrace$ are homeomorphisms.
\end{example}

\begin{definition} Let $\mathcal{U}$ be an open subset of an \'etale groupoid $G$. We say $\mathcal{U}$ is an open bisection of $G$ whenever $r(\mathcal{U}), s(\mathcal{U})$ are open in $G^{(0)}$, and the restricted maps $r\vert_{\mathcal{U}}: \mathcal{U} \rightarrow r(\mathcal{U})$ and $s\vert_{\mathcal{U}}: \mathcal{U} \rightarrow s(\mathcal{U})$ are homeomorphisms.
\end{definition}

The importance of etalicity lies on the description of topological groupoids: the basis of the topology can be described by open bisections. As we see next, not only these sets forms a topological basis for the groupoid but also it is crucial to describe the compactly supported complex continuous functions on the groupoid.

\begin{proposition} On an \'etale groupoid $G$, the family of open bisections of $G$ is an open basis for its topology.
\label{prop:openbisectionbase}
\end{proposition}

\begin{proof}
Given $U$ an open set of $G$, we claim that every $g \in U$ is contained in an open bisection $\mathcal{U}_g$ s.t. $g \in \mathcal{U}_g \subset U$. Indeed, by etalicity, $r$ and $s$ are local homeomorphisms and that grants the existence of two open neighborhoods of $g$, $R_g$ and $S_g$, satisfying that
\begin{equation*}
    r\vert_{R_g}: R_g \to r(R_g) \quad \text{and} \quad s\vert_{S_g}: S_g \to s(S_g)
\end{equation*}
are homeomorphisms. Since every homeomorphism is an open map we notice that $r(R_g)$ and $s(S_g)$ are open neighborhoods of $G^{(0)}$, By taking the non-empty set $\mathcal{U}_g = R_g \cap S_g \cap U$ we have that $r(\mathcal{U}_g)$ is open in $r(R_g)$ and $s(\mathcal{U}_g)$ is open in $s(S_g)$, and hence both $r(\mathcal{U}_g)$ and $s(\mathcal{U}_g)$ are open in $G^{(0)}$ and so
\begin{equation*}
    r\vert_{\mathcal{U}_g}: \mathcal{U}_g \to r(\mathcal{U}_g) \quad \text{and} \quad s\vert_{\mathcal{U}_g}: \mathcal{U}_g \to s(\mathcal{U}_g)
\end{equation*}
are homeomorphisms. We conclude that for every open set $U$ it follows that 
\begin{equation*}
    U = \bigcup_{g \in U}\mathcal{U}_g,
\end{equation*}
that is, the family of open bisections is a basis for the topology.
\end{proof}

\begin{remark} Note that in the proposition above we proved that the open bisections are an analytic basis for the topology on $G$, and then it is also a synthetic basis that generates that topology. 
\end{remark} 

\begin{proposition}\label{prop:GxSecondCountable} Let $G$ be an \'etale groupoid. For every $x \in G^{(0)}$, the subspace topology of $G^x$ and $G_x$ are equivalent to the discrete topology. In addition, if $G$ is second countable, then $G^x$ and $G_x$ are countable sets.
\end{proposition}
\begin{proof}
Given $g \in G^x$, by Proposition \ref{prop:openbisectionbase} there is an open bisection $\mathcal{U}_g$ containing $g$. We claim that $\mathcal{U}_g \cap G^x = \{ g \}$. In fact, for any $h \in \mathcal{U}_g \cap G^x$ we have $r(h) = x$, and by the injectivity of $r$ on $\mathcal{U}_g$ we conclude that $h = g$, and therefore $\{ g \}$ is open in $G^x$, that is, the subspace topology on $G^x$ is the discrete topology. Similarly, for $g \in G^x$, by exchaning $r$ by $s$ in the proof above, we have $\mathcal{U}_g \cap G_x = \{ g \}$.

Now suppose that $G$ is second countable. Then, $G^x$ and $G_x$ are also second countable. Since both fibers are discrete, the families of pairwise disjoint open sets $\{\{ g \}:g \in G^x\}$ and $\{\{ g \}:g \in G_x\}$ are countable, and therefore both fibers are also countable.
\end{proof}

\begin{proposition} \label{prop:G0clopen} For every Hausdorff \'etale groupoid $G$, the set $G^{(0)}$, endowed with the subspace topology, is a clopen subset of $G$.
\end{proposition}
\begin{proof}
First we prove $G^{(0)}$ is closed. In fact, let $(x_d)_D$ be a in $G^{(0)}$ such that $\lim_{d \in D} x_d = x \in G$. By Proposition \ref{prop:arrow_approach_groupoids} $(d)$ we get $x_d = r(x_d)$ for every $d \in D$ and by the continuity of $r$ we obtain 
\begin{equation*}
    x = \lim_{d \in D} x_d = \lim_{d \in D} r(x_d) = r(x) \in G^{(0)},
\end{equation*}
and we conclude that $G^{(0)}$ is closed. Now we prove that $G^{(0)}$ is open as well. Indeed, take $x \in G^{(0)}$ and let $\mathcal{U}_x \subset G$ be an open bisection that contains $x$. The restriction $r\vert_{\mathcal{U}}:\mathcal{U} \to r(\mathcal{U})$ is a homeomorphism and then $r(\mathcal{U}) \subseteq G^{(0)}$ is an open set containing $x$ due to Proposition \ref{prop:arrow_approach_groupoids} $(d)$. Then,
\begin{equation*}
    G^{(0)} = \bigcup_{x \in G^{(0)}} r(\mathcal{U}_x),
\end{equation*}
a union of open sets, and therefore that $G^{(0)}$ is open.
\end{proof}

\begin{definition}[Isotropy bundle]\label{def:isotropy_bundle} Let $G$ be a groupoid. The isotropy bundle of $G$ is the set
\begin{equation*}
    \Iso(G):= \bigcup_{x \in G^{(0)}} G_x^x = \{g \in G: r(g) = s(g)\}.
\end{equation*}
\end{definition}

\begin{remark}\label{remark:isotropy_bundle_is_subgroupoid} In the notations of the definition above, $\Iso(G)$ is a subgroupoid of $G$. 
\end{remark}

\begin{lemma}\label{lemma:isotropyclosed}
If $G$ is a locally compact Hausdorff (LCH) second-countable \'etale groupoid, then for every $g \in G$ such that $r(g) \neq s(g)$, there exists an open bisection $\mathcal{U}$ containing $g$ separating $r(g)$ from $s(g)$, in the sense that $r(\mathcal{U}) \cap s(\mathcal{U}) = \emptyset$. Furthermore, $\Iso(G) \cap \mathcal{U} = \emptyset$ and, in particular, its isotropy bundle is closed.
\end{lemma}
\begin{proof} For the first claim, assume that there is $g \in G \setminus \Iso(G)$, that is, $g \in G$ and $r(g) \neq s(g)$, satisfying
\begin{equation}\label{eq:rU_cap_sU_non_empty}
    r(\mathcal{U}) \cap s(\mathcal{U}) \neq \emptyset
\end{equation}
for all open bisections where $g \in \mathcal{U}$. Because of Proposition \ref{prop:openbisectionbase}, the etalicity of $G$ states that the open bisections forms a basis for the topology. Furthermore, since $G$ is second-countable there is a countable subfamily $\{\mathcal{U}_n\}_{\mathbb{N}}$ of the open bisections, which is also a basis for $G$, and hence every open neighborhood of $g$ contains at least one $\mathcal{U}_m$ for some $m \in \mathbb{N}$. In particular, the equality \eqref{eq:rU_cap_sU_non_empty} ensures the existence of $g_n,h_n \in \mathcal{U}_n$ satisfying $r(g_n) = s(h_n)$ for every $n \in \mathbb{N}$. Since every locally compact Hausdorff second-countable space is metrizable, we conclude by we proved so far that both sequences $(g_n)_\mathbb{N}$ and $(h_n)_\mathbb{N}$ converge to $g$. Consequently, by the continuity of $r$ and $s$, we obtain
\begin{equation*}
    r(g_n) \to r(g) \quad \text{and} \quad r(g_n) = s(h_n) \to s(h),
\end{equation*}
that is, $r(g) = s(g)$, a contradiction because $g \notin \Iso(G)$. Therefore, for any $g \in G \setminus \Iso(G)$ there exists an open bisection $\mathcal{U}_g$ containing $g$ and satisfying $r(\mathcal{U}_g) \cap s(\mathcal{U}_g) = \emptyset$ and then $\mathcal{U}_g \cap \Iso(G) = \emptyset$. Moreover,
\begin{equation*}
    G\setminus \Iso(G) = \bigcup_{g \in G\setminus \Iso(G)} \mathcal{U}_g
\end{equation*}
is open and therefore $\Iso(G)$ is closed.
\end{proof}

The next lemma is an essential auxiliary resuts that allow us to grant that involution and convolution product defined to construct the groupoid C$^*$-algebreas are well-defined and they work as a sort of closure of the groupoid operations on the family of the open bisections: its first statement grants that the groupoid inversion operation applied on a open bisection returns another open bisection, while the second one shows that by taking the allowed products between two open bisections we obtain another open bisection.

\begin{lemma} \label{lemma:inverse_product_maps_on_open_bisections_is_open_bisection}
Consider a LCH second countable \'etale groupoid $G$. Then,
\begin{itemize}
    \item[$(i)$] for every open bisection $\mathcal{U} \subset G$, its image under the inverse operation, given by 
    \begin{equation*}
        \mathcal{U}^{-1} = \{ g^{-1}: g \in \mathcal{U}\},
    \end{equation*}
    is an open bisection;
    \item[$(ii)$] for every pair of open bisections $(\mathcal{U},\mathcal{V}) \subset G \times G$, their image under the product operation restricted to the composable ordered pairs, given by
    \begin{align*}
        \mathcal{U}\mathcal{V} = \{ gh: g \in \mathcal{U}, h \in \mathcal{V}, (g,h) \in G^{(2)}\},
    \end{align*}
    is an open bisection.
\end{itemize}
\end{lemma}

\begin{proof} For $(i)$, we recall that the inverse map is continuous and it is its own inverse, that is, it is bicontinuous. In particular, $\mathcal{U}^{-1}$ is open. We claim that $r\vert_{\mathcal{U}^{-1}}:\mathcal{U}^{-1} \to r(\mathcal{U}^{-1})$ is an homeomorphism. Since this map is continuous and the inverse map is bicontinuous, we just need to prove the injectivity. In particular, the fact that the inverse map is a bijection gives that, for any $g_1, g_2 \in \mathcal{U}^{-1}$ s.t. $r(g_1) = r(g_2)$, there are $h_1, h_2 \in \mathcal{U}$ satisfying
\begin{equation*}
   g_1 = h_1^{-1} \quad \text{and} \quad g_2 = h_2^{-1}.
\end{equation*}
We have
\begin{align*}
s(h_1) = r(g_1) = r(g_2) = s(h_2),
\end{align*}
and then $h_1 = h_2$ because $\mathcal{U}$ is an open bisection, and therefore $g_1 = g_2$, proving the claim. The proof of this claim can be made into that for the statement that $s\vert_{\mathcal{U}^{-1}}:\mathcal{U}^{-1} \to s(\mathcal{U}^{-1})$ by mutatis mutandis, and therefore $\mathcal{U}^{-1}$ is an open bisection.

Before we prove $(ii)$, we show that we can take $s(\mathcal{U}) = r(\mathcal{V})$ w.l.o.g. In fact, $\mathcal{U}$ and $\mathcal{V}$ are open bisections and hence the set $W= s(\mathcal{U}) \cap r(\mathcal{V})$ is open. Then, the sets $\mathcal{U}_0 \subseteq \mathcal{U}$ and $\mathcal{V}_0 \subseteq \mathcal{V}$ defined as
\begin{equation*}
    \mathcal{U}_0 := s\vert_{\mathcal{U}}^{-1}(W) \quad \text{and} \quad \mathcal{V}_0 = r\vert_{\mathcal{V}}^{-1}(W)
\end{equation*}
are also open bisections, since they are pre-images of an open set under continuous functions and they are contained in open bisections, and they satisfy $s(\mathcal{U}_0) = W = r(\mathcal{V}_0)$. By construction we have $\mathcal{U} \times \mathcal{V} \cap G^{(2)} \supseteq \mathcal{U}_0 \times \mathcal{V}_0 \cap G^{(2)}$. We assert this inclusion is actually an equality. Indeed, let $(g,h) \in \mathcal{U} \times \mathcal{V} \cap G^{(2)}$ and denote $x = s(g)$. By the equivalence \ref{prop:arrow_approach_groupoids} $(a)$, one gets $x= s(g) = r(h)$, and then $x \in W$. Also, $g = s\vert_{\mathcal{U}}^{-1}(x)$ and $h = r\vert_{\mathcal{V}}^{-1}(x)$, then $g \in \mathcal{U}_0$ and $h \in \mathcal{V}_0$, proving the assertion. Consequently, we have the following:
\begin{align*}
\mathcal{UV} &= \lbrace gh : g \in \mathcal{U}, h \in \mathcal{V}, (g,h) \in G^{(2)} \rbrace = \lbrace gh : (g,h) \in \mathcal{U} \times \mathcal{V} \cap G^{(2)} \rbrace\\
&= \lbrace gh : (g,h) \in \mathcal{U}_0 \times \mathcal{V}_0 \cap G^{(2)}\rbrace = \mathcal{U}_0 \mathcal{V}_0.
\end{align*}
We conclude that $\mathcal{U}$ and $\mathcal{V}$ can be chosen satisfying $s(\mathcal{U}) = r(\mathcal{V})$ without loss of generality. 

The proof of $(ii)$ follows next. Suppose $s(\mathcal{U}) = r(\mathcal{V})$. Define the homeomorphism $\Upsilon: \mathcal{U} \to \mathcal{V}$, $\Upsilon(g):= r\vert_{\mathcal{V}}^{-1} \circ s\vert_{\mathcal{U}}(g)$, and the map $\psi:\mathcal{U} \to \mathcal{U} \times \mathcal{V}$, $\psi(g):= (g, \Upsilon(g))$. Observe that for every $g \in \mathcal{U}$ it holds that
\begin{align*}
r(\Upsilon(g)) = r(r\vert_{\mathcal{V}}^{-1} \circ s(g)) = s(g),
\end{align*}
and then $(g,\Upsilon(g)) \in \mathcal{U}\times \mathcal{V}\cap G^{(2)}$, that is, $\psi(\mathcal{U}) \subseteq \mathcal{U}\times \mathcal{V}\cap G^{(2)}$. On the other hand, for any $(g,h) \in \mathcal{U} \times \mathcal{V} \cap G^{(2)}$, that is, $g \in \mathcal{U}$, $h\in \mathcal{V}$ and $s(g) = r(h)$ and hence
\begin{align*}
h = r\vert_{\mathcal{V}}^{-1} \circ s\vert_{\mathcal{U}}(g) = \Upsilon(g),
\end{align*}
and therefore $(g,h) = (g, \Upsilon(g)) = \psi(g)$, i.e., $\psi(\mathcal{U}) = \mathcal{U} \times \mathcal{V} \cap G^{(2)}$ and then $\psi$ is surjective. Also, it is injective because $r\vert_{\mathcal{V}}^{-1}$ and $s\vert_{\mathcal{U}}$ are injective as well, therefore $\psi$ is a bijection. Furthermore, this map is continuous. We prove that projection on the first coordinate $\pi:\mathcal{U} \times \mathcal{V} \cap G^{(2)} \to \mathcal{U}$, $\pi(g,h):=g$, which is a continuous map, is the inverse of $\psi$. Indeed, for every $g \in \mathcal{U}$,
\begin{align*}
\pi \circ \psi (g) &= \pi(g,\Upsilon(g)) = g.
\end{align*}
At same time, for any $(g,h) \in \mathcal{U} \times \mathcal{V} \cap G^{(2)}$, the bijectivity of $\psi$ means that $(g,h) = (g, \Upsilon(g)) = \psi(g)$ and hence
\begin{align*}
(\psi \circ \pi)(g,h) = \psi(g) = (g, \Upsilon(g)) = (g,h).
\end{align*}
Therefore $\pi$ is the inverse of $\psi$, and consequently $\psi$ is a homeomorphism. Furthermore, this implies that $\mathcal{U} \times \mathcal{V} \cap G^{(2)}$ an open subset of $\mathcal{U} \times \mathcal{V}$, because $\mathcal{U}$ is open.

Now let $p:\mathcal{U} \times \mathcal{V} \cap G^{(2)} \rightarrow \mathcal{UV}$ be the product operation on the composable pairs of $\mathcal{U} \times \mathcal{V}$. It holds that
\begin{align*}
r\vert_{\mathcal{UV}}\circ p(g,h) = r(gh) = r(g) = r\vert_{\mathcal{U}} \circ \pi(g,h),
\end{align*}
for all $(g,h) \in \mathcal{U} \times \mathcal{V} \cap G^{(2)}$. In other words, it is true that
\begin{align}
\label{eq:rUVp}
r\vert_{\mathcal{UV}} \circ p = r\vert_{\mathcal{U}} \circ \pi.
\end{align}
The composition of homeomorphisms is also a homeomorphism, so it is $r\vert_{\mathcal{U}} \circ \pi$ and by, $\eqref{eq:rUVp}$, $r\vert_{\mathcal{UV}}\circ p$ is a homeomorphism as well. Consequently, $p$ is surjective and $r\vert_{\mathcal{UV}}$ is injective.

We state that $p$ is also injective. Indeed, given $(g_1,h_1),(g_2,h_2) \in \mathcal{U} \times \mathcal{V} \cap G^{(2)}$ such that $p(g_1, h_1) = p(g_2, h_2)$, we have the sequence of implications
\begin{equation*}
    p(g_1, h_1) = p(g_2, h_2) \stackrel{(\bullet)}{\implies}r\vert_{\mathcal{U}} \circ \pi(g_1, h_1) = r\vert_{\mathcal{U}} \circ \pi(g_2, h_2) \implies r\vert_{\mathcal{U}}(g_1) = r\vert_{\mathcal{U}}(g_2),
\end{equation*}
where in $(\bullet)$ we used \eqref{eq:rUVp}. Hence $g_1 = g_2$ because $\mathcal{U}$ is an open bisection, and then
\begin{equation*}
    r(h_1) = s(g_1) = s(g_2) = r(h_2)
\end{equation*}
by the hypothesis $(g_1,h_1),(g_2,h_2) \in \mathcal{U} \times \mathcal{V} \cap G^{(2)}$. However, $\mathcal{V}$ is an open bisection and consequently $h_1 = h_2$. Therefore $(g_1,h_1)=(g_2,h_2)$, i.e., $p$ is injective and then it is a continuous bijection, as well it is $r\vert_{\mathcal{UV}}$\footnote{$r\vert_{\mathcal{UV}}$ is surjective by construction.}, and the composition $r\vert_{\mathcal{UV}} \circ p$ is a homeomorphism. It follows that both functions are homeomorphisms. The proof that $s\vert_{\mathcal{UV}}$ follows analogously.

Summarizing, we showed that $\mathcal{UV}$ is an open set s.t. $r\vert_{\mathcal{UV}}$ and $s\vert_{\mathcal{UV}}$ are homeomorphisms and therefore it is an open bisection.
\end{proof}

\section{Groupoid C$^*$-algebras: the full groupoid C$^*$-algebra}

In the previous section, we presented and studied groupoids and some important topological features we can introduce on their structure. We also mentioned that the groupoid C$^*$-algebras are constructed from a $*$-algebra of compactly supported continuous complex functions on the respective groupoids. Continuity is a topological property and that justifies the necessity of introducing a topology on a groupoid. However we may beforehand mention that $*$-algebra aforementioned is in general non-commutative, and therefore the product between these functions will not be the usual one, but it is a convolution. For instance, we will see here that the etalicity allows us, under some conditions, to decompose those continuous functions on sums of functions of same properties but with support contained in open bisections. 

This section is based on the works \cite{Frausino2018, Lima2019, Sims2017}.

\begin{definition}\label{def:C_c_G} Let $G$ be a LCH second countable \'etale groupoid. Given a function $f:G \to \mathbb{C}$, the support of $f$ is the set
 \begin{equation*}
    \supp f := \overline{\{g \in G: f(g) \neq 0\}}.
 \end{equation*}
We say that $f$ is compactly supported or it has a compact support when $\supp f$ is compact. In addition, we define the set 
 \begin{equation*}
    C_c(G) = \{ f: G \rightarrow \mathbb{C} : \text{ $f$ is continuous and compactly supported} \}
 \end{equation*}
of all compactly supported continuous complex functions on $G$.
\end{definition}

We endow $C_c(G)$ with the usual addition and product by scalar operations, turning it into a vector space. What we do next is to construct the convolution product and the involution, in order to endow a $*$-algebra structure to $C_c(G)$.

\begin{lemma}\label{lemma:function_sum_supported_open_bisections} Let $G$ be a LCH second countable \'etale groupoid. For every $f \in C_c(G)$ there is a finite family $\{\mathcal{U}_i\}_{i=1}^n$ of open bisections and a finite family $\{f_i\}_{i=1}^n$ of functions in $C_c(G)$, where $\supp f_i \subseteq \mathcal{U}_i$ satisfying
\begin{equation}\label{eq:function_sum_supported_open_bisections}
    f = \sum_{i = 1}^n f_i.
\end{equation}
Furthermore, if $f$ is non-negative, then we can choose each $f_i$ as non-negative.
\end{lemma}

\begin{proof} Given $f \in C_c(G)$ with $K = \supp f$. By Proposition \ref{prop:openbisectionbase} the open bisections form a basis and then there exists an open cover of $K$ by open bisections. Since $K$ is compact, that cover a subcover $\{\mathcal{U}_i\}_{i=1}^n$ of open bisections. Now, define $\mathcal{U}_{n+1} := G\setminus K$, which is open because $K$ is closed due to the Hausdorff property, and hence $\{\mathcal{U}_i\}_{i=1}^{n+1}$ is an open cover of $G$. Let $\{P_i\}_{i=1}^{n+1}$ be a partition of the unit subordinate to $\{\mathcal{U}_i\}_{i=1}^{n+1}$. We get
\begin{equation*}
    f = \sum_{i=1}^{n+1} fP_i = \sum_{i=1}^n fP_i,
\end{equation*}
where the second equality holds because $fP_{n+1} = 0$, since $\supp f \subseteq \mathcal{U}_{n+1}^c$ and $\supp P_{n+1} \subseteq \mathcal{U}_{n+1}$. For each $i = 1,...,n$, the function $f_i := fP_i$ is continuous and $\supp f_i \subseteq \mathcal{U}_i$. Therefore the family of functions $\{f_i\}_{i=1}^n$ is supported on open bisections and it satisfies \eqref{eq:function_sum_supported_open_bisections}. Moreover, since every $P_i$ assumes values in $[0,1]$, if $f$ is non-negative then each $f_i$ is non-negative as well. 
\end{proof}

\begin{remark} It is straightforward to notice that in Lemma \ref{lemma:function_sum_supported_open_bisections} we have $f_i \in C_c(\mathcal{U}_i)$ for each $i$. For any open bisection $\mathcal{U}$, $C_c(\mathcal{U})$ is a vector subspace of $C_c(G)$.
\end{remark}

\begin{lemma} \label{lemma:inv_conv_functions_open_bisections} Let $G$ be a LCH second countable \'etale groupoid.
\begin{itemize}
    \item[$(i)$] Let $\mathcal{U}$ be an open bisection and $f \in C_c(\mathcal{U})$. Define the complex function $f^*$ on $G$ by
    \begin{equation*}
        f^*(g):= \overline{f(g^{-1})}.
    \end{equation*}
    Then $f^* \in C_c(\mathcal{U}^{-1})$ and hence $f^* \in C_c(G)$.
    \item[$(ii)$] For any two open bisections $\mathcal{U}_1$ and $\mathcal{U}_2$, let $f_1 \in C_c(\mathcal{U}_1)$ and $f_2 \in C_c(\mathcal{U}_2)$. Define the function $f_1 \cdot f_2$ on $G$ by
    \begin{equation*}
        (f_1 \cdot f_2) (g):= \sum_{g_1g_2 = g} f_1(g_1)f_2(g_2)
    \end{equation*}
    Then $f_1 \cdot f_2 \in C_c(\mathcal{U}_1\mathcal{U}_2)$ and hence $f_1 \cdot f_2 \in C_c(G)$.
\end{itemize}
\end{lemma}

\begin{proof} In order to prove $(i)$, we recall that the inverse map is a continuous because $G$ is a topological groupoid, and it is clear that the complex conjugation is also a continuous. Since $f \in C_c(\mathcal{U})$, the composition of these maps $f^*$ is continuous. We claim that $\supp f^* = (\supp f)^{-1}$. In fact,
\begin{align*}
\supp f^* &= \overline{\left\{ g \in G : f^*(g) \neq 0 \right\}} = \overline{\left\{ g \in G : \overline{f(g^{-1})} \neq 0 \right\}} = \overline{\left\{ g \in G : f(g^{-1}) \neq 0 \right\}}\\
&= \overline{\left(\left\{ g \in G : f(g) \neq 0 \right\}\right)^{-1}} \stackrel{(\bullet)}{=} \left(\overline{\left\{ g \in G : f(g) \neq 0 \right\}}\right)^{-1} = (\supp f)^{-1},
\end{align*}
where in $(\bullet)$ is justified by the fact that the inverse map is continuous. Also by the continuity of the inverse map, $(\supp f)^{-1}$ is compact. Since $\supp f \subseteq \mathcal{U}$, we have $\supp f^* = (\supp f)^{-1} \subseteq \mathcal{U}^{-1}$ and hence $f^* \in C_c(\mathcal{U}^{-1})$. By Lemma \ref{lemma:inverse_product_maps_on_open_bisections_is_open_bisection} $(i)$, $\mathcal{U}^{-1}$ is an open bisection and therefore $f^* \in C_c(G)$.

Now we prove $(ii)$. Given $g \notin \mathcal{U}_1 \mathcal{U}_2$, there is not a pair $(g_1,g_2) \in \mathcal{U}_1 \times \mathcal{U}_2 \cap G^{(0)}$ s.t. $g = g_1g_2$ and hence
\begin{equation*}
    (f_1 \cdot f_2 )(g) = \sum_{g_1g_2 = g} f_1(g_1)f_2(g_2) = 0,
\end{equation*}
implying that $\supp (f_1 \cdot f_2) \subseteq \mathcal{U}_1 \mathcal{U}_2$. By Lemma \ref{lemma:inverse_product_maps_on_open_bisections_is_open_bisection} $(ii)$, the set $\mathcal{U}_1 \mathcal{U}_2$ is an open bisection and then the maps
\begin{equation*}
    u_1: \mathcal{U}_1 \mathcal{U}_2 \to \mathcal{U}_1, \quad \text{and} \quad u_2: \mathcal{U}_1 \mathcal{U}_2 \to \mathcal{U}_2,
\end{equation*}
given by
\begin{equation*}
    u_1:= r\vert_{\mathcal{U}_1}^{-1} \circ r, \quad \text{and} \quad u_2:= s\vert_{\mathcal{U}_2}^{-1} \circ s,
\end{equation*}
are homeomorphisms. Now we state that for every $g \in \mathcal{U}_1 \mathcal{U}_2$ there exist unique elements $g_1 \in \mathcal{U}_1$ and $g_2 \in \mathcal{U_2}$ where $g = g_1g_2$, that is, $u_1(g) = g_1$ and $u_2(g) = g_2$. Indeed, for $h_1 \in \mathcal{U}_1$ and $h_2 \in \mathcal{U}_2$ with $(h_1,h_2)$ composable s.t. $g = h_1h_2$, we have
\begin{equation*}
    r(g) = r(h_1 h_2) = r(h_1) = r(g_1) \quad \text{and} \quad s(g) = s(h_1 h_2) = s(h_2) = s(g_2).
\end{equation*}
However, $\mathcal{U}_1$ and $\mathcal{U}_2$ are open bisections and hence
\begin{equation*}
    h_1 = r\vert_{\mathcal{U}_1}^{-1} \circ r(g) = g_1 \quad \text{and} \quad h_2 = s\vert_{\mathcal{U}_2}^{-1} \circ s(g) = g_2,
\end{equation*}
proving the statement. We conclude that for each $g \in \mathcal{U}_1 \mathcal{U}_2$ it holds that
\begin{equation*}
    (f_1\cdot f_2)(g) = f_1(u_1(g)) f_2(u_2(g)).
\end{equation*}
The continuity of $u_1$, $u_2$, $f_1$ and $f_2$, it follows that $f_1 \cdot f_2$ is continuous. Moreover,
\begin{align*}
    \supp(f_1 \cdot f_2) &= \overline{\{g \in \mathcal{U}_1 \mathcal{U}_2:(f_1 \cdot f_2)(g) \neq 0 \}} = \overline{\{g \in \mathcal{U}_1 \mathcal{U}_2:f_1(u_1(g)) f_2(u_2(g)) \neq 0 \}}\\
    &= \overline{\{g \in \mathcal{U}_1 \mathcal{U}_2: f_1(u_1(g)) \neq 0 \}\cap \{g \in \mathcal{U}_1 \mathcal{U}_2: f_2(u_2(g)) \neq 0 \}}\\
    &= \overline{\{g \in \mathcal{U}_1: f_1(g) \neq 0 \}\cap \{g \in \mathcal{U}_2: f_2(g) \neq 0 \}} \subseteq (\supp f_1) \cap (\supp f_2), 
\end{align*}
hence $\supp(f_1 \cdot f_2)$ is a closed subset of a compact set, and then it is compact. Therefore we have that $f_1 \cdot f_2 \in C_c(\mathcal{U}_1 \mathcal{U}_2)$, since $\supp(f_1 \cdot f_2) \subseteq \mathcal{U}_1 \mathcal{U}_2$. We conclude that $f_1 \cdot f_2 \in C_c(G)$ because $\mathcal{U}_1 \mathcal{U}_2$ is an open bisection.
\end{proof}
 
The previous lemma makes well-defined the functions $f^*$ and $f_1 \cdot f_2$ for $f,f_1$ and $f_2$ complex compactly supported functions on $G$ with support contained in some open bisection. We extend these results for the whole space $C_c(G)$.

\begin{lemma} Let $G$ be a LCH second countable \'etale groupoid. For every $f_1, f_2 \in C_c(G)$ we have
\begin{equation*}
    (f_1 \cdot f_2)(g) := \sum_{g_1g_2 = g} f_1(g_1)f_2(g_2) = \sum_{h \in G_{s(g)}} f_1(gh^{-1})f_2(h) = \sum_{h \in G^{r(g)}} f_1(h)f_2(h^{-1}g).
\end{equation*}
Moreover, the sums above are finite.
\end{lemma}

\begin{proof} Given $g_1, g_2 \in G$ such that $g_1 g_2 = g$. By $(G3)$ in Definition \ref{def:groupoid} we have that $(g_1^{-1},(g_1^{-1})^{-1})$ and $(g_2,g_2^{-1})$ belong to $G^{(2)}$, and by $(G1)$ in the same definition it holds that $(g_1^{-1},(g_1^{-1})^{-1})=(g_1^{-1},g_1)$. Still from same definition, $(G2)$ gives $(g_1^{-1},g_1g_2) = (g_1^{-1},g), (g_1g_2,g_2^{-1}) = (g,g_2^{-1}) \in G^{(2)}$ and we obtain the following equivalences:
\begin{equation*}
 g_1 = gg_2^{-1}\iff g_1g_2 = g \iff g_2 = g_1^{-1}g.
\end{equation*}
Since $(g_1^{-1},g), (g,g_2^{-1}) \in G^{(2)}$ we also get by Lemma \ref{lemma:properties_r_s} $(iii)$ and Proposition \ref{prop:arrow_approach_groupoids} $(a)$ that
\begin{equation*}
    r(g) = s(g_1^{-1}) = r(g_1) \quad \text{and} \quad s(g) = r(g_2^{-1}) = s(g_2),
\end{equation*}
that is $g_1 \in G^{r(g)}$ and $g_2 \in G_{s(g)}$. Then, every $h \in G^{r(g)}$ we may take $g_1 = h$ and then $g_2 = h^{-1}g$, implying to
\begin{equation*}
    (f_1 \cdot f_2)(g) = \sum_{h \in G^{r(g)}} f_1(h)f_2(h^{-1}g).
\end{equation*}
Analogously, for each $h \in G_{s(g)}$ we may take $g_2 = h$ and then $g_1 = gh^{-1}$, from which we obtain
\begin{equation*}
    (f_1 \cdot f_2)(g) = \sum_{h \in G_{s(g)}} f_1(gh^{-1})f_2(h).
\end{equation*}
The sums above are finite for every $g$ because both $G^{r(g)}$ and $G_{s(g)}$ are closed\footnote{Due to the continuity of $r$ and $s$.} and the functions $f_1$ and $f_2$ are compactly supported. In fact,
\begin{equation*}
    \{h \in G_{s(g)}:f_1(h) \neq 0\} = G_{s(g)}\cap \supp f_1 \quad \text{and} \quad \{h \in G^{r(g)}:f_2(h) \neq 0\} = G_{r(g)}\cap \supp f_2
\end{equation*}
are intersecions closed sets with compact sets, and therefore they are compact sets. By Proposition \ref{prop:GxSecondCountable} we have that both subspace topologies of $G_{s(g)}$ and $G^{s(g)}$ are discrete spaces, and then both $G_{s(g)}\cap \supp f_1$ and $G_{r(g)}\cap \supp f_2$ are also discrete. Since every discrete space is compact if and only if it is finite, they are finite sets.  
\end{proof}

\begin{lemma}\label{lemma:linearity_prod_inv} Let $G$ be a LCH second countable \'etale groupoid. Given $f,f_1,f_2 \in C_c(G)$ and $\lambda \in \mathbb{C}$, the following is true:
\begin{itemize}
 \item[$(a)$] $(f_1 + \lambda f_2) \cdot f = (f_1 \cdot f) + \lambda (f_2 \cdot f)$ and $f \cdot (f_1 + \lambda f_2) = (f \cdot f_1) + \lambda (f \cdot f_2)$;
 \item[$(b)$] $(f_1 + \lambda f_2)^* = f_1^* + \overline{\lambda} f_2^*$.
\end{itemize}
\end{lemma}

\begin{proof}
In $(a)$ we prove the first equality only, the second one follows by similar steps. Given $g \in G$, we have
\begin{align*}
((f_1 + \lambda f_2)\cdot f)(g) &= \sum_{g_1g_2 = g} (f_1 + \lambda f_2)(g_1)f(g_2) = \sum_{g_1g_2 = g} f_1(g_1)f(g_2) + \lambda\sum_{g_1g_2 = g} f_2(g_1)f(g_2)\\
&= (f_1 \cdot f)(g) + \lambda (f_2 \cdot f)(g).
\end{align*}

Now we prove $(b)$. For every $g \in G$, it holds that
\begin{align*}
(f_1 + \lambda f_2)^*(g) &= \overline{(f_1 + \lambda f_2)(g^{-1})} = \overline{f_1(g^{-1})} + \overline{\lambda f_2(g^{-1})} = f_1^*(g) + \overline{\lambda} f_2^*(g).
\end{align*}
\end{proof}

\begin{theorem}\label{thm:well-defined_product_involution_G} Let $G$ be a locally compact Hausdorff second countable \'etale groupoid. Given $f,f_1,f_2 \in C_c(G)$ then $f^*,f_1 \cdot f_2 \in C_c(G)$. 
\end{theorem}

\begin{proof} By Lemma \ref{lemma:function_sum_supported_open_bisections}, we may write
\begin{equation*}
    f = \sum_{i=1}^n \phi_i, \quad f_1 = \sum_{j=1}^{n_1} \phi^1_j \quad \text{and} \quad f_2 = \sum_{k=1}^{n_2} \phi^2_k,
\end{equation*}
where $n,n_1,n_2 \in \mathbb{N}$, and the functions $\phi_i$, $\phi^1_j$ and $\phi^2_k$ belong to $C_c(G)$ and are supported in open bisections for every $i = 1,...,n$, $j=1,...,n_1$ and $k = 1,...,n_2$. Lemma \ref{lemma:linearity_prod_inv} gives that
\begin{equation*}
    f^* = \sum_{k=1}^n \phi_k^* \quad \text{and} \quad (f_1 \cdot f_2)(g) = \sum_{j=1}^{n_1} \sum_{k=1}^{n_2} \phi^1_j \cdot \phi^2_k,
\end{equation*}
and since $C_c(G)$ is a vector space, we conclude that it is sufficient to prove that the statement of the theorem holds for functions supported in open bisections, which is precisely the result in Lemma \ref{lemma:inv_conv_functions_open_bisections}. 
\end{proof}

The well-definition of $f^*$ and $f_1 \cdot f_2$ for arbitrary functions in $C_c(G)$ allows us to define the maps
\begin{equation*}
 f \mapsto f^* \quad \text{and} \quad (f_1,f_2) \mapsto f_1 \cdots f_2
\end{equation*}
as operations in $C_c(G)$ as follows.

\begin{definition} Let $G$ be a locally compact Hausdorff second countable \'etale groupoid. We define the involution $*:C_c(G) \to C_c(G)$ and convolution product, or simply convolution, $\cdot :C_c(G) \times C_c(G) \to C_c(G)$, defined by
\begin{equation*}
    f^*(g):= \overline{f(g^{-1})} \quad \text{and} \quad (f_1 \cdot f_2 )(g) = \sum_{g_1g_2 = g} f_1(g_1)f_2(g_2),
\end{equation*}
for every $f, f_1, f_2 \in C_g(G)$ and $g \in G$.
\end{definition}

\begin{theorem} Let $G$ be a locally compact Hausdorff second countable \'etale groupoid. The vector space $C_c(G)$ endowed with the operations $\cdot$ and $*$ is a $*$-algebra.
\label{thm:CcGast-algebra}
\end{theorem}
\begin{proof}
As we observed previously in this section, $C_c(G)$ is a vector space. Moreover, we also proved in Theorem \ref{thm:well-defined_product_involution_G} that the involution and convolution product are well-defined. Furthermore, in Lemma \ref{lemma:linearity_prod_inv} we showed that the involution is conjugate-linear and that the convolution is linear. We prove now the remaining axioms of $*$-algebra.

\begin{itemize}
\item \textbf{Associativity of the convolution:} given $f_1, f_2, f_3 \in C_c(G)$ and $g \in G$, we have 
\begin{align*}
(f_1\cdot (f_2 \cdot f_3))(g) &= \sum_{g_1 h = g} f_1(g_1) (f_2 \cdot f_3)(h)= \sum_{g_1 h = g} \sum_{g_2 g_3 = h} f_1(g_1) f_2(g_2) f_3(g_3)\\
&= \sum_{g_1 g_2 g_3 = g} f_1(g_1) f_2(g_2) f_3(g_3) = \sum_{hg_3 = g}\sum_{g_1 g_2 = h} f_1(g_1) f_2(g_2) f_3(g_3)\\
&= \sum_{hg_3 = g} \left( \sum_{g_1 g_2 = h} f_1(g_1) f_2(g_2) \right)f_3(g_3) = \sum_{hg_3 = g} (f_1 \cdot f_2)(h) f_3(g_3)\\
&= ((f_1 \cdot f_2) \cdot f_3)(g).
\end{align*}

\item \textbf{Involution is its own inverse:} let $f \in C_c(G)$ and $g \in G$. We have
\begin{equation*}
    (f^*)^*(g) = \overline{f^*(g^{-1})} = \overline{\overline{f((g^{-1})^{-1})}} = f(g).
\end{equation*}

\item \textbf{Distribution of the involution over the convolution:} let $f_1, f_2 \in C_c(G)$ and $g \in G$. It holds that
\begin{align*}
(f_1 \cdot f_2)^*(g) &= \overline{(f_1 \cdot f_2)(g^{-1})} = \sum_{g_1g_2 = g^{-1}} \overline{f_1(g_1)f_2(g_2)} \\
&= \sum_{g_1g_2 = g^{-1}} f_2^*(g_2^{-1})f_1^*(g_1^{-1}) = \sum_{g_2^{-1}g_1^{-1} = g} f_2*(g_2^{-1})f_1*(g_1^{-1}).
\end{align*}
By setting $h_1 = g_2^{-1}$ and $h_2 = g_1^{-1}$ one gets
\begin{align*}
\sum_{g_2^{-1}g_1^{-1} = g} f_2^*(g_2^{-1})f_1^*(g_1^{-1})&= \sum_{h_1h_2 = g} f_2^*(h_1)f_1^*(h_2) = (f_2^* \cdot f_1^*)(g),
\end{align*}
that is,
\begin{equation*}
    (f_1 \cdot f_2)^* = f_2^* \cdot f_1^*.
\end{equation*}
\end{itemize}
\end{proof}

The construction of the full groupoid C$^*$-algebra of $G$ emerges from the construction of the $*$-algebra $C_c(G)$, and it is fulfilled by endowing $C_c(G)$ with a norm that satisfies the C$^*$ norm property, followed by its norm completion. These will be our next steps from now. In order to construct the aformentioned norm, we use the commutative $*$-subalgebra $C_c(G^{(0)})$ of $C_c(G)$.

\begin{lemma} \label{lemma:CcG0astalgebra} Given a LCH second countable \'etale groupoid $G$, the set $C_c(G^{(0)})$ endowed with the ihnerited operations of $C_c(G)$ is a $*$-subalgebra of $C_c(G)$. Moreover, $C_c(G^{(0)})$ is commutative, the restriction of the convolution product this $*$-subalgebra coincides with the pointwise multiplication, and the restriction of involution is given by $f^*(x) = \overline{f(x)}$, $f \in C_c(G^{(0)})$.
\end{lemma}

\begin{proof}
$G^{(0)}$ is open due to Proposition \ref{prop:G0clopen} and hence $C_c(G^{(0)})$ is a subspace of $C_c(G)$. 

For the involution, let $f \in C_c(G^{(0)})$. Theorem \ref{thm:CcGast-algebra} gives that $f^* \in C_c(G)$. Let $g \in G$ satisfying $f^*(g) \neq 0$, that is, $f(g^{-1}) \neq 0$ and then $g^{-1} \in G^{(0)}$. By Lemma \ref{lemma:properties_r_s} $(iii)$ and Proposition \ref{prop:arrow_approach_groupoids} $(d)$, we have
\begin{equation*}
    r(g) = s(g^{-1}) = g^{-1} = r(g^{-1}) = s(g),
\end{equation*}
then $(g,g) \in G^{(2)}$ due to Proposition \ref{prop:arrow_approach_groupoids} $(i)$. Then we have
\begin{equation*}
    g = gg^{-1} g = g s(g) = gg^{-1} = r(g) = g^{-1},
\end{equation*}
that is $g = g^{-1}$ and therefore $g \in G^{(0)}$. Hence,
\begin{equation*}
    \{g \in G: f^*(g) = 0\} \subseteq G^{(0)},
\end{equation*}
and we obtain $\supp f^* \subseteq G^{(0)}$ because $G^{(0)}$ is closed by Proposition \ref{prop:G0clopen}. Therefore, $f^* \in C_c(G^{(0)})$, that is, $*$ is algebrically closed in $C_c(G^{(0)}$. In addition, note that we proved $x \in G^{(0)}$ implies $x = x^{-1}$ and therefore $f^*(x) = \overline{f(x)}$ for every $f \in C_c(G^{(0)}$.

For the convolution, let $f_1, f_2 \in C_c(G^{(0)})$. Theorem \ref{thm:CcGast-algebra} gives that $f_1 \cdot f_2 \in C_c(G)$. By definition of convolution, for every $g \in G$ satisfying $(f_1 \cdot f_2)(g) \neq 0$, there exists $g_1,g_2 \in G$ such that $g = g_1 g_2$, $f_1(g_1) \neq 0$ and $f_2(g_2) \neq 0$. Since $f_1, f_2 \in C_c(G^{(0)})$, we necessarily have $g_1,g_2 \in G^{(0)}$. Note that by construction $(g_1,g_2) \in G^{(2)}$ and then $s(g_1) = r(g_2)$ due to Proposition \ref{prop:arrow_approach_groupoids} $(a)$. Then,
\begin{equation*}
    g_1 = s(g_1) = r(g_2) = g_2.
\end{equation*}
Hence,
\begin{equation*}
    g = g_1 g_2 = g_1 g_1 = r(g_1) g_1 = g_1,
\end{equation*}
and therefore $g = g_1 = g_2$, that is,
\begin{equation*}
    \{g \in G: (f_1\cdot f_2)(g) = 0\} \subseteq G^{(0)},
\end{equation*}
and recalling that $G^{(0)}$ is closed we conclude that $\supp (f_1\cdot f_2) \subseteq G^{(0)}$, i.e., $f_1\cdot f_2 \in C_c(G^{(0)})$. Moreover,
\begin{equation*}
    (f_1 \cdot f_2) (g) = \sum_{g_1 g_2 = g} f_1(g_1) f_2(g_2) = f_1(g) f_2(g).
\end{equation*}
Consequently, $C_c(G^{(0)})$ is commutative.
\end{proof}

\begin{lemma}\label{lemma:C_cG0_union_C_star_algebras} Let $G$ be a LCH second countable \'etale groupoid. The set $C_c(G^{(0)})$ is a union of C$^*$-algebras. 
\end{lemma}

\begin{proof} Given $K \subseteq G^{(0)}$ compact, define $\mathscr{C}(K)$ as the set of all complex continuous functions with support contained in $K$. We endow $\mathscr{C}(K)$ with the usual vector space structure, the point-wise multiplication and the complex conjugation as an involution. Then, $\mathscr{C}(K)$ has the $*$-algebra structure, and we also endow this space with the supremum norm, that is, $\| f\|_\infty:=\sup_{x\in K}|f(x)|$, for every $f \in \mathscr{C}(K)$. Since every element has support on $K$, it is straightforward to conclude that $\mathscr{C}(K)$ is a C$^*$-algebra. By Lemma \ref{lemma:CcG0astalgebra} the unit space $C_c(G^{(0)})$ is a $*$-subalgebra of $C_c(G^{(0)})$. We claim that
\begin{equation*}
    C_c(G^{(0)})=\bigcup_{f\in C_c(G^{(0)})}\mathscr{C}(\supp f).
\end{equation*}
In fact, for every $f\in C_c(G^{(0)})$, we have $f\in \mathscr{C}(\supp f)$ and hence $C_c(G^{(0)})\subseteq \bigcup_{f\in C_c(G^{(0)})}\mathscr{C}(\supp f)$. Conversely, if $f'\in \bigcup_{f\in C_c(G^{(0)})}\mathscr{C}(\supp f)$ then $h\in \mathscr{C}(\supp f)$ for some $f\in C_c(G^{(0)})$, and then $\supp(f')\subset supp(f)$. By the compactness of $\supp(f)$, it follows that $\supp(f')$ is compact and therefore $C_c(G^{(0)})\supseteq \bigcup_{f\in C_c(G^{(0)})}\mathscr{C}(\supp f)$. 
\end{proof}

\begin{proposition} \label{prop:reprnormbound}
Let $G$ be a locally compact Hausdorff second countable \'etale groupoid. For every $f \in C_c(G)$, there exists a constant $K_f \geq 0$ such that $\| \pi(f) \| \leq K_f$ for every $*$-representation $\pi: C_c(G) \rightarrow \mathfrak{B} (\mathcal{H})$ of $C_c(G)$ on a Hilbert space $\mathcal{H}$. In particular, if $f$ is supported on an open bisection, then we can take $K_f = \| f \|_\infty$.
\end{proposition}

\begin{proof}
For any $*$-representation $\pi$ of $C_c(G)$, the restriction $\pi\vert_{C_c(G^{(0)})}$ is a $*$-representation of $C_c(G^{(0)})$. The lemma \ref{lemma:C_cG0_union_C_star_algebras} gives that every $f \in C_c(G^{(0)})$ belongs to some $\mathscr{C}(K)$, $K$ compact, and then $\pi\vert_{\mathscr{C}(K)}$ is a C$^*$-representation of $\mathscr{C}(K)$ and therefore\footnote{Every homomorphism between C$^*$-algebras is bounded by the norm.} $\|\pi(f)\| \leq \| f \|_\infty$. 

Lemma \ref{lemma:function_sum_supported_open_bisections} gives that any $f \in C_c(G)$ can be written as a sum of functions $f_1,...,f_n \in C_c(G)$ supported on the open bisections $\mathcal{U}_1,...,\mathcal{U}_n$, respectively. By Lemma \ref{lemma:inv_conv_functions_open_bisections} $(i)$, we have $f_k* \in C_c(\mathcal{U}_k^{-1})$, and then the item $(ii)$ of the same lemma implies $f_k* \cdot f_k \in C_c(\mathcal{U}^{-1}\mathcal{U})$ for every $k$. Nevertheless, $\mathcal{U}\mathcal{U}^{-1} = s(\mathcal{U})$ for every open bisection $\mathcal{U}$. Indeed,
\begin{align*}
\mathcal{U}^{-1}\mathcal{U} &= \{ gh : g\in \mathcal{U}^{-1}, h \in \mathcal{U}, s(g) = r(h) \} = \lbrace g_1^{-1}g_2 : g_1, g_2 \in \mathcal{U}, r(g_1) = r(g_2) \rbrace\\
&= \lbrace g^{-1}g : g \in \mathcal{U}\} = s(\mathcal{U}),
\end{align*}
where in the last equality we used the fact that $\mathcal{U}$ is an open bisection. Consequently, $f_k^*f_k \in C_c(s(\mathcal{U_k}))$ for every $k$ and then
\begin{align*}
\| \pi(f_i) \|^2 = \| \pi(f_i* \cdot f_i) \| \leq \| f_i* \cdot f_i \|_\infty = \Vert f_i \Vert^2.
\end{align*}
Now, by taking $K_f = \sum_{i=1}^n \| f_i \|$ and the triangle inequality, we have $\| \pi(f) \| \leq K_f$.
\end{proof}

\begin{proposition}
\label{prop:pifneqzero}
Let $G$ be a locally compact Hausdorff second countable \'etale groupoid and $f \in C_c(G)$ satisfying $f \neq 0$. There exists a $*$-representation $\pi$ of $C_c(G)$ such that $\pi(f) \neq 0$.
\end{proposition}

\begin{proof} If $f \neq 0$, there exists $h' \in G$ satisfying $f(h') \neq 0$. Note that $h' \in G_x$, where $x = s^{-1}(h')$. By Proposition \ref{prop:GxSecondCountable}, $G_x$ is countable and then the vector space
\begin{align*}
\ell^2(G_x) := \left\lbrace \lbrace z_g \rbrace_{g \in G_x} \in \mathbb{C}^{G_x} : \sum_{g \in G_x} \vert z_g \vert^2 < \infty \right\rbrace
\end{align*}
is a Hilbert space, where the inner product is defined by
\begin{align*}
\langle z, w \rangle := \sum_{g \in G_x} z_g \overline{w_g}.
\end{align*}
There exists a natural representation $\pi_x$ of $C_c(G)$ on $\mathfrak{B}(\ell^2(G_x))$ defined by
\begin{align*}
(\pi_x(f') z)_g  := \sum_{h_1 h_2 = g} f'(h_1) z_{h_2}, \quad f' \in C_c(G).
\end{align*}
Equivalently, we may write 
\begin{align}\label{eq:representation_pifneqzero}
(\pi_x(f') z)_g  = \sum_{h \in G_{s(g)}} f'(gh^{-1})z_h = \sum_{h \in G^{r(g)}} f'(h) z_{h^{-1}g},
\end{align}
where in the first identity we set $h_2 = h$, $h_1 = gh^{-1}$, and in the second one we put $h_1 = h$, $h_2 = h^{-1}g$. It is straightforward that $\pi_x$ is linear. We prove now that $\pi_x$ is well-defined. By linearity of $\pi_x$ and Lemma \ref{lemma:function_sum_supported_open_bisections} it is sufficient to show that the representation is well-defined for functions supported on an open bisection. So given $f' \in C_c(\mathcal{U})$, $\mathcal{U}$ open bisection, define
\begin{equation}\label{eq:set_L}
    \mathcal{L}:= \{g \in G_x: \exists h \in G^{r(g)}, f'(h)\neq 0\}.
\end{equation}
Given $g \in \mathcal{L}$, and $h$ as in \eqref{eq:set_L}, we have that $h \in G^{r(g)} \cap \mathcal{U}$, and since $\mathcal{U}$ is an open bisection, we have that $h$ is unique, so without ambiguity we may write $h = h^{r(g)}$, and by the last equality of \eqref{eq:representation_pifneqzero} we obtain
\begin{align*}
(\pi_x(f') z)_g = f'(h^{r(g)})z_{(h^{r(g)})^{-1}g}
\end{align*}
for every $g \in G_x$. Moreover, $g \in G_x \setminus \mathcal{L}$ implies $(\pi_x(f')z)_g = 0$. On the other hand, for $g_1, g_2 \in \mathcal{L}$ satisfying $(h^{r(g_1)})^{-1}g_1 = (h^{r(g_2)})^{-1}g_2$ , then $g_1 = g_2$. Indeed, the hypothesis of this statement implies $s(h^{r(g_1)}) = s(h^{r(g_2)})$. Since $h^{r(g_1)}$ and $h^{r(g_2)}$ belong to the open bisection $\mathcal{U}$, we obtain $h^{r(g_1)} = h^{r(g_2)}$, and then $g_1 = g_2$. Furthermore, $(h^{r(g)})^{-1}g \in G_x$ for every $g \in \mathcal{L}$ and hence the family $\{(h^{r(g)})^{-1}g\}_{g \in \mathcal{L}}$  does not contain equal elements for different $g$'s. We have
\begin{align*}
\| \pi_x(f') z \|^2 &= \sum_{g \in G_x} | (\pi_x(f') z)_g |^2 = \sum_{g \in \mathcal{L}} | (\pi_x(f') z)_g |^2 = \sum_{g \in \mathcal{L}} | f'(h^{r(g)}) z_{(h^{r(g)})^{-1}g} |^2\\
&\leq \Vert f' \Vert_\infty^2 \sum_{g \in \mathcal{L}} | z_{(h^{r(g)})^{-1}g} |^2 \leq \Vert f' \Vert_\infty^2 \sum_{h \in G_x} | z_h |^2
\leq \Vert f' \Vert_\infty^2 \Vert z \Vert^2 < \infty,
\end{align*}
and then $\pi_x(f') \in \mathfrak{B}(\ell^2(G_x))$, implying that the map $\pi_x$ is well defined. Now we prove that $\pi_x(f)\neq 0$. In fact, by taking $z \in \ell^2(G_x)$ given by $z_{h'} = \delta_{x,h'}$, where $\delta_{a,b}$ is the Kronecker delta, we have
\begin{equation*}
    (\pi_x(f)z)_{h'} = \sum_{h \in G^{r(h')}} f(h) z_{h^{-1}h'} = f(h') z_{h'^{-1}h'} = f(h')z_x = f(h') \neq 0.
\end{equation*}
It remains to prove that $\pi_x$ is a $*$-representation. Since, $\pi_x$ is linear, we just need to prove preservation of $*$ and $\cdot$ under $\pi_x$. Let $f' \in C_c(G)$ and $z \in \ell^2(G_x)$, by definitions of $\pi_x$ and inner product in $\ell^2(G_x)$ we have
\begin{align*}
\langle z, \pi_x(f'^*)z \rangle &= \sum_{g \in G_x} z_g \overline{\left( \pi_x(f'^*)z \right)_g} = \sum_{g \in G_x} z_g \sum_{h \in G_x} \overline{f'^*(gh^{-1})z_h} = \sum_{g \in G_x} z_g \sum_{h \in G_x} f'(hg^{-1})\overline{z_h}\\
&= \sum_{h \in G_x} \left(\sum_{g \in G_x} f'(hg^{-1}) z_g \right) \overline{z_h} = \sum_{h \in G_x} \left( \pi_x(f')z \right)_h \overline{z_h} = \langle \pi_x(f')z, z \rangle.
\end{align*}
Therefore, $\pi_x(f'^*) = (\pi_x(f'))^*$. Now take $f_1, f_2 \in C_c(G)$ and $z \in \ell^2(G_x)$. We have
\begin{align*}
\Big( \pi_x(f_1) (\pi_x(f_2) z )\Big)_g &= \sum_{g_1 h = g} f_1(g_1) (\pi_x(f_2)z)_h = \sum_{g_1 h = g} f_1(g_1) \sum_{g_2g_3 = h} f_2(g_2)z_{g_3}\\
&= \sum_{g_1 g_2 g_3 = g} f_1(g_1) f_2(g_2)z_{g_3} = \sum_{hg_3 = g} \left(\sum_{g_1 g_2 = h} f_1(g_1) f_2(g_2) \right)z_{g_3}\\
&= \sum_{hg_3 = g} (f_1 \cdot f_2)(h) z_{g_3} = \Big(\pi_x(f_1 \cdot f_2)z\Big)_g,
\end{align*}
and hence $\pi_x(f_1 \cdot f_2) = \pi_x(f_1) \cdot \pi_x(f_2)$. We conclude that $\pi_x$ is in fact a $*$-representation of $C_c(G)$ satisfying $\pi_x(f) \neq 0$ for $f \neq 0$.
\end{proof}

The next result introduces the full groupoid C$^*$-algebra, by introducing a norm on $C_c(G)$ that consists in the supremum over all $*$-representations of $C_c(G)$. We show that this supremum is well-defined and that it is an actual norm, and it makes $C_c(G)$ a dense subset.

\begin{theorem}
\label{thm:CcGdense}
Let $G$ be a locally compact Hausdorff second countable \'etale groupoid. There exists a C$^*$-algebra $C^*(G)$ containing $C_c(G)$ such that its norm satisfies the following:
\begin{align*}
\Vert f \Vert = \sup \lbrace \Vert \pi(f) \Vert \text{$:$ } \pi: C_c(G) \rightarrow B(H_\pi) \text{ is a $\ast$-representation of $C_c(G)$} \rbrace, \quad f \in C_c(G).
\end{align*}
Moreover, $C_c(G)$ is dense in $C^*(G)$.
\end{theorem}

\begin{proof}
By Proposition \ref{prop:reprnormbound}, we have that the set
\begin{align*}
\lbrace \| \pi(f) \|: \pi \text{ is a $*$-representation of $C_c(G)$} \rbrace
\end{align*}
has a upper bound for every $f \in C_c(G)$. Also, the trivial representation $\pi = 0$ is a $*$-representation of $C_c(G)$ and then the set above is non-empty. Hence, the map $\rho: C_c(G) \rightarrow \left[ 0, \infty \right)$ defined as
\begin{align*}
\rho(f) = \sup_{\pi \in \mathcal{R}}  \| \pi(f) \| ,
\end{align*}
where $\mathcal{R}$ is the set of all $*$-representations of $C_c(G)$, is well defined. We prove that $\rho$ is a sub-multiplicative norm. Indeed, let $\lambda \in \mathbb{C}$ and $f \in C_c(G)$. We have
\begin{align*}
\rho(\lambda f) = \sup_{\pi \in \mathcal{R}} \| \pi(\lambda f) \| = \vert \lambda \vert \sup_{\pi \in \mathcal{R}} \| \pi(f) \| = \vert \lambda \vert \rho(f).
\end{align*}
Now, let $f_1, f_2 \in C_c(G)$. We obtain
\begin{align*}
\rho(f_1 + f_2) = \sup_{\pi \in \mathcal{R}} \| \pi(f_1 + f_2) \| \leq \sup_{\pi \in \mathcal{R}} \| \pi(f_1) \| + \sup_{\pi \in \mathcal{R}} \| \pi(f_2) \| = \rho(f_1) + \rho(f_2).
\end{align*}
and
\begin{align*}
\rho(f_1 \cdot f_2) = \sup_{\pi \in \mathcal{R}} \| \pi(f_1 \cdot f_2) \| = \sup_{\pi \in \mathcal{R}} \| \pi(f_1) \pi(f_2) \| \leq \sup_{\pi \in \mathcal{R}} \| \pi(f_1) \| \sup_{\pi \in \mathcal{R}} \| \pi(f_2) \| = \rho(f_1)\rho(f_2).
\end{align*}
Let $f \in C_c(G)$ satisfying $f \neq 0$. By Proposition \ref{prop:pifneqzero}, there exists a $*$-representation of $C_c(G)$ $\pi'$ satisfying $\pi'(f) \neq 0$ and therefore $\rho(f) > 0$. We conclude so far that $\rho$ is a submultiplicative norm on $C_c(G)$. Also, for every $f \in C_c(G)$, we have
\begin{align*}
\rho(f^*) = \sup_{\pi \in \mathcal{R}} \| \pi(f^*) \| = \sup_{\pi \in \mathcal{R}} \| \pi(f)^* \| = \sup_{\pi \in \mathcal{R}} \| \pi(f) \|
= \rho(f).
\end{align*}
Now we prove that $\rho$ satisfies the C$^*$-identity for the norm, that is, $\rho(f^*f) = \rho(f)^2$. Let $f \in C_c(G)$, it follows that
\begin{align*}
\rho(f^*f) = \sup_{\pi \in \mathcal{R}} \| \pi(f^*f) \| = \sup_{\pi \in \mathcal{R}} \| \pi(f)^* \pi(f) \| = \sup_{\pi \in \mathcal{R}} \| \pi(f)\|^2 = \left(\sup_{\pi \in \mathcal{R}} \| \pi(f)\|\right)^2 = \rho(f)^2.
\end{align*}
The existence of C$^*$ algebra $C^*(G)$ is granted by defining it as the completion of $C_c(G)$ with respect to the norm $\rho$. It is straightforward that $C_c(G)$ is dense on $C^*(G)$
\end{proof}

\begin{definition}[Full groupoid C$^*$-algebra] Let $G$ be a locally compact Hausdorff second countable \'etale groupoid, the full groupoid C$^*$-algebra of $G$, denoted by $C^*(G)$, is the norm completion of $C_c(G)$ as in Theorem \ref{thm:CcGdense}.
\end{definition}

We mention here that there exists another groupoid C$^*$-algebra, called reduced groupoid C$^*$-algebra, that can be seen as a C$^*$-subalgebra of $C^*(G)$. Under a condition called amenability, the full and reduced groupoid algebras coincide. The present thesis do not uses the notion of reduced groupoid algebra, since the Generalized Renault-Deaconu groupoid, the unique groupoid used in our research, is amenable under our conditions. If the reader is interested on a more deep approach in this topic, see \cite{Frausino2018,Renault1980,Sims2017}.

The final results of this chapter are basically some topological properties of the full groupoid C$^*$-algebra.

\begin{lemma}
\label{lemma:C0G0inCast}
Given $G$ a LCH second countable \'etale groupoid, we have that $C_0(G^{(0)})$ is a C$^*$-subalgebra of $C^*(G)$. The norm in $C_0(G^{(0)})$ coincides with the uniform norm and $C_c(G^{(0)})$ is a dense subset of $C_0(G^{(0)})$.
\end{lemma}

\begin{proof} This lemma is essentially Lemma 3.3.16 of \cite{Lima2019} and we omit the proof here. 
\end{proof}

The last aim on this chapter is to prove that $C^*(G)$, under our usual assumptions for $G$, is separable. In order to prove such result we use a version of the Stone-Weierstrass Theorem, and the next definition is used to state this theorem.

\begin{definition}[Separating sets] Let $C,D$ be two sets, a family $\mathcal{F}$ of functions from $C$ to $D$ is said to be a separating set for $C$ if for every $x,y \in C$ s.t. $x \neq y$ there exists $f \in \mathcal{F}$ s.t. $f(x) \neq f(y)$.
\end{definition}

Now we present the statement of the Stone-Weierstrass Theorem for complex valued functions.

\begin{theorem}[Stone-Weierstrass Theorem] Consider a locally compact space $X$ and a subalgebra $\mathcal{A}$ of $C_0(X)$ satisfying the following:
\begin{itemize}
    \item[$(i)$] $\mathcal{A}$ is closed under the usual complex conjugation involution: if $f \in \mathcal{A}$, then $\overline{f} \in \mathcal{A}$, where $\overline{f}(x) := \overline{f(x)}$, for every $x \in X$;
    \item[$(ii)$] $\mathcal{A}$ is a separating set for $X$;
    \item[$(iii)$] for every $x \in X$, there exists $f \in \mathcal{A}$ s.t. $f(x) \neq 0$.
\end{itemize}
Then $\mathcal{A}$ is a dense subset of $C_0(X)$.
\end{theorem}

\begin{lemma}\label{lemma:LCH_has_precompact_basis} Let $X$ be a LCH topological space. For every basis $\mathcal{B}$, the family
\begin{equation*}
    \mathcal{B}_c:= \{B \in \mathcal{B}: \overline{B} \text{ is compact}\}
\end{equation*}
is a basis for $X$.
\end{lemma}

\begin{proof} Observe that $\mathcal{B}_c \subseteq \mathcal{B}$ and therefore it is a family of open sets, then topology generated by $\mathcal{B}_c$ is contained in the topology of $X$. So it is sufficient to prove the inverse inclusion. Given an open set $V$ and $x \in V$, since $X$ is locally compact, there exists a compact neighborhood $K_x$ of $x$ contained in $V$. Since $\mathcal{B}$ is a basis, we also have that there exists $B_x \in \mathcal{B}$ such that $x \in B_x \subseteq K_x$. By the Hausdorff property, we have that every compact is closed and hence $\overline{B_x} \subseteq K_x$. Therefore $B_x$ has compact closure and then $B_x \in \mathcal{B}_c$. Consequently we have
\begin{equation*}
    V = \bigcup_{B\in \mathcal{B}_c: B \subseteq V} B,
\end{equation*}
and therefore $\mathcal{B}_c$ is a basis for the topology on $X$. 
\end{proof}

\begin{lemma}\label{lemma:CcXseparable}
Given a non-empty LCH second countable topological space $X$, then $C_c(X)$ is separable on the uniform norm.
\end{lemma}

\begin{proof} If $X$ is a singleton, then there is nothing to be proven. So assume $X$ is not a singleton. By Lemma \ref{lemma:LCH_has_precompact_basis}, every basis of $X$ contains a basis $\mathcal{B}_c$ of relatively compact open sets. In particular, since $X$ is second countable, $\mathcal{B}_c$ can be chosen countable. In this case, define 
\begin{equation*}
    \mathcal{B}^{(2)}_c = \lbrace (U, V)  \in \mathcal{B}_c \times \mathcal{B}_c : \overline{U} \cap \overline{V} = \emptyset \rbrace.
\end{equation*}
Since $\mathcal{B}_c$, we have that $\mathcal{B}^{(2)}_c$ is countable as well. Also, $\mathcal{B}^{(2)}_c\neq \emptyset$ because $X$ is not a singleton. By Urysohn's Lemma \ref{lemma:Urysohn_LCH}, for each pair $(U, V) \in \mathcal{B}^{(2)}_c$, there exists a real continuous function $f_{U, V}$ on $X$, such that
\begin{align*}
0 \leq f_{U, V} \leq 1, \quad (f_{U, V})\vert_{U} = 1, \quad \text{and} \quad (f_{U, V})\vert_{V} = 0.
\end{align*}
For each pair $(U,V) \in \mathcal{B}^{(2)}$, choose a function as above, and let $\mathfrak{F}$ be the set of the chosen functions. It is straightforward that $\mathfrak{F}$ is countable. Now, consider the complex algebra $\mathcal{A}$ generated by $\mathfrak{F}$, endowed with usual sum and products, and the subalgebra $\mathcal{A}_0$, also generated by $\mathfrak{F}$, and with scalars in $\mathbb{Q} + i\mathbb{Q}$. Observe that $\mathcal{A}_0$ is countable and dense in $\mathcal{A}$. Observe that $\overline{f} \in \mathcal{A}$ for every $f \in \mathcal{A}$, because the generators of $f \in \mathcal{A}$ are real-valued functions. We claim that $\mathcal{A}$ is a separating set for $X$. In fact, given $x_1, x_2 \in X$, we have by Hausdorff property that there exist two disjoint open sets $U_1$ and $U_2$ s.t. $x_1 \in U_1$ and $x_2 \in U_2$. By the local compactness of $X$, there are two compact sets $K_1 \subseteq U_1$ and $K_2 \subseteq U_2$ compact neighborhoods of $x_1$ and $x_2$, respectively. Since $\mathcal{B}_c$ is a basis, there are $B_1, B_2 \in \mathcal{B}_c$ satisfying
\begin{equation*}
    x_1 \in B_1 \subseteq K_1 \quad \text{and} \quad x_2 \in B_2 \subseteq K_2.
\end{equation*}
It is clear that $\overline{B_1} \cap \overline{B_2} = \emptyset$, and then $(B_1,B_2) \in \mathcal{B}^{(2)}_c$. We have that
\begin{equation*}
    f_{B_1,B_2}(x_1) = 1 \quad \text{and} \quad f_{B_1,B_2}(x_2) = 0,
\end{equation*}
and therefore $\mathcal{A}$ is a separating set for $X$. Since $x_1$ is arbitrary, there exists a generator function $f$ such that $f(x_1) = 1$. Hence, by the Stone-Weierstrass Theorem, $\mathcal{A}$ is dense in $C_0(X)$. In particular $\mathcal{A}_0$ is dense in $C_c(X)$.
\end{proof}

\begin{corollary}\label{cor:CcUseparable} Let $X$ be a LCH second countable space. For every open subset $U$ of $X$, $C_c(U)$ is separable on the uniform norm.
\end{corollary}

\begin{proof} $C_c(X)$ is a separable metric space by Lemma \ref{lemma:CcXseparable}. Then the statement holds because $C_c(U)$ can be seen as a metric subspace of $C_c(X)$ and every metric subspace of a separable metric space is also separable.
\end{proof}

\begin{proposition}
If $G$ is a LCH second countable \'etale groupoid, then $C^*(G)$ is separable.
\end{proposition}

\begin{proof} Lemma \ref{lemma:LCH_has_precompact_basis} grants the existence of a countable family $\mathcal{I}$ of open bisections with compact support that covers $G$. Corollary \ref{cor:CcUseparable} gives that $C_c(\mathcal{U})$ contains countable dense subset $\mathcal{A}_{\mathcal{U}}$ with respect to the supremum norm. Let $\mathcal{A}_0$ be the set of finite sums of elements in $\cup_{\mathcal{U} \in I} \mathcal{A}_{\mathcal{U}}$. Observe that $\mathcal{A}_0$ is countable.

Now, for given $f \in C_c(G)$, there exists a family $\{\mathcal{U}_k\}_{k=1}^n \subseteq \mathcal{I}$, $n \in \mathbb{N}$, covering $\supp(f)$. So consider $\{P_k\}_{k=1}^n$ a partition of unit subordinate to $\{\mathcal{U}_k\}_{k=1}^n$.

For every $\epsilon > 0$, there exists $\varphi_k \in \mathcal{A}_{\mathcal{U}_k}$ s.t.
\begin{equation*}
    \|\varphi_i - P_i f\|_\infty < \frac{\epsilon}{n}.
\end{equation*}
Proposition \ref{prop:reprnormbound} gives that
\begin{equation*}
    \| \pi(\varphi_i) - \pi(P_i f) \| < \frac{\epsilon}{n}.
\end{equation*}
for every $*$-representation of $C_c(G)$. Now, take $\varphi = \sum_{k=1}^n \varphi_k \in \mathcal{A}_0$. We obtain
\begin{align*}
\| \varphi - f \| = \left\| \sum_{k=1}^n \varphi_k - \sum_{k=1}^n P_k f_k \right\| \leq  \sum_{k=1}^n \| \varphi_k - P_k f_k \| = \sum_{k=1}^n \sup_{\pi \in \mathcal{R}}\|\pi(\varphi_k) - \pi(P_k f_k)\|
< n \frac{\epsilon}{n} = \epsilon.
\end{align*}
We conclude that $\mathcal{A}_0$ is dense in $C_c(G)$. By density of $C_c(G)$ in $C^*(G)$, we obtain that $C^*(G)$ is separable.
\end{proof}

\section{Haar systems and quasi-invariant measures}

Fixed a LCH second countable groupoid $G$, there is a way to create a Borel measure on $G$ starting from another Borel measure on $G^{(0)}$. This is possible when we have a family of mesures $\lbrace \lambda^x \rbrace_{x \in G^{(0)}}$, with $\lambda^x$ supported on $G^x$. This construction, under some hipotheses, defines a Haar system, which is a generalization of the concept of Haar measure, studied in Group Theory. In this section we construct the notion of Haar system and present some of its properties. In particular, we are interested in the case when the a priori measure on $G^{(0)}$ is quasi-invariant, since in this case it is intrinsically related to the KMS states on groupoid C$^*$-algebras.

First, we define the notion of Radon measure, that is necessary to define Haar systems.

\begin{definition}[Radon measures] A Borel measure $\mu$ on a locally compact Hausdorff
space $X$ will be called a \textit{Radon measure} when it is non-negative and satisfies the following
\begin{enumerate}
    \item  $\mu(K)<\infty$  for all compact sets $K \subseteq X$.
    \item $\mu(E) = \inf{\{\mu (V ) : E \subseteq V, V \text{ is open}\}}$ for all Borel sets $E$
    \item $\mu (E) = \sup{\{\mu (K) : K \subseteq E, K\text{ is compact}\}}$ for all open sets $E$ and all Borel sets $E$ such that $\mu(E) < \infty$.
\end{enumerate}
\end{definition}

\begin{remark} In \cite{Cohn2013}, Radon measures are called regular measures and they are defined for every $\sigma$-algebra that contains the Borel $\sigma$-algebra. In this thesis we work with Borel measures only. 
\end{remark}

In particular, we will use the following result.

\begin{lemma}\label{lemma:measure_finite_on_compacts_is_Radon}
Let $X$ be a locally compact second countable Hausdorff space
and let $\mu$ be a Borel measure on $X$. If $\mu$ is finite on compact sets of $X$ then $\mu$ is Radon measure. Furthermore, any Radon measure $\mu$ on $X$
satisfies:
$$\mu(E) = \sup{\{\mu(K) : K \subseteq E, K\text{ is compact}\}}$$
for all Borel sets $E$.
\end{lemma}
\begin{proof}
The first statement comes from Proposition 7.2.3 and the second one comes from Proposition 7.2.6 of \cite{Cohn2013}.
\end{proof}

Now, we define the concept of Haar system.

\begin{definition}
\label{def:Haar_system}
Consider a LCH groupoid $G$. A (left) Haar system is a family of Radon measures $\lbrace \lambda^x \rbrace_{x \in G^{(0)}}$ satisfying the following properties.

\begin{itemize}
\item[$(i)$] For every $x \in \mathcal{G}^{(0)}$, $\supp \lambda^x = G^x$;

\item[$(ii)$] given $f \in C_c(G)$, the map
\begin{align*}
x \mapsto \int_G f(g)d\lambda^x(g)
\end{align*}
belongs to $C_c(G^{(0)})$;

\item[$(iii)$] given $f \in C_c(G)$ and $h \in G$, we have
\begin{align*}
\int_G f(hg)d\lambda^{s(h)}(g) = \int_G f(g)d\lambda^{r(h)}(g).
\end{align*}
\end{itemize}
\end{definition}

\begin{remark} The item $(ii)$ in definition above is stated as it is in Definition 2.2.2 of \cite{Paterson1999}. In Definition 2.2. of \cite{Renault1980}, the map
\begin{align*}
x \mapsto \int_G f(g)d\lambda^x(g)
\end{align*}
just need to be continuous. However, they are equivalent statements. In fact, for $f \in C_c(G)$, let $K = \supp f$ and $F:G^{(0)} \to \mathbb{C}$ be the continuous function defined by
\begin{equation*}
    F(x) := \int_G f d\lambda^x.
\end{equation*}
Since $K$ is compact, we have that $r(K)$ is compact as well. If $x \notin r(K)$, that is, $\{x\} \cap r(K) = \emptyset$, then $G^x \cap K \subseteq G^x \cap r^{-1}(r(K)) = \emptyset$, hence $G^x \cap K = \emptyset$, and then $F(x) = 0$. Therefore $F \in C_c(G^{(0)})$.
\end{remark}

\begin{remark}\label{remark:Haar_system_generalizes_Haar_measure} The notion of Haar system above generalizes the concept of Haar measure. As it is in \cite{Cohn2013}, given a LCH topological group\footnote{A topological group is a group endowed with a topology such that the group operations are continuous.} $G$, we recall that $G$ is a topological groupoid as well, with $G^{(0)} = \{e\}$. A (left) Haar measure is a non-zero Radon measure $\mu$ on $G$ such that
\begin{equation*}
    \mu(g B) = \mu (B),
\end{equation*}
for every $g \in G$ and every $B$ measurable subset of $G$, where $g B = \{gb : b \in B\}$. It is known that a Radon measure $\mu$ on a Haar measure if and only if $\{\lambda_e\}$, $\lambda^e = \mu$, is a Haar system. A proof for this fact can be found in \cite{Lima2019}.
\end{remark}

Accordingly to Lemma \ref{prop:G0clopen}, for a Hausdorff \'etale groupoid $G$, we necessarily have that $G^{(0)}$ is open. In this case, by Lemma 2.7 of \cite{Renault1980}, if there exists the Haar system, it is unique and it is the set of counting measures on the $r$-fibers. In this thesis we only work with Hausdorff \'etale groupoids, then we will focus on the Haar system of these counting measures from now.

\begin{proposition} \label{prop:counting_measures_on_bundles_are_Haar_systems} For a LCH second countable \'etale groupoid $G$, consider the family of measures $\lbrace \lambda^x \rbrace_{x \in G^{(0)}}$, where $\lambda^x$ is the counting measure on $r$-fiber $G^x$. Such family is a left Haar system.
\end{proposition}

\begin{proof}
We prove that each property of Definition \ref{def:Haar_system} holds. By Proposition \ref{prop:GxSecondCountable} we have that $G^x$ closed discrete subset of $G$, then for every $x \in G^{(0)}$, the measure $\lambda^x$ is well-defined for the Borel $\sigma$-algebra on $G$.

\begin{itemize}
\item[$(i)$] Fix $x \in G^{(0)}$. Since $G^x$ is closed, and by definition the counting measure $\lambda^x$ is non-zero only for non-empty subsets of $G^x$, we have that $\supp \lambda^x = G^x$.

Now we prove that $\lambda^x$ is a Radon measure. Again by Proposition $\ref{prop:GxSecondCountable}$, we may take $\mathcal{U}_g$ open bisection satisfying $\mathcal{U}_g \cap G^x = \lbrace g \rbrace$. Given $K \subset G$ compact set, we have that $K' = K \cap G^x$ is compact as well and $\lambda^x(K) = \lambda^x(K')$. By compactness, there are $g_1, \hdots, g_n \in G^x$ such that $K' \subseteq \bigcup_{\ell = 1}^n\mathcal{U}_{g_\ell}$. Since $G$ is second countable, Proposition \ref{prop:GxSecondCountable} gives that $G^x$ is countable and then,
\begin{align*}
\lambda^x(K) &= \lambda^x(K') \leq \lambda^x\left(\bigcup_{\ell = 1}^n\mathcal{U}_{g_\ell}\right) \leq \sum_{\ell=1}^n \lambda^x(\mathcal{U}_{g_\ell}) = \sum_{\ell=1}^n \lambda^x(g_\ell) = n,
\end{align*}
and we conclude that the measure $\lambda^x$ is finite on compact subsets, and by Lemma \ref{lemma:measure_finite_on_compacts_is_Radon} $\lambda^x$ is Radon measure.

\item[$(ii)$] Consider $\mathcal{U}$ an open bisection and take $f \in C_c(\mathcal{U})$. Also, set $K = \supp f$, $V = r(\mathcal{U})$ and $W = r(K)$. Note that $V$ is open and $W$ is compact. So we may define $\tilde{f} \in C_c(r(V))$ by
\begin{align*}
\tilde{f}(y) = \begin{cases}
                    f(r\vert_{\mathcal{U}}^{-1}(y)), \quad \text{if $y \in W$;}\\
                    0, \quad \text{otherwise.}
               \end{cases}
\end{align*}
It is straightforward that
\begin{align*}
\tilde{f}(y) = \int_G f(g) d\lambda^y(g),
\end{align*}
for every $y \in G^{(0)}$. The general case, for $f \in C_c(G)$, is straightforward from the previous one, due to Lemma \ref{lemma:function_sum_supported_open_bisections} and the linearity of the integral.

\item[$(iii)$] For fixed $f \in C_c(G)$ and $h \in G$, we have
\begin{align*}
\int_G f(hg) d\lambda^{s(h)}(g) &= \sum_{g \in G^{s(h)}} f(hg),
\end{align*}
and since the map $\psi: G^{s(h)} \to G^{r(h)}$, $\psi(g):= hg$, is bijective\footnote{Its inverse map $\psi^{-1}: G^{r(h)} \to G^{s(h)}$ is given by $\psi^{-1}(g):= h^{-1}g$.}, one gets
\begin{align*}
\sum_{g \in G^{s(h)}} f(hg) = \sum_{g \in G^{s(h)}} f(\psi(g)) = \sum_{\tilde{g} \in G^{r(h)}} f(\tilde{g}) = \int_G f(\tilde{g}) d\lambda^{r(h)} (\tilde{g}),
\end{align*}
that is,
\begin{align*}
\int_G f(hg) d\lambda^{s(h)}(g) &= \int_G f(g) d\lambda^{r(h)}(g). \tag*{\qedhere}
\end{align*}
\end{itemize}
\end{proof}

Now we define measures associated to the left Haar system.

\begin{definition}\label{def:right_Haar_system} Suppose that there exists a left Haar system $\lbrace \lambda^x \rbrace_{x \in G^{(0)}}$. For each $x \in G^{(0)}$ we define the measure $\lambda_x$ by setting
\begin{equation*}
   \lambda_x(B) = \lambda^x(B^{-1}), 
\end{equation*}
for every $B$ measurable subset of $G$, where $B^{-1}:= \{g^{-1}: g \in B\}$.
\end{definition}

Now we show the system of measures $\{\lambda_x\}$ in the definition above are counting measures on each $s$-fiber $G_x$, $x \in G^{(0)}$, when the left Haar system is the set of counting measures on each respective $r$-fiber.

\begin{lemma}\label{lemma:left_Haar_system_counting_r_fiber_implies_right_Haar_system_counting_s_fiber} Let $G$ be LCH second countable \'etale, $x \in G^{(0)}$, and let $\lambda^x$ be the counting measure on the $r$-fiber $G^x$. Then the measure $\lambda_x$ is the counting measure on the $s$-fiber $G_x$. 
\end{lemma}

\begin{proof}
Given $B$ be a measurable set and denote by $\vert B \vert$ its number of elements. Then,
\begin{align*}
\lambda_x(B) &= \lambda^x(B^{-1}) = \vert \{ g : g \in B^{-1} \cap G^x \rbrace \} = \vert \{ h^{-1} : h \in B \cap G_x \} \vert,
\end{align*}
where in the last we used that the inversion map is a bijection from $B^{-1}$ onto $B$ and the same holds from $G^x$ onto $G_x$. The bijectivity of the inverse map also implies that
\begin{align*}
\vert \{ h^{-1} : h \in B \cap G_x \} \vert &= \vert \lbrace h : h \in B \cap G_x \rbrace \vert,
\end{align*}
where we used the change of variables $h \mapsto h^{-1}$. Therefore, $\lambda_x$ is the counting measure on $G_x$.
\end{proof}

Now, we define two induced measures from the Haar systems.

\begin{definition}
\label{def:auxiliary_measure_quasi_invariance}
Suppose that $G$ is LCH and consider a Haar system $\{ \lambda^x \}_{x \in G^{(0)}}$ and a Radon measure $\mu$ on $\mathcal{B}_{G^{(0)}}$. We define the induced measures $\mu_r$ and $\mu_s$, given by 
\begin{align*}
\mu_r(B) := \int_{G^{(0)}} \lambda^x(B)d\mu(x) \quad \text{and} \quad \mu_s(B) = \int_{G^{(0)}} \lambda_x(B)d\mu(x),
\end{align*}
for every $B \in \mathcal{B}_G$.

Notation: $\nu = \int_{G^{(0)}} \lambda^x d\mu(x) $, $\nu^{-1} = \int_{G^{(0)}} \lambda_x d\mu(x)$. 
\end{definition}

\begin{remark}\label{remark:equivalent_definition_mu_r_and_mu_s} The definition above is equivalent to the following: for every $f \in C_c(G)$, we define the measuress $\mu_r$ and $\mu_s$ by the identities
\begin{align*}
    \int_G f(g) d\nu(g) := \int_{G^{(0)}} \int_{G^{x}} f(g) d\lambda^x(g) d\mu(x) \quad \text{and} \quad \int_G f(g) d\nu^{-1}(g) := \int_{G^{(0)}} \int_{G_{x}} f(g) d\lambda_x(g) d\mu(x).
\end{align*}
\end{remark}

\begin{remark} The measures $\mu_r$ and $\mu_s$ as presented in Definition \ref{def:auxiliary_measure_quasi_invariance} when the Haar system is the counting measure are also used out of the context of groupoids. In \cite{Kimura2015}, there is a similar construction used to define conformal measures. 
\end{remark}

In the paricular case when $\lambda^x$ is the counting measure on $G^x$ for each $x \in G^{(0)}$, we have
\begin{align*}
\int f(g) d\mu_r(g) = \int_{G^{(0)}} \sum_{g \in G^{x}} f(g) d\mu(x)  \quad \text{and} \quad 
\int f(g) d\mu_s(g) = \int_{G^{(0)}} \sum_{g \in G_{x}} f(g) d\mu(x),
\end{align*}
for every $f \in C_c(G)$.

Now we present the notion of quasi-invariant measure.

\begin{definition}[quasi-invariant measure] Let $G$ be a LCH \'etale groupoid, $\{\lambda^x\}_{x \in G^{(0)}}$ be a Haar system and consider a Radon measure $\mu$ on $G^{(0)}$. $\mu$ is said to be quasi-invariant when $\mu_r \sim \mu_s$.
\end{definition}

\begin{remark} Consider the group $(0, \infty)$ endowed with the usual product. Item 3 of Corollary 3.14 of \cite{Hahn1978} allows us to choose the Radon-Nikodym derivative $\frac{d \mu_r}{d \mu_s}$ being a homomorphism from $G$ to $(0, \infty)$.
\end{remark}

\begin{proposition}
\label{prop:RN_derivative_mu_r_mu_s_onisotropy} Consider a LCH second countable \'etale groupoid $G$ and a Haar system $\lbrace \lambda^x \rbrace_{x \in G^{(0)}}$ of counting measures $\lambda^x$ on $G^x$. Let $\mu$ be a quasi-invariant measure on $\mathcal{B}_{G^{(0)}}$. For $\mu$-a.e. $x$ and all $g \in G_x^x$, we have $\frac{d \mu_r}{d \mu_s} = 1$.
\end{proposition}

\begin{proof} Consider the isotropy bundle $\Iso(G) = \cup_{x \in G^{(0)}} G_x^x$. By Lemma \ref{lemma:isotropyclosed} we have that $\Iso(G)$ is closed and hence it is a Borel set. For every positive measurable function $f$ with support cointained in $\Iso(G)$, we have
\begin{align}
\label{eqn:quasiinvariantfpos}
\int_G f(g) d\mu_r(g) = \int_G f(g) \frac{d \mu_r}{d \mu_s}(g) d\mu_s(g),
\end{align}
and by definition we also have
\begin{align}
\int_G f(g) d\mu_r(g) &= \int_{G^{(0)}} \sum_{g \in G^x} f(g) d\mu(x) = \int_{G^{(0)}} \sum_{g \in G_x^x} f(g) d\mu(x), \label{eqn:fdmur}
\end{align}
where in the last equality we used the fact that $\supp f \subseteq \Iso(G)$. Since $\supp \left(f \frac{d \mu_r}{d \mu_s} \right) \subseteq \Iso(G)$ as well, one gets
\begin{align}
\int_G f(g) \frac{d \mu_r}{d \mu_s}(g) d\mu_s(g) &= \int_{G^{(0)}} \sum_{g \in G_x} f(g) \frac{d \mu_r}{d \mu_s}(g) d\mu(x) = \int_{G^{(0)}} \sum_{g \in G_x^x} f(g) \frac{d \mu_r}{d \mu_s}(g) d\mu(x). \label{eqn:fDeltadmur}
\end{align}
By inserting the identities $\eqref{eqn:fdmur}$ and $\eqref{eqn:fDeltadmur}$ in $\eqref{eqn:quasiinvariantfpos}$, one obtains
\begin{align}\label{eq:integral_sum_1_minus_delta}
\int_{G^{(0)}} \sum_{g \in G_x^x} \left(\frac{d \mu_r}{d \mu_s}(g) - 1\right) f(g) d\mu(x) = 0
\end{align}
Now, define the sets
\begin{equation*}
    P_+ := \left\{ g \in \Iso(G) : \frac{d \mu_r}{d \mu_s}(g) \geq 1 \right\} \quad \text{and} \quad P_- := \left\{ g \in \Iso(G) : \frac{d \mu_r}{d \mu_s}(g) \geq 1 \right\}.
\end{equation*}
For $f$
If we choose $f = \mathbbm{1}_{P_+}$, equation \eqref{eq:integral_sum_1_minus_delta} becomes
\begin{align*}
\int_{G^{(0)}} \sum_{g \in G_x^x \cap P_+} \left(\frac{d \mu_r}{d \mu_s}(g) - 1\right)d\mu(x) = 0,
\end{align*}
and then $\frac{d \mu_r}{d \mu_s} = 1$ $\mu$-a.e. on $P_+$ because
\begin{equation*}
    \sum_{g \in G_x^x \cap P_+} \left(\frac{d \mu_r}{d \mu_s}(g) - 1\right) \geq 0
\end{equation*}
for every $x \in G^{(0)}$. Similar proof holds for $P_-$.
\end{proof}

In this thesis, we are interested to relate quasi-invariant measures to KMS states via continuous $1$-cocycles, defined as follows.

\begin{definition}\label{def:cocycle} Let $G$ be a groupoid. A $1$-cocycle is a function $c:G \to \mathbb{R}$ satisfying
\begin{equation*}
    c(gh) = c(g) + c(h), \quad (g,h) \in G^{(2)}.
\end{equation*}
\end{definition}

For a LCH \'etale groupoid $G$ a continuous 1-cocycle defines a C$^*$-dynamical system $(C^*(G),\tau)$, where $\tau = \{\tau_t\}_{t \in \mathbb{R}}$ is the one-parameter group of automorphisms given by
\begin{equation}\label{eq:KMS_cocycle}
    \tau_t(f)(g) = e^{itc(g)} f(g),
\end{equation}
for every $f \in C_c(G)$ and $g \in G$, and extended (uniquely) to $C^*(G)$. For $f \in C_c(G)$, we can extend \eqref{eq:KMS_cocycle} in $t$ to the whole complex plane, and for $\beta>0$, a KMS$_\beta$ state $\varphi$ on $C^*(G)$ is precisely a state that satisfies the KMS condition on $C_c(G)$, that is,
\begin{equation*}
    \varphi(f_1 \tau_{i\beta}(f_2)) = \varphi(f_2 f_1)
\end{equation*}
for every $f_1,f_2 \in C_c(G)$. The quasi-invariant probability measures s.t.
\begin{equation*}
    \frac{d \mu_r}{d \mu_s} = e^{-\beta c},
\end{equation*}
where $c$ is a continuous $1$-cocycle, are strictly related to KMS$_\beta$ states. For the generalized Renault-Deaconu groupoid (see section \ref{section:GRD_groupoid}) and a 1-cocycle associated to a potential (see chapter \ref{ch:TF_on_Generalized_Countable_Markov_shifts}), we discuss this relation in Remark \ref{remark:KMS_quasi_invariant}.

\chapter{Cuntz-Krieger/Exel-Laca algebras and the generalized Markov shifts}
\label{ch:CK_EL}

This chapter focuses on two particular universal C$^*$-algebras, namely the Cuntz-Krieger algebras \cite{CK1980} and, specially, their generalization for the infinite alphabet, the Exel-laca algebras \cite{EL1999}. We will explain how these algebras are related to the Markov shift spaces and how they are represented in terms of groupoid C$^*$-algebras. 

\section{Cuntz-Krieger algebras}

In 1963, J. Dixmier proved the existence of a separable simple\footnote{We say a C$^*$-algebra is simple if it does not contains non-trivial closed two sided ideals.} infinite\footnote{A simple unital C$^*$-algebra is said to be infinite if it contains an element $a$ such that $a^*a = 1$ and $aa^* \neq 1$.} C$^*$-algebra \cite{Dixmier1963}. However, the first concrete examples of this type of algebra was constructed almost fifteen years later, in 1977, by J. Cuntz \cite{Cuntz1977}, the Cuntz algebras. In 1980, J. Cuntz and W. Krieger \cite{CK1980} generalized the Cuntz algebras to a bigger class of C$^*$-algebras, whose structure encodes the Markov shift spaces in its generators, and these algebras became famous, being called Cuntz-Krieger algebras. Before we formally introduce the Cuntz-Krieger algebras in terms of universal C$^*$-algebra generated by a set under some relations, we construct them in a less abstract approach.

Given separable infinite dimensional Hilbert space $\mathcal{H}$ and $n \in \mathbb{N}$, we split $\mathcal{H}$ as
\begin{equation*}
    \mathcal{H} = \mathcal{H}_1 \oplus \mathcal{H}_2 \oplus \cdots \oplus \mathcal{H}_n,
\end{equation*}
where each $\mathcal{H}_i$ is also separable and infinite dimensional. Now consider a $n \times n$ transition matrix $A$ that every row and every column is non-zero. For each $i \leq n$, we consider the infinite dimensional separable subspace 
\begin{equation*}
    \bigoplus_{j: A_{i,j}=1}\mathcal{H}_j
\end{equation*}
and we may choose for each $i$ an isometric isomorphism
\begin{equation*}
    S_i:\bigoplus_{j: A_{i,j}=1}\mathcal{H}_j \to \mathcal{H}_i.
\end{equation*}

\begin{example} If we take $n = 4$ and the transition matrix
\begin{equation*}
    A = \begin{pmatrix}
            1 & 1 & 1 & 0 \\
            1 & 0 & 0 & 1 \\
            0 & 1 & 1 & 1 \\
            1 & 1 & 0 & 1 
        \end{pmatrix},
\end{equation*}
then the isometric isomorphisms are
\begin{align*}
    S_1&: \mathcal{H}_1 \oplus \mathcal{H}_2 \oplus \mathcal{H}_3 \to \mathcal{H}_1, \\ 
    S_2&: \mathcal{H}_1 \oplus \mathcal{H}_4 \to \mathcal{H}_2, \\
    S_3&: \mathcal{H}_2 \oplus \mathcal{H}_3 \oplus \mathcal{H}_4 \to \mathcal{H}_3, \\ 
    S_4&: \mathcal{H}_1 \oplus \mathcal{H}_2 \oplus \mathcal{H}_4 \to \mathcal{H}_4.
\end{align*}
\end{example}

By extending each $S_i$ to $\mathcal{H}$ as $S_i a = 0$ for $a \notin \bigoplus_{j: A_{i,j}=1}\mathcal{H}_j$, we obtain a family $\{S_i\}_{i=1}^n$ of partial isometries on $\mathcal{H}$, i.e.
\begin{equation*}
    S_iS_i^*S_i = S_i, \quad i= 1,..., n,
\end{equation*}
which satisfies the relations
\begin{equation} \label{eq:CK}
    \sum_{j=1}^n S_jS_j^* = 1 \quad \text{and} \quad S_i^*S_i = \sum_{j=1}^nA(i,j) S_jS_j^*.
\end{equation} 
The relations in \eqref{eq:CK} are called Cuntz-Krieger relations, and the next result shows that these relations are suitable to create a universal algebra. From now on we refer the left equation in \eqref{eq:CK} as (CK1) and the right one as (CK2).

\begin{theorem}\label{thm:CK_are_universal_algebras} Consider the set $D = \{S_i\}_{i=1}^n$, $n \in \mathbb{N}$, and $\mathscr{R}$ the collection of the following relations on $D$:
\begin{itemize}
    \item $(S_iS_i^*S_i - S_i,0)$, $i= 1,..., n$;
    \item $\left(1-\sum_{j=1}^n S_jS_j^*,0\right)$;
    \item and $\left(S_i^*S_i - \sum_{j=1}^nA(i,j) S_jS_j^*,0\right)$, $i = 1,...,n$;
\end{itemize}
where $A$ is a $n \times n$ matrix with entries in $\{0,1\}$. The pair $(D,\mathscr{R})$ is admissible. 
\end{theorem}
\begin{proof} Let $\mathcal{A}_D$ be the free associative complex $*$-algebra generated by the set $D$ and $\Theta: D \to B$ be a representation for the pair $(D,\mathscr{R})$, where $B$ is a C$^*$-algebra. By the GNS construction, there is always a Hilbert space $\mathcal{H}$ that admits a faithful representation $\Phi:B \to \mathscr{B}(\mathcal{H})$. Since $\mathscr{B}(\mathcal{H})$ is also a C$^*$-algebra, then $\Phi$ is an isometric map. Then,
\begin{equation*}
    \|\Theta(S_i)\| = \|\Phi \circ \Theta(S_i)\| = \sup_{\substack{h \in \mathcal{H},\\ \|h\|=1}} \frac{\|\Phi \circ \Theta(S_i)h\|}{\|h\|}.
\end{equation*}
On the other hand, for $h \in \mathcal{H}$ satisfying $\|h\| =1$, we have that
\begin{align*}
    \|\Phi \circ \Theta(S_i)h\|^2 &= \left(\Phi \circ \Theta(S_i)h,\Phi \circ \Theta(S_i)h\right) = \left(\Phi \circ \Theta(S_i)^*\Phi \circ \Theta(S_i)h,h\right) \\
    &= \left(\Phi \circ \Theta(S_i)^*\Phi \circ \Theta(S_i)h,h\right) \leq \left(h,h\right) = \|h\|^2,
\end{align*}
where the inequality above becomes an equality when $h \notin \ker\left(\Theta(S_i)\right)$. Since $\|h\| > 0$, it follows that
\begin{align*}
    \frac{\|\Phi \circ \Theta(S_i)h\|}{\|h\|} &\leq 1,
\end{align*}
and then
\begin{equation*}
    \|\Theta(S_i)\| \leq 1.
\end{equation*}
Since the representation $\Theta$ is arbitrary, we conclude that the pair $(D,\mathscr{R})$ is admissible.
\end{proof}

The theorem above grants that the Cuntz-Krieger algebras are in fact well defined and now we define them formally as next.

\begin{definition}\label{def:CK_algebras} Consider $n \in \mathbb{N}$ and a $n \times n$ transition matrix $A$. The Cuntz-Krieger algebra $\mathcal{O}_A$ is the universal C$^*$-algebra generated by a family of partial isometries $\{S_i\}_{i=1}^n$ satisfying the Cuntz-Krieger relations \eqref{eq:CK} for $A$.
\end{definition}

\begin{remark}\label{remark:non_zero_row_column} Observe that Theorem \ref{thm:CK_are_universal_algebras} does not depend on the matrix $A$ and hence Definition \ref{def:CK_algebras} has the same independency as well. However, supposing that the $j$-th row of $A$ has only zeros, we necessarily have that $S_j = 0$. On the other hand, if the $j$-th column of $A$ is zero, we have that $S_iS_j = 0$ for every $i \in S$, this case is proven in Lemma 2.1 of \cite{CK1980}. 
\end{remark}

For the rest of this section we assume the following standing hypothesis.

\begin{mdframed} \textbf{Standing hypothesis:} any transition matrix for \emph{finite} alphabets has only non-zero rows and only non-zero columns.
\end{mdframed}

In addition, we consider the projections $P_i = S_i S_i^*$ and $Q_i = S_i^* S_i$, $i = 1,...n$. We present next some properties of these projections and the elements $S_i$.

\begin{lemma}\label{lemma:CKprop} Consider a Cuntz-Krieger algebra $\mathcal{O}_A$, where $A$ is a $n\times n$ matrix. The following assertions are true:
    \begin{itemize}
        \item[$(i)$] $i \neq j \implies P_iP_j = 0$;
        \item[$(ii)$] $S_i^*S_j = \delta_{ij}S_i^*S_j$;
        \item[$(iii)$] $S_i^*S_iS_j = A(i,j)S_j$;
        \item[$(iv)$] $A(i,j) = 1 \implies \ran S_j \subseteq (\ker S_i)^{\perp}$;
        \item[$(v)$] $P_i = \sum_{j=1}^nS_i P_jS_i^*$.
    \end{itemize}
\end{lemma}

\begin{proof} Indeed,
    \begin{itemize}
        \item[$(i)$] accordingly to (CK1), we have that $\sum_{k = 1}^n P_k = 1$. Fix $i,j \{1,...,n\}$ with $i \neq j$. Then,
            \begin{equation*}
                P_j + P_jP_iP_j = P_j^3 +P_jP_iP_j = P_j(P_j + P_i)P_j \leq P_j \left(\sum_{k = 1}^n P_k\right) P_j \leq P_j 1 P_j = P_j \implies P_j P_i P_j = 0.
            \end{equation*}
            On the other hand, we have that
            \begin{equation*}
                P_j P_i P_j = P_j(P_i)^2P_j = (P_i^*P_j^*)^* P_i P_j = (P_i P_j)^* P_i P_j = (P_i P_j)^2,
            \end{equation*}
        and we conclude that $P_i P_j = 0$;    
        \item[$(ii)$] here we can use an equivalent definition\footnote{Remember that $a$ is a partial isometry iff $aa^*$ is a projection iff $a = aa^*a$.} for partial isometries and get $S_i^*S_j = (S_i^*S_iS_i^*)(S_jS_j^*S_j) = S_i^*P_iP_jS_j = \delta_{ij}S_i^*S_j$;
        \item[$(iii)$] $S_i^*S_iS_j \stackrel{\text{(CK2)}}{=} \sum_{k=1}^n A(i,k) S_k S_k^* S_j \stackrel{(ii)}{=} A(i,j)S_j$;
        \item[$(iv)$] from (CK2) we have that
            \begin{equation*}
                S_i^* S_i = \sum_{k=1,k\neq j}^n A(i,k) S_k S_k^* + S_j S_j^*,
            \end{equation*}
        hence $\ran (S_j S_j^*) \subseteq \ran (S_i^* S_i)$ because the projections $\{P_k\}_{k=1}^n$ and are orthogonal each other. Then,
        \begin{equation*}
            (\ker(S_i))^\perp = \ran(S_i^*S_i) \supseteq \ran (S_j S_j^*) = \ran(S_j). 
        \end{equation*}
        Therefore $\ran S_j \subseteq (\ker S_i)^{\perp}$;
        \item[$(v)$] note that $\sum_{j=1}^n P_j = 1$, hence
            \begin{equation*}
                P_i = S_iS_i^* = S_i1S_i^* = S_i\left(\sum_{j=1}^n P_j\right)S_i^* = \sum_{j=1}^nS_i P_jS_i^*.
            \end{equation*}
    \end{itemize}
\end{proof}

\begin{remark} On the item $(iv)$ in the lemma above we used that $\ker (S_i)^{\perp} = \ran (S_i^*S_i)$. We justify that claim as follows. Let $x \neq 0$ element of a Hilbert space $\mathcal{H}$ which has a representation of $\mathcal{O}_A$ s.t.\footnote{If $S_i = 0$, then the result is immediate.} $S_i \neq 0$. Take $z \in \ker  S_i$, and suppose $x \in \ran (S_i^*S_i)$. We prove that $(z,x) = 0$. Indeed, by hypothesis there exists $y \in \mathcal{H}$ s.t. $S_i^*S_i y = x$, then 
\begin{align*}
    (z,x) =  (z,S_i^*S_iy) = (S_i z, S_i x) = (0,S_ix) = 0,
\end{align*}
and hence $\ran (S_i^*S_i) \subseteq \ker(S_i)^\perp$. Conversely, if $x \in \ker(S_i)^\perp$, that is $S_ix \neq 0$, then
\begin{align*}
     0<(S_i x, S_i x)=  (S_i^*S_i x,  x),
\end{align*}
so $S_i^*S_i x \neq 0$ and therefore $x \in \ran(S_i^*S_i)$. We conclude that $(\ker (S_i))^{\perp} = \ran (S_i^*S_i)$.
\end{remark}

An important question answered by Cuntz and Krieger in \cite{CK1980} is related to the uniqueness of the Cuntz-Krieger algebras. By `uniqueness' we mean that for every two families we say $\{S_i\}_{i=1}^n$ and $\{\widehat{S}_i\}_{i=1}^n$ of non-zero partial isometries that satisfy the Cuntz-Krieger relations \eqref{eq:CK}, then the mapping
\begin{equation*}
    S_i \mapsto \widehat{S}_i, \quad i \in \{1,...,n\},
\end{equation*}
extends to an $*$-isomorphism from the C$^*$-algebra generated by $\{S_i\}_{i=1}^n$ to the C$^*$-algebra generated by $\{\widehat{S}_i\}_{i=1}^n$. In other words, all the faithful surjective representations of $\mathcal{O}_A$ on any C$^*$-algebra are $*$-isomorphic. For the finite alphabet case, there is an important condition to ensure such uniqueness: the condition (I) in \cite{CK1980}, which we explain now. First, let $\mathcal{S} \subset S$ be the subset of symbols defined as follows: $i \in \mathcal{S}$ if and only if there exist at least two distinct admissible words $i_0,...,i_{r-1}$ and $j_0,...,j_{s-1}$, where $r, s \geq 2$, satisfying $i_0 = i_{r-1} = j_0 = j_{s-1} = i$ and $i_k,j_\ell \neq i$ for $0 < k < r-1$ and $0 < \ell < s-1$. The condition (I) is the hypothesis for the transition matrix follows:
\begin{itemize}
    \item[(I)] for every $i \in S$ there exists an admissible word $i_0 \cdots i_{r-1}$, $r\geq 1$, such that $i_0 = i$ and $i_{r-1} \in \mathcal{S}$.
\end{itemize}
The figure \ref{fig:symbolic_graphs_condition_I} illustrates examples of the absence and occurrence of the condition (I).

\begin{figure}[H]
\centering

\caption{Two examples of symbolic graphs. The figure I represents a Markov shift space that satisfies $\mathcal{S} = \emptyset$, while in the figure II the graph represents a non-transitive Markov shift space such that $\mathcal{S} = \{1,2,3,4,5\}$.\label{fig:symbolic_graphs_condition_I}}

\scalebox{0.7}{
\tikzset{every picture/.style={line width=0.75pt}} 

\begin{tikzpicture}[x=0.75pt,y=0.75pt,yscale=-1,xscale=1]

\draw  [fill={rgb, 255:red, 0; green, 0; blue, 0 }  ,fill opacity=1 ] (149.35,287.75) .. controls (152.42,283.61) and (151.55,277.76) .. (147.41,274.69) .. controls (143.27,271.62) and (137.42,272.49) .. (134.35,276.63) .. controls (131.28,280.78) and (132.15,286.62) .. (136.29,289.69) .. controls (140.43,292.76) and (146.28,291.89) .. (149.35,287.75) -- cycle ;
\draw  [fill={rgb, 255:red, 0; green, 0; blue, 0 }  ,fill opacity=1 ] (212.42,70.34) .. controls (217.58,70.32) and (221.74,66.13) .. (221.73,60.98) .. controls (221.72,55.82) and (217.53,51.65) .. (212.37,51.67) .. controls (207.22,51.68) and (203.05,55.87) .. (203.06,61.03) .. controls (203.08,66.18) and (207.27,70.35) .. (212.42,70.34) -- cycle ;
\draw  [fill={rgb, 255:red, 0; green, 0; blue, 0 }  ,fill opacity=1 ] (103.15,152.78) .. controls (106.91,149.26) and (107.11,143.36) .. (103.59,139.59) .. controls (100.07,135.83) and (94.16,135.63) .. (90.4,139.15) .. controls (86.63,142.67) and (86.43,148.58) .. (89.96,152.34) .. controls (93.48,156.11) and (99.38,156.31) .. (103.15,152.78) -- cycle ;
\draw  [fill={rgb, 255:red, 0; green, 0; blue, 0 }  ,fill opacity=1 ] (337.75,141.65) .. controls (336.06,136.78) and (330.74,134.2) .. (325.87,135.9) .. controls (321,137.59) and (318.43,142.91) .. (320.12,147.78) .. controls (321.81,152.65) and (327.13,155.22) .. (332,153.53) .. controls (336.87,151.84) and (339.45,146.52) .. (337.75,141.65) -- cycle ;
\draw  [fill={rgb, 255:red, 0; green, 0; blue, 0 }  ,fill opacity=1 ] (291.02,288.82) .. controls (295.11,285.68) and (295.88,279.83) .. (292.74,275.74) .. controls (289.6,271.65) and (283.75,270.87) .. (279.66,274.01) .. controls (275.57,277.15) and (274.79,283.01) .. (277.93,287.1) .. controls (281.07,291.19) and (286.93,291.96) .. (291.02,288.82) -- cycle ;
\draw    (96.77,136.97) .. controls (82.06,74.62) and (159.45,25.17) .. (202.19,55.05) ;
\draw [shift={(204.12,56.47)}, rotate = 217.8] [fill={rgb, 255:red, 0; green, 0; blue, 0 }  ][line width=0.08]  [draw opacity=0] (10.72,-5.15) -- (0,0) -- (10.72,5.15) -- (7.12,0) -- cycle    ;
\draw    (219.4,56) .. controls (259.23,24.15) and (350.06,71.21) .. (332.02,132.65) ;
\draw [shift={(331.12,135.47)}, rotate = 289.08] [fill={rgb, 255:red, 0; green, 0; blue, 0 }  ][line width=0.08]  [draw opacity=0] (10.72,-5.15) -- (0,0) -- (10.72,5.15) -- (7.12,0) -- cycle    ;
\draw    (337.94,145.71) .. controls (378.22,183.1) and (359.42,267.1) .. (297.86,281.23) ;
\draw [shift={(295.02,281.82)}, rotate = 349.22] [fill={rgb, 255:red, 0; green, 0; blue, 0 }  ][line width=0.08]  [draw opacity=0] (10.72,-5.15) -- (0,0) -- (10.72,5.15) -- (7.12,0) -- cycle    ;
\draw    (278.83,288.67) .. controls (250.12,324.31) and (194.95,340.34) .. (149.22,291.18) ;
\draw [shift={(147.83,289.67)}, rotate = 407.95] [fill={rgb, 255:red, 0; green, 0; blue, 0 }  ][line width=0.08]  [draw opacity=0] (10.72,-5.15) -- (0,0) -- (10.72,5.15) -- (7.12,0) -- cycle    ;
\draw    (131.85,283.19) .. controls (62.61,268.05) and (59.93,181.17) .. (87.76,154.29) ;
\draw [shift={(89.96,152.34)}, rotate = 501.08] [fill={rgb, 255:red, 0; green, 0; blue, 0 }  ][line width=0.08]  [draw opacity=0] (10.72,-5.15) -- (0,0) -- (10.72,5.15) -- (7.12,0) -- cycle    ;
\draw  [fill={rgb, 255:red, 0; green, 0; blue, 0 }  ,fill opacity=1 ] (535.35,285.93) .. controls (538.42,281.79) and (537.55,275.94) .. (533.41,272.87) .. controls (529.27,269.8) and (523.42,270.67) .. (520.35,274.81) .. controls (517.28,278.95) and (518.15,284.8) .. (522.29,287.87) .. controls (526.43,290.94) and (532.28,290.07) .. (535.35,285.93) -- cycle ;
\draw  [fill={rgb, 255:red, 0; green, 0; blue, 0 }  ,fill opacity=1 ] (598.42,68.51) .. controls (603.58,68.5) and (607.74,64.31) .. (607.73,59.16) .. controls (607.72,54) and (603.53,49.83) .. (598.37,49.85) .. controls (593.22,49.86) and (589.05,54.05) .. (589.06,59.2) .. controls (589.08,64.36) and (593.27,68.53) .. (598.42,68.51) -- cycle ;
\draw  [fill={rgb, 255:red, 0; green, 0; blue, 0 }  ,fill opacity=1 ] (489.15,150.96) .. controls (492.91,147.44) and (493.11,141.54) .. (489.59,137.77) .. controls (486.07,134.01) and (480.16,133.81) .. (476.4,137.33) .. controls (472.63,140.85) and (472.43,146.76) .. (475.96,150.52) .. controls (479.48,154.29) and (485.38,154.48) .. (489.15,150.96) -- cycle ;
\draw  [fill={rgb, 255:red, 0; green, 0; blue, 0 }  ,fill opacity=1 ] (723.75,139.83) .. controls (722.06,134.96) and (716.74,132.38) .. (711.87,134.07) .. controls (707,135.77) and (704.43,141.09) .. (706.12,145.95) .. controls (707.81,150.82) and (713.13,153.4) .. (718,151.71) .. controls (722.87,150.02) and (725.45,144.7) .. (723.75,139.83) -- cycle ;
\draw  [fill={rgb, 255:red, 0; green, 0; blue, 0 }  ,fill opacity=1 ] (677.02,287) .. controls (681.11,283.86) and (681.88,278) .. (678.74,273.91) .. controls (675.6,269.82) and (669.75,269.05) .. (665.66,272.19) .. controls (661.57,275.33) and (660.79,281.18) .. (663.93,285.27) .. controls (667.07,289.36) and (672.93,290.14) .. (677.02,287) -- cycle ;
\draw    (482.77,135.15) .. controls (468.06,72.79) and (545.45,23.35) .. (588.19,53.22) ;
\draw [shift={(590.12,54.65)}, rotate = 217.8] [fill={rgb, 255:red, 0; green, 0; blue, 0 }  ][line width=0.08]  [draw opacity=0] (10.72,-5.15) -- (0,0) -- (10.72,5.15) -- (7.12,0) -- cycle    ;
\draw    (605.4,54.18) .. controls (645.23,22.33) and (736.06,69.39) .. (718.02,130.83) ;
\draw [shift={(717.12,133.65)}, rotate = 289.08] [fill={rgb, 255:red, 0; green, 0; blue, 0 }  ][line width=0.08]  [draw opacity=0] (10.72,-5.15) -- (0,0) -- (10.72,5.15) -- (7.12,0) -- cycle    ;
\draw    (723.94,143.89) .. controls (764.22,181.28) and (745.42,265.28) .. (683.86,279.4) ;
\draw [shift={(681.02,280)}, rotate = 349.22] [fill={rgb, 255:red, 0; green, 0; blue, 0 }  ][line width=0.08]  [draw opacity=0] (10.72,-5.15) -- (0,0) -- (10.72,5.15) -- (7.12,0) -- cycle    ;
\draw    (664.83,286.85) .. controls (636.12,322.49) and (580.95,338.52) .. (535.22,289.36) ;
\draw [shift={(533.83,287.85)}, rotate = 407.95] [fill={rgb, 255:red, 0; green, 0; blue, 0 }  ][line width=0.08]  [draw opacity=0] (10.72,-5.15) -- (0,0) -- (10.72,5.15) -- (7.12,0) -- cycle    ;
\draw    (517.85,281.37) .. controls (448.61,266.23) and (445.93,179.35) .. (473.76,152.47) ;
\draw [shift={(475.96,150.52)}, rotate = 501.08] [fill={rgb, 255:red, 0; green, 0; blue, 0 }  ][line width=0.08]  [draw opacity=0] (10.72,-5.15) -- (0,0) -- (10.72,5.15) -- (7.12,0) -- cycle    ;
\draw  [fill={rgb, 255:red, 0; green, 0; blue, 0 }  ,fill opacity=1 ] (870.42,115.51) .. controls (875.58,115.5) and (879.74,111.31) .. (879.73,106.16) .. controls (879.72,101) and (875.53,96.83) .. (870.37,96.85) .. controls (865.22,96.86) and (861.05,101.05) .. (861.06,106.2) .. controls (861.08,111.36) and (865.27,115.53) .. (870.42,115.51) -- cycle ;
\draw    (491.59,141.77) .. controls (622.84,107.35) and (516.57,287.36) .. (484.97,155.19) ;
\draw [shift={(484.5,153.18)}, rotate = 437.16] [fill={rgb, 255:red, 0; green, 0; blue, 0 }  ][line width=0.08]  [draw opacity=0] (10.72,-5.15) -- (0,0) -- (10.72,5.15) -- (7.12,0) -- cycle    ;
\draw    (877.73,106.16) .. controls (911.98,124.89) and (924.88,166.92) .. (882.48,196.82) ;
\draw [shift={(880.5,198.18)}, rotate = 326.19] [fill={rgb, 255:red, 0; green, 0; blue, 0 }  ][line width=0.08]  [draw opacity=0] (10.72,-5.15) -- (0,0) -- (10.72,5.15) -- (7.12,0) -- cycle    ;
\draw  [fill={rgb, 255:red, 0; green, 0; blue, 0 }  ,fill opacity=1 ] (873.42,213.51) .. controls (878.58,213.5) and (882.74,209.31) .. (882.73,204.16) .. controls (882.72,199) and (878.53,194.83) .. (873.37,194.85) .. controls (868.22,194.86) and (864.05,199.05) .. (864.06,204.2) .. controls (864.08,209.36) and (868.27,213.53) .. (873.42,213.51) -- cycle ;
\draw    (864.06,199.2) .. controls (824.31,187.42) and (815.83,129.56) .. (858.39,107.51) ;
\draw [shift={(861.06,106.2)}, rotate = 515.28] [fill={rgb, 255:red, 0; green, 0; blue, 0 }  ][line width=0.08]  [draw opacity=0] (10.72,-5.15) -- (0,0) -- (10.72,5.15) -- (7.12,0) -- cycle    ;
\draw    (873.4,204.18) .. controls (841.66,255.92) and (809.82,334.39) .. (679,287.71) ;
\draw [shift={(677.02,287)}, rotate = 379.98] [fill={rgb, 255:red, 0; green, 0; blue, 0 }  ][line width=0.08]  [draw opacity=0] (10.72,-5.15) -- (0,0) -- (10.72,5.15) -- (7.12,0) -- cycle    ;

\draw (449,30) node [anchor=north west][inner sep=0.75pt]  [font=\Huge] [align=left] {{\fontfamily{ptm}\selectfont II.}};
\draw (71,31) node [anchor=north west][inner sep=0.75pt]  [font=\Huge] [align=left] {{\fontfamily{ptm}\selectfont I.}};
\draw (74,123) node [anchor=north west][inner sep=0.75pt]  [font=\Large]  {$1$};
\draw (206,23) node [anchor=north west][inner sep=0.75pt]  [font=\Large]  {$2$};
\draw (342,128) node [anchor=north west][inner sep=0.75pt]  [font=\Large]  {$3$};
\draw (287.02,291.82) node [anchor=north west][inner sep=0.75pt]  [font=\Large]  {$4$};
\draw (124,292) node [anchor=north west][inner sep=0.75pt]  [font=\Large]  {$5$};
\draw (868,219.18) node [anchor=north west][inner sep=0.75pt]  [font=\Large]  {$7$};
\draw (865,75.18) node [anchor=north west][inner sep=0.75pt]  [font=\Large]  {$6$};
\draw (517,294.18) node [anchor=north west][inner sep=0.75pt]  [font=\Large]  {$5$};
\draw (665.02,294) node [anchor=north west][inner sep=0.75pt]  [font=\Large]  {$4$};
\draw (728,124.18) node [anchor=north west][inner sep=0.75pt]  [font=\Large]  {$3$};
\draw (593,23) node [anchor=north west][inner sep=0.75pt]  [font=\Large]  {$2$};
\draw (459,122.18) node [anchor=north west][inner sep=0.75pt]  [font=\Large]  {$1$};

\end{tikzpicture}
}
\end{figure}
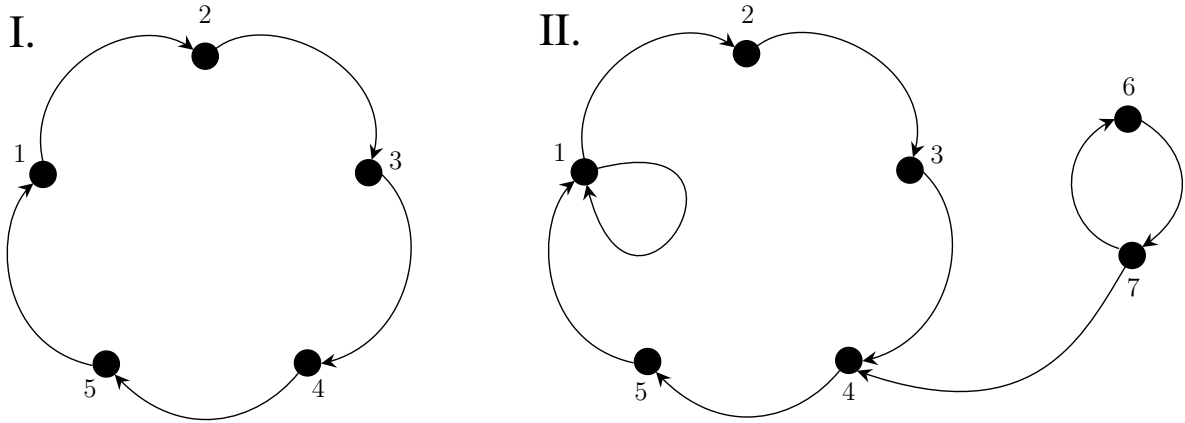

The next lemma shows a sufficient condition for $A$ to satisfy the condition (I).
\begin{lemma}\label{lemma:transitive_not_permutation_condition_I} Let $A$ be a $n \times n$ transition matrix. If $A$ is transitive and it is not a permutation matrix, then it satisfies the condition (I).
\end{lemma}

\begin{proof} Let $i \in S$. Since $A$ is not a permutation matrix, there exists $j \in S$ such that $A(j,k_1) = A(j,k_2) = 1$ for $k_1,k_2 \in S$ such that $k_1 \neq k_2$. By transitivity of the matrix, there exist the admissible words $w$ and $v$ of shortest length such that $w_0 = k_1$, $w_{|w|-1} = j$, $v_0 = k_2$ and $v_{|v|-1} = j$. Note that the unique letters of $w$ and $v$ equal to $j$ are the last ones. Then, the words $jw$ and $jv$ are admissible and distinct satisfying $(jw)_p \neq j$ for $0 < p < |w|$, and $(jv)_q \neq j$ for $0 < q < |v|$, and therefore $j \in \mathcal{S}$. Again by transitivity of $A$, there exists an admissible word $i_0...i_{r-1}$ such that $i_0 = i$ and $i_{r-1} = j \in \mathcal{S}$ and we conclude that $A$ satisfies the condition (I). 
\end{proof}

\begin{example} By Lemma \ref{lemma:transitive_not_permutation_condition_I} the following matrices satisfy the condition (I):
\begin{equation*}
 \begin{pmatrix}
    1 & 1 & 0 & 0\\
    0 & 0 & 1 & 0\\
    0 & 0 & 0 & 1\\
    1 & 0 & 0 & 0
 \end{pmatrix}, \quad
 \begin{pmatrix}
    1 & 1 & 1 & 1\\
    1 & 1 & 1 & 1\\
    1 & 1 & 1 & 1\\
    1 & 1 & 1 & 1
 \end{pmatrix}, \quad
  \begin{pmatrix}
    1 & 0 & 1 & 1\\
    1 & 1 & 0 & 1\\
    1 & 1 & 1 & 1\\
    1 & 0 & 1 & 0
 \end{pmatrix}.
\end{equation*}
\end{example}

The next example shows that the condition for the matrix in the previous lemma is not necessary to satisfy the condition (I).

\begin{example} The matrix
\begin{equation*}
\begin{pmatrix}
    1 & 1 & 0 & 0\\
    1 & 1 & 0 & 0\\
    0 & 0 & 1 & 1\\
    0 & 0 & 1 & 1
 \end{pmatrix}
\end{equation*}
is not transitive since there is not admissible word connecting the symbols $1$ and $3$. However, we have that $S = \mathcal{S}$. In fact, the table \ref{tab:condition_I_non_transitive_matrix} shows, for each $i \in S$, two admissible distinct words starting and ending in $i$, and that have not $i$ in the remaining positions.

\begin{table}[H]
\centering
\begin{tabular}{c|c|c|c|c}
  $i$ & $1$ & $2$ & $3$ & $4$\\ \hline
  words & $11$ & $22$ & $33$ & $44$\\
    & $121$ & $212$ & $343$ & $434$
\end{tabular}
\caption{Examples of distinct pairs of admissible words such that, for each symbol $i \in S$, that starts and ends with $i$ and such that their remaining letters are not $i$.\label{tab:condition_I_non_transitive_matrix}}
\end{table}
Then $\mathcal{S} = S$ and the validity of the condition (I) is straightforward.
\end{example}

Under the point of view of the symbolic dynamics, a transition matrix that every row and every column is non-zero satisfies the condition (I) if and only if the shift space $\Sigma_A$ has not isolated points. Also, $A$ does not satisfies the condition (I) if and only if there exist an admissible word $w$, $|w| \geq 1$ such that $P_w = Q_w$, where
\begin{equation*}
    P_w := S_w S_w^* \quad \text{and} \quad Q_w := S_w^* S_w,
\end{equation*}
with $S_w := S_{w_0} \cdots S_{w_{|w|-1}}$. The condition (I) is the classification criteria in order to verify if the uniqueness. In fact, the Uniqueness Theorem 2.13 of \cite{CK1980} states that, if the transition matrix satisfies the condition (I) and there two families of partial isometries, we say $\{S_i\}_{i=1}^n$ and $\{S_i'\}_{i=1}^n$, satisfying the relation Cuntz-Krieger relations \eqref{eq:CK}, then the mapping
\begin{equation*}
    S_i \mapsto S_i', \quad i = 1,...,n;
\end{equation*}
extends to an isomorphism between the algebras generated by those families. The next section presents the generalization of the Cuntz-Krieger algebras, the Exel-Laca algebras \cite{EL1999}. In this general context there is also a uniqueness theorem in the sense here presented, and we prove that these algebras are in fact generalizations of the Cuntz-Krieger algebras. Besides that, we show that the general uniqueness theorem is also a an result that extends the original theorem for the finite alphabet, and in this case we present now a definition and an equivalence for the condition (I) that we recall further, in the next section.  

\begin{definition}[Circuits in the symbolic graph] Given a transition matrix $A$ for a countable alphabet, a circuit is any finite admissible word on the symbolic graph of $A$, $x_0 x_1 \cdots x_{n-1}$, $n \in \mathbb{N}$, such that $A(x_{n-1},x_0) = 1$. Given a circuit $x_0 x_1 \cdots x_{n-1}$, we define the following:
\begin{itemize}
    \item we say $x_0 x_1 \cdots x_{n-1}$ has an exit when there exists a symbol $y \in S$, which we call exit symbol, such that $A(x_k,y) = 1$ for some $k \in \{0,...,n-1\}$ and $y \neq x_{k+1}$, where we set $x_n = x_0$. In the case of absence of an exit symbol, we say that $x_0 x_1 \cdots x_{n-1}$ is a terminal circuit;
    \item the reduced form of $x_0 x_1 \cdots x_{n-1}$ is its shortest subword $x_0 x_1 \cdots x_p$, $p \leq n-1$, such that $x_{p+1} \cdots x_{n-1} = x_0 \cdots x_{n-p-2}$.
\end{itemize}
\end{definition}

In the definition above, a reduced form of a circuit consists in the same path of the a priori circuit but `cutting off' its redundancies. For the rest of this section, we identify a circuit by its closed path in the symbolic graph, and this is realized by a natural equivalence relation: two circuits are equivalent if they have the same reduced form, up to a cyclic index permutation. Observe the reduced forms are also circuits and then they can be choosen as representatives of the equivalence classes. It is straightforward there exists a bijection between each one of these equivalence classes and the finite closed paths in the symbolic graph. 

Our last objective in this section is to show, for finite alphabet, that the condition (I) is equivalent to the condition (L) as follows:
\begin{itemize}
    \item[(L)] every circuit has an exit.
\end{itemize}
This result is Lemma 3.3 of \cite{KumjPaskRaeb1998} and we repeat the proof here. Part of the proof is the following lemma.

\begin{lemma}\label{lemma:A_no_zero_rows_no_zero_columns_implies_symbol_in_circuit_or_connecting_circuits} Suppose $A$ is a $n\times n$ transition matrix whitout zero rows and without zero columns, then for every $i \in S$ we have that $i$ is contained in a circuit or there are two circuits (not necessarily distincts) such that there exists an admissible word containing $i$ connecting them in the symbolic graph of $A$. 
\end{lemma}

\begin{proof} Since $A$ does not have zero rows neither zero columns, for every $i \in S$ there exists an bi-infinite admissible path in the graph of $A$ containing $i$:
\begin{equation*}
    \cdots \to b_3 \to b_2 \to b_1 \to i \to a_1 \to a_2 \to a_3 \to \cdots
\end{equation*}
Note that $|S| < \infty$, and hence there exists at least two symbols $p$ and $q$ such that
\begin{itemize}
    \item there are admissible paths connecting $q$ to $i$ and $i$ to $p$, that is, we may take the bi-infinite word above satisfying $a_r = p$ and $b_s = q$ for some $r,s \in \mathbb{N}$;
    \item $p$ repeats in forward the direction and $q$ repeats in the backward direction, that is, $r',s' \in \mathbb{N}$ s.t. $r'>r$, $s'>s$,  $a_{r'} = p$ and $b_{s'} = q$.
\end{itemize}
Now, we have two possibilities: there are $p$ and $q$ as above such that $p = q$, and then $i$ belongs to the circuit $b_s \cdots i \cdots a_r$; or for every $p$ and $q$ as above, we have $p \neq q$, and then, after choosing $p$ and $q$, the word $w^1 i w^2$ given by
\begin{align*}
    w^1= \begin{cases}
            \emptyset, \quad \text{if } b_1 = q,\\
            b_{s-1}\cdots b_1, \quad \text{otherwise};
          \end{cases}
            \quad \text{and} \quad
    w^2= \begin{cases}
            \emptyset, \quad \text{if } a_1 = p,\\
            a_1\cdots a_{r-1}, \quad \text{otherwise};
          \end{cases}
\end{align*}
where $r$ and $s$ are the smallest values as possible, is an admissible word connecting a circuit containing $q$ to a circuit containing $a$. 
\end{proof}

\begin{remark} Equivalently, the lemma above says that for every symbol $i \in S$, there exists a (possibly empty) admissible word $w$ s.t. $iw$ ends in a circuit.
\end{remark}

Now we prove the equivalence between (I) and (L).

\begin{proposition}\label{prop:equivalence_condition_I} Let $A$ be a $n \times n$ transition matrix. $A$ satisfies the condition (I) if and only if every circuit in the symbolic graph has an exit. 
\end{proposition}

\begin{proof} Suppose that $A$ satisfies the condition (I) and let $x_0 x_1 \cdots x_{n-1}$ be a circuit. Given $x_k$, $k \in \{0,...,n-1\}$, by condition (I), there exists an admissible path $i_0 \cdots i_{r-1}$, $r \geq 1$, s.t. $i_0 = x_k$ and $i_{r-1} \in \mathcal{S}$. And by definition of $\mathcal{S}$ there are two distinct circuits containing $i_{r-1}$, and then at least one of them is not a subpath of $x_0 x_1 \cdots x_{n-1}$. Therefore (L) holds. 

Conversely, suppose (L) and let $i \in S$. By Lemma \ref{lemma:A_no_zero_rows_no_zero_columns_implies_symbol_in_circuit_or_connecting_circuits} there exists an admissible path $w$, which can be empty, such that $(iw)_{|w|}$ belongs to a circuit. By hypothesis, there is a symbol $j$ in this circuit that connects to an exit. Observe that the circuit grants the existence of a path  $u$, $|u| \geq 22$, with $u_0 = u_{|u|-1} = j$ and $u_m \neq j$ for $0<m<|u|-1$, and s.t. it is a subpath in the circuit. Now, let $k \in S$ be the exit symbol. We have the following cases:
\begin{itemize}
    \item[$1.$] if there exists a word connecting $k$ to $iw$ or to any symbol of the circuit, then we have a path $v$, $|v| > 2$, and s.t. $v_0 = v_{|v|-1} = j$, $v_1 = k$ and $v_n \neq j$ for $0<n<|v|-1$. Certainly $u \neq v$, because the path $jk$ is $v$ and but it is not in $u$. Hence, $j \in \mathcal{S}$;
    \item[$2.$] if there is not a word connecting $k$ to $iw$ or to any symbol of the circuit, then consider the path $iwt$, where $t$ is a path that connects the path $w$ to the symbol $k$, then analyse the cases again at previous item and then this one.
\end{itemize}
Note that after each repetition of the items above we have less options of symbols that can satisfy the item $2.$ in the place of $k$. Since $|S| < \infty$, the process must terminate. Therefore $A$ satisifies (I). 
\end{proof}

\begin{remark} Observe that, for a general matrix with $S = \mathbb{N}$, the condition (L) is weaker than (I), which is also proved in Lemma 3.3 of \cite{KumjPaskRaeb1998}. In fact, the proof that (I) implies (L) does not need the hypothesis $|S| < \infty$. Moreover, consider the matrix given by
\begin{equation*}
    A(1,2) = A(2n,2n+2) = A(2n+1,2n-1) = 1, \quad n \in \mathbb{N},
\end{equation*}
and zero in the remaining entries. This matrix has not zero rows neither zero columns. Moreover, it satisfies condition (L), but condition (I) does not hold. Its symbolic graph is given by
\[
\begin{tikzcd}
\cdots\arrow[r]& \circled{7}\arrow[r]&\circled{5}\arrow[r]&\circled{3}\arrow[r]&\circled{1}\arrow[r]&\circled{2}\arrow[r]&\circled{4}\arrow[r]&\circled{6}\arrow[r]&\cdots
\end{tikzcd}
\]

\end{remark}

\section{Exel-Laca algebras}

One of the most natural questions about the Cuntz-Krieger algebras was the possibility to extend these algebras to an infinite alphabet. Note that for an infinite matrix, the series in \eqref{eq:CK} do not converge in general. However, at this point, we will ignore this problem in order to show a brief heuristic approach to the infinite countable alphabet case for the Cuntz-Krieger algebra as it was done for this one, and after that we present its formal definition, which is the Exel-Laca algebra. Later, we prove that the Exel-Laca algebras are in fact a generalization of the Cuntz-Krieger algebras, showing that their respective relations are equivalent for finite alphabets.

Given separable infinite dimensional Hilbert space $\mathcal{H}$ and $n \in \mathbb{N}$, $\mathcal{H}$ as
\begin{equation*}
    \mathcal{H} = \bigoplus_{i=1}^\infty \mathcal{H}_i,
\end{equation*}
where each $\mathcal{H}_i$ is also separable and infinite dimensional. Now consider a infinite transition matrix $A$ such that every row and every column is non-zero. For each $i \in \mathbb{N}$, we consider the infinite dimensional separable subspace 
\begin{equation*}
    \bigoplus_{j: A_{i,j}=1}\mathcal{H}_j
\end{equation*}
and we may choose for each $i$ an isometric isomorphism
\begin{equation*}
    S_i:\bigoplus_{j: A_{i,j}=1}\mathcal{H}_j \to \mathcal{H}_i.
\end{equation*}

\begin{example} If we take transition matrix of the renewal shift
\begin{equation*}
    A = \begin{pmatrix}
            1 & 1 & 1 & 1 & 1 & \cdots \\
            1 & 0 & 0 & 0 & 0 & \cdots \\
            0 & 1 & 0 & 0 & 0 & \cdots \\
            0 & 0 & 1 & 0 & 0 & \cdots \\
            0 & 0 & 0 & 1 & 0 & \cdots \\
            \vdots & \vdots & \vdots & \vdots & \vdots & \ddots \\
        \end{pmatrix},
\end{equation*}
then the isometric isomorphisms are
\begin{align*}
    S_1&: \bigoplus_{i=1}^\infty \mathcal{H}_i \to \mathcal{H}_1, \\ 
    S_p&: \mathcal{H}_{p-1} \to \mathcal{H}_p, \quad p \neq 1.
\end{align*}
\end{example}

By extending each $S_i$ to $\mathcal{H}$ as $S_i a = 0$ for $a \notin \bigoplus_{j: A_{i,j}=1}\mathcal{H}_j$, we obtain a family $\{S_i\}_{i \in \mathbb{N}}$ of partial isometries on $\mathcal{H}$, i.e.
\begin{equation*}
    S_iS_i^*S_i = S_i, \quad i= 1,..., n.
\end{equation*}

As we observed in the beginning of this section, it is not possible to write the relations \ref{eq:CK}. However, Exel and Laca \cite{EL1999} constructed the version of the Cuntz-Krieger relations for this general setting. In order to construct the Exel-Laca algebra, we need to introduce a unital universal C$^*$-algebra, which eventually coincides with the Exel-Laca algebra, and then we define the Exel-Laca algebra as it C$^*$-subalgebra. First, we define the following.
\begin{definition} Let $A$ be a transition matrix on an countable alphabet $S$. For every $X,Y \subset \mathbb{N}$ finite sets and $j \in \mathbb{N}$ we set
\begin{equation*}
    A(X,Y,j):= \prod_{x \in X} A(x,j) \prod_{y \in Y} (1-A(y,j)).
\end{equation*} 
\end{definition}
Now prove the following.

\begin{theorem}\label{thm:EL_are_universal_algebras} Consider a set of generators $D=\{1\}\sqcup\{S_i\}_{i \in S}$, $S$ countable satisfying the following family $\mathscr{R}$ of relations on $D$:
\begin{itemize}
    \item $(1-1^*,0)$, $(1-1^2,0)$, $(1S_i-S_i1,0)$ and $(S_i-S_iS_i^*S_i,0)$ for every $i \in S$;
    \item $(S_i^*S_iS_j^*S_j - S_j^*S_jS_i^*S_i,0)$ for every $i,j \in S$;
    \item $(S_i^*S_j, 0)$ whenever $i \neq j$;
    \item $\left(\left(\prod_{x \in X} S_x^*S_x\right) \left(\prod_{y \in Y} (1-S_y^*S_y)\right) - \sum_{j \in \mathbb{N}} A(X,Y,j)S_j S_j^*,0\right)$ for every pair $X,Y$ of finite subsets of $S$ such that $A(X,Y,j)$ is non-zero only for a finite number of $j$'s.
\end{itemize}
Then, the pair $(D,\mathscr{R})$ is admissible and therefore the unversal C$^*$-algebra C$^*(D,\mathscr{R})$ exists.
\end{theorem}

\begin{proof} The proof for $S_i$ is exactly the same as in Theorem \ref{thm:CK_are_universal_algebras}, and for every representation $\Theta_0$ for the pair $(D,\mathscr{R})$, we have that $\Theta_0(1) = \Theta_0(1)^* = \Theta_0(1)^2$ and then $\Theta(1)$ is a a projection, so $\|\Theta_0(1)\| \leq 1$. We conclude that $(D,\mathscr{R})$ is an admissible pair and therefore the C$^*$-algebra C$^*(D,\mathscr{R})$ exists.
\end{proof}

\begin{remark} It is important to observe that, given an admissible pair $(D,\mathscr{R})$, and two elements $F,G \in \widetilde{\mathcal{A}}_D$ such that $(F-G,0) \in \mathscr{R}$ means that $F=G$. Also, to affirm that there exists $1 \in D$ such that the relations $(1-1^*,0)$, $(1-1^2,0)$ and $(1d-d1,0)$ for every $\mathfrak{e} \in \mathscr{E}$ are in $\mathscr{R}$, is equivalent to say that the universal C$^*$-algebra C$^*(D,\mathscr{R})$ is unital.
\end{remark}

\begin{definition}[Exel-Laca algebra]\label{def:EL_algebra} Given a transition matrix $A$, we define the universal unital C$^*$-algebra $\widetilde{\mathcal{O}}_A$ generated by a family of partial isometries $\{S_j:j \in S\}$, and that satisfies the relations of $\mathscr{R}$ in Theorem \ref{thm:EL_are_universal_algebras}, that is,
\begin{itemize}
    \item[(EL1)] $S_i^*S_i$ and $S_j^*S_j$ commute for every $i,j \in S$;
    \item[(EL2)] $S_i^*S_j = 0$ whenever $i \neq j$;
    \item[(EL3)] $(S_i^*S_i)S_j = A(i,j)S_j$ for all $i,j \in S$;
    \item[(EL4)] for every pair $X,Y$ of finite subsets of $\mathbb{N}$ such that the quantity $A(X,Y,j)$ is non-zero only for a finite number of $j$'s we have
        \begin{equation*}
            \left(\prod_{x \in X} S_x^*S_x\right) \left(\prod_{y \in Y} (1-S_y^*S_y)\right) = \sum_{j \in \mathbb{N}} A(X,Y,j)S_j S_j^*.
        \end{equation*}
\end{itemize}
Also, the relations (EL1)-(EL4) are called the Exel-Laca relations. The Exel-Laca algebra $\mathcal{O}_A$ is the universal C$^*$-subalgebra of $\widetilde{\mathcal{O}}_A$ generated by the family $\{S_j:j \in S\}$.
\end{definition} 

\begin{remark} By the definition above, we observe that there are only two possibilities: $\widetilde{\mathcal{O}}_A \simeq \mathcal{O}_A$ or $\widetilde{\mathcal{O}}_A$ is the canonical unitization of $\mathcal{O}_A$. Further in this section, we present precise conditions that makes these algebras $*$-isomorphic. 
\end{remark}

For the particular case $|S| < \infty$, the universal algebra $\widetilde{\mathcal{O}}_A$ is the Cuntz-Krieger algebra, as we show next.

\begin{proposition}\label{prop:equivalence_CK_EL} Let $A$ be a transition $n \times n$ matrix for the (finite) alphabet $S$. Then, $\widetilde{\mathcal{O}}_A$ is isomorphic to the Cuntz-Krieger algebra $\mathcal{O}_A$. 
\end{proposition}

\begin{proof} We claim that the Exel-Laca relations are equivalent to the Cuntz-Krieger relations for the finite alphabet case. Indeed, let $\mathscr{S} = \{S_k\}_{k=1}^n$ be a family of partial isometries. Suppose $\mathscr{S}$ satisfies the Exel-Laca relations. By the finiteness of the alphabet we have that $A(X,Y,j)$ is always non-zero for a finite quantity of $j$'s. Also note that $A(\emptyset,\emptyset,j)=1$ for every $j$, so (EL4) gives
\begin{equation*}
    \sum_{i = 1}^n S_i S_i^* = \sum_{i = 1}^n A(\emptyset,\emptyset,i) S_i S_i^* = \left(\prod_{x \in \emptyset} S_x^*S_x\right)\left(\prod_{k \in \emptyset}(1- S_y^*S_y)\right) = 1 \cdot 1 = 1,
\end{equation*}
and the validity of (CK1) is proved. Now, observe that for every $i,j \in S$ we have $A(\{i\},\emptyset,j)=A(i,j)$. Again from (EL4), it follows that
\begin{equation*}
    \sum_{j = 1}^n A(i,j)S_j S_j^* = \sum_{j = 1}^n A(\{i\},\emptyset,j)S_j S_j^*  = \left(\prod_{x \in \{i\}} S_x^*S_x\right)\left(\prod_{y \in \emptyset}(1- S_y^*S_y)\right) = S_i^*S_i,
\end{equation*}
and (CK2) is also proved. Conversely, suppose now that $\mathscr{S}$ satisfies the Cuntz-Krieger relations. The items $(ii)$ and $(iii)$ of Lemma \ref{lemma:CKprop} are precisely the relations (EL2) and (EL3), and then these relations are automatically valid. For (EL1), we use (CK2), Lemma \ref{lemma:CKprop} $(i)$ and the fact that $S_kS_k^*$ is a projection for every $k \in S$ to obtain
\begin{align*}
   S_i^*S_iS_j^*S_j &= \left(\sum_{k=1}^n A(i,k) S_k S_k^*\right) \left(\sum_{\ell=1}^n A(j,\ell) S_\ell S_\ell^* \right) = \sum_{k,\ell=1}^n A(i,k)A(j,\ell) S_k S_k^*  S_\ell S_\ell^* \\
   &= \sum_{k=1}^n A(i,k)A(j,k) S_k S_k^*  S_k S_k^* = \sum_{k=1}^n A(i,k)A(j,k) S_k S_k^*,
\end{align*}
that is
\begin{equation}\label{eq:prod_Q_i_Q_j}
   S_i^*S_iS_j^*S_j = \sum_{k=1}^n A(i,k)A(j,k) S_k S_k^*,
\end{equation}
and, since the RHS of \eqref{eq:prod_Q_i_Q_j} does not change by exchanging the positions of $i$ and $j$, we conclude that $S_i^*S_iS_j^*S_j = S_j^*S_j S_i^*S_i$, that is, (EL1) is satisfied. By applying the result above for a product of projections $S_x^*S_x$, $x \in X \subseteq \{1,...,n\}$, we have
\begin{equation}\label{eq:prod_sisi}
    \prod_{x \in X} S_x^*S_x = \sum_{k=1}^n \left(\prod_{x \in X} A(x,k)\right)S_kS_k^*.
\end{equation}
On other hand, if we subtract (CK1) from (CK2) we obtain
\begin{equation*}
    1-S_i^*S_i = \sum_{j=1}^n(1-A(i,j))S_jS_j^*, \quad i \in S,
\end{equation*}
and hence
\begin{align*}
    \prod_{y \in Y} (1-S_y^*S_y) &= \prod_{y \in Y} \left( \sum_{l=1}^n(1-A(y,\ell))S_lS_l^*\right)\stackrel{(\ddagger)}{=}  \sum_{l=1}^n \left( \prod_{y \in Y}(1-A(y,\ell))S_\ell S_\ell^* \right)\\
    &\stackrel{(\bullet)}{=}  \sum_{l=1}^n \left( \prod_{y \in Y}(1-A(y,\ell))\right)S_\ell S_\ell^*, \quad \text{for every } Y \subseteq \{1,...,n\},
\end{align*}
where $(\ddagger)$ and $(\bullet)$ are exactly the same arguments as used in \eqref{eq:prod_Q_i_Q_j}. Then,
\begin{equation}\label{eq:prod_1_minus_sisi}
    \prod_{y \in Y} (1-S_y^*S_y) = \sum_{l=1}^n \left( \prod_{y \in Y}(1-A(y,\ell))\right)S_\ell S_\ell^*.
\end{equation}
We get
\begin{align*}
    \left(\prod_{i \in X} S_i^*S_i\right) \left(\prod_{j \in Y} (1-S_j^*S_j)\right) &\stackrel{(\bullet \bullet)}{=} \left(\sum_{k=1}^n \left(\prod_{i \in X} A(i,k)\right)S_kS_k^*\right) \left(\sum_{l=1}^n \left( \prod_{j \in Y}(1-A(j,l))\right)S_lS_l^*\right)\\
    &=\sum_{k,l=1}^n \left(\prod_{i \in X} A(i,k)\right)  \left( \prod_{j \in Y}(1-A(j,l))\right)S_kS_k^*S_lS_l^* \\
    &\stackrel{(\ddagger)}{=}\sum_{k=1}^N \left(\prod_{i \in X} A(i,k)\right)  \left( \prod_{j \in Y}(1-A(j,k))\right)S_kS_k^*S_kS_k^* \\
    &\stackrel{(\bullet)}{=}\sum_{k=1}^N \left(\prod_{i \in X} A(i,k)\right)  \left( \prod_{j \in Y}(1-A(j,k))\right)S_kS_k^* \\
    &=\sum_{k=1}^N A(X,Y,k)S_kS_k^*,
\end{align*}
where in $(\bullet \bullet)$ we used the identities \eqref{eq:prod_sisi} and \eqref{eq:prod_1_minus_sisi}, in $(\ddagger)$ we used Lemma \ref{lemma:CKprop} $(i)$ as previously in the proof. Also again, in $(\bullet)$ we used that $S_jS_j^*$ is a projection. Then,
\begin{equation*}
    \left(\prod_{i \in X} S_i^*S_i\right) \left(\prod_{j \in Y} (1-S_j^*S_j)\right) = \sum_{k=1}^N A(X,Y,k)S_kS_k^*,
\end{equation*}
which is precisely the relation (EL4). By the uniqueness of the universal C$^*$-algebras, Theorem \ref{thm:uniqueness_universal_algebra}, we conclude that $\widetilde{\mathcal{O}}_A$ and the Cuntz-Krieger algebras are isomorphic.
\end{proof}

\begin{corollary} Let $A$ be a transition $n \times n$ matrix for the (finite) alphabet $S$. Then, the Exel-Laca algebra is isomorphic to the Cuntz-Krieger algebra. Moreover, in this case, the Exel-Laca algebra is unital.
\end{corollary}

\begin{proof} By Proposition \ref{prop:equivalence_CK_EL} we have that the Cuntz-Krieger algebra is isomorphic to $\widetilde{\mathcal{O}}_A$. In particular, the Cuntz-Krieger algebra is generated by the family of partial isometries, and this family also generates the Exel-Laca algebra by definition. We conclude that these three C$^*$-algebras are isomorphic. Consequently, the Exel-Laca algebra is unital for finite alphabet.
\end{proof}

Now we state the version of the Uniqueness Theorem of the Cuntz-Krieger algebra, in the sense of the equivalence between faithful representations, for the Exel-Laca algebras, which is Corollary 13.2 of \cite{EL1999}.

\begin{theorem}[Uniqueness Theorem for Exel-Laca algebras]\label{thm:uniqueness_thm_EL_algebras} Let $A$ be a transition matrix on an alphabet $S$ with no identically zero rows and suppose that its symbolic graph has no terminal circuits. Let $\{S_i\}_{i\in S}$ and $\{T_i\}_{i\in S}$ be two families of non-zero partial isometries on Hilbert spaces, both of them satisfying the Exel-Laca relations. Then the C$^*$-algebras generated by $\{S_i\}_{i\in S}$ and $\{T_i\}_{i\in S}$ are isomorphic to each other under an isomorphism $\Psi$ such that $\Psi(S_i) = T_i$ for all $i \in S$.
\end{theorem}

The next result grants that transitivity is a sufficient condition for a transition matrix on a infininte alphabet to satisfy the hipotheses of Theorem \ref{thm:uniqueness_thm_EL_algebras}.

\begin{proposition}Suppose that $A$ is a transitive transition matrix on a countably infinite alphabet. Then $A$ has not indentically zero rows and every circuit is not terminal.
\end{proposition} 

\begin{proof} Let $x_0\cdots x_{n-1}$, $n \in \mathbb{N}$, be a circuit. Since the alphabet is infininte, there exists $z \in \mathbb{N}$ such that $z \neq x_p$ for every $p \in \{0,...,n-1\}$. By the transitivity there exists an admissible word $w$, $|w| \geq 2$, satisfying $w_0 = x_0$ and $w_{|w|-1} = z$. Since $z$ is not a symbol in the word $x_0\cdots x_{n-1}$, there exists the smallest $k \in \{0,...,|w|-1\}$ such that $w_k$ is not a symbol in the word $x_0\cdots x_{n-1}$ and therefore $w_k$ is an exit for the circuit $x_0\cdots x_{n-1}$. Therefore the symbolic graph of $A$ has not terminal circuits. 
\end{proof}

We assume for the rest of the thesis the following standing hypothesis.

\begin{mdframed} \textbf{Standing hypothesis:} any transition matrix $A$ has only non-zero rows and only non-zero columns, and it is transitive.
\end{mdframed}

\subsection{Representation of the EL algebras in $\mathfrak{B}(\ell^2(\Sigma_A))$}

We present now a very special faithful representation of $\widetilde{\mathcal{O}}_A$ that connects the Exel-Laca algebras with the countable Markov shifts.

\begin{definition}\label{def:representation_EL} Given a transition matrix $A$ its Markov shift space $\Sigma_A$, we consider the (non-separable) Hilbert space $\mathfrak{B}(\ell^2(\Sigma_A))$ and its canonical basis $\{\delta_x\}_{x \in \Sigma_A}$, given by
\begin{equation*}
    (\delta_x)_y = \begin{cases}
                        1 \text{ if }x=y,\\
                        0 \text{ otherwise}.
                    \end{cases}
\end{equation*}
Let $\pi: \widetilde{\mathcal{O}}_A \to \mathfrak{B}(l^2(\Sigma_A))$ be a representation, defined by $\pi(S_j):= T_j$, where
\begin{equation*}
    T_s(\delta_x) = \begin{cases}
                        \delta_{sx} \text{ if } A(s,x_0)=1,\\
                        0 \text{ otherwise};
                    \end{cases} \text{with} \quad
    T_s^*(\delta_x)=\begin{cases}
            \delta_{\sigma(x)} \text{ if } x \in [s],\\
            0 \text{ otherwise}.
        \end{cases}.
\end{equation*}
We also define the projections $P_s:=T_sT_s^*$ and $Q_s:=T_s^*T_s$, given by
\begin{equation*}
    P_s(\delta_\omega)=\begin{cases}
                            \delta_\omega \text{ if } \omega \in [s], \\
                            0 \text{ otherwise;} 
                        \end{cases}  \text{and} \quad
    Q_s(\delta_\omega)=\begin{cases}
                        \delta_\omega \text{ if } \omega \in \sigma([s]),\\
                            0 \text{ otherwise.}
                        \end{cases}
\end{equation*}
\end{definition}

By Proposition $9.1$ of \cite{EL1999}, the representation above is the unique one such that $\pi(S_j):= T_j$. Also, by the standing hypothesis of transitivity, we have that $A$ does not have terminal circuits and then $\pi$ is faithful (see Proposition 12.2 in \cite{EL1999}). The next proposition characterizes $\widetilde{\mathcal{O}}_A$ in terms of the representation of Definition \ref{def:representation_EL}.

\begin{proposition}\label{prop:O_A_closure_span} $\widetilde{\mathcal{O}}_A$ is isomorphic to the closure of the linear span of the terms $T_\alpha \left(\prod_{i \in F}Q_i\right)T_\beta^*$, where $\alpha$ and $\beta$ are admissible finite words or the empty word and $F \subseteq S$ is finite.
\end{proposition}

\begin{proof} We recall that $\widetilde{\mathcal{O}}_A \simeq C^*(\{T_i: i \in S\}\cup\{1\})$. First, we will prove that 
\begin{equation}\label{eq:spanO_A}
    \spann\left\{T_\alpha \left(\prod_{i \in F}Q_i\right)T_\beta^*:F \text{ finite}; \text{ }\alpha, \beta \text{ finite admissible words, including empty words} \right\}
\end{equation}
is a $*$-algebra. Indeed, the vector space properties are trivially satisfied, as well as the closeness of the involution. For the algebra product, take two generators in $\widetilde{\mathcal{O}}_A$, $T_\alpha \left(\prod_{i \in F}Q_i\right)T_\beta^*$ and $T_{\alpha'} \left(\prod_{j \in F'}Q_i\right)T_{\beta'}^*$ like in \eqref{eq:spanO_A}, with $\beta = \beta_1 \cdots \beta_n$ and $\alpha' = \alpha_1'\cdots \alpha_m'$; $n,m \in \mathbb{N}$. We wish that the product
\begin{equation}\label{eq:prod_span}
    T_\alpha \left(\prod_{i \in F}Q_i\right)T_\beta^*T_{\alpha'} \left(\prod_{j \in F'}Q_i\right)T_{\beta'}^*
\end{equation}
can be written as a linear combination of terms like the generators of \eqref{eq:spanO_A} and hence we need to study the term $T_\beta^*T_{\alpha'}$. From the axiom $(EL3)$ for the Cuntz-Krieger algebra for infinite matrices
we have that
\begin{equation} \label{eq:Q_iT_j}
    Q_i T_j = A(i,j) T_j,
\end{equation}
and consequently
\begin{equation} \label{eq:Q_iT_j_star}
    T_j^* Q_i = A(i,j) T_j^*.
\end{equation}
We have three cases to analyze as follows.

\begin{itemize}
    \item[$(a)$] If $n=m$, then by the axiom $(EL2)$ and \eqref{eq:Q_iT_j} we get
        \begin{equation*}
            T_\beta^*T_{\alpha'} = T_{\beta_n}^* \cdots T_{\beta_2}^* \delta_{\beta_1,\alpha'_1}Q_{\beta_1} T_{\alpha'_2} \cdots T_{\alpha'_n}
            = \delta_{\beta_1,\alpha'_1} T_{\beta_n}^* \cdots T_{\beta_2}^*  T_{\alpha'_2} \cdots T_{\alpha'_n} = \cdots = \delta_{\beta,\alpha'} Q_{\beta_n},
        \end{equation*}
    where $\delta_{\beta,\alpha'}$ is the Kronecker delta. So,
    \begin{equation*}
    T_\alpha \left(\prod_{i \in F}Q_i\right)T_\beta^*T_{\alpha'} \left(\prod_{j \in F'}Q_j\right)T_{\beta'}^*= 
                            \delta_{\beta,\alpha'} T_\alpha \left(\prod_{i \in F}Q_i\right)Q \left(\prod_{j \in F'}Q_j\right)T_{\beta'}^*
    \end{equation*}
    where $Q = Q_{\beta_n}$ if $n>0$ and $Q = 1$ otherwise. We conclude that the product above belongs to \eqref{eq:spanO_A} in this case;
    \item[$(b)$] if $n>m$, by similar calculations done in the earlier case using \eqref{eq:Q_iT_j_star} instead of \eqref{eq:Q_iT_j} and defining $\overline{\beta}:=\beta_1\cdots \beta_m$ we obtain $T_\beta^*T_{\alpha'} = \delta_{\overline{\beta},\alpha'} T_{\beta_n}^* \cdots T_{\beta_{m+1}}^*$. By using \eqref{eq:Q_iT_j_star} several but finite times on the term $T_{\beta_{m+1}}^*\left(\prod_{j \in F'}Q_j\right)$ we have that
    \begin{equation*}
    T_\alpha \left(\prod_{i \in F}Q_i\right)T_\beta^*T_{\alpha'} \left(\prod_{j \in F'}Q_j\right)T_{\beta'}^*= \delta_{\overline{\beta},\alpha'}\left(\prod_{j \in F'}A(j,\beta_{m+1})\right) T_\alpha \left(\prod_{i \in F}Q_i\right) T_{\beta'\beta_{m+1}\cdots \beta_n}^*.
    \end{equation*}
    We conclude that the product above also belongs to \eqref{eq:spanO_A};
    \item[$(c)$] for $n<m$ the proof is similar to the previous item by using the \eqref{eq:Q_iT_j} instead of \eqref{eq:Q_iT_j_star}.
\end{itemize}
We conclude that \eqref{eq:spanO_A} is a $*$-subalgebra of the $C^*$-algebra $\widetilde{\mathcal{O}}_A$, and hence 
\begin{equation}\label{eq:spanO_A_closed}
    B=\overline{\spann\left\{T_\alpha \left(\prod_{i \in F}Q_i\right)T_\beta^*:F \text{ finite}; \text{ }\alpha, \beta \text{ finite admissible words} \right\}}
\end{equation}
is a $C^*$-subalgebra of $\widetilde{\mathcal{O}}_A$. On other hand, if we take $F= \emptyset$, $\alpha = s$, $s\in S$ and $\beta$ the empty sequence, then we conclude that $T_s \in B$ for all $s \in S$. Also, if we take $F= \emptyset$ and $\alpha = \beta$ empty sequence, it follows that $1$ belongs to \eqref{eq:spanO_A_closed}. Since $B$ is a $C^*$-subalgebra of $\widetilde{\mathcal{O}}_A$ which contains its generators, we have that $\widetilde{\mathcal{O}}_A \subseteq B$ and therefore the result follows.
\end{proof}    

\begin{remark}\label{remark:O_A_non_unital} By similar proof as in Proposition \ref{prop:O_A_closure_span}, it is straightforward to verify that 
\begin{equation*}
    \mathcal{O}_A \simeq \overline{\spann\left\{ \begin{array}{l l}
         & F \text{ finite}; \text{ }\alpha, \beta \text{ finite admissible words};\\
         T_\alpha \left(\prod_{i \in F}Q_i\right)T_\beta^*: &F \neq \emptyset \text{ or } \alpha \text{ is not the empty word} \\
         &\text{or }  \beta \text{ is not an empty word}
    \end{array}\right\} }.
\end{equation*}
\end{remark}

\begin{remark}\label{remark:TQT_on_H} For an element $T_\alpha \left(\prod_{i \in F}Q_i\right)T_\beta^* \in \widetilde{\mathcal{O}}_A$ as in the Proposition \ref{prop:O_A_closure_span}, we have that
\begin{equation*}
    T_\alpha \left(\prod_{i \in F}Q_i\right)T_\beta^*(\delta_x) = \begin{cases}
                                                                    \delta_{\alpha \sigma^{|\beta|}x}, \quad \text{if } x \in [\beta] \text{ and } A(i,x_{|\beta|})=1 \text{ for every } i \in F,\\
                                                                    0, \quad \text{otherwise};
                                                                  \end{cases}
\end{equation*}
for every $x \in \Sigma_A$.
\end{remark}

Now we define the commutative C$^*$-subalgebras of $\widetilde{\mathcal{O}}_A$ such that their spectra are the generalized Markov shifts. 

\begin{definition} Let $\widetilde{\mathcal{D}}_A$ be the commutative unital $C^*$-subalgebra of $\widetilde{\mathcal{O}}_A$ given by
\begin{equation*}
    \widetilde{\mathcal{D}}_A:= \overline{\spann\left\{T_\alpha \prod_{i \in F}Q_i T_\alpha^*: F \text{ finite}; \alpha \text{ finite word} \right\}}
\end{equation*}
and denote by $\mathcal{D}_A$ its C$^*$-subalgebra defined by
\begin{equation*}
    \mathcal{D}_A:= \overline{\spann\left\{T_\alpha \prod_{i \in F}Q_i T_\alpha^*: F \text{ finite}; \alpha \text{ finite word}; F \neq \emptyset \text{ or } \alpha \text{ is not the empty word}\right\}}.
\end{equation*}
\end{definition}

The proof of Proposition \ref{prop:O_A_closure_span} for $\alpha = \beta$ and $\alpha' = \beta'$ shows that $\widetilde{\mathcal{D}}_A$ is a unital $C^*$-subalgebra of $\widetilde{\mathcal{O}}_A$. Moreover, by easy calculations we note that $\widetilde{\mathcal{D}}_A$ is in fact commutative.

\begin{remark}\label{remark:D_A_subalgebra_on_diagonal_operators} As a particular case of Remark \ref{remark:TQT_on_H} we have that 
\begin{equation*}
    T_\alpha \left(\prod_{i \in F}Q_i\right)T_\alpha^*(\delta_x) = \begin{cases}
                                                                    \delta_{x}, \quad \text{if } x \in [\alpha] \text{ and } A(i,x_{|\alpha|})=1 \text{ for every } i \in F,\\
                                                                    0, \quad \text{otherwise};
                                                                  \end{cases}
\end{equation*}
for every $x \in \Sigma_A$. Then $T_\alpha \left(\prod_{i \in F}Q_i\right)T_\alpha^*$ is a diagonal operator and therefore both $\widetilde{\mathcal{D}}_A$ and $\mathcal{D}_A$ are C$^*$-subalgebras of diagonal operators of $\ell^2(\Sigma_A)$.
\end{remark}

Now, we will obtain a more suitable set of generators for $\widetilde{\mathcal{D}}_A$ which will allow to see its spectrum as a set of configurations on the Cayley tree. We present the notion of partial representation of a group.

\begin{definition}[Partial representations] Given a group $G$, denote its identity element by $e$, and consider a Hilbert space $\mathcal{H}$. A partial representation of $G$ on $\mathcal{H}$ is a map $u:G \to \mathfrak{B}(\mathcal{H})$ such that, for every $g,h \in G$, it satisfies
\begin{itemize}
    \item[(PR1)] $u(g)u(h)u(h^{-1}) = u(gh)u(h^{-1})$;
    \item[(PR2)] $u(g^{-1}) = u(g)^*$;
    \item[(PR3)] $u(e) = 1$,
\end{itemize}
\end{definition}

\begin{remark} For every partial representation $u:G \to \mathfrak{B}(\mathcal{H})$ defined as above, it follows that 
\begin{equation*}
    u(g^{-1})u(g)u(h) = u(g^{-1})u(gh).
\end{equation*}
Indeed, by exchanging $g$ and $h^{-1}$ in (PR1), we get 
\begin{equation*}
    u(h^{-1})u(g^{-1})u(g) = u(h^{-1}g^{-1})u(g).
\end{equation*}
And by applying the involution in both sides of the equality above, it follows that
\begin{align*}
    u(g)^*u(g^{-1})^*u(h^{-1})^* = u(g)^*u(h^{-1}g^{-1})^*
\end{align*}
and by (PR2) we obtain
\begin{equation*}
    u(g^{-1})u(g)u(h) = u(g^{-1})u(gh).
\end{equation*}
\end{remark}

\begin{proposition}\label{prop:e_g_is_projection} Given a partial representation $u$ of a group $G$ on a Hilbert space $\mathcal{H}$, we have that $u(g)$ is a partial isometry for every $g \in G$. The elements in $\mathfrak{B}(\mathcal{H})$ defined by
\begin{equation*}
    e(g):= u(g)u(g)^*, \quad g \in G
\end{equation*}
are projections.
\end{proposition}

\begin{proof} For every $g \in G$, we have that
\begin{equation*}
    u(g)u(g)^*u(g) \stackrel{(PR2)}{=} u(g)u(g^{-1})u(g) \stackrel{(PR1)}{=} u(gg^{-1})u(g) = u(e)u(g) \stackrel{(PR3)}{=} u(g),
\end{equation*}
where in first two equalities above we used, respectively, the properties (PR2) and (PR1). In the last equality above we used (PR3). Therefore, $u(g)$ is a partial isometry, and therefore $e(g)$ is a projection. \end{proof}


\begin{definition} Let $\mathbb{F}$ be a free group. We say that a partial representation of $\mathbb{F}$ on any Hilbert space is
\begin{itemize}
    \item[$(a)$] semi-satured if $u(gh) = u(g)u(h)$ whenever $|gh|=|g|+|h|$, for $gh$ reduced;
    \item[$(b)$] orthogonal if $u(i)^*u(j) = 0$ when $i,j \in S$ and $i \neq j$, where $S$ is the generator set of $\mathbb{F}$.
\end{itemize}
\end{definition}

Consider the free group $\mathbb{F}_S$ generated by the alphabet $S$ and let the map

\begin{align*}
    T:\mathbb{F}_S &\to \widetilde{\mathcal{O}}_A\\
    s &\mapsto T_s \\
    s^{-1} &\mapsto T_{s^{-1}}:= T_s^*.
\end{align*}
Also, for any word $g$ in $\mathbb{F}_S$, take its reduced form $g=x_1...x_n$ and define that $T$ realizes the mapping
\begin{equation*}
    g \mapsto T_g:= T_{x_1} \cdots T_{x_n},
\end{equation*}
and that $T_e=1$. We are imposing conditions only on the reduced words in order to make $T$ well-defined. For example, note that $e=i^{-1}i$ for any $i \in S$. If we would include non reduced words we would get $1 = T_e = Q_i$ for all $i \in S$, which is not true. We prove that $T$ defines a partial action that is semi-satured and orthogonal, and the next lemma will be used for it. From now, we will denote $\mathbb{F}_S$ by $\mathbb{F}$. In addition, denote by $\mathbb{F}_+$ the positive cone of $\mathbb{F}$, i.e., the unital sub-semigroup of $\mathbb{F}$ generated by $S$. The set $\mathbb{F}_+^{-1}$ is the collection of the inverse elements of $\mathbb{F}_+$, that is, the elements $g \in \mathbb{F}$ such that $g^{-1} \in \mathbb{F}_+$.  

\begin{lemma}\label{lemma:g_non_zero} For any $g \in \mathbb{F}$ reduced which is not in the form $\alpha\beta^{-1}$, with $\alpha, \beta \in \mathbb{F}_+$, (including the cases $g=\mu$, $g=\nu^{-1}$ and $g \neq e$), it follows that $T_g = 0$.
\end{lemma}

\begin{proof} For a given $g$ as in the statement of the lemma, its reduced form is $g=g_0\cdots g_p$, $p > 1$, and there exists $i \in \{0,...,p-1\}$ such that $g_i = s^{-1}$ and $g_{i+1}=t$, $s,t \in S$. Hence,
\begin{equation*}
    T_g = T_{g_0} \cdots T_{s^{-1}} T_t \cdots T_{g_p} = T_{g_0} \cdots T_{s}^* T_t \cdots T_{g_p} \stackrel{(\clubsuit)}{=} T_{g_0} \cdots 0 \cdots T_{g_p}=0,
\end{equation*}
where in $(\clubsuit)$ we used (EL2).
\end{proof}

\begin{proposition}[Proposition 3.2 of \cite{EL1999}] \label{prop:T_partial_rep} The map $T$ constructed as above is a partial representation of $\mathbb{F}$ on $\mathcal{O}_A$ that is semi-satured and orthogonal.
\end{proposition}

\begin{proof} (PR2) and (PR3) are straightforward as well as the semi-saturation property. The orthogonality comes directly from (EL2). We divide the proof of (PR1) in claims.

\textbf{Claim 1:} For $n \in \mathbb{N}$, $x_1,...,x_n \in S$ and $\alpha = x_1 \cdots x_n$, it follows that $T_\alpha^*T_\alpha = \Lambda T_{x_n}^*T_{x_n}$, where
\begin{equation*}
    \Lambda = \prod_{k=1}^{n-1} A(x_k,x_{k+1}).
\end{equation*}

\textbf{Proof of the claim 1:} the result is obvious for $n=1$. By induction, assume that the result follows for $n-1$, $n \geq 2$, i.e., for $\beta = x_1 \cdots x_{n-1}$ we have $T_\beta^*T_\beta = \Lambda' T_{x_{n-1}}^*T_{x_{n-1}}$, where 

\begin{equation*}
    \Lambda' = \prod_{k=1}^{n-2} A(x_k,x_{k+1}).
\end{equation*}
We have that 

\begin{align*}
    T_\alpha^* T_\alpha &= T_{\beta x_n}^* T_{\beta x_n} = T_{x_n}^*T_{\beta}^* T_{\beta}T_{x_n} \stackrel{(\spadesuit)}{=} \Lambda' T_{x_n}^* T_{x_{n-1}}^*T_{x_{n-1}}T_{x_n} \\
    &\stackrel{(\diamondsuit)}{=} \Lambda' A(x_{n-1},x_n) T_{x_n}^* T_{x_n} = \Lambda T_{x_n}^* T_{x_n},
\end{align*}
where in $(\spadesuit)$ we used the induction hypothesis and in $(\diamondsuit)$ we used (EL3). This claim is proved. 

Note that the claim above also shows that for all $\alpha \in \mathbb{F}_+$ it follows that $T_\alpha^*T_\alpha$ is idempotent because $\Lambda \in \{0,1\}$, and hence it is a partial isometry.

\textbf{Claim 2:} If $\alpha, \beta \in \mathbb{F}_+$ with $|\alpha| = |\beta|$ but $\alpha \neq \beta$, then $T_\alpha^*T_\beta = 0$.

\textbf{Proof of the claim 2:} let $m= |\alpha| = |\beta|$. For $m=1$, we have that $(EL2)$ implies $T_\alpha^*T_\beta =0$. For $m>1$, write $\alpha = \alpha'x$ and $\beta = \beta'y$, with $\alpha',\beta' \in \mathbb{F}_+$ and $x,y \in S$. It follows that 

\begin{equation*}
    T_\alpha^*T_\beta = T_{\alpha'x}^*T_{\beta'y} = T_x^* T_{\alpha'}^*T_{\beta'}T_y.
\end{equation*}
Assume by induction that the result follows, i.e., $T_\mu^*T_\nu = 0$ for $\mu, \nu \in \mathbb{F}_+$ such that $|\mu|=|\nu|=m-1$ and $\mu \neq \nu$. Now, suppose that for $m$ we have $T_\alpha^*T_\beta \neq 0$. By the induction hypothesis, we conclude that $\alpha' = \beta'$, and the claim 1 implies that
\begin{equation*}
    T_{\alpha'}^*T_{\beta'} = T_{\alpha'}^*T_{\alpha'} = \Lambda T_{z}^*T_{z}
\end{equation*}
for any $z \in S$ and $\Lambda \in \{0,1\}$. We conclude that
\begin{equation*}
    0 \neq T_\alpha^*T_\beta = \Lambda T_xT_{z}^*T_{z}T_y^* = \Lambda A(z,y) T_x T_y^*,
\end{equation*}
and by the case for $m=1$ we conclude that $x=y$ and therefore $\alpha=\beta$, leading to a contradiction. The claim 2 is proved.

\textbf{Claim 3:} For all $\alpha, \beta \in \mathbb{F}_+$, if $T_\alpha^*T_\beta \neq 0$, then $\alpha^{-1} \beta \in \mathbb{F}_+\cup \mathbb{F}_+^{-1}$.

\textbf{Proof of the claim 3:} w.l.o.g. assume $|\alpha|\leq |\beta|$. Then we can write $\beta = \beta' \gamma$ such that $|\beta'| = |\alpha|$, and $\beta', \gamma \in \mathbb{F}_+$, then
\begin{equation*}
    0 \neq T_\alpha^*T_\beta = T_\alpha^* T_{\beta'} T_\gamma \implies T_\alpha^* T_{\beta'} \neq 0.
\end{equation*}
By the claim 2 we have that $\alpha = \beta'$, hence $\alpha^{-1} \beta = \gamma \in \mathbb{F}_+ \subseteq \mathbb{F}_+\cup \mathbb{F}_+^{-1}$. Note that for $|\alpha| \geq |\beta|$ we obtain $\alpha^{-1} \beta \in \mathbb{F}_+^{-1}$. The claim 3 is proved.

Let $g \in \mathbb{F}$ and consider $e(g) = T_g T_g^*$. Observe that $T_\alpha$ is a partial isometry for $\alpha \in \mathbb{F}_+$. Indeed,
\begin{equation*}
    T_\alpha T_\alpha^* T_\alpha \stackrel{(\dagger)}{=} \Lambda T_\alpha T_{\alpha_{|\alpha|-1}}^* T_{\alpha_{|\alpha|-1}} = \Lambda T_\alpha,
\end{equation*}
where in $(\dagger)$ we used the claim 1 and in the last equality is justified by the fact that $T_{\alpha_{|\alpha|-1}}$ is a partial isometry, since $\Lambda \in \{0,1\}$. Consequently, $e(\alpha)$ is a projection and hence it is self-adjoint and idempotent.

\textbf{Claim 4:} For all $\alpha, \beta \in \mathbb{F}_+$ the operators $e(\alpha)$ and $e(\beta)$ commute.

\textbf{Proof of the claim 4:} the proof is divided in two cases:

\begin{itemize}
    \item[$(i)$] $\alpha^{-1}\beta \notin \mathbb{F}_+\cup \mathbb{F}_+^{-1}$. In this case we have by the claim 3 that
    \begin{equation*}
        e(\alpha)e(\beta) = T_{\alpha}T_{\alpha}^*T_{\beta}T_{\beta}^* = 0,
    \end{equation*}
    and $e(\beta)e(\alpha)=0$ by the similar argument.
    \item[$(ii)$] $\alpha^{-1}\beta \in \mathbb{F}_+\cup \mathbb{F}_+^{-1}$. W.l.o.g. write $\alpha^{-1}\beta = \gamma \in \mathbb{F}_+$. We have that
    \begin{align*}
        e(\alpha)e(\beta) &= T_{\alpha}T_{\alpha}^*T_{\alpha\gamma}T_{\alpha\gamma}^* = T_{\alpha}T_{\alpha}^*T_{\alpha}T_{\gamma}T_{\gamma}^*T_\alpha^* \stackrel{(*)}{=} T_{\alpha}T_{\gamma}T_{\gamma}^*T_\alpha^* \\
        &\stackrel{(**)}{=} T_{\alpha}T_{\gamma}T_{\gamma}^*T_\alpha^*T_{\alpha}T_{\alpha}^* = T_{\beta}T_{\beta}^*T_{\alpha}T_\alpha^* = e(\beta)e(\alpha).
    \end{align*}
    For the equalities $(*)$ and $(**)$ we used the fact that $T_\alpha$ is a partial isometry. The proof is analoguous when $\gamma \in \mathbb{F}_+^{-1}$. 
\end{itemize}
The claim 4 is proved.

\textbf{Claim 5:} For every $x \in S$ and $\alpha \in \mathbb{F}_+$, the operators $Q_x$ and $T_\alpha QT_\alpha^*$ commute, where $Q$ is either the identity operator or the initial projection $Q_i$ of one of the isometries $T_i$.

\textbf{Proof of the claim 5:} the case $|\alpha|=0$ is a direct consequence of (EL1). For $|\alpha|>0$, write $\alpha = y \alpha'$ with $y \in S$ and $\alpha' \in \mathbb{F}_+$ and note that (EL3) gives
\begin{align*}
    Q_x T_\alpha = Q_x T_y T_{\alpha'} = A(x,y) T_{\alpha'},
\end{align*}
and by applying the involution on the identity above we get $T_\alpha^* Q_x = A(x,y) T_{\alpha'}^*$. Therefore we obtain
\begin{equation*}
    Q_x T_\alpha QT_\alpha^* = A(x,y) T_{\alpha'} QT_\alpha^*= T_\alpha QT_\alpha^* Q_x,
\end{equation*}
and the claim is proved. 

\textbf{Claim 6:} For every $g,h \in \mathbb{F}$ the operators $e(g)$ and $e(h)$ commute.

\textbf{Proof of the claim 6:} by Lemma \ref{lemma:g_non_zero}, we may assume that $g = \mu \nu^{-1}$ and $h=\alpha\beta^{-1}$, with $\mu, \nu, \alpha, \beta$ being positive. Also, there are no problems if we assume that $|g| = |\mu|+|\nu|$ and $|h| = |\alpha|+|\beta|$. It follows that
\begin{equation*}
    e(g) = T_\mu T_\nu^*T_\nu T_\mu^* \quad \text{and} \quad e(h) = T_\alpha T_\beta^*T_\beta T_\alpha^*.
\end{equation*}
Observe that $e(g)e(h) = e(h)e(g)=0$ if $T_\mu^* T_\alpha = 0$, so suppose from now that $T_\mu^* T_\alpha \neq 0$, which implies that $\mu^{-1}\alpha \in \mathbb{F}_+\cup \mathbb{F}_+^{-1}$, by the claim 3. So w.l.o.g. we may assume that $\mu^{-1}\alpha = \gamma \in \mathbb{F}_+$.

Now, we use the notation $Q=T_\nu^*T_\nu$. By the claim 1, if $|\nu|>0$ then $Q$ is zero or equal to some $Q_x$, $x \in S$. On the other hand, if $\nu=e$ then $Q=1$. The same is valid for $Q' = T_\beta^* T_\beta$ and $Q''=T_\mu^*T_\mu$. We have that
\begin{align*}
    e(g)e(h) &= T_\mu T_\nu^*T_\nu T_\mu^*T_{\mu\gamma}T_\beta^*T_\beta T_{\mu\gamma}^* = T_\mu T_\nu^*T_\nu T_\mu^*T_{\mu}T_\gamma T_\beta^*T_\beta T_\gamma^*T_{\mu}^* = T_\mu Q Q'' T_\gamma Q' T_\gamma^*T_{\mu}^*\\
    &\stackrel{(\heartsuit)}{=} T_\mu T_\gamma Q' T_\gamma^*Q Q''T_{\mu}^* \stackrel{(\ddagger)}{=} T_\mu T_\gamma Q' T_\gamma^*Q''QT_{\mu}^* = T_\mu T_\gamma T_\beta^* T_\beta T_\gamma^*T_\mu^*T_\mu T_\nu^*T_\nu T_{\mu}^*\\ 
    &= T_{\mu\gamma} T_\beta^* T_\beta T_{\mu\gamma}^*T_\mu T_\nu^*T_\nu T_{\mu}^* = e(h)e(g),
\end{align*}
where in $(\heartsuit)$ we used the claim 5 and in $(\ddagger)$ we used (EL1). The claim 6 is proved. 

Finally, we prove (PR1), i.e.,
\begin{equation*}
    T_g T_h T_{h^{-1}} = T_{gh}T_{h^{-1}}, \quad \forall g,h \in \mathbb{F},
\end{equation*}
by induction in $|g|+|h|$. The result is obvious for $|g| =0$ or $|h| =0$. Otherwise, write $g = g'x$ and $h=yh'$, where $x,y \in S\cup S^{-1}$, $|g|=|g'|+1$ and $|h|=|h'|+1$. Suppose as induction hypothesis that the result follows for all $|g|+|h| \leq n-1$, $n>2$. If $|g|+|h|=n$, we analyze the only two following possible cases:

\begin{itemize}
    \item[$(a)$] $x \neq y^{-1}$. In this case we have that $|gh|=|g|+|h|$, and by the orthogonality it follows that $T_{gh}=T_gT_h$, and hence
    \begin{equation*}
        T_g T_h T_{h^{-1}} = T_{gh}T_{h^{-1}};
    \end{equation*}
    \item[$(b)$] $x = y^{-1}$. In this case we have that
    \begin{align*}
        T_g T_h T_{h^{-1}} &= T_{g'x} T_{x^{-1}h'} T_{(h')^{-1}x} = T_{g'}T_x T_{x}^*T_{h'} T_{h'}^* T_{x} = T_{g'}e(x) e(h') T_{x} \stackrel{(\clubsuit)}{=} T_{g'} e(h') e(x) T_{x} \\
        &= T_{g'} T_{h'} T_{h'}^*T_x T_{x}^* T_{x} \stackrel{(\star)}{=} T_{g'} T_{h'} T_{h'}^*T_x \stackrel{(\star \star)}{=} T_{g'h'} T_{h'}^*T_x = T_{gh} T_{h^{-1}},
    \end{align*}
    where in $(\clubsuit)$ we used the claim 6, in $(\star)$ the fact that $T_x$ is a partial isometry, and in $(\star \star)$ the induction hypothesis.
\end{itemize}
\end{proof}

From now on, we denote the projections $e(g)$ by $e_g$.

\begin{proposition}\label{prop:D_A_isomorphic_e_g} $\widetilde{\mathcal{D}}_A \simeq C^*(\{e_g:g \in \mathbb{F}\})$.
\end{proposition}

\begin{proof} The main idea of the proof is to show that the faithful representation of the $C^*$-algebra $\widetilde{\mathcal{D}}_A$ in $\mathfrak{B}(\ell^2(\Sigma_A))$ coincides with the $C^*$-subalgebra $\mathfrak{U} = C^*(\{e_g:g \in \mathbb{F}\})$ contained in $\mathfrak{B}(\ell^2(\Sigma_A))$, which implies that they are isomorphic. We will show that the terms $T_\alpha \left(\prod_{i \in F} Q_i\right) T_\alpha^*$ can be written as terms in $\mathfrak{U}$ and conversely that the terms $e_g$ can be written as terms in $\widetilde{\mathcal{D}}_A$. 

Let $g \in \mathbb{F}$. W.l.o.g. we may assume that $T_g \neq 0$. By Lemma \ref{lemma:g_non_zero} we have that $g = \alpha \beta^{-1}$ such that $\alpha, \beta \in \mathbb{F}_+$, with $\alpha = \alpha_0 \cdots \alpha_t$ and $\beta = \beta_0 \cdots \beta_u$ for the respective cases that $\alpha$ and $\beta$ are not $e$. Assume that $g$ is already its reduced form, i.e., $g = \alpha \beta^{-1}, \alpha, \beta^{-1}$ or $e$. By the axiom $(EL3)$ we have
\begin{align*}
    e_g &= T_{\alpha} T_{\beta}^* T_{\beta} T_{\alpha}^* = T_{\alpha} T_{\beta_u}^*\cdots T_{\beta_1}^* Q_{\beta_0} T_{\beta_1} \cdots T_{\beta_{u}} T_{\alpha}^* 
    = T_{\alpha} T_{\beta_{u}}^*\cdots T_{\beta_1}^* T_{\beta_1} T_{\beta_{u}} T_{\alpha}^*
    = \cdots \\
    &= T_{\alpha} T_{\beta_{u}}^* T_{\beta_{u}} T_{\alpha}^* = T_{\alpha} Q_{\beta_u} T_{\alpha}^* \in \widetilde{\mathcal{D}}_A,
\end{align*}
and we conclude that $\mathfrak{U} \subseteq \widetilde{\mathcal{D}}_A$. The result above is similar for $\alpha =e$ or $\beta = e$. For the opposite inclusion, let $\alpha \in \mathbb{F}_+$ admissible or $\alpha = e$ in its reduced form, and $F \subseteq \mathbb{N}$ finite. If $\alpha = e$, we have that 
\begin{equation*}
    T_\alpha \left(\prod_{i \in F} Q_i \right) T_\alpha^* = \prod_{i \in F} Q_i = \prod_{i \in F} e_{i^{-1}} \in  \mathfrak{U}.
\end{equation*}
On other hand, if $F = \emptyset$ and $\alpha \neq e$ is an admissible word, we have:
\begin{equation*}
    T_\alpha \left(\prod_{i \in F} Q_i \right) T_\alpha^* = T_\alpha  T_\alpha^* =  e_{\alpha} \in \mathfrak{U}.
\end{equation*}
Now, suppose that $\alpha = \alpha_0 \cdots \alpha_{t} \neq e$ reduced and $F \neq \emptyset$. We will prove that 
\begin{equation}\label{eq:prod_e_alpha}
    T_\alpha \left(\prod_{j \in F} Q_j \right) T_\alpha^* = e_\alpha \prod_{\substack{j \in F \\ j \neq \alpha_{t}}} e_{\alpha j^{-1}}
\end{equation}
by induction in $|F|$. If $|F| = 1$ we have that
\begin{align*}
    T_\alpha \left(\prod_{j \in F} Q_j \right) T_\alpha^* = T_\alpha T_{i^{-1}} T_i T_{\alpha^{-1}},
\end{align*}
where $i \in F$. If $\alpha_t = i$, since $e_i$ is a projection we get
\begin{equation*}
    T_\alpha T_{i^{-1}} T_i T_{\alpha^{-1}} = T_{\alpha'} T_i T_{i^{-1}} T_i T_{i^{-1}} T_{(\alpha')^{-1}} = T_{\alpha'}  T_i T_{i^{-1}} T_{(\alpha')^{-1}} = e_{\alpha},
\end{equation*}
where $\alpha' = e$ if $t=0$ and $\alpha' = \alpha_0 \cdots \alpha_{t - 1}$ if $t>0$. On other hand, if $\alpha_t \neq i$, it follows that
\begin{equation*}
    T_\alpha T_{i^{-1}} T_i T_{\alpha^{-1}} = T_{\alpha i^{-1}}  T_{(\alpha i^{-1})^{-1}} = e_{\alpha i^{-1}},
\end{equation*}
and it is easy to use (PR1) to verify that $e_\alpha e_{\alpha i^{-1}} = e_{\alpha i^{-1}}$. So, anyway we have that \eqref{eq:prod_e_alpha} is true for $|F|=1$. Now, suppose the validity of \eqref{eq:prod_e_alpha} for $|F| = n-1$, $n>1$. For $|F| = n$, fix $k \in F$. One can use (PR1) and the claim 1 of the proof of Proposition \ref{prop:T_partial_rep} in order to obtain
\begin{align*}
    T_\alpha \left(\prod_{i \in F} Q_i \right) T_\alpha^* &= T_\alpha T_\alpha^* T_\alpha \left(\prod_{i \in F} Q_i \right) T_\alpha^* =   T_\alpha Q_{\alpha_t} \left(\prod_{i \in F} Q_i \right) T_\alpha^* = T_\alpha Q_k Q_{\alpha_t} \left(\prod_{i \in F\setminus\{k\}} Q_i \right) T_\alpha^*\\ 
    &= T_\alpha Q_k T_\alpha^* T_\alpha \left(\prod_{i \in F\setminus\{k\}} Q_i \right) T_\alpha^* =  e_\alpha \left(\prod_{\substack{j \in \{k\}\\ j \neq \alpha_t}} e_{\alpha j^{-1}}\right) e_\alpha \left(\prod_{\substack{j \in F\setminus\{k\}\\ j \neq \alpha_t}} e_{\alpha j^{-1}}\right), 
\end{align*}
where in the last equality we used \eqref{eq:prod_e_alpha} for $|F|=1$ and the induction hypothesis. Since the $e_g$'s commute and they are projections, we conclude that
\begin{align*}
    T_\alpha \left(\prod_{i \in F} Q_i \right) T_\alpha^* &= e_\alpha \left(\prod_{\substack{j \in F\\ j \neq \alpha_{|\alpha|-1}}} e_{\alpha j^{-1}}\right),
\end{align*}
as we wished to prove. The direct consequence of the results above is that $T_\alpha \left(\prod_{i \in F} Q_i \right) T_\alpha^* \in \mathfrak{U}$, for all $\alpha$ admissible finite word and for every $F \subseteq S$ finite. Then,
\begin{equation*}
    \spann \left\{T_\alpha\left(\prod_{i \in F}Q_i\right)T_\alpha^*: \alpha\text{ admissible }, 0\leq|\alpha|< \infty ,0\leq |F|< \infty \right\} \subseteq \mathfrak{U},
\end{equation*}
and since $\mathfrak{U}$ is a C$^*$-algebra we conclude that the closure of the left hand side of the relation above is still contained in $\mathfrak{U}$, i.e., $\widetilde{\mathcal{D}}_A \subseteq \mathfrak{U}$. The proof is complete. 
\end{proof}

Now we have all the necessary background to introduce and study the space $X_A$, and this is the main object studied in the next section.

\section{Generalized Countable Markov shifts}

The construction of the commutative C$^*$-algebras $\mathcal{D}_A$ and $\widetilde{\mathcal{D}}_A$ leads to the construction of the Generalized Markov shift space through the Gelfand Representation Theorem for C$^*$-algebras presented in Corollary \ref{corollary:Gelfand_C_star}, as the spectrum of these algebras, and it is defined next.

\begin{definition}[Generalized Markov shift space]\label{def:X_A} Given an irreducible transition matrix $A$ on the alphabet $\mathbb{N}$, the generalized Markov shift spaces are the sets
\begin{equation*}
    X_A := \text{spec}\,\mathcal{D}_A \quad \text{and} \quad \widetilde{X}_A := \text{spec}\,\widetilde{\mathcal{D}}_A,
\end{equation*}
both endowed by the weak$^*$ topology.
\end{definition}

\begin{remark}\label{remark:unital_vs_full_row_of_ones} We remind the reader that when when $\mathcal{O}_A$ is unital, then $\mathcal{D}_A$ is also unital. Consequently, $\widetilde{\mathcal{O}}_A = \mathcal{O}_A$ and then $\widetilde{\mathcal{D}}_A = \mathcal{D}_A$ and then $\widetilde{X}_A = X_A$. Besides that, $X_A$ is locally compact and $\widetilde{X}_A$ is always compact. In particular, for every matrix with a full row of $1$'s gives $X_A$ compact. In fact, if there exists a symbol $j \in \mathbb{N}$ such that $A(j,n) = 1$ for every $n \in \mathbb{N}$, then $Q_j = T_j^* T_j = 1 \in \mathcal{D}_A$ and therefore $X_A$ is compact. 
\end{remark}

\begin{remark}\label{remark:subalgebra_character} Given a commutative algebra $B$ and $J$ a closed self-adjoint two-sided ideal of $B$, then the set of characters of $J$ is $\widehat{J} = \{\varphi \in \widehat{B}: \varphi\restriction_J \neq 0\}$. Therefore,
\begin{equation*}
    \widehat{\mathcal{D}}_A = \{\varphi\restriction_{\mathcal{D}_A}:\varphi \in \widehat{\widetilde{D}}_A, \quad \varphi\restriction_{\mathcal{D}_A} \neq 0\} = X_A.
\end{equation*}
It is straightforward to conclude that $X_A \subseteq \widetilde{X}_A$.
\end{remark}

\begin{proposition}\label{prop:X_A_tilde_X_A} $X_A = \widetilde{X}_A \setminus \{\varphi_0\}$, where $\varphi_0$ is the character in $\widetilde{X}_A$ given by
\begin{equation*}
    \varphi_0(e_g) := \begin{cases}
                        1, \quad \text{if } g=e;\\
                        0, \quad \text{otherwise.}
                      \end{cases}
\end{equation*}
\end{proposition}

\begin{proof} Suppose that $\varphi_0 \in X_A$. Then for $g \neq e$ we have
\begin{equation*}
    \varphi_0(T_gT_g^*) = \varphi_0(e_g) = 0 \implies \varphi_0\restriction_{\mathcal{D}_A} = 0,
\end{equation*}
which is not a character for the algebra $\mathcal{D}_A$. By Remark \ref{remark:subalgebra_character} we have that $X_A \subseteq \widetilde{X}_A\setminus \{\varphi_0\}$. Conversely, for given $\varphi \in \widetilde{X}_A$ such that $\varphi \neq \varphi_0$, by Proposition \ref{prop:D_A_isomorphic_e_g} there exists $g \neq e$ that we have $\varphi(e_g)=1$. Then $\varphi\restriction_{\mathcal{D}_A} \neq 0$. Then  $\varphi\restriction_{\mathcal{D}_A}$ is a character for $\varphi\restriction_{\mathcal{D}_A}$ by Remark \ref{remark:subalgebra_character}, hence $\widetilde{X}_A\setminus \{\varphi_0\} \subseteq X_A$.
\end{proof}

\begin{remark} Observe that if $\mathcal{D}_A$ is not unital, then $\widetilde{X}_A$ is the Alexandrov compactification of $X_A$, where the compactification point of $X_A$ is the character $\varphi_0$. 
\end{remark}

The next result connects equivalent conditions to $\mathcal{O}_A$ be unital.

\begin{theorem}[Theorem 8.5 of \cite{EL1999}]\label{thm:O_A_unital} The following are equivalent:
\begin{itemize}
    \item[$(i)$] $\mathcal{O}_A=\mathcal{\widetilde{O}}_A$;
    \item[$(ii)$] $\mathcal{O_A}$ is unital;
    \item[$(iii)$] $\varphi_0 \notin \widetilde{X}_A$;
    \item[$(iv)$] On the space $\{0,1\}^S$ (column space of the matrix $A$, endowed with the product topology), the null vector is not an accumulation point\footnote{We recall that null vector is an accumulation point of the column space if for every neighborhood $V$ of $\{0\}_{s \in S}$ there exists an infinite set of $j$'s such that the column $c_j$ is in $V$.} of the columns of $A$
    \item[$(v)$] There is $Y\subseteq S$ such that $A(\emptyset,Y,j)$ has finite support on $j$
\end{itemize}
\end{theorem}

At this point, the name \emph{Generalized Markov shift space} must be justified: the Markov shift space is not only included in $X_A$ but also is dense subset, and moreover, $\Sigma_A$ coincides with $X_A$ when it is locally compact, that is, when $A$ is row-finite. All these facts are presented and proven next.

\begin{definition}\label{def:i_1} Given $A$ an irreducible matrix and its corresponding Markov shift space $\Sigma_A$, we define the map $i_1: \Sigma_A \xhookrightarrow{} X_A$, $\Sigma_A \ni \omega \mapsto \varphi_{\omega}\restriction_{\mathcal{D}_A} \in X_A,$ where $\varphi_{\omega}$ is the evaluation map
\begin{equation}\label{eq:varphi_omega_definition}
    \varphi_\omega(R)= (R\delta_\omega,\delta_\omega), \quad R \in \mathfrak{B}(l^2(\Sigma_A)).
\end{equation}
will be denoted by $i_1$ and the context will let evident if the codomain is $X_A$ or $\widetilde{X}_A$.
\end{definition}

\begin{lemma}\label{lemma:i_1_inj_cont} The map $i_1$ is an injective continuous function.
\end{lemma}

\begin{proof} The injectivity is straightforward: let $\omega,\eta \in \Sigma_A$ such that $\varphi_\omega = \varphi_\eta$. If $\omega \neq \eta$, then there exists $n \in \mathbb{N}$ such that $\omega_0...\omega_{n-1} \neq \eta_0...\eta_{n-1}$. By taking $R = T_{\omega_0...\omega_{n-1}}T_{\omega_0...\omega_{n-1}}^* \in \mathcal{D}_A$ we have that
\begin{align*}
    \varphi_\omega(T_{\omega_0...\omega_{n-1}}T_{\omega_0...\omega_{n-1}}^*) &= (T_{\omega_0...\omega_{n-1}}T_{\omega_0...\omega_{n-1}}^* \delta_\omega, \delta_\omega) = 1 \text{ and}\\
    \varphi_\eta(T_{\omega_0...\omega_{n-1}}T_{\omega_0...\omega_{n-1}}^*) &= (T_{\omega_0...\omega_{n-1}}T_{\omega_0...\omega_{n-1}}^* \delta_\eta, \delta_\eta) = 0,
\end{align*}
and hence we have that $\varphi_\omega \neq \varphi_\eta$, a contradiction, and we conclude that $i_1$ is injective. Now we prove its continuity. It is sufficient to prove that $\varphi_{\omega^n}(R)\to \varphi_\omega(R)$ for every $R$ in the form $R=T_\alpha(\prod_{i\in F}Q_i)T_\alpha^*$, since these are generators of $\mathcal{D}_A$ and the elements of the spectrum ar $*$-homomorphisms\footnote{Notice that the characters are continuous functions because they are $*$-homomorphisms between $C^*$-algebras.}. Let $(\omega^n)_{n \in \mathbb{N}}$ be a sequence in $\Sigma_A$ converging to $\omega \in \Sigma_A$. Take $R=T_\alpha(\prod_{i\in F}Q_i)T_\alpha^*$ an arbitrary generator of $\mathcal{D}_A$. We have that
\begin{align*}
    \varphi_{\omega_n}(R)&=
    \begin{cases}
        1,\quad \text{if $\omega^n\in [\alpha]$ and $A(i,\omega^n_{|\alpha|})=1 \forall i\in F$}, \\
        0,\quad \text{otherwise};
    \end{cases}\\
    \varphi_{\omega}(R)&=
    \begin{cases}
        1,\quad \text{if $\omega\in [\alpha]$ and $A(i,\omega_{|\alpha|})=1\forall i\in F$}, \\
        0,\quad \text{otherwise}.
    \end{cases}
\end{align*}
For large enough $m \in \mathbb{N}$, we have that $\omega_0^n...\omega^n_{|\alpha|}=\omega_0...\omega_{|\alpha|}$ for every $n > m$, and then we have $\varphi_{\omega^n}(R)=\varphi_{\omega}(R)$. Since $F$ and $\alpha$ are arbitrary, we have for every generator $R \in \mathcal{D}_A$ that
\begin{equation*}
    \varphi_{\omega^n}(R)\to \varphi_\omega (R),
\end{equation*}
that is,  $\varphi_{\omega^n}\stackrel{w^*}{\to} \varphi_\omega$. Therefore $i_1$ is continuous.
\end{proof}

The lemma above shows that in fact $\Sigma_A$ is included in $X_A$. Besides that, the inclusion is continuous. From now on, we omit the notations of restriction to $\widetilde{\mathcal{D}}_A$ and $\mathcal{D}_A$ on $\varphi_\omega$. The next result we prove that $\Sigma_A$, seen as a subset of $X_A$ (or $\widetilde{X}_A$), is weak$^*$-dense.

\begin{proposition}\label{prop:i_1_Sigma_A_dense} $i_1(\Sigma_A)$ is dense in $X_A$ and in $\widetilde{X}_A$ (when considered).
\end{proposition}

\begin{proof} We claim that the elements of $i_1(\Sigma_A)$ separate points in $\mathcal{D}_A$. Suppose that there exists $R \in \mathcal{D}_A$ such that $\varphi_\omega (R) = 0$ for every $\omega \in \Sigma_A$. By Remark \ref{remark:D_A_subalgebra_on_diagonal_operators}, $R$ is a diagonal operator. Suppose that $\varphi_\omega(R) = 0$ for every $\omega \in \Sigma_A$, then
\begin{equation*}
    \varphi_\omega(R) = (R \delta_\omega, \delta_\omega) = 0,
\end{equation*}
then $\spann\{\delta_\omega:\omega \in \Sigma_A\} \subseteq \ker R$, and therefore $\overline{\spann\{\delta_\omega:\omega \in \Sigma_A\}} \subseteq \ker R$ since $\ker R$ is closed. We conclude that $R =0$ and the claim is proved. By Theorem \ref{thm:dense_character_general} it follows that $i_1(\Sigma_A) = \{\varphi_\omega\}_{\omega \in \Sigma_A}$ is dense in $X_A$. Since $X_A$ is dense in $\widetilde{X}_A$, we have that $i_1(\Sigma_A)$ is also dense on $\widetilde{X}_A$. 
\end{proof}

Now, we prove that if $A$ is row-finite, then $\Sigma_A = X_A$, that is, $i_1$ is a surjection. 

\begin{proposition}\label{prop:Sigma_A_locally_compact_coincides_with_X_A} If $\Sigma_A$ is locally compact, then $i_1$ is surjective. 
\end{proposition}

\begin{proof} Let $\varphi \in X_A$. By Proposition \ref{prop:i_1_Sigma_A_dense}, there exists a sequence $(\varphi_{\omega^n})_{n\in \mathbb{N}}$ in $i_1(\Sigma_A)$, where $\omega^n \in \Sigma_A$ for every $n \in \mathbb{N}$, such that $\varphi_{\omega^n}\stackrel{w^*}{\to} \varphi$. Since $\varphi$ is a character, we have necessarily that $\varphi \neq 0$, and by Proposition \ref{prop:D_A_isomorphic_e_g}, there exists $g \in \mathbb{F}$ such that $\varphi(e_g) \neq 0$. W.l.o.g. we may assume $g = \alpha$ positive admissible word\footnote{It is straightforward from the proof of Proposition \ref{prop:D_A_isomorphic_e_g}.}, and Proposition \ref{prop:e_g_is_projection} implies that $\varphi(e_\alpha) = \varphi(e_\alpha^2)$ and hence $\varphi(e_\alpha) = 1$. Then there exists $N \in \mathbb{N}$ such that for every $n > N$ we have
\begin{equation*}
    \varphi_{\omega^n}(e_\alpha) = 1,
\end{equation*}
that is,
\begin{equation*}
    (T_\alpha T_\alpha^* \delta_{\omega^n},\delta_{\omega^n} ) = 1,
\end{equation*}
and hence for every $n >N$ we have that $\omega^n \in [\alpha]$. On the other hand, since $\Sigma_A$ is locally compact, Proposition \ref{prop:local_compactness_Markov} grants that $[\alpha]$ is compact, and so the sequence $(\varphi_{\omega^n})_{n\in \mathbb{N}}$ has a convergent subsequence $(\varphi_{\omega^{n_p}})_{p\in \mathbb{N}}$ in $[\alpha]$, converging to some $\omega \in \Sigma_A$. By the continuity of $i_1$ proved in Lemma \ref{lemma:i_1_inj_cont} we have that $\lim_p \varphi_{\omega^{n_p}} = \varphi_\omega$
\begin{equation*}
    \varphi = \lim_n \varphi_{\omega^n} = \lim_p \varphi_{\omega^{n_p}} = \varphi_\omega
\end{equation*}
and therefore $i_1$ is surjective.
\end{proof}

\begin{remark} At this point due to the density of $i_1(\Sigma_A)$ on $X_A$, one could ask if the inclusion $i_1$ has some familiar topological feature. When $\Sigma_A$ is compact, then $i_1$ is simply a homeomorphism between compact metric spaces. If $\Sigma_A$ is non-compact and locally compact, one could ask if $i_1$ would be the compactification of $\Sigma_A$ with a Martin Boundary similarly as in O. Shwartz' paper \cite{Shwartz2019}, which is a construction for locally compact spaces. In this case the answer is `no', since the Martin Boundary adds extra points to the space. Futhermore, it cannot be any compactification, since in this case $X_A = \Sigma_A$ and therefore it is not compact. Maybe the most interesting case is $\Sigma_A$ not locally compact and $X_A$ compact, as it happens, for instance, when $A$ has a full row of $1$'s, as presented in Remark \ref{remark:unital_vs_full_row_of_ones}. In this case we have that $i_1$ is in fact a compactification of $\Sigma_A$, and certainly it is not the Stone-C\v{e}ch compactification \cite{Engelking1989} since both $\Sigma_A$ and $X_A$ are metric spaces and $\Sigma_A$ is not compact.
\end{remark}

Proposition \ref{prop:D_A_isomorphic_e_g} gives us a easier way to see $X_A$ (repect. $\widetilde{X}_A$). Given $\varphi \in X_A$ or $\widetilde{X}_A$, we can determine its image completely simply by taking its values on the generators $(e_g)_{g \in \mathbb{F}}$. Since $e_g$ is idempotent for any $g$, it follows that $\varphi(e_g) \in \{0,1\}$. By endowing $\{0,1\}^{\mathbb{F}}$ is endowed the product topology of the discrete topology in $\{0,1\}$, we introduce now the mapping that realizes the characters of the generalized Markov shift spaces as configurations on the Cayley tree generated by $\mathbb{F}$.

\begin{definition} Let $A$ be a transition matrix and consider its respective generalized Markov shift space $X_A$. Define the map $i_2: X_A \to \{0,1\}^{\mathbb{F}}$ (repect. for $\widetilde{X}_A$) given by $i_2(\varphi) := \xi$, where
\begin{align*}
    \xi_g = \pi_g(\xi) := \varphi(e_g), \quad g \in \mathbb{F}.
\end{align*}
An element $\xi \in \{0,1\}^\mathbb{F}$ is called a configuration. We say that a configuration is filled in $g \in \mathbb{F}$ when $\xi_g = 1$.  
\end{definition}

\begin{remark} It is important to notice the difference between how the configurations of $\{0,1\}^\mathbb{F}$ are presented here and in \cite{EL1999}. Here, for a given configuration $\xi$ and a word $g \in \mathbb{F}$ we will use $\xi_g = 1$ instead of $g \in \xi$ as used in \cite{EL1999}. Our choice of notation is motivated by the Markov shift notation for sequences in $\Sigma_A$.
\end{remark}

The next proposition shows that the topological properties of $X_A$ are preserved under $i_2$.

\begin{proposition}\label{prop:i_2_top_embedding} The inclusion $i_2$ is a topological embedding.
\end{proposition}

\begin{proof} First we show that $i_2$ is injective and continuous. The injectivity is straightforward: given $\varphi, \psi \in \widetilde{X}_A$ s.t. $i_2(\varphi) = i_2(\psi)$, it follows that $\varphi(e_g)=\psi(e_g)$ for all $g \in \mathbb{F}$. Since $\{e_g:g \in \mathbb{F}\}$ generates $\widetilde{\mathcal{D}}_A$, there exists an unique *-homomorphism which extends the function
\begin{equation*}
    e_g \mapsto \varphi(e_g), \quad g \in \mathbb{F}
\end{equation*}
and such uniqueness implies that $\varphi = \psi$, i.e., $i_2$ is injective. For the continuity, let $(\varphi_n)_{n \in \mathbb{N}}$ be a sequence in $\widetilde{X}_A$ such that $\varphi_n \rightharpoonup \varphi \in \widetilde{X}_A$. The topology in $\widetilde{X}_A$ is the weak$^*$ topology and it is metrizable. We have the following equivalences:
\begin{equation*}
    \varphi_n \rightharpoonup \varphi \iff \varphi_n(e_g) \rightharpoonup \varphi(e_g), \quad \forall g \in \mathbb{F} \iff \{\varphi_n(e_g)\}_{g \in \mathbb{F}} \to \{\varphi(e_g)\}_{g \in \mathbb{F}}, 
\end{equation*}
where the last convergence above is the precisely the one in the product topology. We observe that every injective continuous map from a compact space to a Hausdorff space is a topological embedding. And this is exactly the case we are dealing: $\widetilde{X}_A$ is compact, $\{0,1\}^\mathbb{F}$ is Hausdorff and $i_2$ is continuous and injective. We conclude that $i_2$ is a topological embedding. The proof is the same for $X_A$ when it is compact, and for $X_A$ non-compact, it follows straightforward by taking the restriction on $i_2$ from the $\widetilde{X}_A$ case. 
\end{proof}

Now we can see the characters in $X_A$ and $\widetilde{X}_A$ as configurations in the Cayley graph generated by $\mathbb{F}$, where the words $g$ are the vertices and the oriented edges multiply by the right the word in the source vertex by a letter $a$, leading to the range vertex. Of course, the inverse way of the edge represents a multiplication by the inverse of the correspondent letter $a$.
\begin{figure}[H]
\begin{center}
		\begin{tikzpicture}[scale=1.5,decoration={markings, mark=at position 0.5 with {\arrow{>}}}]
		\node[circle, draw=black, fill=black, inner sep=1pt,minimum size=1pt] (0) at (0,0) {};
		\node[circle, draw=black, fill=black, inner sep=1pt,minimum size=1pt] (1) at (3,0) {};
    	\draw[postaction={decorate}, >=stealth] (0)  to (1);
   	    \node[above] at (0,0) {$g$};
   	    \node[above] at (1.5,0) {$a$};
        \node[above] at (3,0) {$ga$};
		\end{tikzpicture}
	\end{center}
\end{figure}
 
The next corollary is straightforward.

\begin{corollary}\label{cor:i_2_i_1_Sigma_A_dense} $i_2 \circ i_1(\Sigma_A)$ is dense in $i_2(X_A)$. Moreover, if $\mathcal{O}_A$ is not unital, then $i_2 \circ i_1(\Sigma_A)$ is dense in $i_2(\widetilde{X}_A)$.
\end{corollary}

From now we will describe $X_A$ (respect. $\widetilde{X}_A$) by its copy $i_2(X_A)$ (respect. $i_2(\widetilde{X}_A)$) contained in $\{0,1\}^\mathbb{F}$, except when the map $i_2$ is explicitly needed.

\begin{definition} \label{def:convex_configurations} A configuration $\xi$ is said to be convex if for any two $a,b \in \mathbb{F}$ filled in $\xi$, the whole shortest path in the Cayley tree between $a$ and $b$ is also filled in $\xi$. 
\end{definition}

\begin{remark} Note that a configuration $\xi$ is convex and it is filled in $e$ if and only if, for every $g \in \mathbb{F}$ filled in $\xi$, the subwords of $g$ are also filled.
\end{remark}

Now, we present some properties of the configurations in $X_A$ (respect $\widetilde{X}_A$), and for such objective we define the following set, as it is done in \cite{EL1999}.

\begin{definition}\label{def:set_Omega_A_tau} For a given transition matrix $A$, we define the set
\begin{equation} \label{eq:Omega_A_tau}
    \Omega_A^{\tau} = \left\{ \begin{array}{l}
         \xi \in \{0,1\}^\mathbb{F}: \xi_e = 1, \text{ } \xi \text{ convex}, \\
         \text{if $\xi_\omega=1$, then there exists at most one $y \in \mathbb{N}$ s.t. $\xi_{\omega y}=1$}, \\
         \text{if $\xi_\omega= \xi_{\omega y}=1$, $y \in \mathbb{N}$, then for all $x \in \mathbb{N}$ } (\xi_{\omega x^{-1}}=1 \iff A(x,y)=1)
    \end{array}\right\}.
\end{equation} 
\end{definition}

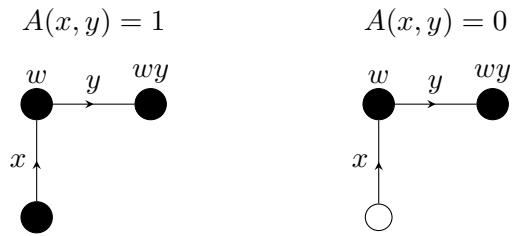
\begin{figure}[H]
\begin{center}
		\begin{tikzpicture}[scale=1.5,decoration={markings, mark=at position 0.5 with {\arrow{>}}}]
		\node[circle, draw=black, fill=black, inner sep=1pt,minimum size=5pt] (0) at (0,0) {0};
		\node[circle, draw=black, fill=black, inner sep=1pt,minimum size=5pt] (1) at (0,1) {1};
		\node[circle, draw=black, fill=black, inner sep=1pt,minimum size=5pt] (2) at (1,1) {2};
    	\draw[postaction={decorate}, >=stealth] (0)  to (1);
    	\draw[postaction={decorate}, >=stealth] (1)  to (2);
   		\node[left] at (0,0.5) {$x$};
   	    \node[above] at (0.5,1) {$y$};
   	    \node[above] at (0,1.1) {$w$};
        \node[above] at (1,1.1) {$wy$};
        \node[above] at (0.5,1.5) {$A(x,y)=1$};
       \node[circle, draw=black, fill=white, inner sep=1pt,minimum size=10pt] (3) at (3,0) {};
		\node[circle, draw=black, fill=black, inner sep=1pt,minimum size=5pt] (4) at (3,1) {4};
		\node[circle, draw=black, fill=black, inner sep=1pt,minimum size=5pt] (5) at (4,1) {5};
    	\draw[postaction={decorate},>=stealth] (3)  to (4);
    	\draw[postaction={decorate}, >=stealth] (4)  to (5);
   		\node[left] at (3,0.5) {$x$};
   	    \node[above] at (3.5,1) {$y$};
   	    \node[above] at (3,1.1) {$w$};
        \node[above] at (4,1.1) {$wy$};
        \node[above] at (3.5,1.5) {$A(x,y)=0$};
		\end{tikzpicture}
	\end{center}
	\caption{ Representation of the last condition of $\Omega_A^\tau$. The black dots represents that the configuration $\xi$ is filled. }
\end{figure}

\begin{remark} The set $\Omega_A^\tau$ generates the Toeplitz-Cuntz-Krieger algebra $\mathcal{TO}_A$. Theorem 4.6 of \cite{EL1999} presents a isomorphic correspondence between $\mathcal{TO}_A$ and $C(\Omega_A^\tau) \rtimes \mathbb{F}$.
\end{remark}

The set above is compact and it contains $X_A$ (respect. $\widetilde{X}_A$), and these two properties are proved next.

\begin{lemma} \label{lemma:Omega_A_tau_closed} $\Omega_A^{\tau}$ is a closed set.
\end{lemma}

\begin{proof} Let $(\xi^n)_\mathbb{N}$ be a sequence in $\Omega_A^{\tau}$ that converges to $\xi \in \{0,1\}^\mathbb{F}$, i.e.,
\begin{equation} \label{eq:convergence_xi_g}
    \forall g \in \mathbb{F}, \exists N_g \in \mathbb{N}:n > N_g \implies \xi^n_g = \xi_g.
\end{equation}
Now, let us prove that $\xi \in \Omega_A^{\tau}$.
\begin{enumerate}
    \item $\xi_e = 1$. It follows by applying \eqref{eq:convergence_xi_g} for $g=e$;
    \item $\xi$\textbf{ is convex.} By the previous condition we simply need to prove that if $\xi_g =1$, then $\xi_h=1$ for every $h$ subword of $g$. Indeed, let $\xi_g =1$ and $h$ be a subword of $g$. By \eqref{eq:convergence_xi_g} we have that for $n > N_g$ we have $\xi^n_g = \xi_g = 1$, hence $\xi^n_h = 1$ for all $n > N_g$, and therefore $\xi_h = 1$;
    \item \textbf{if $\xi_\omega=1$, then there exists at most one $y \in \mathbb{N}$ s.t. $\xi_{\omega y}=1$.} If $\xi_\omega = 1$, $\xi_{\omega x} = 1$ and $\xi_{\omega y} = 1$, then by \eqref{eq:convergence_xi_g} there exist $N_{\omega}, N_{\omega x}, N_{\omega y} \in \mathbb{N}$ such that $\mathbb{N}:n > N_g \implies \xi^n_g = \xi_g$, $g = \omega, \omega x, \omega y$. By taking $n > \max\{N_{\omega}, N_{\omega x}, N_{\omega y}\}$ we get that $\xi^n_{\omega y} = \xi^n_{\omega x} = 1$ and by definition of $\Omega_A^{\tau}$, we get that $x=y$;
    \item \textbf{if $\xi_\omega= \xi_{\omega y}=1$, $y \in \mathbb{N}$, then for all $x \in \mathbb{N}$ }$(\xi_{\omega x^{-1}}=1 \iff A(x,y)=1)$. Similarly as the previous condition, we can take $N = \max\{N_\omega,N_{\omega y},N_{\omega x^{-1}}\}$ and for $n >N$ we have that $\xi^n_g=\xi_g = 1$, for $g = \omega, \omega x^{-1}, \omega y$. We conclude that $\xi_{\omega x^{-1}} = 1 = \xi^n_{\omega x^{-1}}$ and the second equality happens if and only if $A(x,y) = 1$.
\end{enumerate}
\end{proof}

\begin{corollary}$\Omega_A^{\tau}$ is a compact subset of $\{0,1\}^\mathbb{F}$.
\end{corollary}    

\begin{proof} From Lemma \ref{lemma:Omega_A_tau_closed}, $\Omega_A^{\tau}$ is a closed subset of $\{0,1\}^\mathbb{F}$. On the other hand, $\{0,1\}^\mathbb{F}$ is the product space of compact space and then by the Tychonoff's Theorem it is compact. Since every closed subset of a compact space is also compact, and therefore $\Omega_A^{\tau}$ is compact. 
\end{proof}

\begin{proposition}\label{prop:X_A_in_Omega_tau} $i_2(\widetilde{X}_A) \subseteq \Omega_A^{\tau}$.
\end{proposition}

\begin{proof} It is sufficient to prove that $i_2(i_1(\Sigma_A)) \subseteq \Omega_A^{\tau}$ since $i_2(i_1(\Sigma_A))$ is dense in $i_2(\widetilde{X}_A)$, due to Corollary \ref{cor:i_2_i_1_Sigma_A_dense}, and $\Omega_A^{\tau}$ is closed, by Lemma\ref{lemma:Omega_A_tau_closed}, and therefore $i_2(\widetilde{X}_A) = \overline{i_2(i_1(\Sigma_A))} \subseteq \Omega_A^{\tau}$. By Lemma \ref{lemma:g_non_zero}, for any reduced word $g \in \mathbb{F}$ which is not in the form $g = \alpha \beta^{-1}$ for $\alpha, \beta \in \mathbb{F}_+$ (the case $g = e$ is included), we have that $e_g=0$. Given the character $\xi \in i_2(i_1(\Sigma_A))$, we have that $\xi = i_2(\varphi_\omega)$, $\omega \in \Sigma_A$, where
\begin{equation*}
    \varphi_\omega(R) = (R \delta_\omega,\delta_\omega), \quad \forall R \in \widetilde{\mathcal{D}}_A.
\end{equation*}
By direct checking we have for any $\nu \in \Sigma_A$ that
\begin{equation*}     e_{\alpha \beta^{-1}} \delta_\nu = \begin{cases}                                             \delta_\nu,\text{ if } \nu \in [\alpha]\text{ and } A(\beta_{|\beta|-1},\nu_{|\alpha|}) = 1;\\                                            0,\text{ otherwise;}                                       \end{cases}
\end{equation*}
where $|\alpha|$ is the length of $\alpha$ and $\alpha = \alpha_0\cdots\alpha_{|\alpha|}$ (with analogue conventions for $\beta$). Also, we conventioned that $[\alpha]=\Sigma_A$ for $|\alpha|=0$ and `$A(\beta_{|\beta|-1},\alpha_{|\alpha|}) = 1$' simply vanishes for $|\beta| = 0$. We conclude that 
\begin{equation}\label{eq:evaluation_character_generator}
    \xi_{\alpha \beta^{-1}} = i_2(\varphi_{\omega})_{\alpha\beta^{-1}} = \varphi_\omega(e_{\alpha \beta^{-1}}) =  (e_{\alpha \beta^{-1}} \delta_\omega,\delta_\omega) =                                                 \begin{cases}                                            1,\text{ if } \omega \in [\alpha]\text{ and } A(\beta_{|\beta|-1},\omega_{|\alpha|}) = 1;\\                                            0,\text{ otherwise;}
                          \end{cases}
\end{equation}
Now we are able to prove the proposition.
\begin{enumerate}
    \item[$(a)$] $\xi_e = 1$. It is obvious since $\xi_e = \varphi_\omega(e_e) = \varphi_\omega(1) = 1$;
    \item[$(b)$] $\xi_e$ \textbf{is convex.} Since $\xi_e = 1$, it is sufficient to prove that for given $g \in \mathbb{F}$ s.t. $\xi_g = 1$ we have that $\xi_h = 1$ for every subword $h$ of $g$. Indeed, let $g = \alpha \beta^{-1}$ be a finite reduced word such that $\xi_{\alpha\beta^{-1}} = 1$, $\alpha, \beta \in \mathbb{F}_+$. From \eqref{eq:evaluation_character_generator} we obtain that $\omega \in [\alpha]$ and $A(\beta_{|\beta|-1}, \omega_{|\alpha|})=1$. We have two possibilities for any $h$ reduced subword of $g$, namely $h=\alpha \beta_{|\beta|-1}^{-1}\cdots \beta_{p}^{-1}$, $0 \leq p \leq |\beta|-1$, and $h=\alpha_0 \cdots \alpha_p$, $0 \leq p \leq |\alpha|-1$. Both of the possibilities give us that $\xi_h=1$, since for all $\omega \in [\alpha]$ we have that $\omega \in [\alpha_0 \cdots \alpha_p]$, $0 \leq p \leq |\alpha|-1$. On other hand the condition $A(\beta_{|\beta|-1},\omega_{|\alpha|}) = 1$ remains unchanged for $h=\alpha \beta_{|\beta|-1}^{-1}\cdots \beta_{p}^{-1}$, $0 \leq p \leq |\beta|-1$ and it is not necessary when $h=\alpha_0 \cdots \alpha_p$, $0 \leq p \leq |\alpha|-1$, when you take $\omega = h$. Therefore we can use \eqref{eq:evaluation_character_generator} again to conclude that $\xi_h=1$;
    \item[$(c)$] \textbf{if $\xi_g=1$, then there exists at most one $y \in \mathbb{N}$ s.t. $\xi_{g y}=1$}. Suppose $\xi_g = \xi_{g y} = \xi_{g x}=1$, with $x,y \in \mathbb{N}$. Since $e_g = 0$ unless $g = \alpha \beta^{-1}$ with $\alpha,\beta \in \mathbb{F}_+$, we have two possibilities: $|\beta| \neq 0$ and $|\beta| = 0$. If $|\beta| \neq 0$ we have necessarily that $x=y=\beta_0$, otherwise we get $\xi_{g y} = \xi_{g x}=0$, a contradiction. If $|\beta| = 0$, then we have that $g = \alpha$ and by \eqref{eq:evaluation_character_generator} we get $\omega \in [\alpha x]\cap [\alpha y]$. However, $[\alpha x]\cap [\alpha y] \neq \emptyset$ if and only if $x=y$;
    \item[$(d)$] \textbf{if $\xi_g= \xi_{g y}=1$, $y \in \mathbb{N}$, then for all $x \in \mathbb{N}$ }$(\xi_{g x^{-1}}=1 \iff A(x,y)=1)$\textbf{.} Once again we recall that $e_g = 0$ unless $g = \alpha \beta^{-1}$ with $\alpha,\beta \in \mathbb{F}_+$. Also again we divide the proof in the same two possibilities as in $(c)$, namely $|\beta| \neq 0$ and $|\beta| = 0$. For $|\beta| \neq 0$, again the unique non contraditory possibility is that $y=\beta_0$ and $\xi_{gx^{-1}} = 1$ if and only if $A(x,\beta_0) =1$, i.e., $A(x,y) =1$. Now, if $|\beta| = 0$ we have by \eqref{eq:evaluation_character_generator} that $\omega \in [\alpha y]$, so We have the equivalences
    \begin{equation*}
        \xi_{gx^{-1}} = 1 \iff \xi_{\alpha x^{-1}} = 1 \stackrel{\text{\eqref{eq:evaluation_character_generator}}}{\iff} \omega \in [\alpha] \land A(x,\omega_{|\alpha|}) = 1
    \end{equation*}
    Hence,  $\xi_{gx^{-1}} = 1$ if and only if $A(x,y) = 1$ because $\omega_{|\alpha|} = y$.
    \end{enumerate}
\end{proof}

The proposition above shows that the elements of $\widetilde{X}_A$ satisfies the rules that define $\Omega_A^\tau$. Now, we introduce the definition of stem and root of a configuration as in \cite{EL1999}.

\begin{definition} By a \textit{positive word} in $\mathbb{F}$ we mean any finite or infinite sequence $\omega = \omega_0 \omega_1 \cdots$ satisfying $\omega_j \in \mathbb{F}_+$ for every $j$, including the empty word $e$. As well as it is defined for the classical Markov shift spaces, a positive word $\omega$ in $\mathbb{F}$ is said to be admissible when $A(\omega_j,\omega_{j+1}) = 1$ for every $j$. Given an either finite or infinite positive word $\omega = \omega_0 \omega_1 \cdots$, define the set
\begin{equation*}
    \llbracket \omega \rrbracket := \{e,\omega_0,\omega_0 \omega_1,\omega_0\omega_1\omega_2,\cdots\}
\end{equation*}
of the subwords of $\omega$. 
\end{definition}

\begin{remark} Observe that if $\omega$ is an infinite positive word, then $\omega \notin \llbracket \omega \rrbracket$.
\end{remark}

The next lemma shows the existence well-definition of a special positive admissible word: the stem of a configuration. In order to prove it we name the properties that define $\Omega_A^\tau$ as the rules below:
\begin{itemize}
    \item[$(R1)$] $\xi_e = 1$;
    \item[$(R2)$] $\xi$ is convex;
    \item[$(R3)$] if $\xi_g=1$, then there exists at most one $y \in \mathbb{N}$ s.t. $\xi_{g y}=1$;
    \item[$(R4)$] if $\xi_g= \xi_{g y}=1$, $y \in \mathbb{N}$, then for all $x \in \mathbb{N}$ it follows that
    \begin{equation*}
        \xi_{g x^{-1}}=1 \iff A(x,y)=1.
    \end{equation*}
\end{itemize}

\begin{lemma}\label{lemma:stem} Let $\xi \in \Omega_A^{\tau}$. There exists a unique positive admissible word $\omega$ such that
\begin{equation}\label{eq:identity_stem}
    \{g \in \mathbb{F}: \xi_g = 1\} \cap \mathbb{F}_+ = \llbracket \omega \rrbracket.
\end{equation}
\end{lemma}

\begin{proof} Let $\xi \in \Omega_A^{\tau}$ and define
\begin{equation*}
    G(\xi):=\{g \in \mathbb{F}: \xi_g = 1\} \cap \mathbb{F}_+.
\end{equation*}
We observe that $(R4)$ implies that every element in $G(\xi)$ is admissible. Due to $(R1)$ and $(R2)$, we have that if $\omega \in G(\xi)$ then every subword of $\omega$ is also contained in $G(\xi)$. We claim that given $\omega,\omega' \in G(\xi)$ we necessarily have either $\omega \in \llbracket \omega' \rrbracket$ or $\omega' \in \llbracket \omega \rrbracket$. In fact, w.l.o.g. we may assume $|\omega'| \leq |\omega|$. If $\omega'= e$ the claim is straightforward, so suppose $\omega' \neq e$. If $\omega' \notin \llbracket \omega \rrbracket$, then there exists $k \in \{0,...,|\omega'|-1\}$ s.t. $\omega_k \neq \omega'_k$  and we may take the smallest $k$ with such property. Then, the element
\begin{equation*}
    g = \begin{cases}
            e, \quad \text{if } k = 0;\\
            \omega'_0 \cdots \omega'_{k-1}, \quad \text{otherwise};
        \end{cases}
\end{equation*}
is such that $\xi_g = 1$, $\xi_{g\omega_k} = 1$ and $\xi_{g\omega'_k} = 1$, a contradiction, because of $(R3)$. So the claim holds. Consequently, for each $n \in \mathbb{N}$, $G(\xi)$ has at most one element of length\footnote{We set that $e$ has length zero.} $n$. We have two cases:
\begin{itemize}
    \item $G(\xi)$ is finite. As we noticed in the beginning of the proof, every subword of $\omega$ belongs to $G(\xi)$ and then $\llbracket \omega \rrbracket \subseteq G(\xi)$. The inverse inclusion is a consequence of the claim above and the fact that $\omega$ is the largest word in $G(\xi)$;
       
    \item $G(\xi)$ is (countably) infinite. In this case, take the infinite sequence $\omega$ that coincides with the sequence of $G(\xi)$ with length $n \in \mathbb{N}_0$ when on its first $n+1$ coordinates. By construction, $\llbracket \omega \rrbracket \subseteq G(\xi)$. The inverse inclusion follows by the same argument as in the previous case.
\end{itemize}
\end{proof}

The existence and uniqueness of the word in the lemma above for each configuration motivates part of the next definition.

\begin{definition}[stems and roots] Let $\xi \in \Omega_A^\tau$. The stem of $\xi$, denoted by $\kappa(\xi)$, is the positive admissible word such that
\begin{equation*}
    \{g \in \mathbb{F}: \xi_g = 1\} \cap \mathbb{F}_+ = \llbracket \omega \rrbracket.
\end{equation*}
We say that a configuration $\xi \in \Omega_A^\tau$ is a \textit{bounded element} if its stem has finite length. If $\xi$ is not bounded we call it \textit{unbounded}.

Given $g \in \mathbb{F}$ s.t. $\xi_g = 1$, the root of $g$ relative to $\xi$, denoted by $R_\xi(g)$, is defined by
\begin{equation*}
    R_\xi(g) := \{x \in S: \xi_{g x^{-1}} = 1\}.
\end{equation*}
\end{definition}

\begin{remark} Observe that if $\kappa(\xi) \neq e$, then $\kappa_{|\kappa|-1} \in R_\xi(\kappa(\xi))$. 
\end{remark}

Roughly speaking, we may compare the notions of stem and root as a hydrographic basin as follows. The stem of a configuration correponds to the `positive main river' of a configuration, that is, the longest path of positive finite words which are filled in the configuration. On the other hand, the root of an element of $g \in \mathbb{F}$ on a configuration corresponds to the `affluent edges that disembogue to the vertex $g$'. Observe that the inclusion $i_2 \circ i_1 : \Sigma_A \to \{0,1\}^{\mathbb{F}}$ necessarily maps the elements of $\Sigma_A$ to unbounded configurations. Moreover, every unbounded configuration in $\Omega_A^\tau$ comes from an element of $\Sigma_A$. This relation is presented and proved next.

\begin{proposition}\label{prop:bijection_sigma_unbounded_elements} The inclusion $i_2 \circ i_1 : \Sigma_A \to \Omega_A^\tau$ is a bijection between $\Sigma_A$ and the unbounded elements of $\Omega_A^\tau$. 
\end{proposition}

\begin{proof} The fact that every $\omega \in \Sigma_A$ is mapped to an unbounded configuration $\xi(\omega) \in \Omega_A^\tau$ is a straightforward consequence of Definition \ref{def:i_1}, by evaluating $\xi(\omega)_\alpha = \varphi_\omega(e_\alpha) = 1$ for every $\alpha \in \llbracket \omega \rrbracket$, and hence it is a unbounded configuration. Conversely, given an unbounded configuration $\xi \in \Omega_A^\tau$, we observe that $(R4)$ imposes that $\xi$ is the unique configuration of $\Omega_A^\tau$ having a stem $\kappa(\xi)$. On the other hand $\kappa(\xi) \in \Sigma_A$ since it is an infinite positive admissible sequence. Since the stem of $i_2 \circ i_1(\kappa(\xi))$ is $\kappa(\xi)$, we conclude that every unbounded element of $\Omega_A^\tau$ comes from $\Sigma_A$, concluding the proof. 
\end{proof}

\begin{corollary}\label{cor:unbounded_dense_in_X_A} $X_A$ ($\widetilde{X}_A$, when it is considered) is the closure of the unbounded elements of $\Omega_A^\tau$.
\end{corollary}

\begin{proof} It is a straightforward consequence of Corollary \ref{cor:i_2_i_1_Sigma_A_dense} and Proposition \ref{prop:bijection_sigma_unbounded_elements}. 
\end{proof}

Now we can characterize the bounded elements of $X_A$, denoted by the set $Y_A$, that is,
\begin{equation*}
    Y_A:= X_A \setminus \Sigma_A,
\end{equation*}
where we ommited the inclusions $i_1$ and $i_2$. What we do next is to characterize the bounded elements of $\Omega_A^\tau$ and $Y_A$. In order to realize that, we consider the most important root for every bounded configuration, namely $R_{\xi}(\kappa(\xi))$, that is, the root of the stem of a configuration relative to this one. Observe that $Y_A = \emptyset$ if and only if $\Sigma_A$ is locally compact, and then the existence of bounded configurations requires that the alphabet be infinitely countable, and hence from now and until the end of this thesis, we adopt the following standing hypothesis.

\begin{mdframed} \textbf{Standing hypothesis:} the alphabet on the Exel-Laca algebras considered is $\mathbb{N}$. Moreover, the transition matrix is not row-finite, that is, there exists at least one symbol in $\mathbb{N}$ that is an infinite emmiter in the symbolic graph.
\end{mdframed}

The hypothesis above combined with the transitivity of the matrix implies that the corresponding Markov shifts $\Sigma_A$ are not locally compact.

The next proposition characterizes the bounded elements of $\Omega_A^\tau$.

\begin{proposition}\label{prop:stem_root_pair_for_bounded_configurations_of_Omega_A_tau_uniqueness} Let $\xi \in \Omega_A^\tau$ a bounded configuration. Then $\xi$ is uniquely determined by its stem and $R_\xi(\kappa(\xi))$. Conversely, every pair $(\kappa,R)$, where $\kappa$ is a finite positive admissible word and $R \subseteq \mathbb{N}$ satisfying\footnote{Observe that if $\kappa = e$, then $\kappa_{|\kappa|-1} \in R$ is a vacuous condition.} $\kappa_{|\kappa|-1} \in R$, determines a configuration $\xi\in \Omega_A^\tau$ s.t. $\kappa = \kappa(\xi)$ and $R = R_\xi(\kappa(\xi)) = R_\xi(\kappa)$. 
\end{proposition}

\begin{proof} Suppose that there exists $\xi,\eta \in \Omega_A^\tau$ s.t. $\kappa(\xi) = \kappa(\eta) = \kappa$ and $R_\xi(\kappa) = R_\eta(\kappa)$. Given $g = \alpha \beta^{-1}$ irreducible we have by $(R4)$ that, if $\beta \neq e$, then
\begin{equation*}
    \xi_g = 1 \iff \xi_{\alpha \beta_{|\beta|-1}^{-1}} = 1,
\end{equation*}
and the same is valid for $\eta$. We have two possibilities:
\begin{itemize}
    \item $\alpha \neq \kappa$. By Lemma \ref{lemma:stem} and $(R4)$\footnote{If $\beta = e$, then $(R4)$ is not necessary.} we have that
    \begin{equation*}
        \xi_g = 1 \iff \alpha \in \llbracket \kappa \rrbracket \setminus \{\kappa\} \text{ and } A(\beta_{|\beta|-1}^{-1},\kappa_{|\alpha|}) = 1\text{\footnote{Also, if $\beta = e$, then the condition on the matrix is not necessary.}} \iff \eta_g = 1.
    \end{equation*}
    \item $\alpha = \kappa$. It is straightforward that $\xi_\kappa = \eta_\kappa = 1$, so suppose $\beta \neq e$, we have by the hypothesis $R_\xi(\kappa) = R_\eta(\kappa)$ that
    \begin{equation*}
        \xi_g = 1 \iff \eta_g = 1.
    \end{equation*}
\end{itemize}
We conclude that $\xi = \eta$. Now, take a pair $(\kappa,R)$ as in the statement of this proposition. We construct the following configuration of $\xi \in \{0,1\}^\mathbb{F}$ by setting the following
\begin{itemize}
    \item[$(i)$] $\xi_\omega = 1$ for every $\omega \in \llbracket \kappa \rrbracket$;
    \item[$(ii)$] $\xi_{\kappa j^{-1}} = 1$ for every $j \in R$;
    \item[$(iii)$] $\xi_{\alpha \beta^{-1}} = 1$ if and only if $\alpha \in  \llbracket \kappa \rrbracket \setminus \{\kappa\}$, $\beta$ is positive admissible and $A(\beta_{|\beta|-1},\kappa_{|\alpha|}) = 1$;
    \item[$(iv)$] $\xi_{\kappa \beta^{-1}} = 1$ for every $\beta$ positive admissible such that $\beta_{|\beta|-1} \in R$;
    \item[$(v)$] $\xi_g = 0$ for every $g \in \mathbb{F}$ in the remaining possibilities.
\end{itemize}
We have that $\xi \in \Omega_A^\tau$. Indeed, $(i)$ implies $(R1)$; $(R2)$ is a consequence from the fact that every path between two filled words is a set of elements of $g \in \mathbb{F}$ described in the conditions $(i)-(iv)$; $(R4)$ comes from the conditions $(iii)$ and $(iv)$, by the fact that every subword of $\kappa$ is also admissible, and that the lack of admissibility leads to the case $(v)$. Let us prove $(R3)$. For $(R3)$ suppose that there exists $g \in \mathbb{F}$, that $\xi_{g} = \xi_{gi} = \xi_{gj} = 1$ for two distinct elements $i,j \in \mathbb{N}$. Since $\xi_g = 1$, we have necessarily that $g = \alpha \beta^{-1}$, where $\alpha$ and $\beta$ are admissible. Moreover, we must have $\beta \neq e$, otherwise we would have two different positive admissible words of same length which implies that one of them does not belong to $\llbracket \kappa \rrbracket$ and hence $\xi_{gi} = 0$ or $\xi_{gj} = 0$, a contradiction. However, if $\beta \neq e$ then we have $g = \alpha \beta^{-1} i$ and by $(v)$ we have that $\xi_g = 0$. We conclude that $(R3)$ holds. Therefore $\xi \in \Omega_A^\tau$.
\end{proof}

For the next proposition, we see the root relative to the stem of a configuration of this configuration itself as an infinite matrix column in $\{0,1\}^\mathbb{N}$, endowed with the product topology.

\begin{proposition}\label{prop:characterization_of_elements_of_Y_A} Let $\xi \in \Omega_A^\tau$. Then $\xi \in Y_A$ if and only if it satisfies simultaneously the two following items:
\begin{itemize}
    \item[$(a)$] $\kappa(\xi) = e$ or $\kappa(\xi)$ ends in an infinite emmiter letter; 
    \item[$(b)$] If $\kappa(\xi)\neq e$, $R_\xi(\kappa(\xi))$ is an accumulation point of the sequence $(\mathfrak{c}(i))_{A(\kappa(\xi)_{|\kappa(\xi)|-1},i)=1}$ of columns of $A$, with these columns seen as elements of $\{0,1\}^\mathbb{N}$. If $\kappa(\xi)= e$, $R_\xi(\kappa(\xi))$ is an accumulation point of the sequence $(\mathfrak{c}(i))_{i \in \mathbb{N}}$.
\end{itemize}
\end{proposition} 

\begin{proof} Let $\xi \in Y_A$. By Corollary \ref{cor:unbounded_dense_in_X_A}, there exists a sequence $\{\xi^n\}_{\mathbb{N}}$ in $\Sigma_A$ converging to $\xi$. Then, we have that
\begin{equation*}
    (\xi^n)_{\kappa(\xi)} \to 1 \quad \text{and} \quad (\xi^n)_{\kappa(\xi) a} \to 0,
\end{equation*}
where the second convergence holds for every $a \in \mathbb{N}$. Since the convergence is inherited from the product topology of $\{0,1\}^\mathbb{F}$, for each $a \in \mathbb{N}$, the sequence  $\{(\xi^n)_{\kappa(\xi) a}\}_\mathbb{N}$ is eventually constant. Hence, we necessarily have that $\kappa(\xi) = e$ or it ends in an infinite emmiter, otherwise it would exist a symbol $b \in \mathbb{N}$ such that $(\xi^n)_{\kappa(\xi) b} = 1$ for infinitely many values of $n$, leading to a contradiction, because this fact would imply that $\xi_{\kappa(\xi)b}=1$. On the other hand, suppose that $\kappa(\xi) \neq e$, we have
\begin{equation}
    R_{\xi^n}(\kappa(\xi)) = \{j \in \mathbb{N}: A(j,a_n) = 1\} = c(a_n)
\end{equation}
where $a_n \in \mathbb{N}$ is the unique symbol due to $(R3)$ s.t. $\xi^n_{\kappa(\xi)a_n} = 1$ and $c(a_n)$ is the $a_n$-th column of $A$ when we see $R_{\xi^n}(\kappa(\xi))$ as an element of $\{0,1\}^{\mathbb{N}}$. Since $\xi^n \to \xi$ we have that
\begin{equation*}
    R_{\xi^n}(\kappa(\xi))_j = \xi^n_{\kappa(\xi) j^{-1}} \to \xi_{\kappa(\xi) j^{-1}} = R_\xi(\kappa(\xi)_j),
\end{equation*}
for every $j \in \mathbb{N}$, that is, $c(a_n) \to R_\xi(\kappa(\xi)$, and therefore $R_\xi(\kappa(\xi))$ is an accumulation point of the set $\{c(i): A(\kappa(\xi)_{|\kappa(\xi)|-1},i)\}$ of the columns of $A$, and therefore it is an accumulation point since this set is countable. The proof is similar for $\kappa(\xi) = e$.

Conversely, suppose that $\xi \in \Omega_A^\tau$ satisfies $(a)$ and $(b)$. We construct a sequence $\{\xi^n\}_{\mathbb{N}}$ of elements $\Sigma_A$ converging to $Y_A$, by setting
\begin{equation*}
    \xi^n_{\kappa(\xi)} = \xi^n_{\kappa(\xi)i(n)} = 1,
\end{equation*}
where $\{i(n)\}_{\mathbb{N}}$ is an enumeration of distinct elements of an infinite subset of $\{k \in \mathbb{N}:A(\kappa(\xi)_{|\kappa(\xi)|},k)=1\}$ if $|\kappa(\xi)| > 0$, and an infinit subset of $\mathbb{N}$ otherwise, in both cases satisfying $\mathfrak{c}(i(n)) \to R_\xi(\kappa(\xi)$. In addition, for each $n \in \mathbb{N}$ we choose $\omega_n \in [\kappa(\xi)i(n)]$ and set $\xi^n_{u}:= 1$ for every $u \in \llbracket \omega_n \rrbracket$. This choice determines uniquely each $\xi^n$. It is straightforward that $\xi^n_{\kappa(\xi)} = 1$ for every $n \in \mathbb{N}$,
\begin{equation*}
    \xi^n_{\kappa(\xi)a} \to 0 \quad \text{and} \quad R_{\xi^n}(\kappa(\xi)) = \mathfrak{c}(i(n)) \to R_\xi(\kappa(\xi)).
\end{equation*}
We constructed a sequence in $\Sigma_A$ that converges to an element $\xi \in \Omega_A^\tau$, by density, we necessarily have that $\xi \in Y_A$.
\end{proof}

In order to provide more intuition to the reader, we discuss how the convergence of elements of elements of $\Sigma_A$ to elements of $Y_A$ occurs. Let $\xi \in Y_A$. Due to Corollary \ref{cor:unbounded_dense_in_X_A}, there must exist a sequence $\{\xi^n\}_{\mathbb{N}}$ in $\Sigma_A$ that converges to $\xi$. Let $\omega$ be the stem of $\xi$. We must have necessarily that $\xi^n_{\omega j} \to 0$ for every $j \in \mathbb{N}$. In particular, only need to consider those $j$ s.t. $\omega j$ is admissible. Then we must have for each one of these particular $j$'s that $\xi^n_{\omega j} = 1$ for a finite quantity of $n$'s. So w.l.o.g. we may suppose that $\xi^n_{\omega j_n}$ for each $n$ and that $j_n \neq j_m$ if $n \neq m$. That make us understand that must be an infinite different possible $j_k$'s s.t. $\omega j_k$ is admissible, that is, $\omega = e$ or $\omega$ ends in an infinite emmiter. The figure \ref{fig:sequence_of_Sigma_A_converging_to_Y_A} illustrates this discussion.

\begin{figure}[H]
\centering
\caption{A sequence of unbounded elements of $X_A$ converging to an element of $Y_A$ viewed close to the stem $\omega$ of the limit configuration. The black dots represents the filled vertices and the white dots represents the non-filled vertices. The oriented edge from the vertex $\omega$ to $\omega j_k$, $k \in \mathbb{N}$ represents the multiplication of $\omega$ by $j_k$. \label{fig:sequence_of_Sigma_A_converging_to_Y_A}}
\begin{tikzpicture}

\node[circle,fill=black,inner sep=0pt,minimum size=8pt,label=below:{$\omega$}]      (maintopic)                              {};
\node[circle,fill=black,minimum size=8pt, label=right:{$\omega j_1$}]        (j1)       [above right= 1.5cm and 1cm of maintopic] {};
\node[circle,fill=white,draw,minimum size=8pt,label=right:{$\omega j_2$}]        (j2)       [above right= 1cm and 1.1cm of maintopic] {};
\node[circle,fill=white,draw,minimum size=8pt,label=right:{$\omega j_3$}]        (j3)       [above right= 0.5cm and 1.2cm of maintopic] {};
\node[circle,fill=white,draw,minimum size=8pt,label=right:{$\omega j_4$}]        (j4)       [above right= 0cm and 1.3cm of maintopic] {};

\draw[->,shorten >=0.1cm](maintopic.north) edge[bend left=45] (j1.west);
\draw[->,shorten >=0.1cm](maintopic.north) edge[bend left=45] (j2.west);
\draw[->,shorten >=0.1cm](maintopic.north) edge[bend left=45] (j3.west);
\draw[->,shorten <=0.7cm,ultra thick](j3.east) -- ($(j3.west)+(0:2.2)$);
\draw[->,shorten >=0.1cm](maintopic.north) edge[bend left=45] (j4.west);

\node[circle,fill=black,inner sep=0pt,minimum size=8pt,label=below:{$\omega$}]      (maintopic2) [right= 3.5cm of maintopic]                              {};
\node[circle,fill=white,draw,minimum size=8pt, label=right:{$\omega j_1$}]        (j12)       [above right= 1.5cm and 1cm of maintopic2] {};
\node[circle,fill=black,draw,minimum size=8pt,label=right:{$\omega j_2$}]        (j22)       [above right= 1cm and 1.1cm of maintopic2] {};
\node[circle,fill=white,draw,minimum size=8pt,label=right:{$\omega j_3$}]        (j32)       [above right= 0.5cm and 1.2cm of maintopic2] {};
\node[circle,fill=white,draw,minimum size=8pt,label=right:{$\omega j_4$}]        (j42)       [above right= 0cm and 1.3cm of maintopic2] {};

\draw[->,shorten >=0.1cm](maintopic2.north) edge[bend left=45] (j12.west);
\draw[->,shorten >=0.1cm](maintopic2.north) edge[bend left=45] (j22.west);
\draw[->,shorten >=0.1cm](maintopic2.north) edge[bend left=45] (j32.west);
\draw[->,shorten <=0.7cm,ultra thick](j32.east) -- ($(j32.west)+(0:2.2)$);
\draw[->,shorten >=0.1cm](maintopic2.north) edge[bend left=45] (j42.west);

\node[circle,fill=black,inner sep=0pt,minimum size=8pt,label=below:{$\omega$}]      (maintopic3) [right= 3.5cm of maintopic2]                              {};
\node[circle,fill=white,draw,minimum size=8pt, label=right:{$\omega j_1$}]        (j13)       [above right= 1.5cm and 1cm of maintopic3] {};
\node[circle,fill=white,draw,minimum size=8pt,label=right:{$\omega j_2$}]        (j23)       [above right= 1cm and 1.1cm of maintopic3] {};
\node[circle,fill=black,draw,minimum size=8pt,label=right:{$\omega j_3$}]        (j33)       [above right= 0.5cm and 1.2cm of maintopic3] {};
\node[circle,fill=white,draw,minimum size=8pt,label=right:{$\omega j_4$}]        (j43)       [above right= 0cm and 1.3cm of maintopic3] {};

\draw[->,shorten >=0.1cm](maintopic3.north) edge[bend left=45] (j13.west);
\draw[->,shorten >=0.1cm](maintopic3.north) edge[bend left=45] (j23.west);
\draw[->,shorten >=0.1cm](maintopic3.north) edge[bend left=45] (j33.west);
\draw[->,shorten <=0.7cm,ultra thick](j33.east) -- ($(j33.west)+(0:2.2)$);

\draw[->,shorten >=0.1cm](maintopic3.north) edge[bend left=45] (j43.west);


\node[circle,fill=black,inner sep=0pt,minimum size=8pt,label=below:{$\omega$}]      (maintopic4) [right= 3.5cm of maintopic3]                              {};
\node[circle,fill=white,draw,minimum size=8pt, label=right:{$\omega j_1$}]        (j14)       [above right= 1.5cm and 1cm of maintopic4] {};
\node[circle,fill=white,draw,minimum size=8pt,label=right:{$\omega j_2$}]        (j24)       [above right= 1cm and 1.1cm of maintopic4] {};
\node[circle,fill=white,draw,minimum size=8pt,label=right:{$\omega j_3$}]        (j34)       [above right= 0.5cm and 1.2cm of maintopic4] {};
\node[circle,fill=black,draw,minimum size=8pt,label=right:{$\omega j_4$}]        (j44)       [above right= 0cm and 1.3cm of maintopic4] {};
\node     (dots)       [right= 2.0cm of j34] {$\bullet\bullet\bullet$};

\draw[->,shorten >=0.1cm](maintopic4.north) edge[bend left=45] (j14.west);
\draw[->,shorten >=0.1cm](maintopic4.north) edge[bend left=45] (j24.west);
\draw[->,shorten >=0.1cm](maintopic4.north) edge[bend left=45] (j34.west);
\draw[->,shorten <=0.7cm,ultra thick](j34.east) -- ($(j34.west)+(0:2.2)$);
\draw[->,shorten >=0.1cm](maintopic4.north) edge[bend left=45] (j44.west);

\end{tikzpicture}
\end{figure}
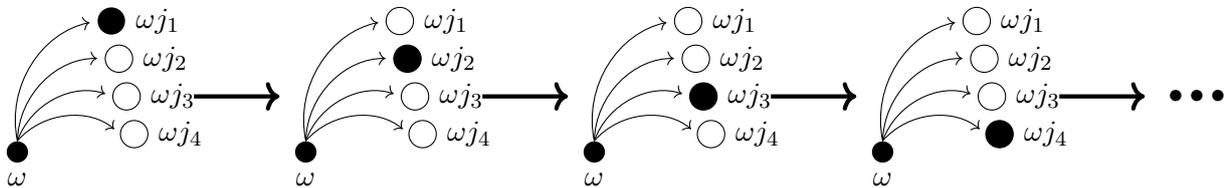

Also, we have the converse density, that is, the density of $Y_A$ on $X_A$, as proved next.

\begin{proposition}\label{prop:Y_A_dense} $Y_A$ is dense in $X_A$.
\end{proposition}

\begin{proof} We recall that  there exists a symbol $i \in \mathbb{N}$ such that $|\{j \in \mathbb{N}: A(i,j) = 1\}| = \infty$. Let $\xi \in \Sigma_A$, that is, $\kappa(\xi) = x = x_0 x_1 x_2 \cdots \in \Sigma_A$. By transitivity of $\Sigma_A$, there exists an admissible word $w_k$ such that $x_k w_k i$ is also admissible. Since $A(i,j) = 1$ is satisfied for an infinite values of $j$, by \ref{prop:bijection_sigma_unbounded_elements} and Corollary \ref{cor:unbounded_dense_in_X_A} we may construct for each $n \in \mathbb{N}_0$ at least one $\xi^n \in Y_A$ satisfying $\kappa(\xi^n) = x_0 \cdots x_n w^n i$. For each $n$ choose one of the possible $\xi^n$ as previously. We claim that $\xi^n \to \xi$. Indeed, for every $p \in \mathbb{N}_0$ the sequences $\{(\xi^n)_{x_0\cdots x_p}\}_{n \in \mathbb{N}_0}$ are constant for $n>p$ and therefore convergent, and moreover
\begin{equation*}
    (\xi^n)_{x_0\cdots x_p} \to 1 = \xi_{x_0\cdots x_p}.
\end{equation*}
Consequently for every $g \in \mathbb{F}$ in the form $g = \alpha \gamma^{-1}$ or $\gamma^{-1}$ irreducible, where $\alpha$ and $\gamma$ are admissible words, we have that $(\xi^n)_g$ is constant for $n > |\alpha|$ (we consider $|\alpha| = 0$ for $g = \gamma^{-1}$) due to $(R4)$. So  
\begin{equation*}
    (\xi^n)_{g} \to \xi_g,
\end{equation*}
and therefore $\xi^n \to \xi$.
\end{proof}

We warn the reader that, although in many cases $Y_A$ is countable, this is not granted, even under the transitivity hypothesis, we can have $Y_A$ uncountable, as we show in the next example. Futhermore, this example shows a transition matrix $A$ s.t. both $\Sigma_A$ and $X_A$ are not compact. 

\begin{example}[Uncountable $Y_A$] \label{exa:uncountable_Y_A} Consider the transition matrix $A$ as follows: each of its columns is a periodic sequences in $\{0,1\}^\mathbb{N}$, excluding the zero sequence. Also, we order the columns by increasing minimal period. Set
\begin{align*}
    A(n,1)= \begin{pmatrix}
             1\\
             1\\
             1\\
             1\\
             1\\
             1\\
             \vdots
            \end{pmatrix}, \quad
    A(n,2)= \begin{pmatrix}
             1\\
             0\\
             1\\
             0\\
             1\\
             0\\
             \vdots
            \end{pmatrix}
    \quad \text{ and } \quad
    A(n,2)= \begin{pmatrix}
             0\\
             1\\
             0\\
             1\\
             0\\
             1\\
             \vdots
            \end{pmatrix}.
\end{align*}
And for the remaining columns with same minimal period, take any ordering\footnote{It does not matter the choice of ordering among sequences of same minimal period, since it is a finite set of columns.}. Denote by $p(n)$, $n \in \mathbb{N}$, the minimal period of the sequence of the $n$-th column. It is straightforward that $p(n) < n$ for ever $n > 2$. We prove the following claims.

\textbf{Claim: for $b \in \mathbb{N}$, $b > 2$, there exists $c \in \mathbb{N}$ satisfying $A(c,b) = 1$ and $c \leq p(b)$.} In fact, suppose that for every $c \leq p(b)$ we have $A(c,b) = 0$. Then we have that the $b$-th column has only zeros in its first $p(b)$-entries as it is shown in figure \ref{fig:uncountable_Y_A}, and since $p(b)$ is the period of such sequence, it follows that $A(n,b) = 0$ for every $n \in \mathbb{N}$, a contradiction because every column of $A$ is not the zero column.

\begin{figure}[H]
 \centering
 \caption{A contradiction on the matrix constructed in this example if, for some $b >2$, we would have $A(c,b) = 0$ for every $c \leq p(b)$. In this case, the first $p(b)$ terms of $A(n,b)$ must be zero, and by periodicity, we conclude that $A(n,b) = 0$, for every $n \in \mathbb{N}$.\label{fig:uncountable_Y_A}}
 \includegraphics[scale = .35]{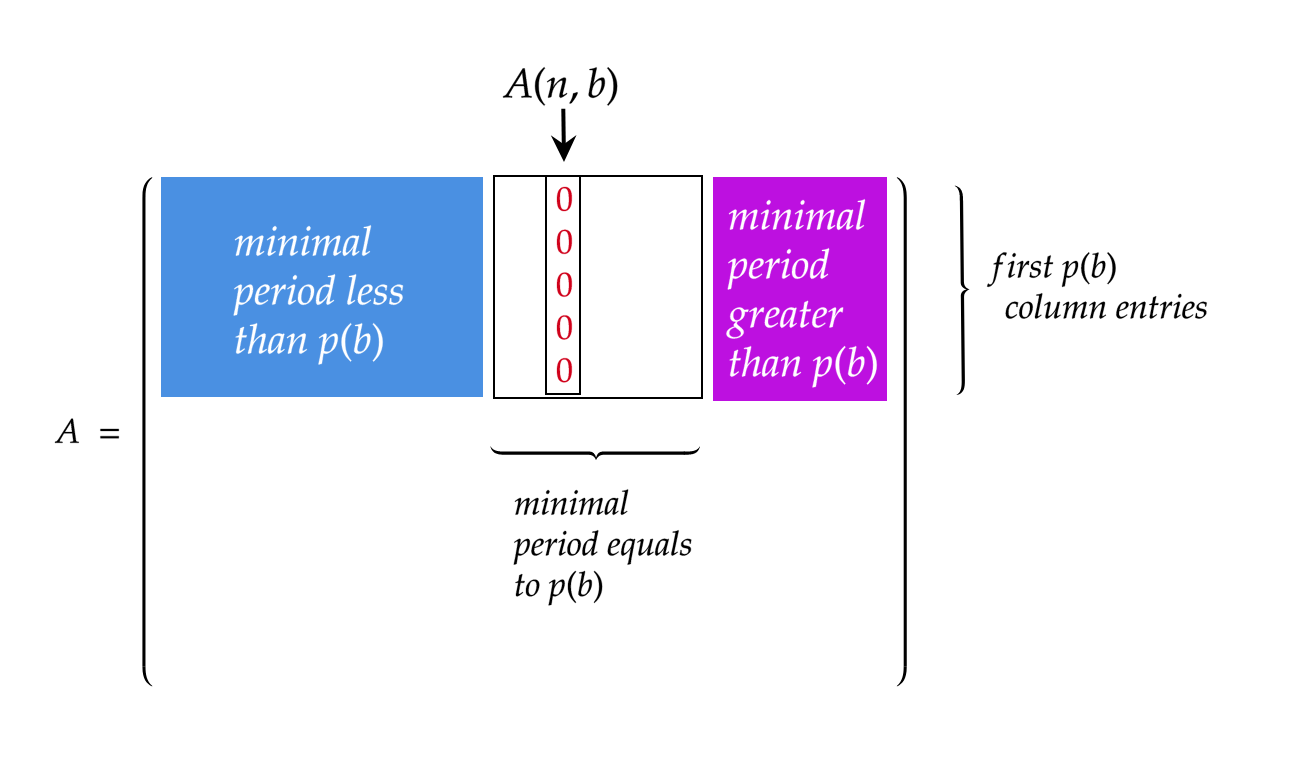}
\end{figure}

\textbf{Claim: $A$ is transitive.} Indeed, let $a,b\in \mathbb{N}$. The case $b=1$ and $b=2$ are clear: since for $b=1$, we have $A(a,b)=1$ for every $a$ because $A(n,1)$ is the column of the constant sequence of $1$'s. On the other hand if $b=2$, then $A(1,b)=1$ and $A(a,1)=1$, and so $a1b$ is admissible. 

Now, suppose $b>2$. By the previous claim, there exists $c_1 \in \mathbb{N}$ $c_1 \leq p(b) < b$ s.t. $A(c_1,b) = 1$. If $c_1\leq 2$, then $a1c_1b$ is admissible. Otherwise, then we restart the process, by taking $c_2<p(c_1)<c_1$ such that $A(c_2,c_1)=1$ and verifying if $c_2 \leq 2$ or not. Since the process must terminate, after $m \in \mathbb{N}$ iterations, we find $c_1,...,c_m \in \mathbb{N}$, with $c_m\leq 2$, and s.t. $a1c_nc_{n-1}\cdots c_1b$ is admissible.

\textbf{Claim: $A$ is topologically mixing.} It is straightforward from the last claim and the fact that $A(1,1) = 1$. In fact, by the proof above, for every $a,b \in \mathbb{N}$, if $b = 1$, then, for every $n \in \mathbb{N}$ we have the admissible word $a1^n 1$, where
\begin{equation*}
    1^n = \underbrace{1\cdots 1}_{n \text{ times}}.
\end{equation*}
Now, if $b = 2$, then $a1b$ is admissible. Again, we have that $a1^nb$ is admissible for every $n \in \mathbb{N}$. If $b >2$, by the proof of last claim, there exists $m \in \mathbb{N}$ s.t. $a1c_m \cdots c_1 b$ is admissible and therefore $a1^nc_m \cdots c_1 b$ for every $n \in \mathbb{N}$, and we conclude that $A$ is topologically mixing.

\textbf{Claim: the accumulation points of the sequence $(\mathfrak{c}(n))_{\mathbb{N}}$, where $\mathfrak{c}(n)$ is the $n$-th column of $A$ is a uncountable set.} This is straightforward from the fact that periodic sequences are dense in the uncountable set $\{0,1\}^{\mathbb{N}}$. In fact, for every $w = a_1 a_2 \cdots \in \{0,1\}^\mathbb{N}$, consider as the first element of the subsequence take $v^1 = A(n,1)$. Now, for $k>1$, take $v^k = (a_1,...,a_k 1)^\infty$. For each $m \in \mathbb{N}$, $(v^k)_m$ is eventually constant and equals to $a_m$, and then $v^k \to w$. Then the set of accumulation points of the sequence is $(\mathfrak{c}(n))_{\mathbb{N}}$.

\textbf{Claim: $\Sigma_A$ and $X_A$ are not compact.} For $\Sigma_A$ is straightforward, observe for instance, that $A(n,1) = 1$ for infinite quantity\footnote{Actually every row has infinite quantity of $1$'s, that is, every symbol is an infinite emitter.} of $n$'s. By the last claim, every element in $\{0,1\}^\mathbb{N}$ is an accumulation point of $(\mathfrak{c}(n))_{\mathbb{N}}$, so in particular the null vector is it, and by Theorem \ref{thm:O_A_unital} $(iv)$, we have that $\mathcal{O}_A$ is not unital, then $\mathcal{D}_A$ is not unital, and therefore $X_A$ is not compact.
\end{example}

From now we shall use the following notation: we denote by $\mathfrak{C}(A)$ the set of columns of $A$, also the elements of $\mathfrak{C}$ are denoted as follows: given $j \in \mathbb{N}$ we denote by $\mathfrak{c}(j)$ the $j$-th column of the matrix $A$, that is,
\begin{equation*}
    \mathfrak{c}(j)_i = \begin{cases}
                1, \quad \text{if } A(i,j) = 1;\\
                0, \quad \text{otherwise.}
             \end{cases}
\end{equation*}


Now, we prove the compatibility between the classical and generalized Markov shift spaces in terms of Measure Theory.

\begin{proposition}\label{prop:Sigma_A_Borel_subset_in_X_A}
$i_2\circ i_1(\Sigma_A)$ is a measurable set in the Borel $\sigma$-algebra of $i_2(X_A)$.
\end{proposition}
\begin{proof} Using the fact that $i_2$ is a homeomorphism,
\begin{equation*}
  i_2(X_A\setminus i_1(\Sigma_A)) = i_2(X_A)\setminus i_2\circ i_1(\Sigma_A) =\{\xi\in i_2(X_A):|\kappa(\xi)|<\infty\}=\bigcup_{\alpha\in \mathcal{L}}\{\xi\in i_2(X_A):\kappa(\xi)=\alpha\},  
\end{equation*}
where $\mathcal{L}$ is the countable set of all admissible finite words. Note that
\begin{align*}
    H(\alpha):&=\{\xi\in i_2(X_A):\kappa(\xi) =\alpha\} = \{\xi\in i_2(X_A):\xi_\alpha=1 \text{ and } \xi_{\alpha s}=0,\text{ }\forall s\in \mathbb{N}\} \\
    &= \{\xi\in i_2(X_A):\xi_\alpha=1\} \cap \{\xi\in i_2(X_A): \xi_{\alpha s}=0,\text{ }\forall s\in \mathbb{N} \}.
\end{align*}
Since $H(\alpha)$ is an intersection of two closed sets, this means that $H(\alpha)$ is closed in $i_2(X_A)$.  As $i_2(X_A)\setminus i_2\circ i_1(\Sigma_A)$ is an countable union of those sets, we conclude $i_2(X_A)\setminus i_2\circ i_1(\Sigma_A)$ is a $F_\sigma$, hence $i_2\circ i_1(\Sigma_A)$ is a $G_\delta$, a Borel set. 
\end{proof}

\begin{proposition}\label{prop:Borel_sets_preserved_from_Sigma_A_to_X_A}For every Borel set $B \subseteq \Sigma_A$, $i_1(B)$ is a Borel set in $X_A$.
\end{proposition} 

\begin{proof} It is equivalent to prove that $i_2 \circ i_1(B)$ is a Borel set. Also, it is sufficient to prove the result for the cylinders in $\Sigma_A$ because they do form an countable basis of the topology. Given a cylinder set $[\alpha]\subseteq \Sigma_A$, we have that $i_2 \circ i_1([\alpha]) = \{\xi \in i_2(X_A): \kappa(\xi) \text{ is infinite and $\xi_\alpha=1$}\}$.

Denoting by $\pi_g$ the projection of a word $g$ in the Cayley tree it follows that
\begin{equation*}
    i_2 \circ i_1([\alpha]) = \left(\bigcap_{\nu \in \llbracket\alpha\rrbracket}\pi_{\nu}^{-1}(\{1\})\right)\cap i_2\circ i_1(\Sigma_A)
\end{equation*}
is a Borel set in $i_2(X_A)$. 
\end{proof}

Since the Borel sets of $\Sigma_A$ are preserved in the sense that they are also Borel sets of $X_A$ we are able to see Borel measures of $\Sigma_A$ as measures of $X_A$ that lives on $\Sigma_A$. Conversely, when we restrict Borel measures on $X_A$ to $\Sigma_A$ we obtain a Borel measure on $\Sigma_A$. This compatibility of Borel $\sigma$-algebras and measures is crucial for a stronger result: the compatibility of conformal measures on classical and generalized Markov shift spaces. This will be proved in chapter \ref{ch:TF_on_Generalized_Countable_Markov_shifts}.

We now introduce the generalized shift map for $X_A$, inherited from $\Sigma_A$.

\begin{definition}[Generalized shift map] Consider the generalized Markov shift space. The generalized shift map is the function $\sigma :X_A \setminus \{\xi \in X_A:\kappa(\xi) = e\} \to X_A$ defined by
\begin{equation*}
    (\sigma(\xi))_g = \xi_{x_0^{-1} g}, \quad g \in \mathbb{F},
\end{equation*}
where $x_0 = \kappa(\xi)_0$, that is, $x_0$ is the first letter of the stem of $\xi$. 
\end{definition}

Observe that $\sigma$ is simply a translation on the configurations of the Cayley tree of generated by $\mathbb{N}$. It is straightforward that such map is well-defined surjective continuous map. More than that, $\sigma$ is a local homeomorphism. Next, we define the $Y_A$-families for $X_A$, a notion that will be used later to characterize extremal conformal measures on $X_A$ that vanish on $\Sigma_A$, in the next chapter.

\begin{definition}[$Y_A$-family] Let $\{\xi^{0,\mathfrak{e}}\}_{\mathfrak{e} \in \mathscr{E}}$, the collection of all configurations on $X_A$ (or $\widetilde{X}_A$) that have empty stem. We define the $Y_A$-family of $\xi^{(0,\mathfrak{e})}$ as the set
\begin{equation}
    Y_A(\xi^{0,\mathfrak{e}}):= \bigsqcup_{n \in \mathbb{N}_0} \sigma^{-n}(\xi^{0,\mathfrak{e}}).
\end{equation}
\end{definition}

\begin{remark} When $\widetilde{X}_A$ is considered, $Y_A(\varphi_0)$ is a singleton. 
\end{remark}

We prove some properties of the $Y_A$-families.

\begin{proposition}[Properties of the $Y_A$-families] \label{prop:properties_of_Y_A_families} Consider the family $\{\xi^{0,\mathfrak{e}}\}_{\mathfrak{e} \in \mathscr{E}}$ of all distinct configurations in $Y_A$ with empty stem. The following properties hold:
\begin{itemize}
    \item[$(i)$] $Y_A(\xi^{0,\mathfrak{e}}) \cap Y_A(\xi^{0,\mathfrak{e}'}) = \emptyset$ whenever $\mathfrak{e} \neq \mathfrak{e}'$;
    \item[$(ii)$] $Y_A = \bigsqcup_{\mathfrak{e} \in \mathscr{E}}Y_A(\xi^{0,\mathfrak{e}})$;
    \item[$(iii)$] $Y_A(\xi^{0,\mathfrak{e}})$ is $\sigma$-invariant for all $\mathfrak{e} \in \mathscr{E}$, in the sense that
    \begin{equation*}
        \sigma(Y_A(\xi^{0,\mathfrak{e}}) \cap \Dom \sigma) = Y_A(\xi^{0,\mathfrak{e}});
    \end{equation*}
    \item[$(iv)$] fixed $\mathfrak{e} \in \mathscr{E}$, $R_\xi(\kappa(\xi)) = R_{\xi^{0,\mathfrak{e}}}(\kappa(\xi^{0,\mathfrak{e}}))$ for every $\xi \in Y_A(\xi^{0,\mathfrak{e}})$;
    \item[$(v)$] $Y_A(\xi^{0,\mathfrak{e}})$ is countable for every $\mathfrak{e} \in \mathscr{E}$;
    \item[$(vi)$] $Y_A(\xi^{0,\mathfrak{e}})$ is a Borel subset of $X_A$ for every $\mathfrak{e} \in \mathscr{E}$.
\end{itemize}
\end{proposition}

\begin{proof} We divide the proof accordingly to the properties labeled in the statement.
\begin{itemize}
    \item[$(i)$] Let $\xi \in Y_A(\xi^{0,\mathfrak{e}}) \cap Y_A(\xi^{0,\mathfrak{e}'})$, where $\mathfrak{e} \neq \mathfrak{e}'$. Then $\xi^{0,\mathfrak{e}} = \sigma^n(\xi) = \xi^{0,\mathfrak{e}'} $, where $n = |\kappa(\xi)|$, a contradiction because $\xi^{0,\mathfrak{e}} \neq \xi^{0,\mathfrak{e}'}$.
    \item[$(ii)$] The inclusion
    \begin{equation*}
        Y_A \supseteq \bigsqcup_{\mathfrak{e} \in \mathscr{E}}Y_A(\xi^{0,\mathfrak{e}})
    \end{equation*}
    is straightforward. For the converse, take $\xi \in Y_A$. If $\kappa(\xi) = e$, there is nothing to be proven, so suppose $\kappa(\xi) = \kappa \neq e$. Then $\sigma^{\kappa}(\xi) = \xi^{0,\mathfrak{e}}$ for some $\mathfrak{e} \in \mathscr{E}$ and hence $\xi \in Y_A(\xi^{0,\mathfrak{e}})$, proving the inclusion
    \begin{equation*}
        Y_A \subseteq \bigsqcup_{\mathfrak{e} \in \mathscr{E}}Y_A(\xi^{0,\mathfrak{e}}).
    \end{equation*}
    \item[$(iii)$] Given $\xi \in Y_A(\xi^{0,\mathfrak{e}}) \cap \Dom \sigma$ and since the translation $\sigma$ preserves the root of the stem, we have that $\sigma(\xi) \in Y_A(\xi^{0,\mathfrak{e}})$.
    \item[$(iv)$] it is straightforward from definition of $Y_A(\xi^{0,\mathfrak{e}})$.
    \item[$(v)$] $Y_A(\xi^{0,\mathfrak{e}})$ is at most a countable union of countable sets and therefore it is countable.
    \item[$(vi)$] Since $X_A$ is Hausdorff, every singleton is closed and then it is a Borel subset of $X_A$. By $(v)$, each $Y_A(\xi^{0,\mathfrak{e}})$ is a countable union of Borel sets, and therefore the $Y_A$-families are Borel subsets of $X_A$.
\end{itemize}
\end{proof}

By conditions $(i)$ and $(ii)$ we have that the collection of $Y_A$-families forms a partition of $Y_A$. By $(iv)$, two distinct configurations in a same $Y_A$-family necessarily have distinct stems. Then, for fixed empty stem configuration $\xi^{0,\mathfrak{e}}$, there exists a bijection between the all possible stems in configurations of $Y_A(\xi^{0,\mathfrak{e}})$ and their respective configurations. Considering this fact, we define the following.

\begin{definition} Given a transition matrix s.t. its respective set of empty stem configurations $\{\xi^{0,\mathfrak{e}}\}_{\mathfrak{e} \in \mathscr{E}}$ in $Y_A$. We define the sets
\begin{equation*}
    \mathfrak{R}_\mathfrak{e}:= \{\omega: \kappa(\xi) = \omega \text{ for some } \xi \in Y_A(\xi^{0,\mathfrak{e}})\},
\end{equation*}
for each $\mathfrak{e} \in \mathscr{E}$. If $\mathscr{E} = \{\mathfrak{e}\}$, that is, $\mathscr{E}$ is a singleton, we simply write $\mathfrak{R}_\mathfrak{e} = \mathfrak{R}$. 
\end{definition}

We now focus on some particularly special examples. They will be recalled further when we construct the Thermodynamic Formalism on generalized Markov shift spaces, since in these cases we explore concrete examples of a new type of phase transition that is not detected in the classical theory.

\subsection{The Generalized Renewal shift}
\label{subsec:Generalized_Renewal_shift}

We first characterize the generalized renewal shift space. By taking the alphabet $\mathbb{N}$, we recall the renewal transition matrix, presented in Example \ref{exa:renewal_shift} of chapter \ref{ch:Markov_shift_space}: $A(1,n) = A(n+1,n) = 1$ for every $n \in \mathbb{N}$ and zero in the remaining entries. From the same example, we also recall its symbolic graph:
\[
\begin{tikzcd}
\circled{1}\arrow[loop left]\arrow[r,bend left]\arrow[rr,bend left]\arrow[rrr,bend left]\arrow[rrrr, bend left]&\circled{2}\arrow[l]&\circled{3}\arrow[l]&\circled{4}\arrow[l]&\arrow[l]\cdots
\end{tikzcd}
\]
First we observe that the C$^*$-algebra $\mathcal{O}_A$ for the renewal matrix is unital. In fact, in the representation of $\mathcal{O}_A$ on $\mathfrak{B}(\ell^2(\Sigma_A))$, Remark \ref{remark:O_A_non_unital} gives
\begin{equation*}
    1 = Q_1 = T_1^* T_1 \in \mathcal{O}_A
\end{equation*}
Then we have that $\mathcal{O}_A$ is unital and by Theorem \ref{thm:O_A_unital} we conclude that $X_A$ is compact\footnote{Alternatively we also can prove that $X_A$ by noticing that $A(\emptyset,\{1\},j) = 0$ for every $j \in \mathbb{N}$ and again by using Theorem \ref{thm:O_A_unital}.}. On the other hand, we observe that $1$ is the unique infinite emitter in the symbolic graph. Observe that
\begin{equation*}
    \mathfrak{c}(j)_i = \begin{cases}
                1, \quad \text{if } i \in \{1,j+1\};\\
                0, \quad \text{otherwise.}
             \end{cases}
\end{equation*}
Then, the unique accumulation point of the set of matrix columns of $A$ is
\begin{equation*}
    c = 
.
\end{equation*}
We conclude that there exists a unique $Y_A$-family, where 
\begin{equation*}
    \mathfrak{R} = \{e\} \sqcup \{\omega \text{positive admissible word}, |\omega| \geq 1, \omega_{|\omega|-1}= 1\}.
\end{equation*}
We denote by $\xi^0$ the unique configuration in $Y_A$ with empty stem, and it is represented in the figure \ref{figure.emptyconfigurationrenewal}.

\begin{figure}[H]
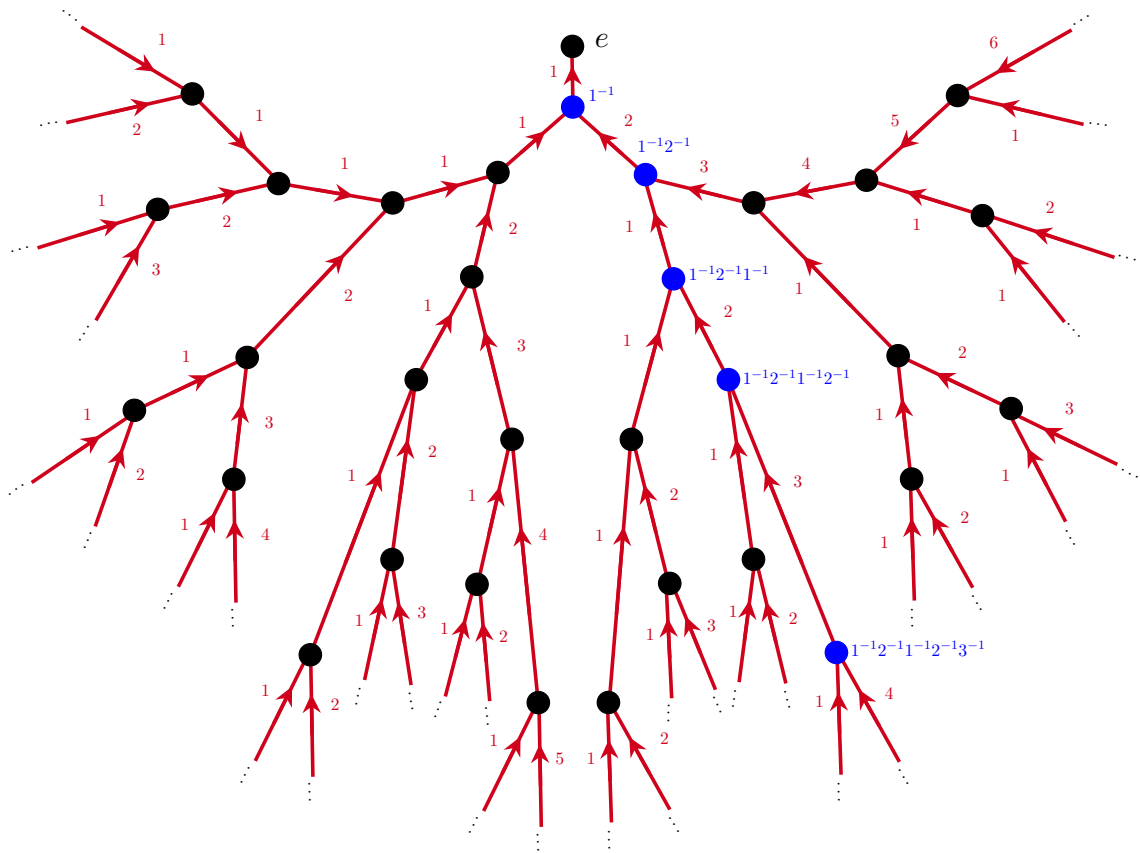

    \centering
\scalebox{0.6}{
\tikzset{every picture/.style={line width=0.75pt}} 

%

}
\caption{The empty stem configuration of the renewal shift. Only the filled vertices are shown. The blue vertices are some explicit examples of filled elements of $\mathbb{F}$.\label{figure.emptyconfigurationrenewal}}
\end{figure}

The configuration of the figure \ref{fig:321_renewal} presents the unique configuration relative to the stem $321$. Its uniqueness is a straightforward consequence of Proposition \ref{prop:properties_of_Y_A_families}.

\begin{figure}[H]
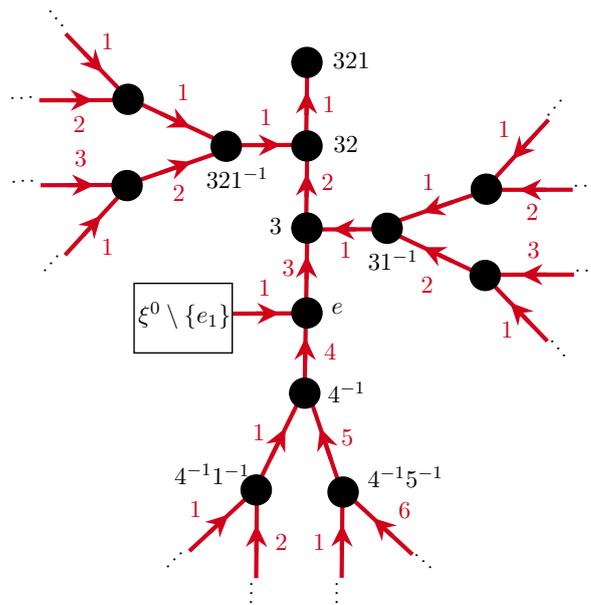

    \centering

\scalebox{0.7}{
\tikzset{every picture/.style={line width=0.75pt}} 



}
\caption{321 stem configuration for the renewal shift. The directed graph denoted by $\xi^0\setminus \{e_1\}$ is 
obtained removing the root $e$ and the edge $e_1$, labeled with $1$, from the directed graph associated to the 
empty configuration $\xi^0$ in the figure \ref{figure.emptyconfigurationrenewal}. \label{fig:321_renewal}}
\end{figure}

Since this shift space is a simple concrete and explicit example, we show explicitly how the shift map acts on it. The figure \ref{fig:321_renewal_shift_action} shows the shift action on the configuration of the figure \ref{fig:321_renewal}, and the figure \ref{fig:321_renewal_shift_action_detailed} presents the same shift action in the same configuration with focus on where each of branch disemboguing to the stem is translated.

\begin{figure}[H]
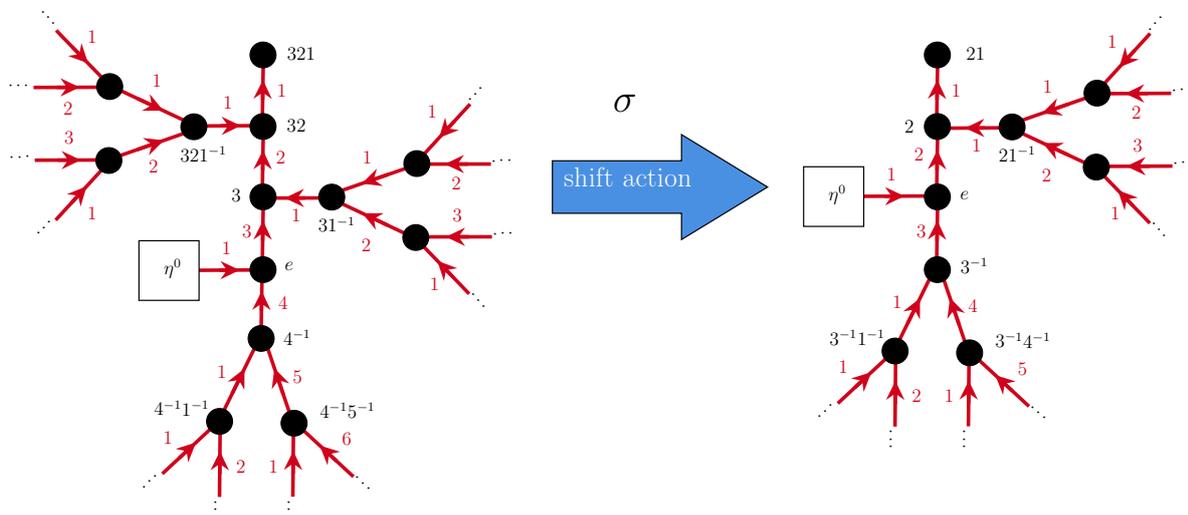

    \centering

\scalebox{0.6}{

\tikzset{every picture/.style={line width=0.75pt}} 



}
\caption{Shift action on the 321 stem configuration for the renewal shift, focusing on the transportation of the branches disemboguing in the stem. The directed graph denoted by $\xi^0\setminus \{e_1\}$ is 
obtained removing the root $e$ and the edge $e_1$, labeled with $1$, from the directed graph associated to the 
empty configuration $\xi^0$ in the figure 3. \label{fig:321_renewal_shift_action}.}
\end{figure}

\begin{figure}[H]
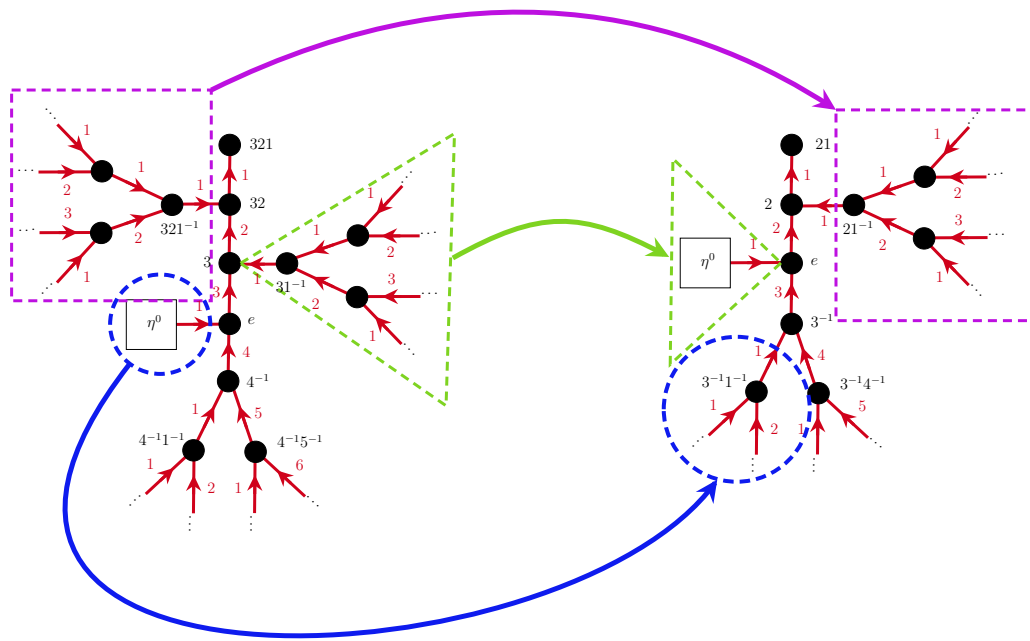

    \centering

\scalebox{0.5}{

\tikzset{every picture/.style={line width=0.75pt}} 

%

}
\caption{Shift action on the 321 stem configuration for the renewal shift. The directed graph denoted by $\xi^0\setminus \{e_1\}$ is obtained removing the root $e$ and the edge $e_1$, labeled with $1$, from the directed graph associated to the 
empty configuration $\xi^0$ in the figure \ref{figure.emptyconfigurationrenewal}. Each displayed branch is translated under $\sigma$ accordingly to the shape that sorrounds it, i.e., the circled in the left configuration goes to the circled one in the right configuration, and the analogous transportation occurs for the squared and triangled branches. \label{fig:321_renewal_shift_action_detailed}.}

\end{figure}

The next proposition is the precise counting number of elements in $\mathfrak{R}$ with length $n$, that is, the number of elements of $Y_A$ with stem of length $n$, for each $n \in \mathbb{N}_0$.

\begin{proposition}\label{prop:cardinality_words_length} There is exactly one element in $Y_A$ whose stem has
length zero, namely $\xi^0$, and for each $n \geq 1$, there are exactly
$2^{n-1}$ elements in $Y_A$ whose stem has length $n$.
\end{proposition}

\begin{proof} Every stem of a configuration in $Y_A$ is admissible and it ends with $1$. Hence, for any $n \in \mathbb{N}$ we have that
\begin{equation*}
    \{\omega\xi^0 \in Y_A: |\omega| = n\} = \sigma^{-(n-1)}(1),
\end{equation*}
where $\sigma$ is the shift map restricted to $Y_A$. Since the transition matrix is given by $A(s+1,s) = A(1,s) = 1$ for all $s \in \mathbb{N}$ and zero in the rest of entries we conclude that $|\sigma^{-1}(\eta)| = 2$ for all $\eta \in Y_A$. It follows that
\begin{equation*}
    |\{\omega\xi^0 \in Y_A: |\omega| = n\}| = |\sigma^{-(n-1)}(1)| = 2^{n-1}.
\end{equation*}
Indeed, it is obvious for $n=1$. Then suppose that the result above follows for $k \in \mathbb{N}$, i.e.,
\begin{equation*}
    |\sigma^{-(k-1)}(1)| = 2^{k-1},
\end{equation*}
then for a word $\eta$ in $\sigma^{-(k-1)}(1)$ we have that $|\sigma^{-1}(\eta)| = 2$ and then
\begin{equation*}
    |\sigma^{-k}(1)| = |\sigma^{-1}(\sigma^{-(k-1)}(1))| = 2 |\sigma^{-(k-1)}(1)| = 2^{k}.
\end{equation*}
The proof is complete. 
\end{proof}

\subsection{The Pair Renewal shift}
\label{subsec:Pair_Renewal}

This shift space corresponds to a slightly modification on the renewal shift, by adding another infinite emmiter to the symbolic graph, and this is described now. Consider the transition matrix satisfying
\begin{equation*}
    A(1,n)=A(2,2n)=A(n+1,n) = 1, \forall n\in \mathbb{N}
\end{equation*}
and zero otherwise. This shift will be called from now by \emph{Pair Renewal Shift} and its explicit associated transition matrix is
\begin{equation*}
    A = \begin{pmatrix}
    1 & 1 & 1 & 1 & 1 & 1 & 1 & 1 & 1 & \cdots\\
    1 & 1 & 0 & 1 & 0 & 1 & 0 & 1 & 0 & \cdots\\
    0 & 1 & 0 & 0 & 0 & 0 & 0 & 0 & 0 & \cdots\\
    0 & 0 & 1 & 0 & 0 & 0 & 0 & 0 & 0 & \cdots\\
    0 & 0 & 0 & 1 & 0 & 0 & 0 & 0 & 0 &\cdots\\
    0 & 0 & 0 & 0 & 1 & 0 & 0 & 0 & 0 &\cdots\\
    0 & 0 & 0 & 0 & 0 & 1 & 0 & 0 & 0 &\cdots\\
    0 & 0 & 0 & 0 & 0 & 0 & 1 & 0 & 0 &\cdots\\
    0 & 0 & 0 & 0 & 0 & 0 & 0 & 1 & 0 &\cdots\\
    \vdots & \vdots & \vdots & \vdots & \vdots & \vdots & \vdots & \vdots & \vdots & \ddots
    \end{pmatrix}.
\end{equation*}
Its symbolic graph is the following:
\[
\begin{tikzcd}
\circled{1}\arrow[loop left]\arrow[r,bend left]\arrow[rr,bend left]\arrow[rrr,bend left]\arrow[rrrr, bend left]\arrow[rrrrr, bend left]\arrow[rrrrrr, bend left]&\circled{2}\arrow[loop below]\arrow[l]\arrow[rr,bend right]\arrow[rrrr,bend right]\arrow[rrrrr,bend right]&\circled{3}\arrow[l]&\circled{4}\arrow[l]&\circled{5}\arrow[l]&\circled{6}\arrow[l]&\arrow[l]\cdots
\end{tikzcd}
\]
The current transition matrix $A$ has a full line of $1$`s in the first row and then $\mathcal{O}_A$ is unital as the renewal shift and therefore $X_A$ is compact. The only two possibilities for accumulation points on the sequence of columns of $A$ are
\begin{equation}\label{eq:limit_columns_pair_renewal_shift}
    c_1 = \begin{pmatrix} 1 \\ 1 \\0 \\ 0 \\ 0 \\\vdots \end{pmatrix} \quad \text{and} \quad c_2 = \begin{pmatrix} 1 \\ 0 \\0 \\ 0 \\ 0 \\\vdots \end{pmatrix}.
\end{equation}
By Proposition \ref{prop:characterization_of_elements_of_Y_A}, there are only two distinct emtpty stem configurations, as it is shown in the figure \ref{fig:emptyconfigurationpairrenewal}, and then there are only two $Y_A$-families. Observe that the unique infinite emmiters the pair renewal shift symbolic graph are the symbols $1$ and $2$. Moreover, note that $(\mathfrak{c}(j))_{A(1,j)=1} = (\mathfrak{c}(j))_{\mathbb{N}}$, and then that sequence has $c_1$ and $c_2$ as accumulation points, on the other hand $(\mathfrak{c}(j))_{A(2,j)=1}$ converges to $c_1$, so again by Proposition \ref{prop:characterization_of_elements_of_Y_A} we determine the non-empty stem configurations in $Y_A$: for $\xi \in Y_A$, with $\kappa(\xi) = \omega \neq e$ we have that one of the following holds:
\begin{enumerate}
    \item $\omega$ ends in $1$ and $R_\xi(\omega) = c_1$;
    \item $\omega$ ends in $2$ and $R_\xi(\omega) = c_1$;
    \item $\omega$ ends in $1$ and $R_\xi(\omega) = c_2$.
\end{enumerate}
Alternatively, one could prove the above by applying Proposition \ref{prop:properties_of_Y_A_families} (i)-(ii) instead. Indeed, each empty stem generates a unique $Y_A$-family and it is straightforward that every configuration with no empty stem of $Y_A$ belongs to one of these families. The two first items correspond precisely to the elements of $Y_A(\xi^{0,1})$, where $\xi^{0,1}$ is the empty stem configuration with $\kappa_{\xi^{0,1}}(e) = c_1$; the last condition corresponds exactly to the elements of $Y_A(\xi^{0,2})$ with non-empty stem, where $\xi^{0,2}$ is the remaining empty-stem configuration. We have that
\begin{align*}
    \mathfrak{R}_1 &= \{e\} \sqcup \{\omega \text{ positive admissible word}, |\omega| \geq 1, \omega_{|\omega|-1} \in \{1,2\}\},\\
    \mathfrak{R}_2 &= \{e\} \sqcup \{\omega \text{ positive admissible word}, |\omega| \geq 1, \omega_{|\omega|-1} = 1\}.
\end{align*}

\begin{figure}[H]
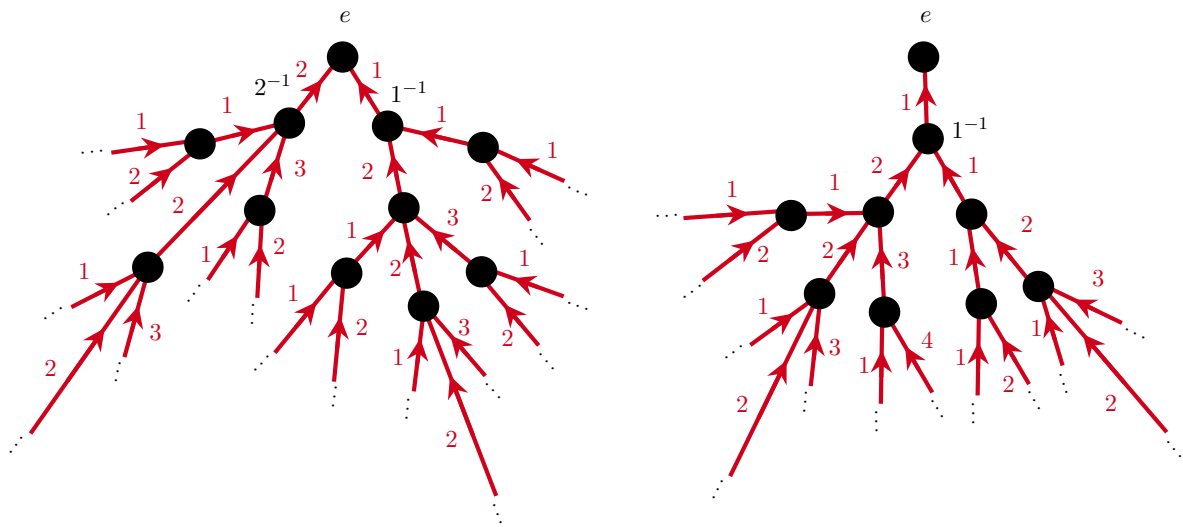

 \caption{The two empty stem configurations of the pair renewal shift. The left configuration $\xi^{0,1}$ satisfies $R_{\xi^{0,1}}(e) = c_1$, while the right one $\xi^{0,2}$ corresponds to $R_{\xi^{0,2}}(e) = c_2$, where $c_1$ and $c_2$ are the column vectors of \eqref{eq:limit_columns_pair_renewal_shift}. \label{fig:emptyconfigurationpairrenewal}}
 \scalebox{.7}{

\tikzset{every picture/.style={line width=0.75pt}} 



}
\end{figure}

The next result theorem gives, for each $k\in\{1,2\}$, the number of configurations in $Y_A(\xi^{0,k})$ with stem of same fixed length. Such theorem will be crucial in the next chapter, when we study phase transition phenomena for the pair renewal shift. In order to prove the next theorem, we say that an edge $c$ that conects two filled vertices of any empty stem configuration is in the $n$-th generation, $n \in \mathbb{N}$ if the unique admissible path that connects this edge to the vertex $e$ has $n$ edges including $c$.

\begin{theorem}\label{thm:counting_configurations_pair_renewal_shift}
 
Consider the pair renewal shift and $n \in \mathbb{N}_0$. Then,
\begin{align*}
    |\sigma^{-n}(\xi^{0,1})| = \frac{1}{4}\left[(1-\sqrt{2})^{n} + (1-\sqrt{2})^{n+1} + (1+\sqrt{2})^{n} + (1+\sqrt{2})^{n+1}\right]
\end{align*}
and
\begin{align*}
|\sigma^{-n}(\xi^{0,2})| = \begin{cases}
                                 1, \quad \text{if $n=0$};\\
                                 \frac{1}{4}\left[(1-\sqrt{2})^{n-1} + (1-\sqrt{2})^{n} + (1+\sqrt{2})^{n-1} + (1+\sqrt{2})^{n}\right], \quad \text{otherwise}.
                              \end{cases}
\end{align*}
\end{theorem}
\begin{proof} For $n = 0$ the result is straightforward in both cases, as well as in the case $n=1$ for $\xi^{0,2}$. Observe that the configuration $\xi^{0,2}$ has the same graph tree in terms of labeled edges than $\xi^{0,1}$ if we remove the vertex $e$ and the edge from $1^{-1}$ to $e$. Then, for $n \geq 1$, we have $|\sigma^{-n}(\xi^{0,2})| = |\sigma^{-(n-1)}(\xi^{0,1})|$, and therefore it is sufficient to prove for the configuration $\xi^{0,1}$. Moreover, the number of configurations in $\sigma^{-n}(\xi^{0,1})$ is the number of edges of $n$-th generation. This counting problem depends of the number of edges that can connect to a labeled edge. By the transition matrix, we have the following:
\begin{itemize}
    \item[$(i)$] if an edge is labeled with an odd number, then there are two edges that connect directly to him: one with an odd label, and another one with an even label;
    \item[$(ii)$] if an edge is labeled with an even number, then there are three edges that connect directly to him: two with odd label, and another one with even label.
\end{itemize}
We can define an iterated map to count the number of edges as follows. Let $G_k$ be the set of edges of generation $k \in \mathbb{N}$, and we make no distinction among all of its odd labeled edges, and we do the same for all their even labeled edges. We consider $\mathbb{R}^2$ endowed with the $\ell^1$-norm, that is, $\|(x_1,x_2)\| = |x_1| + |x_2|$ for every $(x_1,x_2) \in \mathbb{R}^2$. We associate the odd and even labeled edges to the respective canonical basis vectors
\begin{align*}
    e_1 = \begin{pmatrix}
           1\\
           0
          \end{pmatrix}
\quad \text{and} \quad
    e_2 = \begin{pmatrix}
           1\\
           0
          \end{pmatrix}.
\end{align*}
In this setting, the number of elements of $G_k$ is
\begin{equation*}
    \left\|\begin{pmatrix}
           r\\
           s
          \end{pmatrix}\right\|
    = r\left\|\begin{pmatrix}
           1\\
           0
          \end{pmatrix}\right\|
    + s \left\|\begin{pmatrix}
           0\\
           1
          \end{pmatrix}\right\| = r + s,
\end{equation*}
where $r,s\in \mathbb{N}_0$ are respectively the number of odd and even edges in $G_k$. By the characterization done in items $(i)$ and $(ii)$ above, we have that the number of elements in $G_{k+1}$ depends on the number of edges in $G_k$, and how many of these edges are odd and how many are even, since they generate a different number of edges. Furthermore the number of odd and even edges in $G_{k+1}$ is given respectively by $r+2s$ and $r+s$. In terms of vectors in $\mathbb{R}^2$, we have that the total of edges is given by
\begin{equation*}
    \left\|\begin{pmatrix}
           r+2s\\
           r+s
          \end{pmatrix}\right\|
    = \left\|\begin{pmatrix}
           1 & 2\\
           1 & 1
             \end{pmatrix}
            \begin{pmatrix}
            r\\
            s
            \end{pmatrix}\right\|.
\end{equation*}
Note that the matrix 
\begin{equation*}
    M = \begin{pmatrix}
            1 & 2\\
            1 & 1
        \end{pmatrix}
\end{equation*}
does not depend of $k$, hence we have that we can count the number of edges by iterating $M$, that is, the number of elements in $G_k$ is
\begin{equation}\label{eq:cardinality_of_G_k}
    \left\|M^{k-1}\begin{pmatrix}
                   r_0\\
                   s_0
                  \end{pmatrix}\right\|,
\end{equation}
where $r_0, s_0 \in \mathbb{N}_0$ are respectively the number of odd and even edges in $G_1$, which in our case we have $r_0 = s_0 = 1$. It is straightforward that $M$ can be written in diagonal form, since its determinant is $-1$, so we diagonalize to make the calculations easier to obtain the quantity in \ref{eq:cardinality_of_G_k}. The characteristic polynomial of $M$ is
\begin{equation*}
    p(\lambda) = \det (M-\lambda \mathbbm{1}) = \lambda^2 - 2 \lambda - 1,
\end{equation*}
then the eigenvalues of $M$ are
\begin{equation*}
    \lambda_1 = 1-\sqrt{2} \quad \text{and} \quad \lambda_1 = 1+\sqrt{2}
\end{equation*}
and their respective associate eigenvectors are
\begin{align*}
    v_1 = \begin{pmatrix}
            -\sqrt{2}\\
            1
          \end{pmatrix}
\quad \text{and} \quad
    v_2 = \begin{pmatrix}
            \sqrt{2}\\
            1
          \end{pmatrix}.
\end{align*}
Now, observe that
\begin{align*}
    v_1 = -\sqrt{2} e_1 + e_2 \quad \text{and} \quad v_2 = \sqrt{2} e_1 + e_2,
\end{align*}
and conversely
\begin{align*}
    e_1 = -\frac{1}{2\sqrt{2}} v_1 + \frac{1}{2\sqrt{2}} v_2 \quad \text{and} \quad e_2 = \frac{1}{2} v_1 + \frac{1}{2} v_2.
\end{align*}
By representing $\mathbb{R}^2_{\mathcal{C}}$ and $\mathbb{R}^2_{\mathcal{E}}$ for the vector space $\mathbb{R}^2$ on the canonical and eigenvector basis respectively, we have the following change of basis matrices:
\begin{align*}
    P = \begin{pmatrix}
            -\frac{1}{2\sqrt{2}} & \frac{1}{2}\\
            \frac{1}{2\sqrt{2}} & \frac{1}{2}
        \end{pmatrix}
\quad \text{and} \quad
    P^{-1} = \begin{pmatrix}
                -\sqrt{2} & \sqrt{2}\\
                1 & 1
             \end{pmatrix},
\end{align*}
where $P: \mathbb{R}^2_{\mathcal{C}} \to \mathbb{R}^2_{\mathcal{E}}$ and $P^{-1}: \mathbb{R}^2_{\mathcal{E}} \to \mathbb{R}^2_{\mathcal{C}}$. We denote the diagonal form of $M$ by $D$, hence
\begin{equation*}
    M = P^{-1} D P
\end{equation*}
and then for $n \in \mathbb{N}$ we have
\begin{align*}
    |\sigma^{-n}(\xi^{0,1})| = |G_n| = \left\|M^{n-1}\begin{pmatrix}
                   1\\
                   1
                  \end{pmatrix}\right\|
    = \left\|P^{-1} D^{n-1} P \begin{pmatrix}
                            1\\
                            1
                        \end{pmatrix}\right\|
    = \frac{1}{2\sqrt{2}}\left\|\begin{pmatrix}
                \sqrt{2}(R_1+R_2)\\
                -R_1+R_2
             \end{pmatrix}\right\|,
\end{align*}
where
\begin{equation*}
    R_1 = (1-\sqrt{2})^{n} \quad \text{and} \quad R_2 = (1+\sqrt{2})^{n}.
\end{equation*}
It is straightforward to observe that $R_2 \geq |R_1|$, and then
\begin{equation}\label{eq:counting_formula_xi_1_pair_renewal_shift}
    |\sigma^{-n}(\xi^{0,1})| = \frac{1}{4}\left[(1-\sqrt{2})^{n} + (1-\sqrt{2})^{n+1} + (1+\sqrt{2})^{n} + (1+\sqrt{2})^{n+1}\right].
\end{equation}
Observe that the formula \ref{eq:counting_formula_xi_1_pair_renewal_shift} also holds for $n=0$. Consequently we have that
\begin{align*}
|\sigma^{-n}(\xi^{0,2})| = \frac{1}{4}\left[(1-\sqrt{2})^{n-1} + (1-\sqrt{2})^{n} + (1+\sqrt{2})^{n-1} + (1+\sqrt{2})^{n}\right],
\end{align*}
for $n \in \mathbb{N}$.
\end{proof}

\subsection{Prime Renewal shift}
\label{subsec:Prime_Renewal_shift}

This example is one step furhter on generalizating the renewal shift. Take matrix $A$ as follows: for each $p$ prime number and  $n \in \mathbb{N}$ we have
\begin{equation*}
A(n+1,n)=A(1,n) = A(p,p^n) = 1
\end{equation*}
and zero for the other entries of $A$. From now on, we will call this shift \emph{Prime Renewal Shift}.
The transition matrix is
\begin{equation*}
    A = \begin{pmatrix}
    1 & 1 & 1 & 1 & 1 & 1 & 1 & 1 & 1 & \cdots\\
    1 & 1 & 0 & 1 & 0 & 0 & 0 & 1 & 0 & \cdots\\
    0 & 1 & 1 & 0 & 0 & 0 & 0 & 0 & 1 & \cdots\\
    0 & 0 & 1 & 0 & 0 & 0 & 0 & 0 & 0 & \cdots\\
    0 & 0 & 0 & 1 & 1 & 0 & 0 & 0 & 0 &\cdots\\
    0 & 0 & 0 & 0 & 1 & 0 & 0 & 0 & 0 &\cdots\\
    0 & 0 & 0 & 0 & 0 & 1 & 1 & 0 & 0 &\cdots\\
    0 & 0 & 0 & 0 & 0 & 0 & 1 & 0 & 0 &\cdots\\
    0 & 0 & 0 & 0 & 0 & 0 & 0 & 1 & 0 &\cdots\\
    \vdots & \vdots & \vdots & \vdots & \vdots & \vdots & \vdots & \vdots & \vdots & \ddots
    \end{pmatrix}.
\end{equation*}
Now, consider the space of the columns of $A$ endowed with the product topology. The possible accumulation points of sequences of elements of $\mathfrak{C}(A)$ are
\begin{align*}
    c(p) = \begin{pmatrix}
               1\\
               0\\
               0\\
               \vdots\\
               0\\
               1 \quad(p\text{-th coordinate})\\
               0\\
                  \vdots
           \end{pmatrix}
\quad \text{or} \quad
    c(1)  = \begin{pmatrix}
               1\\
               0\\
               0\\
               0\\
               \vdots
             \end{pmatrix}.
\end{align*}
where $p$ is a prime number. For each $p$, one may take the sequence $\{\mathfrak{c}(p^n)\}_\mathbb{N}$ that converges to $c(p)$, and the sequence $\{\mathfrak{c}(p_n)\}_{p_n \text{is the $n$-th prime}}$ converges to $1$.
The prime renewal shift has a transition matrix with a full line of $1$'s in the first row and then $\mathcal{O}_A$ is unital as the renewal shift and therefore $X_A$ is compact. The next proposition is a straightforward consequence of Proposition \ref{prop:characterization_of_elements_of_Y_A}.


\begin{proposition}\label{prop:characterization_Y_A_prime_renewal_shift} Consider the prime renewal shift. The elements of $Y_A$ has as stems the finite words ending with $1$ or any prime number $p$, otherwise they have empty stem. Moreover,
\begin{itemize}
    \item[$(i)$] for each finite positive admissible word $\omega$ ending in a prime number $p$, there exists exactly one configuration in $Y_A$ with stem $\omega$;
    \item[$(ii)$] for each $k \in \{1\}\cup\{p \in \mathbb{N}:p\text{ is a prime number}\}$ and each $\omega$ positive admissible word ending in $1$, there exists exactly one configuration with stem $\omega$ s.t. $\omega j^{-1}$ is filled only for $j \in \{1,k\}$;
    \item[$(iii)$] for each $k \in \{1\}\cup\{p \in \mathbb{N}:p\text{ is a prime number}\}$, there exists exactly one configuration with empty stem such that $j^{-1}$ is filled only for $j \in \{1,k\}$.
\end{itemize}
\end{proposition}

We denote by $\xi^{0,p}$ the configuration with empty stem such that $\xi^{0,p}_{1^{-1}}=\xi^{0,p}_{p^{-1}}=1$, where $p$ is prime or $1$. By Proposition \ref{prop:characterization_Y_A_prime_renewal_shift}, it is straightforward that the $Y_A$-families $Y_A(\xi^{0,p})$ and are characterized by the sets of words
\begin{align*}
    \mathfrak{R}_1 &= \{e\} \sqcup \{\omega \text{ positive admissible word}, |\omega| \geq 1, \omega_{|\omega|-1} = 1\}, \text{ if }p =1;\\
    \mathfrak{R}_p &= \{e\} \sqcup \{\omega \text{ positive admissible word}, |\omega| \geq 1, \omega_{|\omega|-1} \in \{1,p\}\}, \text{ if }  p \text{ is a prime number}.
\end{align*}
Now we estimate, for every $Y_A$-family above, the number of elements of stem with length $n$, for each $n \in \mathbb{N}$.

\begin{proposition}\label{prop:control_number_configurations} For the prime renewal shift we have that
\begin{equation*}
    2^{n-1}\leq |\{\xi \in Y_A(\xi^{0,p}): |\kappa(\xi)|=n\}|\leq 3^n, \quad \forall n \in \mathbb{N},
\end{equation*}
for every $p \in \{1\}\cup\{q \in \mathbb{N}: q \text{ is a prime number}\}$.
\end{proposition}

\begin{proof} Given $j \in \mathbb{N}$, let us estimate $|\{i \in \mathbb{N}:A(i,j) = 1\}|$. We have two possibilities:
\begin{enumerate}
    \item $j = p^m$, with $m \in \mathbb{N}$ and $p$ prime number, and in this case $A(i,j) = 1$ if and only if $i = p^m+1$, $i = p$ or $i = 1$;
    \item Otherwise, $A(i,j) = 1$ if and only if $i = j+1$ or $i = 1$.
\end{enumerate}
If $p = 1$, then there exists only one $\xi$ in $Y_A(\xi^0(1))$ such that $|\kappa(\xi)| = 1$. Hence, by the possibilities shown above, we have
\begin{equation*}
    2^{n-1}\leq|\sigma^{-n}(\xi^0(1))| \leq 3^{n-1}, \quad \forall n \in \mathbb{N}.
\end{equation*}
For $p$ prime number, we have a similar calculation but with the difference that there exist exactly two different $\xi$ in $Y_A(\xi^{0,p})$ such that $|\kappa(\xi)| = 1$, therefore $
    2^{n}\leq|\sigma^{-n}(\xi^{0,p})| \leq 3^{n}, \ \forall  n \in \mathbb{N}$.
Then,
\begin{equation*}
    2^{n-1}\leq |\{\xi \in Y_A(\xi^{0,p}): |\kappa(\xi)|=n\}|\leq 3^n, \quad \forall n \in \mathbb{N},
\end{equation*}
for every $p \in \{1\}\cup\{q \in \mathbb{N}: q \text{ is a prime number}\}$.
\end{proof}

\subsection{Alternating Renewal shift}

We finish this section with an example of a matrix $A$ generating an unital Exel-Laca algebra $\mathcal{O}_A$ that has no full line of $1$'s, in contrast with the previous examples of the renewal class. Consider the matrix $A$ defined by the entries
\begin{equation*}
    A(1,2n) = A(2,2n-1) = A(n+1,n) = 1, \quad \text{ for all } n \in \mathbb{N},
\end{equation*}
and zero for the remaining entries. The matrix $A$ is the following:
\begin{equation*}
    A = \begin{pmatrix}
    0 & 1 & 0 & 1 & 0 & 1 & 0 & \cdots\\
    1 & 0 & 1 & 0 & 1 & 0 & 1 & \cdots\\
    0 & 1 & 0 & 0 & 0 & 0 & 0 & \cdots\\
    0 & 0 & 1 & 0 & 0 & 0 & 0 & \cdots\\
    0 & 0 & 0 & 1 & 0 & 0 & 0 &\cdots\\
    0 & 0 & 0 & 0 & 1 & 0 & 0 &\cdots\\
    0 & 0 & 0 & 0 & 0 & 1 & 0 &\cdots\\
    \vdots & \vdots & \vdots & \vdots & \vdots & \vdots & \vdots & \ddots
    \end{pmatrix}.
\end{equation*}
Although $A$ has not a full line of 1's, the corresponding Exel-Laca algebra is unital. Indeed, we have that $Q_1 + Q_2 = 1$, and therefore
\begin{equation*}
    1 \in \overline{\spann\left\{ \begin{array}{l l}
         & F \text{ finite}; \text{ }\alpha, \beta \text{ finite admissible words};\\
         T_\alpha \left(\prod_{i \in F}Q_i\right)T_\beta^*: &F \neq \emptyset \text{ or } \alpha \text{ is not the empty word} \\
         &\text{or }  \beta \text{ is not an empty word}
    \end{array}\right\} } \simeq \mathcal{O}_A,
\end{equation*}
that is, $\mathcal{O}_A$ is unital, and hence $X_A$ is compact. The graph associated to the matrix $A$ has 2 infinite emitters and therefore $\Sigma_A$ is not locally compact.  The only two possibilities for accumulation point on the space of columns of $A$ are
\begin{equation*}
    c_1 = \begin{pmatrix} 1 \\ 0 \\0 \\ 0 \\ 0 \\\vdots \end{pmatrix} \quad \text{and} \quad c_2 = \begin{pmatrix} 0 \\ 1 \\0 \\ 0 \\ 0 \\\vdots \end{pmatrix}.
\end{equation*}
By Proposition \ref{prop:characterization_of_elements_of_Y_A}, the possible configurations in $\xi \in Y_A$ are precisely following ones:
\begin{enumerate}
    \item $\kappa(\xi)$ ends in `$1$' or $e$ and $R_\xi(\kappa(\xi)) = c_1$;
    \item $\kappa(\xi)$ ends in `$2$' or $e$ and $R_\xi(\kappa(\xi)) = c_2$.
\end{enumerate}
The empty stem configurations are displayed in figure \ref{fig:empty_stem_alternating_renewal}.

\begin{figure}[H]
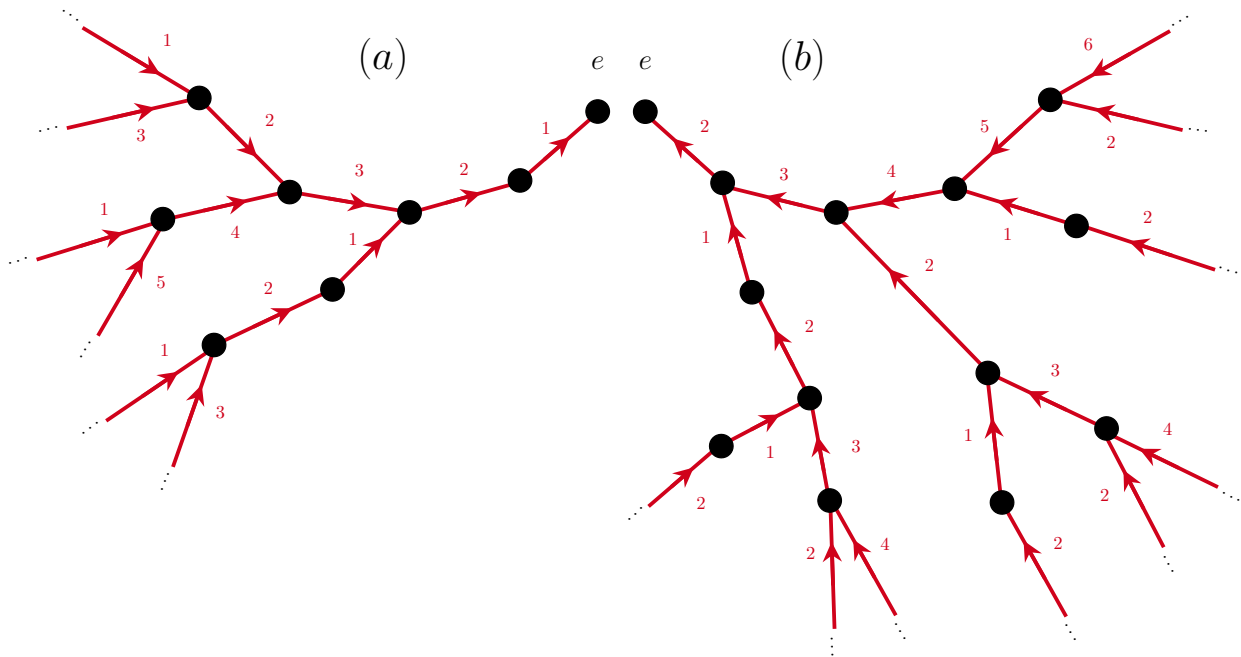

    \centering
\scalebox{0.63}{

\tikzset{every picture/.style={line width=0.75pt}} 



}
\caption{The empty stem configurations of the alternating renewal and only the filled vertices are shown. The configuration $(a)$ is the one such that the root of $e$ relative to this configuration is $c_1$, and the remaining one $(b)$ it satisfies that the root of $e$ relative to this configuration is $c_2$. \label{fig:empty_stem_alternating_renewal}}
\end{figure}

\section{Topological Aspects of Generalized Markov shifts}

Now we study the generalized Markov shift space $X_A$ more deeply under the topological point of view. We recall that the topology on $X_A$ is the subspace topology of the product topology on $\{0,1\}^\mathbb{F}$. In other words: the topology of $X_A$ is generated by the subbasis $\{\pi_\omega^{-1}(B)\cap X_A: \omega \in \mathbb{F}, B \subseteq \{0,1\}\}$. It is straightforward to notice that the topology of $X_A$ is second countable. Consider the \textit{generalized cylinder sets} defined by $C_\omega = \{\xi \in X_A: \xi_\omega =1\}$, where $\omega \in \mathbb{F}$. We have that
\begin{equation*}
    \pi_\omega^{-1}(B)\cap X_A = \begin{cases}
                                    \emptyset, \quad \text{if }B = \emptyset, \\
                                    C_\omega, \quad \text{if }B = \{1\}, \\
                                    C_\omega^c, \quad \text{if }B = \{0\}, \\
                                    X_A, \quad \text{if }B = \{0,1\}.
                                 \end{cases}
\end{equation*}
If we consider the topological basis $\mathcal{B}$ of the finite intersections of the elements of the family $$\{C_g:g=\alpha \gamma^{-1}\}\cup\{C_g^c:g=\alpha \gamma^{-1}\},$$then we already have a topological basis under the hypotheses of Theorem 8.2.17 of \cite{Bogachev2007}. However, many of the elements of $\mathcal{B}$ are redundant and we do not need to test the convergence for every $C_g$, $g \in \mathbb{F}$, as well as the are not needed to be part of the basis. Indeed, since for every $g \in \mathbb{F}$ which is not in the reduced form $\alpha \gamma^{-1}$ for all $\alpha, \gamma$ finite positive admissible words, we have that $\xi_g = 0$ for all $\xi \in X_A$, then $\pi_g^{-1}(B) \in \{\emptyset,X_A\}$ and such terms may be excluded of the basis without changes in the topology. In addition, we can simplify the basis even further by the following results. The main idea is to simplify the topological basis and prove the convergence of the measures on it. In the classical case, the convergence of measures on cylinders is sufficient to grant the weak$^*$ convergence and the same idea will be used here.

\begin{theorem}\label{thm:cylinders_are_compact} Let $\alpha$ be a positive admissible word and suppose that $\mathcal{O}_A$ is non-unital. Then the set $C_\alpha$ is compact on $X_A$.
\end{theorem}

\begin{proof} It is equivalent to prove that $C_\alpha \subseteq X_A$ is closed on $\widetilde{X}_A$. Since the topology on $\widetilde{X}_A$ is the product topology on $\{0,1\}^\mathbb{F}$ we have that any sequence $(\xi^n)_\mathbb{N}$ in $\widetilde{X}_A$ converges to an element $\xi \in \widetilde{X}_A$ if and only if it converges coordinate-wise. Let $(\xi^n)_\mathbb{N}$ be a sequence in $C_\alpha$ converging to $\xi \in \widetilde{X}_A$. Then $\xi_\alpha = \lim_{n \to \infty} (\xi^n)_\alpha = \lim_{n \to \infty} 1 = 1$. Therefore $\xi \in C_\alpha$, and then $C_\alpha$ is closed in $\widetilde{X}_A$, that is, $C_\alpha$ is compact in $X_A$.
\end{proof}

\begin{proposition}\label{prop:cylinder_reducing_inverse} For any $\alpha, \gamma$ positive admissible words with $\gamma= \gamma_0 \cdots \gamma_{|\gamma|-1}$, $|\gamma|>1$, we have that $C_{\alpha \gamma^{-1}}=C_{\alpha \gamma_{|\gamma|-1}^{-1}}$.
\end{proposition}

\begin{proof} The inclusion $C_{\alpha \gamma^{-1}}\subseteq C_{\alpha \gamma_{|\gamma|-1}^{-1}}$ is obvious from the convexity property and the fact that $\xi_e = 1$ for all configurations in $X_A$ give us the equivalence
\begin{equation} \label{eq:equiv_convex}
    \xi_\omega = 1 \iff \xi_\nu = 1, \forall \nu \in \llbracket \omega \rrbracket,
\end{equation}
which is valid for the particular case when $\xi \in C_{\alpha\gamma^{-1}}$ $\omega = \alpha \gamma^{-1}$ with $\nu = \alpha \gamma_{|\gamma|-1}^{-1}$. For the opposite inclusion, we recall that for any configuration $\xi \in X_A$ the following holds: if $\omega \in \mathbb{F}$ and $y \in \mathbb{N}$ satisfy $\xi_{\omega} = \xi_{\omega y}=1$, then for any $x \in \mathbb{N}$
\begin{equation}\label{eq:equiv_past_filled}
    \xi_{\omega x^{-1}} = 1 \iff A(x,y) = 1.
\end{equation}
In particular, we may apply the equivalence above for $\xi \in C_{\alpha \gamma_{|\gamma|-1}^{-1}}$, $\omega = \alpha \gamma_{|\gamma|-1}^{-1}$, $y = \gamma_{|\gamma|-1}$ and $x = \gamma_{|\gamma|-2}$. The rest of the proof follows by repeating the process for $\omega = \alpha (\gamma_{p} \cdots \gamma_{|\gamma|-1})^{-1}$, $y = \gamma_{p}$ and $x = \gamma_{p-1}$ for $p \in \{1,..., \gamma_{|\gamma|-2}\}$.
\end{proof}

In order to proceed with more simplifications of the basis, we define the following. For given positive admissible words $\alpha$ and $\gamma$ we define the finite sets $F_\alpha := \{\xi \in Y_A: \kappa(\xi) \in  \llbracket \alpha \rrbracket\setminus\{\alpha\}\}$, $F_\alpha^* := \{\xi \in Y_A: \kappa(\xi) \in  \llbracket \alpha \rrbracket\}$ and $F_\gamma^\alpha := \{\xi \in F_\gamma: \alpha \in \llbracket \kappa(\xi) \rrbracket\}$. Note that $F_\gamma^\gamma = \emptyset$ and that $F_\gamma^\alpha = \emptyset$ if $\alpha \notin \llbracket\gamma\rrbracket$. Also, define for any $H,I \subseteq \mathbb{N}$ the sets $G(\alpha,H):=\{\xi \in X_A: \kappa(\xi)=\alpha, \xi_{\alpha j^{-1}}=1, \forall j \in H\}$, with $G(\alpha,j):=G(\alpha,\{j\})$, for $j \in \mathbb{N}$; $K(\alpha,H):=\{\xi \in X_A: \kappa(\xi)=\alpha, \xi_{\alpha j^{-1}}=0, \forall j \in H\}$, with $K(\alpha,j):= K(\alpha,\{j\})$, for $j \in \mathbb{N}$; and $GK(\alpha,H,I):= G(\alpha,H) \cap H(\alpha,I)$, with $GK(\alpha,j,k) := GK(\alpha,\{j\},\{k\})$, for $j,k \in \mathbb{N}$. Finally, for $m \in \{0,1,..., |\alpha|-1\}$, define the words $\delta^0(\alpha) = e$ and $\delta^m(\alpha) = \alpha_0 \cdots \alpha_{m-1}$, for $m \neq 0$. 

\begin{proposition}\label{prop:general_C_alpha_inverse_j} For every $\alpha$ positive admissible word and for every $j \in \mathbb{N}$ s.t. $j \neq \alpha_{|\alpha|-1}$ we have that
\begin{equation*}
    C_{\alpha j^{-1}} = G(\alpha,j)\sqcup \bigsqcup_{k:A(j,k)=1}C_{\alpha k}.
\end{equation*}
\end{proposition}

\begin{proof} The inclusion
\begin{equation*}
    C_{\alpha j^{-1}} \supseteq G(\alpha,j)\sqcup \bigsqcup_{k:A(j,k)=1}C_{\alpha k}
\end{equation*}
is straightforward. For the opposite inclusion, let $\xi \in C_{\alpha j^{-1}}$. If $\kappa(\xi)= \alpha$, then $\xi \in G(\alpha,j)$. Suppose then $\kappa(\xi)\neq \alpha$. By convexity and $\xi_e=1$ we get $\alpha \in \llbracket\kappa(\xi)\rrbracket\setminus\{\kappa(\xi)\}$, and hence $\alpha k \in \llbracket\kappa(\xi)\rrbracket$ for some $k \in \mathbb{N}$. Since $\xi_{\alpha j^{-1}}=1$, we have necessarily that $A(j,k)=1$. 
\end{proof}
\begin{remark} On the notations of the proposition above, it is straightforward to notice that
\begin{equation*}
    G(\alpha,j) = \overline{\bigsqcup_{k:A(j,k)=1}C_{\alpha k}}\mathbin{\Big\backslash}\left\{\bigsqcup_{k:A(j,k)=1}C_{\alpha k}\right\} = \partial\left(\bigsqcup_{k:A(j,k)=1}C_{\alpha k}\right),
\end{equation*}
in other words,
\begin{equation*}
    C_{\alpha j^{-1}} = \overline{\bigsqcup_{k:A(j,k)=1}C_{\alpha k}}. 
\end{equation*}
\end{remark}

\begin{proposition}\label{prop:general_C_alpha_complement} Let $\alpha$ be a non-empty positive admissible word. Then,
\begin{equation*}
    C_\alpha^c = F_\alpha\sqcup\bigsqcup_{m=0}^{|\alpha|-1}\bigsqcup_{k\neq \alpha_m} C_{\delta^m(\alpha)k}.
\end{equation*}
\end{proposition}
\begin{proof} The inclusion 
\begin{equation*}
    C_\alpha^c \supseteq F_\alpha\sqcup\bigsqcup_{m=0}^{|\alpha|-1}\bigsqcup_{k\neq \alpha_m} C_{\delta^m(\alpha)k}
\end{equation*}
is straightforward. Let $\xi \in C_\alpha^c$. If $\kappa(\xi) \in \llbracket \alpha \rrbracket$ then we necessarily have that $\xi \in F_\alpha$. Suppose that $\kappa(\xi) \notin \llbracket \alpha \rrbracket$. Then there exists $m \in \{0,...,|\alpha|-1\}$ and $k \in \mathbb{N}\setminus\{\alpha_m\}$ such that $\delta^m(\alpha)k \in \llbracket\kappa(\xi)\rrbracket$ and therefore $\xi \in C_{\delta^m(\alpha)k}$. 
\end{proof}

\begin{proposition}\label{prop:general_C_alpha_j_inverse_complement} Let $\alpha$ be a non-empty positive admissible word and $j \neq \alpha_{|\alpha|-1}$. Then,
\begin{equation*}
    C_{\alpha j^{-1}}^c = K(\alpha,j) \sqcup F_\alpha \sqcup \left( \bigsqcup_{m=0}^{|\alpha|-1}\bigsqcup_{p\neq \alpha_m}  C_{\delta^m(\alpha )p}\right)\sqcup \bigsqcup_{\substack{p: A(j,p)=0}}  C_{\alpha p}.
\end{equation*}
\end{proposition}
\begin{proof} By Proposition \ref{prop:general_C_alpha_inverse_j} and Proposition \ref{prop:general_C_alpha_complement} we obtain
\begin{align*}
    C_{\alpha j^{-1}}^c =  \bigcap_{k:A(j,k)=1} \left(\left(G(\alpha,j)^c\cap F_{\alpha k}\right)\sqcup\bigsqcup_{m=0}^{|\alpha|}\bigsqcup_{p\neq (\alpha k)_m} G(\alpha,j)^c\cap C_{\delta^m(\alpha k)p}\right).
\end{align*}
In order to characterize the sets $G(\alpha,j)^c\cap F_{\alpha k}$ we use the following identity,
\begin{align*}
    G(\alpha,j)^c &= \Sigma_A \sqcup G_1\sqcup G_2 \sqcup K(\alpha,j),\text{ where }G_1 = C_\alpha^c \cap Y_A \quad \text{and}\\
    G_2 &= \left\{\xi \in Y_A:\alpha \in \llbracket\kappa(\xi)\rrbracket\setminus\{\kappa(\xi)\}\right\}.
\end{align*}
It is straightforward to verify that
\begin{equation*}
    F_{\alpha k} \cap \Sigma_A = F_{\alpha k} \cap G_2 = \emptyset, \quad  F_{\alpha k} \cap G_1 = F_\alpha, \quad F_{\alpha k} \cap K(\alpha,j) = K(\alpha,j).
\end{equation*}
Now we use the decomposition
\begin{equation*}
    G(\alpha,j)^c = G_3\sqcup K(\alpha,j),\text{ where  }  G_3 =\left\{\xi \in X_A:\kappa(\xi) \neq \alpha \right\},
\end{equation*}
in order to obtain for $m \in \{0,...,|\alpha|\}$ and $p \neq (\alpha k)_m$ the identities
\begin{equation*}
    C_{\delta^m(\alpha k)p}\cap G_3 = C_{\delta^m(\alpha k)p} \quad \text{and} \quad C_{\delta^m(\alpha k)p}\cap K(\alpha,j) = \emptyset.
\end{equation*}
Thus we have \vspace{-3mm}
\begin{align*}
    C_{\alpha j^{-1}}^c &= K(\alpha,j) \sqcup F_\alpha \sqcup \bigcap_{k:A(j,k)=1}\left( \bigsqcup_{m=0}^{|\alpha|-1}\bigsqcup_{p\neq \alpha_m}  C_{\delta^m(\alpha )p}\right)\sqcup \bigcap_{k:A(j,k)=1}\bigsqcup_{p\neq k}  C_{\alpha p}\\
    &= K(\alpha,j) \sqcup F_\alpha \sqcup \left( \bigsqcup_{m=0}^{|\alpha|-1}\bigsqcup_{p\neq \alpha_m}  C_{\delta^m(\alpha )p}\right)\sqcup \bigsqcup_{\substack{p: A(j,p)=0}}  C_{\alpha p}.
\end{align*} 
\end{proof}

\begin{remark} Proposition \ref{prop:general_C_alpha_inverse_j} and Proposition \ref{prop:general_C_alpha_j_inverse_complement} are also valid for $\alpha = e$, removing the hypothesis $j \neq \alpha_{|\alpha|-1}$. The proof is the same as presented above.
\end{remark}

We moved some auxiliary lemmas to the appendix to avoid technicalities. The lemmas are used to prove the next theorem and identities involving intersections of generalized cylinders, the most obvious is the following:
\begin{equation}\label{eq:C_cap_C}
    C_\alpha \cap C_\gamma = \begin{cases}
                                C_\alpha, \quad \text{if } \gamma \in \llbracket \alpha \rrbracket,\\
                                C_\gamma, \quad \text{if } \alpha \in \llbracket \gamma \rrbracket,\\
                                \emptyset, \quad \text{otherwise};
                            \end{cases} 
\end{equation}
which holds for every $\alpha$ and $\gamma$ positive admissible words.

Summarizing the results, Proposition \ref{prop:cylinder_reducing_inverse} grants that we only need to consider the basis of the finite intersections of generalized cylinders in the form $C_\alpha$ and $C_{\alpha j^{-1}}$, $\alpha$ positive admissible word and $j \in \mathbb{N}$, and their complements. However, Proposition \ref{prop:general_C_alpha_inverse_j}, \ref{prop:general_C_alpha_complement} and \ref{prop:general_C_alpha_j_inverse_complement} show that $C_{\alpha j^{-1}}$ as much as $C_\alpha$ and $C_{\alpha j^{-1}}$ can be written as a disjoint union of generalized cylinders disjointly united with a subset of $Y_A$. It is remarkable that for any finite intersection of these elements we obtain again a disjoint union between a subset of $Y_A$ and a disjoint union of cylinders. It is important to emphasize that all these open sets are actually clopen sets. 

The next result presents all the possible intersections between two elements of the subbasis.

\begin{theorem}\label{thm:huge_generators_intersections} Given $\alpha,\gamma$ positive admissible words and $j,k \in \mathbb{N}$ we have that
\begin{align}
    C_\alpha \cap C_\gamma^c &= 
    \begin{cases}
        F_{\gamma}^{\alpha}\sqcup \bigsqcup_{n=|\alpha|}^{|\gamma|-1} \bigsqcup_{j\neq \gamma_n} C_{\delta^{n}(\gamma)j},\quad \alpha \in \llbracket \gamma \rrbracket,\\
        C_\alpha,\quad \alpha \notin \llbracket \gamma \rrbracket \text{ and } \gamma \notin \llbracket \alpha \rrbracket, \\
        \emptyset, \quad \text{otherwise};
    \end{cases} \label{eq:C_cap_C_comp}\\
    C_\alpha^c \cap C_\gamma^c &= \begin{cases}
                                    C_\alpha^c, \quad \alpha \in \llbracket \gamma \rrbracket,\\
                                    C_\gamma^c, \quad \gamma \in \llbracket \alpha \rrbracket,\\
                                    F_{\alpha'}^* \sqcup F_\alpha^{\alpha'\alpha_{|\alpha'|}} \sqcup F_\gamma^{\alpha'\gamma_{|\alpha'|}} \sqcup  \left(\bigsqcup_{m = 0}^{|\alpha|-1} \bigsqcup_{k\neq \alpha_m} C_{\delta^m(\alpha)k}\right) \\
                                    \qquad \qquad \qquad \sqcup \left(\bigsqcup_{|\alpha'| < n < |\gamma|-1} \bigsqcup_{p\neq \gamma_n} C_{\delta^n(\gamma )p}\right), \quad \text{otherwise};
                                  \end{cases}\label{eq:C_comp_cap_C_comp}\\
    C_\alpha \cap C_{\gamma j^{-1}} &= \begin{cases}
                                    C_{\gamma j^{-1}},\quad \alpha \in \llbracket \gamma \rrbracket,\\
                                    C_\alpha,\quad \gamma \in \llbracket \alpha \rrbracket \setminus \{\alpha\} \text{ and } A(j,\alpha_{|\gamma|})=1, \\
                                    \emptyset, \quad \text{otherwise};
                                    \end{cases}\label{eq:C_cap_C_inverse}\\
    C_\alpha \cap C_{\gamma j^{-1}}^c &= \begin{cases}
                                    K(\gamma,j) \sqcup F_\gamma^\alpha \sqcup \left( \bigsqcup_{n=|\alpha|}^{|\gamma|-1}\bigsqcup_{p\neq \gamma_m} C_{\delta^n(\gamma)p}\right)\sqcup \bigsqcup_{\substack{p: A(j,p)=0}} C_{\gamma p},\quad \alpha \in \llbracket \gamma \rrbracket,\\
                                    C_\alpha,\quad \gamma \in \llbracket \alpha \rrbracket \setminus \{\alpha\} \text{ and } A(j,\alpha_{|\gamma|})=0, \text{ or } \alpha \notin \llbracket \gamma \rrbracket \text{ and }\alpha \notin \llbracket \gamma \rrbracket,\\
                                    \emptyset, \quad \text{otherwise};
                                    \end{cases}\label{eq:C_cap_C_inverse_comp}\\
    C_\alpha^c \cap C_{\gamma j^{-1}} &= \begin{cases}
                                        \emptyset,\quad \alpha \in \llbracket \gamma \rrbracket,\\
                                        G(\gamma,j) \sqcup F_\alpha^{\gamma \alpha_{|\gamma|}} \sqcup  \mathfrak{C}[\alpha,\gamma,j], \quad \gamma \in \llbracket \alpha \rrbracket \setminus \{\alpha\} \text{ and }A(j,\alpha_{|\gamma|}) = 1,\\
                                        G(\gamma,j) \sqcup  \mathfrak{C}[\alpha,\gamma,j], \quad \gamma \in \llbracket \alpha \rrbracket \setminus \{\alpha\} \text{ and }A(j,\alpha_{|\gamma|}) = 0,\\
                                        C_{\gamma j^{-1}}, \quad \text{otherwise};
                                        \end{cases}\label{eq:C_comp_cap_C_inverse}\\
    C_\alpha^c \cap C_{\gamma j^{-1}}^c &= \begin{cases}
                                C_\alpha^c,\quad \alpha \in \llbracket \gamma \rrbracket,\\
                                K(\gamma,j) \sqcup F_\gamma \sqcup F_\alpha^{\gamma \alpha_{|\gamma|}} \sqcup \left(\bigsqcup_{n=0}^{|\gamma|-1}\bigsqcup_{p\neq \gamma_n}  C_{\delta^n(\gamma )p}\right) \\
                                \qquad \qquad \qquad \sqcup \mathfrak{D}[\alpha,\gamma,j], \quad \gamma \in \llbracket \alpha \rrbracket \setminus \{\alpha\} \text{ and }A(j,\alpha_{|\gamma|}) = 0,\\
                                K(\gamma,j) \sqcup F_\gamma \sqcup \left(\bigsqcup_{n=0}^{|\gamma|-1}\bigsqcup_{p\neq \gamma_n}  C_{\delta^n(\gamma )p} \right)\\
                                \qquad \qquad \qquad \sqcup \mathfrak{D}[\alpha,\gamma,j], \quad \gamma \in \llbracket \alpha \rrbracket \setminus \{\alpha\} \text{ and }A(j,\alpha_{|\gamma|}) = 1,\\
                                K(\gamma,j) \sqcup F_\alpha^{\alpha' \alpha_{|\gamma|}} \sqcup F_\gamma^{\alpha'\gamma_{|\alpha'|}} \sqcup \left(\bigsqcup_{m = 0}^{|\alpha|-1} \bigsqcup_{k\neq \alpha_m} C_{\delta^m(\alpha)k}\right) \\\qquad \qquad \qquad \sqcup \left(\bigsqcup_{|\alpha'| < n < |\gamma|-1} \bigsqcup_{p\neq \gamma_n} C_{\delta^n(\gamma )p}\right)\\
                                \qquad \qquad \qquad \sqcup \bigsqcup_{\substack{p: A(j,p)=0}} C_{\gamma p} , \quad \text{otherwise}; \end{cases}\label{eq:C_comp_cap_C_inverse_comp}
\end{align}

\begin{align}
    C_{\alpha j^{-1}} \cap C_{\gamma l^{-1}} &= \begin{cases}
                                G(\alpha,\{j,l\}) \sqcup \bigsqcup_{k:A(j,k)A(l,k)=1}C_{\alpha k},\quad \alpha = \gamma,\\
                                G(\gamma,l) \sqcup  \bigsqcup_{m:A(l,m)=1}  C_{\gamma m}, \quad \alpha \in \llbracket \gamma \rrbracket \setminus \{\gamma\} \text{ and }A(j,\gamma_{|\alpha|}) = 1,\\
                                 G(\alpha,j) \sqcup  \bigsqcup_{k:A(j,k)=1}  C_{\alpha k}, \quad \gamma \in \llbracket \alpha \rrbracket \setminus \{\alpha\} \text{ and }A(l,\alpha_{|\gamma|}) = 1,\\
                            \emptyset, \quad \text{otherwise}; \end{cases}\label{eq:C_inverse_cap_C_inverse}\\
    C_{\alpha j^{-1}} \cap C_{\gamma l^{-1}}^c &= \begin{cases}
                                GK(\alpha,j,l) \sqcup \bigsqcup_{k:A(j,k) = 1,A(l,k)=0}C_{\alpha k},\quad \alpha = \gamma,\\
                                K(\gamma,l) \sqcup G(\alpha,j) \sqcup F_{\gamma}^{\alpha \gamma_{|\alpha|}} \sqcup \mathfrak{C}[\gamma,\alpha,j]\\
                                \qquad \qquad \qquad \sqcup  \bigsqcup_{m:A(l,m)=0}  C_{\gamma m}, \quad \alpha \in \llbracket \gamma \rrbracket \setminus \{\gamma\} \text{ and }A(j,\gamma_{|\alpha|}) = 1,\\
                                G(\alpha,j) \sqcup \mathfrak{C} (\gamma,\alpha,j), \quad \alpha \in \llbracket \gamma \rrbracket \setminus \{\gamma\} \text{ and }A(j,\gamma_{|\alpha|}) = 0,\\    
                                G(\alpha,j) \sqcup  \bigsqcup_{k:A(j,k)=1}  C_{\alpha k}, \quad \gamma \in \llbracket \alpha \rrbracket \setminus \{\alpha\} \text{ and }A(l,\alpha_{|\gamma|}) = 0, \\
                                \qquad \qquad \qquad \text{ or if } \alpha \notin \llbracket \gamma \rrbracket \text{ and } \gamma \notin \llbracket  \alpha \rrbracket\\
                                \emptyset, \quad \text{otherwise}; \end{cases}\label{eq:C_inverse_cap_C_inverse_comp}\\
    C_{\alpha j^{-1}}^c \cap C_{\gamma l^{-1}}^c &= \begin{cases}
                                K(\alpha,\{j,l\}) \sqcup F_\alpha \sqcup \left(\bigsqcup_{k: A(j,k) = A(l,k) = 0} C_{\alpha k}\right)\\
                                \qquad \qquad \qquad \sqcup \left(\bigsqcup_{m = 0}^{|\alpha| - 1} \bigsqcup_{k\neq \alpha_m} C_{\delta^m(\alpha)k}\right),\quad \alpha = \gamma,\\
                                K(\alpha,j) \sqcup K(\gamma,l) \sqcup F_\alpha \sqcup F_\gamma^{\alpha \gamma_{|\alpha|}} \sqcup \mathfrak{D}[\gamma,\alpha,l] \sqcup \left( \bigsqcup_{m = 0}^{|\alpha| - 1} \bigsqcup_{k\neq \alpha_m} C_{\delta^m(\alpha)k} \right) \\
                                \qquad \qquad \qquad \sqcup  \left(\bigsqcup_{p:A(A(l,p)=0)}  C_{\gamma p}\right), \quad \alpha \in \llbracket \gamma \rrbracket \setminus \{\gamma\} \text{ and }A(j,\gamma_{|\alpha|}) = 0,\\
                                K(\alpha,j) \sqcup F_\alpha \sqcup \mathfrak{D}[\gamma,\alpha,l]\\
                                \qquad \qquad \qquad \sqcup \left( \bigsqcup_{m = 0}^{|\alpha| - 1} \bigsqcup_{k\neq \alpha_m} C_{\delta^m(\alpha)k} \right), \quad \alpha \in \llbracket \gamma \rrbracket \setminus \{\gamma\} \text{ and }A(j,\gamma_{|\alpha|}) = 1,\\
                                K(\alpha,j) \sqcup K(\gamma,l) \sqcup F_\gamma \sqcup F_\alpha^{\gamma \alpha_{|\gamma|}} \sqcup \mathfrak{D}[\alpha,\gamma,j] \sqcup \left( \bigsqcup_{n = 0}^{|\gamma| - 1} \bigsqcup_{p\neq \gamma_n} C_{\delta^n(\gamma)p} \right)\\
                                \qquad \qquad \qquad \sqcup  \left(\bigsqcup_{k:A(A(j,k)=0)}  C_{\alpha k} \right), \quad \gamma \in \llbracket \alpha \rrbracket \setminus \{\alpha\} \text{ and }A(l,\alpha_{|\gamma|}) = 0,\\
                                K(\gamma,l) \sqcup F_\gamma \sqcup \mathfrak{D}[\alpha,\gamma,j]\\
                                \qquad \qquad \qquad \sqcup \left( \bigsqcup_{n = 0}^{|\gamma| - 1} \bigsqcup_{p\neq \gamma_n} C_{\delta^n(\gamma)p} \right), \quad \gamma \in \llbracket \alpha \rrbracket \setminus \{\alpha\} \text{ and }A(l,\alpha_{|\gamma|}) = 1,\\
                                K(\alpha,j) \sqcup K(\gamma,l) \sqcup F_{\alpha'} \sqcup F_\alpha^{\alpha'\alpha_{|\alpha'|}} \sqcup F_\gamma^{\alpha'\gamma_{|\alpha'|}} \sqcup  \left(\bigsqcup_{k:A(A(j,k)=0)}  C_{\alpha k} \right) \\\qquad \qquad \qquad \sqcup \left( \bigsqcup_{p:A(A(l,p)=0)}  C_{\gamma p}\right)
                                \sqcup \left(\bigsqcup_{m = 0}^{|\alpha|-1} \bigsqcup_{k\neq \alpha_m} C_{\delta^m(\alpha)k}\right)\\
                                \qquad \qquad \qquad \sqcup \left(\bigsqcup_{|\alpha'| < n < |\gamma|-1} \bigsqcup_{p\neq \gamma_n} C_{\delta^n(\gamma )p}\right), \quad \text{otherwise}; \end{cases}\label{eq:C_inverse_comp_cap_C_inverse_comp}
\end{align}
where
\begin{equation*}
    \mathfrak{C}[\alpha,\gamma,j] := \begin{cases}
    \bigsqcup_{\substack{p:A(j,p)=1,\\p \neq \alpha_{|\gamma|}}} C_{\gamma p}, \quad |\alpha| = |\gamma| + 1,\\
    \left(\bigsqcup_{m=|\gamma|+1}^{|\alpha|-1} \bigsqcup_{k \neq \alpha_m} C_{\delta^m(\alpha)k}\right) \sqcup \bigsqcup_{\substack{p:A(j,p)=1,\\p \neq \alpha_{|\gamma|}}} C_{\gamma p}, \quad |\alpha| > |\gamma| + 1,
    \end{cases}
\end{equation*}
and
\begin{equation*}
    \mathfrak{D}[\alpha,\gamma,j] := \begin{cases}
    \bigsqcup_{\substack{p:A(j,p)=0,\\p \neq \alpha_{|\gamma|}}} C_{\gamma p}, \quad |\alpha| = |\gamma| + 1,\\
    \left(\bigsqcup_{m=|\gamma|+1}^{|\alpha|-1} \bigsqcup_{k \neq \alpha_m} C_{\delta^m(\alpha)k}\right) \sqcup \bigsqcup_{\substack{p:A(j,p)=0,\\p \neq \alpha_{|\gamma|}}} C_{\gamma p}, \quad |\alpha| > |\gamma| + 1.
    \end{cases}
\end{equation*}
Also, $\alpha'$ is the longest word in $\llbracket \alpha \rrbracket \cap \llbracket \gamma \rrbracket$.
\end{theorem}

\begin{proof} See appendix \ref{ape:proof_giant_theorem_generalized_cylinders}. 
\end{proof}

\begin{proposition}\label{prop:X_A_decomposition} Consider the set $X_A$. Then,
\begin{equation*}
    X_A = \left(\bigsqcup_{j \in \mathbb{N}}G(e,j)\right) \sqcup \left(\bigsqcup_{j \in \mathbb{N}}C_j\right).
\end{equation*}
\end{proposition}

\begin{proof} The inclusion
\begin{equation*}
    X_A \supseteq \left(\bigsqcup_{j \in \mathbb{N}}G(e,j)\right) \sqcup \left(\bigsqcup_{j \in \mathbb{N}}C_j\right)
\end{equation*}
is straightforward. For the opposite inclusion, let $\xi \in X_A$. We have two possibilities, namely $\kappa(\xi) \neq e$ or $\kappa(\xi) = e$. In the first case, we have necessarily that $\xi \in C_j$ for some $j \in \mathbb{N}$. Now if $\kappa(\xi) = e$, we have necessarily that $\xi \in G(e,j)$ for some $j \in \mathbb{N}$, otherwise $\xi = \phi$, the configuration filled only in $e$. This is not possible because $\phi \notin X_A$ if $X_A$ $\mathcal{O}_A$ is unital and $\phi \in \widetilde{X}_A \setminus X_A$ otherwise. 
\end{proof}

\begin{corollary} It is true that
\begin{equation*}
    X_A = \bigcup_{j \in \mathbb{N}}\sigma(C_j).
\end{equation*}
\end{corollary}
\begin{proof} Let $\xi \in X_A$. By Proposition \ref{prop:X_A_decomposition} we have two possibilities: $\xi \in C_j$, for some $j \in \mathbb{N}$, or $\xi \in G(e,j)$ for some $j \in \mathbb{N}$. In the first case, we have that $\xi_j = 1$ for some $j \in \mathbb{N}$ and hence $R_\xi(j) = \{i \in \mathbb{N}: \xi_{ji^{-1}}\} = \{i \in \mathbb{N}: A(i,j) = 1\}$, which is not empty because $A$ is transitive. Then we have at least one\footnote{Actually $\xi \in \sigma(C_{ij})$ for every $i$ such that $A(i,j) = 1$.} $i \in \mathbb{N}$ such that
\begin{equation*}
    \xi \in \sigma(C_{ij}) \subseteq \sigma(C_i),
\end{equation*}
where $A(i,j) = 1$. Now, if $\xi \in G(e,j)$ for some $j \in \mathbb{N}$, then $\xi_{j^{-1}} = 1$. It is straightforward to notice that $j\xi \in C_j$ and $\sigma(j\xi) = \xi$. Therefore, $\xi \in \sigma (C_j)$. 
\end{proof}

\subsection{Cylinder topology: Renewal shift}
\label{subsec:cylinder_topology_Renewal}

In the particular case of the renewal shift, we have a simple characterization and we present it now. Observe that, for every generalized cylinder $C_\alpha$ on a positive word $\alpha$, if it does not end with $1$, we have
\begin{equation*}
    C_{\alpha} = C_{\gamma},
\end{equation*}
where $\gamma$ is the unique shortest word that it starts with $\alpha$ and ends in $1$.

\begin{lemma}\label{lemma:G_K_renewal}
For the renewal shift we have
\begin{equation*}
    G(\alpha,F) = \begin{cases}
        \{\alpha \xi^0\}, \quad \alpha_{|\alpha| - 1} = 1 \text{ and } F = \{1\},\\
        \emptyset, \quad \text{otherwise};
    \end{cases}
\end{equation*}
and
\begin{equation*}
    K(\alpha,F) = \begin{cases}
        \{\alpha \xi^0\}, \quad \alpha_{|\alpha| - 1} = 1 \text{ and } F \neq \{1\},\\
        \emptyset, \quad \text{otherwise};
    \end{cases}
\end{equation*}
for $\alpha$ non-empty finite admissible word.
\end{lemma}

\begin{proof} From the renewal shift we have that there exists $\xi \in Y_A$ such that $\kappa(\xi) = \alpha$, that is, $\xi = \alpha \xi^0$, if and only if $\alpha_{|\alpha|-1} = 1$. We complete the proof by recalling that $R_{\alpha\xi^0}(\alpha) = \{1\}$ for $\alpha$ ending with $1$.
\end{proof}

\begin{remark} The lemma above is also valid for $\alpha = e$, just by removing the condition $\alpha_{|\alpha| - 1} = 1$.
\end{remark}

\begin{corollary}\label{cor:C_alpha_j_inv_renewal} Consider the set $X_A$ of the renewal Markov shift. For any $\alpha$ positive admissible word, $|\alpha|\neq 0$, and any $j \in \mathbb{N}$, we have that $C_{\alpha j^{-1}}$ is the empty set or $C_{\alpha j^{-1}} = C_\gamma$, where $\gamma$ is a positive admissible word.
\end{corollary}

\begin{proof} If $j = \alpha_{|\alpha|-1}$ then the statement is obvious. So consider $j \neq \alpha_{|\alpha|-1}$. By Lemma \ref{lemma:G_K_renewal} we have $G(\alpha,F) = \emptyset$. Also,
\begin{equation*}
    \bigsqcup_{k:A(j,k)=1} C_{\alpha k} = \begin{cases}
        C_{\alpha (\alpha_{|\alpha|-1} -1)}, \quad \alpha_{|\alpha|-1} > 1 \text{ and } j = 1,\\
        C_{\alpha (j-1)}, \quad \alpha_{|\alpha|-1} = 1,\\
        \emptyset, \quad \text{otherwise};
    \end{cases}
\end{equation*}
and therefore, by Proposition \ref{prop:general_C_alpha_inverse_j} we conclude that
\begin{equation*}
    C_{\alpha j^{-1}} = \begin{cases}
        C_{\alpha (\alpha_{|\alpha|-1} -1)}, \quad \alpha_{|\alpha|-1} > 1 \text{ and } j = 1,\\
        C_{\alpha (j-1)}, \quad \alpha_{|\alpha|-1} = 1,\\
        \emptyset, \quad \text{otherwise}.
    \end{cases}
\end{equation*}
\end{proof}
\begin{remark} If $\alpha = e$, then
\begin{equation*}
    C_{j^{-1}} = \begin{cases}
        X_A, \quad j = 1,\\
        \emptyset, \quad \text{otherwise};
    \end{cases}
\end{equation*}
which is straightforward from the renewal shift.
\end{remark}

\begin{remark} In the renewal shift, for given positive admissible words $\alpha$ and $\gamma$, we have that
\begin{align*}
    F_\alpha &= \{\eta \xi^0 \in Y_A: \eta_{|\eta|-1} = 1, \eta \in \llbracket \alpha \rrbracket \setminus \{\alpha\}\}, \quad F_\alpha^* = \{\eta \xi^0 \in Y_A: \eta_{|\eta|-1} = 1, \eta \in \llbracket \alpha \rrbracket\},\\
    F_\gamma^\alpha &= \{\eta \xi^0 \in F_\gamma: \alpha \in \llbracket \eta \rrbracket\},
\end{align*}
all finite sets.
\end{remark}

\begin{corollary}\label{cor:C_alpha_j_inv_comp_renewal} Let $\alpha$ be a non-empty positive admissible word and $j \neq \alpha_{|\alpha|-1}$. Then,
\begin{equation*}
    C_{\alpha j^{-1}}^c = \begin{cases}
        C_{\alpha (\alpha_{|\alpha|-1} -1)}^c, \quad \alpha_{|\alpha|-1} > 1 \text{ and } j = 1,\\
        C_{\alpha (j-1)}^c, \quad \alpha_{|\alpha|-1} = 1,\\
        X_A, \quad \text{otherwise}
    \end{cases}
\end{equation*}
\end{corollary}

\begin{proof} It is straightforward from the proof of Corollary \ref{cor:C_alpha_j_inv_renewal}. 
\end{proof}

\begin{remark} Equivalently, by Proposition \ref{prop:general_C_alpha_complement}, we have
\begin{equation}\label{eq:c_alpha_inverse_comp_renewal}
    C_{\alpha j^{-1}}^c = \begin{cases}
        F_{\alpha (\alpha_{|\alpha|-1} -1)}\sqcup\bigsqcup_{m=0}^{|\alpha|}\bigsqcup_{k\neq ({\alpha (\alpha_{|\alpha|-1} -1)})_m} C_{\delta^m({\alpha (\alpha_{|\alpha|-1} -1)})k}, \quad \alpha_{|\alpha|-1} > 1 \text{ and } j = 1,\\
        F_{\alpha(j-1)}\sqcup\bigsqcup_{m=0}^{|\alpha|} \bigsqcup_{k\neq (\alpha (j-1))_m} C_{\delta^m(\alpha (j-1))k}, \quad \alpha_{|\alpha|-1} = 1,\\
        X_A, \quad \text{otherwise}.
    \end{cases}
\end{equation}
\end{remark}

The previous results in this subsection allow us to reduce the subbasis of the renewal shift to generalized cylinders on positive words and their complements. Moreover, for these cylinders, we only need to consider the words ending in `$1$'. We present the pairwise intersections of the elements of the subbasis for the generalized renewal shift in the next corollary.

\begin{corollary} For any positive admissible words $\alpha$ and $\gamma$ we have that
\begin{align*}
    C_\alpha \cap C_\gamma &= \begin{cases}
                                C_\alpha, \quad \text{if } \gamma \in \llbracket \alpha \rrbracket,\\
                                C_\gamma, \quad \text{if } \alpha \in \llbracket \gamma \rrbracket,\\
                                \emptyset, \quad \text{otherwise};
                            \end{cases} \\
    C_\alpha \cap C_\gamma^c &= \begin{cases}
                                F_{\gamma}^{\alpha}\sqcup \bigsqcup_{n=|\alpha|}^{|\gamma|-1} \bigsqcup_{j\neq \gamma_n} C_{\delta^{n}(\gamma)j},\quad \alpha \in \llbracket \gamma \rrbracket,\\
                                C_\alpha,\quad \alpha \notin \llbracket \gamma \rrbracket \text{ and } \gamma \notin \llbracket \alpha \rrbracket, \\
                                \emptyset, \quad \text{otherwise};
                               \end{cases}\\
    C_\alpha^c \cap C_\gamma^c &= \begin{cases}
                                    C_\alpha^c, \quad \alpha \in \llbracket \gamma \rrbracket,\\
                                    C_\gamma^c, \quad \gamma \in \llbracket \alpha \rrbracket,\\
                                    F_{\alpha'}^* \sqcup F_\alpha^{\alpha'\alpha_{|\alpha'|}} \sqcup F_\gamma^{\alpha'\gamma_{|\alpha'|}} \sqcup  \left(\bigsqcup_{m = 0}^{|\alpha|-1} \bigsqcup_{k\neq \alpha_m} C_{\delta^m(\alpha)k}\right)\\
                                    \qquad \qquad \qquad \sqcup \left(\bigsqcup_{|\alpha'| < n < |\gamma|-1} \bigsqcup_{p\neq \gamma_n} C_{\delta^n(\gamma )p}\right), \quad \text{otherwise}. \end{cases}
\end{align*}
\end{corollary}

\begin{proof} The first identity is straightforward, while the rest of them are precisely the equations \eqref{eq:C_cap_C_comp} and \eqref{eq:C_comp_cap_C_comp} from Theorem \ref{thm:huge_generators_intersections}.
\end{proof}

From the corollary above, it is straightforward that every finite intersection among the elements of the subbasis is in the form
\begin{equation}\label{eq:basis_renewal_is_finite_set_union_countable_cylinders}
    F \sqcup \bigsqcup_{n \in \mathbb{N}} C_{w(n)},
\end{equation}
where $F \subseteq Y_A$ is finite and $w(n)$ is a positive admissible word for each every $n \in \mathbb{N}$.

\subsection{Cylinder topology: Pair Renwewal shift}
\label{subsec:cylinder_topology_Pair}

As well as it was done in subsection \ref{subsec:cylinder_topology_Renewal}, we also characterize the cylinders and intersections for the pair renewal shift. Particularly for this section, we use the following notation: for $\alpha$ positive admissible word, we denote by $\alpha \xi^{0,k}$ the configuration of stem $\alpha$ that belongs to $Y_A(\xi^{0,k})$,$k \in \{1,2\}$.

First, we notice that, for every generalized cylinder $C_\alpha$ on a positive word $\alpha$, if it does not end with $1$ or $2$, then
\begin{equation*}
    C_{\alpha} = C_{\gamma},
\end{equation*}
where $\gamma$ is the unique shortest word that it starts with $\alpha$ and ends in $2$.

\begin{lemma}\label{lemma:G_K_Pair_renewal}
For the pair renewal shift we have
\begin{equation*}
    G(\alpha,F) = \begin{cases}
        \{\alpha \xi^{0,1},\alpha \xi^{0,2}\}, \quad \alpha_{|\alpha| - 1} = 1 \text{ and } F = \{1\};\\
        \{\alpha \xi^{0,1}\}, \quad \alpha_{|\alpha| - 1} = 1 \text{ and } F \in \{\{1,2\}, \{2\}\};\\
        \{\alpha \xi^{0,1}\}, \quad \alpha_{|\alpha| - 1} = 2 \text{ and } F \in \{\{1\},\{1,2\},\{2\}\} ;\\
        \emptyset, \quad \text{otherwise},
    \end{cases}
\end{equation*}
and
\begin{equation*}
    K(\alpha,F) = \begin{cases}
        \{\alpha \xi^{0,1},\alpha \xi^{0,2}\}, \quad \alpha_{|\alpha| - 1} = 1 \text{ and } \{1,2\} \not\subseteq F;\\
        \{\alpha \xi^{0,2}\}, \quad \alpha_{|\alpha| - 1} = 1,\text{ }  1 \notin F \text{ and }  2 \in F;\\
        \emptyset, \quad \text{otherwise},
    \end{cases}
\end{equation*}
for $\alpha$ non-empty finite admissible word.
\end{lemma}

\begin{proof} From the pair renewal shift we have that there exists $\xi \in Y_A$ such that $\kappa(\xi) = \alpha$, if and only if $\xi = \alpha \xi^{0,1}$ or $\alpha \xi^{0,2}$. We complete the proof by using that
\begin{itemize}
    \item $\xi = \alpha \xi^{0,1}$ if and only if $\alpha_{|\alpha|-1} \in \{1,2\}$ and $R_\xi(\alpha) = \{1,2\}$;
    \item $\xi = \alpha \xi^{0,2}$ if and only if $\alpha_{|\alpha|-1} = 1$ and $R_\xi(\alpha) = \{1\}$.
\end{itemize}
\end{proof}

\begin{remark} The lemma above is also valid for $\alpha = e$, just by removing the conditions for $\alpha_{|\alpha| - 1}$.
\end{remark}

\begin{corollary}\label{cor:C_alpha_j_inv_Pair_renewal} Consider the set $X_A$ of the pair renewal Markov shift. For any $\alpha$ positive admissible word, $|\alpha|\neq 0$, and any $j \in \mathbb{N}$, we have for $j \neq \alpha_{|\alpha|-1}$ that
\begin{equation*}
    C_{\alpha j^{-1}} = \begin{cases}
                            C_{\alpha (\alpha_{|\alpha|-1} - 1)}, \quad \text{if } \alpha_{|\alpha|-1}>2 \text{ and } j=1;\\
                            C_{\alpha (\alpha_{|\alpha|-1} - 1)}, \quad \text{if } \alpha_{|\alpha|-1}>2 \text{ and }, \text{ } j=2 \text{ and } \alpha_{|\alpha|-1} \text{ is even};\\
                            \{\alpha \xi^{0,1}\}\sqcup C_{\alpha 1} \sqcup \bigsqcup_{k \in \mathbb{N}} C_{\alpha (2k)}, \quad  \text{if } \alpha_{|\alpha|-1}=2 \text{ and } j=1;\\
                            \{\alpha \xi^{0,1}\}\sqcup C_{\alpha 1} \sqcup \bigsqcup_{k \in \mathbb{N}} C_{\alpha (2k)}, \quad  \text{if } \alpha_{|\alpha|-1}=1 \text{ and } j=2;\\
                            \emptyset, \quad \text{otherwise};
                        \end{cases}
\end{equation*}
\end{corollary}

\begin{proof} By Lemma \ref{lemma:G_K_Pair_renewal} we have $G(\alpha,j) = \emptyset$ if $\alpha_{|\alpha|-1} > 2$ and $\{\alpha \xi^{0,1}\}$ if when $\alpha_{|\alpha|-1} = 2$ and $j=1$, and when $\alpha_{|\alpha|-1} = 1$ and $j=2$. The remaining cases are empty. By Proposition \ref{prop:general_C_alpha_inverse_j}, it remains to analyze the unions in the form
\begin{equation*}
    \bigsqcup_{k:A(j,k)=1} C_{\alpha k}.
\end{equation*}
It is straightforward from the pair renewal transition matrix that
\begin{equation*}
    \bigsqcup_{k:A(j,k)=1} C_{\alpha k} = \begin{cases}
                                            C_{\alpha (\alpha_{|\alpha|-1} - 1)}, \quad \text{if } \alpha_{|\alpha|-1}>2 \text{ and } j=1;\\
                                            C_{\alpha (\alpha_{|\alpha|-1} - 1)}, \quad \text{if } \alpha_{|\alpha|-1}>2 \text{ and }, \text{ } j=2 \text{ and } \alpha_{|\alpha|-1} \text{ is even};\\
                                            C_{\alpha 1} \sqcup \bigsqcup_{k \in \mathbb{N}} C_{\alpha (2k)}, \quad  \text{if } \alpha_{|\alpha|-1}=2 \text{ and } j=1;\\
                                            C_{\alpha 1} \sqcup \bigsqcup_{k \in \mathbb{N}} C_{\alpha (2k)}, \quad  \text{if } \alpha_{|\alpha|-1}=1 \text{ and } j=2;\\
                                            \emptyset, \quad \text{otherwise};
                                          \end{cases}.  \tag*{\qedhere}
\end{equation*}
\end{proof}

\begin{remark} When $j = \alpha_{|\alpha|-1}$, then $C_{\alpha j^{-1}}$ is a cylinder or $X_A$. If $\alpha = e$, then 
\begin{equation*}
    C_{j^{-1}} = \begin{cases}
                    X_A, \quad j = 1,\\
                    \{\xi^{0,1}\}\sqcup C_1 \sqcup \bigsqcup_{k \in \mathbb{N}} C_{(2k)},\quad j = 2,\\        
                    \emptyset, \quad \text{otherwise};
                 \end{cases}
\end{equation*}
which is straightforward from the pair renewal shift.
\end{remark}

\begin{remark} In the pair renewal shift, for given positive admissible words $\alpha$ and $\gamma$, we have that
\begin{align*}
    F_\alpha &= \{\eta \xi^0 \in Y_A: \eta_{|\eta|-1} \in \{1,2\}, \eta \in \llbracket \alpha \rrbracket \setminus \{\alpha\}\}, \quad F_\alpha^* = \{\eta \xi^0 \in Y_A: \eta_{|\eta|-1} = 1, \eta \in \llbracket \alpha \rrbracket\},\\
    F_\gamma^\alpha &= \{\eta \xi^0 \in F_\gamma: \alpha \in \llbracket \eta \rrbracket\},
\end{align*}
all finite sets.
\end{remark}

\begin{corollary}\label{cor:C_alpha_j_inv_comp_Pair_renewal} Let $\alpha$ be a non-empty positive admissible word and $j \neq \alpha_{|\alpha|-1}$. Then,
\begin{equation*}
    C_{\alpha j^{-1}}^c = \begin{cases}
                            C_{\alpha (\alpha_{|\alpha|-1} - 1)}^c, \quad \text{if } \alpha_{|\alpha|-1}>2 \text{ and } j=1;\\
                            C_{\alpha (\alpha_{|\alpha|-1} - 1)}^c, \quad \text{if } \alpha_{|\alpha|-1}>2 \text{ and }, \text{ } j=2 \text{ and } \alpha_{|\alpha|-1} \text{ is even};\\
                            F_\alpha \sqcup \left( \bigsqcup_{m=0}^{|\alpha|-1}\bigsqcup_{p\neq \alpha_m}  C_{\delta^m(\alpha )p}\right), \quad  \text{if } \alpha_{|\alpha|-1}=2 \text{ and } j=1;\\
                            \{\alpha \xi^{0,2}\} \sqcup F_\alpha \sqcup \left( \bigsqcup_{m=0}^{|\alpha|-1}\bigsqcup_{p\neq \alpha_m}  C_{\delta^m(\alpha )p}\right)\sqcup \bigsqcup_{\substack{p\in \mathbb{N}}}  C_{\alpha (2p+1)}, \quad  \text{if } \alpha_{|\alpha|-1}=1 \text{ and } j=2;\\
                            X_A, \quad \text{otherwise};
                        \end{cases}
\end{equation*}
\end{corollary}

\begin{proof} It is straightforward from Proposition \ref{prop:general_C_alpha_j_inverse_complement}, Lemma \ref{lemma:G_K_Pair_renewal} and Corollary \ref{cor:C_alpha_j_inv_Pair_renewal}. 
\end{proof}

Unlike the renewal shift case in subsection \ref{subsec:cylinder_topology_Renewal}, the cylinders in the form $C_{\alpha j^{-1}}$ cannot be written as another cylinder. However, they are pairwise disjoint unions of positive generalized cylinders, jointly with finite subsets of $Y_A$. And the same occurs with their complements. On the other hand, we only need to consider when $\alpha$ ends in a symbol of $\{1,2\}$. Similarly to the renewal shift case, from Theorem \ref{thm:huge_generators_intersections}, it is straightforward that every finite intersection among the elements of the subbasis is in the form
\begin{equation}\label{eq:basis_pair_renewal_is_finite_set_union_countable_cylinders}
    F \sqcup \bigsqcup_{n \in \mathbb{N}} C_{w(n)},
\end{equation}
where $F \subseteq Y_A$ is finite and $w(n)$ is a positive admissible word for each every $n \in \mathbb{N}$.

\section{The Groupoid approach of the Exel-Laca algebras} \label{section:GRD_groupoid}

In this section the groupoid C$^*$-algebras theory meets the Exel-Laca algebras, as it is in Renault's work \cite{Renault1999}. Our main goal is to present the isomorphism between the Exel-Laca algebras and the full groupoid C$^*$-algebra of the generalized Renault-Deaconu groupoid, proved in the aforementioned paper. Consider a locally compact topological space $X$ and a local homeomorphism $T: \Dom T \to \ran T$, where $\Dom T$ and $\ran T$ are open subsets of $X$. The pair $(X,T)$ is called \emph{Singly Generated Dynamical System} (SGDS). Since we often use the maps $T^p$, $p \in \mathbb{N}$, defined on domains where these maps make sense, we make it clear what are the domains of these respective iterated maps and that these maps are local homeomorphisms, as it is stated and proved in the next lemma.

\begin{lemma}\label{lemma:Tp_local_homeomorphism} Let $(X,T)$ be a SGDS. For every $p \in \mathbb{N}$, consider the maps
\begin{equation*}
    T^p: \Dom (T^p) \to X,
\end{equation*}
where, the domain above is the largest one in $\Dom T$. Then,
\begin{equation}\label{eq:domain_Tp}
    \Dom (T^p) = T^{-(p-1)}(\Dom T),
\end{equation}
where $T^0$ is the identity map. Moreover, $T^p$ is a local homeomorphism.
\end{lemma}

\begin{proof} We have that $x \in \Dom (T^p)$ if and only if $T^p(x)$ is well-defined, that is $T^{p-1}(x) \in \Dom(T)$, i.e., $x \in T^{-(p-1)}(\Dom(T))$ and then the validity of the equation \eqref{eq:domain_Tp} is straightforward. Observe that it is immediate that $\Dom(T^{k+1}) \subseteq \Dom(T^k)$ for every $k \in \mathbb{N}_0$. Now fix $x \in \Dom(T^p)$. Since $T$ is a local homeomorphism, for each $k = 0,...,p-1$, there exists an open set $U_{T^k(x)}$ containing $T^k(x)$ such that the map $T$ restricted to $U_{T^k(x)}$ is a homeomorphism onto its image. Then the set
\begin{equation*}
    U = \bigcap_{k=0}^{p-1}T^{-k}(U_{T^k(x)})
\end{equation*}
is an open set containing $x$ such that $T^p\vert_U$ is a homeomorphism onto its image. 
\end{proof}

\begin{remark}\label{remark:Tp_homeomorphism_implies_T^k_homeomorphism_for_k_leq_p} Note that if $U$ is an open set such that $T^p\vert_U$ is a homeomorphism onto its image, then $T^k\vert_U$ is also a homeomorphism onto its image for every $k \leq p$. 
\end{remark}

The main object of this section is the generalized Renault-Deaconu groupoid\footnote{The term `generalized' is used to indicate that the local homeomorphsim $T$ is not necessarily defined in the whole space $X$.}, and it is presented next.

\begin{definition}[generalized Renault-Deaconu groupoid]\label{def:Renault_Deaconu_groupoid} Let $(X,T)$ be a SGDS. The generalized Renault-Deaconu groupoid is given by
\begin{equation}\label{eq:ADR_groupoid_generalized}
    \mathcal{G}(X,T) = \left\{  (x,k,y) \in X \times \mathbb{Z} \times X \text{ }\text{ }:\begin{array}{l l}
                            & \exists n,m \in \mathbb{N}_0 \text{ s.t. } k=n-m, \\
                            &  x \in \Dom (T^n), y \in \Dom (T^m), T^n(x) = T^m(y)\\
                          \end{array}\right\}
\end{equation}
where $T^0$ is the identity map defined on the whole space $X$, and set of composable parts is defined as
\begin{equation*}
    \mathcal{G}^{(2)}:= \left\{\big((x,k,z),(w,l,y)\big) \in \mathcal{G}(X,T) \times \mathcal{G}(X,T): z = w \right\},
\end{equation*}
on which in it we define the product
\begin{equation*}
    (x,k,z)(z,l,y) := (x,k+l,y).
\end{equation*}
The inverse map is given by
\begin{equation*}
    (x,k,y)^{-1} := (y,-k,x).
\end{equation*} 
\end{definition}

\begin{remark} Observe that in fact the product and the inverse maps are well-defined. Indeed, if $\big((x,k,z),(z,l,y)\big) \in \mathcal{G}^{(2)}$, then there exist $n,m,p,q\in \mathbb{N}$ such that $k = n-m$, $l = p-q$, $x \in \Dom (T^n)$, $z \in \Dom (T^m)\cap \Dom (T^p)$, $ y \in \Dom (T^q)$, and satisfying the indentities
\begin{equation*}
    T^n(x) = T^m(z) \quad \text{and} \quad T^p(z) = T^q(y).
\end{equation*}
W.l.o.g. we may suppose that $m \geq p$, and then we may write $m = p+d$, $d \in \mathbb{N}$. In the proof of Lemma \ref{lemma:Tp_local_homeomorphism}, we observe that $\Dom(T^{k+1}) \subseteq \Dom(T^k)$ for every $k \in \mathbb{N}_0$, that is, if $w \in \Dom T^t$ for some $t \in \mathbb{N}_0$ then $T^s(w) \in \Dom (T^{t-s})$ for every $s \leq t$. Then, $T^q(y) = T^p(z) \in \Dom (T^d)$, and therefore $y \in \Dom(T^{q+d})$. Moreover,
\begin{equation}
    T^n(x) = T^m(z) = T^d(T^p(z)) = T^d(T^q(y)) = T^{d+q}(y),
\end{equation}
and observe that $k+l = n-m +p-q = n-(d+q)$ and therefore $(x,k+l,y) \in \mathcal{G}(X,T)$. The well-definition of the inverse map is straightforward.
\end{remark}

A simple calculation let us determine the set of units $\mathcal{G}^{(0)}$:
\begin{equation*}
    (x,k,y)(x,k,y)^{-1} = (x,k,y)(y,-k,x)= (x,0,x),
\end{equation*}
that is
\begin{equation*}
    \mathcal{G}^{(0)} = \{(x,0,x) : x \in X\}.
\end{equation*}
And the range and source maps, denoted respectively $r$ and $s$, as
\begin{equation*}
    r(x,k,y) = (x,0,x) \quad \text{and} \quad s(x,k,y) = (y,0,y).
\end{equation*}

In order to construct the groupoid C$^*$-algebra of $\mathcal{G}(X,T)$, we construct a topology that, under some circunstances which we describe further, it fits to the theory presented in chapter \ref{ch:Groupoid_algebras}.

\begin{definition}[Topology generators for $\mathcal{G}(X,T)$]\label{def:basis_RD_groupoid} For the groupoid $\mathcal{G}(X,T)$, define the sets
\begin{equation*}
    W(n,m,U,V) := \left\{(x,n-m,y): x \in U, y \in  V, T^n(x) = T^m(y) \right\},
\end{equation*}
where $n,m \in \mathbb{N}_0$ and the sets $U$ and $V$ are respective open subsets of $\Dom(T^n)$ and $\Dom(T^m)$, such that $T^n\vert_U$ and $T^m\vert_V$ are injective.
\end{definition}

The generators above form a topological basis for $\mathcal{G}(X,T)$, as we prove next.

\begin{proposition}\label{prop:generators_form_a_basis_for_RD_groupoid} The family of all sets $W(n,m,U,V)$ as in Definition \ref{def:basis_RD_groupoid} forms a topological basis for $\mathcal{G}(X,T)$. 
\end{proposition}

\begin{proof} It is straightforward that such family covers $\mathcal{G}(X,T)$. Let two elements of that family, we say $W(n_1,m_1,U_1,V_1)$ and $W(n_2,m_2,U_2,V_2)$, with non-empty intersection. We claim that such intersection is actually an element of the same family. Indeed, if the intersection is non-empty, then $n_1 - m_1 = n_2 - m_2$ and w.l.o.g. we may assume $n_1 \geq n_2$ and consequently $m_1 \geq m_2$. Let
\begin{equation*}
    U = U_1 \cap U_2 \quad \text{and} \quad  V = V_1 \cap V_2.
\end{equation*}
Note that $U$ is such that $T^{r}\vert_U$ and $T^{s}\vert_V$ are homeomorphisms onto their respective images for every $r \leq n_1$ and $s \leq m_1$. We prove that
\begin{equation}\label{eq:W_cap_W}
    W(n_1,m_1,U_1,V_1)\cap W(n_2,m_2,U_2,V_2) = W(n_2,m_2,U,V).
\end{equation}
In fact, suppose that $(x,k,y) \in W(n_1,m_1,U_1,V_1)\cap W(n_2,m_2,U_2,V_2)$ then $x \in U_1 \cap U_2 = U$, $y \in V_1 \cap V_2 = V$ and $T^{n_2}(x) = T^{m_2}(y)$ and therefore $(x,k,y) \in W(n_2,m_2,U,V)$. Conversely let $(x,k,y) \in W(n_2,m_2,U,V)$. Then, $x \in U_1 \cap U_2$ and $y \in V_1 \cap V_2$ and $T^{n_2}(x) = T^{m_2}(y)$. Note that $x \in \Dom T^{n_1}$ and $y \in \Dom T^{m_1}$, and hence $T^{n_2+k}(x) = T^{m_2+k}(y)$ for every $k \in \{0,...,n_1 - n_2\}$ (note that $n_1-n_2 = m_1 - m_2$). In particular, $T^{n_1}(x) = T^{m_1}(y)$. Therefore, $(x,k,y) \in W(n_1,m_1,U_1,V_1)\cap W(n_2,m_2,U_2,V_2)$. It is straightforward that the family of the sets in the form as in Definition \ref{def:basis_RD_groupoid} is a topological basis. 
\end{proof}

Now we show that the topology generated by the basic open sets of Definition \ref{def:basis_RD_groupoid} makes $\mathcal{G}(X,T)$ a topological groupoid. From now on we endow $\mathcal{G}(X,T)$ with such topology. We need the following two auxiliary lemmas.

\begin{lemma} 
\label{lemma:RenaultDeaconuCoordinatesConverge}
Consider a net $\{ (x_d, k_d, y_d) \}_{d \in D}$ in $\mathcal{G}(X,T)$ converging to $(x,k,y)$. Then $x_d \to x$, $y_d \to y$ and there exists $d_0$ such that $k_d = k$ for every $d \succ d_0$. Consequently, w.l.o.g. we may take $k_d$ constant. 
\end{lemma}

\begin{proof}
Let $n, m \in \mathbb{N}_0$ such that $T^n(x) = T^m(y)$ and $k = n - m$. Let $U, V$ be open neighborhoods of $x, y$ such that $T^n\vert_U$ and $T^m\vert_V$ are injective, respectively. Then there exists $d_0$ such that for every $d \succ d_0$, $(x_d, k_d, y_d) \in W(n,m,U,V)$, hence
\begin{align*}
x_d \in U, \quad y_d \in V, \quad \text{and} \quad k_d = n -m = k.
\end{align*}
Therefore $k_d$ is eventually constant, $x_d \rightarrow x$ and $y_d \rightarrow y$.
\end{proof}

For the next lemma, given $n,m \in \mathbb{N}$ and $x,y \in X$, we define $\mathfrak{P}^{n,m}(x,y)$ being the set of pairs $(U,V)$ of subsets of $X$ satisfying the following:
\begin{itemize}
 \item $U$ and $V$ are open neighborhoods of $x$ and $y$, respectively;
 \item $T\vert_{U}^n$ and $T\vert_{V}^m$ are injective\footnote{And then $x \in \Dom(T^{n})$ and $y \in \Dom(T^{m})$.}.
\end{itemize}

\begin{lemma} 
\label{lemma:RenaultDeaconuConvergeceSufficient}
Let $n_0, m_0 \in \mathbb{N}$ and $x,y \in X$. Consider a net $\{ (x_d, k, y_d) \}_{d \in D}$ in $\mathcal{G}(X,T)$, and suppose that for every $(U,V) \in \mathfrak{P}^{n_0,m_0}(x,y)$, there exists $d_0$ satisfying $(x_d, k, y_d) \in W(n_0,m_0,U,V)$ for $d \succ d_0$.

Then, for every $(U,V) \in \mathfrak{P}^{n,m}(x,y)$, where $n, m \in \mathbb{N}$, and $(x, k, y) \in W(n,m,U,V)$, there exists $d_0$ such that
\begin{align*}
(x_d, k, y_d) \in W(n,m,U,V) \quad \text{for $d \succ d_0$}.
\end{align*}
Consequently, $(x_d, k, y_d) \to (x,k,y)$.
\end{lemma}

\begin{proof}
Note that $k = n_0 - m_0$. Let $n,m \in \mathbb{N}$ such that $T^n(x) = T^m(y)$ (hence $x \in \Dom(T^{n})$ and $y \in \Dom(T^{m})$) and $k = n - m$ and consider $U$ and $V$ open neighborhoods of $x$ and $y$, respectively, s.t. $T\vert_U^{n}$ and $T\vert_V^{m}$ are injective. Observe that for every $U'$ and $V'$ respective open neighborhoods of $x$ and $y$, the open sets
\begin{equation*}
    \widetilde{U} = U'\cap U \cap U'' \quad \text{and} \quad \widetilde{V} = V'\cap V \cap V'',
\end{equation*}
satisfy $(U'',V'') \in \mathfrak{P}^{n_0,m_0}(x,y)$. By hypothesis, there exists $d_0 \in D$ such that $d \succ d_0$ implies $(x_d,k,y_d) \in W(n_0,m_0,\widetilde{U},\widetilde{V})$. We have two possibilities:
\begin{itemize}
    \item $n \geq n_0$. In this case, let $p = n-n_0 = m-m_0\geq 0$, then
    \begin{equation*}
        T^n(x_d) = T^p(T^{n_0}(x_d)) = T^p(T^{m_0}(y_d)) = T^m(y_d).
    \end{equation*}
    Hence, $(x_d,k,y_d) \in W(n,m,\widetilde{U},\widetilde{V}) \subseteq W(n,m,U,V)$ for $d \succ d_0$;
    \item $n < n_0$. We necessarily have that $T^p\vert_{T^{n_0-p}(\widetilde{U})}$ and $T^q\vert_{T^{m_0-q}(\widetilde{V})}$ are homeomorphisms onto their respective images for every $p \in \{0,...,n_0\}$ and $q \in \{0,...,m_0\}$, then by taking $p = n-n_0 = m-m_0<0$ we have 
    \begin{align*}
        T^n(x_d) &= T^n\vert_{\widetilde{U}}(x_d) = T^p\vert_{T^{n_0}(\widetilde{U})}(T^{n_0}\vert_{\widetilde{U}}(x_d)) = T^p\vert_{T^{n_0}(\widetilde{U})}(T^{m_0}\vert_{\widetilde{V}}(y_d)) \\
        &= T^p\vert_{T^{m_0}(\widetilde{V})}(T^{m_0}\vert_{\widetilde{V}}(y_d)) = T^m(y_d).
    \end{align*}
    And then we also have $(x_d,k,y_d) \in W(n,m,\widetilde{U},\widetilde{V}) \subseteq W(n,m,U,V)$.
\end{itemize}
It is straightforward that $(x_d, k, y_d) \to (x,k,y)$.
\end{proof}

\begin{corollary} 
\label{corollary:RenaultDeaconuConvergenceSufficient}
Let $(x, n - m, y) \in \mathcal{G}(X,T)$. Let $\{ x_d \}_{d \in D}$, $\{ y_d \}_{d \in D}$ be sequences in $X$ such that $x_d \to x$ and $y_d \to y$. If $T^n(x_d) = T^m(y_d)$ for each $d$, then $(x_d, n - m, y_d) \to (x, n-m, y)$ in $\mathcal{G}(X,T)$.
\end{corollary}

\begin{proof} Let $U$ and $V$ open neighborhoods of $x$ and $y$, respectively, and s.t. $T\vert_{U}^n$ and $T\vert_V^m$ are injective. There exists $d_0 \in D$ such that $d \succ d_0$ implies $x_d \in U$ and $y_d \in V$. Since $T^n(x_d) = T^m(y_d)$, it follows that $(x_d, n-m, y_d) \in W(n,m,U,V)$, and Lemma \ref{lemma:RenaultDeaconuConvergeceSufficient} gives that $(x_d, n-m, y_d) \to (x,n-m,y)$ in $\mathcal{G}(X,T)$.
\end{proof}

Now we show that the Hausdorff property on $X$ implies that $\mathcal{G}(X,T)$ is not only an \'etale topological groupoid, but also it is Hausdorff topological space.

\begin{theorem}\label{thm:RD_groupoid_is_topological_and_Hausdorff_for_X_Hausdorff} Let $(X,T)$ be a SGDS such that $X$ is Hausdorff and consider its generalized Renault-Deaconu groupoid $\mathcal{G}(X,T)$ endowed with the topology generated by the basic open sets in Definition \ref{def:basis_RD_groupoid} is an \'etale topological groupoid. Moreover, $\mathcal{G}(X,T)$ is Hausdorff.
\end{theorem}

\begin{proof} Let $\theta : \mathcal{G}^{(2)} \to \mathcal{G}(X,T)$ and $\psi : \mathcal{G}(X,T) \to \mathcal{G}(X,T)$ be the product and inverse operations of $\mathcal{G}(X,T)$, respectively. We divide the proof that $\mathcal{G}(X,T)$ is a topological groupoid in three steps as follows.

\begin{itemize}
    \item[\textbf{Step 1.}] $\mathcal{G}^{(2)}$ is closed in $\mathcal{G} \times \mathcal{G}$.

    Consider a net $\{ (g_d,h_d) \}_{d \in D}$ in $\mathcal{G}^{(2)}$ converging to $(g,h) \in \mathcal{G} \times \mathcal{G}$. Then $g_d = (x_d, k_d, y_d)$, $h_d = (y_d, \ell_d, z_d)$ for each $d \in D$. Assume $g = (x,k,y)$ and $h = (y', l ,z)$. By Lemma \ref{lemma:RenaultDeaconuCoordinatesConverge}, we obtain $y_d \to y$ and $y_d \to y'$ in $X$. Because $X$ is Hausdorff, we have $y = y'$, and therefore $(g,h) \in \mathcal{G}^{(2)}$.

    \item[\textbf{Step 2.}] $\psi$ is continuous.

    Consider  a basic open set $W(m,n,V,U)$ of $\mathcal{G}$. Then,
    \begin{align*}
    \psi^{-1}(W(m,n,V,U)) &= \left\{ (x,k,y) \in \mathcal{G}(X,T): (y,-k, x) \in W(m,n,V,U) \right\} \\
    &= \left\{ (x,k,y) \in \mathcal{G}(X,T): y \in V, x \in U, T^m(y) = T^n(x), k = n - m \right\} \\
    &= W(n,m,U,V).
    \end{align*}
    We conclude that the inverse map is continuous.

    \item[\textbf{Step 3.}] $\theta$ is continuous.

    Let $\{(g_d,h_d)\}_{d \in D}$ be a net in $\mathcal{G}^{(2)}$ converging to $(g,h)$. By step 1, it follows that $(g,h) \in \mathcal{G}^{(2)}$. So we may set $g = (x,k,y), h = (y,\ell,z)$, and by Lemma \ref{lemma:RenaultDeaconuCoordinatesConverge} we may also write $g_i = (x_i,k,y_i)$ and $h_i = (y_i,\ell,z_i)$ for each $d \in D$.

    Consider $n_1, m_1, n_2, m_2 \in \mathbb{N}_0$ such that $x \in \Dom(T^{n_1})$, $y \in \Dom(T^{n_2}) \cap \Dom(T^{m_1})$, $z \in \Dom(T^{m_2})$, and satisfying
    \begin{align*}
    k = n_1 - m_1, \quad T^{n_1}(x) = T^{m_1}(y) \quad \text{and} \quad \ell = n_2 - m_2, \quad T^{n_2}(y) = T^{m_2}(z).
    \end{align*}
    Since $(x_d, k, y_d) \to (x, k, y)$ and $(y_d, \ell, z_d) \to (y, \ell, z)$, it follows that for $U$, $V$ and $W$ open neighborhoods of $x$, $y$ and $z$ s.t. $T^{n_1}\vert_U$, $T^{\max\{m_1,n_2\}}\vert_V$ and $T^{m_2}\vert_W$ are homeomorphisms, there exists $d_0\in D$ such that
    \begin{align*}
    (x_d, k, y_d) \in W(n_1,m_1,U,V) \quad \text{and} \quad (y_d, k, z_d) \in W(n_2,m_2,V,W), \quad \text{for} \quad d \succ d_0.
    \end{align*}
    Hence, for $d \succ d_0$, we have two possibilities:
    \begin{itemize}
     \item[$\bullet$] $m_1 \geq n_2$. In this case we have,
     \begin{align*}
        T^{n_1}(x_d) &= T^{m_1}(y_d) = T^{m_1}\vert_V(y_d) = T^{m_1-n_2}\vert_{T^{n_2}(V)}(T^{n_2}\vert_V(y_d)) \\
        &= T^{m_1-n_2}\vert_{T^{n_2}(V)}(\underbrace{T^{m_2}\vert_W(z_d)}_{\in T^{n_2}(V)}) = T^{m_1+m_2-n_2}(z_d) = T^{m_1+m_2-n_2}\vert_W(z_d).
     \end{align*}
     Observe that $k = n_1 - (m_1+m_2-n_2) = k + \ell$, and then $(x_d,k+\ell,z_d) \in W(n_1,m_1+m_2-n_2,U,W)$ for every $d \succ d_0$, and by Lemma \ref{lemma:RenaultDeaconuConvergeceSufficient} we conclude that $(x_d,k+\ell,z_d) \to (x,k+\ell,z)$;
     \item[$\bullet$] $m_1 < n_2$. In this case we have,
     \begin{align*}
        T^{m_2}(z_d) &= T^{n_2}(y_d) = T^{n_2}\vert_V(y_d) = T^{n_2-m_1}\vert_{T^{m_1}(V)}(T^{m_1}\vert_V(y_d)) \\
        &= T^{n_2-m_1}\vert_{T^{m_1}(V)}(\underbrace{T^{n_1}\vert_U(x_d)}_{\in T^{m_1}(V)}) = T^{n_2-m_1+n_1}(x_d) = T^{n_2-m_1+n_1}\vert_U(x_d).
     \end{align*}
     Similarly to the previous case, we observe that $k = (n_2 - m_1 +n_1) - m_2 = k + \ell$, and then $(x_d,k+\ell,z_d) \in W(n_1+n_2-m_1,m_2,U,W)$ for every $d \succ d_0$, and again by Lemma \ref{lemma:RenaultDeaconuConvergeceSufficient} we obtaain that $(x_d,k+\ell,z_d) \to (x,k+\ell,z)$.
    \end{itemize}
    Therefore $\theta$ is continuous.
\end{itemize}
Since the inverse and product maps are continuous, we conclude that $\mathcal{G}(X,T)$ is a topological groupoid. Now we show its etalicity. Indeed, take $(x, n - m, y) \in \mathcal{G}(X,T)$ satisfying
\begin{equation*}
    x \in \Dom(T^n), \quad y \in \Dom(T^m) \quad \text{and} \quad T^n(x) = T^m(y).
\end{equation*}
Since $T$ is a local homeomorphism, there are $U \subseteq \Dom(T^n)$ and $V \subseteq \Dom(T^m)$ open neighborhoods of $x$ and $y$, respectively, s.t.
\begin{align*}
\text{$T^n(U)$ is open and } &\text{$T^n\vert_U: U \to T^n(U)$ is a homeomorphism},\\
\text{$T^m(V)$ is open and } &\text{$T^m\vert_V: B \to T^m(B)$ is a homeomorphism}.
\end{align*}
Hence, $(x, n-m, y) \in W(n,m,U,V)$. Since $\mathcal{G}$ is a topological groupoid, $r$ is continuous (Remark \ref{remark:r_and_s_continuous}). We prove that $r$ is a local homeomorphism by showing that the following claims are true.

\begin{itemize}
\item[\textbf{Claim 1.}] $r\vert_{W(n,m,U,V)}$ is injective.

Suppose there exist $x_1, x_2 \in U$, $y_1, y_2 \in V$ such that $r(x_1, n-m, y_1) = r(x_2, n-m, y_2)$. Then $y_1 = y_2$. In addition,
\begin{align*}
x_1 = T^{-n}\vert_U(T^m(y_1)) = T^{-n}\vert_U(T^m(y_2)) = x_2.
\end{align*}
Therefore $(x_1, n-m, y_1) = (x_2, n-m, y_2)$ and the claim is proved.

\item[\textbf{Claim 2.}] $r(W(n,m,U,V))$ is open.
\begin{align*}
r(W(n,m,U,V))
&= \{ (y,0,y) \in \mathcal{G}(X,T): (x, n-m, y) \in W(n,m,U,V) \}\\
&= \{ (y,0,y) \in \mathcal{G}(X,T): x \in U, y \in V, T^n(x) =T^m(y) \}\\
&= \{ (y,0,y) \in \mathcal{G}(X,T): y \in V, x = T\vert_U^{-n}(T^m(y)) \in U \}\\
&= \{ (y,0,y) \in \mathcal{G}(X,T): y \in V, T^m(y) \in T^n(U) \} \\
&= \{ (y,0,y) \in \mathcal{G}(X,T): y \in V, y \in T\vert_V^{-m}(T^n(U))\}\\
&= W(0,0,C,C),
\end{align*}
where $C = V \cap T\vert_V^{-m}(T^n(U))$, an open set.

\item[\textbf{Claim 3.}] $r\vert_{W(n,m,U,V)}^{-1}$ is continuous.

Let $\{ (y_d, 0, y_d) \}_{d \in D}$ be a net in $r(W(n,m,U,V))$ converging to some $(y,0,y)$ in $r(W(n,m,U,V))$. Then $y_d \to y \in V$. Define the net $x_d = T\vert_U^{-n}(T^m(y_d))$. Then $x_d \to x = T\vert_U^{-n}(T^m(y))$. Observe that $x_d$ is the unique element in $U$ s.t. $T^n(x_d) = T^m(y_d)$, then
\begin{align*}
(x_d, n - m, y_d) = r\vert_{W(n,m,U,V)}^{-1}(y_d, 0, y_d).
\end{align*}
Analogously, $(x, n - m, y) = r\vert_{W(n,m,U,V)}^{-1}(y, 0, y)$. Then, it follows from Corollary \ref{corollary:RenaultDeaconuConvergenceSufficient} that $(x_d, n-m, y_d) \to (x,n-m,y)$.
\end{itemize}
We conclude that $r\vert_{W(n,m,U,V)}$ is a homeomorphism onto its image. Since $W(n,m,U,V)$ is an arbitrary basic set, it follows that $r$ is a local homeomorphism. The proof that $s$ is a local homeomorphism holds by similar steps, and therefore $\mathcal{G}(X,T)$ is \'etale.

Finally, we prove that the groupoid topology is Hausdorff. Indeed, let $(x_1,k_1,y_1),(x_2,k_2,y_2) \in \mathcal{G}(X,T)$, where $(x_1,k_1,y_1) \neq (x_2,k_2,y_2)$. Let $n_1,n_2,m_1,m_2 \in \mathbb{N}$ s.t. $k_i = n_i - m_i$, $x_i \in \Dom(T^{n_i})$, $y_i \in \Dom(T^{m_i})$ and $T^{n_i}(x_i) = T^{m_i}(y_i)$, for $i = 1,2$. We have two possibilities:
\begin{itemize}
\item $n_1 - m_1 \neq n_2 - m_2$. Then, $(x_1,k_1,y_1) \in W(n_1,m_1,\Dom(T^{n_1}),\Dom(T^{m_1}))$, $(x_2,k_2,y_2) \in W(n_2,m_2,\Dom(T^{n_2}),\Dom(T^{m_2}))$ and
\begin{equation*}
    W(n_1,m_1,\Dom(T^{n_1}),\Dom(T^{m_1})) \cap W(n_2,m_2,\Dom(T^{n_2}),\Dom(T^{m_2})) = \emptyset;
\end{equation*}

\item $n_1 - m_1 = n_2 - m_2$. In this case, we have $x_1 \neq x_2$ or $y_1 \neq y_2$. W.l.o.g. suppose $x_1 \neq x_2$. Since $X$ is Hausdorff, we can choose $U_i \subseteq \Dom(T^{n_i})$ open neighborhood of $x_i$ and such that $T^{n_i}\vert_{U_i}$ is injective, respectively for $i = 1,2$, and satisfying $U_1 \cap U_2 = \emptyset$. Then $(x_i,k_i,y_i) \in W(n_1,m_1,U_1,V_1)$, where $V_i \subseteq \Dom(T^{m_i})$ is any open neighborhood of $y_i$ s.t. $T^{m_i}\vert_{V_i}$ is injective, $i=1,2$. Also, $W(n_1,m_1,U_1,V_1) \cap W(n_2,m_2,U_2,V_2) = \emptyset$.
\end{itemize}
Therefore, $\mathcal{G}(X,T)$ is Hausdorff.
\end{proof}

\begin{corollary} If $X$ is Hausdorff, then the family of open basic sets of Definition \ref{def:basis_RD_groupoid} are open bisections. 
\end{corollary}

\begin{proof} It is straightforward from the proof of Theorem \ref{thm:RD_groupoid_is_topological_and_Hausdorff_for_X_Hausdorff} from the part where it is proved that $G(X,T)$ is \'etale. 
\end{proof}

The next result shows how the other topological properties on $X$ reflect on the properties of the groupoid topology.

\begin{lemma}\label{lemma:injective_basis} Let $(X,T)$ be a SGDS and $p \in \mathbb{N}_0$. If $X$ second countable, then there exists a countable basis for $\Dom(T^p)$ such that the restriction of $T^p$ to its basic open sets is injective. 
\end{lemma}

\begin{proof} Every subspace of a second countable of a topological space is also second countable, then $\Dom(T^p)$ is second countable. Let $\mathcal{B} = \{B_i\}_{i \in \mathbb{N}}$ be a basis for $\Dom(T^p)$ and $x$ an element in such space. Since $T^p$ is a local homeomorphism, there exists $U_x$ an open neighborhood of $x$ such that $T^p\vert_{U_x}$ is injective. Since $\mathcal{B}$ is a basis, there exists $\iota(x) \in \mathbb{N}$ such that $x \in B_{\iota(x)} \subseteq U_x$, and hence $T^p\vert_{B_{\iota(x)}}$ is also injective. Then, there exists there exists a function $\iota: B_i \to \mathbb{N}$ such that
\begin{equation*}
    B_i = \bigcup_{x \in B_i} B_{\iota(x)} \stackrel{(\bullet)}{=} \bigcup_{j \in J_i} B_j^i,
\end{equation*}
where in $(\bullet)$ we used the fact that $\mathcal{B}$ is countable, $J_i := \iota(B_i) \subseteq \mathbb{N}$ for each $i \in \mathbb{N}$, and $B_j^i \in \mathcal{B}$ is such that $T^p\vert_{B_j^i}$ is injective. We have that the family
\begin{equation*}
    \left\{B_j^i: i \in \mathbb{N}, j \in J_i\right\}
\end{equation*}
is a sub-basis for the topology of $\Dom(T^p)$ such that $T^p$ restricted to each one of these elements is injective. Consequently, the same occurs to the basis generated by that sub-basis. 
\end{proof}

\begin{theorem}\label{thm:X_second_countable_implies_RDG_second_countable_and_same_for_local_comapactness} Given a $(X,T)$ a SGDS and let $\mathcal{G}(X,T)$ be its generalized Renault-Deaconu groupoid. The following holds:
\begin{itemize}
    \item[$(a)$] if $X$ is second countable, then then $\mathcal{G}(X,T)$ is second countable;
    \item[$(b)$] if $X$ is locally compact and Hausdorff, then $\mathcal{G}(X,T)$ is locally compact;
\end{itemize} 
\end{theorem}

\begin{proof} Item $(a)$: let $n,m \in \mathbb{N}_0$. By Lemma \ref{lemma:injective_basis}, both $\Dom(T^n)$ and $\Dom(T^m)$ have the respective countable bases
\begin{equation*}
    \mathcal{B}_n:=\{B_i(n):i \in \mathbb{N}\} \quad \text{and} \quad \mathcal{B}_m:=\{B_j(m):j \in \mathbb{N}\}
\end{equation*}
such that $T^n\vert_{B_i(n)}$ and $T^n\vert_{B_j(m)}$ are injective for every $i,j \in \mathbb{N}$. We claim the countable family
\begin{equation}\label{eq:basis_G_countable}
    \left\{W(n,m,B_i(n),B_j(m)): n,m \in \mathbb{N}_0; i,j \in \mathbb{N}\right\}
\end{equation}
is a basis for $\mathcal{G}(X,T)$. In fact, given
\begin{align*}
    (x,k,y) &\in W(n_1,m_1,B_{i_1}(n_1),B_{j_1}(m_1)) \cap W(n_2,m_2,B_{i_2}(n_2),B_{j_2}(m_2))\\
    &\stackrel{(\dagger)}{=} W(n,m,B_{i_1}(n_1)\cap B_{i_2}(n_2),B_{j_1}(m_1)\cap B_{j_2}(m_2)),
\end{align*}
where in $(\dagger)$ we used the equation \eqref{eq:W_cap_W}, $n = \min\{n_1,n_2\}$ and $m = \min\{m_1,m_2\}$. Since $\mathcal{B}_n$ and $\mathcal{B}_m$ are basis, there are $B(n) \in \mathcal{B}_n$ and $B(m) \in \mathcal{B}_m$ such that $x \in B(n) \subseteq B_{i_1}(n_1)\cap B_{i_2}(n_2)$ and $y \in B(m) \subseteq B_{j_1}(m_1)\cap B_{j_2}(m_2)$. Hence,
\begin{equation*}
    (x,k,y) \in W(n,m,B(n),B(m)) \subseteq W(n_1,m_1,B_{i_1}(n_1),B_{j_1}(m_1)) \cap W(n_2,m_2,B_{i_2}(n_2),B_{j_2}(m_2)).
\end{equation*}
Since the family \eqref{eq:basis_G_countable} covers $\mathcal{G}(X,T)$ we conclude that such collection is a countable basis, and therefore $\mathcal{G}(X,T)$ is second countable.

Item $(b)$: given $n,m \in \mathbb{N}_0$, let $U \subseteq \Dom(T^n)$ and $V \subseteq \Dom(T^m)$ be open sets s.t. $\overline{U}$ and $\overline{V}$ are compact subsets of $\Dom(T^n)$ and $\Dom(T^m)$, respectively, and s.t. $T^n\vert_U$ and $T^m\vert_V$ are injective. Then, $W(n,m,U,V) \subseteq W(n,m,\overline{U},\overline{V})$, where
\begin{align*}
W(n,m,\overline{U},\overline{V}) = \{ (x, n-m, y) \in \mathcal{G}(X,T): T^n(x) = T^m(y), x \in \overline{U}, y \in \overline{V} \}.
\end{align*}
Let $\{ (x_d, n - m, y_d) \}_{d \in D}$ be a net in $W(n,m,\overline{U},\overline{V})$. Then $\{ (x_d, y_d) \}_{d \in D}$ is a net in the compact set $\overline{U \times V}$. Hence, there exists converging subnet $\lbrace (x'_{d'}, y'_{d'}) \rbrace_{d' \in D'}$ satisfying $x'_{d'} \rightarrow x$ for some $x \in \overline{U}$, and $y'_{d'} \rightarrow y$ for some $y \in \overline{V}$. By continuity of $T$, we have $T^n(x) = T^m(y)$, and therefore $W(n,m,\overline{U},\overline{V})$ is compact. We conclude that $\mathcal{G}(X,T)$ is locally compact. 
\end{proof}

The next result allows us to identify topologically $X$ with $\mathcal{G}^{(0)}$.

\begin{proposition} Given $(X,T)$ a SGDS, the unit space $\mathcal{G}^{(0)}$ of its generalized Renault-Deaconu groupoid, endowed with the subspace topology, is homeomorphic to $X$.
\end{proposition}

\begin{proof} It is straightforward that the map $\rho:\mathcal{G}^{(0)} \to X$, given by $\rho((x,0,x)) = x$ is a bijection. We prove that such map is a homeomorphism. Let $n, m \in \mathbb{N}_0$, $U \subseteq \Dom(T^n)$ and $V \subseteq \Dom(T^m)$ open sets. If $W(n,m,U,V) \cap \mathcal{G}^{(0)} \neq \emptyset$, then $n = m$. In this case,
\begin{align*}
\rho(W(n,m,U,V) \cap \mathcal{G}^{(0)}) = \rho(\lbrace (x,0,x): x \in U \cap V \rbrace) = U \cap V,
\end{align*}
which is open in $X$, and then $\rho^{-1}$ is continuous. On the other hand, for any open set $U \subset X$, we have
\begin{align*}
\rho^{-1}(U) = \lbrace (x,0,x): x \in U \rbrace = W(0,0,U,U),
\end{align*}
and then $\rho$ is continuous. We conclude that $\rho$ is a homeomorphism.
\end{proof}

By setting $X = X_A$ and $T = \sigma$, we have that $\mathcal{G}(X_A,\sigma)$ is a locally compact Hausdorff second countable \'etale topological groupoid, where $X_A \simeq \mathcal{G}^{(0)}$.

For a continuous function $F:U\to \mathbb{R}$, we think of $\mathbb{R}$ as an additive group, we can define a continuous homomorphism $c_F:\mathcal{G}(X,\sigma)\to \mathbb{R}$ as 
\begin{equation}\label{eq:cocycle}
    c_F(x,n-m,y)=\sum_{j=0}^{n-1}F(\sigma^j(x)) - \sum_{j=0}^{m-1}F(\sigma^j(y)).
\end{equation}

\begin{definition}[Markov partition]\label{def:Markov_partition} Let $X$ be a compact, Hausdorff and totally disconnected topological space and consider two open sets $U$, $V \in X$, and $T:U\to V$ a local homeomorphism from $U$ onto $V$. A Markov partition for $(X,T)$ is a partition of $U$ by a family $\{U_i: i \in I\}$ of non-empty pairwise disjoint compact open sets such that
\begin{itemize}
    \item[$i.$] the restriction $T_i := T\restriction_{U_i}$ is a homeomorphsim from $U_i$ onto a compact open set $V_i = T(U_i) \subseteq V$;
    \item[$ii.$] for all $(i,j) \in I \times I$, either $U_i \subseteq V_j$ or $U_i \cap V_j = \emptyset$;
    \item[$iii.$] the Boolean algebra $\mathcal{B}_0$ generated by $\{X,U_i,V_i,i \in I\}$, seen as a Boolean subalgebra of the power set of $X$, is a generator in the sense that
    \begin{equation*}
        \bigvee_{n=0}^\infty T^{-n}\mathcal{B}_0 
    \end{equation*}
    is the family of all compact open subsets of $X$.
\end{itemize}
\end{definition}

In particular, for our context, $(X_A, \sigma)$ is a SGDS that admits a natural Markov partition, as it is proved and present next.

\begin{proposition} $(X_A,\sigma)$ and $(\widetilde{X}_A,\sigma)$ are SGDS that admit the dense Markov partition $\mathfrak{C}:=\{C_i\}_S$, where we consider $\widetilde{X}_A$ instead of $X_A$ when the last one is not compact. 
\end{proposition}

\begin{proof} Observe that $X_A$ and $\widetilde{X}_A$ are totally disconnected, since they admit a basis of clopen sets. In addition, it is straightforward that $X_A$ is locally compact, that $\Dom(\sigma) = \bigsqcup_{i \in S}C_i$ and $\ran(\sigma) = X_A$ are open, and $\sigma$ is a local homeomorphism, since every $x \in \Dom(\sigma)$ belongs to a $C_i$, for some $i \in S$, and $\sigma\vert_{C_i}:C_i \to \sigma(C_i)$ is a homeomorphism for every $i \in S$. It is straightforward that the sets of the family $\mathfrak{C}$ are non-empty and and pairwise disjoint. By therorem \ref{thm:cylinders_are_compact} we have that $C_i$ is compact for every $i \in S$ when $X_A$ is not compact (the remaining case also holds, by compactness of $X_A$). From now in this proof, we suppose that $X_A$ is compact, since the same proof holds for $\widetilde{X}_A$ when that is not true. We prove the conditions stated in Definition \ref{def:Markov_partition} hold:

\begin{itemize}
    \item[$(i)$] $\sigma$ is a local homeomorphism and hence it is a continuous open map. By continuity of $\sigma$, we have that $\sigma(C_i)$ is also compact, and by $\sigma$ being an open map, we have that $\sigma(C_i)$ is open, for every $i \in S$. Observe that $X_A$ compact implies that $\sigma(C_i)$ is also closed. As we mentioned before, the restriction $\sigma\vert_{C_i}$ is a homeomorphism onto its image;
    \item[$(ii)$] note that $A(i,j) = 1 \iff C_j \subseteq \sigma(C_i)$ and $A(i,j) = 0 \iff C_j \not\subseteq \sigma(C_i) \iff C_j \cap \sigma(C_i) = \emptyset$;
    \item[$(iii)$] let
    \begin{equation*}
        C \in \bigvee_{n=0}^\infty \sigma^{-n}\mathcal{B}_0,
    \end{equation*}
    where $\mathcal{B}_0$ is the Boolean algebra generated by $\{X_A, C_i,\sigma(C_i): i \in S\}$. Then, there exists $k \in \mathbb{N}_0$ s.t. $C \subseteq \Dom(\sigma^k)$ and $C \in \sigma^{-k}\mathcal{B}_0$, that is, $C = \sigma^{-k}(B)$ for some $B \in \mathcal{B}_0$. Since the generators of $\mathcal{B}_0$ are clopen, we have that every element in $\mathcal{B}_0$ is clopen as well. By continuity of $\sigma^k$, we have that $C$ is clopen, and therefore it is an open compact set. Conversely, let $C$ be a compact open set of $X_A$. By compactness,
    \begin{equation*}
        C = \bigcup_{i=1}^n O_i,
    \end{equation*}
    where $O_i$ is a basic set, that is, a finite intersection elements of the subbasis. Hence, it is sufficient to prove that the elements of the sub-basis are compact. Given $\alpha$ an admissible word and $j \in S$, we have that
    \begin{align*}
        C_\alpha &= \bigcap_{k = 0}^{|\alpha|-1} \sigma^{-k}(C_{\alpha_k}), \quad   C_\alpha^c = \bigcup_{k = 0}^{|\alpha|-1} \sigma^{-k}(C_{\alpha_k}^c),\\
        C_{\alpha j^{-1}} &= \sigma^{-|\alpha|}(\sigma(C_j))\cap \bigcap_{k = 0}^{|\alpha|-1} \sigma^{-k}(C_{\alpha_k}),\\
        C_{\alpha j^{-1}}^c &= \sigma^{-|\alpha|}(\sigma(C_j)^c)\cup \bigcup_{k = 0}^{|\alpha|-1} \sigma^{-k}(C_{\alpha_k}^c).
    \end{align*}
    Since
    \begin{equation*}
        \bigvee_{k=0}^\infty \sigma^{-k}\mathcal{B}_0
    \end{equation*}
    is a Boolean algebra, we have that each $O_i$ is clopen and hence $C$ is also clopen, and therefore $C$ is open and compact.
\end{itemize}
We conclude that $\mathfrak{C}$ is a Markov partition. Its density is straightforward.

\end{proof}

Now, define the set $\mathcal{J}_A$ of accumulation points of the sequence $\{\mathfrak{c}_j\}_{j \in \mathbb{N}}$ of columns of the transition matrix $A$.

For every $i \in I$, define the bisections of $\mathcal{G}(X,T)$ by $G_i = \{(x, 1, T (x)), x \in U_i\}$. Hence, it is straightforward that $r(G_i) = U_i$ and $s(G_i) = V_i$. We may identify these bisections with their rescpective characteristic functions $R_i:=\mathbbm{1}_{G_i}$ in the C$^*$-algebra $C^*(\mathcal{G}(X,T))$. Also, observe that
\begin{equation*}
    A(i,j) = \begin{cases}
                1, \quad \text{if } U_j \subseteq V_i,\\
                0, \quad \text{if } U_j \subseteq V_i^c.
             \end{cases}
\end{equation*}

In our context we have that $\{R_i\}_{i \in I}$ is a family of partial isometries that satisfies (EL1)-(EL4). In addition, every totally disconnected compact SGDS endowed with a dense and countable Markov partition can be codified by a matrix like it is above. In particular, if that matrix generated is such that every circuit has an exit, the Uniqueness Theorem \ref{thm:uniqueness_thm_EL_algebras} holds for similar partial isometries as constructed for the generalized Renault-Deaconu relative to this SGDS. Furthermore, this stabilishes an isomorphism between its groupoid C$^*$-algebra and the Exel-Laca algebras as it is in Proposition 4.8 of \cite{Renault1999} as it is presented below.

\begin{proposition}\label{prop:isomorphism_between_O_A_and_RDA} Given a SGDS $(X,T)$ endowed with a dense Markov partition that is generated by a transition matrix $A$. Suppose that its symbolic graph satisfies that every circuit has an exit. Then we have that
\begin{itemize}
    \item[$(a)$] if $0 \notin \mathcal{J}_\mathcal{A}$ then $\mathcal{O}_A \simeq C^*(\mathcal{G}(X,T))$;
    \item[$(b)$] if $0 \in \mathcal{J}_\mathcal{A}$ then $\mathcal{O}_A \simeq C^*(\mathcal{G}(X\setminus \{\varphi\},T))$.
\end{itemize}
In the statement above, $\varphi$ is the (unique) element of the complement of $\Dom T \cup \ran T$, when such element exists.
\end{proposition}

\begin{remark} $\varphi$ as above do exists if and only if $0 \in \mathcal{J}_A$ (see Theorem \ref{thm:O_A_unital}).
\end{remark}

For the generalized Markov shift space, we have that $\mathcal{O}_A \simeq C^*(\mathcal{G}(X_A,\sigma))$, since in the proposition above, $X = X_A$ if $0 \notin \mathcal{J}_A$ and $X = \widetilde{X}_A$ otherwise.

\chapter{Thermodynamic Formalism on Generalized Countable Markov Shifts}
\label{ch:TF_on_Generalized_Countable_Markov_shifts}
\label{ch:term_form_X_A}

In this chapter, we introduce the Thermodynamic Formalism for the generalized Markov shift space. Among the results, we emphasize the equivalences between the notions of conformality in this generalized context, the compatibility between the notion of conformal measure in the classical and the generalized settings, and the existence of new conformal measures which are not detected in the classical theory. In particular, these new measures actually let evident a new type of phase transition, which we present in concrete examples, as the renewal shift, the pair renewal shift and the prime renewal shift. Later in the chapter, we adapt the Denker-Yuri results for iterated function systems on the generalized setting.

\section{Weak$^*$ convergence of measures on $X_A$}

When we study the thermodynamic formalism on a dynamical system, we are interested in potentials which carry a factor $\beta>0$, corresponding to the inverse of the temperature, and the behavior of special measures (conformal measures, equilibrium measures, etc) is related to $\beta$, which generate nets of measures indexed by $\beta$. The control of the weak$^*$ convergence of probability measures over the configuration space (symbolic space) is one of the first facts which we want to know, so we present next the notion of convegence for net of measures in metric spaces. Given a metric space $X$, we denote by $C_b(X)$ the set of the bounded continous real functions on $X$.

\begin{definition}[weak$^*$ convergence of measures] Let $X$ be a metric space and denote by $M_{\text{fin}}(X)$ the set of finite Borel measures on $X$. Let $\{\mu_\beta\}_\beta$ be a net in $M_{\text{fin}}(X)$. Given a measure $\mu \in M_{\text{fin}}(X)$, we say that $\{\mu_\beta\}_\beta$ converges weakly$^*$ to $\mu$ when, for every $f \in C_b(X)$, we have
\begin{equation*}
    \lim_\beta \int_X f d\mu_\beta = \int_X f d\mu.
\end{equation*}
When this happens, we use the notation $\mu_\beta \rightharpoonup \mu$. 
\end{definition}

In particular, we restrict the definition above to the set of the Borel probabilities on $X$. 

\begin{remark} The nets we are interested here are those which the directed set associated to them is a semi-finite interval $[a, \infty)$ with the reverse order. In this case, the net convergence is equivalent to state that every sequence $\{\beta_k\}_{k \in \mathbb{N}}$ in $(a, \infty)$ converging to $a$ satisfies $\mu_{\beta_k} \rightharpoonup \mu$. 
\end{remark}

In the standard case of countable Markov shifts, it is known that that a sequence of probability measures $(\mu_n)_{n \in \mathbb{N}}$ converges in the weak$^*$ topology (weak topology for the probabilists) to a probability measure $\mu$ when, for each cylinder set $[x_0, x_1,..., x_{m-1}]$, we have $\lim_{n \to \infty} \mu_{n}([x_0, x_1,..., x_{m-1}]) = \mu([x_0, x_1,..., x_{m-1}])$. For a recent reference on this topic, we mention \cite{IommiVelozo2019}.

When we use the approach on taking limits of measures on generalized cylinders for the study of the weak$^*$ convergence, we find some differences and similarities with respect to the usual symbolic space. In fact, unlike in the standard Markov shift case, the intersections of generalized cylinders (associated to any $g \in \mathbb{F}$) does not necessarily gives another cylinder. Moreover, the complement of a cylinder is not a union of cylinders as it happens to the standard case. In general, both intersections and complements give a countable union of disjoint cylinders jointly with a subset of $Y_A$ which does not intersects the cylinder union. In other words, the cylinder subbasis it is not a basis, and a typical element of the basis generated by the subbasis has the form
\begin{equation}
    F \sqcup \bigsqcup_{n \in \mathbb{N}} C_{w(n)},
\end{equation}
where $w(n)$ is a positive admissible word for each $ n \in \mathbb{N}$, and $F \subseteq Y_A$ (see Theorem \ref{thm:huge_generators_intersections}). However, in many cases, as it happens to the renewal shift and the pair renewal shift, the extra disjoint set of finite words $F$ as above is finite, and, when we take the limit to the critical value where conformal measures living on $Y_A$ disappear, the measure on $F$ vanishes too. So $F$ does not give any contribution to the measure at the limit.

In this thesis, we describe conformal measures which are not detected by the standard theory and we also take limits on nets of these measures. This brief section presents a way how to take these limits and conditions of existence of a limit measure on the weak$^*$ topology. The main result we present below is a less general version of the Theorem 8.2.17 of \cite{Bogachev2007} as follows.

\begin{theorem}\label{thm:convergence_measures_Bogachev} Let $\{\mu_\beta\}$ be a net of Borel probability measures on a separable metric space $X$ and let $\mu$ a probability measure on $X$. Suppose that the equality
\begin{equation*}
    \lim_\beta \mu_\beta(U) = \mu(U)
\end{equation*}
is fulfilled for all elements $U$ of some base $\mathscr{B}$ of the topology of $X$ that is closed with respect to finite intersections. Then the net of measures $\{\mu_\beta\}$ converges to $\mu$ on the weak$^*$ topology.
\end{theorem}

We recall that $X_A$ is a metric subspace of $\{0,1\}^\mathbb{F}$ and its topology is generated by the generalized cylinder sets and their complements, which forms a subbasis. Consequently, the basis of the finite intersections of generalized cylinders is a topological basis for $X_A$, and it is clearly closed under intersections. Moreover, such basis is countable, that is, $X_A$ is second-countable, and therefore it is separable. We conclude that the generalized Markov shifts fulfills the topological requirements for Theorem \ref{thm:convergence_measures_Bogachev}. Observe that, by the theorem above, since for the standard Markov shift spaces the subbasis coincides with the basis and it is closed under finite intersections, it is sufficient to evaluate the limits on the standard cylinder sets.

\section{Conformal Measures on Generalized Markov Shifts}

In this section we present the generalized notions of Ruelle's operator, conformal measures and eigenmeasures when the dynamics is partially defined. In addition, we also present the notion o quasi-invariant measure, which is related to the groupoid structure from the Renault-Deaconu's theory. We prove that there is a similar result to corollary \ref{cor:equivalences_conformality_classical} about the equivalences between the aforementioned measures. Furthermore, we compare and study the connections among these generalized measures to the classical ones for the generalized space $X_A$ in terms of restriction of the generalized measures to $\Sigma_A$ and extension of the classical ones to $X_A$. Properties such as non-singularity and conservativity of the aforementioned measures are discussed in this new setting as well. The notions presented next could be considered in a much more general context, but we specify them to $X_A$ in further sections, when we will obtain new phase transition results.

The first definitions and results in this section are general: we fix a locally compact, Hausdorff and second countable topological space $X$ endowed with a local homeomorphism  $\sigma:U \to X$, where $U$ is an open subset of $X$. In addition, consider its respective Renault-Deaconu groupoid $\mathcal{G}(X,\sigma)$, and a continuous potential $F:U \to \mathbb{R}$. Since we are interested in phase transition phenomena, we also consider the inverse of the temperature $\beta >0$, which will be a factor multiplying $F$, and it can be absorved in the potential if one whishes to state the definitions here in the same way as was done in chapter \ref{ch:Markov_shift_space}. After Remark \ref{remark:KMS_quasi_invariant}, we restrict the study for the set $X_A$.

We define now the Ruelle's transformation, which is the generalized version of the Ruelle's operator.

\begin{definition}\label{def:Ruelle_transformation} The Ruelle's transformation is the linear transformation
\begin{align}
    L_{\beta F}:C_c(U) &\to C_c(X) \nonumber\\
    f  & \mapsto L_{\beta F}(f)(x):= \sum_{\sigma(y)=x}e^{\beta F(y)}f(y).\label{eq:general_Ruelle}
\end{align}
\end{definition}

\begin{proposition} $L_{\beta F}$ as in Definition \ref{def:Ruelle_transformation} is a well-defined linear transformation. 
\end{proposition}
    
\begin{proof} W.l.o.g. assume $\beta = 1$. Let $f \in C_c(U)$ and observe that, for every $x \in X$, the set
\begin{equation*}
    K_x = \supp(f) \cap \sigma^{-1}(x),
\end{equation*}
where $\sigma^{-1}(x) := \sigma^{-1}(\{x\})$, is a compact set, since it is a closed subset of a compact set. For every $x \in X$ and $y \in K_x$ we have the following:
\begin{itemize}
    \item by local compactness of $X$, there exists an open neighborhood $V_y^1(x)$ of $y$, s.t. $\overline{V_y^1(x)}$ is compact;
    \item since $\sigma$ is a local homeomorphism, there exists an open neighborhood $V_y^2(x)$ of $y$, s.t. $\sigma\vert_{V_y^2(x)}$ is a homeomorphism onto its image.
\end{itemize}
Define $V_y(x):= V_y^1(x)\cap V_y^2(x)$. It is straightforward that $\overline{V_y(x)}$ is compact and that $\sigma\vert_{V_y^2(x)}$ is a homeomorphism onto its image. Also,
\begin{equation*}
    \supp(f) \subseteq \bigcup_{x \in X} \bigcup_{y \in K_x} V_y(x).
\end{equation*}
By compactness\footnote{Observe that $\supp(f)$ is a compact subset of $U$ and then it is also a compact subset $X$.} of $\supp(f)$, there exists $n \in \mathbb{N}$ and a family
\begin{equation}\label{eq:family_V_k_L_F_well_defined}
    \{V_k\}_{k = 1}^n \subseteq \{V_y(x): x\in X, y \in K_x\}
\end{equation}
satisfying
\begin{equation}\label{eq:suppf_contained_V_k_union}
    \supp(f) \subseteq \bigcup_{k=1}^n V_k.
\end{equation}
In particular,
\begin{equation*}
    K_x \subseteq \bigcup_{k=1}^n V_k,
\end{equation*} 
then $\supp(f) \cap \sigma^{-1}(x)$ is a finite set for every $x$, and hence the sum
\begin{equation*}
    L_F(f) (x) = \sum_{\sigma(y)=x}e^{\beta F(y)}f(y)
\end{equation*}
is a finite sum for every $x \in X$, and therefore $L_F(f)$ is a well-defined function on $X$. Now we prove that $L_F(f) \in C_c(X)$. Define the set
\begin{equation*}
    H:= \{x \in X : \sigma^{-1}(x)\cap \supp f \neq \emptyset\},
\end{equation*}
and observe that, for every $z \in H^c$, we have $\sigma^{-1}(z) \cap \supp f = \emptyset$, and then
\begin{equation*}
    L_F(f)(z) = \sum_{\sigma(y) = z} e^{\beta F(y)}f(y) = 0,
\end{equation*}
hence $H^c \subseteq \{x\in X: L_F(f)(x) = 0\}$, that is,
\begin{equation}\label{eq:support_interior_contained_in_H}
    \{x\in X: L_F(f)(x)\neq 0\}\subseteq H.
\end{equation}
We claim that
\begin{equation*}
    H \subseteq \bigcup_{k=1}^n \sigma(\overline{V_k}),
\end{equation*}
where the family $\{V_k\}_{k = 1}^n$ is the same as in \eqref{eq:family_V_k_L_F_well_defined}. In fact, for every $x \in H$, we have
\begin{equation*}
    \sigma^{-1}(x) \cap \supp(f) \subseteq \supp(f) \subseteq \bigcup_{k=1}^n V_k \subseteq \bigcup_{k=1}^n \overline{V_k}.
\end{equation*}
Then\footnote{For every function $f:A \to B$ and, it holds that $f(f^{-1}(C)) = C$, if $C$ is a sigleton contained in the image of $f$.},
\begin{equation*}
    \{x\} = \sigma\left(\sigma^{-1}(x) \cap \supp(f)\right) \subseteq \sigma\left(\bigcup_{k=1}^n \overline{V_k}\right) = \bigcup_{k=1}^n \sigma(\overline{V_k}),
\end{equation*}
and the claim is proved. Since each $\overline{V_k}$ is compact, we have that $\sigma(\overline{V_k})$ is compact, and so $\bigcup_{k=1}^n \sigma(\overline{V_k})$ is compact. By \eqref{eq:support_interior_contained_in_H} we have $\supp(L_F(f)) \subseteq \overline{H}$ and therefore $L_F(f) \in C_c(X)$. It remains to prove that $L_F(f)$ is continuous. By \eqref{eq:suppf_contained_V_k_union}, for every $x\in X$ we have that
\begin{equation*}
    |\sigma^{-1}(x)\cap supp(f)|= N_x\leq n.
\end{equation*}
W.l.o.g. assume $N_x > 0$ and write $\sigma^{-1}(x)\cap supp(f) = \{y_1,...,y_{N_x}\}$, with $y_k \in V_k$. For each $y_i$ there exists an open neighborhood $\Omega_i \subseteq V_i$ s.t. $\Omega_i \cap \Omega_j = \emptyset$ whenever $i \neq j$. It is straightforward that the restriction $\sigma\vert_{\Omega_i}$ is a homeomorphism onto its image. We may assume that $\sigma(\Omega_i)=\sigma(\Omega_j)$ for every $i,j\in\{1,...,N_x\}$. Now, consider a sequence $\{x_m\}_{m \in \mathbb{N}}$ converging to $x$. There exists $M \in \mathbb{N}$ s.t. $m \geq M$ implies $x_m \in \sigma(\Omega_i)$, and then we may take $\sigma\vert_{\Omega_i}^{-1}(x_m)$ for each $i$. Note that, for $i \neq j$,  $\sigma\vert_{\Omega_i}^{-1}(x_m) \neq \sigma\vert_{\Omega_j}^{-1}(x_m)$. We claim that
\begin{equation*}
    \sigma^{-1}(x_m)\cap supp(f)\subseteq \{\sigma\vert_{\Omega_i}^{-1}(x_m): i=1,...,N_x\}
\end{equation*}
for $m$ sufficiently large. Indeed, suppose that for infinitely many values of $m$ we have
\begin{equation}\label{eq:contradiction_hypothesis_continuity_L_F_f}
    \sigma^{-1}(x_m)\cap supp(f)\not\subseteq \{\sigma\vert_{\Omega_i}^{-1}(x_m): i=1,...,N_x\},
\end{equation}
so we may consider the subsequence $\{x_{m_k}\}_{k \in \mathbb{N}}$ be of the elements that satisfy \eqref{eq:contradiction_hypothesis_continuity_L_F_f} and $\{z_{m_k}\}\subseteq supp(f)$ those elements that $\sigma(z_{m_k})=x_{m_k}$ s.t. $z_{m_k} \notin \bigcup_{p=1}^{N_x}V_p$ (otherwise, they would be in the form $\sigma\vert_{\Omega_i}^{-1}(x_{n_k})$). Since
\begin{equation*}
    supp(f)\cap (\cup_{p=1}^{N_x} V_p)^c
\end{equation*}
is compact, $\{z_{m_k}\}$ has a convergent subsequence which we also will denote by $\{z_{m_k}\}$. Let $y$ be its limit. By continuity of $\sigma$ we have
\begin{equation*}
    x_{m_k} = \sigma(z_{m_k}) \to \sigma(y).
\end{equation*}
Since $x_m \to x$, we also obtain $x = \sigma(y)$. Since $y \notin V_p$ for every $p \in \{1,...,N_x\}$ we have that
\begin{equation*}
    |\sigma^{-1}(x)\cap supp(f)|>N_x, 
\end{equation*}
a contradiction. The claim is proved. Then,
\begin{align*}
\lim_m L_F(f)(x_m)&= \lim_m \sum_{\sigma(y)=x_m}e^{F(y)}f(y) = \lim_m \sum_{i=1}^{N_x}e^{F(\sigma\vert_{\Omega_i}^{-1}(x_m))}f(\sigma\vert_{\Omega_i}^{-1}(x_m))\\
&= \sum_{i=1}^{N_x}e^{F(\sigma\vert_{\Omega_i}^{-1}(x))}f(\sigma\vert_{\Omega_i}^{-1}(x)) 
=\sum_{\sigma(y)=x}e^{F(y)}f(y)=L_F(f)(x).
\end{align*}
\end{proof}

\begin{definition}[Eigenmeasure associated to the Ruelle Transformation] Consider the Borel $\sigma$-algebra $\mathcal{B}_X$. A measure $\mu$ on $\mathcal{B}_X$ is said to be an \textit{eigenmeasure} with eigenvalue $\lambda$ for the Ruelle transformation $L_{\beta F}$ when
\begin{equation}\label{eq:conformal_eigenmeasure_functions_X_A}
    \int_{X} L_{\beta F}(f)(x)d\mu(x) = \lambda \int_{U} f(x)d\mu(x),
\end{equation}
for all $f \in C_c(U)$.
\end{definition}
In other words, the equation \eqref{eq:conformal_eigenmeasure_functions_X_A} can be rewritten by using \eqref{eq:general_Ruelle} as
\begin{equation}\label{eq:conformal_eigenmeasure_functions_X_A_2}
    \int_{X} \sum_{\sigma(y) = x} e^{\beta F(y)}f(y) d\mu(x) = \lambda \int_{U} f(x)d\mu(x),
\end{equation}
for all $f \in C_c(U)$. As in the standard theory of countable Markov shifts, when a measure $m$ satisfies the equation \eqref{eq:conformal_eigenmeasure_functions_X_A_2} we write $L_{\beta F}^*\mu = \lambda\mu$.

Now we introduce the notions of conformal measure in the senses of Denker-Urba\'nski and Sarig in the generalized setting.

\begin{definition}[Conformal measure - Denker-Urba\'nski] Let $(X,\mathcal{F})$ be a measurable space and $D:U \to [0,\infty)$ also measurable. A measure $\mu$ in $X$ is said to be $D$-conformal in the sense of Denker-Urba\'nski if
\begin{equation}\label{eq:conformal_Ur_sets_potential}
    \mu(\sigma(B)) = \int_B D d\mu,
\end{equation}
for every special set $B \subseteq U$.
\end{definition}

Although the previous definition is very general, we will always consider $\mathcal{F} = \mathcal{B}_X$.

\begin{definition}
Given a Borel measure $\mu$ on $X$ we define the measure $\mu\odot\sigma$ on $U$ by
$$\mu\odot\sigma(E):=\sum_{i\in \mathbb{N}}\mu(\sigma(E_i)).$$
For every measurable set $E\subseteq U$, where the $\{E_i\}_{i \in \mathbb{N}}$ is a family of pairwise disjoint measurable sets such that $\sigma\vert E_i$ is injective for every $i$, and $E =\sqcup_i E_i$. 
\end{definition}

\begin{remark}
We show that $\mu\odot\sigma$ is well defined. First we prove the existence of at least one countable family $\{E_i\}$, as above. Indeed, if $E\subseteq U$, since $\sigma$ is a local homeomorphism for each $x\in$ E there is an open subset $H_x\ni x$ such that $\sigma$ is injective, we have $E\subseteq \cup_{x\in E} H_x$. For each of those $H_x$ there is an open basic set $U_x$ such that $x\in U_x$, but the topology basis is countable, so we can enumerate $\{U_x\}=\{U_1,U_2,...\}$ and we observe that $\sigma$ is injective on each $U_i$. Take $E_1:=E\cap U_1$, $E_n:=E\cap U_n\setminus\bigsqcup_{i=1}^{n-1}E_i$ and we have what we claimed.

Now we shall see that the definition does not depend on the decomposition of $E$. Let $E=\bigsqcup E_i=\bigsqcup F_j$, then $E=\bigsqcup_{i,j}E_i\cap F_j$. Therefore,
$$\sum_i\mu(\sigma(E_i))=\sum_i\mu(\sigma(\sqcup_j E_i\cap F_j))=\sum_i\mu(\sqcup_j\sigma( E_i\cap F_j))=\sum_{i,j}\mu(\sigma(E_i\cap F_j))$$
Doing analogously for $\{F_j\}$ instead of $\{E_i\}$ we conclude that
$$\sum_i\mu(\sigma(E_i))=\sum_j\mu(\sigma(F_j)) $$
We therefore have that the measure $\mu\odot\sigma$ is well defined.
\end{remark}

\begin{definition}[Conformal measure - Sarig]
A Borel measure $\mu$ in $X$ is called $(\beta F, \lambda)$-conformal in the sense of Sarig if there exists $\lambda >0$ such that
\begin{equation*}
\dfrac{d\mu\odot\sigma}{d\mu}(x)= \lambda e^{-\beta F(x)}\quad \mu\; a.e \; x\in U.
\end{equation*}
\end{definition}

\begin{remark} Both notions of conformal measures in the senses of Denker-Urba\'nski and Sarig do not require that the potential be continuous. 
\end{remark}

Now we define the $1$-cocycle associated to the potential for the generalized Renault-Deaconu groupoid.

\begin{definition} Consider the generalized Renault-Deaconu groupoid $\mathcal{G}(X,\sigma)$ and a continuous potential $F:U \to \mathbb{R}$, we define the continous $1$-cocycle $c_F:\mathcal{G}(X,\sigma)\to \mathbb{R}$, associated to the potential $F$, as follows. Let $(x,n-m,y) \in W(n,m,V_1,V_2)$, set
\begin{equation*}
    c_F(x,n-m,y) = \begin{cases}
                        \sum_{i=0}^{n-1}F(\sigma^i(x))-\sum_{i=0}^{m-1}F(\sigma^i(y)), \quad \text{if } x,y \in U;\\
                        \sum_{i=0}^{n-1}F(\sigma^i(x)), \quad \text{if } x \in U \text{ and } y \in U^c \text{ }(m=0);\\
                        -\sum_{i=0}^{m-1}F(\sigma^i(y)), \quad \text{if } x \in U^c \text{ and } y \in U\text{ }(n=0);\\
                        0, \quad \text{otherwise}.
                   \end{cases}
\end{equation*}
\end{definition}


The next theorem is the generalized version of the Corollary \ref{cor:equivalences_conformality_classical}.  

\begin{theorem}\label{thm:equivalences_conformal_measures_generalized_Markov_shift} Let $X$ be locally compact, Hausdorff and second countable space, $U \subseteq X$ open and $\sigma: U \to X$ a local homeomorphism. Let $\mu$ be a Borel measure that is finite on compact sets. For a given continuous potential $F:U\to \mathbb{R}$, the following are equivalent.

\begin{itemize}
    \item[$(i)$] $\mu$ is $e^{\beta F}$-conformal measure in the sense of Denker-Urba\'nski;
    \item[$(ii)$] $\mu$ is a eigenmeasure associated with the Ruelle Transformation $L_{-\beta F}$, that is
    \begin{equation*}
        \int_{X} \sum_{\sigma(y)=x} f(y) e^{-\beta F(y)} d\mu(x) = \int_{U} f(x) d\mu(x),
    \end{equation*}
    for all $f \in C_c(U)$;
    \item[$(iii)$] $\mu$ is $e^{-\beta c_F}$-quasi-invariant on $\mathcal{G}(X,\sigma)$, i.e
    \begin{equation}\label{eq:quasi-invariant-measure}
        \int_{X} \sum_{r(\gamma)=x} e^{\beta c_F(\gamma)} f(\gamma)  d\mu(x) = \int_{X}\sum_{s(\gamma)=x}f(\gamma)  d\mu(x).
    \end{equation}
    for all $f \in C_c(\mathcal{G}(X,\sigma))$;
    \item[$(iv)$] $\mu$ is $(-\beta F, 1)$-conformal in the sense of Sarig.
\end{itemize}
\end{theorem}

\begin{proof} 

$(iii)\implies (ii)$ is analogous to Proposition 4.2 in \cite{Renault2003}, but we repeat the proof. For any $f\in C_c(U)$, consider:
$$\int_X \sum_{\sigma(y)=x}f(y)e^{-\beta F(y)}d\mu(x)=\int_X\sum_{s(\gamma)=x}(fe^{-\beta F}\circ r)(\gamma)\mathbbm{1}_{S}(\gamma)d\mu(x). $$
where $S$ is the set $\{(x,1,\sigma(x)): x\in U\}$. Now, we use that $\mu$ is quasi-invariant to conclude that
$$\int_X\sum_{s(\gamma)=x}(fe^{-\beta F}\circ r)(\gamma)\mathbbm{1}_{S}(\gamma)d\mu(x)=\int_X\sum_{r(\gamma)=x}(fe^{-\beta F})\circ r(\gamma)\mathbbm{1}_{S}(\gamma)e^{\beta c_F(\gamma)}d\mu(x)=\int_U f d\mu. $$
Proving the implication we were interested.
 
For $(ii) \implies (i)$ let $V$ be an open subset of $U$ such that $\sigma\vert_V$ is injective, and let $W = \sigma(V)$. Also denote by $\tau : W \to V$ the inverse of the restriction of $\sigma$ to $V$. We then have two measures of interest on $V$, namely
\begin{equation*}
\tau^*(\mu\vert_W)\quad and \quad e^{\beta F}\mu\vert_V.    
\end{equation*}
We claim that the above measures on $V$ are equal. By the uniqueness part of the Riesz-Markov Theorem, it is enough to prove that
\begin{equation}\label{eq:urbanski_integrated}
    \int_V gd\tau^*(\mu\vert_W)=\int_V ge^{\beta F}d\mu\vert_V,
\end{equation}
for every $g$ in $C_c(V)$. Given such a $g$, we consider its extension to the whole of $U$ by setting it to be zero on $U \setminus V$ . The extended function is then in $C_c(U)$. Defining $f = ge^{\beta F}$ , we then have that

\begin{align*}
    \int_V g e^{\beta F} d\mu\vert_V&=\int_U f d\mu \stackrel{(\ref{thm:equivalences_conformal_measures_generalized_Markov_shift}.ii)}{=}\int_X \sum_{\sigma(y)=x}f(y)e^{-\beta F(y)}d\mu(y)\\
    &=\int_W f(\tau(x))e^{-\beta F(\tau(x))}d\mu(x)=\int_W g(\tau(x))d\mu(x)=\int_V g d\tau^*(\mu).
\end{align*}
This proves \eqref{eq:urbanski_integrated}, and hence also that $\tau^*(\mu\vert_W)=e^{\beta F}\mu\vert_V$. It follows that, for every measurable set $E\subseteq V$,
\begin{equation*}
    \mu(\sigma(E))=\mu(\tau^{-1}(E))=\tau^*(\mu\vert_W)(E)=\int_E e^{\beta F}d\mu.
\end{equation*}
Now, suppose $E\subseteq U$ is a special set, since $\sigma$ is a local homeomorphism and $X$ is second countable, there exists a countable collection of open sets $\{V_i\}_{i\in \mathbb{N}}$ such that $\sigma\vert_{V_i}$ is injective and $E\subseteq \bigcup_{i\in \mathbb{N}}V_i$. Then we have a countable collection of measurable sets $\{E_i\}_{i\in\mathbb{N}}$, pairwise disjoint, such that $E_i\subseteq V_i$ and $E=\sqcup_{i\in \mathbb{N}}E_i$. We conclude, using that $E$ is special, that
\begin{equation*}
    \mu(\sigma(E))=\sum_{i\in \mathbb{N}}\mu(\sigma(E_i))=\sum_{i\in \mathbb{N}}\int_{E_i}e^{\beta F} d\mu=\int_E e^{\beta F}.
\end{equation*}

$(i)\implies (iii)$. We consider the open bisections defined in the preliminaries $W(n,m,C,B)$. W.l.o.g we can consider $\sigma^n(C)=\sigma^m(B)$, since if not we could take open sets $C'\subseteq C$ and $B'\subseteq B$ such that $\sigma^n(C')=\sigma^n(C)\cap \sigma^m(B)=\sigma^m(B')$. Also, we can suppose that $\sigma^n$ is injective when restricted $C$, similarly for $\sigma^m$ and $B$. In this setting, we can define the map $\sigma^{n-m}_{CB}:=\sigma^{-m}_B\circ\sigma^n_C$ and similarly $\sigma^{m-n}_{BC}:=\sigma^{-n}_C\circ \sigma^{m}_B$.
\[
\begin{tikzcd}
 & \sigma^n(C)=\sigma^m(B)  \\
C \arrow{ur}{\sigma^n_C} \arrow[rr,dashed,"\sigma^{n-m}_{CB}"] && B \arrow[ul,"\sigma^m_B"'] \arrow[ll,dashed,"\sigma^{m-n}_{BC}", bend left]
\end{tikzcd}
\]
Let $f \in C_c(\mathcal{G}(X,\sigma)$ s.t. $supp(f)\subseteq W(n,m,C,B)$. Let us see how the equation \eqref{eq:quasi-invariant-measure} on item $(iii)$ simplifies in this case. Observe first the left hand side. If $x \notin C$, clearly there is no $\gamma\in W(n,m,C,B)$ such that $r(\gamma)=x$, so the integration can be done in $C$. Now, for $x\in C$, consider $\gamma_1,\gamma_2\in W(n,m,C,B)$ such that $r(\gamma_1)=r(\gamma_2)=x$. Since the range map is injective in such set we have $\gamma_1=\gamma_2$ and we conclude the summation on the left hand side of equation \eqref{eq:quasi-invariant-measure} have at most one non-zero term for each $x\in C$. Denoting this term by $\gamma_x$, we see this term is written as $\gamma_x=(x,n-m,\sigma_{CB}^{n-m}(x))$. So, the left hand side of equation \eqref{eq:quasi-invariant-measure} is
\begin{equation}\label{eq:theorem3,LHS}
\int_C e^{\beta c_F(\gamma_x)} f(\gamma_x)  d\mu(x)=\int_C e^{\beta c_F((x,n-m,\sigma_{CB}^{n-m}(x)))} f(x,n-m,\sigma_{CB}^{n-m}(x))  d\mu(x).    
\end{equation}
Calculation on the right hand side of equation \eqref{eq:quasi-invariant-measure} is done in a similar fashion, we have:
\begin{equation}\label{eq:theorem3,RHS}
    \int_B f(\sigma_{BC}^{m-n}(y),n-m,y)d\mu(y).
\end{equation}
Now let $g:C\to \mathbb{C}$ defined by $g(x)=f(x,n-m,\sigma_{CB}^{n-m}(x)).$ Observe that $g(\sigma^{m-n}_{BC}(y))=f(\sigma^{m-n}_{BC}(y),n-m,y)$, which is the function in the equation \eqref{eq:theorem3,RHS}. We rewrite the quasi-invariant condition with the considerations from above
\begin{equation}\label{eq:QIC}
    \int_C e^{\beta c_F(x,n-m,\sigma^{n-m}(x))}g(x)d\mu(x)=\int_B g(\sigma^{m-n}_{BC}(y))d\mu(y),
\end{equation}
for $g\in C_c(C)$. We just need to prove that item $(i)$ implies equation \eqref{eq:QIC}. First, observe that item $(i)$ implies for all $C$ open subset of $U$, $\sigma\vert_C$ injective that
\begin{equation}\label{eq:urbanski_the_not_so_good}
\int_C g(x) e^{\beta F(x)} d\mu(x) = \int_{\sigma(C)} g(\sigma^{-1}(x)) d\mu(x),    
\end{equation}
for all $g\in C_c(C)$. We observe as well that equation \eqref{eq:urbanski_the_not_so_good} is equation \eqref{eq:QIC} when $m=0$ and $n=1$. To prove \eqref{eq:QIC} we proceed by induction on $n+m$. If $n+m=0$ we have $C=B$ , $\sigma_{BC}^{n-m}=Id$ and $c_F(x,0,x)=0$, so equation (\ref{eq:QIC}) is clearly satisfied. Take $n\neq 0$.

\[
\begin{tikzcd}[row sep=5em]
 & \mathbb{C} &  \\[-2em]
C\arrow[ur,"g"]\arrow[rr,"\sigma"]\arrow[d,"\sigma^{n-m}_{CB}"']& & \sigma(C)=C'\arrow[ul,"g'"']\arrow[d,"\sigma^{n-1}"]\\
B\arrow[urr,"\sigma^{m-(n-1)}_{BC'}",dashed]\arrow[rr,"\sigma^m_B"']& & \sigma^n(C)=\sigma^m(B)
\end{tikzcd}
\]
 \noindent
Let $g':C'=\sigma(C)\to  \mathbb{C}$ defined by $g'(x)=g(\sigma^{-1}(x))$, $g'\in C_c(C')$. By induction hypothesis,
$$\int_B g'(\sigma_{BC'}^{m-(n-1)}(y))d\mu(y)=\int_{C'}e^{\beta c_F(x,n-1-m,\sigma^{(n-1)-m}(x))}g'(x)d\mu(x).$$
On the other hand, using as reference the figure above
$$\int_B g'(\sigma_{BC'}^{m-(n-1)}(y))d\mu(y)=\int_B g(\sigma^{-1}\sigma_{BC'}^{m-(n-1)}(y))d\mu(y)=\int_B g(\sigma_{BC}^{m-n}(y))d\mu(y)$$
Which is the right hand side of equation \eqref{eq:QIC}. Then,
\begin{equation}\label{eq:g_2}
\int_{C'}e^{\beta c_F(x,n-1-m,\sigma^{(n-1)-m}(x))}g'(x)d\mu=\int_{\sigma(C)}\underbrace{e^{\beta c_F(x,n-1-m,\sigma^{(n-1)-m}(x))}g(\sigma^{-1}(x))}_{g_2(x)}d\mu.    
\end{equation}
The equation \eqref{eq:urbanski_the_not_so_good}, using a change of variables, can be seen as well as
$$\int_C g_2(\sigma(x))e^{\beta F(x)} d\mu(x)=\int_{\sigma(C)}g_2(x)d\mu\quad \forall g_2\in C_c(\sigma(C)). $$
Applying it to \eqref{eq:g_2}, we obtain
\begin{align*}
\int_{C'}e^{\beta c_F(x,n-1-m,\sigma^{(n-1)-m}(x))}g'(x)d\mu(x)&=\int_C g_2(\sigma(x))e^{\beta F(x)} d\mu(x)\\
&= \int_C e^{\beta c_F(\sigma(x),n-1-m,\sigma^{(n-1)-m}(\sigma(x)))}g(x)e^{\beta F(x)}d\mu(x).
\end{align*}
It is left to verify that
$$c_F(\sigma(x),n-1-m,\sigma^{n-1-m}(\sigma(x)))+F(x)=c_F(x,n-m,\sigma^{n-m}(x)),$$
because that is the left hand side of equation \eqref{eq:QIC}.
It is true by the cocycle property of $c_F$ and the fact that $F(x)=c_F(x,1,\sigma(x))$ along with the observation that $$(x,1,\sigma(x))(\sigma(x),n-1-m,\sigma^{n-m}(x))=(x,n-m,\sigma(x)).$$
The implication $(i)\implies (iii)$ is proved for $f$ supported on the open bisection $W(n,m,C,B)$, therefore proved for every $f\in C_c(\mathcal{G}(X,\sigma)).$
 \noindent
Now, we prove $(i) \iff (iv)$. Suppose that $\dfrac{d\mu\odot\sigma}{d\mu}=e^{\beta F}\quad \mu\; a.e\; x\in U$. Take $E\subseteq U$ such that $\sigma|_E$ is injective. Then,
$$\mu(\sigma(E))=\mu\odot\sigma(E)=\int_U \mathbbm{1}_{E}\,d\mu\odot\sigma=\int_U \mathbbm{1}_E e^{\beta F(x)}d\mu(x)=\int_E e^{\beta F(x)} d\mu(x)$$
and we have proved item $(i)$.
 \noindent
Now the converse. Let $E\subseteq$ U and $\{E_i\}_{i\in \mathbb{N}}$ be its decomposition. Hence,
\begin{align*}
\mu\odot\sigma(E)= \sum_{i\in \mathbb{N}}\mu(\sigma(E_i))=\sum_{i\in \mathbb{N}}  \int_X \mathbbm{1}_{E_i} e^{\beta F(x)}\mu(x)=\int_E e^{\beta F(x)} d\mu(x). \nonumber
\end{align*}
Since this is true for every measurable set $E\subseteq U$, we have
$$\dfrac{d\mu\odot\sigma}{d\mu}(x)=e^{\beta F(x)}\quad \mu\; a.e \; x\in U.$$
This concludes the theorem.
\end{proof}

\begin{remark} Similarly to the standard Markov shift case, the theorem above also is generalized to general eigenmeasures, for any eigenvalue, through absortion of the eigenvalue to the potential like it is shown in Remark \ref{remark:equivalences_conformality_classical_general_eigenvalue}. 
\end{remark}

\begin{remark}\label{remark:KMS_quasi_invariant} The above theorem is of particular interest, since it is known that if a measure $\mu$ is $e^{-\beta c_F}$-quasi-invariant on $\mathcal{G}(X,\sigma)$, then the state defined by 
\begin{equation}
\varphi_\mu(f)=\int_X f(x,0,x)d\mu(x),\quad f\in C_c(\mathcal{G}(X,\sigma))
\end{equation}
is a $KMS_\beta$ state of the full groupoid $C^*$-algebra $C^*(\mathcal{G}(X,\sigma))$ for the one parameter group of automorphisms $(\eta_t)$ defined as
\begin{equation*}
    \eta_t(f)(\gamma)=e^{itc_F(\gamma)}f(\gamma) , \quad f\in C_c(\mathcal{G}(X,\sigma)), \gamma\in \mathcal{G}(X,\sigma).
\end{equation*}
More than that, if the subgroupoid $c_F^{-1}(0)$ is \textit{principal}, then every $KMS_\beta$ state of $C^*(\mathcal{G}(X,\sigma))$ for $\eta$ is of the form $\varphi_\mu$, for a $\mu$ $e^{-\beta c_F}$-quasi-invariant measure (see Theorem 3.3.12 of \cite{Renault2009} for these assertions). By principal we mean that the \textit{isotropy} subgroupoid  $\Iso(\mathcal{G}):=\{\gamma\in \mathcal{G}: r(\gamma)=s(\gamma) \}$ coincides with $\mathcal{G}^{(0)}$. Observe that the units of $c_F^{-1}(0)$ are $X$.

As an example, suppose $F>0$. Then for $\gamma=(x,n-m,x)\in \Iso(c_F^{-1}(0))$, we have
$$\sum_{i=0}^{n-1}F(\sigma^i(x))=\sum_{i=0}^{m-1}F(\sigma^i(x))$$
which can only happen if $n=m$, thus $\gamma=(x,0,x)\in X$ and we conclude that $c_F^{-1}(0)$ is principal.
\end{remark}

Some natural questions arise from this new setting, such as when these measures are non-singular and when they are conservative. Also one could ask if these measures when restricted to $\mathcal{B}_{\Sigma_A}$ are also eigenmeasures in the standard Markov shift space, or even if an eigenmeasure in the classical case when extended to $X_A$ is an eigenmeasure in this generalized setting. Our aim now is to answer all these questions. First, we prove that every eigenmeasure on $X_A$ is non-singular in $U$.

\begin{proposition}\label{prop:eigen_measure_nonsingular_Generalized_Markov_shift} Every eigenmeasure on $X_A$ is non-singular in $U$.
\end{proposition}

\begin{proof} For every eigenmeasure with eigenvalue $\lambda >0$, we have
\begin{equation}\label{eq:conformal_eigenmeasure_functions_general}
    \int_{X_A} L_{-\beta F}(f)(x)d\mu(x) = \lambda \int_{U} f(x)d\mu(x),
\end{equation}
and then
\begin{equation}\label{eq:conformal_eigenmeasure_functions_general_modified}
    \int_{X_A} L_{-\beta \left( F + \frac{\log \lambda}{\beta}\right)}(f)(x)d\mu(x) = \int_{U} f(x)d\mu(x),
\end{equation}
so by the equivalence of conformal measures of Theorem \ref{thm:equivalences_conformal_measures_generalized_Markov_shift}, we have that $\mu$ is a $\left(-\beta \left( F + \frac{\log \lambda}{\beta}\right),1\right)$-conformal measure in the sense of Sarig, that is
\begin{equation}\label{eq:conformality_RN_derivative}
    \dfrac{d \mu \odot \sigma}{d\mu}(x) = e^{\beta \left( F(x) + \frac{\log \lambda}{\beta}\right)} = \lambda e^{\beta F(x)}, \quad \mu-a.e. \text{ on } U.
\end{equation}
Suppose that there exists $B \subseteq U$ a Borel set s.t. $\mu(B) = 0$ and $\mu(\sigma^{-1}B) > 0$. Then,
\begin{equation*}
    \mu \odot \sigma (\sigma^{-1}B) = \sum_{i \in \mathbb{N}} \mu (\sigma(C_i \cap \sigma^{-1}B)) \leq \sum_{i \in \mathbb{N}} \mu (B) =0.
\end{equation*}
So, $\mu \odot \sigma (\sigma^{-1}B) = 0$, and combining this with the equation \eqref{eq:conformality_RN_derivative} we have that
\begin{equation*}
    \mu (\sigma^{-1}B) = \lambda^{-1} \int_{\sigma^{-1}B} e^{-\beta F(x)} d \mu \odot \sigma = 0,
\end{equation*}
a contradiction. Conversely, if $B \subseteq U$ is a Borel set satisfying $\mu(B) > 0$ and $\mu(\sigma^{-1}B) = 0$, then there exists $i \in \mathbb{N}$ s.t. $\mu(\sigma(C_i\cap \sigma^{-1}B))>0$ and then $\mu \odot \sigma(B) >0$. But
\begin{equation*}
    \mu (\sigma^{-1}B) = \lambda^{-1} \int_{\sigma^{-1}B} e^{-\beta F(x)} d \mu \odot \sigma = \lambda^{-1} \sum_{k \in \mathbb{N}}\int_{C_k \cap \sigma^{-1}B} e^{-\beta F(x)} d \mu \odot \sigma > 0,
\end{equation*}
and therefore we have a contradiction again. We conclude that $\mu$ is non-singular in $U$.
\end{proof}

\begin{theorem}\label{thm:restriction_eigenmeasures} For $X_A$, let $F:U \to \mathbb{R}$ be a potential and $\lambda > 0$ such that there exists a $(F,\lambda)$-conformal measure $\mu$ on $X_A$. 
\begin{itemize}
    \item[$(a)$] Let $\mu_{\Sigma_A}$ the restriction measure of $\mu$ to $\mathcal{B}_{\Sigma_A}$, defined by
    \begin{equation*}
        \mu_{\Sigma_A}(E) := \mu(E\cap \Sigma_A), \quad E \in \mathcal{B}_{X_A}.
    \end{equation*}
    Then, $\mu_{\Sigma_A}$ is also a $(F,\lambda)$-conformal measure. Moreover, its domain restriction\footnote{Observe that $\mu_{\Sigma_A}\vert_{\mathcal{B}_{\Sigma_A}} = \mu\vert_{\mathcal{B}_{\Sigma_A}}$.} to $\mathcal{B}_{\Sigma_A}$, $\mu_{\Sigma_A}\vert_{\mathcal{B}_{\Sigma_A}}$, is a $(F\vert_{\Sigma_A},\lambda)$-conformal measure.
    \item[$(b)$] Let $\mu_{Y_A}$ the restriction measure of $\mu$ to $\mathcal{B}_{Y_A}$, defined by
    \begin{equation*}
        \mu_{Y_A}(E) := \mu(E\cap Y_A), \quad E \in \mathcal{B}_{X_A}.
    \end{equation*}
    Then, $\mu_{Y_A}$ is also a $(F,\lambda)$-conformal measure. Moreover, its domain restriction\footnote{Observe that $\mu_{Y_A}\vert_{\mathcal{B}_{Y_A}} = \mu\vert_{\mathcal{B}_{Y_A}}$.} to $\mathcal{B}_{Y_A}$, $\mu_{Y_A}\vert_{\mathcal{B}_{Y_A}}$, is a $(F\vert_{Y_A},\lambda)$-conformal measure.
\end{itemize} 
\end{theorem}

\begin{proof} Without loss of generality, we may assume $\lambda = 1$.

\textbf{Proof of (a):} $\mu_{\Sigma_A}$ is well defined since $\Sigma_A$ is a Borel subset of $X_A$, because of Proposition \ref{prop:Sigma_A_Borel_subset_in_X_A}. For every $E \in \mathcal{B}_{X_A}$, $E \subseteq U$, we have
\begin{equation}\label{eq:mu_Sigma_A_odot_sigma}
    \mu_{\Sigma_A} \odot \sigma (E) = \sum_{i \in \mathbb{N}}\mu_{\Sigma_A}(\sigma (E_i)) = \sum_{i \in \mathbb{N}}\mu(\sigma (E_i)\cap \Sigma_A) \stackrel{(\bullet)}{=} \sum_{i \in \mathbb{N}}\mu(\sigma (E_i\cap \Sigma_A)) = \mu \odot \sigma(E \cap \Sigma_A),
\end{equation}
where $(E_i)$ is a decomposition of $E$ such that $\sigma\vert_{E_i}$ is injective. In $(\bullet)$ we used the fact that $\sigma$ is surjective and then $\sigma(\sigma^{-1}(\Sigma_A)) = \sigma(\Sigma_A) = \Sigma_A$. Due to $\mu \odot \sigma \sim \mu$ in $U$ and \eqref{eq:mu_Sigma_A_odot_sigma}, we have that $\mu_{\Sigma_A} \odot \sigma \sim \mu_{\Sigma_A}$ and then the Radon-Nikodym derivative
\begin{equation*}
    \frac{d\mu_{\Sigma_A} \odot \sigma}{d\mu_{\Sigma_A}}
\end{equation*}
is well defined $\mu_{\Sigma_A}$-a.e. Now, for every $E \in \mathcal{B}_{X_A}$, $E \subseteq U$, we have
\begin{align*}
    \mu_{\Sigma_A} \odot \sigma (E) &= \mu \odot \sigma (E \cap \Sigma_A) = \int_{E\cap \Sigma_A} e^{-F(x)} d\mu(x) = \int_{X_A}\mathbbm{1}_{E\cap \Sigma_A} e^{-F(x)} d\mu(x)\\ &= \sup\left\{\int_{X_A}\varphi d\mu: 0 \leq \varphi (x)\leq \mathbbm{1}_{E\cap \Sigma_A}(x) e^{-F(x)} \text{ and } \varphi \text{ is a simple function} \right\}\\
    &\stackrel{(\dagger)}{=} \sup\left\{\int_{X_A}\varphi d\mu_{\Sigma_A}: 0 \leq \varphi (x)\leq \mathbbm{1}_{E\cap \Sigma_A}(x) e^{-F(x)} \text{ and } \varphi \text{ is a simple function} \right\}\\
    &= \int_{X_A} \mathbbm{1}_{E \cap \Sigma_A} e^{-F} d\mu_{\Sigma_A} = \int_{E} e^{-F} d\mu_{\Sigma_A}. 
\end{align*}
The equality $(\dagger)$ is justified as follows: for every simple function
\begin{equation*}
    \varphi = \sum_{i=1}^na_i \mathbbm{1}_{B_i}
\end{equation*}
satisfying $0 \leq \varphi (x)\leq \mathbbm{1}_{E\cap \Sigma_A}(x) e^{-F(x)}$, we may choose $B_i \in \mathcal{B}_{\Sigma_A} \subseteq \mathcal{B}_{X_A}$ for every $i$. Since $\mu = \mu_{\Sigma_A}$ on $\mathcal{B}_{\Sigma_A}$, we have
\begin{align*}
    \int_{X_A}\varphi d\mu = \int_{X_A}\varphi d\mu_{\Sigma_A}.
\end{align*}
We conclude, by uniqueness of the Radon-Nikodym derivative, that 
\begin{equation*}
    \frac{d\mu_{\Sigma_A} \odot \sigma}{d\mu_{\Sigma_A}}(x) = e^{-F(x)} \quad \mu_{\Sigma_A}-\text{a.e. on }U,
\end{equation*}
and therefore $\mu_{\Sigma_A}$ is a $(F,1)$-conformal measure. When we restrict $\mu$ to $\mathcal{B}_{\Sigma_A}$, we have
\begin{equation*}
    \mu\vert_{\mathcal{B}_{\Sigma_A}}(E) = \mu_{\Sigma_A}(E), \quad \forall E \in \mathcal{B}_{\Sigma_A}.
\end{equation*}
It is straightforward that
\begin{equation}
    \mu \odot \sigma\vert_{\mathcal{B}_{\Sigma_A}}(E) = \mu\vert_{\mathcal{B}_{\Sigma_A}} \odot \sigma (E) = \mu_{\Sigma_A} \odot \sigma (E), \quad \forall E \in \mathcal{B}_{\Sigma_A}.
\end{equation}
And therefore $\mu\vert_{\mathcal{B}_{\Sigma_A}}$ is a $(F\vert_{\Sigma_A},1)$-conformal measure.

\textbf{Proof of (b):} it follows from similar steps done in (a), by observing that $\mathcal{B}_{Y_A} \subseteq \mathcal{B}_{X_A}$, since $Y_A = \Sigma_A^c$ is a Borel subset of $X_A$.
\end{proof}

\begin{theorem} \label{thm:extension_eigenmeasures} Let $A$ be a transition matrix.
\begin{itemize}
    \item[$(a)$] Let $F:\Sigma_A \to \mathbb{R}$ be a potential, $\widetilde{F}$ be an extension of $F$ to $U$, and $\lambda >0$. Also, suppose that there exists a $(F,\lambda)$-conformal measure $\mu$ on $\Sigma_A$. Consider its natural extension $\mu_{\text{ext},\Sigma_A}$ on $X_A$, defined by
    \begin{equation*}
        \mu_{\text{ext},\Sigma_A}(E) := \mu(E\cap \Sigma_A), \quad E \in \mathcal{B}_{X_A}.
    \end{equation*}
    Then, $\mu_{\text{ext},\Sigma_A}$ is a $(\widetilde{F},\lambda)$-conformal measure on $X_A$.
    
    \item[$(b)$] Let $F:Y_A \to \mathbb{R}$ be a potential, $\widetilde{F}$ be an extension of $F$ to $U$, and $\lambda >0$. Also, suppose that there exists a $(F,\lambda)$-conformal measure $\mu$ on $Y_A$. Consider its natural extension $\mu_{\text{ext},Y_A}$ on $X_A$, defined by
    \begin{equation*}
        \mu_{\text{ext},Y_A}(E) := \mu(E\cap Y_A), \quad E \in \mathcal{B}_{X_A}.
    \end{equation*}
    Then, $\mu_{\text{ext},Y_A}$ is a $(\widetilde{F},\lambda)$-conformal measure on $X_A$.
\end{itemize}
\end{theorem}

\begin{proof} \textbf{Proof of (a):} observe that $\mu_{\text{ext},\Sigma_A}$ is well defined because $\mathcal{B}_{\Sigma_A} \subseteq \mathcal{B}_{X_A}$. W.l.o.g. we consider the case that $\mu \neq 0$ and $\lambda = 1$. For every $E \in \mathcal{B}_{X_A}$, $E \subseteq U$, we have
\begin{equation*}
    \mu_{\text{ext},\Sigma_A} \odot \sigma (E) = \sum_{i \in \mathbb{N}}\mu_{\text{ext},\Sigma_A}(\sigma (E)) = \sum_{i \in \mathbb{N}}\mu(\sigma (E)\cap \Sigma_A) \stackrel{(\bullet)}{=} \sum_{i \in \mathbb{N}}\mu(\sigma (E\cap \Sigma_A)) = \mu \odot \sigma(E \cap \Sigma_A),
\end{equation*}
where in $(\bullet)$ we used the same arguments for the $(\bullet)$ in Theorem \ref{thm:restriction_eigenmeasures}. Once again, we have $\mu_{\text{ext},\Sigma_A} \odot \sigma \sim \mu_{\text{ext},\Sigma_A}$ on $U$, because $\mu \odot \sigma \sim \mu$ and that both measures $\mu_{\text{ext},\Sigma_A} \odot \sigma$ and $\mu_{\text{ext},\Sigma_A}$ are both zero on every borel set contained in $Y_A$. Then, the Radon-Nikodym derivative
\begin{equation*}
    \frac{d\mu_{\text{ext},\Sigma_A} \odot \sigma}{d \mu_{\text{ext},\Sigma_A}}
\end{equation*}
is well defined on $U$. On the other hand, for every $E \in \mathcal{B}_{X_A}$, $E \subseteq U$, it follows that
\begin{align*}
    \mu_{\text{ext},\Sigma_A} \odot \sigma (E) &= \mu \odot \sigma (E\cap \Sigma_A) = \int_{E\cap \Sigma_A} e^{-F} d\mu = \int_{\Sigma_A}\mathbbm{1}_{E\cap \Sigma_A} e^{-F} d\mu \\
    &= \sup\left\{\int_{\Sigma_A}\varphi d\mu: 0 \leq \varphi \leq \mathbbm{1}_{E\cap \Sigma_A} e^{-F} \text{ and } \varphi \text{ is a simple function} \right\}\\
    &\stackrel{(\dagger)}{=} \sup\left\{\int_{X_A}\varphi d\mu_{\text{ext},\Sigma_A}: 0 \leq \varphi \leq \mathbbm{1}_{E\cap \Sigma_A} e^{-\widetilde{F}} \text{ and } \varphi \text{ is a simple function} \right\}\\
    &= \int_{X_A} \mathbbm{1}_{E \cap \Sigma_A} e^{-\widetilde{F}} d\mu_{\text{ext},\Sigma_A} = \int_{E} e^{-\widetilde{F}} d\mu_{\text{ext},\Sigma_A}. 
\end{align*}
The equality $(\dagger)$ is justified as follows: for any simple function on $\Sigma_A$
\begin{equation*}
    \varphi = \sum_{i=1}a_i \mathbbm{1}_{B_i}
\end{equation*}
such that $0 \leq \varphi \leq \mathbbm{1}_{E\cap \Sigma_A} e^{-F}$, we have that $a_i > 0$ and $B_i \in \mathcal{B}_{\Sigma_A} \subseteq \mathcal{B}_{X_A}$ for all $i$. On the other hand, for any simple function $\widetilde{\varphi}$ on $\mathcal{B}_{X_A}$, its restriction is a simple function on $\mathcal{B}_{\Sigma_A}$. Moreover,   
\begin{align*}
    \int_{X_A}\widetilde{\varphi} d\mu_{\text{ext},\Sigma_A} = \int_{\Sigma_A}\widetilde{\varphi} d\mu_{\text{ext},\Sigma_A} + \int_{Y_A}\widetilde{\varphi} d\mu_{\text{ext},\Sigma_A} = \int_{\Sigma_A}\widetilde{\varphi} d\mu_{\text{ext},\Sigma_A} = \int_{\Sigma_A}\widetilde{\varphi}\vert_{\Sigma_A} d\mu,
\end{align*}
because $\mu_{\text{ext},\Sigma_A} = \mu$ on $\mathcal{B}_{\Sigma_A} = \mathcal{B}_{X_A} \cap \Sigma_A$ and $\mu_{\text{ext},\Sigma_A}(Y_A) = 0$. Also, for any simple function $\varphi$ on $\Sigma_A$ we always have an extension of $\varphi$ because $\Sigma_A$ is a Borel set of $X_A$. For instance, the function
\begin{equation*}
    \widetilde{\varphi}(x):= \begin{cases}
        \varphi(x), \quad \text{if } x \in \Sigma_A,\\
        0, \quad \text{otherwise};
    \end{cases}
\end{equation*}
is simple function which extends $\varphi$. For any extension $\widetilde{\varphi}$ of $\varphi$ it is true that
\begin{align*}
    \int_{X_A}\widetilde{\varphi} d\mu_{\text{ext},\Sigma_A} = \int_{\Sigma_A}\widetilde{\varphi}\vert_{\Sigma_A} d\mu_{\text{ext},\Sigma_A} + \int_{Y_A}\widetilde{\varphi}\vert_{Y_A} d\mu_{\text{ext},\Sigma_A} = \int_{\Sigma_A}\widetilde{\varphi} d\mu_{\text{ext},\Sigma_A} = \int_{\Sigma_A}\varphi d\mu.
\end{align*}
And therefore $(\dagger)$ is justified.
We conclude, by uniqueness of the Radon-Nikodym derivative, that 
\begin{equation*}
    \frac{d\mu_{\text{ext},\Sigma_A} \odot \sigma}{d\mu_{\text{ext},\Sigma_A}}(x) = e^{-\widetilde{F}(x)} \quad \mu_{\text{ext},\Sigma_A}-\text{a.e. on }U,
\end{equation*}
that is, $\mu_{\text{ext},\Sigma_A}$ is a $(\widetilde{F},1)$-conformal measure.

\textbf{Proof of (b):} it is similar to the previous proof. 
\end{proof}

The next corollary is a straightforward consequence of the Theorems \ref{thm:restriction_eigenmeasures} and \ref{thm:extension_eigenmeasures}.

\begin{corollary} \label{cor:decomposition_eigenmeasures} Consider the space $X_A$, a potential $F:U \to \mathbb{R}$ $\beta > 0$, and $\lambda > 0$. Then every $(F,\lambda)$-conformal measure can be decomposed into the linear combination of the two $(F,\lambda)$-conformal measures
\begin{equation*}
    \mu = \mu_{\Sigma_A} + \mu_{Y_A}
\end{equation*}
as defined in the Theorem \ref{thm:restriction_eigenmeasures}.

Conversely, let $F:\Sigma_A \to \mathbb{R}$ and $G: Y_A \cap U \to \mathbb{R}$ be potentials with a common extension $\widetilde{F}:U \to \mathbb{R}$, $\lambda > 0$. Also, let $\nu$ and $\eta$ be, respectively, a $(F,\lambda)$-conformal measure and a $(G,\lambda)$-conformal measure for same $\beta>0$ and associated eigenvalue $\lambda > 0$. Then, any positive linear combination of the measures 
\begin{equation}
    \mu = a_1\nu_{\text{ext},\Sigma_A}+a_2\eta_{\text{ext},Y_A}, \quad a_1, a_2 \in \mathbb{R}_+,
\end{equation}
with  $\nu_{\text{ext},\Sigma_A}$ and $\eta_{\text{ext},Y_A}$ as defined in Theorem \ref{thm:extension_eigenmeasures}, is a $(\widetilde{F},\lambda)$-conformal measure.
\end{corollary}

\begin{remark} In particular, every conformal probability, in the sense of Sarig on $X_A$, can be decomposed into a convex combination of finite conformal measures, with one vanishing on $\Sigma_A$, and another one vanishing on $Y_A$. In particular, if the decomposition is non-trivial in the sense that both $\Sigma_A$ and $Y_A$ are not null, then the decomposition can be taken as a convex conbination of probabilities.
\end{remark}

We define the notion of measure that lives on a set.

\begin{definition} Given $(X,\mathcal{F},\mu)$ a measure space, where $\mu$ is a positive measure, we say that $\mu$ lives on $B \in \mathcal{F}$ if $\mu(B) > 0$ and $\mu(B^c) = 0$.
\end{definition}

Now we characterize measures that live on $Y_A$-families. Every non-zero positive Borel measure $\mu$ on $X_A$ that lives on $Y_A(\xi^{0,\mathfrak{e}})$ for some $\mathfrak{e}$ is necessarily an atomic measure. In other words, these measures are written in the form  $\mu(E) = \sum_{\omega \in \mathfrak{R}_\mathfrak{e}} \mathbbm{1}_{E}(\omega) c_\omega$, where we identify each configuration in $Y_A(\xi^{0,\mathfrak{e}})$ with its stem, since this defines a bijection between $Y_A(\xi^{0,\mathfrak{e}})$ and $\mathfrak{R}_\mathfrak{e}$, and $E \subseteq X_A$ is a measurable set. We define
\begin{equation}\label{eq:coefficients_Y_A_confomal}
    c_\omega := \mu(\{\omega\}), \quad \omega \in \mathfrak{R}_\mathfrak{e},
\end{equation}
From now on, the idea is to consider the family of variables $\{c_\omega\}_{\omega \in \mathfrak{R}_\mathfrak{e}}$. The Denker-Urba\'nski conformality condition \eqref{eq:conformal_Ur_sets_potential} here is written with
\begin{equation*}
    D(\omega) = e^{F(\omega)}, \quad \omega \neq e,
\end{equation*}
and we get the general formulation for the conformal measures in $Y_A$-families in the theorem below.

\begin{theorem}\label{theorem.coeficientes} A measure $\mu$ which lives on a $Y_A$-family $Y_A(\xi^{0,\mathfrak{e}})$, where $\xi^{0,\mathfrak{e}}$ is an empty-stem configuration, satisfies the Denker-Urba\'nski conformality condition if and only if the coefficients $c_\omega$ in \eqref{eq:coefficients_Y_A_confomal} satisfies
\begin{equation*}
    c_\omega D(\omega) = c_{\sigma(\omega)}, \quad \omega \in \mathfrak{R}_\mathfrak{e}\setminus \{e\}.
\end{equation*}
\end{theorem}

\begin{proof}  It is straightforward from the Denker-Urba\'nski conformality condition for characteristic functions on the special set $\{\omega\}$, hence the condition above is necessary. The converse is clear because, for every special set $E$, we have that $e \notin E$ and $$\sum_{\omega\in E}D(\omega)c_\omega=\sum_{\omega\in E}c_{\sigma(\omega)}$$ implies the Denker-Urba\'nski corformality condition.  
\end{proof}

From the theorem above, every non-zero conformal measure living on an $Y_A$-family necessarily gives mass to every point of this family. Observe that the identity $c_{\omega} D(\omega) = c_{\sigma(\omega)}$, $\omega \neq e$, implies
\begin{equation*}
c_{\omega} \prod_{i=0}^{|\omega|-1}D(\sigma^i(\omega)) = c_{e}, \quad \omega \in \mathfrak{R}_\mathfrak{e}\setminus\{e\},
\end{equation*}
where $c_e := \mu(\{\xi^{0,\mathfrak{e}}\})$. The equation above can be rewritten as
\begin{equation}\label{eq:c_omega_new}
c_{\omega} e^{ F_{|\omega|}(\omega)} = c_{e}, \quad \omega \in \mathfrak{R}_\mathfrak{e} \setminus\{e\}.
\end{equation}
where $F_n$ is the Birkhoff's sum. 

In order to construct any potential which would give a $e^F$-conformal probability measure in a $Y_A$-family we must keep $c_e > 0$, otherwise all other $c_\omega$'s are zero by \eqref{eq:c_omega_new}. That is equivalent to impose $c_\omega > 0$ for all $\omega \in \mathfrak{R}_\mathfrak{e}$, since it is a necessary condition to obtain $c_e > 0$. At the same time we wish to have
\begin{equation}\label{eq:c_omega_probability}
    \sum_{\omega \in \mathfrak{R}_\mathfrak{e}}c_\omega = 1,
\end{equation}
which imposes that $\mu$ is in fact a probability measure. The following result is a consequence of this caracterization.

\begin{theorem}\label{thm:uniqueness_conformal_probabilities_Y_A_families} Fixed a potential $F:U \to \mathbb{R}$ and a $Y_A$-family $Y(\xi^{0,\mathfrak{e}})$, there exists at most one $e^F$-conformal probability living on such family. 
\end{theorem}

\begin{proof} Nothing is need to be proven if there are no $e^F$-conformal probabilities living on $Y_A$. So suppose there exists a two collection of strictly positive numbers $\{c_{\omega}\}_{\omega \in \mathfrak{R}_\mathfrak{e}}$ and $\{d_{\omega}\}_{\omega \in \mathfrak{R}_\mathfrak{e}}$ satisfying
\begin{align*}
    \begin{cases}
        c_{\omega} = e^{-F_{|\omega|}(\omega)} c_{e}, \quad \omega \in \mathfrak{R}_\mathfrak{e},\\
        \sum_{\omega \in \mathfrak{R}_\mathfrak{e}}c_\omega = 1;
    \end{cases}
    \quad \text{and} \quad 
    \begin{cases}
        d_{\omega} = e^{-F_{|\omega|}(\omega)} d_{e}, \quad \omega \in \mathfrak{R}_\mathfrak{e},\\
        \sum_{\omega \in \mathfrak{R}_\mathfrak{e}}d_\omega = 1.
    \end{cases}
\end{align*}
Since $c_e$ and $d_e$ are positive numbers, there exists $\lambda \in \mathbb{R}_{+}^*$ s.t. $d_e = \lambda c_e$. Then, for every $\omega \in \mathfrak{R}_\mathfrak{e}$,
\begin{equation*}
    d_{\omega} = e^{-F_{|\omega|}(\omega)} d_{e} = e^{-F_{|\omega|}(\omega)} \lambda c_e = \lambda d_{\omega}.
\end{equation*}
On the other hand,
\begin{equation*}
    1 = \sum_{\omega \in \mathfrak{R}_\mathfrak{e}}d_\omega = \sum_{\omega \in \mathfrak{R}_\mathfrak{e}} \lambda c_\omega = \lambda \sum_{\omega \in \mathfrak{R}_\mathfrak{e}} c_\omega = \lambda,
\end{equation*}
and hence $\lambda = 1$. Therefore, $d_\omega = c_\omega$ for every $\omega \in \mathfrak{R}_\mathfrak{e}$.
\end{proof}

\begin{corollary}\label{cor:eigenmeasures_dimension_1} Given a potential $F:U \to \mathbb{R}$, the space of eigenmeasures for an associated eigenvalue $\lambda$ living on an $Y_A$ family has dimension at most $1$.
\end{corollary}

\begin{proof} It is straightforward from Theorems \ref{thm:equivalences_conformal_measures_generalized_Markov_shift} and \ref{thm:uniqueness_conformal_probabilities_Y_A_families}.
\end{proof}

\begin{lemma}\label{lemma:conformal_Y_A_minus_Y_A_family} Given a potential $F:U \to \mathbb{R}$, suppose that there exists a $e^F$-conformal measure living on $Y_A$ and let $Y_A^1$ be a $Y_A$-family. Then the restriction of $\mu$ to $(Y_A^1)^c$, given by 
\begin{equation*}
    \nu(B)_{(Y_A^1)^c} := \nu(B \cap (Y_A^1)^c), \quad B \in \mathcal{B}_{X_A}
\end{equation*}
is a $e^F$-conformal measure as well. 
\end{lemma}

\begin{proof} Similar to Theorem \ref{thm:restriction_eigenmeasures}. 
\end{proof}

\begin{corollary}\label{cor:extremal_conforma_Y_A_characterization} Let $F:U \to \mathbb{R}$ be a potential. The extremal $e^{F}$-conformal probabilities living on $Y_A$ are precisely the ones living on each $Y_A$-family. 
\end{corollary}

\begin{proof} A $e^{F}$-conformal probability living on a $Y_A$-family is necessarily extremal due to Theorem \ref{thm:uniqueness_conformal_probabilities_Y_A_families}. Conversely, let $\mu$ be an $e^{F}$-conformal probability on $Y_A$ and suppose that $\mu$ gives mass on more than one $Y_A$-family. Let $Y_A^1$ be a $Y_A$-family satisfying $\mu(Y_A^1)>0$. Then $\mu((Y_A^1)^c) >0$. Then the measures given by
\begin{align*}
    \mu_1 (B):= \frac{\mu(B \cap Y_A^1)}{\mu(Y_A^1)} \quad \text{and} \quad \mu_2 (B):= \frac{\mu(B \cap (Y_A^1)^c)}{\mu((Y_A^1)^c)},
\end{align*}
defined for every $B \in \mathcal{B}_{X_A}$ are $e^F$-conformal probabilities on $X_A$ such that
\begin{equation*}
    \mu = \mu(Y_A^1) \mu_1 + \mu((Y_A^1)^c) \mu_2.
\end{equation*}
Since $\mu(Y_A^1) + \mu((Y_A^1)^c) = 1$ and $\mu(Y_A^1) \in (0,1)$, we written $\mu$ as a non-trivial convex combination of two distinct $e^F$-conformal probabilities, we conclude that $\mu$ is not extremal.
\end{proof}

Now we study phase transition phenomena on $X_A$.

\section{Phase Transition on $X_A$}

In this section we present phase transition results for conformal measures and eigenmeasures on $X_A$ and we connect these results to the standard theory on $\Sigma_A$. In particular, we present the length type phase transition, which consists in a change of space where the measure lives, from $Y_A$ to $\Sigma_A$, on decreasing $\beta$. The examples here presented are the renewal, pair renewal and prime renewal shift spaces. It is important to notice that in these examples, the set of empty stem configurations is countable, and therefore $Y_A$ is countable as well. In particular, every conformal probability that lives on a $Y_A$-family is an extremal measure, in the sense that it can only be represented as a convex combination of conformal measures by the trivial one.

\subsection{Phase transitions of conformal probabilities for the generalized renewal shift}

We recall that the generalized renewal shift has only one $Y_A$-family, which is $Y_A$ itself.

\begin{example}We first look the class of potentials that depends only on the length of the word, i.e, $F(\omega)=F(|\omega|)$. In this case, the coefficients $c_\omega$'s have the same property, i.e., $c_\omega = c_{|\omega|}$. This imposition affects directly \eqref{eq:c_omega_probability}, one may rewrite it as
\begin{equation*}
    c_e+\sum_{n \in \mathbb{N}}\sum_{\substack{\omega \in \mathfrak{R} \\ |\omega| = n}}c_\omega \stackrel{\text{prop. \ref{prop:cardinality_words_length}}}{=} c_0 + \sum_{n \in \mathbb{N}} 2^{n-1} c_{n}=1,
\end{equation*}
where $c_n$ is the coefficient $c_\omega$ when $|\omega|=n$. The equality above imposes $c_e \in (0,1)$. We summarize the conditions on the atomic probability $\mu$ which vanishes on $\Sigma_A$
\begin{equation*}
    \begin{cases}
        c_0 + \sum_{n \in \mathbb{N}} 2^{n-1} c_{n}=1,\\
        c_0 \in (0,1).
    \end{cases}
\end{equation*}
By \eqref{eq:c_omega_new} and $|\omega|=n$, we have $c_n e^{F_n(\omega)} = c_0$, then
\begin{equation}\label{eq:S_n_F}
    F_n(\omega) = \log \left(\frac{c_0}{c_n}\right), \quad n \in \mathbb{N}.
\end{equation}
The identity above allow us to determine $F$ in $Y_A\setminus\{\xi^0\}$,
\begin{align*}
    F(\omega) &= \sum_{i=0}^{n-1}F[\sigma^i(\omega)] - \sum_{i=1}^{n-1}F[\sigma^i(\omega)] = \log \left(\frac{c_{n-1}}{c_n}\right), \quad n \in \mathbb{N}.
\end{align*}
For $\alpha>2$, take $c_n = \frac{\alpha-2}{\alpha^n(\alpha-1)}$. The potential which makes $\mu$, defined by the coefficients $c_n$,  a $e^{F}$-conformal probability measure is given by the constant function $F = \log \alpha$ defined in $X_A \setminus \{\xi^0\}$. 
 \noindent
\end{example} 
\begin{remark}
In fact, because of the structure of our renewal shift, potentials that depends only on the length of the word, are the constant ones. To see this, take any continuous potential $F:X_A\setminus\{\xi^0\} \to \mathbb{R}$ such that, for any $\omega \in Y_A$ we have $F(\omega)= g(|\omega|)$ where $g: \mathbb{N} \rightarrow \mathbb{R}$. In this case, $F$ is a constant function. Let $x \neq y$ on $\Sigma_A$ and consider two sequences $(x_n)_{n \in \mathbb{N}}$ and $(y_n)_{n \in \mathbb{N}}$ on $Y_A$ such that $x_n \rightarrow x$ and $y_n \rightarrow y$, with $|x_n| = |y_n|=n$, this choice is possible because $A(1,n) = 1$ for every $n \in \mathbb{N}$. This implies that $F(x)=F(y)$, since $\Sigma_A$ is dense on $X_A$,  $F$ is constant.
\end{remark} 
\begin{theorem}\label{theorem.potentialwithonecoordinate} Consider a potential $F:X_A\setminus\{\xi^0\} \to \mathbb{R}$ and $\beta>0$, we have the following:
\begin{itemize}
    \item[$(i)$] If $\inf F>0$, for $\beta>\frac{\log 2}{\inf F}$, there exists a unique $e^{\beta F}$-conformal probability measure $\mu_\beta$ that vanishes in $\Sigma_A$.
    \item[$(ii)$] If $0\leq \sup F < + \infty$ and $\beta\leq \frac{\log 2}{\sup F}$, there are no $e^{\beta F}$-conformal probability measures that vanish in $\Sigma_A$.
\end{itemize}
\end{theorem}

\begin{proof} The equations \eqref{eq:c_omega_new} and \eqref{eq:c_omega_probability} give us

\begin{equation}\label{eq:conformal_Y_A_formula}
    1 + \sum_{\omega \in \mathfrak{R}\setminus \{e\}} e^{-\beta\sum_{j=0}^{|\omega|-1}F(\sigma^j(\omega))} = \frac{1}{c_e}>0.
\end{equation}
Since $F(\omega) \geq \inf F$ for all $\omega \in R$, by Proposition \ref{prop:cardinality_words_length} we obtain
\begin{align*}
    1 + \sum_{\omega \in \mathfrak{R}\setminus \{e\}} e^{-\beta\sum_{j=0}^{|\omega|-1}F(\sigma^j(\omega))}\leq  1 + \frac{1}{2}\sum_{n \in \mathbb{N}} \left(\frac{2}{e^{\beta \inf F}}\right)^{n}.
\end{align*}
The series $\sum_{n \in \mathbb{N}}  \left(\frac{2}{e^{\beta \inf F}}\right)^{n}$ converges if $\beta > \frac{\log 2}{\inf F}$, therefore the validity of the last inequality grants that the series $\sum_{\omega \in \mathfrak{R}\setminus \{e\}} e^{-\beta\sum_{j=0}^{|\omega|-1}F(\sigma^j(\omega))}$ converges and we obtain the existence of a $e^{\beta F}$-conformal probability measure $\mu_\beta$ that vanishes on $\Sigma_A$, given by the coefficients $c_\omega$ in equation \eqref{eq:c_omega_new}. The uniqueness is straightforward. This proves item (i) and a similar procedure proves item (ii). Indeed, it is clear that
\begin{equation}
-\beta \sum_{j=0}^{n-1}F(\sigma(\omega))\geq -\beta n \sup F.   
\end{equation}
Hence,
\begin{align*}
    \sum_{\omega \in \mathfrak{R}\setminus \{e\}} e^{-\beta\sum_{j=0}^{|\omega|-1}F(\sigma^j(\omega))}\geq 
    1/2\sum_{n\in \mathbb{N}}\left(\frac{2}{e^{\beta \sup F}}\right)^n.
\end{align*}
The last sum diverges if $\beta\leq \frac{\log 2}{\sup F}$, which means that no $e^{\beta F}$-conformal probability measure vanishing on $\Sigma_A$ can be obtained  in such interval.
\end{proof}

\begin{corollary}\label{cor:phase_transition_conformal_potential_1}
Let $F\equiv1$. Then, for the constant $\beta_c=\log 2$, the result follows:
\begin{itemize}
    \item[(i)] For $\beta>\beta_c$ we have a unique $e^\beta$-conformal probability measure that vanishes on $\Sigma_A$.
    \item[(ii)] For $\beta = \beta_c$ there is a unique $e^\beta$-conformal probability measure that vanishes on $Y_A$.
    \item[(iii)] For $\beta < \beta_c$ there are not $e^\beta$-conformal probability measures.
\end{itemize}
\end{corollary}
\begin{proof}
For this potential, we have $\inf F = \sup F = 1$ and we apply Theorem \ref{theorem.potentialwithonecoordinate} for the constant potential $F\equiv 1$. For $\beta_c$, it is a straightforward calculation that the series associated with it diverges. This characterizes the existence and the absence of the $e^\beta$-conformal probability living on $Y_A$. For the aforementioned absence region $\beta_c \leq \log 2$, if there exists a $e^\beta$-conformal probability $\nu$, then necessarily $\nu$ vanishes in $Y_A$. By Theorem \ref{thm:restriction_eigenmeasures} we may restrict the measure to the standard theory, and turn it into a question on finding eigenmeasures for the potential $-\beta F$ by corollary \ref{cor:equivalences_conformality_classical}. In particular, this potential is positive recurrent for every $\beta$. Indeed, we have that
\begin{align*}
    Z_n(-\beta ,[1]) = \sum_{\sigma^n x = x} e^{-\beta n} \mathbbm{1}_{[1]}(x) = 2^{n-1} e^{-\beta n},
\end{align*}
and then 
\begin{align*}
    P_G(-\beta) = \lim_n \frac{1}{n} \log Z_n(-\beta F,[1]) =  \log 2 - \beta.
\end{align*}
And setting $\lambda = e^{P_G(-\beta)} = 2 e^{-\beta}$ we have
\begin{equation*}
    \sum_{n \in \mathbb{N}} \lambda^{-n} Z_n(-\beta,[1]) = \sum_{n \in \mathbb{N}} 2^{-n} e^{\beta n} 2^{n-1} e^{-\beta n} = \infty,
\end{equation*}
and then the potential is recurrent for every $\beta >0$. Observe that
\begin{equation*}
    \mathbbm{1}_{[\varphi_1 = n]}(x) = 1 \iff x = \overline{1,n,n-1,...,2},
\end{equation*}
hence
\begin{equation*}
    Z_n^*(-\beta,[1]) =  \sum_{\sigma^n x = x} e^{-\beta n} \mathbbm{1}_{[\varphi_1 = n]}(x) = e^{-\beta n},
\end{equation*}
therefore
\begin{equation*}
    \sum_{n \in \mathbb{N}} n \lambda^{-n} Z_n^*(-\beta,[1]) = \sum_{n \in \mathbb{N}} 2^{-n} e^{\beta n} e^{-\beta n} = 1 < \infty,
\end{equation*}
and we conclude that in fact the potential is positive recurrent for every $\beta$. By the Generalized RPF Theorem \ref{thm:RPF_Generalized} and the recurrence, there exists an eigenmeasure for every $\beta$, and for every eigenmeasure, it follows that its associated eigenvalue is $\lambda$. So necessarily we only can have an $e^\beta$-conformal measure at $\beta = \log 2$. Moreover, the positive recurrence implies, via Proposition \ref{prop:eigenmeasures_for_positive_recurrent_potentials_dimension_1}, that we have at most one probability eigenmeasure for every $\beta > 0$. Now, we must find $\nu$ that solves the equation
\begin{equation*}
     \int f d\nu = \int L_{-\beta} f d\nu, \quad f \in L^1(\nu),
\end{equation*}
Observe that for every $n \in \mathbb{N}$ we obtain
\begin{equation}\label{eq:eigenmeasure_n_times_conformal}
    \int f d\nu = \int L_{-\beta}^n f d\nu, \quad f \in L^1(\nu),
\end{equation}
and 
\begin{equation}\label{eq:Ruelle_n_times_conformal}
    (L_{-\beta}^n f)(x) = \sum_{\sigma^n(y)=x}e^{-\beta n}f(y) \quad \nu-\text{a.e.}, \quad f\in L^1(\nu).
\end{equation}
For every admissible (positive) word $\alpha$ that ends in `$1$', we get
\begin{align*}
    \nu([\alpha]) \stackrel{\text{\eqref{eq:eigenmeasure_n_times_conformal}}}{=} \sum_{a \in \mathbb{N}}\int_{[a]} L_{-\beta}^{|\alpha|} \mathbbm{1}_{[\alpha]} d\nu \stackrel{\text{\eqref{eq:Ruelle_n_times_conformal}}}{=} \sum_{a \in \mathbb{N}}\int_{[a]} \sum_{\sigma^{|\alpha|}(y)=x}2^{-|\alpha|} \mathbbm{1}_{[\alpha]}(y) d\nu(x).
\end{align*}
Now we claim that for given $a \in \mathbb{N}$, it follows that
\begin{equation*}
    \sum_{\sigma^{|\alpha|}(y)=x} \mathbbm{1}_{[\alpha]}(y) = 1, \quad x \in [a].
\end{equation*}
In fact, the system of equations
\begin{equation}\label{eq:system_equations_determining_Sarig}
    \begin{cases}
        \mathbbm{1}_{[\alpha]}(y) = 1,\\
        \sigma^{|\alpha|}(y)=x;
    \end{cases}
\end{equation}
admits the unique solution $y=\alpha x$ for all $x \in \mathbb{N}^{\mathbb{N}_0}$. Since $\alpha$ is admissible, such solution belongs to $[a]$ if and only if $A(\alpha_{|\alpha|-1},a)=1$, which is satisfied because $\alpha_{|\alpha|-1}=1$. The claim is proved and consequently we get 
\begin{equation}\label{eq:Sarig_measure_potential_1}
    \nu([\alpha]) = 2^{-|\alpha|}.
\end{equation}
By observing that $[n] = [n,n-1,...,1]$ for every $n > 2$ we have
\begin{equation*}
    \sum_{n \in \mathbb{N}}\nu([n]) = \sum_{n \in \mathbb{N}}\nu([n])2^{-n} = 1,
\end{equation*}
and therefore $\nu$ is a probability.
\end{proof}

\begin{remark}\label{remark:eigenmeasures_renewal_Sigma_A} The probability in \eqref{eq:Sarig_measure_potential_1} is the unique eigenmeasure in $\Sigma_A$, for every $\beta > 0$. Indeed, for general $\beta$, the eigenmeasure satisfies
\begin{equation*}
    2e^{-\beta} \int f d\nu = \int L_{-\beta} f d\nu, \quad f \in L^1(\nu),
\end{equation*}
Similarly to the proof in corollary above, we have, for every $n \in \mathbb{N}$ the following:
\begin{equation}\label{eq:eigenmeasure_n_times}
    (2e^{-\beta})^n \int f d\nu = \int L_{-\beta}^n f d\nu, \quad f \in L^1(\nu).
\end{equation}
and in this case for every $\alpha$ admissible ending in `$1$', we obtain
\begin{align*}
    (2e^{-\beta})^{|\alpha|}\nu([\alpha]) \stackrel{\text{\eqref{eq:eigenmeasure_n_times}}}{=} \sum_{a \in \mathbb{N}}\int_{[a]} L_{-\beta}^{|\alpha|} \mathbbm{1}_{[\alpha]} d\nu = \sum_{a \in \mathbb{N}}\int_{[a]} \sum_{\sigma^{|\alpha|}(y)=x}e^{-\beta|\alpha|} \mathbbm{1}_{[\alpha]}(y) d\nu(x).
\end{align*}
By \eqref{eq:system_equations_determining_Sarig} we conclude that
\begin{equation}
    (2e^{-\beta})^{|\alpha|} \nu([\alpha]) = e^{-\beta|\alpha|},
\end{equation}
and therefore $\nu$ is the same as in \eqref{eq:Sarig_measure_potential_1}.
\end{remark}

\begin{remark} As we can see in the proof of corollary \ref{cor:phase_transition_conformal_potential_1}, the $e^\beta$-conformal probabilities are eigenmeasures of the Ruelle's transformation for the potential $F=-1$. One could ask about the existence of eigenmeasure probabilities associated to the eigenvalue $e^{P_G(-\beta)}$ which, unlike in Remark \ref{remark:eigenmeasures_renewal_Sigma_A}, vanishes in $\Sigma_A$. The answer is \textbf{no}. In fact, since in this case we have $P_G(-\beta) = \log 2 - \beta$, we have that $\mu$ is an eigenmeasure probability for $L_{-\beta}$ for the eingevalue $e^{P_G(-\beta)}$ if and only if
\begin{equation*}
    L_{-\beta - P(-\beta)}^* \mu = \mu,
\end{equation*}
that is,
\begin{equation*}
    L_{-\log 2}^* \mu = \mu,
\end{equation*}
and by Theorem \ref{thm:equivalences_conformal_measures_generalized_Markov_shift}, this is equivalent to state that $\mu$ is a $e^{\log 2}$-conformal probability, and by corollary \ref{cor:phase_transition_conformal_potential_1}, $\mu$ necessarily must vanish in $Y_A$.
\end{remark}

The picture \ref{fig:renewal_F_equals_1} compares the standard formalism with the generalized one for the renewal shift and potential $F=1$, as in corollary \ref{cor:phase_transition_conformal_potential_1}.

\begin{figure}[h!]
  \hspace{-.5cm}
 \includegraphics[scale=.4]{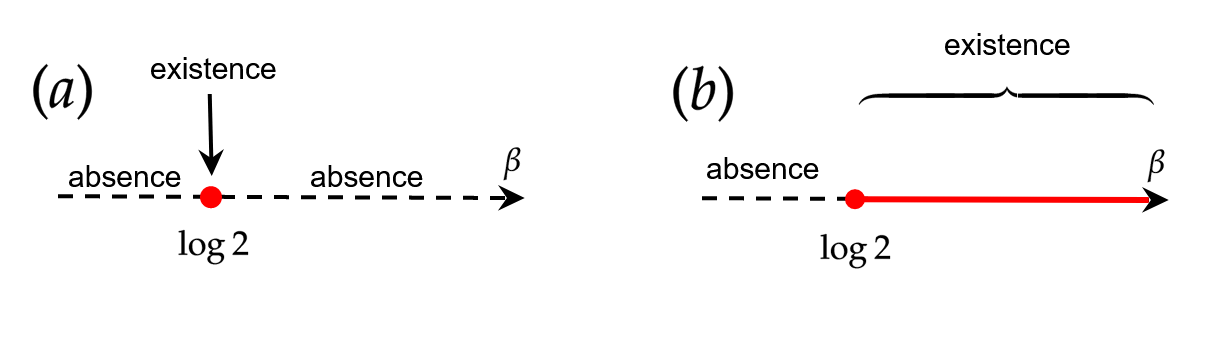}
 
 \caption{The phase transitions on different thermodynamic formalisms for conformal probabilities in the renewal shift space with potential $F=1$. The picture $(a)$ represents the standard formalism on $\Sigma_A$, where we have a unique $e^\beta$-conformal probability for a unique possible inverse of temperature, namely $\beta = \log 2$ (red dot). The picture $(b)$ represents the generalized formalism on $X_A$ and, unlike in $(a)$, we are able to detect not only the previous standard $e^\beta$-conformal probability, but also, for each $\beta > \log 2$, we have a unique $e^\beta$-conformal probability (red interval), in this case the measures live on $Y_A$.  \label{fig:renewal_F_equals_1}}
\end{figure}

Let $\mu_\beta$, $\beta > \log 2$, be the conformal measure of the corollary \ref{cor:phase_transition_conformal_potential_1} for the inverse of temperature $\beta$. This measure is explicitly obtained by the equation \eqref{eq:conformal_Y_A_formula} combined with \eqref{eq:c_omega_new}. By denoting $\mu_\beta(\{\alpha \xi^0\}):= c_{\alpha}^\beta$ we obtain
\begin{equation}\label{eq:conformal_measure_Y_A_explicit}
     c_\omega^\beta := \frac{e^{-\beta |\alpha|}}{1+\frac{1}{2}\sum_{n \in \mathbb{N}}e^{n(\log 2 - \beta)}} = \frac{e^{-|\alpha|\beta}(e^\beta-2)}{e^\beta-1}.
\end{equation}

Perhaps the corollary \ref{cor:phase_transition_conformal_potential_1} can be seen as a pathological fact in contrast with the Sarig's theorem in \cite{Sarig2001} about renewal shifts considering $\Sigma_A$ instead of $X_A$. He proved, for potentials regular enough, if we have a phase transition at some $\beta_c$, then there exist $(-\beta F, e^{P(-\beta F)})$-conformal measures at high temperatures ($\beta < \beta_c$) and these measures do not exist at low temperatures ($\beta > \beta_c$). Our theorem tells us the exact opposite behavior for the $e^{\beta F}$-conformal measures which vanish on $\Sigma_A$. Nevertheless, the Sarig's conformal measure is intrinsically connected with the ones of the corollary \ref{cor:phase_transition_conformal_potential_1}: there exists the weak$^*$-limit of the net of measures $\{\mu_\beta\}_{\beta > \beta_c}$ of the aforementioned corollary when $\beta \to \beta_c$. In order to verify this fact, we need to discuss some particularities of the topology of $X_A$ of the renewal shift to study the weak$^*$ convergence of measure on this space as it follows.

We recall from subsection \ref{subsec:cylinder_topology_Renewal} that every generalized cylinder on a positive admissible word is written by choosing this word ending in `$1$'. Hence,
\begin{equation*}
    \mu_\beta(C_\alpha) = \sum_{\omega \in Y_A: \alpha \in \llbracket\omega\rrbracket} c_\omega^\beta,
\end{equation*}
and then,
\begin{equation*}
    \mu_\beta(C_\alpha) = \sum_{n \in \mathbb{N}}\sum_{\substack{\omega \in Y_A: \alpha \in \llbracket\omega\rrbracket \\ |\omega| = n}} c_n^\beta,
\end{equation*}
where
\begin{equation*}
    c_n^\beta := \frac{e^\beta-2}{e^{n \beta }(e^\beta-1)}.
\end{equation*}
However, if $\alpha \in \llbracket \omega \rrbracket$, then $|\omega|\geq |\alpha|$ and hence
\begin{equation*}
    \mu_\beta(C_\alpha) = \sum_{n = |\alpha|}^\infty\sum_{\substack{\omega \in Y_A: \alpha \in \llbracket\omega\rrbracket \\ |\omega| = n}} c_n^\beta.
\end{equation*}
\begin{lemma} Let $\omega \in \mathfrak{R}$ s.t. $\alpha \in \llbracket \omega \rrbracket$. Write $n=|\omega| = |\alpha| + p$, $p \in \mathbb{N}_0$. If $p=0$, then there exists a unique word in $Y_A$ s.t. $\alpha \in \llbracket \omega \rrbracket$. If $p \in \mathbb{N}$, then there are $2^{p-1}$ words in $Y_A$ which satisfies $\alpha \in \llbracket \omega \rrbracket$.
\end{lemma}

\begin{proof} If $p=0$, the result is obvious, then suppose that $p \in \mathbb{N}$, that is, $|\omega|>|\alpha|$. Then, $\omega = \alpha \gamma$, where both $\alpha$ and $\gamma$ end with `$1$' and $|\gamma|=p$ and therefore the number of possibilities for $\gamma$, and consequently $\omega$, is $2^{p-1}$ due to Proposition \ref{prop:cardinality_words_length}.
\end{proof}

By the lemma above we have that
\begin{equation*}
    \mu_\beta(C_\alpha) = c_{|\alpha|}^\beta + \sum_{n=|\alpha|+1}^\infty 2^{n-|\alpha|-1}c_n^\beta = c_{|\alpha|}^\beta + \sum_{k=1}^\infty 2^{k-1}c_{k + |\alpha|}^\beta,
\end{equation*}
where in the last equality we used the change of index $k = n-|\alpha|$ in the sum. It follows that
\begin{align*}
    \mu_\beta(C_\alpha) &= \frac{e^\beta-2}{e^\beta-1}\left[\frac{1}{e^{|\alpha|\beta}} + \sum_{k=1}^\infty \frac{2^{k-1}}{e^{(k + |\alpha|)\beta}}\right] =  e^{-|\alpha|\beta}.
\end{align*}
By taking the limit, we get
\begin{align*}
    \lim_{\beta \to \log 2}\mu_\beta(C_\alpha) &= 2^{-|\alpha|}.
\end{align*}
Before we proceed the calculation of the limit above for the remaining elements of $\mathcal{B}$, we will show that $\lim_{\beta \to \log 2}\mu_\beta(F)=0$, for every finite $F\subset Y_A$, which implies that the limit for the basic elements can be evaluated just by ignoring the part which is not a union of generalized cylinders and that if there exists the limit measure, then in fact it lives only in $\Sigma_A$.

As we discussed in subsection \ref{subsec:cylinder_topology_Renewal}, any $\xi \in Y_A$ is in the form $\alpha \xi^0$, $\alpha$ positive admissible word ending with $1$. 
We calculate the following limit,
\begin{align*}
   \lim_{\beta \to \log 2}\mu_\beta(\{\alpha\xi^0\})=\lim_{\beta \to \log 2} \frac{e^\beta-2}{e^{|\alpha|\beta}(e^\beta-1)}=0.
\end{align*}
This implies that for finite $F\subset Y_A$, $\lim_{\beta\to \log 2}\mu(F)=0$ and since every element of the basis of $X_A$ of the renewal shift is in the form
\begin{equation*}
    F\sqcup \bigsqcup_{p \in \mathbb{N}} C_{w(p)}, 
\end{equation*}
where $F \subseteq Y_A$ is finite and $w(p)$ is some positive admissible word ending in `$1$' we have that
\begin{align}\label{eq:limit_mu_beta_complement}
    \lim_{\beta \to \log 2}\mu_\beta\left(F\sqcup \bigsqcup_{p \in \mathbb{N}} C_{w(p)}\right) &= \lim_{\beta \to \log 2}\mu_\beta \left(\bigsqcup_{p \in \mathbb{N}} C_{w(p)}\right)\nonumber \\
    &= \lim_{\beta \to \log 2}\sum_{p\in \mathbb{N}}\mu_\beta \left(C_{w(p)}\right) = \sum_{p\in \mathbb{N}}\lim_{\beta \to \log 2}\mu_\beta \left(C_{w(p)}\right).
\end{align}

Now we will compare this result with the measure $\nu$ of \eqref{eq:Sarig_measure_potential_1} from corollary \ref{cor:phase_transition_conformal_potential_1}. Given $\alpha$ positive admissible word ending in `$1$', we have
\begin{align*}
    2^{|\alpha|}\nu(C_\alpha) &= 2^{|\alpha|}\nu([\alpha]) =  \sum_{a \in \mathbb{N}}\int_{[a]} 1 d\nu(x) = \nu(\Sigma_A) = 1.
\end{align*}
Therefore,
\begin{equation*}
    \nu(C_\alpha) = 2^{-|\alpha|}
\end{equation*}
and finally
\begin{equation*}
    \nu(C_\alpha) = \lim_{\beta \to \log 2}\mu_\beta(C_\alpha)
\end{equation*}
By the last equality, the $\sigma$-additivity of measures and the identity \eqref{eq:limit_mu_beta_complement}, we obtain for every basic open set $B$, which is in the form $F\sqcup \bigsqcup_{p \in \mathbb{N}} C_{w(p)}$ (see subsection \ref{subsec:cylinder_topology_Renewal}) that
\begin{equation*}
    \nu\left(F\sqcup \bigsqcup_{p \in \mathbb{N}} C_{w(p)}\right) = \lim_{\beta \to \log 2}\mu_\beta\left(F\sqcup \bigsqcup_{p \in \mathbb{N}} C_{w(p)}\right),
\end{equation*}
that is,
\begin{equation*}
    \nu(B) = \lim_{\beta \to \log 2}\mu_\beta(B).
\end{equation*}
Now we prove the weak$^*$ convergence of the conformal measures of the corollary \ref{cor:phase_transition_conformal_potential_1}.

\begin{theorem} The net $\{\mu_\beta\}_{\beta> \log 2}$ converges to $\nu$ in the weak$^*$ topology for $\beta \to \log 2$.
\end{theorem}
\begin{proof} We know that the basis is closed under finite intersections, and that the net of the statement converges numerically to $\nu$ for every basic set. Since the measures $\mu_\beta$, $\beta > \beta_c$, and $\nu$ are defined on the Borel $\sigma$-algebra of a metric space, the hypotheses of Theorem \ref{thm:convergence_measures_Bogachev} are satisfied and therefore the weak$^*$ convergence holds.
\end{proof}

\section{Empty stems and extremal conformal measures on $Y_A$}

As in the previous section, now we shall study conformal measures which live in $Y_A$, but for shifts that have more than one configuration with empty stem. The reader will notice some similarities with the renewal shift example. Due to corollary \ref{thm:uniqueness_conformal_probabilities_Y_A_families}, we have at most a bijection between the number of extremal conformal probability measures living on $Y_A$ and configurations with empty stem. However, for generalized shifts as the renewal or with similar construction, as the pair and prime renewal shifts, we observe that, for lower temperatures and potentials bounded away from zero, we have in fact a bijection between the extremal conformal measures living on $Y_A$ and the family of empty stem configurations, see Theorems \ref{thm:pair_renewal_general_phase_transition_conditions} and \ref{theorem:infinite_conformal_measures}.

\subsection{Phase Transition on Pair Renewal shift}

For the pair renewal shift and general potential, we have the following result.

\begin{theorem}\label{thm:pair_renewal_general_phase_transition_conditions} For the generalized pair renewal shift $X_A$, let $F:U \to \mathbb{R}$ be a potential. We have the following:
\begin{itemize}
    \item[$(i)$] If $\inf F > 0$, for $\beta > \frac{\log(1+\sqrt{2})}{\inf F}$, there exist two extremal $e^{\beta F}$-conformal probability measures living on $Y_A$;
    \item[$(ii)$] If $0 \leq \sup F < \infty$, for $\beta \leq \frac{\log(1+\sqrt{2})}{\sup F}$, there are no $e^{\beta F}$-conformal probability measures living on $Y_A$.
\end{itemize}
\end{theorem}

\begin{proof} For each $k \in \{1,2\}$ we denote by $c_{\omega,k}$ the coefficient $c_\omega = \mu(\{\omega\})$, $\omega \in \mathfrak{R}_k$, for a measure living on $Y_A(\xi^{0,k})$. The Denker-Urba\'nski conformality condition for extremal measures is
\begin{equation}\label{eq:conformality_DU_PairRS_general}
    c_{\omega,k} = c_{e,k} e^{-\left(\beta F_{|\omega|}\right)(\omega)}, \quad \omega \in \mathfrak{R}_k. 
\end{equation}
In order to make the collection of numbers $\{c_{\omega,k}\}_{\omega \in \mathfrak{R}_k}$ be a probability we must have
\begin{equation}\label{eq:c_omega_probability_Pair_RS_general}
    1 = \sum_{\omega \in \mathfrak{R}_k}c_{\omega,k} = c_{e,k}\left(1+ \sum_{n \in \mathbb{N}}\sum_{\substack{\omega \in \mathfrak{R}_k \\ |\omega| = n}}e^{-\beta F_n(\omega)}\right).
\end{equation}
By freedom of choice on $c_{e,k}$, the condition above is equivalent to state that the right series above converges. If $\inf F > 0$, then
\begin{align}\label{eq:c_omega_probability_Pair_RS_general_upper_bound}
    \sum_{n \in \mathbb{N}} \sum_{\substack{\omega \in \mathfrak{R}_k \\ |\omega| = n}}e^{-\beta F_n(\omega)} \leq \sum_{n \in \mathbb{N}} |\sigma^{-n}(\xi^{0,k})|e^{- \beta n \inf F}
\end{align}
Theorem \ref{thm:counting_configurations_pair_renewal_shift} gives $|\sigma^{-n}(\xi^{0,2})| = |\sigma^{-(n-1)}(\xi^{0,1})|$ for $n \in \mathbb{N}$, and the upper bound in RHS of \eqref{eq:c_omega_probability_Pair_RS_general_upper_bound}, for any $k \in \{1,2\}$, converges if and only if the series
\begin{equation}\label{eq:L_upper_bound}
    \mathfrak{L}^{\text{upper}}(\beta) = \sum_{n \in \mathbb{N}}\left[(1-\sqrt{2})^{n} + (1-\sqrt{2})^{n+1} + (1+\sqrt{2})^{n} + (1+\sqrt{2})^{n+1}\right] e^{-\beta n \inf F}
\end{equation}
converges, where in the series above we constructed by using
\begin{equation*}
    4|\sigma^{-n}(\xi^{0,1})| = (1-\sqrt{2})^{n} + (1-\sqrt{2})^{n+1} + (1+\sqrt{2})^{n} + (1+\sqrt{2})^{n+1}, \quad n \in \mathbb{N}_0.
\end{equation*}
The ratio test for $\mathfrak{L}^{\text{upper}}(\beta)$ is calculated next:
\begin{align*}
    \lim_n \frac{e^{-\beta (n+1) \inf F}\left[(1-\sqrt{2})^{n+1} + (1-\sqrt{2})^{n+2} + (1+\sqrt{2})^{n+1} + (1+\sqrt{2})^{n+2}\right]}{e^{-\beta n \inf F}\left[(1-\sqrt{2})^{n} + (1-\sqrt{2})^{n+1} + (1+\sqrt{2})^{n} + (1+\sqrt{2})^{n+1}\right]} = e^{-\beta \inf F + \log(1+\sqrt{2})}.
\end{align*}
Therefore, the series \eqref{eq:L_upper_bound} converges if $\beta > \frac{\log(1+\sqrt{2})}{\inf F}$ and in this case there exist two extremal $e^{\beta F}$-conformal measures living on $Y_A$.

On the other hand, if $0 \leq \sup F < \infty$, we have
\begin{align}\label{eq:c_omega_probability_Pair_RS_general_lower_bound}
    \sum_{n \in \mathbb{N}} \sum_{\substack{\omega \in \mathfrak{R}_k \\ |\omega| = n}}e^{-\beta F_n(\omega)} \geq \sum_{n \in \mathbb{N}} |\sigma^{-n}(\xi^{0,k})|e^{- \beta n \sup F},
\end{align}
and analogously to the proof of $(i)$, the lower bound in RHS of \eqref{eq:c_omega_probability_Pair_RS_general_lower_bound} diverges if and only if the series
\begin{equation}\label{eq:L_lower_bound}
    \mathfrak{L}^{\text{lower}}(\beta) = \sum_{n \in \mathbb{N}}\left[(1-\sqrt{2})^{n} + (1-\sqrt{2})^{n+1} + (1+\sqrt{2})^{n} + (1+\sqrt{2})^{n+1}\right] e^{-\beta n \sup F}
\end{equation}
diverges. We calculate the ratio test of the series above:
\begin{align*}
    \lim_n \frac{e^{-\beta (n+1) \sup F}\left[(1-\sqrt{2})^{n+1} + (1-\sqrt{2})^{n+2} + (1+\sqrt{2})^{n+1} + (1+\sqrt{2})^{n+2}\right]}{e^{-\beta n \sup F}\left[(1-\sqrt{2})^{n} + (1-\sqrt{2})^{n+1} + (1+\sqrt{2})^{n} + (1+\sqrt{2})^{n+1}\right]} = e^{-\beta \sup F + \log(1+\sqrt{2})}.
\end{align*}
Therefore, the series diverges for $\beta > \frac{\log(1+\sqrt{2})}{\sup F}$. For $\beta = \frac{\log(1+\sqrt{2})}{\sup F}$ the series $\mathfrak{L}^{\text{lower}}(\beta)$ also diverges because in this case we have
\begin{align*}
    \mathfrak{L}^{\text{lower}}\left(\frac{\log(1+\sqrt{2})}{\sup F}\right) &= \sum_{n \in \mathbb{N}}\left[\frac{(1-\sqrt{2})^{n} + (1-\sqrt{2})^{n+1} + (1+\sqrt{2})^{n} + (1+\sqrt{2})^{n+1}}{(1+\sqrt{2})^n}\right]\\
    &= \sum_{n \in \mathbb{N}}\left[\frac{(-1)^n}{(1+\sqrt{2})^{2n}} + \frac{(-1)^{n+1}}{(1+\sqrt{2})^{2n+1}} +2 + \sqrt{2}\right] = \infty.
\end{align*}
Therefore, for $\beta \leq \frac{\log(1+\sqrt{2})}{\sup F}$, we have the absence of extremal $e^{\beta F}$-measures living on $Y_A$, and therefore there are not $e^{\beta F}$-measures living on $Y_A$.
\end{proof}

For the pair renewal shift and potential $F = 1$ we have the following result.

\begin{theorem}\label{thm:existence_conformal_measures_Pair_renewal} For the pair renewal shift and constant potential $F=1$, there exists a critical value $\beta_c = \log(1+\sqrt{2})$ s.t.
\begin{itemize}
    \item[$(i)$] for $\beta > \beta_c$ there exist two extremal $e^\beta$-conformal probabilities living on $Y_A$, each one living on a distinct $Y_A$-family;
    \item[$(ii)$] for $\beta = \beta_c$ there exists a unique $e^\beta$-conformal probability living on $\Sigma_A$;
    \item[$(iii)$] for $\beta < \beta_c$, there are no $e^\beta$-conformal probabilities.
\end{itemize}
\end{theorem}

\begin{proof} 
By Theorem \ref{thm:pair_renewal_general_phase_transition_conditions} for $F = 1$, we have $\sup F = \inf F = 1$ and therefore there exist two extremal $e^\beta$-conformal probabilities living on $Y_A$ for $\beta > \log(1+\sqrt{2})$, and for $\beta \leq \log(1+\sqrt{2})$ we have the absence of these measures. Now, we analyze the existence of these probabilities, but living on $\Sigma_A$. For every $n \in \mathbb{N}$, any $e^\beta$-conformal probability measure that lives on $\Sigma_A$ necessarily satisfies
\begin{equation}\label{eq:conformality_Pair_renewal}
 \mu_\beta(\sigma([n]))=\int_{[n]} e^\beta d\mu_\beta = e^\beta \mu_\beta([n]), \quad n \in \mathbb{N}.
\end{equation}
We also have that,
\begin{equation*}
    \sigma([1])=\Sigma_A, \quad \sigma([2])=[1]\sqcup \bigsqcup_{n \in \mathbb{N}} [2n] \quad \text{and} \quad \sigma([i])=[i-1], \quad i \neq 1.
\end{equation*}
By \eqref{eq:conformality_Pair_renewal}, we have
\begin{align}\label{eq:eigenmeasure_Sigma_A_Pair_Renewal_cylinder_1}
    1 = e^{\beta}\mu_\beta([1]) &\implies \mu_\beta([1])=e^{-\beta},
\end{align}
and we claim that
\begin{align} \label{eq:eigenmeasure_Sigma_A_Pair_Renewal_cylinder_greater_than_2}
    \mu_\beta([n])=e^{-\beta(n-2)}\mu_\beta([2]), \quad n > 2.
\end{align}
In fact, the result is straightforward by using the conformality equation \eqref{eq:conformality_Pair_renewal} for $n = 3$. Now, suppose that \eqref{eq:eigenmeasure_Sigma_A_Pair_Renewal_cylinder_greater_than_2} holds for some $n\in \mathbb{N}\setminus\{1,2\}$, again by \eqref{eq:conformality_Pair_renewal}, we have
\begin{align*}
    \mu_\beta([n]) = \mu_\beta(\sigma([n+1]))= e^{\beta} \mu_\beta([n+1])
\end{align*}
and then
\begin{align*}
    \mu_\beta([n+1]) = e^{-\beta} \mu_\beta([n]) = e^{-\beta} e^{-\beta(n-2)}\mu_\beta([2]) = e^{-\beta(n-1)}\mu_\beta([2])
\end{align*}
and the claim is proved. Now, for the cylinder $[2]$, we have by \eqref{eq:eigenmeasure_Sigma_A_Pair_Renewal_cylinder_1}, the claim above and once more equation \eqref{eq:conformality_Pair_renewal}, that
\begin{equation}\label{eq:eigenmeasure_Sigma_A_Pair_Renewal_cylinder_2}
    e^\beta \mu_\beta([2]) = e^{-\beta}+\mu_\beta([2])+\sum_{n=2}^\infty e^{-2\beta (n-1)}\mu_\beta([2]). 
\end{equation}
Then,
\begin{equation}\label{eq:eigenmeasure_Sigma_A_Renewal_cylinder_2_semifinal}
    \left(e^\beta-\sum_{n=0}^\infty e^{-2\beta n}\right)\mu_\beta([2]) = \frac{2 \sinh(\beta) - 1}{1-e^{-2\beta}} \mu_\beta([2])=  e^{-\beta}
\end{equation}
The number multiplying $\mu_\beta([2])$ in the identity above is zero if, and only if, $\sinh(\beta) = \frac{1}{2}$, and for such $\beta$ it is straightforward that there is not a probability conformal measure, and that $\beta \neq \log(1+\sqrt{2})$. Consider from now on the remaining case that $\sinh(\beta) \neq \frac{1}{2}$. By imposing that $\mu_\beta$ is a probability, that is,
\begin{equation*}
    e^{-\beta}+\mu_\beta([2])+ \mu_\beta([2]) \sum_{n=1}^\infty e^{-n\beta} = \frac{1-e^{-\beta}}{e^{\beta}+e^{-2\beta}-2} = 1,  
\end{equation*}
we turn the equation above, by the substitution $y = e^{-\beta}$, into
\begin{equation*}
    1-y = \dfrac{1}{1-y} \dfrac{y}{y^{-1}-\dfrac{1}{1-y^2}},
\end{equation*}
for $y > 0$ and $y \neq 1$. The equation above becomes,
\begin{equation*}
    -(y^2+y-1)(1-y)^2 = y^2(1-y^2),
\end{equation*}
and since $y \neq 1$, we may divide both sides by $(1-y)$ and we obtain
\begin{equation*}
    y^2+2y-1 = 0,
\end{equation*}
which the roots are $-1 \pm \sqrt{2}$. Since $y > 0$, we must have
\begin{equation*}
    y = -1+\sqrt{2} = \dfrac{1}{1+\sqrt{2}},
\end{equation*}
and therefore the measure $\mu_\beta$ is a probability if and only if $\beta = \log(1+\sqrt{2}) = \beta_c$. In this case, the measure $\mu_{\beta_c}$ satisfies
\begin{align*}
   \mu_{\beta_c}([1]) &= e^{-\beta_c},\\
   \mu_{\beta_c}([2]) &=  e^{-\beta_c}\frac{1-e^{-2\beta_c}}{2 \sinh(\beta_c) - 1},\\
   \mu_{\beta_c}([n]) &= e^{-\beta_c(n-1)}\frac{1-e^{-2\beta_c}}{2 \sinh(\beta_c) - 1}, \quad n>2.
\end{align*}
Its uniqueness is granted because, for every cylinder $[\alpha]$, $|\alpha|>1$, we may apply the conformality equation \eqref{eq:conformality_Pair_renewal} finite times, and therefore the measure for every cylinder $[\alpha]$ depends only on the cylinders of length one. 

Summarizing the results, there exist two conformal probability measures living on $Y_A$ if and only if $\beta > \beta_c$, and there exists a conformal probability measure living on $\Sigma_A$ if and only if $\beta = \beta_c$. By Corollary \ref{cor:decomposition_eigenmeasures}, every probability conformal measure on $X_A$ can be written as a sum of its normalized restrictions on $\Sigma_A$ and $Y_A$, there are no conformal probability measures on $X_A$ for $\beta < \beta_c$, and we conclude that the statement holds. 
\end{proof}

\begin{remark}\label{remark:probability_eigenmeasures_Y_A_pair_renewal_on_cylinders} We observe that a similar characterization for the conformal probability living on $\Sigma_A$ by its value on length one cylinders is possible on each $Y_A$-family when $\beta > \beta_c$, with the difference that we must to consider the measure on the singletons with on the respective empty stems. We have
\begin{equation*}
    \sigma(C_n \cap Y_A(\xi^{0,k})) = \sigma(C_n) \cap Y_A(\xi^{0,k}), \quad n \in \mathbb{N}, \quad k \in \{1,2\},
\end{equation*}
and then,
\begin{equation*}
    \sigma(C_n)\cap Y_A(\xi^{0,k}) = \begin{cases}
                                        Y_A(\xi^{0,k}), \quad \text{if } n = 1;\\
                                        Y_A(\xi^{0,k}) \cap \left(C_1 \sqcup \bigsqcup_{n \in \mathbb{N}} C_{2n}\right), \quad \text{if } n = 2;\\
                                        Y_A(\xi^{0,k}) \cap C_{n-1}, \quad \text{otherwise},
                                     \end{cases}
\end{equation*}
where $k \in \{1,2\}$. The conformality equation \eqref{eq:conformality_Pair_renewal} is analogous for generalized cylinders. We denote by $\mu_{\beta,k}$ the extremal conformal probability living on  $Y_A(\xi^{0,k})$, $k \in \{1,2\}$. We have
\begin{align*}
    \mu_{\beta,k}(C_1) &= e^{-\beta},\\
    \mu_{\beta,k}(C_n) &= e^{-\beta(n-2)}\mu_{\beta,k}(C_2), \quad n>2.
\end{align*}
And since these measures are probabilities we must have
\begin{align*}
    e^\beta \mu_{\beta,k}(C_2)=e^{-\beta}+\mu_{\beta,k}(C_2)+\mu_{\beta,k}(\{\xi^{0,1}\})+\sum_{n=2}^\infty e^{-2\beta (n-1)}\mu_{\beta,k}(C_2).
\end{align*}
By similar series study done for Theorem \ref{thm:pair_renewal_general_phase_transition_conditions} also done for Theorem \ref{thm:existence_conformal_measures_Pair_renewal}, we obtain
\begin{equation*}
    \mu_{\beta,k}(C_2) = \begin{cases}
                           \left(1 + \frac{4 e^\beta}{\mathfrak{L}(\beta)}\right) e^{-\beta}\frac{1-e^{-2\beta}}{2 \sinh(\beta) - 1}, \quad \text{if } k=1;\\
                           \left(1 + \frac{4 e^\beta}{4 + e^{-\beta} \mathfrak{L}(\beta)}\right) e^{-\beta}\frac{1-e^{-2\beta}}{2 \sinh(\beta) - 1}, \quad \text{if } k=2,
                         \end{cases}
\end{equation*}
where
\begin{equation*}
    \mathfrak{L}(\beta) := \sum_{n \in \mathbb{N}_0}\left[(1-\sqrt{2})^{n} + (1-\sqrt{2})^{n+1} + (1+\sqrt{2})^{n} + (1+\sqrt{2})^{n+1}\right] e^{-\beta n}.
\end{equation*}
\end{remark}

The figure \ref{fig:phase_transition_pair_renewal_potential_1} compares the standard formalism of $\Sigma_A$ to the generalized one of $X_A$ for the pair renewal shift and potential $F \equiv 1$.

\begin{figure}[h!]
  \hspace{-.5cm}
 \includegraphics[scale=.3]{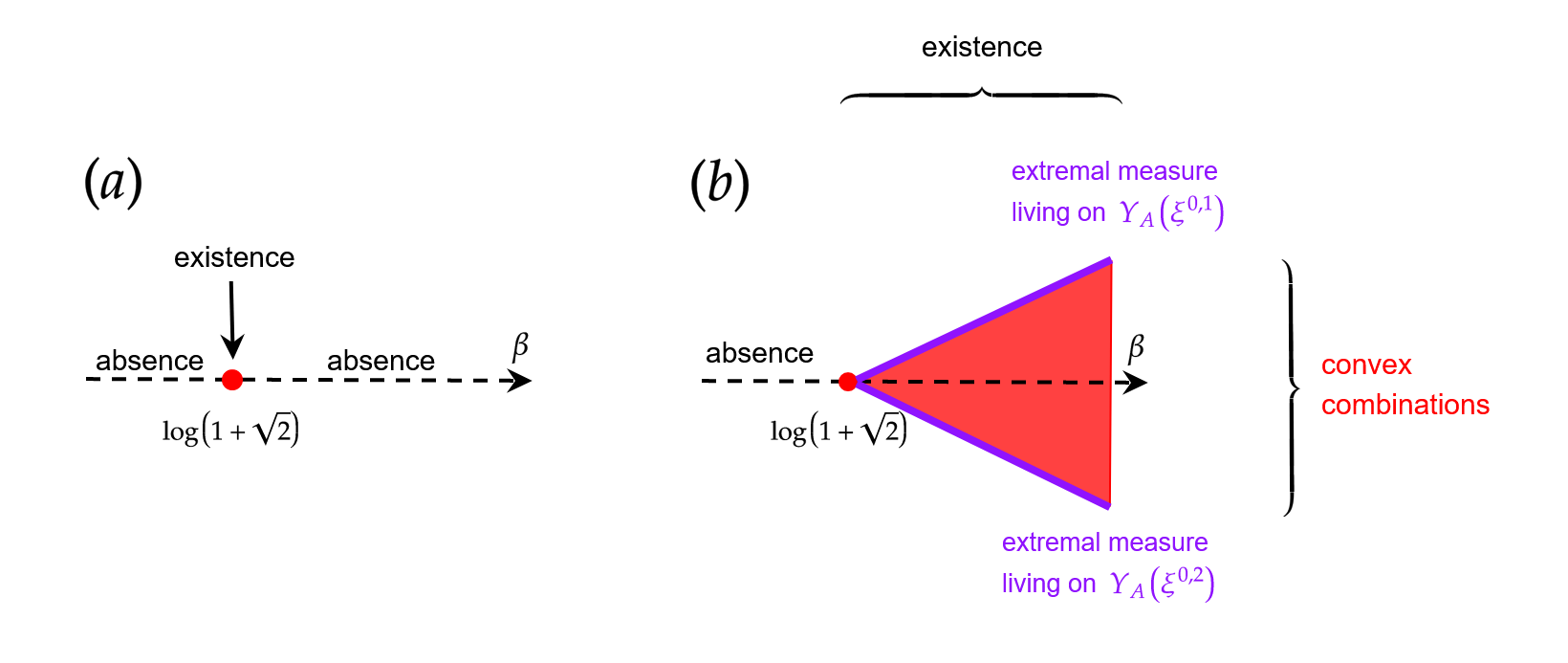}
 
 \caption{The phase transitions on different thermodynamic formalisms for conformal probabilities in the pair renewal shift space with potential $F=1$. The picture $(a)$ represents the standard formalism on $\Sigma_A$, where we have a unique conformal probability for a unique possible inverse of temperature, namely $\beta = \log(1+ \sqrt{2})$ (red dot). The picture $(b)$ represents the generalized formalism on $X_A$ and, unlike in $(a)$, we are able to detect not only the previous standard conformal probability, but also, for $\beta > \log(1+ \sqrt{2})$, two extremal conformal probabilities (purple lines), each one living on a different $Y_A$-family. In this case, by convex combinations (red triangle) we have infinitely many conformal measures for each $\beta > \log(1+ \sqrt{2})$.  \label{fig:phase_transition_pair_renewal_potential_1}}
\end{figure}

Still for the case where $F\equiv 1$, we shall connect the $e^\beta$-conformal measures living on $Y_A$ ($\beta > \log(1+\sqrt{2}$) to the unique one living on $\Sigma_A$ ($\beta = \log(1+\sqrt{2})$) next. Consider the set $\Phi_\beta$ of $e^\beta$-conformal probabilities living on $Y_A$ for the pair renewal shift and potential $F(x) = 1$ as above, which is the set of convex combination of the extremal probabilities $\mu_{\beta,1}$ and $\mu_{\beta,2}$. Also we consider the Hausdorff distance (see \cite{Munkres2000}) between $\Phi_\beta$ and $\mu_{\beta_c}$ on the probability space for $X_A$, which is given by
\begin{equation*}
    d_H(\mu_{\beta_c},\Phi_\beta)=\max \left\{\,\inf _{y\in \Phi_\beta}d(\mu_{\beta_c},y),\,\sup _{y\in \Phi_\beta}d(\mu_{\beta_c},y)\,\right\}= \sup _{y\in \Phi_\beta}d(\mu_{\beta_c},y),
\end{equation*}
where $d$ is a metric compatible with the weak$^*$-topology.

\begin{theorem} $d_H(\mu_{\beta_c},\Phi_\beta) \to 0$ for $\beta \to \beta_c$ by above.
\end{theorem}

\begin{proof} Since for $k \in \{0,1\}$ we have
\begin{equation*}
    \mu_{\beta,k}(\{\xi^{0,k}\}) = \mu_\beta(\{\xi^{0,1}\})=\frac{4}{\mathcal{L}(\beta)}\to 0
\end{equation*}
when $\beta \to \beta_c$ by above, then for every finite set $F \subseteq Y_A$ it holds that
\begin{equation*}
    \mu_{\beta,k}(F) \to 0.
\end{equation*}
Also, it follows by Remark \ref{remark:probability_eigenmeasures_Y_A_pair_renewal_on_cylinders}, the continuity of $H$ and the fact tha the generalized cylinders are clopen that
\begin{align*}
    \lim_{\beta \to \beta^c} \mu_{\beta,1}(C_2) &=  \lim_{\beta \to \beta^c} \left(1 + \frac{4 e^\beta}{\mathfrak{L}(\beta)}\right) H(\beta) = H(\beta_c) = \mu_{\beta_c}([2]),\\
    \lim_{\beta \to \beta^c} \mu_{\beta,2}(C_2) &=  \lim_{\beta \to \beta^c} \left(1 + \frac{4 e^\beta}{4 + e^{-\beta} \mathfrak{L}(\beta)}\right) H(\beta) = H(\beta_c) = \mu_{\beta_c}([2]).
\end{align*}
That is, $\lim_{\beta \to \beta^c} \mu_{\beta,k}(C_2) = \mu_{\beta_c}([2])$, $k \in\{1,2\}$. Consequently,
\begin{equation*}
    \lim_{\beta \to \beta^c}\mu_{\beta,k}(C_n) = \mu_{\beta_c}([n]), \quad n \geq 2, 
\end{equation*}
and for $C_1$ the limit above is straightforward, and by \eqref{eq:conformality_Pair_renewal} the limit above also holds for every generalized cylinder set on positive words. i.e.,
\begin{equation*}
    \lim_{\beta \to \beta^c}\mu_{\beta,k}(C_{\alpha}) = \mu_{\beta_c}([\alpha]), \quad \alpha \text{ positive admissible word}. 
\end{equation*}
We recall that every basic set generated by the subbasis of cylinders and their complements of the generalized pair renewal shift $X_A$ is a countable union of disjoint positive generalized cylinders jointly with a finite subset of $Y_A$ as in \eqref{eq:basis_pair_renewal_is_finite_set_union_countable_cylinders}. Also, the basis of finite intersections of these cylinders and complements of cylinders is, by definition, closed under finite intersections. So, by Theorem \ref{thm:convergence_measures_Bogachev}, both extremal measures $\mu_{\beta,1}$ and $\mu_{\beta,2}$ converges to $\mu_{\beta_c}$ on the weak$^*$ topology. We claim that $\Phi_\beta$ is closed. In fact, Let $\{\eta_n\}_{\mathbb{N}}$ be a sequence in $\Phi_\beta$ converging weakly$^*$ to a Borel probability $\nu$. Since $U$ is open, $C_c(U)$ can be seen as a subspace of $C_c(X_A)$, so by Theorem \ref{thm:equivalences_conformal_measures_generalized_Markov_shift} we have
\begin{equation*}
    \int_U fd\eta=\lim \int_U f d\eta_n=\lim \int_{X_A} L_{-\beta F}(f)d\eta_n=\int_{X_A}L_{-\beta F}(f) d\eta,
\end{equation*}
and hence $\eta \in \Phi_\beta$, and the claim is proved. Since $X_A$ is a compact metric space, we have that the space of probability Borel measures is compact on the weak$^*$ topology and therefore $\Phi_\beta$ is compact, since it is closed. On the other hand, the function
\begin{equation*}
    \nu \mapsto d(\mu_{\beta_c},\nu)
\end{equation*}
is continuous, and by compactness of $\Phi_\beta$, there exists $\nu_\beta\in  \Phi_\beta $ satisfying
\begin{equation*}
    \sup _{\nu \in \Phi_\beta}d(\mu_{\beta_c},\nu)=d(\mu_{\beta_c},\nu_\beta ).
\end{equation*}
We prove now that $\lim_{\beta\to \beta_c}d(\mu_{\beta_c},\eta_\beta)=0$. Indeed, given $\epsilon>0$ there exists $\delta$, s.t. $\beta-\beta_c<\delta$ implies
\begin{align*}
 \left| \int f d\mu_\beta -\int f d\mu_{\beta_c}\right|<\epsilon \quad \text{and} \quad  \left| \int f d\nu_\beta -\int f d\mu_{\beta_c}\right|<\epsilon 
\end{align*}
for every $f\in C(X_A)$. Since $\eta_\beta=\lambda \mu_{\beta,1}+ (1-\lambda) \mu_{\beta,2} $, for some $\lambda \in [0,1]$, we get
\begin{equation*}
    \left|\int f d \nu_\beta -\int f d\mu_{\beta_c}\right|<\epsilon,
\end{equation*}
hence $\lim_{\beta\to \beta_c}d(\mu_{\beta_c},\nu_\beta)=0$, and therefore $\lim_{\beta\to \beta_c}d_H(\mu_{\beta_c},\Phi_\beta)=0$.
\end{proof}

\section{Phase Transition on Prime Renewal shift}

Now we study the conformal measures on the prime renewal shift. We will see that the existence of infinitely many empty configurations can lead to the existence of infinite extremal conformal measures living on $Y_A$. First, we characterize the extremal conformal measures living on $Y_A$ as we realized in the case of the standard renewal shift. Take $p \in \{1\}\cup\{q \in \mathbb{N}: q \text{ is a prime number}\}$, and define
\begin{equation}\label{eq:coefficients_Y_A_confomal_prime_renewal_shift}
    c_{\omega,p}:= \mu(\{\omega\xi^{0,p}\}), \quad \omega \in \mathfrak{R}_p
\end{equation}
where $\xi^{0,p}$ is the empty stem configuration such that and $R_{\xi^{0,p}}(e) = \{1,p\}$, as it is in subsection \ref{subsec:Prime_Renewal_shift}. Consider a potential $F:U \to \mathbb{R}$, where $U$ is the open set of elements for which we can apply the shift map. In $Y_A$ we use the following notation $F(\omega,p):=F(\xi(\omega,p))$, $p$ prime or one. As in the case of the standard renewal shift, we define:
\begin{equation*}
    D(\omega,p) = e^{F(\omega,p)}.
\end{equation*}
By Theorem \ref{theorem.coeficientes}, identity \eqref{eq:c_omega_probability} and Corollary \ref{cor:extremal_conforma_Y_A_characterization}, every extremal conformal probability living on $Y_A$ lives in some $Y_A$-family $Y_A(\xi^{0,p})$ and it satisfies
\begin{align}\label{eq:system_conformal_prime_renewal}
    \begin{cases}
        c_{\omega,p} e^{ F_{|\omega|}(\omega,p)} = c_{e}, \quad \omega \in \mathfrak{R}_p \setminus\{e\};\\
        \sum_{\omega \in \mathfrak{R}_p}c_{\omega,p} = 1.    
    \end{cases}
\end{align}

Now, we have an analogous result to the Theorem \ref{theorem.potentialwithonecoordinate} for the case of countably infinite empty configurations.

\begin{theorem}\label{theorem:infinite_conformal_measures} Consider a potential $F:U \to \mathbb{R}$, we have the following:
\begin{itemize}
    \item[$(i)$] Suppose $\inf F>0$. For each $p \in \{1\}\cup\{p \in \mathbb{N}: p \text{ is a prime number}\}$ and $\beta>\frac{\log 3}{\inf F}$, there exists a unique $e^{\beta F}$-conformal probability measure $\mu_{\beta,p}$ that vanishes out of $Y_A(\xi^0(p))$.
    \item[$(ii)$] Suppose $0\leq \sup F < + \infty$ and $\beta\leq \frac{\log 2}{\sup F}$, there are no $e^{\beta F}$-conformal probability measures which vanish in $\Sigma_A$.
\end{itemize}
\end{theorem}

\begin{proof} For each $p \in \{1\}\cup\{q \in \mathbb{N}: q \text{ is a prime number}\}$, the system \eqref{eq:system_conformal_prime_renewal} gives
\begin{equation*}
    1 = \sum_{\omega \in \mathfrak{R}_p}c_{\omega,p} = c_{e,p}\left(1 + \sum_{n \in \mathbb{N}}\sum_{\substack{\omega \in \mathfrak{R}_p\\ |\omega|=n}}e^{ -\beta F_n(\omega,p)}\right),
\end{equation*}
that is 
\begin{equation*}
    1 + \sum_{n \in \mathbb{N}}\sum_{\substack{\omega \in \mathfrak{R}_p\\ |\omega|=n}}e^{ -\beta F_n(\omega,p)} = \frac{1}{c_{e,p}},
\end{equation*}
where $c_{e,p} \in (0,1)$, and then $c_{e,p}^{-1} \in (1, \infty)$, otherwise we have absence of conformal probabilities on $Y_A(\xi^{0,p})$. Since the LHS of the identity above is greater than $1$ and the terms of the series are positive, then there exists a (unique) conformal probability living on $Y_A(\xi^{0,p})$ if and only if the series above converge. If $\inf F > 0$, we have
\begin{align*}
    \sum_{n \in \mathbb{N}}\sum_{\substack{\omega \in \mathfrak{R}_p\\ |\omega|=n}}e^{-\beta F_n(\omega,p)} \leq  \sum_{n \in \mathbb{N}}\sum_{\substack{\omega \in \mathfrak{R}_p\\ |\omega|=n}}e^{-n \beta \inf F} \stackrel{(\bullet)}{\leq} \sum_{n \in \mathbb{N}} 3^n e^{-n \beta \inf F} = \sum_{n \in \mathbb{N}} e^{n \inf F \left(\frac{\log3}{\inf F} -\beta \right)},
\end{align*}
where in $(\bullet)$ we used Proposition \ref{prop:control_number_configurations}. By last equation above, it is straightforward that $\beta > \frac{\log 3}{\inf F}$ implies the existence of a conformal probability living on $Y_A(\xi^{0,p})$. On the other hand, if $0 \leq \sup F < \infty$, then
\begin{align*}
    \sum_{n \in \mathbb{N}}\sum_{\substack{\omega \in \mathfrak{R}_p\\ |\omega|=n}}e^{-\beta F_n(\omega,p)} \geq  \sum_{n \in \mathbb{N}}\sum_{\substack{\omega \in \mathfrak{R}_p\\ |\omega|=n}}e^{-n \beta \sup F} \stackrel{(\bullet)}{\geq} \frac{1}{2} \sum_{n \in \mathbb{N}} 2^n e^{-n \beta \sup F} = \sum_{n \in \mathbb{N}} e^{n \inf F \left(\frac{\log2}{\sup F} -\beta \right)},
\end{align*}
where in $(\bullet)$ we used again Proposition \ref{prop:control_number_configurations}. It is straightforward that we have the absence of conformal probabilities living on $Y_A(\xi^{0,p})$. Since $p$ is arbitrary, the proof holds for every $Y_A$-family of $X_A$.
\end{proof}

\begin{corollary}
Let $F\equiv 1$. Then, we have the following results.
\begin{itemize}
    \item[$(i)$] For $\beta>\log(3)$ , all $e^{\beta}$-conformal probability measures $\mu_\beta$ that vanishes on $\Sigma_A$ can be written as a convex combination of the measures $\mu_{\beta,p}$.
    \item[$(ii)$] For $\beta\leq \log(2)$ there is no $e^\beta$- conformal probability measure that vanishes on $\Sigma_A$.
\end{itemize}
\end{corollary}
\begin{proof}
Apply Theorem \ref{theorem:infinite_conformal_measures} for the constant potential $F\equiv 1$. For $\beta>\log(3)$, if $\mu_\beta$ is a $e^{\beta}$-conformal measure that vanishes on $\Sigma_A$, denote
$$J=\left\{p\in \{1\}\cup \{p\in \mathbb{N}: p \text{ is prime}\}: \mu_\beta(Y_A(\xi^0(p)))\neq 0\right\}$$
and we have for $E\subseteq Y_A$,
$$\mu_\beta(E)=\sum_{p\in J}\mu_\beta(E\cap Y_A(\xi^0(p)))=\sum_{p\in J}\mu_\beta(Y_A(\xi^0(p)))\dfrac{\mu_\beta(E\cap Y_A(\xi^0(p)))}{\mu_\beta(Y_A(\xi^0(p)))}=\sum_{p\in J}\mu_\beta(Y_A(\xi^0(p)))\mu_{\beta,p}(E).$$
Observe that $\mu_{\beta,p}$ are extremal measures.
For $\beta\leq \log(2)$ is a direct consequence of the previous theorem.
\end{proof}


\section{Pressures}\label{sec:pressures}

Our main goal in this section is to present the concept of pressure at a point $x \in X_A$ and compare it with the Gurevich pressure defined in chapter \ref{ch:Markov_shift_space}. For the pressure of a potential $F:U\to \mathbb{R}$ at a point $x$, denoted by $P(F,x)$, we follow \cite{DenYu2015}. This concept is introduced on a more general setting, called an \textit{Iterated Function System}, i.e a pair consisting of a Polish Space $X$ and a family $\mathcal{V}$ of homeomorphisms $v:D(v)\to v(D(v))\subset X$ defined on a closed nonempty subset $D(v)\subset X$. The authors require as well the existence of point which the orbit is infinite, but since our goal is to apply the theorems on the generalized symbolic space $X_A$, which contains $\Sigma_A$, this hypothesis is automatically satisfied.

For us, the polish space $X$ is the space $X_A$ and the family of homeomorphisms are chosen to be $\{(\sigma\vert_{C_i})^{-1}: i\in \mathbb{N}\}$, the inverses of the shift map when restricted to the \textit{generalized  cylinder sets} $C_i := \{\xi \in X_A: \xi_i =1\}$. 

For every finite admissible word $\alpha$, $|\alpha| \leq n$, $n \in \mathbb{N}$, we define the set $W_n^\alpha$ of the words of length $n$ which ends with $\alpha$.

\begin{definition} Given $\beta > 0$ and a potential $F:U \to \mathbb{R}$ and $x \in X_A$, the \emph{n-th patition function at the point }$x$ is defined by
\begin{equation*}
    Z_n(\beta F,x) = \sum_{\sigma^n(y)=x}e^{\beta F_n(y)},
\end{equation*}
where $F_n$ is the Birkhoff sum of $F$. The \emph{pressure at the point }$x$ is
\begin{equation*}
    P(\beta F,x) :=\limsup_{n\to \infty} \frac{1}{n} \log Z_n(\beta F,x).
\end{equation*}
\end{definition}

Until the end of this section, we assume that $X_A$ comes from the Renewal Shift and the potential $F$ depends only on the first coordinate, i.e $F(x)=F(x_0), x\in U$.
\begin{remark}
With the hypothesis above we can write the Gurevich pressure in another way. Let $a=1$, we have
$$ Z_n(\beta F,[1]) = \sum_{x\in \Sigma_A:\sigma^n(x)=x, x_0=1}e^{\beta F_n(x)}=\sum_{x\in \Sigma_A:\sigma^n(x)=x, x_0=1}e^{\beta F_n(\sigma(x))}=\sum_{\alpha\in W_n^1}e^{\beta F_n(\alpha)}$$
We will use that last expression of $Z_n(\beta F,[1])$ for most of our calculations.
\end{remark}
Observe as well that $2|W_n^1| = |W_{n+|\alpha|}^\alpha|=2^n$ for every $\alpha$ finite admissible word and $n \in \mathbb{N}$. For given $x \in Y_A$ and $n \in \mathbb{N}$ we define the maps $J_n: W_n^1 \to W_{n+|x|}^x$ and $T_n: W_n^1 \to W_{n+|x|}^x$ as
\begin{equation}\label{eq:J_n_T_n}
    J_n(\alpha) = \alpha x \quad \text{and} \quad T_n(\alpha):= \sigma^{\alpha_0}(\alpha)(x_0+\alpha_0)\ldots (x_0+1)x.
\end{equation}

\begin{proposition} Given $x \in Y_A$ and $n \in \mathbb{N}$, the following statements are true:
\begin{itemize}
    \item[$(i)$] both $J_n$ and $T_n$ are injective;
    \item[$(ii)$] $J_n(W_n^1) \cap T_n(W_n^1) = \emptyset$;
    \item[$(iii)$] $J_n(W_n^1) \sqcup T_n(W_n^1) = W_{n+|x|}^x$.
\end{itemize}
\end{proposition}

\begin{proof} $(i):$ the injectivity of $J_n$ is straightforward. For $T_n$, we claim that $T_n(W_n^1) = W_{n+|x|}^{(x_0+1)x}$. The inclusion $T_n(W_n^1) \subseteq W_{n+|x|}^{(x_0+1)x}$ comes directly for the defintion of $T_n$. Now, let $\alpha \in  W_{n+|x|}^{(x_0+1)x}$, then $\alpha = \alpha'(x_0+p)(x_0+p-1)\ldots (x_0+1)x$ for some $p \in \mathbb{N}$ and $\alpha'$ is an admissible word which ends with `$1$' or the empty word. Take the least $p$ with such property and note that $p \leq n$. In addition, observe that $p = n$ if and only if $\alpha'$ is the empty word and it is straightforward to notice that $\alpha = T_n(n(n-1)\ldots 1)$. Now, if $p<n$, then $\alpha'$ ends with `$1$' and it is straightforward that $\alpha = T(p(p-1)\ldots 1\alpha')$, and therefore the claim is proved. Since $|W_n^1|=|W_{n+|x|}^{(x_0+1)x}|=2^{n-1}$, we have necessarily that $T_n$ is injective.

$(ii):$ it is a direct consequence form the facts: $T_n(W_n^1) = W_{n+|x|}^{(x_0+1)x}$ and $J_n(W_n^1) = W_{n+|x|}^{1x}$, two disjoint sets by definition.

$(iii):$ since $A(j,x_0) = 1$ if and only if $j \in \{1,x_0 +1\}$, we have $W_{n+|x|}^{(x_0+1)x} \sqcup W_{n+|x|}^{1x} = W_{n+|x|}^{x}$. 
\end{proof}

By the definition of $Z_n(\beta F,x)$ we have
\begin{align*}
    Z_n(\beta F,x) &= \sum_{\sigma^n(y)=x}e^{\beta F_n(y)} = \sum_{\substack{\alpha \in W_{n+|x|}^x}}e^{\beta F_n(\alpha)} = \sum_{\substack{\alpha \in W_{n+|x|}^{1x}}}e^{\beta F_n(\alpha)} + \sum_{\substack{\alpha \in W_{n+|x|}^{(x_0+1)x}}}e^{\beta F_n(\alpha)} \\
    &= \sum_{\substack{\alpha \in W_{n}^{1}}}e^{\beta F_n(\alpha)} + \sum_{\substack{\alpha \in W_{n+|x|}^{(x_0+1)x}}}e^{\beta F_n(\alpha)},
\end{align*}
that is,
\begin{align}\label{eq:Z_veio_decomp_T_J_geral}
    Z_n(\beta F,x) &= Z_n(\beta F,[1]) + \sum_{\substack{\alpha \in W_{n}^{1}}}e^{\beta F_n(T_n\alpha)}.
\end{align}
The identity above implies that $Z_n(\beta F,x) > Z_n(\beta F,[1])$. In the current section we consider the class of potentials $F: X_A\setminus \{\xi^0\} \to \mathbb{R}$, given by $F(x) = g(x_0) - g(x_0 + 1)$, where $g$ is a continuous function. This potential has summable variations, since all of its variations are zero. 

\begin{lemma}\label{lemma:Z_veio_rewritten_general_final}
For every $\alpha \in W_n^1$ and $n \geq 2$ it is true that
\begin{equation*}
    F_n(T_n \alpha) = F_n(\alpha) + g(x_0+1) - g(x_0 + \alpha_0 +1) + g(\alpha_0 + 1) - g(1).
\end{equation*}
\end{lemma}

\begin{proof} First we claim that
\begin{equation*}
    F_n(T_n \alpha) = F_n(\alpha) + \sum_{i=0}^{\alpha_0 -1}\left[F(x_0 + \alpha_0 - i)-F(\alpha_0-i)\right]
\end{equation*}
for every $\alpha \in W_n^1$ and $n \geq 2$. Indeed, by the definition of Birkhoff's sum and $T_n \alpha$
\begin{equation*}
    F_n(T_n \alpha) = \sum_{i=0}^{n-1}F[\sigma^i(\sigma^{\alpha_0}(\alpha)(x_0+ \alpha_0)\ldots(x_0+1)x)].
\end{equation*}
If $\alpha_0 < n$, then
\begin{align*}
    F_n(T_n \alpha) &= \sum_{i=0}^{n-\alpha_0-1} F(\sigma^{i+\alpha_0}(\alpha)(x_0+ \alpha_0)\ldots(x_0+1)x) + \sum_{i=0}^{\alpha_0-1} F((x_0 +\alpha_0 -i)\ldots (x_0+1)x)\\
    &= \sum_{i=0}^{n-\alpha_0-1} F(\alpha_{\alpha_0 + i}) + \sum_{i=0}^{\alpha_0-1} F(x_0 +\alpha_0 -i) = \sum_{i=\alpha_0}^{n-1} F(\alpha_i) + \sum_{i=0}^{\alpha_0-1} F(x_0 +\alpha_0 -i) \\
    &= \sum_{i=0}^{n-1} F(\alpha_i) -\sum_{i=0}^{\alpha_0-1} F(\alpha_i) + \sum_{i=0}^{\alpha_0-1} F(x_0 +\alpha_0 -i) = F_n(\alpha) + \sum_{i=0}^{\alpha_0 -1}\left[F(x_0 + \alpha_0 - i)-F(\alpha_i)\right].
\end{align*}
If $\alpha_0 = n$ the equality obtained above also holds:
\begin{align*}
    F_n(T_n \alpha) &= \sum_{i=0}^{n-1} F((x_0 +n -i)\ldots (x_0+1)x) =  \sum_{i=0}^{n-1} F(x_0 +\alpha_0 -i) \\
    &= F_n(\alpha)-F_n(\alpha) + \sum_{i=0}^{n-1} F(x_0 +\alpha_0 -i) \\
    &= F_n(\alpha) + \sum_{i=0}^{\alpha_0 -1}\left[F(x_0 + \alpha_0 - i)-F(\alpha_i)\right].
\end{align*}
However, for any $\alpha_i$ with $0 \leq i \leq \alpha_0 -1$ we necessarily have $\alpha_i = \alpha_0-i$, and the claim is proved
Now, note that
\begin{equation*}
    \sum_{i=0}^{\alpha_0 - 1} F(x_0 + \alpha_0 - i) = g(x_0+1) - g(x_0 + \alpha_0 +1) \quad \text{and} \quad \sum_{i=0}^{\alpha_0 - 1} F(\alpha_0 - i) = g(1) - g(\alpha_0 + 1),
\end{equation*}
and therefore
\begin{align*}
    F_n(T_n \alpha) &= F_n(\alpha) + \sum_{i=0}^{\alpha_0 -1}\left[F(x_0 + \alpha_0 - i)-F(\alpha_i)\right] \\
    &= F_n(\alpha) + g(x_0+1) - g(x_0 + \alpha_0 +1) + g(\alpha_0 + 1) - g(1).
\end{align*}
\end{proof}

\begin{theorem}\label{teo:equal_pressures_Y_A} Let $x \in Y_A$ and $\beta >0$. Consider a potential $F:U \to \mathbb{R}$, bounded above, such that $F(x) = g(x_0) - g(x_0 + 1)$, then $P(\beta F, x)=P_G(\beta F)$.
\end{theorem}

\begin{proof} First, we compare the fraction
\begin{equation*}
    \frac{Z_n(\beta F,x)}{Z_n(\beta F,[1])}.
\end{equation*}
We have that
\begin{align*}
    1 < \frac{Z_n(\beta F,x)}{Z_n(\beta F,[1])} &\stackrel{\text{\eqref{eq:Z_veio_decomp_T_J_geral}}}{=}
    \frac{Z_n(\beta F,[1]) + \sum_{\substack{\alpha \in W_{n}^{1}}}e^{\beta F_n(T_n\alpha)}}{Z_n(\beta F,[1])} = 1 + \frac{\sum_{\substack{\alpha \in W_{n}^{1}}}e^{\beta F_n(T_n\alpha)}}{Z_n(\beta F,[1])}\\
    &\stackrel{\text{Lemma \ref{lemma:Z_veio_rewritten_general_final}}}{=} 1 + \frac{\sum_{\substack{\alpha \in W_{n}^{1}}}e^{\beta F_n(\alpha)+\beta (g(x_0+1) - g(x_0+\alpha_0+1) + g(\alpha_0 + 1) -g(1))}}{Z_n(\beta F,[1])}.
\end{align*}
Since the potential is bounded above, there exists $M>0$ such that
\begin{equation*}
    F(m) = g(m) - g(m+1) \leq M, \quad m \in \mathbb{N},
\end{equation*}
thus 
\begin{equation*}
    g(\alpha_0 +1) - g(x_0+\alpha_0 +1) \leq x_0 M
\end{equation*}
and we obtain
\begin{align*}
    \frac{Z_n(\beta F,x)}{Z_n(\beta F,[1])} &\leq 1 + e^{\beta [g(x_0+1)-g(1) + x_0 M]}\frac{\sum_{\substack{\alpha \in W_{n}^{1}}}e^{\beta F_n(\alpha)}}{Z_n(\beta F,[1])} \\
    &= 1 + e^{\beta [g(x_0+1)-g(1) + x_0 M]}\frac{Z_n(\beta F,[1])}{Z_n(\beta F,[1])} = 1 + e^{\beta [g(x_0+1)-g(1) + x_0 M]}
\end{align*}
Therefore, we get
\begin{equation}\label{eq:bound_Z_veio_general_final}
Z_n(\beta F,x) \leq  Z_n(\beta F,[1])(1 + e^{\beta [g(x_0+1)-g(1) + x_0 M]}).
\end{equation}
Now, for the pressures, we have that
\begin{align}\label{eq:dif_pressures_1_general_final}
    P(\beta F, x)\leq P_G(\beta F,[1])+ \limsup_{n\to \infty} \frac{1}{n} \log (1 + e^{\beta [g(x_0+1)-g(1) + x_0 M]})=P_G(\beta F,[1]).
\end{align}
Therefore,
\begin{equation*}
    P(\beta F,x)=P_G(\beta F,[1]). \tag*{\qedhere}
\end{equation*}
\end{proof}

\begin{proposition}\label{prop:equal_pressures_Sigma_A}Let $\beta>0$ and $F:U \to \mathbb{R}$ be a potential depending only on the first coordinate. For every $x\in \Sigma_A$ there exists an $\widetilde{x}\in Y_A$ such that $P(\beta F,x)=P(\beta F,\widetilde{x})$
\end{proposition}

\begin{proof} For $x=x_0x_1\cdots x_{i-1}\cdots$, let $i$ be the least positive integer such that $x_{i-1}=1$ and so we can define $\widetilde{x}=x_0x_1\cdots x_{i-1}\xi_0\in Y_A$. Given $y\in \sigma^{-n}(x)$, $y=y_0\cdots y_{n-1}x_0\cdots x_{i-1}\cdots$, we can define $\widetilde{y}:=y_0\cdots y_{n-1}x_0\cdots x_{i-1}\xi_0\in Y_A$. It is clear that $\widetilde{y}\in \sigma^{-n}(\widetilde{x})$ and it is not difficult to see that $y\to \widetilde{y}$ defines a bijection between $\sigma^{-n}(x)$ and $\sigma^{-n}(\widetilde{x})$. More than that, we have, using the fact that $\phi$ depends only on the first coordinate,
$$F_n(y)=\sum_{i=0}^{n-1}F(\sigma^i(y))=\sum_{i=0}^{n-1}F(y_i)=F_n (\widetilde{y}).$$
With these observations, 
$$Z_n(\beta F,x)=\sum_{y\in \sigma^{-n}(x)}e^{\beta F_n(y)} = \sum_{y\in \sigma^{-n}(\widetilde{x})}e^{\beta F_n(\widetilde{y})}=Z_n(\beta F,\widetilde{x}),$$
concluding that $P(\beta F,x)=P(\beta F,\widetilde{x})$. 
\end{proof}

The most important conclusion of this section is a direct consequence of Theorem \ref{teo:equal_pressures_Y_A} and Proposition \ref{prop:equal_pressures_Sigma_A}.

\begin{corollary}\label{cor:equal_pressures} Let $x \in X_A$ and $\beta >0$. Consider a potential $F:U \to \mathbb{R}$, bounded above, such that $F(x) = g(x_0) - g(x_0 + 1)$. Then $P(\beta F, x)=P_G(\beta F)$.
\end{corollary}

As a subclass of examples, $F$ satisfies the hypothesis of the corollary \ref{cor:equal_pressures} if we take $g$ as being an increasing function. In particular, we are interested in the case $g(x) = \log x$ because in this case $F$ is transient for $\beta > \beta_c$, $\beta_c \in (1,2)$, and it is positive recurrent for $\beta < \beta_c$.

\subsection{Existence of eigenmeasures}

At this point, one could ask if it is always possible to grant the existence of a conformal measures for some potential. In order to answer this question, we extend the notion of summable potentials in Definition \ref{def:regularity_potentials} for $U = \Dom \sigma$ as follows: a continous potential $F:U \to \mathbb{R}$ is said to be summable (or exp-summable) when
\begin{equation*}
    \Sigma_{n \in \mathbb{N}} e^{\sup F\vert_{C_n}} < \infty.
\end{equation*}
Similarly, one can also extend the notion of $n$-th variation and summable variations by extending them to the set $U$. For a summable potential with summable variations on $X_A$, its restriction to $\Sigma_A$ is a continuous potential, which has summable variations and it is summable. In this case, it is known by Theorem 1 of \cite{FreireVargas2018} that, for each $\beta > 1$, there exists a unique equilibrium state $\mu_\beta$. Then, by Theorem 1.2 of a paper of Buzzi and Sarig \cite{BuzSa2003}, $\mu_\beta$ necessarily satisfies
\begin{equation*}
    d\mu_\beta = h_\beta d\nu_\beta,
\end{equation*}
where $h_\beta:\Sigma_A \to \mathbb{R}$ is a positive continuous function and $\nu_\beta$ is a Borel measure finite on cyliders with full support satisfying $\int h_\beta d\nu_\beta = 1$, and such that
\begin{equation*}
    L_{\beta F} h_\beta = \lambda_\beta h_\beta \quad \text{and} \quad L_{\beta F}^* \nu_\beta = \lambda_\beta \nu_\beta
\end{equation*}
for $\lambda_\beta = e^{P_G(\beta F)}$. By Theorem \ref{thm:extension_eigenmeasures}, the eigenmeasure $\nu_\beta$ is also an eigenmeasure on $X_A$ for $\beta F$ by extending back the restriction of $F$ to $\Sigma_A$. As an example of potential with such regularity one could choose
\begin{equation*}
    F(\xi) = -\kappa(\xi)_0.
\end{equation*}
The existence of conformal measures for summable potentials was also studied by R. D. Mauldin and M. Urba\'nski in \cite{MaulUr2001}, where they obtained the existence of probability eigenmeasures, which the eigenvalue is the exponential of the pressure. These existence results in \cite{FreireVargas2018,MaulUr2001} above are very general in terms of the nature of the matrix $A$, since the only hypothesis on it is transitivity, while the nature of potentials is somewhat restrictive.

On the other hand, by Theorem 1.1 of \cite{BuzSa2003}, if we look for the set of equilibrium measures associated to a summable variations potential $\beta F$ defined on $\Sigma_A$, the only options are the absence or the existence of an unique equilibrium measure for $\beta F$ for each $\beta$. In other words, the \textit{phase transition} is in the sense of the existence or not of the equilibrium state. Depending on the recurrence nature of the potential, the study of phase transitions, in the sense of existence-absence of eigenmeasures, can be realized by a change of the type recurrent-transient for some critical value $\beta_c$, by combining the Generalized RPF Theorem \cite{Sarig1999,Daon2013} and the Discriminant Theorem \cite{Sarig2001}.
 
Since $\Sigma_A \subseteq X_A$ and in many cases $X_A$ is compact, it is very natural to expect to find more eigenmeasures (they are constructed via subsequences of atomic measures) for the Ruelle transformation on $X_A$ than in the standard symbolic space $\Sigma_A$.  
 
In the last subsection of this chapter, we shall exhibit a concrete example of a potential $F:U \to \mathbb{R}$ such that its restriction to the classical symbolic space presents a phase transition respect to the existence of the eigenmeasure which disappears when we consider the same potential defined on $U\subseteq X_A$. 

In particular for this section, we shall give a fairly general condition on a potential $F:U\to \mathbb{R}$ for the existence of an eigenmeasure $m$ with eigenvalue $\lambda>0$ for the Ruelle transformation $L_F$, in other words,
\[\int f dm= \int L_{-(\log \lambda-F)} f dm \quad  \text{for every $f\in C_c(U)$},\]
by Theorem \ref{thm:equivalences_conformal_measures_generalized_Markov_shift}, $m$ is a $e^{\log\lambda-F}$-conformal measure. Our result for existence is a direct consequence of Theorem 3.6 on \cite{DenYu2015}, but due to some differences on notation and some minor changes on the statement of the theorem, we will be providing a proof below. 

\begin{theorem}[Denker-Yuri]\label{thm:existence_eigenmeasures_Denker_Yuri}
Let $X_A$ be compact and suppose there exists a $x\in X_A$ such that $P(F,x)$ is finite, then there exists an eigenmeasure $m$ for $L_F$ with eigenvalue $e^{P(F,x)}$.
\end{theorem}
\begin{proof}
We shall use Lemma 3.1 of \cite{DenUr1991} to construct our measure. Given a sequence of real numbers $(a_n)$ there exists a sequence of positive numbers $(b_n)_{n\in \mathbb{N}}$ such that 
\begin{equation}\label{eq:Denker_Urbanski_lemma}
    \sum_{n\in \mathbb{N}}b_n\exp{(a_n-ns)}=\left\{\begin{array}{lr}
    <\infty     & s>c \\
     \infty    & s\leq c
    \end{array}\right.
\end{equation}
and $\lim_n b_n/b_{n+1}=1$, where $c:=\limsup_n a_n/n$. As the pressure $P(F,x)$ is finite, $Z_n(F,x)$ is finite for $n$ large enough. Then, without loss of generality, we assume that $Z_n(F,x)$ is finite for every $n \in \mathbb{N}$. By taking $a_n=\log
Z_n(F,x)$, we have $c=P(F,x)$, and we may define
\begin{equation*}
    M(p,x)=\sum_{n\in\mathbb{N}}b_ne^{-np}Z_n(F,x), \quad p > c,
\end{equation*}
and the measures
\begin{equation}
    m(p,x)=M(p,x)^{-1}\sum_{n\in \mathbb{N}}b_n e^{-np}\sum_{\sigma^n(y)=x}e^{F_n(y)}\delta_y,
\end{equation}
where $\delta_y$ is the Dirac measure on $y$. From \eqref{eq:Denker_Urbanski_lemma} we claim that ${\lim_{p\downarrow P(F,x)}M(p,x)=\infty}$. Indeed, for every sequence $(p_k)_\mathbb{N}$ such that $p_k\downarrow c$, let $S(N,k)=\sum_{n=1}^{N}b_n\exp(-np_k)Z_n(F,x)$. Such sequence of two variables is increasing on both of them, by Monotone Convergence Theorem, we have $\lim_N\lim_kS(N,k)=\lim_k\lim_NS(N,k)$, and the claim is proved because $\lim_N\lim_kS(N,k) = \infty$. Now, fix a sequence $p_k\downarrow P(F,x)$ such that $m(p_k,x)\to m$ weakly. We have
\begin{align*}
    \int L_Fg(y) m(p_k,x)(dy)&= M(p_k,x)^{-1}\sum_{n\in \mathbb{N}} b_n e^{-np_k}\sum_{\sigma^n(y)=x}e^{F_n(y)}L_Fg(y)\\
    &=M(p_k,x)^{-1}e^{p_k}\sum_{n=1}^{\infty} b_n e^{-(n+1)p_k}\sum_{\sigma^{n+1}(w)=x}e^{F_{n+1}(w)}g(w),
\end{align*}
for every $g\in C_c(U)$. Define
\begin{equation}\label{eq:notation_Ruelle_Operator}
    L_F^{n+1} g (x) := \sum_{\sigma^{n+1}(w)=x}e^{F_{n+1}(w)}g(w).
\end{equation}
By using that $b_n = b_n+b_{n+1}-b_{n+1}$, one gets
\begin{align*}
    \int L_Fg(y) m(p_k,x)(dy)&= e^{p_k}M(p_k,x)^{-1}\sum_{n=1}^{\infty}(b_n/b_{n+1}-1)b_{n+1}e^{-(n+1)p_k}L^{n+1}_Fg(x) \\
    &+e^{p_k}\int g(y) m(p_k,x)(dy)-b(1)M(p_k,x)^{-1}L_Fg(x).
\end{align*}
Now, by taking $p_k\downarrow P(F,x)$, the RHS of above expression becomes $e^{P(F,x)}\int g(y)dm(y)$. In fact, note that
\begin{align*}
    b(1)M(p_k,x)^{-1}L_Fg(x) &\to 0,\\
    e^{p_k}\int g(y) m(p_k,x)(dy) &\to e^{P(F,x)}\int g(y)dm(y).
\end{align*}
We claim that 
\begin{equation*}
    M(p_k,x)^{-1}\sum_{n\in \mathbb{N}} b_n e^{-np_k}\sum_{\sigma^n(y)=x}e^{F_n(y)}L_Fg(y) \to 0.
\end{equation*}
Indeed, the term above has the upper bound
\begin{equation}\label{eq:upper_bound_Denker_Yuri_proof}
e^{p_k}M(p_k,x)^{-1}\|g\|_{\infty}\sum_{n=1}^{\infty}|b_n/b_{n+1}-1|b_{n+1}e^{-(n+1)p_k}Z_{n+1}(F,x).
\end{equation}
Now, for every $\epsilon>0$, there exists $N\in\mathbb{N}$ s.t. $|b_n/b_{n+1}-1|<\epsilon$ for every $n>N$, then \eqref{eq:upper_bound_Denker_Yuri_proof} is bounded above by
\[e^{p_k}M(p_k,x)^{-1}\|g\|_{\infty}\left(\sum_{n=1}^{N}|b_n/b_{n+1}-1|b_{n+1}e^{-(n+1)p_k}Z_{n+1}(F,x)+\epsilon \sum_{n=N+1}^{\infty}b_{n+1}e^{-(n+1)p_k}Z_{n+1}(F,x)\right)\]
It is straightforward that the first term in the last expression vanishes when $k$ goes to infinity, since it is a finite sum. The remaining term in same expression is less or equal to $\epsilon$ because
\begin{equation*}
    \sum_{n=N+1}^{\infty}b_{n+1}e^{-(n+1)p_k}Z_{n+1}(F,x) \leq M(p_k,x).
\end{equation*}
Since $\epsilon$ is arbitrary the expression goes to zero as $p_k\downarrow P(F,x)$, proving the claim.
\end{proof}

\begin{remark} We emphasize to the reader that \eqref{eq:notation_Ruelle_Operator} was used just as a notation. We did not define it as a transformation on $C_c(U)$. 
\end{remark}

\begin{theorem} For $A$ s.t. $\Sigma_A$ is topologically mixing, let $F:U \to \mathbb{R}$ be a potential s.t. $F\vert_{\Sigma_A}$ satisfies the Walters' condition and take $\beta >0$ s.t. $P_G(\beta F\vert_{\Sigma_A})< \infty$ and $\beta F\vert_{\Sigma_A}$ is positive recurrent. Also suppose that there exists $x \in X_A$ s.t. $P(\beta F,x) = P_G(\beta F\vert_{\Sigma_A})$. If a Denker-Yuri's probability eigenmeasure $m_\beta$ from Theorem \ref{thm:existence_eigenmeasures_Denker_Yuri} lives in $\Sigma_A$. Then, it coincides with the Sarig's eigenmeasure probability for the same eigenvalue $\lambda = e^{P_G(\beta F\vert_{\Sigma_A})}$.
\end{theorem}

\begin{proof} The restriction of $m_\beta$ to $\mathcal{B}_{\Sigma_A}$ is an eigenmeasure associated to the same eigenvalue as in the statement due to Proposition \ref{thm:restriction_eigenmeasures}. Since $m_\beta(\Sigma_A) = 1$, the restriction to $\Sigma_A$ is a probability as well. Since $F\vert_{\Sigma_A}$ satisfies Walters' condition we have, by positive recurrence and the generalized RPF theorem, the existence of the Sarig's eigenmeasure, and the space of the eigenmeasures has dimension $1$. Then, we necessarily have that $m_\beta$ is the Sarig's eigenmeasure.
\end{proof}

. 

\begin{remark} In the conditions of the theorem above, it is straightforward that the existence of the Denker-Yuri's probability eigenmeasure implies the finiteness of the Sarig's eigenmeasure.
\end{remark}

\subsection{Different thermodynamics for the same potential.\newline
Standard versus generalized symbolic spaces: a concrete example}

In this section, we give an explicit example to illustrate the differences which can be found when we use the standard countable Markov shift $\Sigma_A$ and the generalized one $X_A$ as the configuration space for a fixed potential.

\begin{lemma}\label{lemma:condition_for_eigenmeasures_renewal_potential_log} For the generalized renewal shift, consider $\beta >0$ and the potential $F:U \to \mathbb{R}$ given by
\begin{equation*}
    F(\xi) := \log(x_0) - \log(x_0+1),
\end{equation*}
where $x_0$ is the first coordinate of the stem of $\xi$. If a probability measure $\mu_\beta$ is an eigenmeasure for the Ruelle transformation $L_{\beta F}$, with associated eigenvalue $\lambda > 0$, then
\begin{equation}\label{eq:condition_for_eigenmeasures_renewal_potential_log}
    \mu_\beta(C_n) = \frac{1}{\lambda^n}\frac{1}{(n+1)^\beta}
\end{equation}
for every $n \in \mathbb{N}$. In particular,
\begin{equation}\label{eq:eigenmeasure_probability_sum_renewal_potential_log}
    1 = \mu_\beta(\{\xi^0\}) + \sum_{n \in \mathbb{N}} \frac{1}{\lambda^n}\frac{1}{(n+1)^\beta}.
\end{equation}
\end{lemma}

\begin{proof} By Theorem \ref{thm:equivalences_conformal_measures_generalized_Markov_shift}, $\mu_\beta$ is an eigenmeasure as in the statement of this lemma if and only if it is a $\lambda e^{-\beta F}$-conformal measure in the sense of Denker-Urba\'nski, and then
\begin{equation}\label{eq:condition_conformality_renewal_potential_log_lemma_proof}
    \mu_\beta(\sigma(C_n)) = \int_{C_n} \lambda e^{-\beta F} d\mu_\beta,
\end{equation}
for every $n \in \mathbb{N}$. For $n =1$, we have $\sigma(C_1) = X_A$, and since $\mu_\beta$ is a probability, equation \eqref{eq:condition_conformality_renewal_potential_log_lemma_proof} gives
\begin{equation*}
    1 = \mu_\beta(\sigma(C_1)) = \int_{C_1} \lambda e^{-\beta F} d\mu_\beta = \lambda 2^\beta \mu_\beta(C_1),
\end{equation*}
that is,
\begin{equation}\label{eq:measure_on_C_1_renewal_potential_log_lemma_proof}
    \mu_\beta(C_1) = \frac{1}{\lambda}\frac{1}{2^\beta}.
\end{equation}
Now, for $n \neq 1$, then $\sigma(C_n) = C_{n-1}$ and then \eqref{eq:condition_conformality_renewal_potential_log_lemma_proof} implies
\begin{equation*}
    \mu_\beta(\sigma(C_n)) = \mu_\beta(C_{n-1}) = \int_{C_n} \lambda e^{-\beta F} d\mu_\beta = \lambda \left(\frac{n}{n+1}\right)^{-\beta} \mu_\beta(C_n),
\end{equation*}
i.e.,
\begin{equation}\label{eq:measure_on_C_n_renewal_potential_log_lemma_proof}
    \mu_\beta(C_n) = \frac{1}{\lambda} \left(\frac{n}{n+1}\right)^{\beta} \mu_\beta(C_{n-1}).
\end{equation}
Now we prove the validity of \eqref{eq:condition_for_eigenmeasures_renewal_potential_log}. The result is straightforward for $n = 1$ because of \eqref{eq:measure_on_C_1_renewal_potential_log_lemma_proof}. Now, suppose that \eqref{eq:condition_for_eigenmeasures_renewal_potential_log} holds for $C_n$. By equation \eqref{eq:measure_on_C_n_renewal_potential_log_lemma_proof} we have
\begin{align*}
    \mu_\beta(C_{n+1}) = \frac{1}{\lambda} \left(\frac{n+1}{n+2}\right)^{\beta} \mu_\beta(C_n) = \frac{1}{\lambda} \left(\frac{n+1}{n+2}\right)^{\beta} \frac{1}{\lambda^n}\frac{1}{(n+1)^\beta} = \frac{1}{\lambda^{n+1}}\frac{1}{(n+2)^\beta},
\end{align*}
and the equation \eqref{eq:condition_for_eigenmeasures_renewal_potential_log} holds by induction. The identity \eqref{eq:eigenmeasure_probability_sum_renewal_potential_log}.
\end{proof}

\begin{remark} Observe that the existence of an eigenmeasure probability $\mu_\beta$ as in the statement of the lemma above imposes that the series in RHS of \eqref{eq:eigenmeasure_probability_sum_renewal_potential_log} converges. 
\end{remark}

\begin{lemma}\label{lemma:eigenmeasure_C_alpha_formula} For the generalized renewal shift and the same potential as in Lemma \ref{lemma:condition_for_eigenmeasures_renewal_potential_log}, let $\beta >0$ and $\mu_\beta$ be an eigenmeasure for the Ruelle transformation $L_{\beta F}$, with associated eigenvalue $\lambda > 0$. Then, for $\alpha = \alpha_0 \cdots \alpha_{n-1}$, $n > 1$, positive admissible word, we have
\begin{equation}\label{eq:eigenmeasure_C_alpha_formula}
    \mu_\beta(C_\alpha) = \frac{e^{\beta\sum_{k=0}^{n-2}F(\alpha_k)}}{\lambda^{\alpha_{n-1}+(n-1)}} \frac{1}{(\alpha_{n-1}+1)^\beta},
\end{equation}
where $F(p) := F\vert_{C_p} \equiv \log(p) - \log(p+1)$, $p \in \mathbb{N}$
\end{lemma}

\begin{proof} We recall that, if $\mu_\beta$ is an eigenmeasure as in the statement above, then
\begin{equation}\label{eq:conformality_DU_eigenmeasure_recall}
    \mu_\beta(\sigma(C_\alpha)) = \int_{C_\alpha} \lambda e^{-\beta F} d\mu_\beta,
\end{equation}
for every $\alpha$ positive admissible word s.t. $|\alpha| > 1$.
We prove the lemma by induction. For $n = 2$, the identity above becomes
\begin{equation*}
    \mu_\beta(C_{\alpha_1}) = \mu_\beta(\sigma(C_{\alpha_0 \alpha_1})) = \int_{C_{\alpha_0 \alpha_1}} \lambda e^{-\beta F} d\mu_\beta = \lambda e^{-\beta F(\alpha_0)}\mu_\beta(C_{\alpha_0\alpha_1}),
\end{equation*}
that is,
\begin{equation*}
    \mu_\beta(C_{\alpha_0\alpha_1}) = \frac{e^{\beta F(\alpha_0)}}{\lambda} \mu_\beta(C_{\alpha_1}) = \frac{e^{\beta F(\alpha_0)}}{\lambda^{\alpha_1 + 1}} \frac{1}{(\alpha_1+1)^\beta},
\end{equation*}
where in the last equality we used Lemma \ref{lemma:condition_for_eigenmeasures_renewal_potential_log}. Now, suppose that \eqref{eq:eigenmeasure_C_alpha_formula} holds for some $n > 2$ and let $\alpha = \alpha_0 \cdots \alpha_n$ be a positive admissible word. By \eqref{eq:conformality_DU_eigenmeasure_recall} and the inductive step for the word $\alpha_1 \cdots \alpha_n$, we have
\begin{equation*}
    m_\beta(C_\alpha) = \frac{e^{\beta F(\alpha_0)}}{\lambda} m_\beta(C_{\sigma(\alpha)}) = \frac{e^{\beta\sum_{k=0}^{n-1}F(\alpha_k)}}{\lambda^{\alpha_{n}+n}} \frac{1}{(\alpha_n+1)^\beta}.
\end{equation*}
\end{proof}

\begin{theorem}\label{thm:complete_characterization_eigenmeasures_renewal} Let $A$ be the renewal shift transition matrix and $X_A$ its generalized Markov shift space. Consider the potential $F:U \to \mathbb{R}$ given by $$F(x)=\log(x_0)-\log(x_0+1).$$ Then, for every $\beta > 0$, there exists a unique eigenmeasure associated to the eigenvalue $\lambda_\beta = e^{P_G(\beta F)}$. Moreover, there is critical value $\beta_c$, which is the (real) solution for $\zeta(\beta_c) = 2$ such that
\begin{itemize}
    \item[$(i)$] if $\beta >\beta_c$, then the eigenmeasure lives on $Y_A$;
    \item[$(ii)$] if $\beta \leq \beta_c$, then the eigenmeasure lives on $\Sigma_A$.
\end{itemize}
\end{theorem}

\begin{proof} The proof is a summarization of some results we developed and proved in this thesis:
\begin{itemize}
    \item[(1)] The function given by $g(x_0) = \log(x_0)$ is continuous. In Example \ref{exa:renewal_potential_log} we shown that $F$ is bounded above. Hence, by Corollary \ref{cor:equal_pressures}, we conclude that $P(\beta F,x) = P_G(\beta F)$ for every $x \in X_A$. Moreover, since $\sup F< \infty$, we have by direct calculations for the renewal shift space $\Sigma_A$ that $P_G(\beta F) \leq \log 2 + \beta \sup F < \infty$, for every $\beta >0$. 
    \item[(2)] Since $X_A$ is compact (see subsection \ref{subsec:Generalized_Renewal_shift}) we have by (1) and Theorem \ref{thm:existence_eigenmeasures_Denker_Yuri} that there exists a probability eigenmeasure $m_\beta$ of the Ruelle operator for every $\beta >0$, and the associated eigenvalue is $\lambda_\beta = e^{P(\beta F,x)} = e^{P_G(\beta F)}$. Therefore, the first claim of the statement of this theorem is proved, and by Lemma \ref{lemma:condition_for_eigenmeasures_renewal_potential_log} we have
    \begin{equation}\label{eq:Denker_eigenmeasure_probability_sum}
        1 = m_\beta(\{\xi^0\}) + \sum_{n\in \mathbb{N}} \frac{1}{\lambda_\beta^n}\frac{1}{(1+n)^\beta}.
    \end{equation}
    \item[(3)] Again by Example \ref{exa:renewal_potential_log}, there exists a critical value $\beta_c$, which is precisely the positive solution of $\zeta(\beta_c) = 2$ such that the potential is positive recurrent for $0< \beta < \beta_c$ and transient for $\beta > \beta_c$. We study each case as follows.
    \begin{itemize}
        \item[(3.a)] \textbf{Case $0 < \beta < \beta_c$.} By the Generalized RPF Theorem \ref{thm:RPF_Generalized}, since the potential is positive recurrent, there exists an eigenmeasure $\mu_\beta$ living in $\Sigma_A$ and it is necessarily associated to the eigenvalue $e^{P(\beta F)}$. Accordingly to Theorem 30 of \cite{BelBisEndo2020}, since the potential satisfies $\Var_1 F <\infty$, the eigenmeasure is finite, so consider that $\mu_\beta$ is normalized, that is, it is a probability. By Proposition \ref{prop:eigenmeasures_for_positive_recurrent_potentials_dimension_1}, $\mu_\beta$ is the unique probability eigenmeasure on $\Sigma_A$. Theorem \ref{thm:extension_eigenmeasures} $(a)$ gives that $\mu_\beta$ can be seen that the unique eigenmeasure probability on $X_A$ which lives on $\Sigma_A$. Again by Lemma \ref{lemma:condition_for_eigenmeasures_renewal_potential_log}, we obtain
        \begin{equation}\label{eq:Sarig_eigenmeasure_probability_sum}
            1 = \mu_\beta(\{\xi^0\}) + \sum_{n\in \mathbb{N}} \frac{1}{\lambda_\beta^n}\frac{1}{(1+n)^\beta} = \sum_{n\in \mathbb{N}} \frac{1}{\lambda_\beta^n}\frac{1}{(1+n)^\beta},
        \end{equation}
        and we conclude that
        \begin{equation*}
            \sum_{n\in \mathbb{N}} \frac{1}{\lambda_\beta^n}\frac{1}{(1+n)^\beta}.
        \end{equation*}
        In other words, the series above does not depend on the measure, and by \eqref{eq:Denker_eigenmeasure_probability_sum} we conclude that $m_\beta(\xi^0) = 0$, and then $m_\beta(Y_A) = 0$, so $m_\beta$ is an eigenmeasure probability that lives on $\Sigma_A$ and then $m_\beta = \mu_\beta$. Therefore, given $0 < \beta < \beta_c$, $\mu_\beta$ is the unique probability eigenmeasure associated to the eigenvalue $\lambda_\beta = e^{P_G(\beta F)}$, and it lives on $\Sigma_A$.
        \item[(3.b)] \textbf{Case $\beta = \beta_c$.} By Example \ref{exa:renewal_potential_log} we have that $P_G(\beta_c F) = 0$, that is, $\lambda_{\beta_c} = 1$ and then
        \begin{equation}\label{eq:Denker_eigenmeasure_probability_sum_critical_beta}
            1 = m_{\beta_c}(\{\xi^0\}) + \sum_{n\in \mathbb{N}} \frac{1}{(1+n)^{\beta_c}} = m_\beta(\{\xi^0\}) -1  + \sum_{n\in \mathbb{N}} \frac{1}{n^{\beta_c}} = m_\beta(\{\xi^0\}) -1  + \zeta(\beta_c).
        \end{equation}
        Since $\zeta(\beta_c) = 2$ we conclude that $m_\beta(\{\xi^0\}) = 0$. So $\mu_\beta$ lives on $\Sigma_A$. Moreover this measure is unique, since by Lemma \ref{lemma:eigenmeasure_C_alpha_formula} we obtain a unique value for $m_\beta(C_\alpha)$, $|\alpha| \geq 2$, for every generalized cylinder on positive words, namely
        \begin{align*}
            \mu_{\beta_c}([\alpha]) = \mu_{\beta_c}(C_\alpha) = e^{\beta_c \sum_{i=0}^{n-2}F(\alpha_i)}\frac{1}{(\alpha_{n-1}+1)^{\beta_c}},
        \end{align*}
        and this can be extended for the whole space uniquely.
        \item[(3.c)] \textbf{Case $\beta > \beta_c$.} Once more by Example \ref{exa:renewal_potential_log}, we have that $P_G(\beta F) = 0$ for every $\beta > \beta_c$, i.e., $\lambda_{\beta} = 1$ and then
        \begin{equation}\label{eq:Denker_eigenmeasure_probability_sum_beta_greater_than_critical_beta}
            1 = m_\beta(\{\xi^0\}) + \sum_{n\in \mathbb{N}} \frac{1}{(1+n)^{\beta}} =  m_\beta(\{\xi^0\}) -1  + \zeta(\beta).
        \end{equation}
        Since $1<\zeta(\beta) < 2$ for $\beta > \beta_c$, we have $m_\beta(\{\xi^0\}) = 2 - \zeta(\beta) \in (0,1)$, and then necessarily we must have $m_\beta(Y_A) > 0$. We claim that $m_\beta(\Sigma_A) = 0$. In fact, if $m_\beta(\Sigma_A) > 0$. Let $\nu_\beta$ be the restriction of $m_\beta$ to $\mathcal{B}_{\Sigma_A}$. We observe that $\nu_\beta$ is a non-zero eigenmeasure on $\Sigma_A$ with associated eigenvalue $1$, due to Theorem \ref{thm:restriction_eigenmeasures}, and it is finite because $\mu_\beta$ is a probability. So we may take $\nu_\beta$ normalized in order to be a probability, and by Lemma \ref{lemma:condition_for_eigenmeasures_renewal_potential_log} we have
        \begin{equation*}
            1 = \nu_\beta(\{\xi^0\}) + \sum_{n\in \mathbb{N}} \frac{1}{(1+n)^{\beta}} =  -1  + \zeta(\beta),
        \end{equation*}
        and we obtain $\zeta(\beta) = 2$, a contradiction. Therefore $m_\beta$ lives on $Y_A$ and it is unique because $\xi^0$ is the unique with empty stem.
    \end{itemize}
\end{itemize}
\end{proof}

\begin{remark} Since for $\beta > \beta_c$ the eigenmeasure probability $m_\beta$ in Theorem \ref{thm:complete_characterization_eigenmeasures_renewal} is a $e^{-\beta F}$-conformal measure, and it is explicitly determined by \ref{eq:c_omega_new}, that is,  
\begin{equation*}
    m_\beta(\omega) \equiv c_\omega^\beta = e^{\beta F_{|\omega|}(\omega)} c_e = e^{\beta F_{|\omega|}(\omega)} (2-\zeta(\beta)), \quad \omega \in \mathfrak{R}.
\end{equation*}
For example, we have
\begin{equation*}
    c_1^\beta = 2^{-\beta}(2-\zeta(\beta)), \quad c_{11}^\beta = 2^{-2\beta} (2-\zeta(\beta)), \quad c_{21}^\beta = 3^{-\beta} (2-\zeta(\beta)).
\end{equation*}
\end{remark}

The next two results is a straightforward corollary of Theorem \ref{thm:complete_characterization_eigenmeasures_renewal} above.

\begin{proposition}\label{prop:KMS_renewal_potential_log}  Consider the generalized renewal shift space $X_A$ and the potential $F:U \to \mathbb{R}$ given by
\begin{equation*}
    F(x) = \log(x_0) - \log(x_0 + 1).
\end{equation*}
Also, consider the C$^*$-dynamical system $(C^*(\mathcal{G}(X_A, \sigma)),\tau)$, where $\tau = \{\tau_t\}_{t \in \mathbb{R}}$ is the one-parameter group of automorphisms given by
\begin{equation*}
    \tau_t(f)(\gamma) = e^{-it \beta c_F(\gamma)}f(\gamma)
\end{equation*}
and (uniquely) extended to $C^*(\mathcal{G}(X_A, \sigma))$. Then, we have the following:
\begin{itemize}
    \item[$(a)$] for $\beta \geq \beta_c$ there exists a unique KMS$_\beta$ state on $C^*(\mathcal{G}(X_A,\sigma))$;
    \item[$(b)$] for $\beta < \beta_c$ there are not KMS$_\beta$ states on $C^*(\mathcal{G}(X_A,\sigma))$.
\end{itemize}
\end{proposition}

\begin{proof} Since $P_G(\beta F) = 0$ for $\beta \geq \beta_c$, we have by Theorem \ref{thm:complete_characterization_eigenmeasures_renewal} $(i)$, there exists a unique eigenmeasure probability $m_\beta$, and in this case the eigenvalue is $1$. Then, by Theorem \ref{thm:equivalences_conformal_measures_generalized_Markov_shift}, $m_\beta$ is a $e^{\beta c_F}$-quasi-invariant probability measure, and therefore by Remark \ref{remark:KMS_quasi_invariant} there exists a unique KMS$_\beta$-state, given by
\begin{equation*}
    \varphi_\mu(f)=\int_X f(x,0,x)d\mu(x),\quad f\in C_c(\mathcal{G}(X,\sigma)).
\end{equation*}
For $\beta < \beta_c$, Theorem \ref{thm:complete_characterization_eigenmeasures_renewal} $(ii)$, there exists a unique eigenmeasure probability living on $\Sigma_A$, and in this case the eigenvalue is strictly greater than $1$, so there are not $e^{\beta c_F}$-quasi-invariant probability measures, and therefore there are not KMS$_\beta$ states.
\end{proof}

\begin{corollary}\label{cor:crazy_series_renewal_log_potential} Consider the renewal shift space $\Sigma_A$ and the potential $F:\Sigma_A \to \mathbb{R}$ given by
\begin{equation*}
    F(x) = \log(x_0) - \log(x_0 + 1).
\end{equation*}
Then,
\begin{equation}\label{eq:crazy_series_renewal_log_potential}
    \sum_{n \in \mathbb{N}} \frac{1}{e^{n P_G(\beta F)}} \frac{1}{(n+1)^\beta} = \begin{cases}
                                                                                    1, \quad \text{if } \beta < \beta_c;\\
                                                                                    \zeta(\beta) - 1,  \quad \text{if } \beta \geq \beta_c,
                                                                                 \end{cases}
\end{equation}
where $\beta_c$ is the unique positive real number s.t. $\zeta(\beta_c) = 2$.
\end{corollary}

\begin{remark}
The existence of eigenmeasures living on $\Sigma_A$ for transient potentials like in the previous example already was obtained in the standard thermodynamic formalism for countable Markov shifts by Van Cyr \cite{Cyr2010}, see also \cite{Shwartz2019}, assuming that $\Sigma_A$ is locally compact. In our example, $\Sigma_A$ is not locally compact, and for $\beta$ s.t. $\beta F$ is transient we discovered that the eigenmeasures always exist due to compactness of $X_A$ and Denker and Yuri's \cite{DenYu2015} eigenmeasure existence Theorem \ref{thm:existence_eigenmeasures_Denker_Yuri}. Moreover, the phase transition in this setting is not in the sense of existence or absence of eigenmeasures, but meaning that the eigenmeasure changes the space which it gives mass, from $\Sigma_A$ to $Y_A$. Figure \ref{fig:complete_characterization_eigenmeasures_renewal} compares the phase transitions between standard and generalized formalisms.
\end{remark}

\begin{figure}[H]
 \includegraphics[scale=.4]{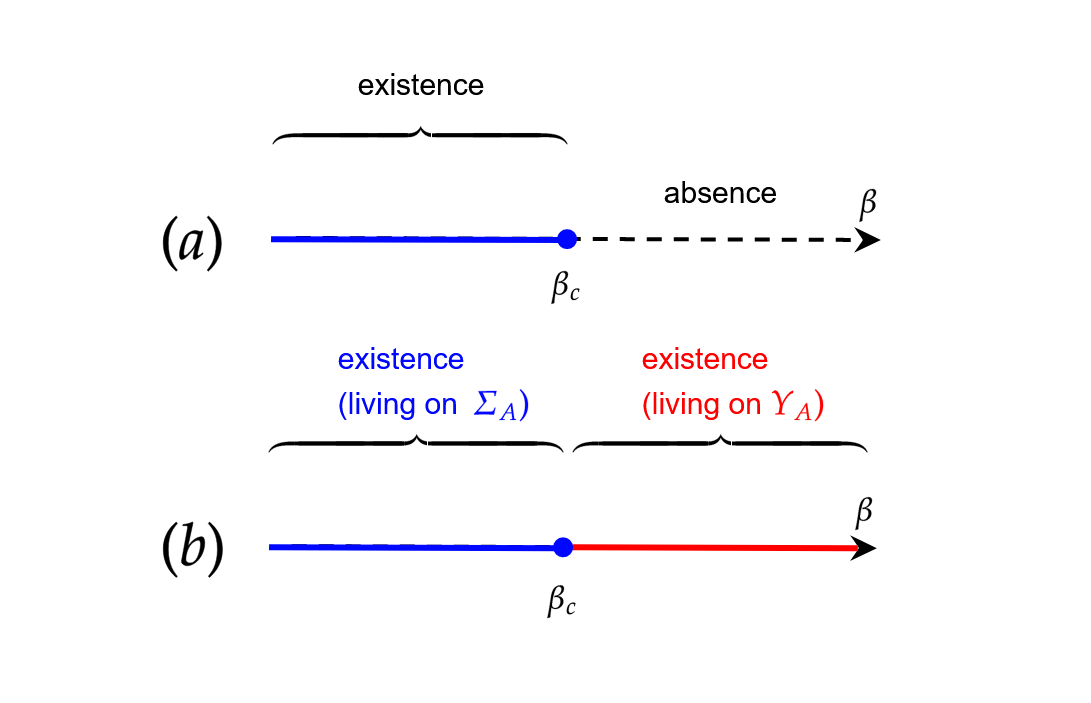}
 
 \caption{The phase transitions on different thermodynamic formalisms for probability eigenmeasures in the renewal shift space with potential as in Theorem \ref{thm:complete_characterization_eigenmeasures_renewal}. In this figure, $\beta_c$ is the unique real positive solution for the equation $\zeta(\beta_c) = 2$. The picture $(a)$ represents the standard formalism on $\Sigma_A$, where we have a unique eigenmeasure for each $\beta \leq \beta_c$ (blue interval). In this case it is not possible to detect any eingenmeasure probability for $\beta> \beta_c$. The picture $(b)$ represents the generalized formalism on $X_A$ and, unlike in $(a)$, we can see more than the standard eigenmeasure probabilities; we can detect a unique eigenmeasure for each $\beta > \log 2$ (red interval), and this measure vanishes on $\Sigma_A$.  \label{fig:complete_characterization_eigenmeasures_renewal}}
\end{figure}

\begin{theorem}\label{thm:weak_convergence_eigenmeasure_renewal} For each $\beta$ be the unique eigenmeasure probability as in Theorem \ref{thm:complete_characterization_eigenmeasures_renewal}. Then, the net $\{m_\beta\}_{\beta > \beta_c}$ converges on the weak$^*$ topology to $m_{\beta_c}$ as $\beta$ goes to $\beta_c$.
\end{theorem}

\begin{proof} Observe that
\begin{equation*}
    \mu_\beta(\{e\}) = c_e^\beta = 2 - \zeta(\beta) \to 0
\end{equation*}
as $\beta \to \beta_c$ by above. Then,
\begin{equation*}
    \lim_{\beta \to \beta^c}\mu_\beta(F) \to 0,
\end{equation*}
by above for every $F \subseteq Y_A$ finite set. For $\beta \geq \beta_c$, the associated eigenvalue is $1$, and by Lemmas \ref{lemma:condition_for_eigenmeasures_renewal_potential_log} and \ref{lemma:eigenmeasure_C_alpha_formula}, we have for every $\alpha$ positive admissible word that $\mu_\beta(C_\alpha)$ is continous on $\beta$ and then
\begin{equation*}
    \mu_\beta(C_\alpha) \to \mu_{\beta_c}(C_\alpha) = \mu_{\beta_c}([\alpha])
\end{equation*}
as $\beta$ decreases to $\beta_c$. Then, for a typical basic element in the form
\begin{equation*}
    F \sqcup \bigsqcup_{n \in \mathbb{N}}C_{w(n)},
\end{equation*}
where $F \subseteq Y_A$ is finite and $w(n)$ is a positive admissible word, one gets
\begin{equation*}
    \mu_\beta\left(F \sqcup \bigsqcup_{n \in \mathbb{N}}C_{w(n)}\right) = \mu_\beta(F) + \sum_{n\in \mathbb{N}} \mu_\beta(C_{w(n)}) \to \sum_{n\in \mathbb{N}}\mu_{\beta_c}([w(n)]) = \mu_{\beta_c}\left(F \sqcup \bigsqcup_{n \in \mathbb{N}}C_{w(n)}\right). 
\end{equation*}
We conclude the proof by Theorem \ref{thm:convergence_measures_Bogachev}.
\end{proof}

\chapter{Conclusions and Further Research}
In this thesis, we introduced the thermodynamic formalism for generalized countable Markov shifts. We were able to extend many notions of conformal measures and the concept of Ruelle's operator. Moreover, we discovered new conformal measures that are not detected by the standard theory. A new kind of phase transition, namely the length-type phase transition, which is the change of set which the conformal measure lives, from $Y_A$ to $\Sigma_A$, as the temperature increases. Furthermore, we gave a complete description of the eigenmeasures on $X_A$ for a couple of potentials. In most cases, we also connected the measures on $Y_A$ with the standard ones in $\Sigma_A$ by weak$^*$ limits on $\beta$.

Among the possible further research directions which can be investigated from this thesis, we highlight the following:
\begin{itemize}
    \item[$1.$] Try to apply similar formalism done for $X_A$ on more realistic physical models. The results obtained in this thesis gives a good opportunity to interact more with the physicists, and a possible path is to verify if $X_A$ may substitute the usual configuration space in some model in Statistical Mechanics, pointing out new conformal measures.
    \item[$2.$] To extend the study of conformal measures on the generalized setting for shift spaces, which are beyond the class of the renewal type shifts and, as well as it was done here, to compare it with the standard setting. In particular, to investigate the length-type phase transition phenomenon on a more general class of shift spaces beyond the renewal class.
    \item[$3.$] To find a suitable definition for the pressure for models on $X_A$, since the Gurevich pressure uses periodic points on its partition function, it does never deal with finite words. In this thesis, we used Denker-Yuri's paper \cite{DenYu2015}, which works for the renewal type shifts; however, we could have, for instance, non-compact $X_A$ and more general potentials. One possible direction on this topic could be, for example, to adapt the definition of pressure developed by Thompson in \cite{Thompson2011}.
    \item[$4.$] To investigate if the length-type phase transition phenomenon also occurs on a more general class of shift spaces beyond the renewal type class and to prove more general results on the occurrence of this phenomenon.
\end{itemize}

\renewcommand{\chaptermark}[1]{\markboth{\MakeUppercase{\appendixname\ \thechapter}} {\MakeUppercase{#1}} }
\fancyhead[RE,LO]{}
\appendix

\chapter{Proof of the Theorem \ref{thm:huge_generators_intersections}}
\label{ape:proof_giant_theorem_generalized_cylinders}

\section{Auxiliary Results}

\begin{lemma}\label{lemma:CF_intersections_general} For any two positive admissible words $\alpha$ and $\gamma$ it is true that
\begin{itemize}
\item[$(i)$] $$\displaystyle
        C_\alpha \cap F_\gamma =  F_{\gamma}^{\alpha};$$
     \item[$(ii)$] 
   \begin{equation*}
        F_\alpha \cap F_\gamma = \begin{cases}
                                    F_\alpha, \quad \alpha \in \llbracket \gamma \rrbracket,\\
                                    F_\gamma, \quad \gamma \in \llbracket \alpha \rrbracket,\\
                                    F_{\alpha'}^*, \quad \text{otherwise},
                                 \end{cases}
    \end{equation*}
    where $\alpha'$ is the longest word in $\llbracket \alpha \rrbracket \cap \llbracket \gamma \rrbracket$;
    \item[$(iii)$] for fixed $n \in \mathbb{N}$ with $n \in \{0,1,...,|\gamma|-1\}$, let be the natural number $j \neq \gamma_n$. We have
    \begin{equation*}
        C_\alpha \cap C_{\delta^{n}(\gamma)j} = \begin{cases}
                                            C_{\delta^{n}(\gamma)j}, \quad            \alpha \in \llbracket             \gamma \rrbracket                 \text{ and } n \geq               |\alpha|,\\
                                            C_\alpha, \quad                   \alpha \notin                     \llbracket \gamma                 \rrbracket,                        \gamma \notin                     \llbracket \alpha                 \rrbracket, n =                  |\alpha'| \text{ and }           j=\alpha_{|\alpha'|},           \\
                                            \emptyset, \quad                  \text{otherwise};
                                         \end{cases}
    \end{equation*}
    where $\alpha'$ is the longest stem in $\llbracket \alpha \rrbracket \cap \llbracket \gamma \rrbracket$.
    \item[$(iv)$] for $j \in \mathbb{N}$,
    \begin{equation*} 
        C_\alpha \cap G(\gamma,j) = \begin{cases}
                                        G(\gamma,j), \quad \alpha \in \llbracket \gamma \rrbracket,\\
                                        \emptyset, \quad \text{otherwise};
                                    \end{cases}
    \end{equation*}
    \item[$(v)$] for $j \in \mathbb{N}$,
    \begin{equation*} 
        C_\alpha \cap K(\gamma,j) = \begin{cases}
                                        K(\gamma,j), \quad \alpha \in \llbracket \gamma \rrbracket,\\
                                        \emptyset, \quad \text{otherwise};
                                    \end{cases}
    \end{equation*}
    \item[$(vi)$] for $H,I \subseteq \mathbb{N}$,
    \begin{equation*} 
        G(\alpha,H) \cap G(\gamma,I) = \begin{cases}
                                        G(\alpha,H\cup I), \quad \alpha = \gamma,\\
                                        \emptyset, \quad \text{otherwise};
                                    \end{cases}
    \end{equation*}
    \item[$(vii)$] for $H,I \subseteq \mathbb{N}$,
    \begin{equation*} 
        K(\alpha,H) \cap K(\gamma,I) = \begin{cases}
                                        K(\alpha,H\cup I), \quad \alpha = \gamma,\\
                                        \emptyset, \quad \text{otherwise};
                                    \end{cases}
    \end{equation*}
    \item[$(viii)$] 
    \begin{equation*}
        F_\alpha \cap G(\gamma,j) = \begin{cases}
                                        G(\gamma,j), \quad \gamma \in \llbracket \alpha \rrbracket \setminus \{\alpha\},\\
                                        \emptyset, \quad \text{otherwise};
                                    \end{cases}
    \end{equation*}
    \item[$(ix)$] 
    \begin{equation*}
        F_\alpha \cap K(\gamma,j) = \begin{cases}
                                        K(\gamma,j), \quad \gamma \in \llbracket \alpha \rrbracket \setminus \{\alpha\},\\
                                        \emptyset, \quad \text{otherwise};
                                    \end{cases}
    \end{equation*}
\end{itemize}
\end{lemma}

\begin{proof} For $(i)$, the proof is straightforward. Now, for the proof of $(ii)$, we notice from the definition of $F_\alpha$ that 
\begin{equation*}
    F_\alpha\cap F_\gamma = \left\{\xi \in Y_A: \kappa(\xi) \in (\llbracket\alpha\rrbracket \cap \llbracket\gamma\rrbracket)\setminus\{\alpha,\gamma\}\right\}.
\end{equation*}
If $\alpha \in \llbracket\gamma\rrbracket$ then $\llbracket\alpha\rrbracket \subseteq \llbracket \gamma \rrbracket$ and therefore $(\llbracket\alpha\rrbracket \cap \llbracket\gamma\rrbracket)\setminus\{\alpha,\gamma\}= \llbracket\alpha\rrbracket\setminus\{\alpha\}$. In this case, we obtain $F_\alpha\cap F_\gamma = F_\alpha$. A similar proof holds for the case that $\gamma \in \llbracket \alpha \rrbracket$. The remaining possibility occurs when $\alpha \notin \llbracket\gamma\rrbracket$ is simultaneous to $\gamma \notin \llbracket \alpha \rrbracket$, and here we consider $\alpha'$ as in the statement. It follows immediately that $(\llbracket\alpha\rrbracket \cap \llbracket\gamma\rrbracket)\setminus\{\alpha,\gamma\}= \llbracket \alpha'\rrbracket$, since we only must have $|\alpha'|<\min\{|\alpha|,|\gamma|\}$. We conclude that $\xi  \in F_\alpha\cap F_\gamma$ if and only if $\kappa(\xi) \in \llbracket\alpha'\rrbracket$ and therefore $F_\alpha\cap F_\gamma = F_{\alpha'}^*$.

To prove $(iii)$ we use \eqref{eq:C_cap_C} and obtain
\begin{equation*}
    C_\alpha \cap C_{\delta^{n}(\gamma)j} = \begin{cases}
                                C_\alpha, \quad \text{if } \delta^{n}(\gamma)j \in \llbracket \alpha \rrbracket,\\
                                C_{\delta^{n}(\gamma)j}, \quad \text{if } \alpha \in \llbracket \delta^{n}(\gamma)j \rrbracket,\\
                                \emptyset, \quad \text{otherwise}.
                            \end{cases} 
\end{equation*}
When $\alpha \in \llbracket \gamma \rrbracket$ it never occurs that $\delta^{n}(\gamma) j \in \llbracket \alpha \rrbracket$, for any $n$ and $j$. Indeed, if $\delta^{n}(\gamma) j \in \llbracket \alpha \rrbracket$ then $\delta^{n}(\gamma) j = \alpha_0 \cdots \alpha_n$ and $n \leq |\alpha|-1$, and consequently $j = \alpha_n = \gamma_n$, a contradiction. If $\alpha \in \llbracket \delta^{n}(\gamma) j \rrbracket$ and $n \geq |\alpha|$, then $C_\alpha \supseteq C_{\delta^{n}(\gamma)j}$ and $C_\alpha \cap C_{\delta^{n}(\gamma)j} = C_{\delta^{n}(\gamma)j}$. On other hand, if $n < |\alpha|$ then it is clear that $C_\alpha \cap C_{\delta^{n}(\gamma)j} = \emptyset$ since $j \neq \alpha_n$. Now, let us consider $\alpha \notin \llbracket \gamma \rrbracket$, which implies that $|\alpha|>0$ because, otherwise, we would get $\alpha = e$, which is subword of every positive admissible word. If $\gamma \in \llbracket \alpha \rrbracket$, observe that $\delta^{n}(\gamma)j \notin \llbracket \alpha \rrbracket$ because $j \neq \gamma_n = \alpha_n$, and hence $C_{\delta^{n}(\gamma)j} \subseteq C_\alpha^c$, that is, $C_{\delta^{n}(\gamma)j} \cap C_\alpha = \emptyset$. For the case that $\gamma \notin \llbracket \alpha \rrbracket$, let $\alpha'$ be the longest word of $\llbracket \alpha \rrbracket \cap \llbracket \gamma \rrbracket$. If $n = |\alpha'|$ and $j=\alpha_{|\alpha'|}$, then it follows directly that $C_\alpha \subseteq C_{\delta^{n}(\gamma)j}$ and hence $C_\alpha \cap C_{\delta^{n}(\gamma)j}= C_\alpha$. If $n = |\alpha'|$ and $j\neq\alpha_{|\alpha'|}$ then it is clear that $C_\alpha \cap C_{\delta^{n}(\gamma)j}= \emptyset$. Now, if $n < |\alpha'|$, then also occurs $C_\alpha \cap C_{\delta^{n}(\gamma)j}= \emptyset$ because $j \neq \gamma_n = \alpha'_n$. The same happens for $n > |\alpha'|$ since $\alpha \notin \llbracket \gamma' \rrbracket$ for any $\gamma' \in \llbracket\gamma\rrbracket$ with $|\gamma'|>|\alpha|$.

The proofs of $(iv)-(ix)$ are straightforward. \qed
\end{proof}

\begin{lemma}\label{lemma:C_delta_cap_C_gamma_p} Let $\alpha$. $\gamma$ be finite admissible words such that $\gamma \in \llbracket \alpha \rrbracket \setminus \{\alpha\}$ and consider $k,p \in \mathbb{N}$, $m \in \{0,1,..., |\alpha|-1\}$, such that $k \neq \alpha_m$. Then,

\begin{equation}
    C_{\delta^m(\alpha)k} \cap C_{\gamma p} = \begin{cases}
               C_{\delta^m(\alpha)k}, \quad \text{if } |\alpha| > |\gamma| + 1, \text{ } p = \alpha_{|\gamma|}, \text{ and } m \geq |\gamma|+1;\\
               C_{\gamma p}, \quad  \text{if } m= |\gamma|, k=p;\\
               \emptyset, \quad \text{otherwise}.
    \end{cases}
\end{equation}
\end{lemma}

\begin{proof} Lemma \ref{lemma:CF_intersections_general} $(iii)$ gives 
\begin{equation}\label{eq:C_delta_cap_C_gamma_p}
    C_{\delta^m(\alpha)k} \cap C_{\gamma p} = \begin{cases}
               C_{\delta^m(\alpha)k}, \quad \gamma p \in \llbracket \alpha \rrbracket \text{ and } m \geq |\gamma|+1,\\
               C_{\gamma p}, \quad \gamma p \notin \llbracket \alpha \rrbracket, \alpha \notin \llbracket \gamma p \rrbracket, m= |\gamma|, k=(\gamma p)_{|\gamma|},\\
               \emptyset, \quad \text{otherwise};
    \end{cases}
\end{equation}
Since $\gamma \in \llbracket \alpha \rrbracket \setminus \{\alpha\}$, we have that $\gamma p \in \llbracket \alpha \rrbracket$ if and only if $p = \alpha_{|\gamma|}$. Note that $\gamma p \notin \llbracket \alpha \rrbracket$ implies $\alpha \notin \llbracket \gamma p \rrbracket$. Also, $(\gamma p)_|\gamma| = p$, and when $k=p$ we necessarily have $p \neq \alpha_{|\gamma|}$. The equality \eqref{eq:C_delta_cap_C_gamma_p} becomes
\begin{equation}\label{eq:C_delta_cap_C_gamma_p_gamma_subword}
    C_{\delta^m(\alpha)k} \cap C_{\gamma p} = \begin{cases}
               C_{\delta^m(\alpha)k}, \quad p = \alpha_{|\gamma|} \text{ and } m \geq |\gamma|+1,\\
               C_{\gamma p}, \quad m= |\gamma|,\text{ } k=p,\\
               \emptyset, \quad \text{otherwise}.
    \end{cases}
\end{equation}
For $|\alpha| = |\gamma| + 1$, if the conditions $p = \alpha_{|\gamma|} \text{ and } m \geq |\gamma|+1$ are satisfied, we have necessarily that  $\gamma p = \alpha$ and then $C_{\delta^m(\alpha)k} \cap C_{\gamma p} = C_{\delta^m(\alpha)k} \cap C_\alpha = \emptyset$, for any $m$ and $k$ considered, because $m \leq |\alpha|-1$ and $k \neq \alpha_m$.\qed
\end{proof}

\begin{lemma} \label{lemma:C_delta_cap_C_delta} Let $\alpha$. $\gamma$ be finite admissible words such that $\gamma \notin \llbracket \alpha \rrbracket$ and $\alpha \notin \llbracket \gamma \rrbracket$, and $\alpha'$ the longest word in $\llbracket \alpha \rrbracket \cap \llbracket \gamma \rrbracket$. Also, consider $k,p \in \mathbb{N}$, $m \in \{0,1,..., |\alpha|-1\}$ and $n \in \{0,1,..., |\gamma|-1\}$ , such that $k \neq \alpha_m$ and $p \neq \gamma_n$. Then,

\begin{equation}
    C_{\delta^m(\alpha)k} \cap C_{\delta^n(\gamma) p} = \begin{cases}
               C_{\delta^m(\alpha)k} = C_{\delta^n(\gamma)p}, \quad \text{if } n=m \leq |\alpha'| \text{ and } p=k,\\
               C_{\delta^m(\alpha)k}, \quad \text{if } n = |\alpha'| < m \text{ and } p=\alpha_n;\\
               C_{\delta^n(\gamma)p}, \quad \text{if } m = |\alpha'| < n \text{ and } k=\gamma_m;\\
               \emptyset, \quad \text{otherwise}.
    \end{cases}
\end{equation}
\end{lemma}

\begin{proof}
The identity \eqref{eq:C_cap_C} gives
\begin{equation*}
    C_{\delta^m(\alpha)k} \cap C_{\delta^n(\gamma) p} = \begin{cases}
        C_{\delta^m(\alpha)k}, \quad \text{if } \delta^n(\gamma) p \in \llbracket \delta^m(\alpha)k \rrbracket,\\
         C_{\delta^n(\gamma) p}, \quad \text{if } \delta^m(\alpha)k \in \llbracket \delta^n(\gamma) p \rrbracket,\\
        \emptyset, \quad \text{otherwise}.
    \end{cases}  
\end{equation*}
The case $\delta^n(\gamma) p \in \llbracket \delta^m(\alpha)k \rrbracket$, which implies that $n \leq m$, occurs if and only if $\gamma_0 \cdots \gamma_{n-1}p = \alpha_0 \cdots \alpha_{n-1}l$, where
\begin{equation*}
    l = \begin{cases}
        k, \quad n=m,\\
        \alpha_n, \quad n<m.
    \end{cases}
\end{equation*}
However, by the definition of $\alpha'$, we have $\gamma_i = \alpha_i$ for all $i = 0,..., n-1$ if and only if $n-1 \leq |\alpha'|-1$, that is, $n \leq |\alpha'|$. If $n = m$, then it is straightforward that $\delta^n(\gamma) p \in \llbracket \delta^m(\alpha)k \rrbracket$ if and only if $p=k$. Now, suppose that $n<m\leq |\alpha'|$, then $\delta^n(\gamma) p \in \llbracket \delta^m(\alpha)k \rrbracket$ if and only if $p = \alpha_n$, which never happens since $p \neq \gamma_n = \alpha_n$ for this case. If $n < |\alpha'|<m$ we have a similar problem as in the previous situation and hence this never happens. Finally, if $n=|\alpha'|<m$ we have that $\delta^n(\gamma) p \in \llbracket \delta^m(\alpha)k \rrbracket$ if and only if $p = \alpha_n \neq \gamma_n$. We conclude that $\delta^n(\gamma) p \in \llbracket \delta^m(\alpha)k \rrbracket$ if and only if $n = m$ and $p=k$, or $n=|\alpha'|<m$ and $p = \alpha_n$. It is analogous to prove that $\delta^m(\alpha) k \in \llbracket \delta^n(\gamma) p \rrbracket$ if and only if $n = m$ and $p=k$, or $m=|\alpha'|<n$ and $k = \gamma_m$. \qed
\end{proof}

\begin{corollary}\label{cor:4_tuple_union_delta} Let $\alpha$ and $\gamma$ be positive admissible words. If $\alpha \in \llbracket \gamma \rrbracket$, then
\begin{align*}
    \bigsqcup_{m=0}^{|\alpha|-1}\bigsqcup_{k\neq \alpha_m} \bigsqcup_{n=0}^{|\gamma|-1}\bigsqcup_{p\neq \gamma_n}  C_{\delta^m(\alpha)k} \cap C_{\delta^n(\gamma )p} &= \bigsqcup_{m = 0}^{|\alpha| - 1} \bigsqcup_{k\neq \alpha_m} C_{\delta^m(\alpha)k}.
\end{align*}
If $\alpha \notin \llbracket \gamma \rrbracket$ and $\gamma \notin \llbracket \alpha \rrbracket$, then
\begin{align*}
    \bigsqcup_{m=0}^{|\alpha|-1}\bigsqcup_{k\neq \alpha_m} \bigsqcup_{n=0}^{|\gamma|-1}\bigsqcup_{p\neq \gamma_n}  C_{\delta^m(\alpha)k} \cap C_{\delta^n(\gamma )p} &= \left(\bigsqcup_{m = 0}^{|\alpha|-1} \bigsqcup_{k\neq \alpha_m} C_{\delta^m(\alpha)k}\right) \sqcup \left(\bigsqcup_{|\alpha'| < n < |\gamma|-1} \bigsqcup_{p\neq \gamma_n} C_{\delta^n(\gamma )p}\right),
\end{align*}
where $\alpha'$ is the longest word in $\llbracket \alpha \rrbracket \cap \llbracket \gamma \rrbracket$.
\end{corollary}

\begin{proof} If $\alpha \in \llbracket \gamma \rrbracket$, then we may write
\begin{align*}
    \bigsqcup_{m=0}^{|\alpha|-1}\bigsqcup_{k\neq \alpha_m} \bigsqcup_{n=0}^{|\gamma|-1}\bigsqcup_{p\neq \gamma_n}  C_{\delta^m(\alpha)k} \cap C_{\delta^n(\gamma )p} &= \left(\bigsqcup_{m = 0}^{|\alpha| - 1} \bigsqcup_{0 \leq n \leq |\alpha|-1} \bigsqcup_{k\neq \alpha_m} \bigsqcup_{p\neq \gamma_n}  C_{\delta^m(\alpha)k} \cap C_{\delta^n(\gamma )p}\right)  \\
    &\sqcup  \left(\bigsqcup_{m = 0}^{|\alpha| - 1} \bigsqcup_{|\alpha| - 1 < n < |\gamma|-1} \bigsqcup_{k\neq \alpha_m} \bigsqcup_{p\neq \gamma_n}  C_{\delta^m(\alpha)k} \cap C_{\delta^n(\gamma )p}\right),
\end{align*}
and it is straightforward that 
\begin{equation*}
    \bigsqcup_{m = 0}^{|\alpha| - 1} \bigsqcup_{0 \leq n \leq |\alpha|-1} \bigsqcup_{k\neq \alpha_m} \bigsqcup_{p\neq \gamma_n}  C_{\delta^m(\alpha)k} \cap C_{\delta^n(\gamma )p} = \bigsqcup_{m = 0}^{|\alpha| - 1} \bigsqcup_{k\neq \alpha_m} C_{\delta^m(\alpha)k},
\end{equation*}
because $\delta^n(\gamma) = \delta^n(\alpha)$ and $\alpha_n = \gamma_n$ for $0 \leq n \leq |\alpha|-1$. It is straightforward that
\begin{equation*}
    \bigsqcup_{m = 0}^{|\alpha| - 1} \bigsqcup_{|\alpha| - 1 < n < |\gamma|-1} \bigsqcup_{k\neq \alpha_m} \bigsqcup_{p\neq \gamma_n}  C_{\delta^m(\alpha)k} \cap C_{\delta^n(\gamma )p} = \emptyset,
\end{equation*}
because if $n > |\alpha|-1$ then $\delta^n(\gamma) = \alpha \gamma'$ for some admissible positive word $\gamma'$ such that $\alpha \gamma'$ is also admissible, while we have that $\delta^m(\alpha)k$ with $k \neq \alpha_m$ in the union above.

Suppose now that $\alpha \notin \llbracket \gamma \rrbracket$ and $\gamma \notin \llbracket \alpha \rrbracket$. Then, we may write
\begin{align*}
    \bigsqcup_{m=0}^{|\alpha|-1}\bigsqcup_{k\neq \alpha_m} \bigsqcup_{n=0}^{|\gamma|-1}\bigsqcup_{p\neq \gamma_n}  C_{\delta^m(\alpha)k} \cap C_{\delta^n(\gamma )p} &= \left(\bigsqcup_{0 \leq m \leq |\alpha'|} \bigsqcup_{0 \leq n \leq |\alpha'|} \bigsqcup_{k\neq \alpha_m} \bigsqcup_{p\neq \gamma_n}  C_{\delta^m(\alpha)k} \cap C_{\delta^n(\gamma )p}\right)  \\
    &\sqcup  \left(\bigsqcup_{0 \leq m \leq |\alpha'|} \bigsqcup_{|\alpha'| < n < |\gamma|-1} \bigsqcup_{k\neq \alpha_m} \bigsqcup_{p\neq \gamma_n}  C_{\delta^m(\alpha)k} \cap C_{\delta^n(\gamma )p}\right)  \\
    &\sqcup  \left(\bigsqcup_{|\alpha'| < m  < |\alpha|-1} \bigsqcup_{0 \leq n \leq |\alpha'|} \bigsqcup_{k\neq \alpha_m} \bigsqcup_{p\neq \gamma_n}  C_{\delta^m(\alpha)k} \cap C_{\delta^n(\gamma )p}\right) \\
    &\sqcup \left(\bigsqcup_{|\alpha'| < m  < |\alpha|-1} \bigsqcup_{|\alpha'| < n < |\gamma|-1} \bigsqcup_{k\neq \alpha_m} \bigsqcup_{p\neq \gamma_n}  C_{\delta^m(\alpha)k} \cap C_{\delta^n(\gamma )p}\right). 
\end{align*}
By applying Lemma \ref{lemma:C_delta_cap_C_delta} on each of the four parcels above, we obtain
\begin{align*}
    \bigsqcup_{0 \leq m \leq |\alpha'|} \bigsqcup_{0 \leq n \leq |\alpha'|} \bigsqcup_{k\neq \alpha_m} \bigsqcup_{p\neq \gamma_n}  C_{\delta^m(\alpha)k} \cap C_{\delta^n(\gamma )p} &= \bigsqcup_{m = 0}^{|\alpha'|} \bigsqcup_{k\neq \alpha_m} C_{\delta^m(\alpha)k},\\
    \bigsqcup_{0 \leq m \leq |\alpha'|} \bigsqcup_{|\alpha'| < n < |\gamma|-1} \bigsqcup_{k\neq \alpha_m} \bigsqcup_{p\neq \gamma_n}  C_{\delta^m(\alpha)k} \cap C_{\delta^n(\gamma )p} &= \bigsqcup_{|\alpha'| < n < |\gamma|-1} \bigsqcup_{p\neq \gamma_n} C_{\delta^n(\gamma )p},\\
    \bigsqcup_{|\alpha'| < m  < |\alpha|-1} \bigsqcup_{0 \leq n \leq |\alpha'|} \bigsqcup_{k\neq \alpha_m} \bigsqcup_{p\neq \gamma_n}  C_{\delta^m(\alpha)k} \cap C_{\delta^n(\gamma )p} &= \bigsqcup_{|\alpha'| < m  < |\alpha|-1} \bigsqcup_{k\neq \alpha_m}   C_{\delta^m(\alpha)k}, \\
    \bigsqcup_{|\alpha'| < m  < |\alpha|-1} \bigsqcup_{|\alpha'| < n < |\gamma|-1} \bigsqcup_{k\neq \alpha_m} \bigsqcup_{p\neq \gamma_n}  C_{\delta^m(\alpha)k} \cap C_{\delta^n(\gamma )p} &= \emptyset,
\end{align*}
and the proof is finished by noticing that
\begin{align*}
    \left(\bigsqcup_{m = 0}^{|\alpha'|} \bigsqcup_{k\neq \alpha_m} C_{\delta^m(\alpha)k}\right) \sqcup \left( \bigsqcup_{|\alpha'| < m  < |\alpha|-1} \bigsqcup_{k\neq \alpha_m}   C_{\delta^m(\alpha)k} \right) &= \bigsqcup_{m = 0}^{|\alpha| - 1} \bigsqcup_{k\neq \alpha_m} C_{\delta^m(\alpha)k}.
\end{align*} \qed
\end{proof}

\section{Proof of the Theorem 4.83}

\textbf{Proof of \eqref{eq:C_cap_C_comp}:} let $\alpha \in \llbracket \gamma \rrbracket$, then by Proposition \ref{prop:general_C_alpha_complement} we have that
\begin{equation*}
    C_\alpha \cap C_\gamma^c = (C_\alpha \cap F_\gamma) \sqcup \bigsqcup_{n=0}^{|\gamma|-1}\left(\bigsqcup_{j\neq \gamma_n}C_\alpha \cap C_{\delta^{n}(\gamma)j}\right).
\end{equation*}
By Lemma \ref{lemma:CF_intersections_general} $(i)$ and $(iii)$ it is clear that \eqref{eq:C_cap_C_comp} holds for this case. Similarly the same proof holds\footnote{Alternatively, it is straightforward to notice that $C_\alpha \subseteq C_\gamma^c$ for this case.} for when $\alpha \notin \llbracket \gamma \rrbracket \text{ and } \gamma \notin \llbracket \alpha \rrbracket$. If\footnote{Also this follows similarly for the case $\alpha \in \llbracket \gamma \rrbracket$.} $\gamma \in \llbracket \alpha \rrbracket$ then for every $\xi \in C_\alpha$ we have that $\xi_\gamma = 1$ and therefore $\xi \notin C_\gamma^c$, implying that $C_\alpha \cap C_\gamma^c= \emptyset$.

\textbf{Proof of \eqref{eq:C_comp_cap_C_comp}:} we notice that if $\alpha \in \llbracket \gamma \rrbracket$ is equivalent to affirm that $C_\alpha \supseteq C_\gamma$, which is equivalent to state that $C_\alpha^c \subseteq C_\gamma^c$ and hence $C_\alpha^c \cap C_\gamma^c = C_\alpha^c$. A similar proof is used when $\gamma \in \llbracket \alpha \rrbracket$. For the last case, when  $\alpha \notin \llbracket\gamma\rrbracket$ and $\gamma \notin \llbracket \alpha \rrbracket$, we write
\begin{align*}
    C_\alpha^c \cap C_\gamma^c &= (F_\alpha \cap F_\gamma) \sqcup \left(\bigsqcup_{n=0}^{|\gamma|-1}\bigsqcup_{j\neq \gamma_n}F_\alpha \cap C_{\delta^{n}(\gamma)j}\right) \sqcup \left(\bigsqcup_{m=0}^{|\alpha|-1}\bigsqcup_{k\neq \alpha_m} F_\gamma \cap C_{\delta^{m}(\alpha)k} \right)\\
    &\sqcup \left(\bigsqcup_{m=0}^{|\alpha|-1}\bigsqcup_{k\neq \alpha_m} \bigsqcup_{j=0}^{|\gamma|-1}\bigsqcup_{j\neq \gamma_j}C_{\delta^{m}(\alpha)k} \cap C_{\delta^{n}(\gamma)j}\right).
\end{align*}
By Lemma \ref{lemma:CF_intersections_general} $(i)$ and $(ii)$, and Corollary \ref{cor:4_tuple_union_delta} we obtain
\begin{align*}
    C_\alpha^c \cap C_\gamma^c &= F_{\alpha'}^* \sqcup \left(\bigsqcup_{n=0}^{|\gamma|-1}\bigsqcup_{j\neq \gamma_n}F_\alpha^{\delta^{n}(\gamma)j}\right) \sqcup \left(\bigsqcup_{m=0}^{|\alpha|-1}\bigsqcup_{k\neq \alpha_m} F_\gamma^{\delta^{m}(\alpha)k} \right)\\
    &\sqcup \left(\bigsqcup_{m = 0}^{|\alpha|-1} \bigsqcup_{k\neq \alpha_m} C_{\delta^m(\alpha)k}\right) \sqcup \left(\bigsqcup_{|\alpha'| < n < |\gamma|-1} \bigsqcup_{p\neq \gamma_n} C_{\delta^n(\gamma )p}\right).
\end{align*}
Since $\alpha'$ is the largest word in $\llbracket \alpha \rrbracket \cap \llbracket \gamma \rrbracket$, we separate in the possibilities as follows. If $0 \leq m < |\alpha'|$ (for $|\alpha'|>0$), we have that $\delta^{m}(\alpha)k \notin \llbracket\kappa(\xi)\rrbracket$ for any $\xi \in Y_A$ such that $\kappa(\xi) \in \llbracket\gamma\rrbracket$ because $k \neq \alpha_m = \gamma_m$. The same occurs for $m = |\alpha'|$ with $k \neq \gamma_m$ and for $m > |\alpha'|$. If $m = |\alpha'|$ with $k = \gamma_m$, we have that $\delta^{m}(\alpha)k \in \llbracket\kappa(\xi)\rrbracket$ for every $\xi \in Y_A$ such that $\kappa(\xi) \in \llbracket\gamma\rrbracket\setminus\{\gamma\}$ and $|\kappa(\xi)|\geq m+1$. Hence, 
\begin{equation*}
    F_\gamma^{\delta^{m}(\alpha)k} = \begin{cases}
                            F_\gamma^{\alpha'\gamma_{|\alpha'|}}, \quad m=|\alpha'|, k = \gamma_m,\\
                            \emptyset, \quad \text{otherwise}.
                           \end{cases}
\end{equation*}
Analogously we obtain
\begin{equation*}
    F_\alpha^{\delta^{n}(\gamma)j} = \begin{cases}
                            F_\alpha^{\alpha'\alpha_{|\alpha'|}}, \quad n=|\alpha'|, j = \alpha_n,\\
                            \emptyset, \quad \text{otherwise}.
                           \end{cases}
\end{equation*}
Therefore,
\begin{align*}
    C_\alpha^c \cap C_\gamma^c &= F_{\alpha'}^* \sqcup F_\alpha^{\alpha'\alpha_{|\alpha'|}} \sqcup F_\gamma^{\alpha'\gamma_{|\alpha'|}} \sqcup  \left(\bigsqcup_{m = 0}^{|\alpha|-1} \bigsqcup_{k\neq \alpha_m} C_{\delta^m(\alpha)k}\right) \sqcup \left(\bigsqcup_{|\alpha'| < n < |\gamma|-1} \bigsqcup_{p\neq \gamma_n} C_{\delta^n(\gamma )p}\right).
\end{align*}

\textbf{Proof\footnote{Alternatively, the equality can be proved by using Lemma \ref{lemma:CF_intersections_general}.} of \eqref{eq:C_cap_C_inverse}:} if $\alpha \in \llbracket \gamma \rrbracket$, then it is straightforward that $C_{\gamma j^{-1}} \subseteq C_{\gamma} \subseteq C_\alpha$ and hence $C_\alpha \cap C_{\gamma j^{-1}}=C_{\gamma j^{-1}}$. A similar argument follows when $\gamma \in \llbracket \alpha \rrbracket \setminus \{\alpha\} \text{ and } A(j,\alpha_{|\gamma|})=1$, leading to $C_{\gamma j^{-1}} \supseteq C_\alpha$. If $\alpha \notin \llbracket \gamma \rrbracket$ and $\gamma \notin \llbracket \alpha \rrbracket \setminus \{\alpha\}$ it is straighforward that $C_{\gamma j^{-1}}^c \supseteq C_\alpha$, and the same happens when $\gamma \in \llbracket \alpha \rrbracket \setminus \{\alpha\}$ and $A(j,\alpha_{|\gamma|})=0$.

\textbf{Proof of\eqref{eq:C_cap_C_inverse_comp}:} Proposition \ref{prop:general_C_alpha_j_inverse_complement} gives
\begin{align*}
    C_\alpha \cap  C_{\gamma j^{-1}}^c &= (C_\alpha \cap K(\gamma,j)) \sqcup (C_\alpha \cap F_\gamma) \sqcup \left( \bigsqcup_{n=0}^{|\gamma|-1}\bigsqcup_{p\neq \gamma_n}  C_\alpha \cap C_{\delta^n(\gamma)p}\right)\sqcup \bigsqcup_{\substack{p: A(j,p)=0}} C_\alpha \cap C_{\gamma p}.
\end{align*}
By Lemma \ref{lemma:CF_intersections_general} $(iii)$ and $(v)$, if $\alpha \in \llbracket \gamma \rrbracket$ then
\begin{align*}
    C_\alpha \cap  C_{\gamma j^{-1}}^c &= K(\gamma,j) \sqcup F_\gamma^\alpha \sqcup \left( \bigsqcup_{n=|\alpha|}^{|\gamma|-1}\bigsqcup_{p\neq \gamma_m} C_{\delta^n(\gamma)p}\right)\sqcup \bigsqcup_{\substack{p: A(j,p)=0}} C_{\gamma p}.
\end{align*}
On other hand, if $\gamma \in \llbracket \alpha \rrbracket \setminus \{\alpha\}$, the same lemma gives us that 
\begin{align*}
    C_\alpha \cap  C_{\gamma j^{-1}}^c &=  \bigsqcup_{\substack{p: A(j,p)=0}}C_\alpha \cap C_{\gamma p} = \begin{cases}
                        C_\alpha, \quad A(j,\alpha_{|\gamma|})=0,\\
                        0, \quad \text{otherwise}.
                   \end{cases}
\end{align*}
Now, if $\alpha \notin \llbracket \gamma \rrbracket$ and $\alpha \notin \llbracket \gamma \rrbracket$, it is straightforward that $C_\alpha \cap  C_{\gamma j^{-1}}^c= C_\alpha$. 

\textbf{Proof of \eqref{eq:C_comp_cap_C_inverse}:} we consider first the case when $\alpha \in \llbracket \gamma \rrbracket$. For given $\xi \in C_{\gamma j^{-1}}$ it follows that $\xi_\gamma = 1$ and hence $\xi_\alpha = 1$, that is, $C_{\gamma j^{-1}} \subseteq C_\alpha$ and therefore $C_\alpha^c \cap C_{\gamma j^{-1}} = \emptyset$. Now, if $\alpha \notin \llbracket \gamma \rrbracket$ and $\gamma \notin \llbracket \alpha \rrbracket$ we necessarily have that $\alpha \notin \llbracket \gamma \rrbracket \setminus \{\gamma\}$. Indeed, since $\gamma \notin \llbracket \alpha \rrbracket$ we have only two possibilites, namely $\alpha \in \llbracket \gamma \rrbracket \setminus \{\gamma\}$ and $\alpha \notin \llbracket \gamma \rrbracket \setminus \{\gamma\}$. If the first one happens we conclude that $\alpha \in \llbracket \gamma \rrbracket$, a contradition. Hence, for given $\xi \in C_{\gamma j^{-1}}$ we have that $\xi_\gamma = 1$ and then $\xi_\alpha = 0$, that is, $C_{\gamma j^{-1}}\subseteq C_\alpha^c$ and therefore $C_\alpha^c \cap C_{\gamma j^{-1}} = C_{\gamma j^{-1}}$. For the remaining case, namely $\gamma \in \llbracket \alpha \rrbracket \setminus \{\alpha\}$, we use the identity
\begin{align*}
    C_\alpha^c \cap C_{\gamma j^{-1}} &= (F_\alpha \cap G(\gamma,j))\sqcup\left(\bigsqcup_{p:A(j,p)=1}F_\alpha \cap C_{\gamma p} \right) \sqcup \left( \bigsqcup_{m=0}^{|\alpha|-1} \bigsqcup_{k \neq \alpha_m} C_{\delta^m(\alpha)k} \cap G(\gamma,j) \right)\\ &\sqcup \left( \bigsqcup_{m=0}^{|\alpha|-1} \bigsqcup_{k \neq \alpha_m} \bigsqcup_{p:A(j,p)=1} C_{\delta^m(\alpha)k} \cap C_{\gamma p} \right),
\end{align*}
which is a direct consequence from Propositions \ref{prop:general_C_alpha_complement} and \ref{prop:general_C_alpha_inverse_j}. Lemma \ref{lemma:CF_intersections_general} $(i)$ gives $F_\alpha \cap C_{\gamma p} = F_\alpha^{\gamma p}$, which is empty for $p \neq \alpha_{|\gamma|}$. It is straightforward from Lemma \ref{lemma:CF_intersections_general} $(iv)$ and $(viii)$ that $F_\alpha \cap G(\gamma,j) = G(\gamma,j)$ and that $C_{\delta^m(\alpha)k}  \cap G(\gamma,j) = \emptyset$ for every $m$ and $k$ considered. The equality
\begin{equation}\label{eq:C_german_equals_triple_union}
     \mathfrak{C}[\alpha,\gamma,j] = \bigsqcup_{m=0}^{|\alpha|-1} \bigsqcup_{k \neq \alpha_m} \bigsqcup_{p:A(j,p)=1} C_{\delta^m(\alpha)k} \cap C_{\gamma p}.
\end{equation}
is a straightforward consequence of Lemma \ref{lemma:C_delta_cap_C_gamma_p}.

\textbf{Proof of \eqref{eq:C_comp_cap_C_inverse_comp}:} if $\alpha \in \llbracket \gamma \rrbracket$, then for any $\xi \in C_\alpha^c$, that is $\xi_\alpha = 0$, we have necessarily that $\xi_\gamma = 0$ and hence $C_\alpha^c \subseteq C_\gamma^c$. In addition, the inclusion $C_{\gamma j^{-1}} \subseteq C_\gamma$ implies $C_{\gamma j^{-1}}^c \supseteq C_\gamma^c$ and then $C_\alpha^c \subseteq C_{\gamma j^{-1}}^c$ and therefore $C_\alpha^c \cap C_{\gamma j^{-1}}^c = C_\alpha^c$. Now, suppose that $\gamma \in \llbracket \alpha \rrbracket \setminus \{\alpha\}$. For any $\xi \in C_{\gamma j^{-1}}$ we have only two possibilities, namely $\xi_\gamma = 0$ or $\xi_\gamma = 1$ with $\xi_{\gamma j^{-1}} = 0$. For the first possibility, we have by \ref{prop:general_C_alpha_j_inverse_complement} that $\xi \in F_\gamma \sqcup \bigsqcup_{n=0}^{|\gamma|-1}\bigsqcup_{p\neq \gamma_n}  C_{\delta^n(\gamma )p} \subseteq C_\alpha^c$ because $\gamma \in \llbracket \alpha \rrbracket \setminus \{\alpha\}$. For the second one, we have necessarily that $\xi \in K(\gamma,j)\sqcup \bigsqcup_{p:A(j,p) = 0}C_{\gamma p}$. However, for $\gamma \in \llbracket \alpha \rrbracket \setminus \{\alpha\}$ it is straightforward that $K(\gamma,j) \subseteq F_\alpha$ and hence $C_\alpha^c \cap K(\gamma,j) = K(\gamma,j)$. As showed in the proof of \eqref{eq:C_comp_cap_C_inverse}, the hypothesis $\gamma \in \llbracket \alpha \rrbracket \setminus \{\alpha\}$ implies that $C_{\gamma p} \cap F_\alpha = F_\alpha^{\gamma p}$, which is empty if $|\alpha| = |\gamma| + 1$ or $p \neq \alpha_{|\gamma|}$. As a consequence of Lemma \ref{lemma:C_delta_cap_C_gamma_p} we have that
\begin{equation*}
    \mathfrak{D}[\alpha,\gamma,j] = C_\alpha^c \cap \bigsqcup_{p:A(j,p)=0}  C_{\gamma p} = \bigsqcup_{m=0}^{|\alpha|-1} \bigsqcup_{k \neq \alpha_m} \bigsqcup_{p:A(j,p)=0} C_{\delta^m(\alpha)k} \cap C_{\gamma p}.
\end{equation*}
For the remaining case $\alpha \notin \llbracket \gamma \rrbracket$ and $\gamma \notin \llbracket \alpha \rrbracket$ we use the identity
\begin{align*}
    C_\alpha^c \cap C_{\gamma j^{-1}}^c &= (F_\alpha \cap K(\gamma,j)) \sqcup (F_\alpha \cap F_\gamma) \sqcup \left(\bigsqcup_{n=0}^{|\gamma|-1}\bigsqcup_{p\neq \gamma_n}F_\alpha \cap  C_{\delta^n(\gamma )p}\right) \sqcup \left( \bigsqcup_{\substack{p: A(j,p)=0}}  F_\alpha \cap C_{\gamma p} \right)\\
    &\sqcup \left( \bigsqcup_{m=0}^{|\alpha|-1}\bigsqcup_{k\neq \alpha_m} \bigsqcup_{n=0}^{|\gamma|-1}\bigsqcup_{p\neq \gamma_n}  C_{\delta^m(\alpha)k} \cap C_{\delta^n(\gamma )p}\right) \sqcup \left( \bigsqcup_{m=0}^{|\alpha|-1}\bigsqcup_{k\neq \alpha_m} C_{\delta^m(\alpha)k} \cap K(\gamma,j)\right)  \\
    &\sqcup \left( \bigsqcup_{m=0}^{|\alpha|-1}\bigsqcup_{k\neq \alpha_m} C_{\delta^m(\alpha)k} \cap F_\gamma\right) \sqcup \left(\bigsqcup_{m=0}^{|\alpha|-1}\bigsqcup_{k\neq \alpha_m} \bigsqcup_{\substack{p: A(j,p)=0}} C_{\delta^m(\alpha)k} \cap C_{\gamma p} \right).
\end{align*}
We have $F_\alpha \cap K(\gamma,j) = F_\alpha \cap F_\gamma = \emptyset$, since $\alpha \notin \llbracket \gamma \rrbracket$ and $\gamma \notin \llbracket \alpha \rrbracket$ implies $\gamma \notin \llbracket \alpha \rrbracket \setminus \{\alpha\}$. Let $\alpha'$ be the longest word in $\llbracket \alpha \rrbracket \cap \llbracket \gamma \rrbracket$. By Lemma \ref{lemma:CF_intersections_general} $(i)$ we get
\begin{equation*}
    \bigsqcup_{n=0}^{|\gamma|-1}\bigsqcup_{p\neq \gamma_n}F_\alpha \cap  C_{\delta^n(\gamma )p} = \bigsqcup_{n=0}^{|\gamma|-1}\bigsqcup_{p\neq \gamma_n} F_\alpha^{\delta^n(\gamma )p} = F_{\alpha}^{\alpha'\alpha_{|\alpha'|}},
\end{equation*}
where the last equality holds because $\delta^n(\gamma) p \in \llbracket \alpha \rrbracket$ if and only if $n = |\alpha'|$ and $p = \alpha_{|\alpha'|}$. Similarly we obtain
\begin{equation*}
    \bigsqcup_{\substack{p: A(j,p)=0}}  F_\alpha \cap C_{\gamma p} = \bigsqcup_{\substack{p: A(j,p)=0}}  F_\alpha^{\gamma p} = \emptyset,
\end{equation*}
where the last equality holds because each $F_\alpha^{\gamma p}$ is empty. Indeed, suppose that there exists a configuration $\xi \in F_\alpha^{\gamma p}$, that is, $\kappa(\xi) \in \llbracket \alpha \rrbracket \setminus \{\alpha\}$ and $\gamma p \in \llbracket \kappa(\xi) \rrbracket$. Then, $\gamma \in \llbracket \kappa(\xi) \rrbracket \subseteq \llbracket \alpha \rrbracket$, a contradition since $\gamma \notin \llbracket \alpha \rrbracket$. Corollary \ref{cor:4_tuple_union_delta}, gives
\begin{align*}
    \bigsqcup_{m=0}^{|\alpha|-1}\bigsqcup_{k\neq \alpha_m} \bigsqcup_{n=0}^{|\gamma|-1}\bigsqcup_{p\neq \gamma_n}  C_{\delta^m(\alpha)k} \cap C_{\delta^n(\gamma )p} &= \left(\bigsqcup_{m = 0}^{|\alpha|-1} \bigsqcup_{k\neq \alpha_m} C_{\delta^m(\alpha)k}\right) \sqcup \left(\bigsqcup_{|\alpha'| < n < |\gamma|-1} \bigsqcup_{p\neq \gamma_n} C_{\delta^n(\gamma )p}\right).
\end{align*}
Lemma \ref{lemma:CF_intersections_general} $(v)$ gives
\begin{equation*}
    \bigsqcup_{m=0}^{|\alpha|-1}\bigsqcup_{k\neq \alpha_m} C_{\delta^m(\alpha)k} \cap K(\gamma,j) = K(\gamma,j),
\end{equation*}
because\footnote{Note that $\gamma_{|\alpha'|} \neq \alpha_{|\alpha'|}$.} $K(\gamma,j) \subseteq C_{\delta^{|\alpha'|}(\alpha)\gamma_{|\alpha'|}}$. Also, we have
\begin{equation*}
    \bigsqcup_{m=0}^{|\alpha|-1}\bigsqcup_{k\neq \alpha_m} C_{\delta^m(\alpha)k} \cap F_\gamma = F_\gamma^{\alpha'\gamma_{|\alpha'|}},
\end{equation*}
because
\begin{equation*}
    C_{\delta^m(\alpha)k} \cap F_\gamma = \begin{cases}
        F_\gamma^{\alpha'\gamma_{|\alpha'|}, \quad m = |\alpha'| \text{ and } k = \gamma_{|\alpha'|}},\\
        \emptyset, \quad \text{otherwise}.
    \end{cases}
\end{equation*}
In addition,
\begin{equation*}
    \bigsqcup_{m=0}^{|\alpha|-1}\bigsqcup_{k\neq \alpha_m} \bigsqcup_{\substack{p: A(j,p)=0}} C_{\delta^m(\alpha)k} \cap C_{\gamma p} = \bigsqcup_{\substack{p: A(j,p)=0}} C_{\gamma p},
\end{equation*}
since $C_{\gamma p} \subseteq C_{\delta^{|\alpha'|}(\alpha)\gamma)_{|\alpha'|}}$ for every $p$.

\textbf{Proof of \eqref{eq:C_inverse_cap_C_inverse}:} Proposition \ref{prop:general_C_alpha_inverse_j} gives
\begin{align*}
    C_{\alpha j^{-1}} \cap C_{\gamma l^{-1}} &= (G(\alpha,j) \cap G(\gamma,l)) \sqcup \left(\bigsqcup_{m:A(l,m)=1} G(\alpha,j) \cap C_{\gamma m} \right) \\
    &\sqcup \left(\bigsqcup_{k:A(j,k)=1} C_{\alpha k} \cap G(\gamma,l) \right) \sqcup \bigsqcup_{k:A(j,k)=1} \bigsqcup_{m:A(l,m)=1} C_{\alpha k} \cap C_{\gamma m}.
\end{align*}
The proof is concluded by taking by applying Lemma \ref{lemma:CF_intersections_general} $(iv)$ and $(vi)$ in order to obtain
\begin{align}
    G(\alpha,j) \cap G(\gamma,l) = \begin{cases}
        G(\alpha,\{j,l\}), \quad \alpha = \gamma,\\
        \emptyset, \quad \text{otherwise};
    \end{cases}\nonumber\\
    G(\alpha,j) \cap C_{\gamma m} = \begin{cases}
        G(\alpha,j), \quad \gamma \in \llbracket\alpha \rrbracket\setminus \{\alpha\} \text{ and } m=\alpha_{|\gamma|},\\
        \emptyset, \quad \text{otherwise};
    \end{cases}\label{eq:G_alpha_j_cap_C_gamma_m}\\
   C_{\alpha k} \cap G(\gamma,l) = \begin{cases}
        G(\gamma,l), \quad \alpha \in \llbracket \gamma \rrbracket\setminus \{\gamma\} \text{ and } k=\gamma_{|\alpha|},\\
        \emptyset, \quad \text{otherwise};
    \end{cases}\label{eq:C_alpha_k_cap_G_gamma_j}
\end{align}
and by noticing that
\begin{align} \label{eq:C_alpha_k_cap_C_gamma_m}
     C_{\alpha k} \cap C_{\gamma m} = \begin{cases}
        C_{\alpha k}, \quad \gamma \in \llbracket\alpha \rrbracket\setminus \{\alpha\} \text{ and } m=\alpha_{|\gamma|},\\
        C_{\gamma m}\quad \alpha \in \llbracket \gamma \rrbracket\setminus \{\gamma\} \text{ and } k=\gamma_{|\alpha|},\\
        C_{\alpha k} = C_{\gamma m}, \quad \alpha = \gamma \text{ and } k = m,\\
        \emptyset, \quad \text{otherwise}.
    \end{cases}
\end{align}
\textbf{Proof of \eqref{eq:C_inverse_cap_C_inverse_comp}:} we have that
\begin{align*}
    C_{\alpha j^{-1}} \cap C_{\gamma l^{-1}}^c &= (G(\alpha,j) \cap K(\gamma,l)) \sqcup (G(\alpha,j) \cap F_\gamma) \sqcup \left( \bigsqcup_{n=0}^{|\gamma|-1} \bigsqcup_{p \neq \gamma_n} G(\alpha,j) \cap C_{\delta^n(\gamma)p} \right) \\
    &\sqcup \left( \bigsqcup_{m:A(l,m)=0} G(\alpha,j) \cap C_{\gamma m} \right)\sqcup \left( \bigsqcup_{k:A(j,k)=1}  C_{\alpha k} \cap K(\gamma,l)\right) \\
    &\sqcup \left( \bigsqcup_{k:A(j,k)=1}  C_{\alpha k} \cap F_\gamma\right) \sqcup \left( \bigsqcup_{k:A(j,k)=1} \bigsqcup_{n=0}^{|\gamma|-1} \bigsqcup_{p \neq \gamma_n} C_{\alpha k} \cap C_{\delta^n(\gamma)p}\right)\\
    &\sqcup \left( \bigsqcup_{k:A(j,k)=1} \bigsqcup_{m:A(l,m)=0} C_{\alpha k} \cap C_{\gamma m}\right)
\end{align*}
due to Propositions \ref{prop:general_C_alpha_inverse_j} and \ref{prop:general_C_alpha_j_inverse_complement}. By definition of $GK(\alpha,j,l)$ we have that
\begin{equation*}
    G(\alpha,j) \cap K(\gamma,l) = \begin{cases}
    GK(\alpha,j,l), \quad \alpha = \gamma,\\
    \emptyset, \quad \text{otherwise}.
    \end{cases}
\end{equation*}
Lemma \ref{lemma:CF_intersections_general} $(viii)$ gives
\begin{equation*}
    G(\alpha,j) \cap F_\gamma = \begin{cases}
    G(\alpha,j), \quad \alpha \in \llbracket\gamma\rrbracket \setminus \{\gamma\},\\
    \emptyset, \quad \text{otherwise}.
    \end{cases}
\end{equation*}
From Lemma \ref{lemma:CF_intersections_general} $(iv)$ we have
\begin{equation*}
    G(\alpha,j) \cap C_{\delta^n(\gamma)p} = \begin{cases}
        G(\gamma,j), \quad \delta^n(\gamma)p \in \llbracket \alpha \rrbracket,\\
        \emptyset, \quad \text{otherwise};
    \end{cases}
\end{equation*}
and from Lemma \ref{lemma:CF_intersections_general} $(i)$ we get
\begin{equation*}
    C_{\alpha k} \cap F_{\gamma} = F_{\gamma}^{\alpha k},
\end{equation*}
which is empty if\footnote{It can be empty even satisfying the opposite of the conditions stated, but these cases are included in the definition of the $F$'s, we are just simplyfing the statements.} $\alpha \notin \llbracket \gamma \rrbracket\setminus \{\gamma\}$ or $k \neq \gamma_{|\alpha|}$, and then
\begin{equation*}
    \bigsqcup_{k:A(j,k)=1}  C_{\alpha k} \cap F_\gamma = \begin{cases}
        F_\gamma^{\alpha \gamma_{|\alpha|}}, \quad \alpha \in \llbracket \gamma \rrbracket\setminus\{\gamma\} \text{ and } A(j,\gamma_{|\alpha|}) = 1,\\
        \emptyset, \quad \text{otherwise}.
    \end{cases}
\end{equation*}

If $\gamma \in \llbracket \alpha \rrbracket$, then $\delta^n(\gamma)p \notin \llbracket \alpha \rrbracket$ because $p \neq \gamma_n = \alpha_n$. The same happens when $\alpha \in \llbracket \gamma \rrbracket \setminus \{\gamma\}$ and $n \leq |\alpha| -1$. Now, if $\alpha \in \llbracket \gamma \rrbracket \setminus \{\gamma\}$ and $n < |\alpha| -1$ we have also that $\delta^n(\gamma)p \notin \llbracket \alpha \rrbracket$ because we must have $p \neq \gamma_n = \alpha_n$. The last possibility consists in $\alpha \notin \llbracket \gamma \rrbracket$ and $\gamma \notin \llbracket \alpha \rrbracket$, where we consider $\alpha'$, the longest word in $\llbracket \alpha \rrbracket \cap \llbracket \gamma \rrbracket$. We have the same result as before, except when $n = |\alpha'|$, where it is straightforward that $\delta^n(\gamma)p \in \llbracket \alpha \rrbracket$ if and only if $p = \alpha_{|\alpha'|} \neq \gamma_{|\alpha'|}$. We conclude that
\begin{equation}\label{eq:union_intersection_G_C_delta}
    \bigsqcup_{n=0}^{|\gamma|-1} \bigsqcup_{p \neq \gamma_n} G(\alpha,j) \cap C_{\delta^n(\gamma)p} = \begin{cases}
        G(\alpha,j), \quad \alpha \notin \llbracket \gamma \rrbracket, \text{ }\gamma \notin \llbracket \alpha \rrbracket,\\
        \emptyset, \quad \text{otherwise}.
    \end{cases}
\end{equation}

The intersection $C_{\gamma m} \cap G(\alpha,j)$ was already explicited in \eqref{eq:G_alpha_j_cap_C_gamma_m} and therefore
\begin{equation*}
    \bigsqcup_{m:A(l,m)=0} G(\alpha,j) \cap C_{\gamma m} = \begin{cases}
        G(\alpha,j), \quad \gamma \in \llbracket \alpha \rrbracket \setminus \{\alpha\} \text{ and } A(l,\alpha_{|\gamma|}) = 0,\\
        \emptyset, \quad \text{otherwise}.
    \end{cases}
\end{equation*}
We also have
\begin{equation*}
     C_{\alpha k} \cap K(\gamma,l) = \begin{cases}
        K(\gamma,l), \quad \alpha \in \llbracket \gamma \rrbracket\setminus \{\gamma\} \text{ and } k=\gamma_{|\alpha|},\\
        \emptyset, \quad \text{otherwise},
    \end{cases}
\end{equation*}
which is proved similarly as in \eqref{eq:C_alpha_k_cap_G_gamma_j} and gives
\begin{equation}\label{eq:union_C_alpha_k_cap_K}
     \bigsqcup_{k:A(j,k)=1}  C_{\alpha k} \cap K(\gamma,l) = \begin{cases}
        K(\gamma,l), \quad \alpha \in \llbracket \gamma \rrbracket\setminus \{\gamma\} \text{ and } A(j,\gamma_{|\alpha|}) = 1,\\
        \emptyset, \quad \text{otherwise}.
    \end{cases}
\end{equation}
Now, if $\alpha \in \llbracket \gamma \rrbracket \setminus \{\gamma\}$, then by \eqref{eq:C_german_equals_triple_union} we have that
\begin{equation*}
    \mathfrak{C}[\gamma,\alpha,j] = \bigsqcup_{k:A(j,k)=1} \bigsqcup_{n=0}^{|\gamma|-1} \bigsqcup_{p \neq \gamma_n} C_{\alpha k} \cap C_{\delta^n(\gamma)p}.
\end{equation*}
The remaining case, when $\alpha \notin  \llbracket \gamma \rrbracket \setminus \{\gamma\}$, is separated into two possibilities, namely $\gamma \in  \llbracket \alpha \rrbracket$ and $\gamma \notin  \llbracket \alpha \rrbracket$. For the first one, we have necessarily that $\delta^n(\gamma) \in \llbracket \alpha \rrbracket \setminus \{\alpha\}$ for every $n$ possible. Since $p \neq \gamma_n = \alpha_n$, we conclude that $ C_{\alpha k} \cap C_{\delta^n(\gamma)p} = \emptyset$ for every $n$ and $k$ and therefore
\begin{equation*}
    \bigsqcup_{k:A(j,k)=1} \bigsqcup_{n=0}^{|\gamma|-1} \bigsqcup_{p \neq \gamma_n} C_{\alpha k} \cap C_{\delta^n(\gamma)p} = \emptyset.
\end{equation*}
The second case is equivalent to the condition $\gamma \notin \llbracket \alpha \rrbracket$ and $\alpha \notin \llbracket \gamma \rrbracket$, and we have that
\begin{equation*}
    C_{\delta^n(\gamma)p}\cap C_{\alpha k} = \begin{cases}
        C_{\alpha k}, \quad n = |\alpha'| \text{ and } p = \alpha_{|\alpha'|},\\
        \emptyset, \quad \text{otherwise};    \end{cases}
\end{equation*}
obtained from Lemma \ref{lemma:CF_intersections_general} $(iii)$ and considering that $C_{\delta^n(\gamma)p}\cap C_{\alpha k} = (C_{\delta^n(\gamma)p}\cap C_{\alpha})\cap C_{\alpha k}$. Therefore, in this case we have
\begin{equation*}
    \bigsqcup_{k:A(j,k)=1} \bigsqcup_{n=0}^{|\gamma|-1} \bigsqcup_{p \neq \gamma_n} C_{\alpha k} \cap C_{\delta^n(\gamma)p} = \bigsqcup_{k:A(j,k)=1} C_{\alpha k}.
\end{equation*}
Summarizing, we get
\begin{equation}\label{eq:triple_union_C_alpha_k_C_delta}
    \bigsqcup_{k:A(j,k)=1} \bigsqcup_{n=0}^{|\gamma|-1} \bigsqcup_{p \neq \gamma_n} C_{\alpha k} \cap C_{\delta^n(\gamma)p} = \begin{cases}
        \mathfrak{C}[\gamma,\alpha,j], \quad \alpha \in \llbracket \gamma \rrbracket \setminus \{\gamma\},\\
        \bigsqcup_{k:A(j,k)=1} C_{\alpha k}, \quad \alpha \notin \llbracket \gamma \rrbracket \text{ and } \gamma \notin \llbracket \alpha \rrbracket,\\
        \emptyset, \quad \text{otherwise}.
    \end{cases}
\end{equation}
From \eqref{eq:C_alpha_k_cap_C_gamma_m} we obtain
\begin{equation}\label{eq:double_union_matix_conditions}
    \bigsqcup_{k:A(j,k)=1} \bigsqcup_{m:A(l,m)=0} C_{\alpha k} \cap C_{\gamma m} = \begin{cases}
        \bigsqcup_{k:A(j,k) = 1, A(l,k)=0}C_{\alpha k}, \quad \alpha = \gamma, \\
        \bigsqcup_{m:A(l,m)=0}  C_{\gamma m}, \quad \alpha \in \llbracket \gamma \rrbracket \setminus \{\gamma\} \text{ and }A(j,\gamma_{|\alpha|}) = 1,\\
        \bigsqcup_{k:A(j,k)=1}  C_{\alpha k}, \quad \gamma \in \llbracket \alpha \rrbracket \setminus \{\alpha\} \text{ and }A(l,\alpha_{|\gamma|}) = 0,\\
        \emptyset, \quad \text{otherwise}.
    \end{cases} 
\end{equation}

\textbf{Proof of \eqref{eq:C_inverse_comp_cap_C_inverse_comp}:} Proposition \ref{prop:general_C_alpha_j_inverse_complement} gives
\begin{align*}
    C_{\alpha j^{-1}}^c \cap C_{\gamma l^{-1}}^c &= (K(\alpha,j) \cap K(\gamma,l)) \sqcup (K(\alpha,j) \cap F_\gamma) \sqcup \left(\bigsqcup_{n=0}^{|\gamma|-1} \bigsqcup_{p \neq \gamma_n} K(\alpha,j) \cap C_{\delta^n(\gamma)p} \right) \\
    &\sqcup \left(\bigsqcup_{p:A(l,p)=0} K(\alpha,j) \cap C_{\gamma p} \right) \sqcup (F_\alpha \cap K(\gamma,l)) \sqcup (F_\alpha \cap F_\gamma) \\
    &\sqcup \left(\bigsqcup_{n=0}^{|\gamma|-1} \bigsqcup_{p \neq \gamma_n} F_\alpha \cap C_{\delta^n(\gamma)p} \right) \sqcup \left( \bigsqcup_{p:A(l,p)=0} F_\alpha \cap C_{\gamma p}\right)\\
    &\sqcup \left(\bigsqcup_{m=0}^{|\alpha|-1}\bigsqcup_{k\neq \alpha_m} C_{\delta^m(\alpha)k} \cap K(\gamma,l)\right) \sqcup \left( \bigsqcup_{m=0}^{|\alpha|-1} \bigsqcup_{k \neq \alpha_m} C_{\delta^m(\alpha)k} \cap F_\gamma \right) \\
    &\sqcup \left(\bigsqcup_{m=0}^{|\alpha|-1}\bigsqcup_{k\neq \alpha_m} \bigsqcup_{n=0}^{|\gamma|-1} \bigsqcup_{p \neq \gamma_n} C_{\delta^m(\alpha)k} \cap C_{\delta^n(\gamma) p}\right)
    \sqcup \left(\bigsqcup_{m=0}^{|\alpha|-1}\bigsqcup_{k\neq \alpha_m} \bigsqcup_{p:A(l,p)=0}   C_{\delta^m(\alpha)k} \cap C_{\gamma p} \right)\\
    &\sqcup \left(\bigsqcup_{k:A(j,k)=0}  C_{\alpha k} \cap K(\gamma,l) \right) \sqcup \left(\bigsqcup_{k:A(j,k)=0}  C_{\alpha k} \cap F_\gamma \right) \\
    &\sqcup \left(\bigsqcup_{k:A(j,k)=0} \bigsqcup_{n=0}^{|\gamma|-1} \bigsqcup_{p \neq \gamma_n} C_{\alpha k} \cap C_{\delta^n(\gamma)p}\right) \sqcup \left(\bigsqcup_{k:A(j,k)=0} \bigsqcup_{p:A(l,p)=0} C_{\alpha k} \cap C_{\gamma p} \right).
\end{align*}
Lemma \ref{lemma:CF_intersections_general} $(vii)$ gives
\begin{equation*} 
    K(\alpha,j) \cap K(\gamma,l) = \begin{cases}
                        K(\alpha,\{j,l\}), \quad \alpha = \gamma,\\
                        \emptyset, \quad \text{otherwise}.
                        \end{cases}
\end{equation*}
Also we have
\begin{equation*}
    K(\alpha,j) \cap F_\gamma = \begin{cases}
                        K(\alpha,j), \quad \alpha \in \llbracket \gamma \rrbracket \setminus \{\gamma\},\\
                        \emptyset, \quad \text{otherwise};
                        \end{cases}
\end{equation*}
and 
\begin{equation*}
    F_\alpha \cap K(\gamma,l)  = \begin{cases}
                        K(\gamma,l), \quad \gamma \in \llbracket \alpha \rrbracket \setminus \{\alpha\},\\
                        \emptyset, \quad \text{otherwise};
                        \end{cases}
\end{equation*}
due to Lemma \ref{lemma:CF_intersections_general} $(ix)$. In addition,
\begin{equation*} 
    \bigsqcup_{n=0}^{|\gamma|-1} \bigsqcup_{p \neq \gamma_n} K(\alpha,j) \cap C_{\delta^n(\gamma)p} = \begin{cases}
        K(\alpha,j), \quad \alpha \notin \llbracket \gamma \rrbracket, \text{ }\gamma \notin \llbracket \alpha \rrbracket,\\
        \emptyset, \quad \text{otherwise};
         \end{cases}
\end{equation*}
and
\begin{equation*} 
    \bigsqcup_{n=0}^{|\alpha|-1} \bigsqcup_{k \neq \alpha_k} C_{\delta^m(\alpha)k} \cap K(\gamma,l) = \begin{cases}
        K(\gamma,l), \quad \alpha \notin \llbracket \gamma \rrbracket, \text{ }\gamma \notin \llbracket \alpha \rrbracket,\\
        \emptyset, \quad \text{otherwise};
         \end{cases}
\end{equation*}
due to a similar proof as it was done to prove \eqref{eq:union_intersection_G_C_delta}. In analogous way, the identities
\begin{equation*}
    \bigsqcup_{p:A(l,p)=0} K(\alpha,j) \cap C_{\gamma p} = \begin{cases}
        K(\alpha,j), \quad \gamma \in \llbracket \alpha \rrbracket\setminus \{\alpha\} \text{ and } A(l,\alpha_{|\gamma|}) = 0,\\
        \emptyset, \quad \text{otherwise};
        \end{cases}
\end{equation*}
and
\begin{equation*}
    \bigsqcup_{k:A(j,k)=0}  C_{\alpha k} \cap K(\gamma,l) = \begin{cases}
        K(\gamma,l), \quad \alpha \in \llbracket \gamma \rrbracket\setminus \{\gamma\} \text{ and } A(j,\gamma_{|\alpha|}) = 0,\\
        \emptyset, \quad \text{otherwise}.
    \end{cases}
\end{equation*}
are proved similarly to what is done to proof \eqref{eq:union_C_alpha_k_cap_K}. As presented in Lemma \ref{lemma:CF_intersections_general} $(ii)$,
\begin{equation*}
        F_\alpha \cap F_\gamma = \begin{cases}
                                    F_\alpha, \quad \alpha \in \llbracket \gamma \rrbracket,\\
                                    F_\gamma, \quad \gamma \in \llbracket \alpha \rrbracket,\\
                                    F_{\alpha'}^*, \quad \text{otherwise}.
                                 \end{cases}
\end{equation*}
Moreover, by Lemma \ref{lemma:CF_intersections_general} $(i)$ we have
\begin{equation*}
    F_\alpha \cap C_{\delta^n(\gamma)p} = \begin{cases}
                                 F_\alpha^{\alpha'\alpha_{|\alpha'|}}, \quad \alpha \notin \llbracket \gamma \rrbracket \setminus \{\gamma\}, \quad  \gamma \notin \llbracket \alpha \rrbracket \setminus \{\alpha\} \text{ and } n = |\alpha'|\\
                                 \emptyset, \quad \text{otherwise};
                                 \end{cases}
\end{equation*}
and the proof is straightforward for $\alpha \notin \llbracket \gamma \rrbracket \setminus \{\gamma\} \text{ and }  \gamma \notin \llbracket \alpha \rrbracket \setminus \{\alpha\}$. For $\alpha = \gamma$, the proof is the same as for $\gamma \in \llbracket \alpha \rrbracket \setminus \{\alpha\}$ and for $\alpha \in \llbracket \gamma \rrbracket \setminus \{\gamma\}$ with $0 \leq n < |\alpha|-1$, because in all these cases we have $p \neq \gamma_n = \alpha_n$. For the remaining case, $\alpha \in \llbracket \gamma \rrbracket \setminus \{\gamma\}$ and $n \geq |\alpha|-1$ it follows that $|\delta^n(\gamma)p|\geq |\alpha|$, however for every $\xi \in F_\alpha$ it is true that $|\kappa(\xi)| \leq |\alpha|-1$ and therefore $F_\alpha \cap C_{\delta^n(\gamma)p} = \emptyset$ for this case. Then,
\begin{equation*}
    \bigsqcup_{n=0}^{|\gamma|-1} \bigsqcup_{p \neq \gamma_n} F_\alpha \cap C_{\delta^n(\gamma)p} =  \begin{cases}
        F_\alpha^{\alpha'\alpha_{|\alpha'|}}, \quad \alpha \notin \llbracket \gamma \rrbracket \setminus \{\gamma\} \text{ and } \gamma \notin \llbracket \alpha \rrbracket \setminus \{\alpha\},\\
        \emptyset, \quad \text{otherwise};
    \end{cases}
\end{equation*}
and
\begin{equation*}
    \bigsqcup_{m=0}^{|\alpha|-1} \bigsqcup_{k \neq \alpha_m} C_{\delta^m(\alpha)k} \cap F_\gamma =  \begin{cases}
        F_\gamma^{\alpha'\gamma_{|\alpha'|}}, \quad \alpha \notin \llbracket \gamma \rrbracket \setminus \{\gamma\} \text{ and } \gamma \notin \llbracket \alpha \rrbracket \setminus \{\alpha\},\\
        \emptyset, \quad \text{otherwise}.
    \end{cases}
\end{equation*}
On the other hand, Lemma \ref{lemma:CF_intersections_general} $(i)$ gives\footnote{It may happen that $\gamma$ is the longest word in $\llbracket \alpha \rrbracket \setminus \{\alpha\}$. In this case we have $F_\alpha^{\gamma p} = F_\alpha^\alpha = \emptyset$.}
\begin{equation*}
    F_\alpha \cap C_{\gamma p} = \begin{cases}
        F_\alpha^{\gamma p}, \quad \gamma \in \llbracket \alpha \rrbracket \setminus \{\alpha\} \text{ and } p=\alpha_{|\gamma|},\\
        \emptyset, \quad \text{otherwise}.
    \end{cases}
\end{equation*}
Hence, we have
\begin{equation*}
    \bigsqcup_{p:A(l,p)=0} F_\alpha \cap C_{\gamma p} = \begin{cases}
        F_\alpha^{\gamma \alpha_{|\gamma|}}, \quad \gamma \in \llbracket \alpha \rrbracket \setminus \{\alpha\} \text{ and } A(l,\alpha_{|\gamma|})=0,\\
        \emptyset, \quad \text{otherwise};
    \end{cases}
\end{equation*}
and
\begin{equation*}
    \bigsqcup_{k:A(j,k)=0}  C_{\alpha k} \cap F_\gamma = \begin{cases}
        F_\gamma^{\alpha \gamma_{|\alpha|}}, \quad \alpha \in \llbracket \gamma \rrbracket \setminus \{\gamma\} \text{ and } A(j,\gamma_{|\alpha|})=0,\\
        \emptyset, \quad \text{otherwise};
    \end{cases}
\end{equation*}
In addition,
\begin{equation*}
    \bigsqcup_{m=0}^{|\alpha|-1}\bigsqcup_{k\neq \alpha_m} \bigsqcup_{p:A(l,p)=0}   C_{\delta^m(\alpha)k} \cap C_{\gamma p} = \begin{cases}
        \mathfrak{D}[\alpha,\gamma,l], \quad \gamma \in \llbracket \alpha \rrbracket \setminus \{\alpha\},\\
        \bigsqcup_{p:A(l,p)=0} C_{\gamma p}, \quad \alpha \notin \llbracket \gamma \rrbracket \text{ and } \gamma \notin \llbracket \alpha \rrbracket,\\
        \emptyset, \quad \text{otherwise};
    \end{cases}
\end{equation*}
and
\begin{equation*}
    \bigsqcup_{k:A(j,k)=0} \bigsqcup_{n=0}^{|\gamma|-1} \bigsqcup_{p \neq \gamma_n} C_{\alpha k} \cap C_{\delta^n(\gamma)p} = \begin{cases}
        \mathfrak{D}[\gamma,\alpha,j], \quad \alpha \in \llbracket \gamma \rrbracket \setminus \{\gamma\},\\
        \bigsqcup_{k:A(j,k)=0} C_{\alpha k}, \quad \alpha \notin \llbracket \gamma \rrbracket \text{ and } \gamma \notin \llbracket \alpha \rrbracket,\\
        \emptyset, \quad \text{otherwise};
    \end{cases}
\end{equation*}
as analogously we proved \eqref{eq:triple_union_C_alpha_k_C_delta}. Also we have
\begin{equation*}
    \bigsqcup_{k:A(j,k)=0} \bigsqcup_{p:A(l,p)=0} C_{\alpha k} \cap C_{\gamma p} = \begin{cases}
        \bigsqcup_{k:A(j,k) = 0, A(l,k)=0}C_{\alpha k}, \quad \alpha = \gamma, \\
        \bigsqcup_{m:A(l,m)=0}  C_{\gamma m}, \quad \alpha \in \llbracket \gamma \rrbracket \setminus \{\gamma\} \text{ and }A(j,\gamma_{|\alpha|}) = 0,\\
        \bigsqcup_{k:A(j,k)=0}  C_{\alpha k}, \quad \gamma \in \llbracket \alpha \rrbracket \setminus \{\alpha\} \text{ and }A(l,\alpha_{|\gamma|}) = 0,\\
        \emptyset, \quad \text{otherwise};
    \end{cases}
\end{equation*}
analogous to the proof of \eqref{eq:double_union_matix_conditions}. The remaining parcel
\begin{align*}
    \bigsqcup_{m=0}^{|\alpha|-1}\bigsqcup_{k\neq \alpha_m} \bigsqcup_{n=0}^{|\gamma|-1}\bigsqcup_{p\neq \gamma_n}  C_{\delta^m(\alpha)k} \cap C_{\delta^n(\gamma )p}
\end{align*}
was already studied in Corollary \ref{cor:4_tuple_union_delta}.
\qed

\chapter{Boolean Algebras} \label{ape:Boolean_algebras}

\begin{definition} Let $X$ be a topological space.
\begin{itemize}
    \item We say $X$ is said to be disconnected when it is the union of two disjoint non-empty open sets. Otherwise, $X$ is said to be connected.
    \item A subset of $X$ is said to be connected if it is connected under its subspace topology.
    \item A subset of $X$ is said to be a connected component of $X$ if it is connected and maximal under the inclusion ordering.
    \item $X$ is said to be totally disconnected if all of its connected components are singletons.
\end{itemize}
\end{definition}

\begin{remark} Connected components of a topological space are closed sets of the whole space and they form a partition, that is, they are disjoint and their union is the whole space.
\end{remark}

\begin{definition}[Boolean Algebra] A Boolean algebra is a six-tuple $(\mathcal{B},\vee,\wedge,\neg,\mathbf{0},\mathbf{1})$, where $\mathcal{B}$ is a non-empty set, $\vee$ and $\wedge$ are two binary operations, resepectively named `and' and `or', $\neg$ is a unary operation called `not', and $\mathbf{0}$ and $\mathbf{1}$ are two distinct elements of $\mathcal{B}$, respectively called `zero' and `one'. Such six-tuple must satisfy, for every $a$, $b$, $c \in \mathcal{B}$ the following properties
\begin{itemize}
    \item idempotency on $\vee$ and $\wedge$:
    \begin{equation*}
        a \vee a = a \quad \text{and} \quad a \wedge a = a;
    \end{equation*}
    \item commutativity on $\vee$ and $\wedge$:
    \begin{equation*}
        a \vee b = b \vee a \quad \text{and} \quad a \wedge b = b \wedge a;
    \end{equation*}
    \item associativity on $\vee$ and $\wedge$:
    \begin{equation*}
        a \vee (b \vee c) = (a \vee b) \vee c \quad \text{and} \quad a \wedge (b \wedge c) = (a \wedge b) \wedge c;
    \end{equation*}
    \item absorvency on $\vee$ and $\wedge$:
    \begin{equation*}
        a \vee (a \wedge b) = a \quad \text{and} \quad a \wedge (a \vee b) = a;
    \end{equation*}
    \item distributive between $\vee$ and $\wedge$:
    \begin{equation*}
        a \vee (b \wedge c) = (a \vee b) \wedge (a \vee c) \quad \text{and} \quad a \wedge (b \vee c) = (a \wedge b) \vee (a \wedge c);
    \end{equation*}
    \item $\mathbf{1} \wedge a = a$ and $\mathbf{0} \vee a = a$;
    \item $a \wedge (\neg a) = \mathbf{0}$ and $a \vee (\neg a) = \mathbf{1}$.
\end{itemize}
\end{definition}

\begin{definition}[Generated Boolean Algebras] Let $\mathcal{B}$ be a Boolean algebra and $\mathcal{A} \subseteq \mathcal{B}$. The Boolean algebra generated by $\mathcal{A}$ is the smallest Boolean subalgebra of $\mathcal{B}$ containing $\mathcal{A}$. Equivalently, the Boolean algebra generated by $\mathcal{A}$ is the intersection of all Boolean subalgebras of $\mathcal{B}$ containing $\mathcal{A}$.
\end{definition}

\begin{proposition}\label{prop:boolean_subalgebra_DNF} Let $\mathcal{B}$ be a Boolean algebra and $\mathcal{A} \subseteq \mathcal{B}$, and denote by $\mathcal{B}(\mathcal{A})$ the Boolean subalgebra generated by $\mathcal{A}$. Every element of $\mathcal{B}(\mathcal{A})$ has a disjunctive normal form (DNF), that is, if $C \in \mathcal{B}(\mathcal{A})$, then
\begin{equation*}
    C = \bigvee_{i=1}^n \bigwedge_{j=1}^{\phi(n)} C_{ij},
\end{equation*}
where $\phi: \{1,...,n\} \to \mathbb{N}$ is a function and either $C_{ij} \in \mathcal{A}$ or $C_{ij}^c \in \mathcal{A}$.
\end{proposition}

\begin{proof} Let $\mathcal{C}$ be all the elements written in DNF using elements of $\mathcal{A}$. It is straighforward that $\mathcal{C} \subseteq \mathcal{B}(\mathcal{A})$. It remains to prove that $\mathcal{C} \supseteq \mathcal{B}(\mathcal{A})$. Observe that the join of two elements of $\mathcal{C}$ is in $\mathcal{C}$. We claim that the complement of an element in $\mathcal{C}$ is also in $\mathcal{C}$. This is proven in three steps:
\begin{itemize}
    \item if $K \in \mathcal{C}$ and either $V \in \mathcal{A}$ or $V^c \in \mathcal{A}$, then $V \wedge K \in \mathcal{C}$. Indeed, if
    \begin{equation*}
        K = \bigvee_{i=1}^n \bigwedge_{j=1}^{\phi(n)} K_{ij},
    \end{equation*}
    then
    \begin{equation*}
        K \wedge V= \bigvee_{i=1}^n \bigwedge_{j=1}^{\phi(n)} \left(K_{ij} \wedge V\right),
    \end{equation*}
    which is in DNF using elements in $\mathcal{A}$;
    \item if $K,V \in \mathcal{C}$, then $K \wedge V \in \mathcal{C}$. In fact, each term
    \begin{equation*}
        \bigwedge_{j=1}^{\phi(n)} \left(K_{ij} \wedge V\right)
    \end{equation*}
    can be written in DNF by using the previous step iteratively. Then, the join of all these terms again is in DNF and therefore $K\wedge V \in \mathcal{C}$;
    \item if $K \in \mathcal{C}$, then $K^c \in \mathcal{C}$. Indeed, if
    \begin{equation*}
        K = \bigvee_{i=1}^n \bigwedge_{j=1}^{\phi(n)} K_{ij},
    \end{equation*}
    then
    \begin{equation*}
        K^c = \bigwedge_{i=1}^n \bigvee_{j=1}^{\phi(n)} K_{ij}^c,
    \end{equation*}
    which is the meet of $n$ elements in $\mathcal{C}$ and then $K^c \in \mathcal{C}$ by the previous step.
\end{itemize}
We conclude that $\mathcal{C}$ is closed under the operation $\wedge$ and complementation and it is straighforward that it is also closed under the join operation. Then $\mathcal{C}$ is a Boolean subalgebra of $\mathcal{B}$ containing $\mathcal{A}$ and hence $\mathcal{C} \supseteq \mathcal{B}(\mathcal{A})$. Therefore $\mathcal{C} = \mathcal{B}(\mathcal{A})$. 
\end{proof}


\backmatter \singlespacing   



\end{document}